\documentclass[amssymb]{revtex4}

\usepackage{array}
\usepackage{scalerel}
\usepackage{tabularx}
\usepackage{multirow}
\usepackage{amsmath}
\usepackage{mathtools}
\usepackage{stmaryrd}
\usepackage{amsthm}
\usepackage{graphicx} 
\usepackage{amssymb}
\usepackage{mathrsfs}
\usepackage{subeqnarray}
\usepackage[all]{xy}
\usepackage{accents}
\usepackage{subfigure}

\usepackage[pdftex, bookmarks, colorlinks=true, plainpages = false, citecolor = red, urlcolor = black, filecolor = blue]{hyperref} 
\usepackage[usenames,dvipsnames]{color}
\usepackage[shortlabels]{enumitem}
\setenumerate{listparindent=\parindent}

\usepackage[titletoc]{appendix}

\newcommand{\itemcolor}[1]{
  \renewcommand{\makelabel}[1]{\color{#1}\hfil ##1}}

\DeclareFontFamily{U}{mathb}{\hyphenchar\font45}
\DeclareFontShape{U}{mathb}{m}{n}{
      <5> <6> <7> <8> <9> <10> gen * mathb
      <10.95> mathb10 <12> <14.4> <17.28> <20.74> <24.88> mathb12
      }{}
\DeclareSymbolFont{mathb}{U}{mathb}{m}{n}
\DeclareMathSymbol{\curvearrowright}{3}{mathb}{'361}

\usepackage{calligra}
\usepackage{tikz-cd}
\usepackage{thm-restate}
\DeclareMathAlphabet{\mathcalligra}{T1}{calligra}{m}{n}
\DeclareFontShape{T1}{calligra}{m}{n}{<->s*[2.2]callig15}{}

\def\dlceil{\left\lceil\kern-4.75pt\left\lceil}
\def\drceil{\right\rceil\kern-4.75pt\right\rceil}

\usepackage{etoolbox}
\patchcmd{\appendices}{\quad}{. }{}{}

\global\long\def\ud{\mathrm{d}}

\global\long\def\bRnn{\mathbb{R}_{\geq 0}}
\global\long\def\bZ{\mathbb{Z}}
\global\long\def\bN{\mathbb{N}}
\global\long\def\bZpos{\mathbb{Z}_{>0}}

\global\long\def\bZnn{\mathbb{Z}_{\geq 0}}

\global\long\def\bC{\mathbb{C}}

\global\long\def\ii{\mathfrak{i}}

\newcommand{\sH}{\mathcal{H}}

\newcommand{\re}{\mathfrak{Re}}

\global\long\def\pder#1{\frac{\partial}{\partial#1}}

\global\long\def\pder#1{\frac{\partial}{\partial#1}}
\global\long\def\pdder#1{\frac{\partial^{2}}{\partial#1^{2}}}

\newcommand{\rad}{\textnormal{rad}\,}
\newcommand{\Span}{\textnormal{span}\,}
\newcommand{\sleg}[1]{b_{#1}}

\newcommand{\conf}{\varphi}

\newcommand{\End}{\textnormal{End}\,}

\newcommand{\Hom}{\textnormal{Hom}\,}

\newcommand{\LS}{\mathsf{L}}
\newcommand{\preLS}{\mathsf{pre}\,\text{-}\,\mathsf{L}}
\newcommand{\Gen}{U}

\newcommand{\Lgen}{L}
\newcommand{\Rgen}{R}

\newcommand{\Quo}{\mathsf{Q}}
\newcommand{\LP}{\mathsf{LP}}
\newcommand{\preLP}{\mathsf{pre}\,\text{-}\,\mathsf{LP}}
\newcommand{\LD}{\mathsf{LD}}
\newcommand{\preLD}{\mathsf{pre}\,\text{-}\,\mathsf{LD}}

\newcommand{\ThetaNet}{\Theta}

\newcommand{\smin}{s_{\textnormal{min}}}
\newcommand{\smax}{s_{\textnormal{max}}}
\newcommand{\pmin}{\mathfrak{p}}
\newcommand{\ppmin}{\bar{\pmin}}
\newcommand{\TL}{\mathsf{TL}}
\newcommand{\preTL}{\mathsf{pre}\,\text{-}\mathsf{TL}}
\newcommand{\WJ}{\mathsf{JW}}

\newcommand{\PS}{\mathsf{P}}

\newcommand{\Dom}{\mathsf{Non}}
\newcommand{\Tot}{\mathsf{Tot}}

\newcommand{\BarAction}{\left\bracevert\phantom{A}\hspace{-9pt}\right.}

\newcommand{\im}{\textnormal{im}\,}

\newcommand{\cheque}{{\scaleobj{0.85}{\vee}}}
\newcommand{\SpecialPattern}{\mathsf{SP}}
\newcommand{\SpecialDiagram}{\mathsf{SD}}

\newcommand{\Gram}{\mathscr{G}}

\newcommand{\np}{d}

\newcommand{\OneVec}[1]{\vec{#1}}
\newcommand{\Summed}{n}
\newcommand{\ExtraDefect}{v}
\newcommand{\sIndex}{s}

\newcommand{\DefectSet}{\mathsf{E}}
\newcommand{\multii}{\varsigma}
\newcommand{\multiii}{\varpi}

\newcommand\und[1]{{\underset{\vspace{1pt} \scaleobj{1.4}{\check{}}}{#1}}}
\newcommand{\fds}{{\underset{\vspace{1pt} \scaleobj{1.4}{\check{}}}{\multii}}}
\newcommand{\lds}{{\hat{\multii}}}

\newcommand{\idem}{E}

\newcommand\munderbar[1]{\underaccent{\bar}{#1}}

\newcommand{\super}[1]{^{\scaleobj{0.85}{(#1)}}}
\newcommand{\sub}[1]{_{\scaleobj{0.85}{(#1)}}}
\newcommand{\superscr}[1]{^{\scaleobj{0.85}{#1}}}

\newcommand{\BiForm}[2]{(#1 \BarAction #2)}

\newcommand{\Dim}{D}

\newcommand{\WJProj}{P}
\newcommand{\WJproj}{\WJProj}
\newcommand{\WJProjHat}{\hat{\WJProj}}
\newcommand{\WJEmb}{I}
\newcommand{\WalkMultii}{\varrho}

\newcommand{\one}{\mathbf{1} \hspace*{-.25em} \textnormal{l}}

\newsavebox\CBox
\newcommand\hcancel[2][0.5pt]{%
  \ifmmode\sbox\CBox{$#2$}\else\sbox\CBox{#2}\fi%
  \makebox[0pt][l]{\usebox\CBox}%
  \rule[0.5\ht\CBox-#1/2]{\wd\CBox}{#1}}
  
\newcommand{\Mod}[1]{\ (\mathrm{mod}\ #1)}

\newcommand{\DPle}{\prec}

\definecolor{amber}{rgb}{1.0, 0.75, 0.0}

\newcommand{\blue}{\textcolor{blue}}

\newcommand{\red}{\textcolor{red}}

\newcommand{\be}{\begin{equation}}
\newcommand{\ee}{\end{equation}}
\newcommand{\bea}{\begin{eqnarray}}
\newcommand{\eea}{\end{eqnarray}}
\newcommand{\bean}{\begin{eqnarray*}}
\newcommand{\eean}{\end{eqnarray*}}

\theoremstyle{plain}
\newtheorem*{theorem*}{Theorem}
\newtheorem{theorem}{Theorem}
\numberwithin{theorem}{section} 
\newtheorem{prop}[theorem]{Proposition}
\numberwithin{prop}{section} 
\newtheorem{cor}[theorem]{Corollary}
\numberwithin{cor}{section} 
\newtheorem{lem}[theorem]{Lemma}
\numberwithin{lem}{section} 
\newtheorem{conj}[theorem]{Conjecture}
\numberwithin{conj}{section} 

\newtheorem{quest}[theorem]{Question}
\numberwithin{quest}{section} 

\newtheorem{claim}[theorem]{Claim}

\numberwithin{InductAssump}{section} 

\numberwithin{comment}{section}

\theoremstyle{definition}
\newtheorem{remark}[theorem]{Remark}
\numberwithin{remark}{section} 

\numberwithin{recipe}{section} 
\newtheorem{defn}[theorem]{Definition}
\numberwithin{defn}{section} 
\newtheorem{example}[theorem]{Example} 
\numberwithin{example}{section} 


\newcounter{parentnumber}

\def\clap#1{\hbox to 0pt{\hss#1\hss}}
\def\mathclap{\mathpalette\mathclapinternal}

\def\mathclapinternal#1#2{%
\clap{$\mathsurround=0pt#1{#2}$}}

\makeatletter
\DeclareRobustCommand{\cev}[1]{%
  \mathpalette\do@cev{#1}%
}
\newcommand{\do@cev}[2]{%
  \fix@cev{#1}{+}%
  \reflectbox{$\m@th#1\vec{\reflectbox{$\fix@cev{#1}{-}\m@th#1#2\fix@cev{#1}{+}$}}$}%
  \fix@cev{#1}{-}%
}
\newcommand{\fix@cev}[2]{%
  \ifx#1\displaystyle
    \mkern#23mu
  \else
    \ifx#1\textstyle
      \mkern#23mu
    \else
      \ifx#1\scriptstyle
        \mkern#22mu
      \else
        \mkern#22mu
      \fi
    \fi
  \fi
}

\numberwithin{equation}{section}

\renewcommand{\thesection}{\arabic{section}} 

\makeatother

\allowdisplaybreaks

\begin{document}
\title{Standard modules, radicals, and the valenced Temperley-Lieb algebra \vspace*{.5cm}}


\author{\bf Steven M. Flores}
\affiliation{\blue{\tt \small steven.miguel.flores@gmail.com} \\ 
Department of Mathematics and Systems Analysis, \\ 
P.O. Box 11100, FI-00076, Aalto University, Finland}

\author{\bf Eveliina Peltola}
\affiliation{\blue{\tt \small eveliina.peltola@unige.ch} \\ 
Section de Math\'{e}matiques, Universit\'{e} de Gen\`{e}ve, \\
2--4 rue du Li\`{e}vre,  C.P. 64, 1211 Gen\`{e}ve, Switzerland
\vspace*{.5cm}}


\begin{abstract}
\begingroup
\setlength{\parindent}{1.5em}
\setlength{\parskip}{.5em}
This article concerns a 
generalization of the Temperley-Lieb algebra, important in applications to conformal field theory. 
We call this algebra the valenced Temperley-Lieb algebra.
We prove salient facts concerning this algebra and its representation theory, which are 
both of independent interest and used in our subsequent work~\cite{fp3, fp2, fp1}, 
where we uniquely and explicitly characterize the monodromy invariant correlation functions of certain conformal field theories. 
%
\endgroup
\end{abstract}

\maketitle

\vspace*{-1.5cm}
\renewcommand{\tocname}{}
{\hypersetup{linkcolor=black}
\tableofcontents
}

\begingroup
\setlength{\parindent}{1.5em}
\setlength{\parskip}{.5em}

\section{Introduction} \label{Intro}

The Temperley-Lieb algebra is ubiquitous in the mathematics and physics literature.
Named after its discovers H.~Temperley and E.~Lieb, it initially found its role as an algebra related to transfer matrices 
in integrable statistical mechanics models~\cite{tl, pen, pm, bax}.
Later, V.~Jones independently discovered this algebra as a tool for constructing invariants of knots and links~\cite{vj, vj2}. 
This new application established the Temperley-Lieb algebra as a key ingredient 
in the theory of quantum groups~\cite{mj2, lk2, cp, ck, gras, krt} and topological quantum computation~\cite{vt, ckl}.

One of the most important aspects of the Temperley-Lieb algebra, especially in 
applications to physics, is its representation theory, now already well-understood.
The pioneering works include the book~\cite{pm} of P.~Martin and the articles~\cite{hw, gwe} of F.~Goodman and H.~Wenzl,
of combinatorial nature, the more algebraic work of B.~Westbury~\cite{bw},
as well as the rather general framework of cellular algebras developed by J.~Graham and G.~Lehrer in~\cite{gl, gl2}. 
As a very concrete approach, the recent survey~\cite{rsa} by D.~Ridout and Y.~Saint-Aubin is 
perhaps the most comprehensive and accessible treatment of this topic.

The purpose of the present article is to consider a natural generalization of the Temperley-Lieb algebra, 
which we call the ``valenced Temperley Lieb algebra," and to concretely understand its representation theory.
This algebra is motivated by applications to conformal field theory (CFT). 
It is crucial in our subsequent work~\cite{fp3, fp2, fp1}, where we uniquely and 
explicitly characterize monodromy invariant correlation functions of certain CFTs.

In this article, we classify the simple modules of the valenced Temperley Lieb algebra, and give numerous criteria for its semisimplicity.
Using graphical calculus \`a la Kauffman and Lins~\cite{kl}, 
we find the dimensions of and explicit bases for the radicals of the standard modules.
We also find explicit formulas for determinants of Gram matrices on the standard modules by diagonalization.
As a special case, our results imply the corresponding facts for the ordinary Temperley-Lieb algebra, 
and some of our results for the latter are also new.


We organize the introduction as follows. First, in section~\ref{TLSecIntro} we 
collect important results about the Temperley-Lieb algebra and its representation theory. 
In section~\ref{MainResultSec}, we introduce the valenced Temperley Lieb algebra, 
and list the main results regarding this algebra and its representation theory,
presenting them in parallel to the known results about the Temperley-Lieb algebra stated in section~\ref{TLSecIntro}.
Then, in section~\ref{MotivationSec} we briefly discuss our motivation from conformal field theory.
We conclude with the outline and some literary remarks.

\subsection{Background: Temperley-Lieb algebra} \label{TLSecIntro}

To begin, we review definitions and basic properties of the Temperley-Lieb algebra and its standard modules. 
First, for each $n \in \bZnn$, we define an \emph{$n$-link diagram} to be any planar geometric object comprising two vertical lines, 
$n$ distinct marked points (``nodes'') on each line,
and $n$ simple, nonintersecting, planar curves (``links'') between the lines, joining the nodes pairwise.  
The links are determined up to homotopy. Examples of link diagrams are
\begin{align}\label{LinkExs} 
\vcenter{\hbox{\includegraphics[scale=0.275]{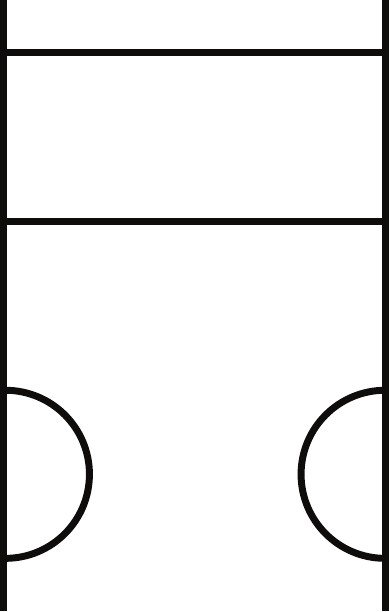}}} \qquad \qquad \text{and} \qquad \qquad
\vcenter{\hbox{\includegraphics[scale=0.275]{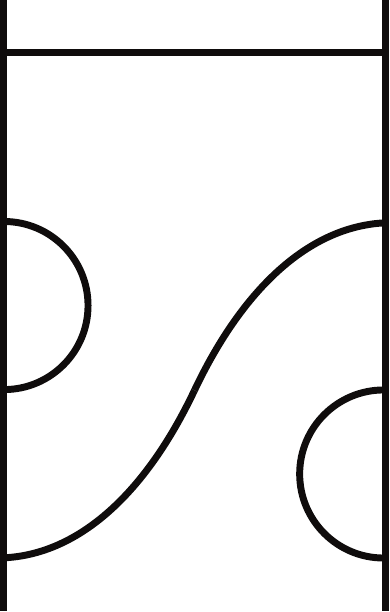}  .}}
\end{align} 
We consider the complex vector space $\TL_n$ of all \emph{tangles}, that is, formal linear combinations of $n$-link diagrams.
We can concatenate two link diagrams in this vector space in a natural manner, as exemplified below:
\begin{alignat}{2} 
\label{TLmult1} 
& \vcenter{\hbox{\includegraphics[scale=0.275]{e-TLalgebra1.pdf}}} \quad 
\vcenter{\hbox{\includegraphics[scale=0.275]{e-TLalgebra2.pdf}}} \quad := \quad 
\vcenter{\hbox{\includegraphics[scale=0.275]{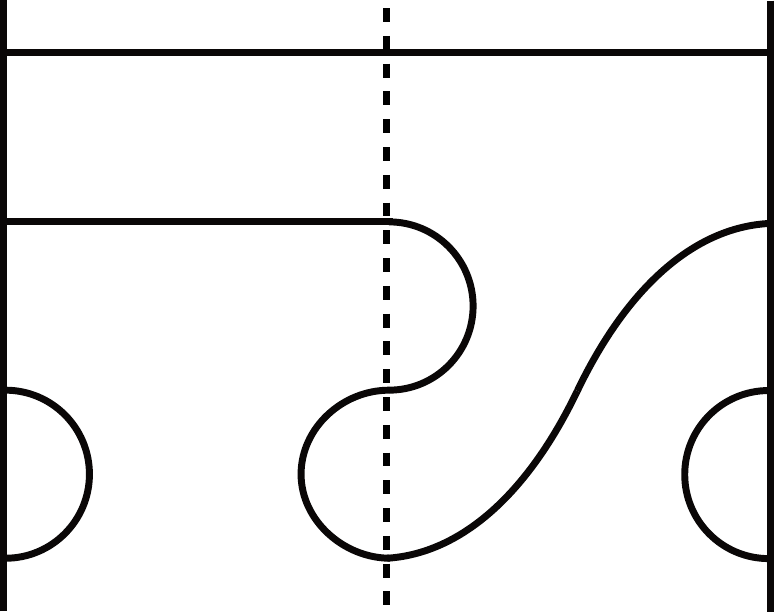}}} \quad = \quad 
\vcenter{\hbox{\includegraphics[scale=0.275]{e-TLalgebra1.pdf} .}}
\end{alignat}
Concatenation on link diagrams
forms a number $k \geq 0$ of internal loops. We remove the loops and multiply the resulting tangle by $\nu^k$,
where $\nu$ is a complex number, called the \emph{loop fugacity}. For instance,
\begin{align}
\label{TLmult4}
& \vcenter{\hbox{\includegraphics[scale=0.275]{e-TLalgebra2.pdf}}} \quad 
\vcenter{\hbox{\includegraphics[scale=0.275]{e-TLalgebra1.pdf}}} \quad := \quad 
\vcenter{\hbox{\includegraphics[scale=0.275]{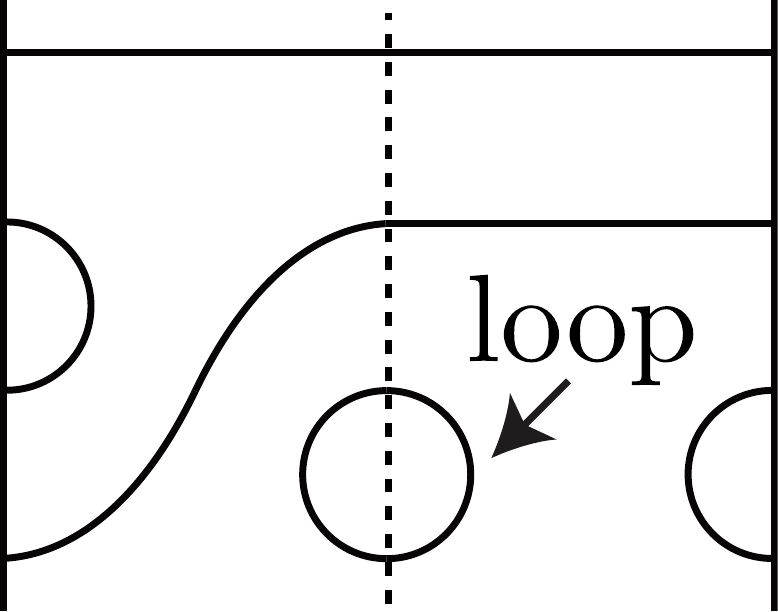}}} \quad = \quad \nu \,\, \times \,\, 
\vcenter{\hbox{\includegraphics[scale=0.275]{e-TLalgebra2.pdf}  .}}
\end{align}
For fixed $\nu \in \bC$, this concatenation recipe endows the vector space $\TL_n$ 
with the structure of an associative, unital algebra, the \emph{Temperley-Lieb algebra} $\TL_n(\nu)$.
Its unit is the link diagram (independent of $\nu$)
\begin{align}\label{Units} 
\mathbf{1}_{\TL_n} \quad = \quad \vcenter{\hbox{\includegraphics[scale=0.275]{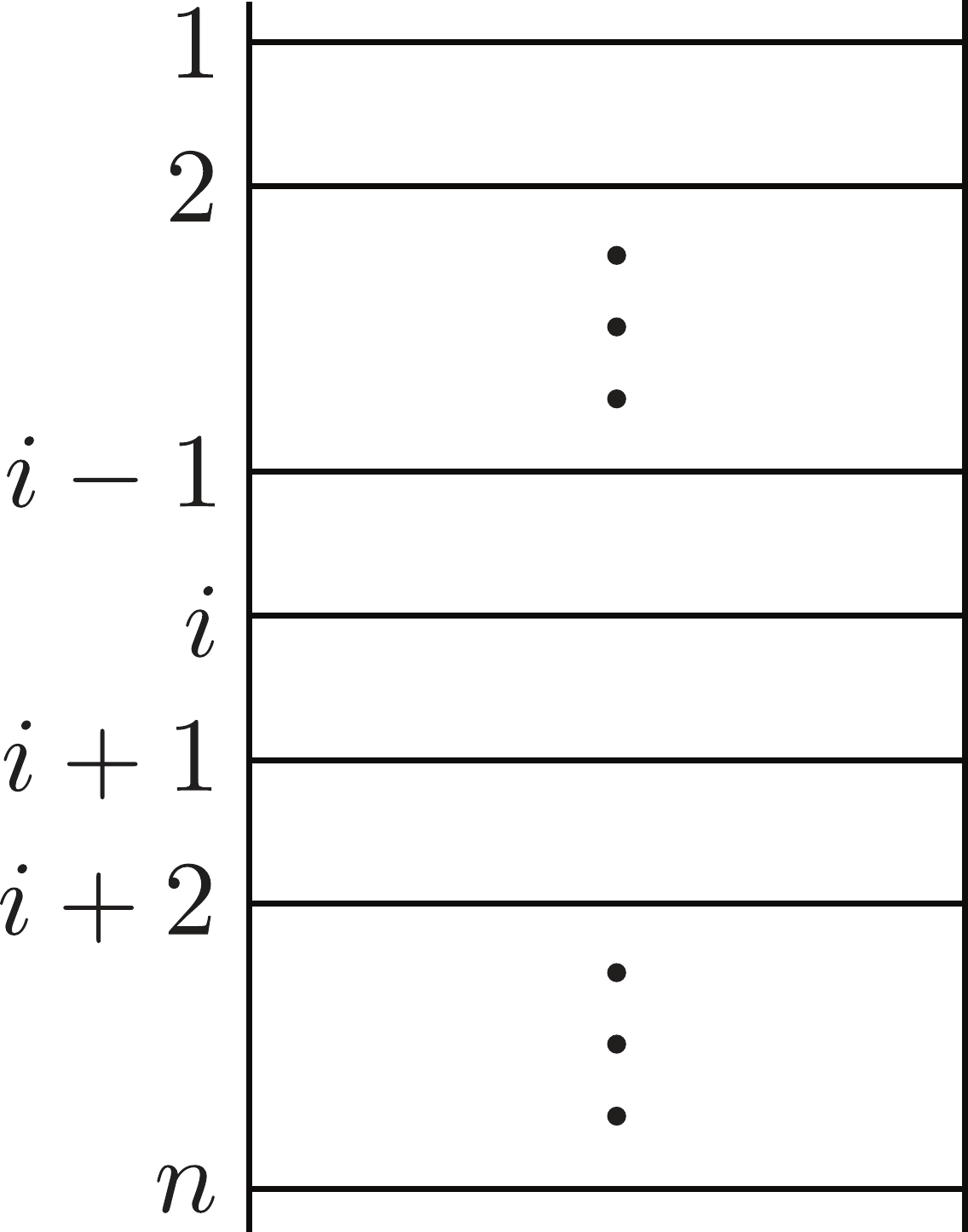}  .}}
\hphantom{\mathbf{1}_{\TL_n} \quad = \quad}
\end{align}

In the representation theory of $\TL_n(\nu)$, ``link patterns'' are important.  To form a link pattern,
we take any $n$-link diagram having $s$ crossing links, where $s$ is necessarily some number in the set
\begin{align}\label{DefectSet} 
\DefectSet_n := \{ n \; \text{mod} \; 2, \; (n \; \text{mod} \; 2) + 2, \; \ldots, \; n \}, 
\end{align} 
divide it vertically in half, discard the right half, and rotate the left half by $\pi/2$ radians:
\begin{align} 
\vcenter{\hbox{\includegraphics[scale=0.275]{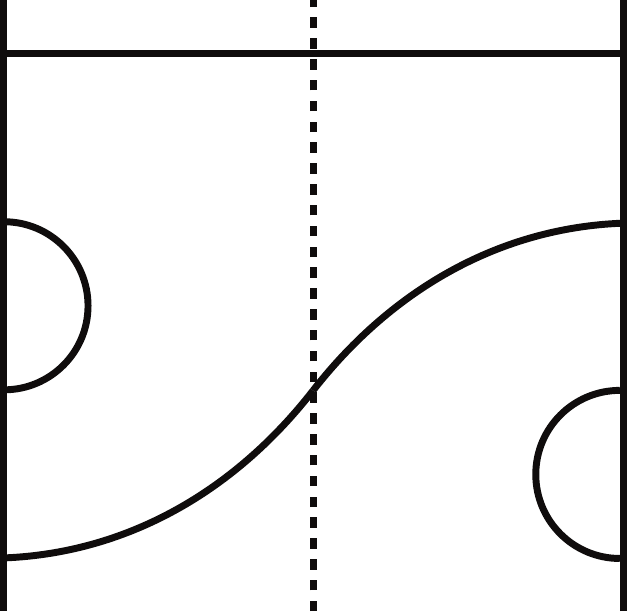}}} \qquad \qquad & \longmapsto \qquad \qquad
\vcenter{\hbox{\includegraphics[scale=0.275]{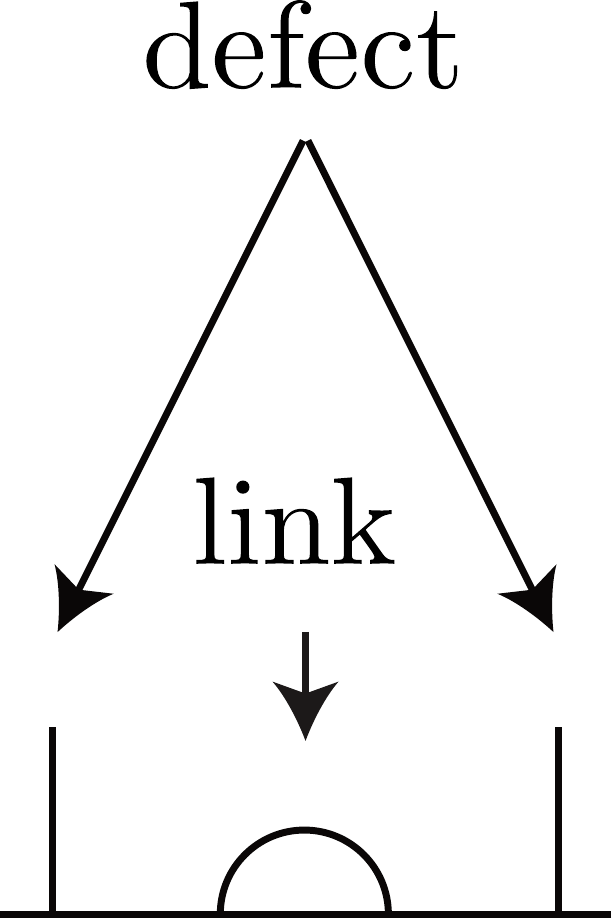}}} \\[1em]
\nonumber
\text{link diagram} \qquad \qquad & \longmapsto \qquad \qquad \text{link pattern}.
\end{align}
We call the remaining left half an \emph{$(n,s)$-link pattern}, and we call each of the broken links in it a \emph{defect}.  
We also call a formal linear combination of $(n,s)$-link patterns with complex coefficients an \emph{$(n,s)$-link state}.

We can concatenate an $n$-link diagram to an $(n,s)$-link pattern (rotated back $-\pi/2$ radians) 
from the left to form a new link pattern.  
Again, we remove any $k$ loops formed by the concatenation and multiply the result by $\nu^k$:
\begin{align}
\label{loopex}
& \vcenter{\hbox{\includegraphics[scale=0.275]{e-TLalgebra1.pdf}}} \quad 
\vcenter{\hbox{\includegraphics[scale=0.275]{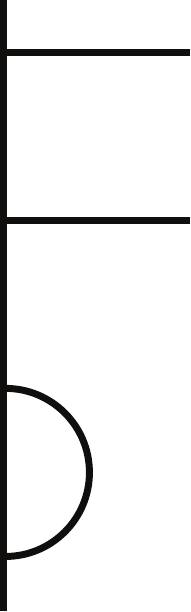}}} \quad
= \quad \vcenter{\hbox{\includegraphics[scale=0.275]{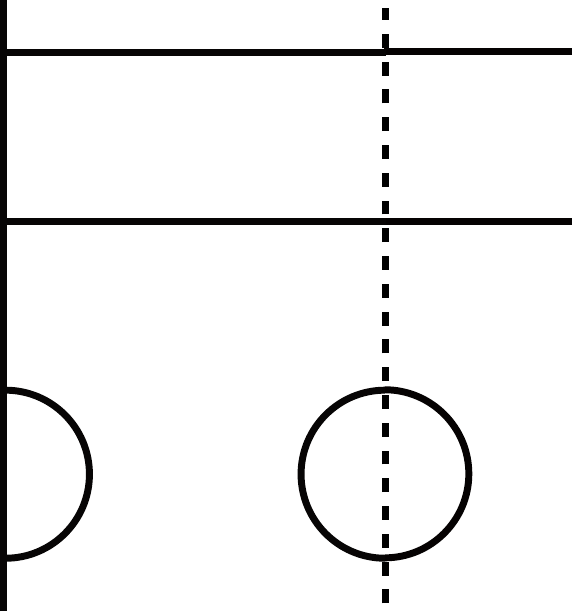}}}  \quad
 = \quad \nu \,\, \times \,\, 
\vcenter{\hbox{\includegraphics[scale=0.275]{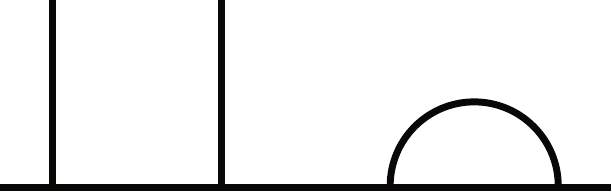} .}}
\end{align}
In order to preserve the number $s$ of defects, we regard all diagrams containing ``turn-back paths'' as zero:
\begin{align}
\label{turnbackex} 
& \vcenter{\hbox{\includegraphics[scale=0.275]{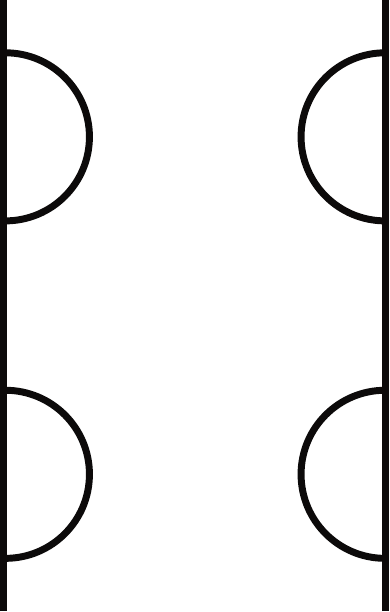}}} \quad 
\vcenter{\hbox{\includegraphics[scale=0.275]{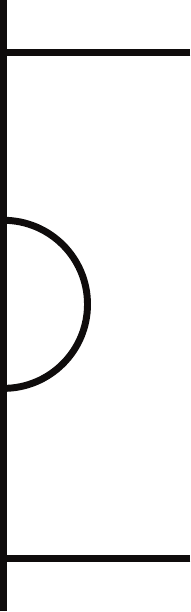}}} \quad
= \quad \vcenter{\hbox{\includegraphics[scale=0.275]{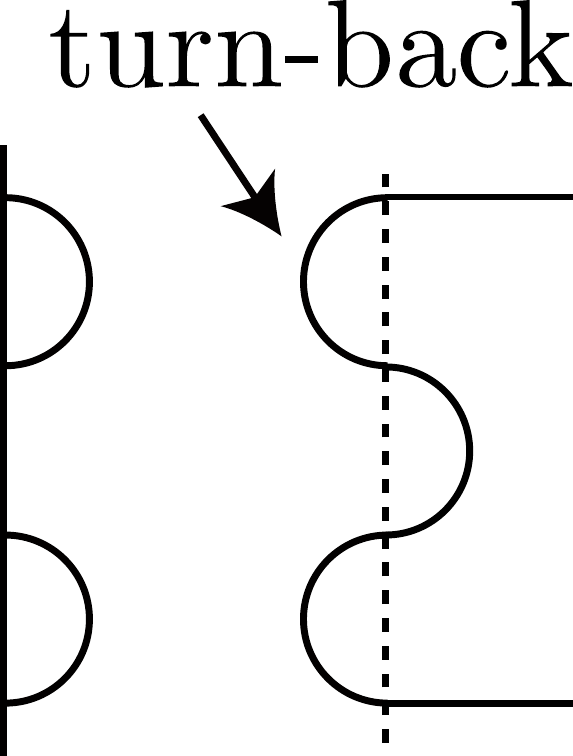}}} \quad
= \quad 0 \,\, \times \,\, 
\vcenter{\hbox{\includegraphics[scale=0.275]{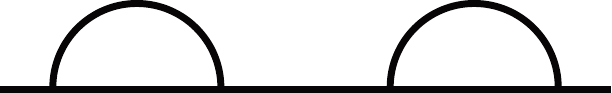}}} 
\quad = \quad 0 .
\end{align}
We thus define an action of $\TL_n(\nu)$ on the complex vector space of $(n,s)$-link states. 
We call this $\TL_n(\nu)$-module a \emph{standard module} and denote it by $\smash{\LS_n\super{s}}$.  
We also define the \emph{link state module} to be the direct sum module
\begin{align} 
\LS_n :=  \bigoplus_{s \, \in \, \DefectSet_n} \LS_n\super{s}.
\end{align}

A certain bilinear form, and in particular its radical, is key to understanding the representation theory of $\TL_n(\nu)$. 
In section~\ref{BilinFormSec}, we define this bilinear form on $\LS_n$ via pairwise concatenation of link patterns, as exemplified below:
\begin{align} \label{PicBiForm}
\bigg( \; \raisebox{1pt}{\includegraphics[scale=0.275]{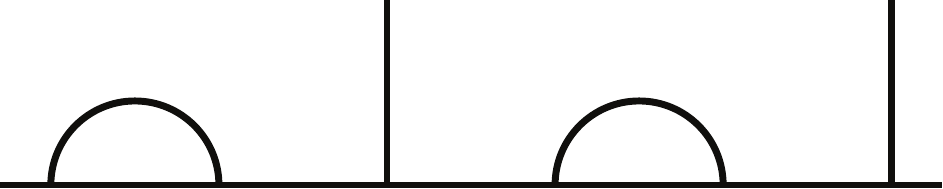}} \;\;  \text{\scalebox{1}[2]{$\BarAction$}} \;\;  
\raisebox{1pt}{\includegraphics[scale=0.275]{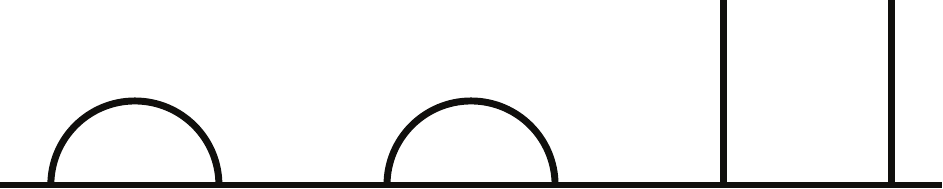}} \; \bigg)
\quad = \quad \bigg( \; \vcenter{\hbox{\includegraphics[scale=0.275]{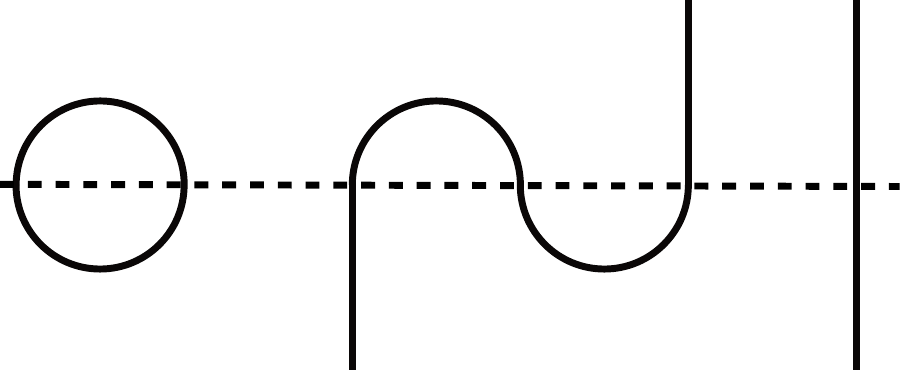}}} \; \bigg).
\end{align}
As before, we replace each internal loop by a multiplicative factor of $\nu$ 
and each turn-back path by a multiplicative factor of zero, 
and now, we also replace each ``through-path'' by a multiplicative factor of one, thus arriving with a complex number: 
\begin{align} 
\bigg( \; \vcenter{\hbox{\includegraphics[scale=0.275]{e-Connectivities10.pdf}}} \; \bigg) \quad = \quad \nu \times 1 \times 1 \quad = \quad \nu , \\
\bigg( \; \vcenter{\hbox{\includegraphics[scale=0.275]{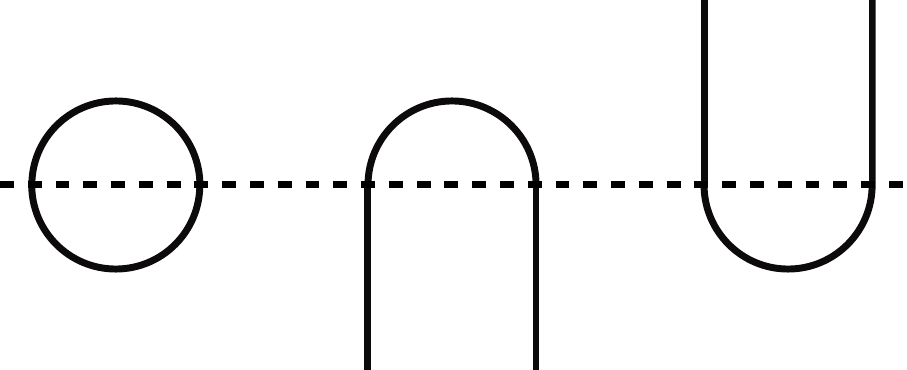}}} \; \bigg) \quad = \quad \nu \times 0 \times 0 \quad = \quad 0 .
\end{align}
We also define the \emph{radical} of $\LS_n$ with respect to the bilinear form to be the vector space
\begin{align} 
\rad\LS_n := \; & \big\{\alpha \in \LS_n \, \big| \, \text{$\BiForm{\alpha}{\beta} = 0$, for all $\beta \in \LS_n$} \big\}.
\end{align}
Properties of the bilinear form ensure that
the radical is a $\TL_n(\nu)$-submodule of $\LS_n$ (see section~\ref{StdModulesSect}). 
It equals a direct sum of the radicals of 
the standard modules $\smash{\LS_n\super{s}}$, which themselves are $\TL_n(\nu)$-submodules of $\smash{\LS_n\super{s}}$.
Hence, we have
\begin{align} 
\rad \LS_n = \bigoplus_{s \, \in \, \DefectSet_n} \rad \LS_n\super{s}, \qquad \text{where} \quad
\rad\smash{\LS_n\super{s}} 
:= \; & \big\{\alpha \in\smash{\LS_n\super{s}} \, \big| \, \text{$\BiForm{\alpha}{\beta} = 0$, for all $\beta \in \smash{\LS_n\super{s}}$} \big\}.
\end{align}
For each standard module, we denote the corresponding quotient module as 
\begin{align} 
\Quo_n\super{s} := \LS_n\super{s} / \rad \LS_n\super{s} . 
\end{align} 
In fact, the nontrivial quotients $\smash{\Quo_n\super{s}}$ form the complete set of non-isomorphic simple $\TL_n(\nu)$-modules~\cite{gl2,rsa}.

\bigskip


Next, we summarize salient properties of the Temperley-Lieb algebra and its representation theory (\ref{result1}--\ref{result9}).
Most of them are well-known, but also follow as special cases of the results of the present article.
First, we introduce notation.
\begin{itemize}[leftmargin=*] 
\item We parameterize the loop fugacity parameter $\nu \in \bC$ by a nonzero complex number $q \in \bC^\times = \bC \setminus \{0\}$ as follows:
\begin{align}\label{fugacity} 
\nu = -q - q^{-1}. 
\end{align} 
\item For each $q \in \bC^\times$, we define
\begin{align} \label{MinPower} 
\pmin(q) : = 
 \begin{cases} \infty, & \text{$q$ is not a root of unity}, \\
p, & \text{$q=e^{\pi ip'/p}$ for coprime $p,p' \in \bZpos$}, 
 \end{cases} \qquad \qquad
\ppmin(q) := 
\begin{cases} \infty, & q \in \{\pm1\}, \\ \pmin(q), & q \not\in \{\pm1\} . 
\end{cases} 
\end{align}
\item For each $k \in \bZnn$, we define $\Delta_k$ to be the following integer:
\begin{align} \label{DeltaDefn} 
\Delta_k = \Delta_k(q) := 
\begin{cases} -1, & \text{$k = 0$ and $\pmin(q) = \infty$}, \\  k \pmin(q) - 1, & \text{otherwise}. 
\end{cases} 
\end{align} 
\item For each $s \in \bZnn$, we define $k_s \in \bZnn$ and $R_s \in \{0,1,\ldots,\pmin(q) - 1\}$ to be the unique integers such that
\begin{align} 
\label{skDefn} s = \Delta_{k_s} + R_s . 
\end{align}

\item We define the ``generic parameter'' set
\begin{align} 
\Dom_n\super{s} 
\hphantom{:}= \Big\{ q \in \bC^\times \,\big|\, 
\text{either $R_s = 0$, or $\frac{n - s}{2} \in \{ 0, 1, \ldots, \pmin(q) - 1 - R_s\}$}  \Big\},
\end{align}
whose complement within $\bC$ has Lebesgue measure zero. We denote
\begin{align}
\Dom_n := \bigcap_{s \, \in \, \DefectSet_n} \Dom_n\super{s} 
= \big\{ q \in \bC^\times \,\big|\, \textnormal{either $n < \ppmin(q),$ or if $n$ is odd, $q = \pm \ii$} \big\}.
\end{align}
It follows from~\eqref{fugacity} that $q = \pm \ii$ is and only if $\nu = 0$, and (c.f. lemma~\ref{qtonulem} in section~\ref{rofSect31})
\begin{align}
n < \ppmin(q) \qquad \Longleftrightarrow \qquad
\nu^2 \neq 4\cos^2\left(\frac{\pi p'}{p}\right) \quad \parbox{5cm}{\textnormal{for any $p',p \in \bZpos$ coprime \\ and satisfying $0 < p' < p \leq n.$}}
\end{align} 
Finally, we denote
\begin{align} \label{Tot} 
\Tot_n\super{s} := 
\begin{cases} \emptyset, & s \neq 0, \\ \{\pm \ii \}, & s = 0. 
\end{cases}
\end{align}
\end{itemize}

\begin{enumerate}[leftmargin=*, label = $\TL$\arabic*$_{n}$., ref = $\TL$\arabic*$_{n}$]
\itemcolor{red}
\item \label{result1} \textnormal{\cite[above theorem~\red{2.4}]{rsa}:} 
We have $\dim \smash{\LS_n\super{s}} = \Dim_n\super{s}$, 
where $\smash{\{\Dim_n\super{s}\}_{s \in \DefectSet_n}}$ is the unique solution to the recursion
\begin{align} 
\Dim_n\super{s} \hspace{.1cm}
 = \hspace{.5cm} \sum_{\mathclap{r \, \in \, \DefectSet_{n-1} \, \cap \, \{ s \pm 1\} }}  
\quad \Dim_{n-1}\super{r} = \begin{cases} \Dim_{n-1}\super{1}, & s = 0, \\[5pt] 
\Dim_{n-1}\super{s-1} + \Dim_{n-1}\super{s+1}, & s \in \{1,2,\ldots, n-1\}, \\[5pt] 
\Dim_{n-1}\super{n-1}, & s = n, 
\end{cases} 
\qquad \quad \text{and} \quad \qquad \Dim_1\super{1} = 1 .
\end{align}

\item \label{result2} \textnormal{\cite[above theorem~\red{2.4}]{rsa}:} 
With $C_n := \frac{1}{n+1} \binom{2n}{n} = \smash{D_{2n}\super{0}}$ denoting the $n$:th Catalan number, we have
\begin{align} \label{Dim34}
\dim \TL_n(\nu) = C_n = \sum_{s \, \in \, \DefectSet_n} \big(\dim \LS_n\super{s}\big)^2. 
\end{align}

\item \label{result4} \textnormal{\cite[theorem~\red{2.4}]{rsa}:} 
Let $\mathsf{A}_n(\nu)$ be the associative, unital algebra with  
generators $\{ \Gen_i \}_{i=1}^{n-1}$
and relations
\begin{alignat}{2}
\label{WordRelations1} 
\Gen_i \Gen_{i \pm 1} \Gen_i &= \Gen_i, \qquad  &&\text{if $1 \leq i\pm1 \leq n-1$}, \\ 
\label{WordRelations2} 
\Gen_i^2 &= \nu \Gen_i, \qquad && \\
\label{WordRelations3} 
\Gen_i \Gen_j &= \Gen_j \Gen_i, \qquad  &&\text{if $|i-j| > 1$},
\end{alignat} 
for all $i,j \in \{ 1, 2, \ldots, n - 1 \}$.  
There exists a unique isomorphism $f_n \colon \mathsf{A}_n(\nu) \longrightarrow \TL_n(\nu)$ of algebras such that
\begin{align}\label{ExtMe} 
f_n(\Gen_i) \quad = \quad \vcenter{\hbox{\includegraphics[scale=0.275]{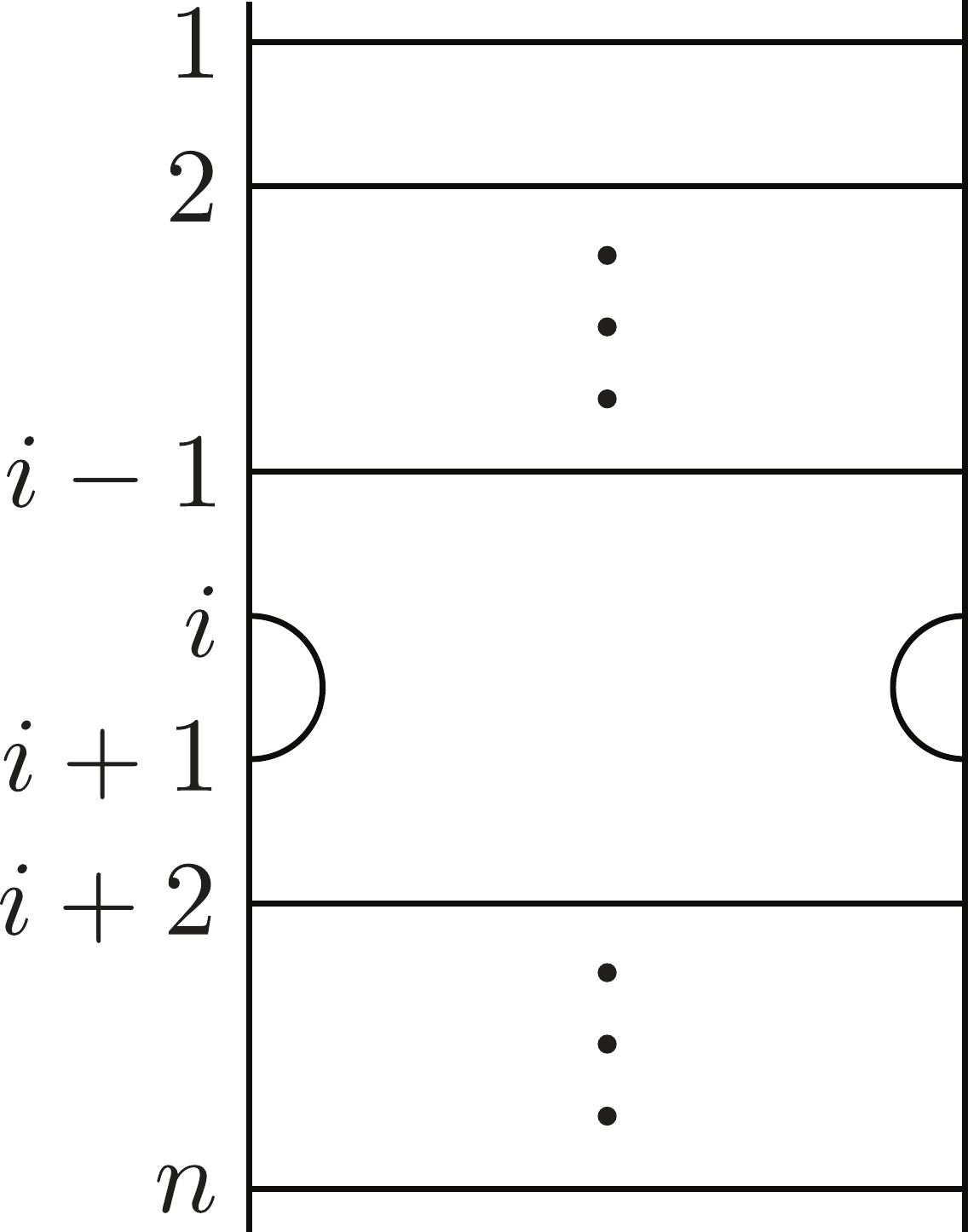}  ,}} 
\hphantom{f_n(\Gen_i) \quad = \quad}
\end{align}
for all $i \in \{1,2, \ldots, n-1\}$.  
Abusing notation, we let $\Gen_i$ also denote the diagram on the right side of~\eqref{ExtMe}.

\item \label{result5} \textnormal{\cite[proposition~\red{3.3}]{rsa}:} 
If $\rad \smash{\LS_n\super{s}} \neq \smash{\LS_n\super{s}}$, then the following hold:
\begin{enumerate}[leftmargin=*, label = \arabic*., ref = \arabic*]
\itemcolor{red}
\item 
The quotient module $\smash{\Quo_n\super{s}}$ is simple, and $\smash{\rad \LS_n\super{s}}$ is the unique maximal proper submodule of~$\smash{\LS_n\super{s}}$.

\smallskip

\item 
The standard module $\smash{\LS_n\super{s}}$ is indecomposable.
\end{enumerate} 

\item \label{result6} \textnormal{\cite[corollary~\red{3.7}]{rsa}:} 
If $\rad \smash{\LS_n\super{s}} \neq \smash{\LS_n\super{s}}$ and 
$\rad \smash{\LS_n\super{r}} \neq \smash{\LS_n\super{r}}$, then we have
\begin{align}
\LS_n\super{s} \cong \LS_n\super{r} \quad \Longleftrightarrow \quad s = r 
\qquad \qquad \textnormal{and} \qquad \qquad 
\Quo_n\super{s} \cong 
\Quo_n\super{r} \quad \Longleftrightarrow \quad s = r . 
\end{align} 

\item \label{result10} \textnormal{[Corollary~\ref{PreFaithfulCor} of the present article]:} 
The link state representation of $\TL_n(\nu)$ on $\LS_n$ induced by the action 
\begin{align} 
(T,\alpha) \quad \longmapsto \quad T\alpha ,
\end{align}
for all tangles $T \in \TL_n(\nu)$ and link states $\alpha \in \LS_n$ is faithful if and only if $\rad \LS_n = \{0\}$.  

\item \label{resultGramN} \textnormal{\cite[proposition~\red{4.5} and theorem~\red{4.7}]{rsa}:}
The Gram determinant $\det \smash{\Gram_n\super{s}}$ of the bilinear form $\BiForm{\cdot}{\cdot}$
on $\smash{\LS_n\super{s}}$ has an explicit formula, given in~\eqref{RidoutDet}.
In particular, if $n < \ppmin(q)$, then $\det \smash{\Gram_n\super{s}} \neq 0$, for all $s \in \DefectSet_n$.

\item \label{resultRadN} \textnormal{[Proposition~\ref{BigTailLem} of the present article]:}
The collection
$\smash{\big\{ \hcancel{\,\alpha} \,\big|\, \alpha \in \smash{\LP_n\super{s}}, \, \textnormal{tail}(\alpha) \in \mathsf{R}_n\super{s} \big\}}$,
where $\smash{\LP_n\super{s}}$ is the set of $(n,s)$-link patterns, 
$\textnormal{tail}(\alpha)$ is defined via~(\ref{TailDef},~\ref{Jindex2}), 
and $\mathsf{R}_n\super{s}$ is defined beneath~\eqref{Alltails},
is a basis for $\smash{\rad \LS_n\super{s}}$.

\item \label{result7} 
\textnormal{\cite[proposition~\red{5.1}]{rsa}:} 
We have $\dim \rad \smash{\LS_n\super{s}} = \smash{\hcancel{\Dim}_n\super{s}}$, 
where $\smash{\{\hcancel{\Dim}_n\super{s}\}_{s \in \DefectSet_n}}$ is the unique solution to the recursion
\begin{align} 
\hspace*{-3mm}
\hcancel{\Dim}_n\super{s} = \begin{cases} 0, & R_s = 0, \\ 
\hcancel{\Dim}_{n-1}\super{s-1} + \Dim_{n-1}\super{s+1}, & R_s = \pmin(q) - 1, \\ 
\hcancel{\Dim}_{n-1}\super{s-1} + \hcancel{\Dim}_{n-1}\super{s+1}, & R_s \in \{1, 2, \ldots, \pmin(q) -2\},
\end{cases}  
\quad \quad \text{and} \quad \qquad \hcancel{\Dim}_1\super{1} = 0 ,
\end{align}
involving the numbers from item~\ref{result1}, with the convention that $\hcancel{\Dim}_{n-1}\super{-1} = 0$.
In particular, we have
\begin{align}\label{radCond} 
\rad \smash{\LS_n\super{s}} = \{0\} \qquad \Longleftrightarrow \qquad q \in \Dom_n\super{s}, 
\end{align}
so $\rad \LS_n$ is trivial if and only if $q \in \Dom_n$. Also, we have
\begin{align}\label{radCond2} 
\rad \smash{\LS_n\super{s}} = \LS_n\super{s} \qquad \Longleftrightarrow \qquad q \in \Tot_n\super{s}.
\end{align}

\item \label{result8} 
\textnormal{\cite[theorem~\red{8.1}]{rsa}:} 
\begin{enumerate}[leftmargin=*, label = \arabic*., ref = \arabic*]
\itemcolor{red}
\item 
If $\nu = 0$ and $n \in 2\bZpos$, then the collection $\smash{\big\{ \Quo_n\super{s} \,\big| \, s \in \DefectSet_n, \, s \neq 0 \big\}}$ 
is the complete set of non-isomorphic simple $\TL_n(\nu)$-modules.

\smallskip

\item 
If $\nu \neq 0$ or $n \not\in 2\bZpos$, then the collection $\smash{\big\{ \Quo_n\super{s} \,\big|\, s \in \DefectSet_n \big\}}$ 
is the complete set of non-isomorphic simple $\TL_n(\nu)$-modules.
\end{enumerate}

\item \label{result9} 
\textnormal{[Consequences of \cite[corollary~\red{4.6}, proposition~\red{5.1}, theorem~\red{8.1}]{rsa}, and theorem~\ref{BigSSTHM} of the present article]:} 
\hspace*{1mm} 
The following statements are equivalent:
\begin{enumerate}[leftmargin=*, label = \arabic*., ref = \arabic*]
\itemcolor{red}
\item 
The Temperley-Lieb algebra $\TL_n(\nu)$ is semisimple, i.e., its Jacobson radical $\rad \TL_n(\nu)$ is trivial.

\smallskip

\item 
We have $\rad \LS_n = \{0\}$.

\smallskip

\item 
The link state representation induced by the action of $\TL_n(\nu)$ on $\LS_n$ is faithful.

\smallskip

\item 
The link state representation induces an isomorphism of algebras from $\TL_n (\nu)$ to 
\raisebox{1pt}{$\smash{\underset{s \, \in \, \DefectSet_n}{\bigoplus}}$}$\smash{\End \LS_n\super{s}}$.

\smallskip

\item 
The collection $\smash{\big\{ \LS_n\super{s} \,\big| \, s \in \DefectSet_n \big\}}$ 
is the complete set of non-isomorphic simple $\TL_n(\nu)$-modules.

\smallskip

\item 
We have $q \in \Dom_n$.
\end{enumerate}
\end{enumerate}

Items~\ref{result1}--\ref{result4} have been well-known since the seminal works of V.~Jones~\cite{vj} and L.~Kauffman~\cite{lk2}.
Items~\ref{result5}--\ref{result6} and~\ref{result8}--\ref{result9}, as well as other salient results  
on the representation theory of $\TL_n(\nu)$ have been proven using combinatorial methods by P.~Martin~\cite{pm}, 
F.~Goodman and H.~Wenzl~\cite{gwe}, and B.~Westbury~\cite{bw}, 
and using the formalism of cellular algebras by J.~Graham and G.~Lehrer~\cite{gl, gl2}.
The article~\cite{rsa} gives a pedestrian survey of these results, including proofs, which we refer to above.
Some of the properties in item~\ref{result9} are not explicitly stated in the literature.
They also follow from theorem~\ref{BigSSTHM} of the present article.

We have not found item~\ref{result10} explicitly stated in the literature. 
Formulas for the Gram determinant in~\ref{resultGramN} appear in~\cite{gl2, rsa}, and
recursion relations similar to item~\ref{result7} appear in~\cite{jm79, bw, gl2, rsa}.
The explicit basis for the radical stated in item~\ref{resultRadN} does not seem to appear in the literature.
Our proof in section~\ref{rofSect1} for item~\ref{resultRadN}
makes use of diagram calculus inspired by Temperley-Lieb recoupling theory~\cite{kl}.

\subsection{Main results: valenced Temperley-Lieb algebra} \label{MainResultSec}

The purpose of this article is to obtain results analogous to items~\ref{result1}--\ref{result9} 
for a natural generalization of the Temperley-Lieb algebra, the ``valenced Temperley Lieb algebra."
We define this algebra via ``valenced tangles.''

Throughout, we let $\multii = (\sIndex_1, \sIndex_2, \ldots, \sIndex_\np)$ be a multiindex of nonnegative entries and 
we denote the sum of its entries by
$\Summed := \sIndex_1 + \sIndex_2 + \dotsm + \sIndex_\np$. We also fix a nonzero complex number $q \in \bC^\times$.
We call a collection of $s$ parallel links within a tangle a ``cable of size $s$," and we illustrate it as one link with label ``$s$" next to it:
\begin{align} 
s
\begin{cases}
\quad \vcenter{\hbox{\includegraphics[scale=0.275]{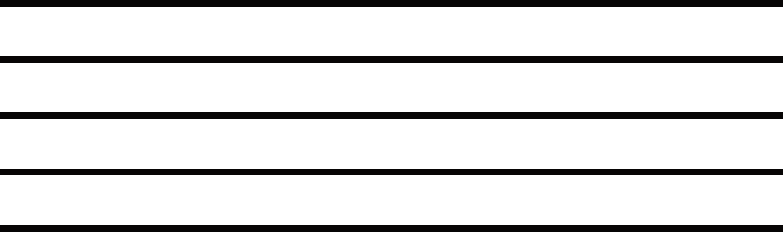}}}
\end{cases}
\; = \quad \raisebox{1pt}{\includegraphics[scale=0.275]{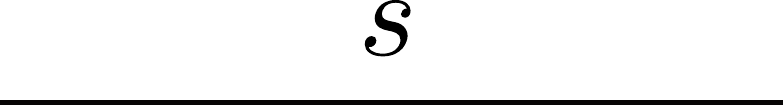}  .} 
\end{align} 
The terminus of such a cable comprises $s$ adjacent nodes, each hosting exactly one endpoint of a link within the cable. 
From now on, we allow the possibility that multiple links terminate at a common node. We illustrate this as 
\begin{align}\label{boxval} 
\begin{rcases}
\vcenter{\hbox{\includegraphics[scale=0.275]{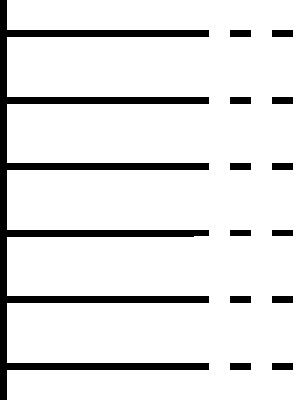}}} \quad
\end{rcases} s
\quad 
\qquad \Longrightarrow \qquad 
\begin{rcases}
\vcenter{\hbox{\includegraphics[scale=0.275]{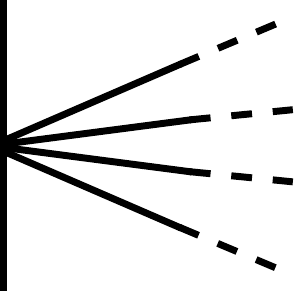}}} \quad
\end{rcases} \; s
\quad 
\qquad \Longrightarrow \qquad 
\quad \vcenter{\hbox{\includegraphics[scale=0.275]{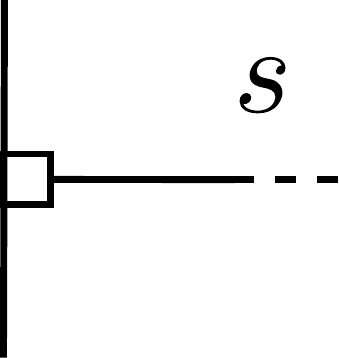}  .}}
\end{align} 
We call the number $s$ of links anchored to a node the ``valence" of that node. 
We call every diagram of the form 
\begin{align} 
\vcenter{\hbox{\includegraphics[scale=0.275]{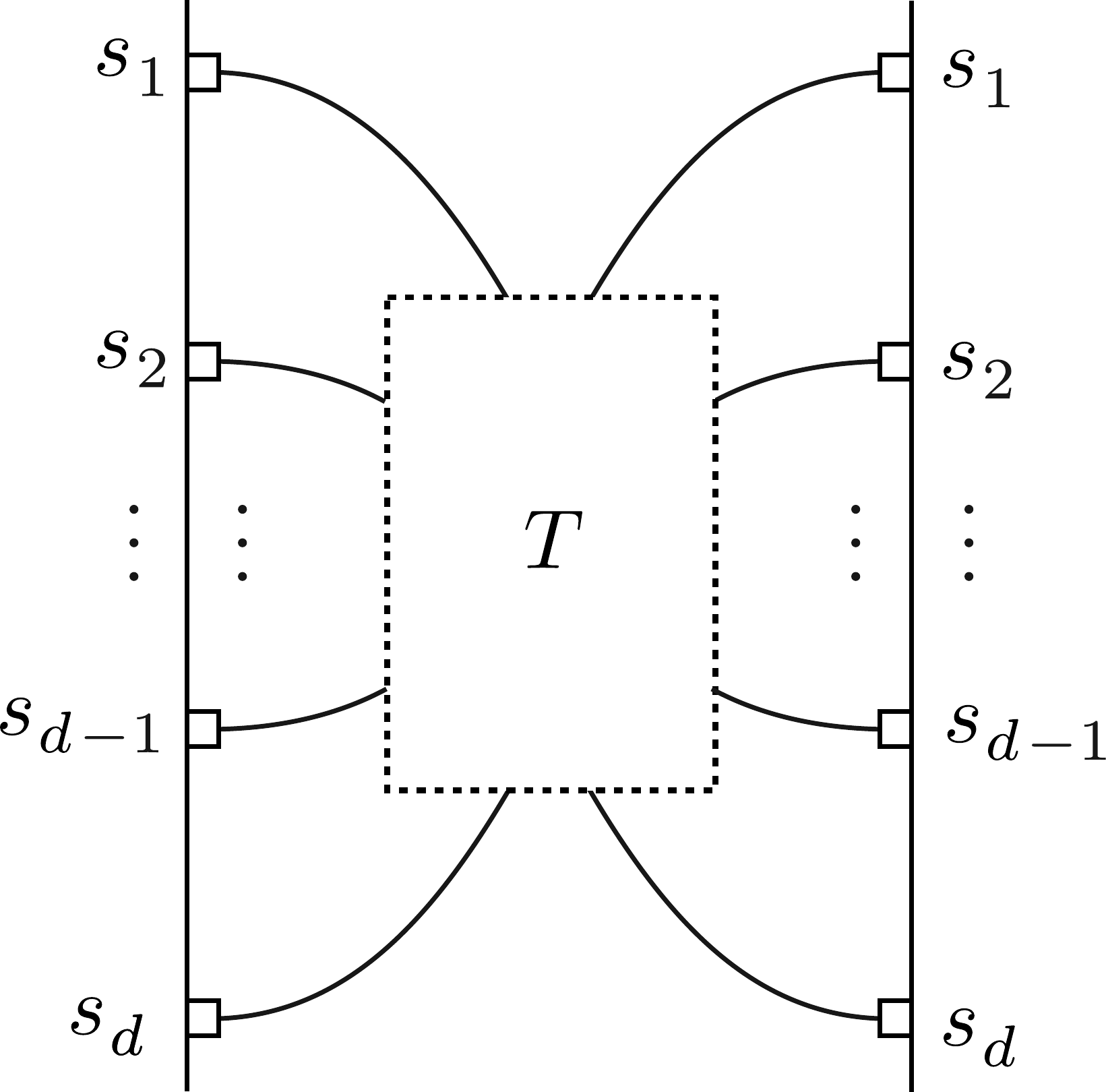}  ,}}
\end{align} 
with $T \in \TL_\Summed(\nu)$, a ``$\multii$-valenced tangle." 
Examples of $(1,2)$-valenced tangles are
\begin{align} \label{AnExnotall0}
\vcenter{\hbox{\includegraphics[scale=0.275]{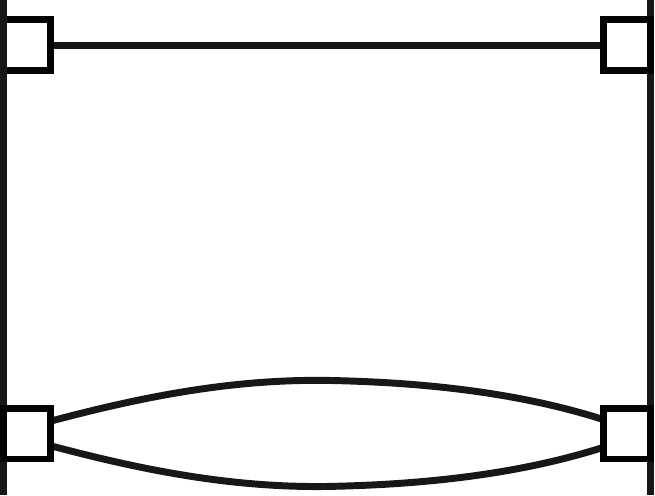}  ,}} \qquad \qquad
\vcenter{\hbox{\includegraphics[scale=0.275]{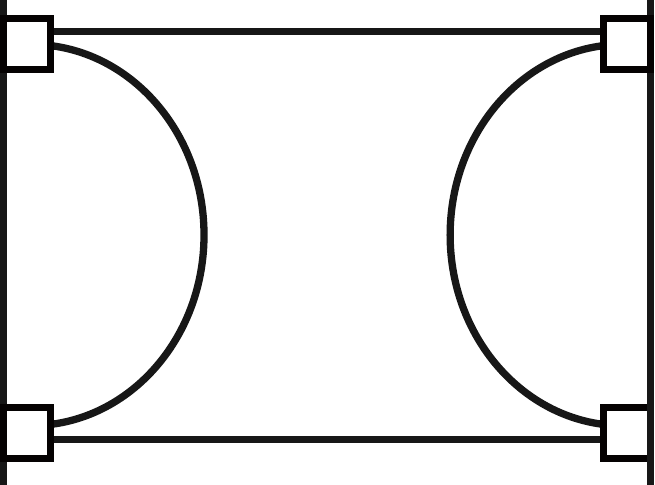}  ,}} \qquad \qquad \text{and} \qquad \qquad
\vcenter{\hbox{\includegraphics[scale=0.275]{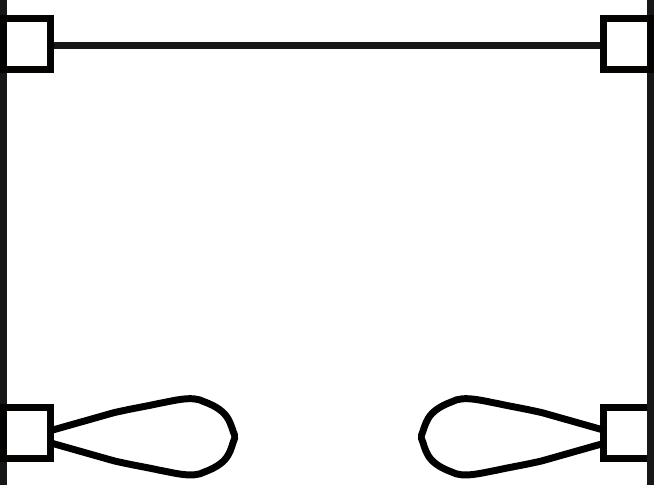}  .}} 
\end{align} 
We restrict our attention to $\multii$-valenced tangles lacking ``loop links," i.e., links with both endpoints at the same node,
by regarding all $\multii$-valenced tangles containing loop links as zero. The third tangle in~\eqref{AnExnotall0} provides an example:
\begin{align} \label{AnEx0}
\vcenter{\hbox{\includegraphics[scale=0.275]{e-WJ_example5_valenced.pdf}}} \quad = \quad 0 .
\end{align} 
We denote by $\TL_\multii$ the space of all $\multii$-valenced tangles modulo those containing loop links
(defined formally in section~\ref{WJLinkDiagSect}).
Next, we fix a fugacity $\nu \in \bC$ parameterized as in~\eqref{fugacity}, and assume that $\max \multii < \ppmin(q)$.
Then we endow the space of $\multii$-valenced tangles with a structure of an associative algebra,
whose multiplication is defined by concatenation of diagrams, as detailed in section~\ref{ValencedCompositionSec}, 
and whose unit element is the valenced tangle
(independent of $\nu$)
\begin{align}\label{ValencedCompProjIntro} 
\mathbf{1}_{\TL_\multii} \quad = \quad \vcenter{\hbox{\includegraphics[scale=0.275]{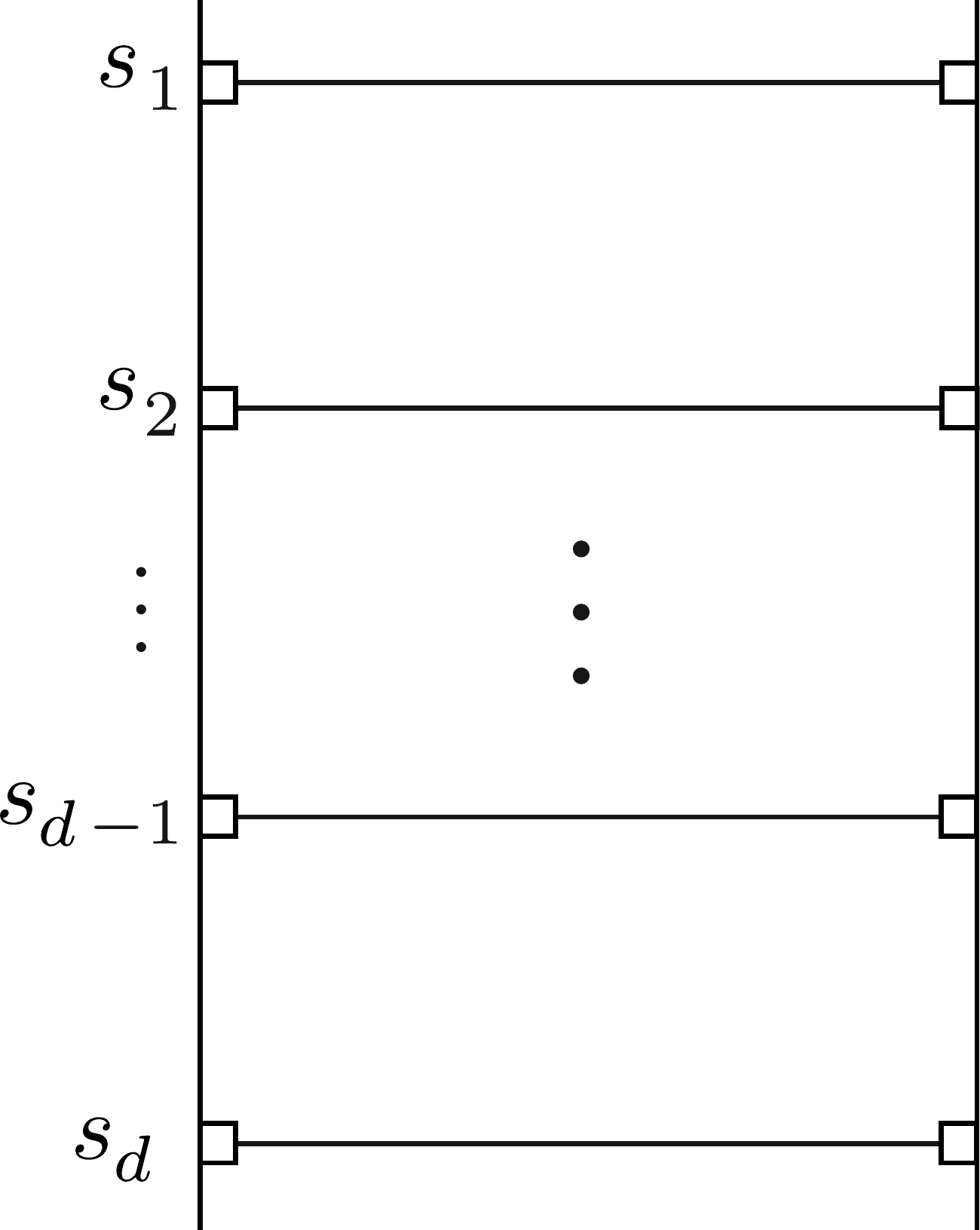} .}}
\hphantom{\mathbf{1}_{\TL_\multii} \quad = \quad}
\end{align} 
We denote the algebra thus obtained by $\TL_\multii(\nu)$ and call it the ``valenced Temperley-Lieb algebra.''
We prove in appendix~\ref{AppWJ} that this algebra 
is isomorphic to a subalgebra $\WJ_\multii(\nu) \subset \TL_\Summed(\nu)$ of the ordinary Temperley-Lieb algebra,
that we call the ``Jones-Wenzl algebra.''  We study the latter algebra in the companion article~\cite{fp0}.

The main purpose of this article is to understand the representation theory of the valenced Temperley-Lieb algebra.
In particular, we consider  ``valenced standard modules" $\smash{\LS_\multii\super{s}}$, which are $\TL_\multii(\nu)$-modules defined analogously 
to the Temperley-Lieb algebra standard modules that appeared in section~\ref{TLSecIntro}.
Using notation~\eqref{boxval}, elements in the valenced standard modules are ``$(\multii,s)$-valenced link states" of the form
\begin{align}\label{JWLinkState2} 
\vcenter{\hbox{\includegraphics[scale=0.275]{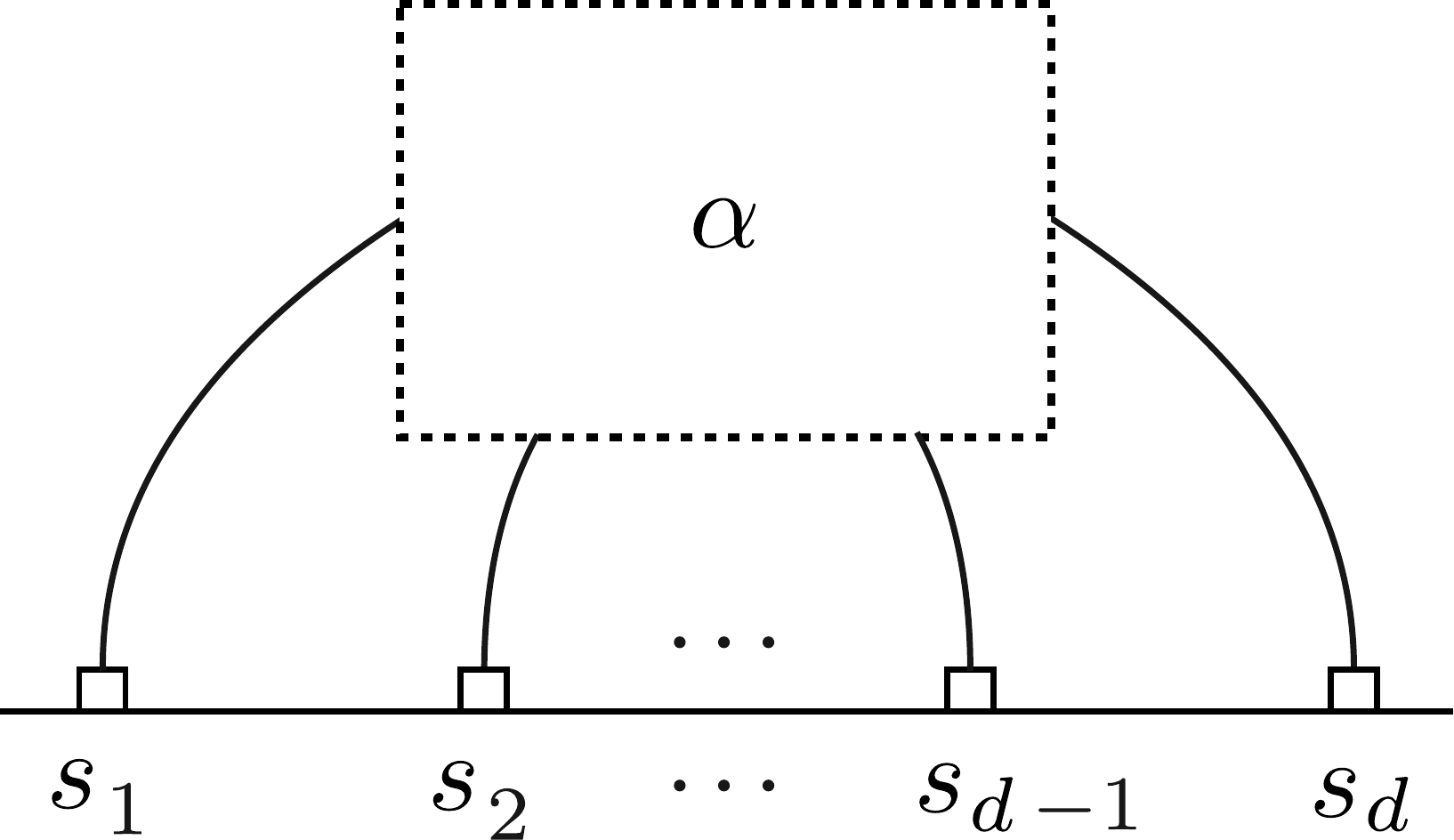} ,}}
\end{align} 
where $\alpha \in \smash{\LS_\Summed\super{s}}$. These objects are defined in detail in section~\ref{DiagramAlgebraSect}.
Examples of $((3,2,2),3)$-valenced link states are
\begin{align}
\vcenter{\hbox{\includegraphics[scale=0.275]{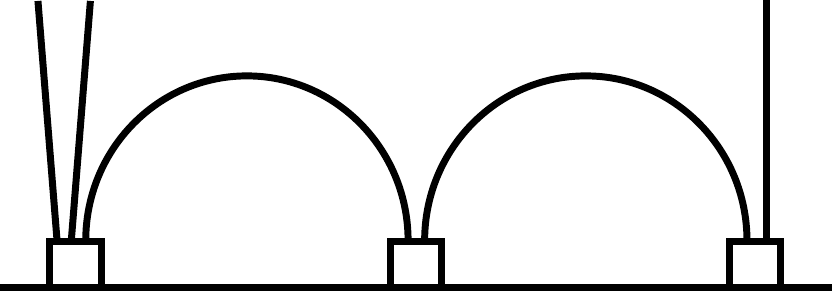}}}
\qquad \qquad \text{and} \qquad \qquad
\vcenter{\hbox{\includegraphics[scale=0.275]{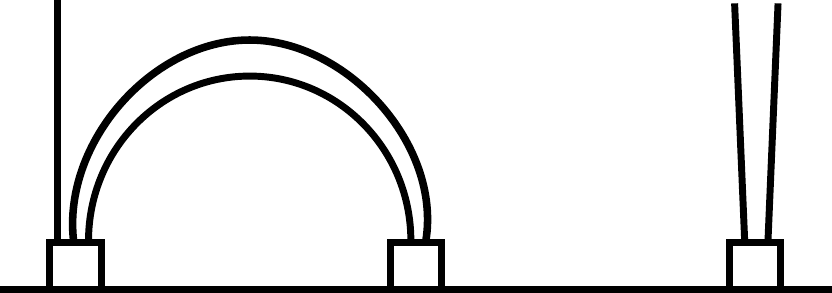}  .}}
\end{align} 
Again, we restrict our attention to $(\multii,s)$-valenced link states lacking loop links, and we let $\smash{\LS_\multii\super{s}}$ denote the space of all such objects. 
We also define the ``valenced link state module" to be the direct sum 
\begin{align} \label{LinkStateModule}
\LS_\multii :=  \bigoplus_{s \, \in \, \DefectSet_\multii} \LS_\multii\super{s} ,
\end{align}
where $\DefectSet_\multii$ denotes the set of all integers $s \geq 0$ such that the module $\smash{\LS_\multii\super{s}}$ is nontrivial. 
(See equation~\eqref{DefSet2} and lemmas~\ref{SpecialDefLem}--\ref{SminLem} in section~\ref{DiagramAlgebraSect} 
for a complete determination of the set $\DefectSet_\multii$.)
When $\max \multii < \ppmin(q)$, the space $\smash{\LS_\multii\super{s}}$ (and hence $\LS_\multii$)
has the structure of a $\TL_\multii(\nu)$-module,
where the action is defined by diagram concatenation in section~\ref{ValencedCompositionSec}.

The $\TL_\multii(\nu)$-module $\LS_\multii$ has a natural bilinear form $\BiForm{\cdot}{\cdot}$, 
which we define diagrammatically in section~\ref{BilinFormSec}. 
We denote the radical of this bilinear form by
\begin{align} 
\rad \LS_\multii := \; & \big\{\alpha \in \LS_\multii \, \big| \, \text{$\BiForm{\alpha}{\beta} = 0$, for all $\beta \in \LS_\multii$} \big\}.
\end{align}
It follows from lemma~\ref{EasyLem2} in section~\ref{StdModulesSect}
that the radical is a $\TL_\multii(\nu)$-submodule of $\LS_\multii$. 
Because the standard modules $\smash{\LS_\multii\super{s}}$ for different $s$ 
are orthogonal, $\rad \LS_\multii$
equals a direct sum of the radicals of the standard modules $\smash{\LS_\multii\super{s}}$,
\begin{align} 
\rad \LS_\multii = \bigoplus_{s \, \in \, \DefectSet_\multii} \rad \LS_\multii\super{s}, \qquad \text{where} \quad
\rad\smash{\LS_\multii\super{s}} := \; & \big\{\alpha \in\smash{\LS_\multii\super{s}} \, \big| \, \text{$\BiForm{\alpha}{\beta} = 0$, 
for all $\beta \in \smash{\LS_\multii\super{s}}$} \big\} ,
\end{align}
and, for each $s \in \DefectSet_\multii$,
$\rad \smash{\LS_\multii\super{s}}$ is a $\TL_\multii(\nu)$-submodule of $\smash{\LS_\multii\super{s}}$.
We denote the corresponding quotient module by 
\begin{align}\label{QuoMod} 
\Quo_\multii\super{s} := \LS_\multii\super{s} / \rad \LS_\multii\super{s}. 
\end{align}

The goal of this article is to find results for the valenced Temperley-Lieb algebra that generalize the analogous results about the Temperley-Lieb algebra 
stated as items~\ref{result1}--\ref{result9} above. The following is a list of our findings:

\begin{enumerate}[leftmargin=*, label = $\TL$\arabic*$_{\multii}$., ref = $\TL$\arabic*$_{\multii}$]
\itemcolor{red}
\item \label{result11} \textnormal{[Lemma~\ref{LSDimLem2}]:} 
We have $\dim \smash{\LS_\multii\super{s}} = \Dim_\multii\super{s}$, 
where $\smash{\{\Dim_\multii\super{s}\}_{s \in \DefectSet_\multii}}$ is the unique solution to the recursion
\begin{align} 
\label{PreRecursion2} 
\Dim_\multii\super{s} \hspace{.1cm}
 = \hspace{.5cm} \sum_{\mathclap{r \, \in \, \DefectSet_{\lds} \, \cap \, \DefectSet\sub{s,t} }}  
\quad \Dim_{\lds}\super{r} 
\quad \qquad \text{and} \quad \qquad \Dim \sub{s}\super{s} = 1 ,
\end{align}
where $\multii = (\sIndex_1, \sIndex_2, \ldots, \sIndex_\np)$, and we denote
$\lds := (\sIndex_1, \sIndex_2, \ldots, \sIndex_{\np-1})$ and $t := \sIndex_\np$.

\item \label{result22} \textnormal{[Special case of corollary~\ref{WJDimLem1}]:} 
We have
\begin{align} \label{PreDim}
\dim \TL_\multii(\nu) = \sum_{s \, \in \, \DefectSet_\multii} \big(\dim \LS_\multii\super{s}\big)^2.
\end{align}

\item \label{result33} \textnormal{[Proposition~\ref{GeneratorPropTwo}]:}
Suppose $\Summed < \ppmin(q)$.
Then the unit~\eqref{ValencedCompProjIntro} 
together with all $\multii$-valenced tangles of the form
\begin{align}
\vcenter{\hbox{\includegraphics[scale=0.275]{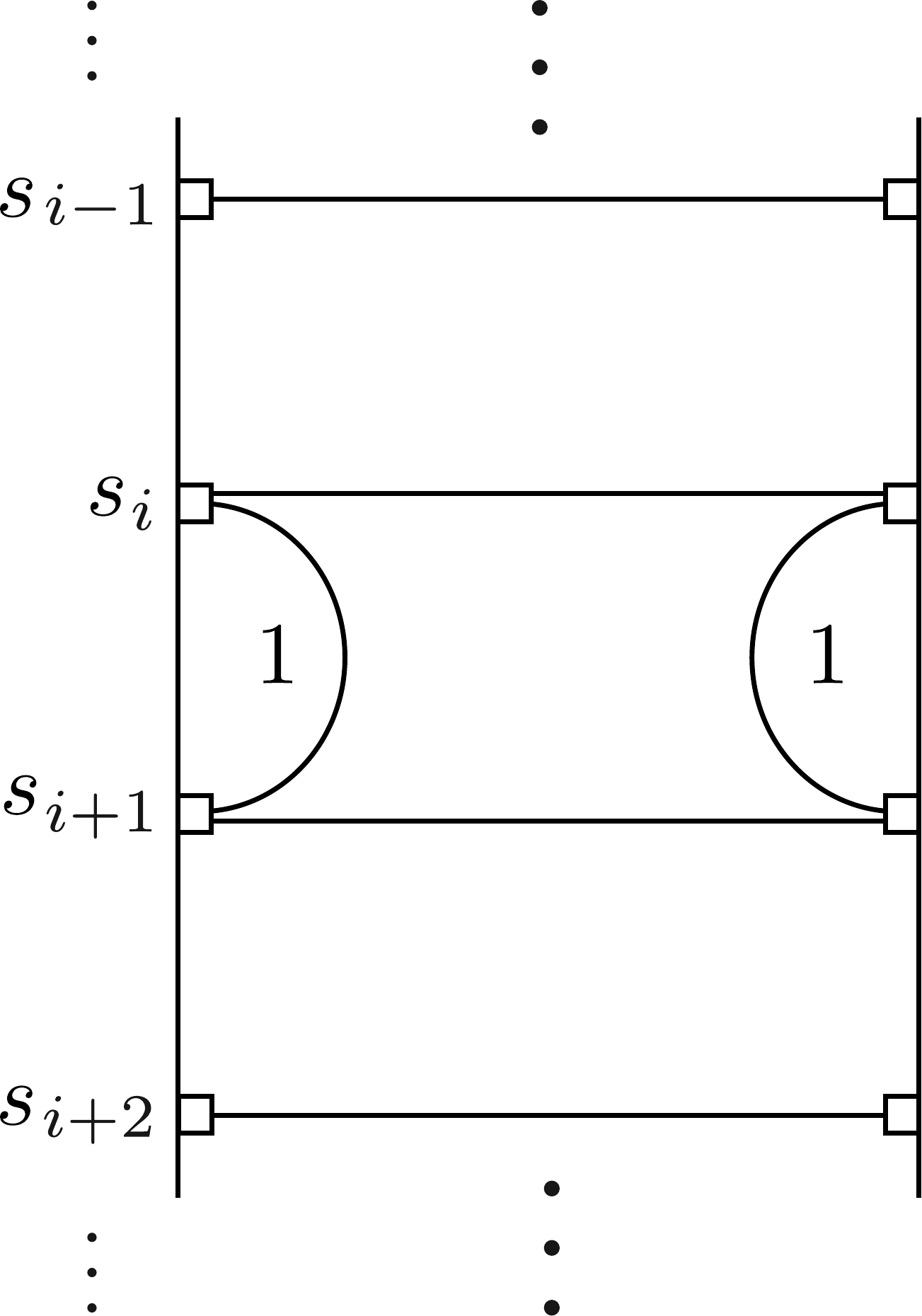} ,}}
\end{align} 
with $i \in \{1,2,\ldots,\np-1\}$ forms a minimal generating set for the valenced Temperley-Lieb algebra $\TL_\multii(\nu)$.

\item \label{result55} \textnormal{[Proposition~\ref{GenLem2}]:} 
Suppose $\max \multii < \ppmin(q)$. If $\rad \smash{\LS_\multii\super{s}} \neq \smash{\LS_\multii\super{s}}$,
then the following hold:
\begin{enumerate}[leftmargin=*, label = \arabic*., ref = \arabic*]
\itemcolor{red}
\item 
The quotient module $\smash{\Quo_\multii\super{s}}$ is simple, and $\smash{\rad \LS_\multii\super{s}}$ 
is the unique maximal proper submodule of~$\smash{\LS_\multii\super{s}}$.

\smallskip

\item 
The standard module $\smash{\LS_\multii\super{s}}$ is indecomposable.
\end{enumerate}

\item \label{result66} \textnormal{[Corollary~\ref{nonisoCor2}]:} 
Suppose $\max \multii < \ppmin(q)$. If $\rad \smash{\LS_\multii\super{s}} \neq \smash{\LS_\multii\super{s}}$ 
and $\rad \smash{\LS_\multii\super{r}} \neq \smash{\LS_\multii\super{r}}$, then we have
\begin{align}
\LS_\multii\super{s} \cong \LS_\multii\super{r} \quad \Longleftrightarrow \quad s = r 
\qquad \qquad \textnormal{and} \qquad \qquad 
\Quo_\multii\super{s} \cong 
\Quo_\multii\super{r} \quad \Longleftrightarrow \quad s = r . 
\end{align} 

\item \label{result1010} 
\textnormal{[Corollary~\ref{PreFaithfulCor}]:} 
Suppose $\max \multii < \ppmin(q)$. The link state representation of $\TL_\multii (\nu)$ on $\LS_\multii$ induced by the action 
\begin{align}
(T,\alpha) \quad \longmapsto \quad T\alpha ,
\end{align}
for all valenced tangles $T \in \TL_\multii(\nu)$ and valenced link states $\alpha \in \LS_\multii$ is faithful
if and only if $\rad \LS_\multii = \{0\}$. 

\item \label{resultGram} \textnormal{[Propositions~\ref{GramDetLem} and~\ref{VanishDetLem2}]:}
Suppose $\max \multii < \ppmin(q)$. 
The Gram determinant $\det \smash{\Gram_\multii\super{s}}$ of the bilinear form $\BiForm{\cdot}{\cdot}$
on $\smash{\LS_\multii\super{s}}$ has an explicit formula, given in~\eqref{DetFormula}.
In particular, if $\Summed < \ppmin(q)$, then $\det \smash{\Gram_\multii\super{s}} \neq 0$, for all $s \in \DefectSet_\multii$.

\item \label{resultRad} \textnormal{[Theorem~\ref{BigTailLem2}]:}
Suppose $\max \multii < \ppmin(q)$. The collection
$\smash{\big\{ \hcancel{\,\alpha} \,\big|\, \alpha \in \smash{\LP_\multii\super{s}}, \, \textnormal{tail}(\alpha) \in \mathsf{R}_\multii\super{s} \big\}}$,
where $\smash{\LP_\multii\super{s}}$ is the set of $(\multii,s)$-valenced link patterns, 
$\textnormal{tail}(\alpha)$ is defined via~(\ref{TailDef},~\ref{Jindex2}), 
and $\mathsf{R}_\multii\super{s}$ is defined in~\eqref{R12},
is a basis for $\smash{\rad \LS_\multii\super{s}}$.

\item \label{result77} \textnormal{[Corollaries~\ref{DnsLemAndRadDimCor2},~\ref{RadDimCor3},~\ref{GridCor}, and proposition~\ref{WholeRadicalImpliesSsmallLem}]:} 
Suppose $\max \multii < \ppmin(q)$. We have $\dim \rad \smash{\LS_\multii\super{s}} = \smash{\hcancel{\Dim}_\multii\super{s}}$, 
where $\smash{\{\hcancel{\Dim}_\multii\super{s}\}_{s \in \DefectSet_\multii}}$ is the unique solution to the recursion
\begin{align} 
\label{PreRadRecurs2}
\hcancel{\Dim}_\multii\super{s} = 
\sum_{r \, \in \, \DefectSet_{\lds} \, \cap \, \DefectSet\sub{s,t}} 
\Big(\one{ \Big\{\Delta_{k_s} < \, \frac{r+s-t}{2} \Big\} } & \one{ \Big\{ \frac{r+s+t}{2} \, < \, \Delta_{k_s+1} \Big\} } \hcancel{\Dim}_{\lds}\super{r} \\
\nonumber
 + & \one{ \Big\{ \Delta_{k_s+1} \, \leq \, \frac{r+s+t}{2} \Big\}} \Dim_{\lds}\super{r} \Big), 
 \qquad \quad \text{and} \qquad \quad \hcancel{\Dim}\sub{s}\super{s} = 0 ,
\end{align}
involving the numbers from item~\ref{result11},
with $\Delta_k$ defined in~\eqref{DeltaDefn} and $\lds$ and $t$ defined in item~\ref{result11}. 

\noindent
With the set $\smash{\Dom_\multii\super{s}}$ 
of full Lebesgue measure defined via (\ref{Domns2},~\ref{Dommultii}) in section~\ref{rofSect31}, 
we have
\begin{align}\label{PreIfonlyIf2} 
\rad \LS_\multii\super{s} = \{0\} \qquad \Longleftrightarrow \qquad \rad \LS_n\super{s} = \{0\} \qquad \Longleftrightarrow \qquad 
q \in \smash{\Dom_\multii\super{s}} ,
\end{align} 
and this in turn implies that $\rad \LS_\multii$ is trivial if and only if 
$q \in \Dom_\multii :=$\raisebox{2pt}{$\smash{\underset{s \, \in \, \DefectSet_\multii}{\bigcap}}$} $\smash{\Dom_\multii\super{s}}$. 

\noindent
Also, with the set $\smash{\Tot_\multii\super{s}}$ 
of zero Lebesgue measure defined via~\eqref{TotDefn} in section~\ref{rofSect32}, 
we have
\begin{align}\label{radCond3} 
\rad \smash{\LS_\multii\super{s}} = \LS_\multii\super{s} \qquad \Longleftrightarrow \qquad q \in \Tot_\multii\super{s}.
\end{align}

\item  \label{result88} \textnormal{[Proposition~\ref{SimpleModuleProp}]:} 
Suppose $\max \multii < \ppmin(q)$. The collection 
$\smash{\big\{ \Quo_\multii\super{s} \,\big| \, s \in \DefectSet_{\multii}, \dim \Quo_\multii\super{s} > 0 \big\}}$ 
is the complete set of  \\ \hspace*{1mm} non-isomorphic simple $\TL_\multii(\nu)$-modules.

\item \label{result99} \textnormal{[Theorem~\ref{BigSSTHM}]:} 
Suppose $\max \multii < \ppmin(q)$. The following statements are equivalent:
\begin{enumerate}[leftmargin=*, label = \arabic*., ref = \arabic*]
\itemcolor{red}
\item
The valenced Temperley-Lieb algebra $\TL_\multii(\nu)$ is semisimple, i.e., its Jacobson radical $\rad\TL_\multii(\nu)$ is trivial.

\smallskip

\item
We have $\rad \LS_\multii = \{0\}$.

\smallskip

\item
The link state representation induced by the action of $\TL_\multii (\nu)$ on $\LS_\multii$ is faithful.

\smallskip

\item
The link state representation induces an isomorphism of algebras from $\TL_\multii (\nu)$ to 
\raisebox{1pt}{$\smash{\underset{s \, \in \, \DefectSet_\multii}{\bigoplus}}$}$\smash{\End \LS_\multii\super{s}}$.

\smallskip

\item
The collection $\smash{\big\{ \LS_\multii\super{s} \,\big| \, s \in \DefectSet_\multii \big\}}$ 
is the complete set of non-isomorphic simple $\TL_\multii(\nu)$-modules.

\smallskip

\item
We have $q \in \Dom_\multii$.
\end{enumerate}
\end{enumerate}

\subsection{Motivation: correlation functions of conformal field theory} \label{MotivationSec}

For us, the main reason to introduce the valenced Templerley-Lieb algebra is its value in applications to conformal field theory (CFT).
In this section, we briefly explain the particular application treated in forthcoming work~\cite{fp1}.
Mathematically, many aspects of CFT are still poorly understood and are presently the subject of active research. However, the problem that we are interested in is well-posed 
and our solution to it is rigorous.
We invite the reader to consult the physics literature for background on CFT, for example~\cite{fms, mh, sr}.

In a CFT, the fundamental objects are conformal fields and their correlation functions. There are different ways to rigorously define such fields in the mathematics literature, 
for example as random distributions~\cite{km} or as formal Laurent series~\cite{sch}. Regardless of these different approaches, the correlation functions make perfect sense 
as functions of several variables. The Temperley-Lieb algebra and its valenced generalization naturally arise when considering the monodromy of 
certain correlation functions~\cite{df1, ffk, mr, gs, fw, fuc, gras}.
Throughout, the central charge of the CFT in question relates to the fugacity parameter $\nu$ via a parameter $\kappa > 0$ as 
\begin{align}
\nu = -2\cos \left( \frac{4\pi}{\kappa} \right) \qquad \Longleftrightarrow \qquad c = \frac{(6-\kappa)(3\kappa-8)}{2\kappa} ,
\end{align} 
and we assume that $\kappa \in (0,8)$ 
is irrational (i.e., $q = \exp(4\pi \ii/\kappa)$ is not a root of unity). 
Under this assumption, all of the equivalent properties in item~\ref{result99} hold.

First, let us consider a special case of the problem we are interested in.
We denote by $\psi_1$ the spinless primary conformal field whose (holomorphic and antiholomorphic) conformal weight equals the quantity
\begin{align}
\sleg{1} := \frac{6-\kappa}{2\kappa} ,
\end{align} 
that is, the Kac weight $h_{1,2}$ or $h_{2,1}$ indexed by the second entry in the first row or column of the Kac table.
We consider the following $n$-point CFT correlation function, with $n \in 2\bZpos$:
\begin{align}\label{corrfunc} 
F_n( \boldsymbol{z}, \bar{\boldsymbol{z}}) =  \langle \psi_1(z_1, \bar{z}_1) \psi_1(z_2, \bar{z}_2) \dotsm \psi_1(z_n, \bar{z}_n) \rangle ,
\end{align} 
where we treat $z_i$ and $\bar{z}_i$ as independent variables, rather than complex conjugates.
The domain of this function is the set of all points $(\boldsymbol{z}, \bar{\boldsymbol{z}}) \in \bC^{2n}$ with $\boldsymbol{z}= (z_1, z_2, \ldots, z_n) \in \bC^n$ 
and $\bar{\boldsymbol{z}} = (\bar{z}_1, \bar{z}_2, \ldots, \bar{z}_n) \in \bC^n$, where no two coordinates of $\boldsymbol{z}$ (resp.~$\bar{\boldsymbol{z}}$) are equal. 
Correlation functions of type~\eqref{corrfunc} should satisfy the following key properties: 
\begin{enumerate}
\itemcolor{red}
\item \label{mon1} The following two decoupled systems of partial differential equations of BPZ type~\cite{bpz}:
\begin{align}
\label{SLEBSAsys1}
\left[ \frac{\kappa}{4}\pdder{z_{i}} 
+ \sum_{j \neq i}\left(\frac{1}{z_{j}-z_{i}}\pder{z_{j}}-\frac{\sleg{1}}{(z_{j}-z_{i})^{2}}\right) \right]
F_n(\boldsymbol{z}, \bar{\boldsymbol{z}}) = 0,  \qquad \text{for all $i\in\{1,2\ldots,n\}$}, \\
\label{SLEBSAsys2}
\left[ \frac{\kappa}{4}\pdder{\bar{z}_{i}} 
+ \sum_{j\neq i}\left(\frac{1}{\bar{z}_{j}-\bar{z}_{i}}\pder{\bar{z}_{j}}-\frac{\sleg{1}}{(\bar{z}_{j}-\bar{z}_{i})^{2}}\right) \right]
F_n(\boldsymbol{z}, \bar{\boldsymbol{z}}) = 0,  \qquad \text{for all $i\in\{1,2\ldots,n\}$}.
\end{align} 

\item \label{mon2} Invariance under all monodromy transformations (defined below) and coordinate permutations.

\item \label{mon4} 
Covariance under all conformal transformations $\conf \colon \bC \longrightarrow \bC$, that is, 
\begin{align}
F_n \big( \conf(\boldsymbol{z}), \conf(\bar{\boldsymbol{z}}) \big)
= \Bigg( \prod_{i=1}^n \partial \conf(z_{i})^{-\sleg{1}} \bar{\partial} \conf(\bar{z}_{i})^{-\sleg{1}} \Bigg) F_n (\boldsymbol{z}, \bar{\boldsymbol{z}}) .
\end{align}

\item \label{mon3} 
There exist numbers $C, p \in \bRnn$ such that the magnitude of correlation function~\eqref{corrfunc} is globally bounded:
\begin{align} 
\big| F_n( \boldsymbol{z}, \bar{\boldsymbol{z}}) \big| \, \leq \,
C\prod_{i<j}^n \max \Big\{ \big|(z_i-z_j)(\bar{z}_i - \bar{z}_j)\big|^p, \big|(z_i-z_j)(\bar{z}_i - \bar{z}_j)\big|^{-p} \Big\} . 
\end{align} 
\end{enumerate}

In item~\ref{mon2} above, we use the convention that the monodromy transformation that winds $z_i$ counterclockwise around $z_j$ 
also simultaneously winds $\bar{z}_i$ clockwise around $\bar{z}_j$, as we illustrate below in~\eqref{half-twist}.

The following question motivates our work, and we answer it in our forthcoming article~\cite{fp1}: 

\begin{quest} \label{BigQuestion} 
What is the dimension of the space of all functions with properties~\ref{mon1}-\ref{mon3}? Can we find an explicit basis?
\end{quest}

Using representation theory and recent results from the series of articles~\cite{sfk1,sfk2,sfk3,sfk4,kp2,kp}, 
in~\cite{fp1} we prove that this space is one-dimensional and we obtain an explicit formula for a function that spans it. 
To find this function, we use $(n,0)$-link states. First, for each $(n,0)$-link pattern, we define
\begin{align}\label{CGform} 
\sH_n[\alpha](\boldsymbol{z}) := \bigg(\prod_{1 \leq i < j \leq n} (z_i - z_j)^{2/\kappa} \bigg) \int_{\Gamma_\alpha} \bigg( \prod_{i = 1}^{\vphantom{n/2} n} \prod_{j = 1}^{n/2} (w_i - z_j)^{-4/\kappa} \bigg) \bigg( \prod_{1 \leq i < j \leq n/2} (w_i - w_j)^{8/\kappa} \bigg) \ud \boldsymbol{w}, 
\end{align} 
where 
the integration surface $\Gamma_\alpha$ is the product of simple contours which, after identifying the $i$:th node of $\alpha$ from the left with the coordinate $z_i$ for each $i \in \{1,2,\dots,n\}$, follow the paths of the links in $\alpha$: 
\begin{align}
\alpha \quad = \quad \vcenter{\hbox{\includegraphics[scale=0.275]{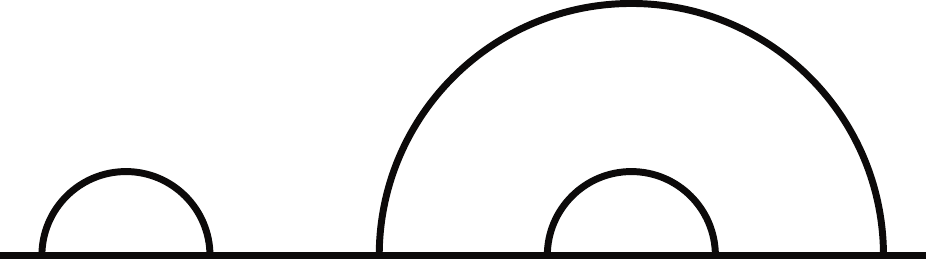}}}
\qquad \qquad \Longrightarrow \qquad \qquad
\Gamma_\alpha \quad = \quad \vcenter{\hbox{\includegraphics[scale=0.275]{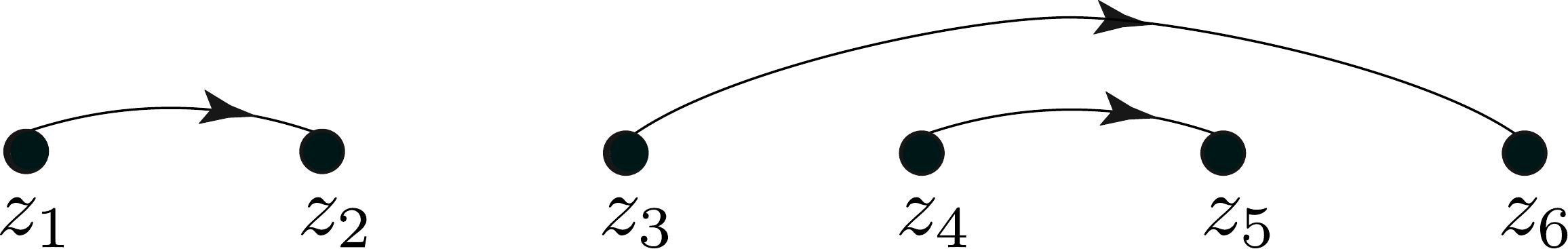} .}}
\end{align} 
We define the function $\sH_n[\alpha]$ by linearly extending~\eqref{CGform} from all $(n,0)$-link patterns 
to all link states $\alpha \in \smash{\LS_n\super{0}}$.
All such functions satisfy PDE system~\eqref{SLEBSAsys1} and the properties stated in items~\ref{mon4} and~\ref{mon3}
with $\bar{\boldsymbol{z}}$ dropped~\cite{dub, kp2}.

Before continuing, we clarify some technical details that arise with the definition of $\sH_n[\alpha]$:

\begin{itemize}[wide, labelwidth=!, labelindent=0pt] 
\item 
So far, our choice of the integration surface $\Gamma_\alpha$ makes sense only when $\re(z_1) < \re(z_2) < \cdots < \re(z_n)$.
Nevertheless, we can use functions of the form~\eqref{CGform} to construct a correlation function of type~\eqref{corrfunc}
that is single-valued when analytically continued from this starting region into its full domain along any path.

\item We have not explicitly specified a branch choice
for the factors in the integrand of~\eqref{CGform}. However, 
such a choice affects the function $\sH_n[\alpha]$ by a single multiplicative factor, which will end up being irrelevant in our application.

\item If $\kappa \in (0,4)$, then the improper integrals in the formula~\eqref{CGform} for $\sH_n[\alpha]$ diverge. 
However, we can renormalize these divergent quantities by replacing their integration contours with Pochhammer contours, 
without affecting our results; see~\cite[appendix~\red{A}]{fp0} and~\cite[section~\red{II}]{sfk3}.
\end{itemize}

Choosing a basis $\mathsf{B} \subset \smash{\LS_n\super{0}}$ for the $\TL_n(\nu)$-standard module,
for each element $\alpha \in \mathsf{B}$ we let $\alpha^\cheque$ denote the dual of $\alpha$ 
with respect to the bilinear form on $\smash{\LS_n\super{0}}$,  
determined by the rule
\begin{align} \label{LSndual}
\BiForm{\alpha^\cheque}{\beta} = \delta_{\alpha,\beta} , \quad \text{for all $\alpha, \beta \in \mathsf{B}$.} 
\end{align} 
By our assumption that $\kappa$ is irrational and item~\ref{result7}, the radical of $\smash{\LS_n\super{0}}$ is trivial, so
the dual basis is well-defined.

\begin{claim} \label{corrformulaClaim} 
The following sum spans the one-dimensional space of functions that satisfy the above properties~\ref{mon1}--\ref{mon3}:
\begin{align}\label{corrformula} 
\sum_{\alpha \, \in \,\mathsf{B}} \sH_n[\alpha](\boldsymbol{z}) \; \sH_n[\alpha^\cheque](\bar{\boldsymbol{z}}) . 
\end{align} 
\end{claim}

In~\cite{fp1}, we prove claim~\ref{corrformulaClaim} using the following ideas.
First, we use the key fact, already known in the physics literature~\cite{ffk, mr, gs, fw},
that the response of the function $\sH_n[\alpha](\boldsymbol{z}) \sH_n[\alpha^\cheque](\bar{\boldsymbol{z}})$ 
to 
analytically continuing its $i$:th and $(i+1)$:st holomorphic and antiholomorphic coordinates simultaneously 
along the half-twists

\vspace*{-4mm}
\begin{align}\label{half-twist} 
\vcenter{\hbox{\includegraphics[scale=0.275]{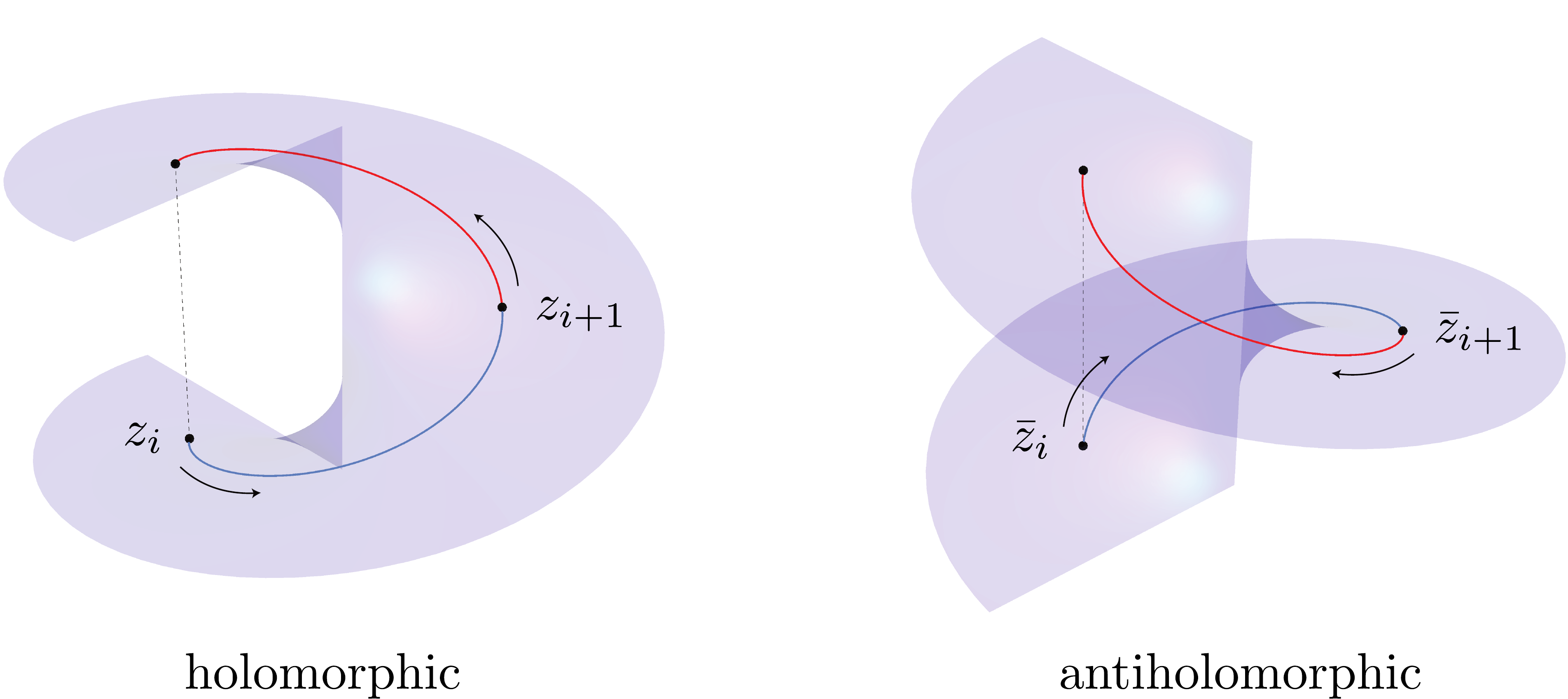}}} 
\end{align} 
matches the response of the corresponding tensor product link state $\alpha \otimes \alpha^\cheque$ to the following $\TL_n(\nu)$-actions: 
an action on $\alpha$, corresponding to the holomorphic 
coordinates of $\boldsymbol{z}$, by the ``braid generator" tangle~\cite{vj3}
\begin{align}\label{braid-gen} 
\vcenter{\hbox{\includegraphics[scale=0.275]{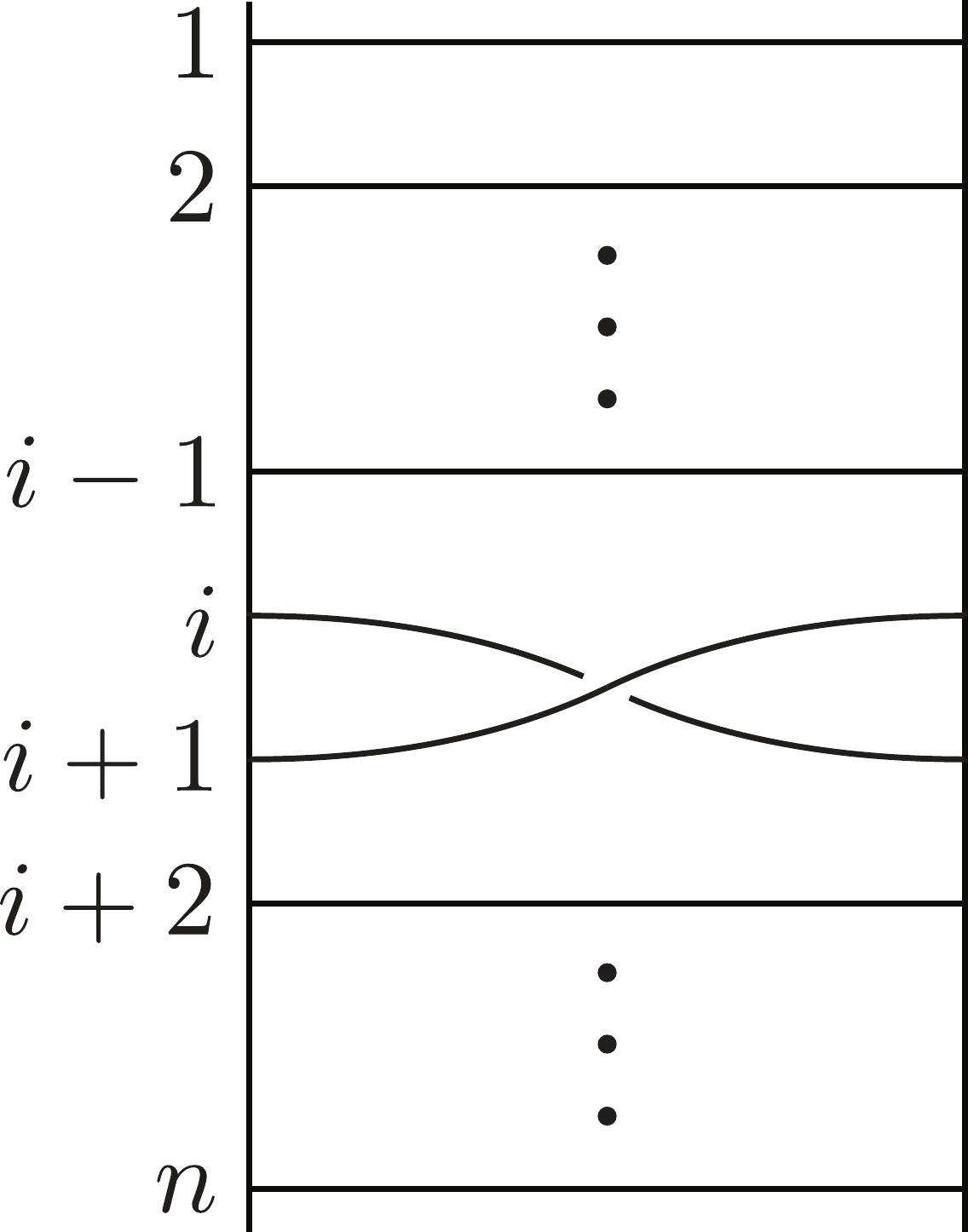}}} 
\quad := \quad q^{1/2} \times \; \vcenter{\hbox{\includegraphics[scale=0.275]{e-TLalgebra5.pdf}}} 
\quad + \quad q^{-1/2} \times \;
\vcenter{\hbox{\includegraphics[scale=0.275]{e-TLalgebra6.pdf} ,}}
\end{align} 
and a simultaneous action on $\alpha^\cheque$, corresponding to the antiholomorphic 
coordinates of $\boldsymbol{\bar{z}}$, by the ``inverse braid" 
\begin{align}\label{braid-gen2} 
\vcenter{\hbox{\includegraphics[scale=0.275]{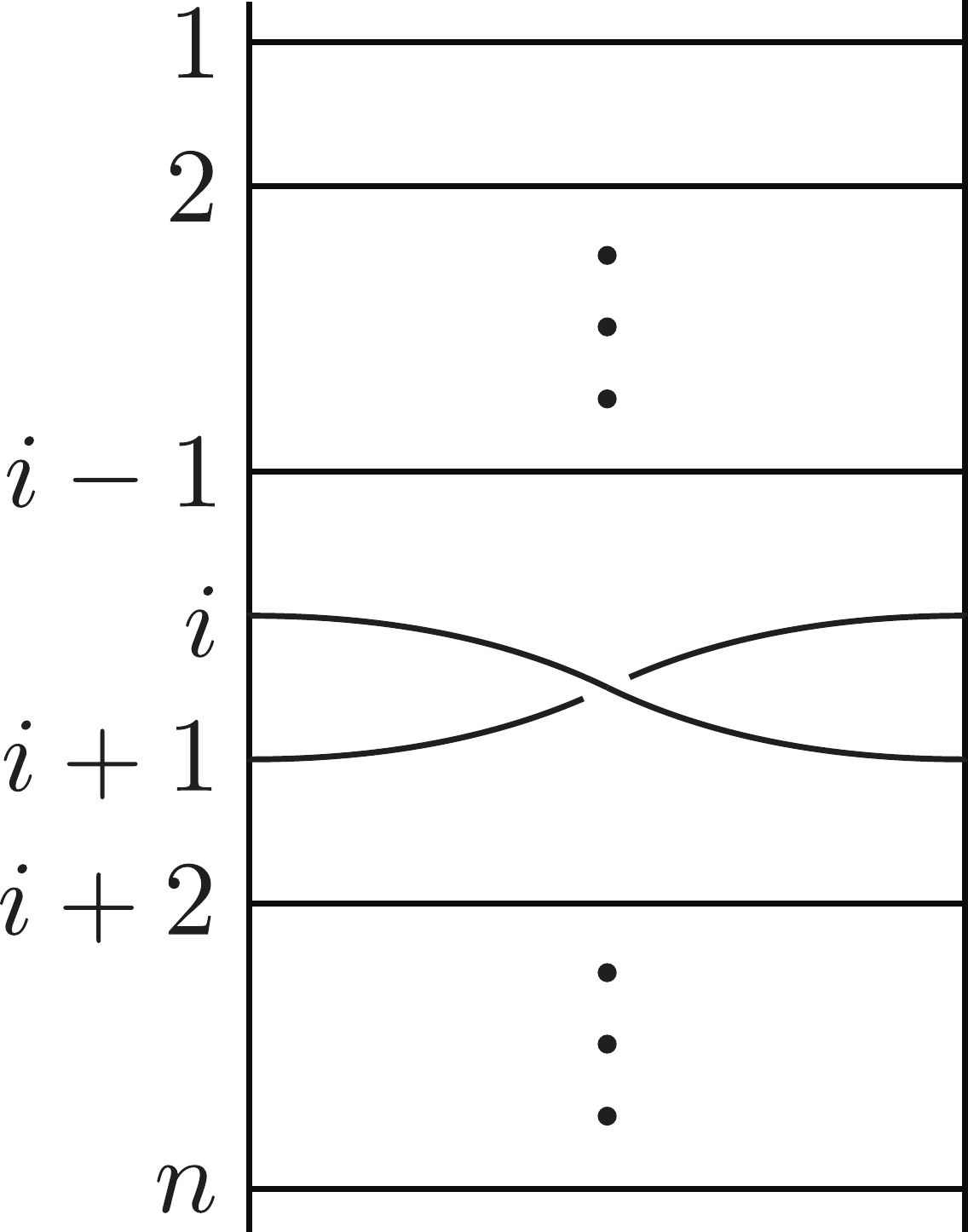}}} 
\quad := \quad q^{-1/2} \times \; \vcenter{\hbox{\includegraphics[scale=0.275]{e-TLalgebra5.pdf}}} 
\quad + \quad q^{1/2} \times \;
\vcenter{\hbox{\includegraphics[scale=0.275]{e-TLalgebra6.pdf} .}} 
\end{align}

The next crucial observation is that tangles~\eqref{braid-gen} or~\eqref{braid-gen2} generate all of $\TL_n(\nu)$ (by item~\ref{result4}).
Thus, invariance of function~\eqref{corrformula} under monodromy transformations and coordinate 
permutations is morally equivalent to invariance of its 
link state counterpart
\begin{align} \label{corralphas}
\sum_{\alpha \, \in \, \mathsf{B}} \alpha \otimes \alpha^\cheque 
\end{align} 
under the corresponding $\TL_n(\nu)$-action on it. 
We explain this in more detail in forthcoming work~\cite{fp1}.

After realizing all monodromy transformations and coordinate permutations of the function 
$\sH_n[\alpha](\boldsymbol{z}) \sH_n[\alpha](\bar{\boldsymbol{z}})$ 
as $\TL_n(\nu)$-actions on the corresponding tensor product link state $\alpha \otimes \alpha^\cheque$, 
we prove in~\cite{fp1} that there exists a unique one-dimensional subspace of link states in 
$\smash{\LS_n\super{0}} \otimes \smash{\LS_n\super{0}}$ that is invariant under the $\TL_n(\nu)$-action, 
spanned by~\eqref{corralphas}. 
This fact provides a key ingredient for concluding that the space of functions satisfying properties~\ref{mon1}--\ref{mon3} above 
is also one-dimensional, spanned by~\eqref{corrformula}. 
Its proof is an almost routine application of Schur's lemma, and the main detail that allows us to use this lemma is the fact that the standard 
module $\smash{\LS_n\super{0}}$ is simple, a well-known fact by item~\ref{result9}.

Next, using ideas from the above discussion, we explain how to construct monodromy 
and coordinate-permutation invariant multi-point CFT correlation functions comprising a more general class of Kac operators. 
We denote by $\psi_s$ the spinless primary conformal field whose (holomorphic and antiholomorphic) conformal weight equals 
\begin{align}
\sleg{s} = \frac{s(2s+4-\kappa)}{2\kappa} ,
\end{align} 
namely the Kac weight $h_{1,s+1}$ or $h_{s+1,1}$ indexed by the $(s+1)$:th entry in the first row or column of the Kac table. Then, 
we consider the following $\np$-point CFT correlation function with respect to 
the multiindex $\multii = (\sIndex_1, \sIndex_2, \ldots, \sIndex_\np)$:
\begin{align}\label{corrfunc2} 
F_\multii ( \boldsymbol{z}, \bar{\boldsymbol{z}}) 
= \langle \psi_{\sIndex_1}(z_1, \bar{z}_1) \psi_{\sIndex_2}(z_2, \bar{z}_2) \dotsm \psi_{\sIndex_\np}(z_\np, \bar{z}_\np) \rangle . 
\end{align} 
This function should satisfy properties similar to~\ref{mon1}--\ref{mon3} above, except that in item~\ref{mon1}, we replace 
PDE system~(\ref{SLEBSAsys1},~\ref{SLEBSAsys2}) by a more complicated collection of BPZ partial differential equations~\cite{bpz, bsa},
and in item~\ref{mon4}, we replace the conformal weight $\sleg{1}$ with the more general conformal weights $\sleg{\sIndex_i}$.

In~\cite{fp3, fp2, fp1}, we study question~\ref{BigQuestion} in this more general setting. Again, we find that 
the space of functions having the desired properties is one-dimensional. 
As before, a key ingredient to the proof of this is 
the fact that the $\TL_\multii(\nu)$-module $\smash{\LS_\multii\super{0}}$ is simple, by item~\ref{result99}.
To construct a nonzero function in this space, we define
\begin{align}\label{CGform2} 
\sH_\multii[\alpha](\boldsymbol{z}) := \bigg(\prod_{i < j}^\np (z_i - z_j)^{2\sIndex_i \sIndex_j/\kappa} \bigg) \int_{\Gamma_\alpha} \bigg( \prod_{i = 1}^{\vphantom{\Summed/2}\np} \prod_{j = 1}^{\Summed/2} (w_i - z_j)^{-4\sIndex_i/\kappa} \bigg) 
\bigg( \prod_{1 \leq i < j \leq \Summed/2} (w_i - w_j)^{8/\kappa} \bigg) \, \ud \boldsymbol{w} ,
\end{align} 
for each nonzero $(\multii,0)$-valenced link pattern $\alpha$, where $\Gamma_\alpha$ is a certain $\alpha$-dependent integration contour.
Then choosing an arbitrary basis $\mathsf{B}$ for $\smash{\LS_\multii\super{0}}$, with dual basis as in~\eqref{LSndual},
we prove that the function 
\begin{align}\label{corrformula2} 
\sum_{\alpha \, \in \,\mathsf{B}} \sH_\multii[\alpha](\boldsymbol{z}) \; \sH_\multii[\alpha^\cheque](\bar{\boldsymbol{z}}) 
\end{align}
spans the sought one-dimensional solution space.
To prove this claim, we realize the monodromy of function~\eqref{corrformula2} as an appropriate 
$\TL_\multii(\nu)$-action on the tensor product valenced link state $\alpha \otimes \alpha^\cheque$,
with explicit generating set for the valenced Temperley-Lieb algebra $\TL_\multii(\nu)$ given in item~\ref{result33} above. 
For example,
the ``braid" tangles
\begin{align}
\vcenter{\hbox{\includegraphics[scale=0.275]{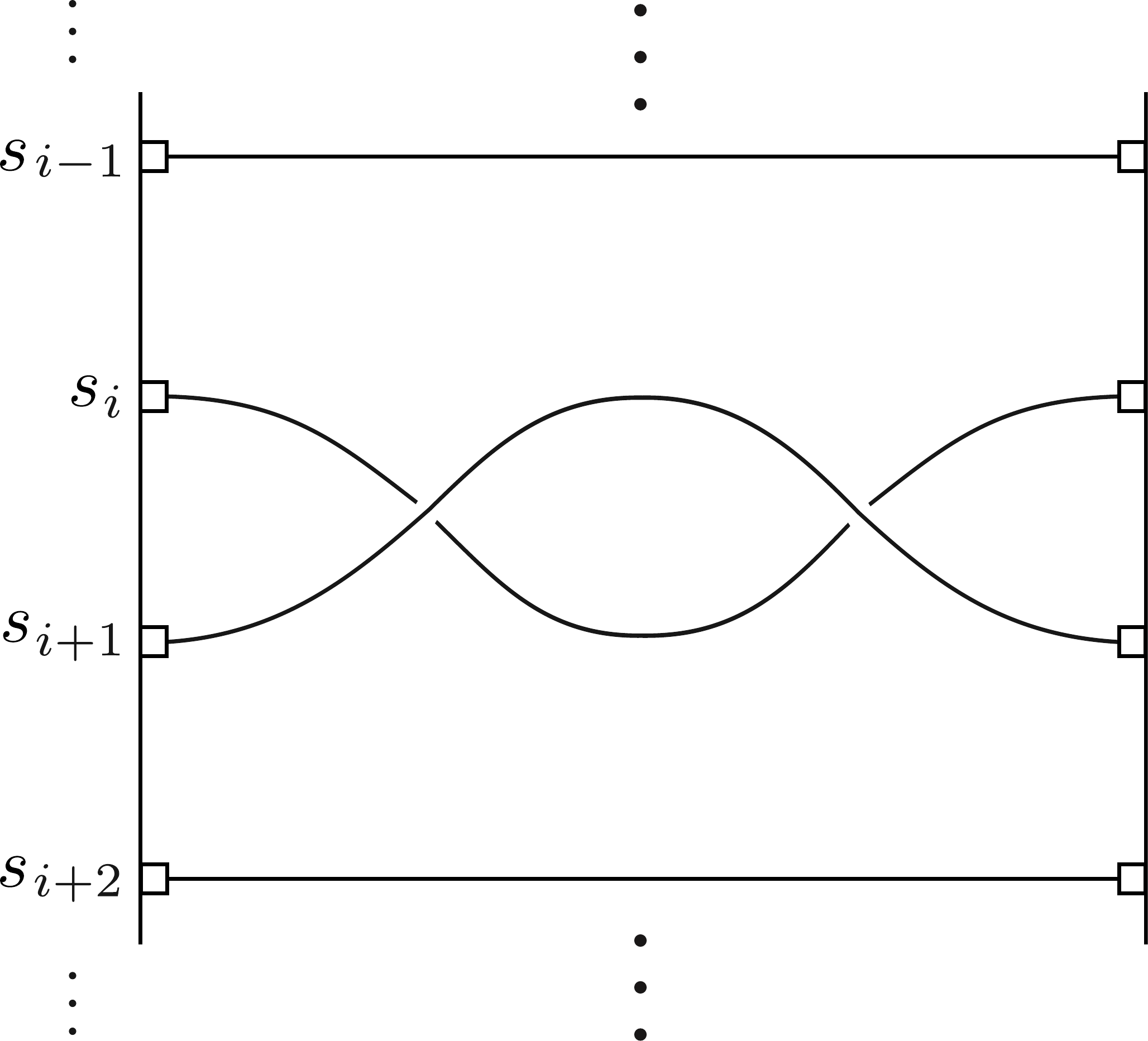}}} 
\quad = \quad \sum_{k=0}^{\min(\sIndex_i, \sIndex_{i+1})} c_k \times \;
\vcenter{\hbox{\includegraphics[scale=0.275]{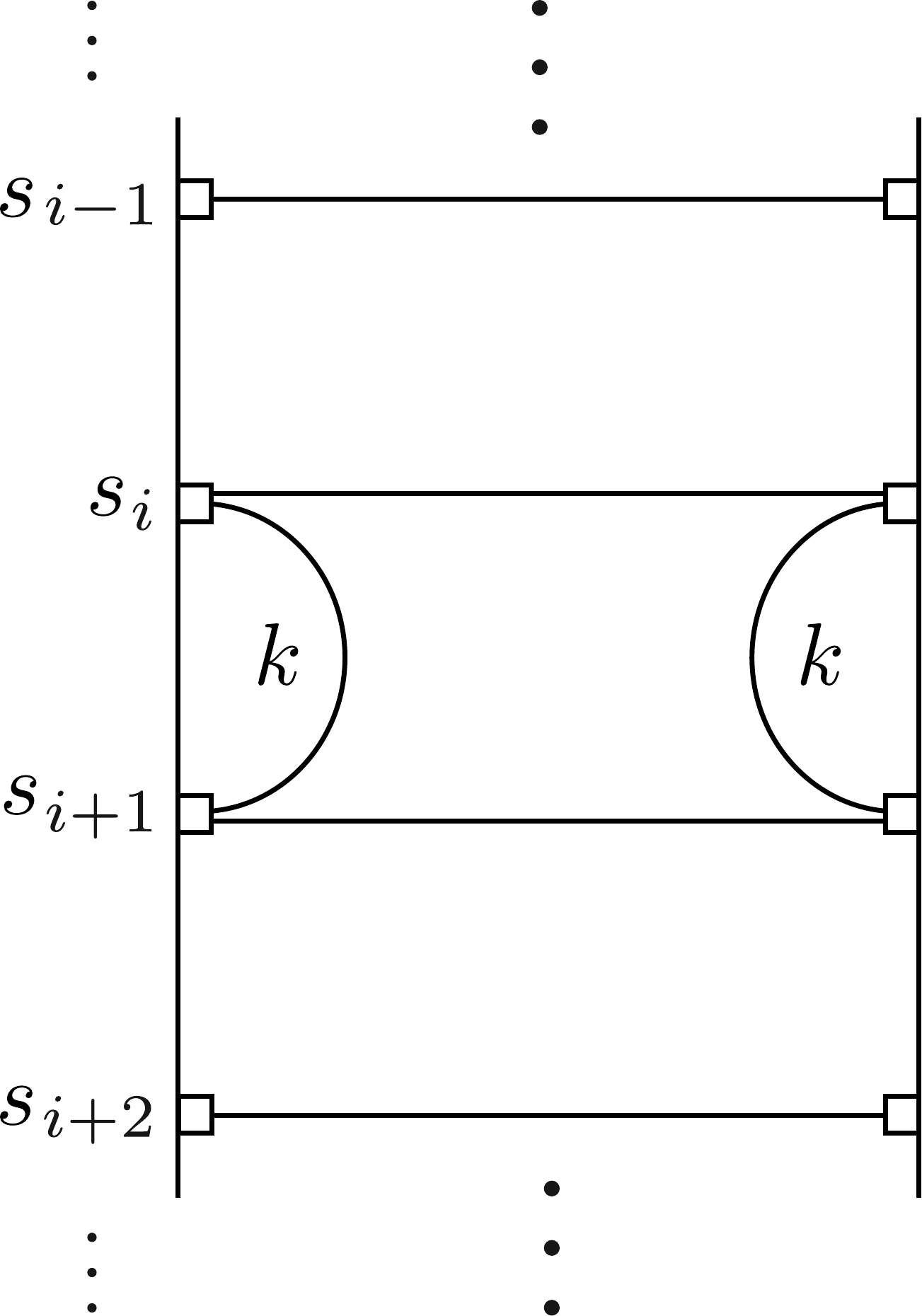} ,}} 
\end{align} 
with $i \in \{1,2,\ldots,\np-1\}$ and certain coefficients  $c_k \in \bC$, give rise to the monodromy when 
the holomorphic coordinates $\boldsymbol{z}_i$ and $\boldsymbol{z}_{i+1}$ wind around each other one full turn.

In summary, 
the determination of a unique monodromy and coordinate-permutation invariant 
correlation function~\eqref{corrfunc2} motivates the work and key results that we present in this article.
Furthermore, in~\cite{fp3}, we discuss a ``quantum Schur-Weyl duality''~\cite{mj2, ppm, mma}
between the valenced Temperley-Lieb algebra $\TL_\multii(\nu)$ and 
the Hopf algebra $U_q(\mathfrak{sl}_2)$. This duality endows certain tensor product representations 
of the latter algebra with a bimodule structure, where the two algebras $U_q(\mathfrak{sl}_2)$ and $\TL_\multii(\nu)$ have commuting actions.
From the above discussion, we know that the action of $\TL_\multii(\nu)$ is closely related to the monodromy of correlation function~\eqref{corrfunc2}.
In~\cite{fp2}, we discuss solution spaces of the PDE systems of BPZ type.
Finally, in~\cite{fp1} we use the results of~\cite{fp0, fp3, fp2} together with 
the results of the present article to prove claim of type~\ref{corrformulaClaim} for any multiindex $\multii$.

\subsection{Organization of this article} \label{OrganizationSec}

In section~\ref{DiagramAlgebraSect}, we define the 
valenced Temperley-Lieb algebra $\TL_\multii(\nu)$,
the principal concern of this article. We also introduce notation and present some key results. 
In detail, in section~\ref{WJLinkDiagSect} we define valenced link diagrams, tangles, link patterns, and link states, 
and the diagram spaces $\smash{\TL_\multii^\multiii}$ and $\smash{\LS_\multii}$. 
In the next section~\ref{CombSect}, we collect simple but useful observations of combinatorial nature.
In section~\ref{ValencedCompositionSec}, 
we define compositions of valenced diagrams. 
Section~\ref{ValTLdefSec} concerns the valenced Temperley-Lieb algebra $\TL_\multii(\nu)$.
In proposition~\ref{GeneratorPropTwo}, whose proof follows from~\cite[theorem~\red{1.1}]{fp0}, 
we present two minimal generating sets for this algebra.

In sections~\ref{StdModulesSect} and~\ref{FinalResultSect}, we focus on the representation theory of 
the valenced Temperley-Lieb algebra $\TL_\multii(\nu)$. 
Much of this depends on the symmetric, invariant bilinear form of the link state module $\smash{\LS_\multii}$. 
Detailed understanding of this bilinear form, and in particular its radical, is a major undertaking of this article. 
We define the bilinear form in section~\ref{BilinFormSec} via network evaluations. 
In section~\ref{LinkStateModSect}, we present fundamental properties of the standard 
modules $\smash{\LS_\multii\super{s}}$: e.g., proposition~\ref{GenLem2} says that 
the radical of the bilinear form on $\smash{\LS_\multii\super{s}}$ is its maximal submodule and 
the corresponding quotient module is simple (if not trivial).
We show later in section~\ref{TLSimpleModSect} that these are in fact all of the non-isomorphic simple $\TL_\multii (\nu)$-modules.
In section~\ref{FaithfulSect3}, we prove that the representation of $\TL_\multii (\nu)$ on the link state module $\LS_\multii$ 
is faithful if and only if the radical of the latter is trivial. 
We relate this representation to the Jacobson radical of $\TL_\multii(\nu)$ in section~\ref{SemiSect}.


In section~\ref{GramMatrixSect}, we study the Gram determinant 
of the bilinear form on the standard modules $\smash{\LS_\multii\super{s}}$.
Fundamental properties of the radical of $\smash{\LS_\multii\super{s}}$, 
such as its dimension, and fundamental properties of the Gram matrix, such as its nullity, 
are interdependent; understanding the latter gives useful information 
about the former. In section~\ref{ConformalBlocksSect}, we define ``trivalent link states" 
(which correspond to conformal blocks in CFT, via the ``spin-chain Coulomb gas map"~\cite{fp2}), 
and in section~\ref{ConformalBlocksProp} we disseminate their key properties. 
One of these properties, stated in proposition~\ref{IndOrthBasisLem}, is that 
the trivalent link states form an orthogonal basis if, for example, $q$ in~\eqref{fugacity} is not a root of unity. 
Then, in proposition~\ref{GramDetLem} in section~\ref{DetSect} we make use of the orthogonality property to 
give an explicit formula for the Gram determinant.
Finally, in section~\ref{RecursSect} we present recursions and useful formulas for the Gram determinants.

From the explicit formulas for the Gram determinant, 
it follows that the radical of the link state module $\LS_\multii$ is trivial if, 
for example, $q$ in~\eqref{fugacity} is not a root of unity. 
Thus, in section~\ref{RadicalSect}, we assume that $q$ is a root of unity, and
we determine a basis for and the dimension of the radical of each standard module 
$\smash{\LS_\multii\super{s}}$ (proposition~\ref{BigTailLem} and corollary~\ref{DnsLemAndRadDimCor},
and proposition~\ref{BigTailLem2} and corollary~\ref{DnsLemAndRadDimCor2}). 
Using these results, we determine in section~\ref{rofSect31} all values of $q$ for which the radical 
of $\smash{\LS_\multii\super{s}}$ is trivial (corollary~\ref{GridCor}). 
Finally, in section~\ref{rofSect32} we study cases where the radical of $\smash{\LS_\multii\super{s}}$ equals the entire space, a phenomenon outlawed as a precondition 
to most results presented in section~\ref{StdModulesSect}.

In the final section~\ref{FinalResultSect}, we present general 
results on the representation theory of the valenced Temperley-Lieb algebra $\TL_\multii(\nu)$.
The first section~\ref{RecapSec} is devoted to a remainder of basic concepts from the representation theory of algebras.
We determine the complete set of non-isomorphic simple $\TL_\multii(\nu)$-modules 
in proposition~\ref{SimpleModuleProp} in section~\ref{TLSimpleModSect}.
In theorem~\ref{BigSSTHM} in section~\ref{SemiSect},
we present several equivalent conditions for the semisimplicity of $\TL_\multii(\nu)$.
One of them is that the standard modules $\smash{\LS_\multii\super{s}}$ 
constitute the complete set of non-isomorphic simple $\TL_\multii(\nu)$-modules. 
In proposition~\ref{Jacobsonprop}, we identify the Jacobson radical of $\TL_\multii(\nu)$ as the kernel of 
the representation of $\TL_\multii(\nu)$ on the quotient of its link state module $\LS_\multii$
modulo its radical $\rad \LS_\multii$.

In the appendices, we give background and proofs for some technical results needed in this article. 
In appendix~\ref{TLRecouplingSect}, we present results from Temperley-Lieb recoupling theory~\cite{kl}. 
Then in appendix~\ref{AppWJ}, we discuss the relation of the valenced Temperley-Lieb algebra $\TL_\multii(\nu)$
to a subalgebra of the ordinary Temperley-Lieb algebra $\TL_n(\nu)$, the Jones-Wenzl algebra $\WJ_\multii(\nu)$.
We study this algebra in terms of generators and relations in our companion article~\cite{fp0}.
In appendix~\ref{RadicalAppendix}, we address a technical detail for section~\ref{GramMatrixSect}, 
concerning the definition of the trivalent link states at roots of unity.
In the last appendix~\ref{CategorySect}, we briefly discuss a categorical framework for diagram algebras.

\begin{center}
\bf Relation to previous work
\end{center}

In~\cite{gl, gl2}, J.~Graham and G.~Lehrer develop and use a general theory of cellular algebras,
a powerful category-theoretic approach for obtaining strong results about the representation theory 
of certain types of algebras, including the Temperley-Lieb algebra $\TL_n(\nu)$. 
In fact, the valenced Temperley-Lieb algebra $\TL_\multii(\nu)$ is also cellular, and we discuss this point of view in~\cite{fp0}.
However, the abstract theory of~\cite{gl} alone does not give full explicit information of the structure of the representations,
such as explicit bases for the radicals of the standard modules, 
nor does it seem to help in constructing the principal indecomposable modules. 

In the article~\cite{rsa} of D.~Ridout and Y.~Saint-Aubin, similar results are obtained with more concrete but related techniques,
following and motivated by the works of V.~Jones~\cite{vj}, L.~Kauffman~\cite{lk2}, G.~James and G.~Murphy~\cite{jm79}, 
P.~Martin~\cite{pm}, F.~Goodman and H.~Wenzl~\cite{gwe}, and B.~Westbury~\cite{bw}.
In particular, the authors identify explicitly all simple and principal indecomposable modules of the Temperley-Lieb algebra.
Certain specific central elements in $\TL_n(\nu)$ play a special role in the analysis of the non-semisimple case.
Unfortunately, finding such central elements for the valenced Temperley-Lieb algebras seems in general to be a difficult task.

In our work, we follow a concrete approach inspired by~\cite{rsa} and predecessors.
Combining ideas in this spirit  with explicit graphical calculus (termed ``Temperley-Lieb recoupling theory'') 
\`a la Kauffman and Lins~\cite{kl}, 
we provide elementary tools to understand the representation theory of the valenced Temperley-Lieb algebra $\TL_\multii(\nu)$. 
Furthermore, we establish explicit structural results about the representations and their radicals, such as specific bases for them.
It does not seem easy to obtain such results without our graphical approach, which makes, e.g., diagonalization of Gram matrices and 
explicit constructions of basis elements in the radicals tractable. 
Furthermore, in~\cite{fp0}, we find generators and relations for $\TL_\multii(\nu)$, and in~\cite{fp3},
we relate the graphical calculus of $\TL_\multii(\nu)$ to that for the Hopf algebra $U_q(\mathfrak{sl}_2)$~\cite{fk},
hence explicating a quantum Schur-Weyl duality (already well-known for the ordinary Temperley-Lieb algebra 
and the fundamental representations of $U_q(\mathfrak{sl}_2)$~\cite{mj2, ppm, mma}).

In recent work~\cite{ilz}, using cellular methods, K.~Iohara, G.~Lehrer, and R.~Zhang consider 
a semisimple quotient of the Temperley-Lieb algebra $\TL_n(\nu)$
when $q$ in~\eqref{MinPower} is a root of unity, 
and relate it to a fusion category of certain quantum $\mathfrak{sl}_2$-tilting modules.
They also prove a form of a quantum Schur-Weyl duality. 

We also mention another abstract but powerful approach for analyzing representation theory, known as 
``categorification.'' (See e.g.~\cite{bfk, fks, bck, fss, ss} and references therein.)
The rough idea is to associate to an algebra a category with certain properties and to use this abstract framework to study the representation theory
of this algebra. From this approach, very general results follow in an elegant way that avoids explicit calculations.

\begin{center}
\bf Acknowledgements
\end{center}

We are very grateful to K.~Kyt\"ol\"a for numerous discussions and encouragement,
and cordially thank J.~Bellet\^ete, A.~Langlois-R\'emillard, D.~Ridout, and Y.~Saint-Aubin 
for careful comments on the first version of the manuscript
that helped to improve the presentation.
We also thank P.~Di~Francesco, C.~Hongler, J.~Jacobsen, V.~Jones, R.~Kashaev, R.~Kedem, V.~Pasquier, D.~Radnell, 
and H.~Saleur for helpful discussions.
E.P. is supported by the ERC AG COMPASP, the NCCR SwissMAP, and the Swiss NSF,
and she also acknowledges the earlier support from the Vilho, Yrj\"o and Kalle V\"ais\"al\"a Foundation.
During this work, S.F. was supported by the
Academy of Finland grant number 288318, 
``Algebraic structures and random geometry of stochastic lattice models.''
Finally, E.P. wishes to acknowledge the kind hospitality of the Institute Mittag-Leffler 
and the Mathematisches Forschungsinstitut Oberwolfach 
during the preparation of the first version of this article.

\section{Diagram algebras} \label{DiagramAlgebraSect}

In this section, we present fundamental definitions and results concerning the valenced Temperley-Lieb
diagram algebra. In section~\ref{WJLinkDiagSect}, 
we define valenced link diagrams, tangles, link patterns, and link states.
In section~\ref{CombSect}, we collect key results about them, of a combinatorial nature, for use throughout this article. 
In sections~\ref{ValencedCompositionSec} and~\ref{ValTLdefSec}, we define bilinear maps that 
determine concatenation of valenced tangles with valenced tangles and valenced link states.
In particular, these determine a diagrammatic multiplication for the valenced Temperley-Lieb algebra $\TL_\multii(\nu)$,
as well as an action of it on its standard modules.


\subsection{Valenced tangles and link states} \label{WJLinkDiagSect}

First, we formally define valenced link diagrams, tangles, link patterns, and link states.  We denote
\begin{align} \label{PositiveIndexSetDef}
\bZpos^\# := \bZpos \cup \bZpos^2 \cup \bZpos^3 \cup \dotsm , 
\end{align}
and we let $\multii$ and $\multiii$ denote two multiindices with $\np_\multii$ and $\np_\multiii$ nonnegative integer entries respectively, 
\begin{align} \label{MultiindexNotation}
\multii = (\sIndex_1, \sIndex_2,\ldots, \sIndex_{\np_\multii}) \in \{ (0) \} \cup \bZpos^\#, \qquad \qquad \multiii = (p_1, p_2, \ldots, p_{\np_\multiii}) \in \{ (0) \} \cup \bZpos^\# , 
\end{align}
and such that $\Summed_\multii + \Summed_\multiii = 0 \Mod 2$, where $\Summed_\multii$ and $\Summed_\multiii$ denote the respective sums 
\begin{align}\label{ndefn} 
\Summed_\multii := \sIndex_1 + \sIndex_2 + \cdots + \sIndex_{\np_\multii}, \qquad \qquad \Summed_\multiii := p_1 + p_2 + \cdots + p_{\np_\multiii}. 
\end{align}

We define a $(\multii, \multiii)$-\emph{valenced link diagram} to be any collection of the following planar geometric objects:

\begin{enumerate}[leftmargin=*] 
\itemcolor{red}
\item two vertical lines, 
\item $\np_\multii$ distinct marked points, called \emph{nodes}, on the left line and $\np_\multiii$ nodes on the right line, and
\item $\frac{1}{2}(\Summed_\multii + \Summed_\multiii)$ 
planar curves, called \emph{links}, that may intersect themselves or each other only at their endpoints and that are 
arranged in such a way that $\sIndex_i$ (resp.~$p_j$) endpoints reside at the $i$:th left (resp.~$j$:th right) node.
Each link is determined only up to a homotopy that preserves its endpoints.
\end{enumerate}
We call the $i$:th entry $\sIndex_i$ of $\multii$ (resp.~$j$:th entry $p_j$ of $\multiii$) the \emph{valence} of the $i$:th left (resp.~$j$:th right) node.  
As in~\eqref{boxval}, we illustrate a node with valence $s$ as a small box that sits over the node itself, with a cable of size $s$ exiting it:
\begin{align}
\begin{rcases}
\vcenter{\hbox{\includegraphics[scale=0.275]{e-valenced_node2.pdf}}} \quad
\end{rcases} s
\quad 
\qquad \Longrightarrow \qquad 
\begin{rcases}
\vcenter{\hbox{\includegraphics[scale=0.275]{e-valenced_node1.pdf}}} \quad
\end{rcases} \; s
\quad 
\qquad \Longrightarrow \qquad 
\quad \vcenter{\hbox{\includegraphics[scale=0.275]{e-valenced_box.pdf}  .}}
\end{align} 
We sort the links of each valenced link diagram into two types, called \emph{crossing links} 
and \emph{turn-back links},
and we define a \emph{loop link} to be a link with both endpoints at the same node:
\begin{align} 
\vcenter{\hbox{\includegraphics[scale=0.275]{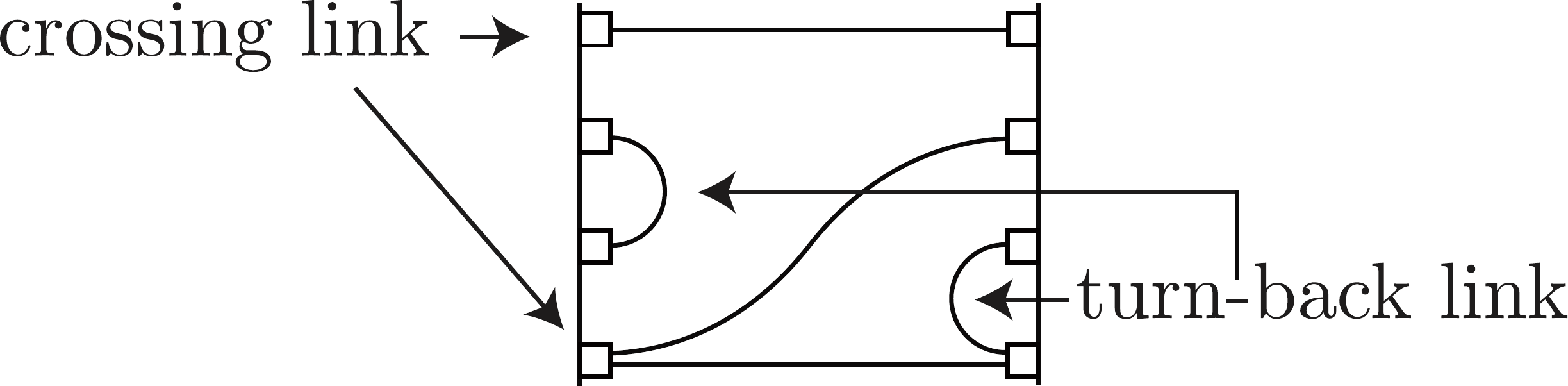}, }}
\qquad\qquad\qquad\qquad
\vcenter{\hbox{\includegraphics[scale=0.275]{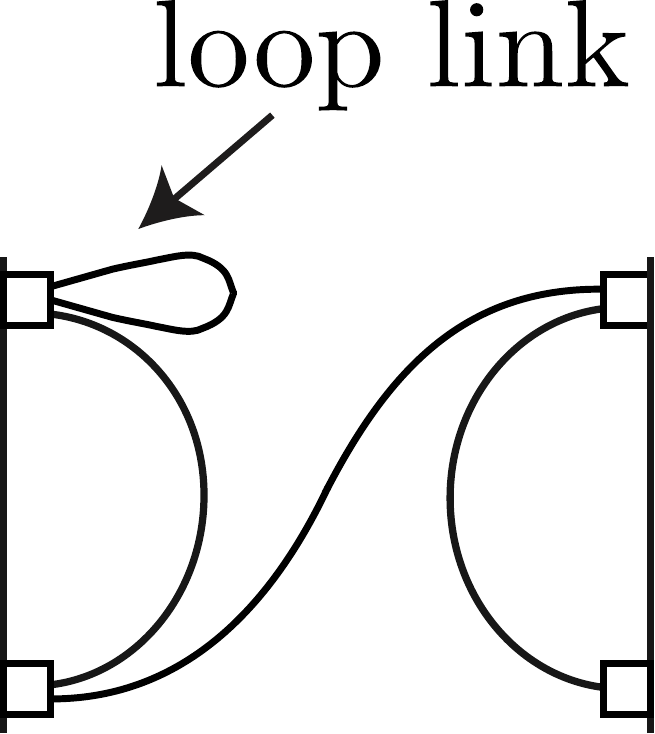} .}}
\end{align}
For any two multiindices $\smash{\multii,\multiii \in \{ (0) \} \cup \bZpos^\#}$, we set 
\begin{align} 
\label{ValLinkDiagramDef0}
\preLD_\multii^\multiii := & \; \text{the collection of all $(\multii,\multiii)$-valenced link diagrams}, \\
\label{ValGenTLDef0}
\preTL_\multii^\multiii := & \; \Span \preLD_\multii^\multiii ,
\end{align}
and we call an element of $\smash{\preTL_\multii^\multiii}$ a $(\multii,\multiii)$\emph{-valenced tangle}.  
Hence, $\smash{\preLD_\multii^\multiii}$ is a basis for $\smash{\preTL_\multii^\multiii}$.

Next, we define a $(\multii, s)$-\emph{valenced link pattern} to be any planar geometric object formed by
\begin{enumerate}[leftmargin=*]
\itemcolor{red}
\item dividing a $(\multii,\multiii)$-valenced link diagram with exactly $s$ crossing links in half, 
\item discarding the right half, and 
\item rotating the left half by $\pi/2$ radians.
\end{enumerate} 
The division breaks the $s$ crossing link in half, resulting in a valenced link pattern with $s$ defects and a number of links. 
For example, we have
\begin{align} 
\vcenter{\hbox{\includegraphics[scale=0.275]{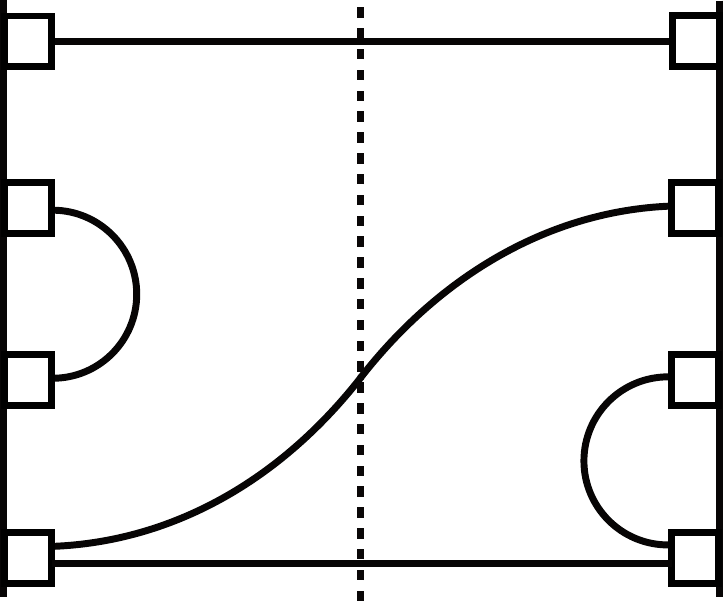}}} \qquad \qquad & \longmapsto \qquad \qquad
\vcenter{\hbox{\includegraphics[scale=0.275]{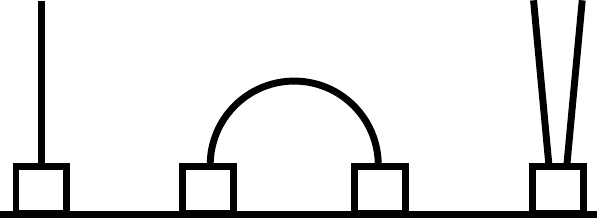}}} \\[1em]
\nonumber
\text{$(\multii, \multiii)$-valenced link diagram} \qquad \qquad & \longmapsto \qquad \qquad \text{$(\multii, s)$-valenced link pattern}.
\end{align}
For any multiindex $\multii \in \{ (0) \} \cup \smash{\bZpos^\#}$ and $s \in \bZnn$ 
such that there exists a $(\multii,s)$-valenced link pattern, 
we set
\begin{align} 
\label{ValLinkPattDef0}
\preLP_\multii\super{s} := & \; \text{the collection of all $(\multii,s)$-valenced link patterns}, \\
\label{ValLinkStateDef0}
\preLS_\multii\super{s} := & \; \Span \preLP_\multii\super{s} ,
\end{align}
and we call an element of $\smash{\preLS_\multii\super{s}}$ a \emph{$(\multii,s)$-valenced link state}.

Rather than working with these vector spaces, we exclude tangles and link states that contain loop links.  
To formalize this, we define the following equivalence relations on the latter spaces $\smash{\preTL_\multii^\multiii}$ and $\smash{\preLS_\multii\super{s}}$:
\begin{align}
&T \sim S \qquad \Longleftrightarrow \qquad T = S + K , \quad 
&&\left\{\parbox{6.8cm}{\textnormal{where $K \in \smash{\preTL_\multii^\multiii}$ is a linear combination} \\ 
\textnormal{of link diagrams, each having a loop link}}\right\}, \\
&\alpha \sim \beta \qquad \Longleftrightarrow \qquad \alpha = \beta + \gamma , \quad 
&&\left\{\parbox{6.5cm}{\textnormal{where $\gamma \in \smash{\preLS_\multii\super{s}}$ is a linear combination} \\ 
\textnormal{of link patterns, each having a loop link}}\right\} ,
\end{align}
and we set
\begin{align}
\LD_\multii^\multiii := & \; \preLD_\multii^\multiii / \sim ,
 \qquad \qquad \TL_\multii^\multiii := \preTL_\multii^\multiii / \sim \; = \; \Span \LD_\multii^\multiii  , \\
\LP_\multii\super{s} :=  & \; \preLP_\multii\super{s} / \sim ,
 \qquad \qquad \;\, \LS_\multii\super{s} := \preLS_\multii\super{s} / \sim \;\; = \; \Span \LP_\multii\super{s} .
\end{align}
We identify the $(\multii, \multiii)$-valenced tangles 
and the $(\multii, s)$-valenced link states with their respective equivalence classes: 
\begin{align} 
\label{EquivClassTangle}
\vcenter{\hbox{\includegraphics[scale=0.275]{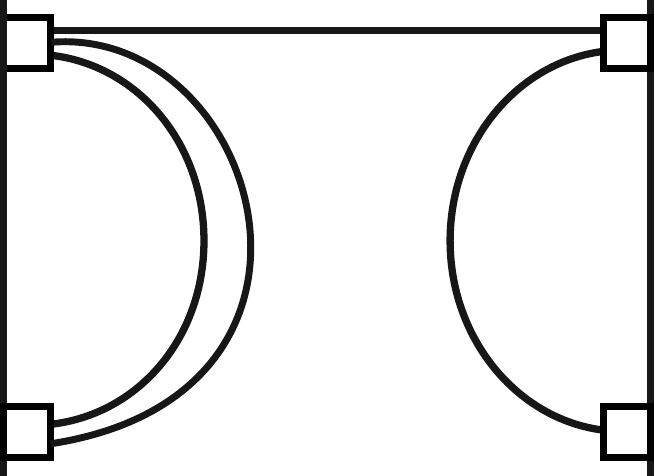}}}
\qquad \qquad & \longleftrightarrow \qquad \qquad
\left[ \quad
\vcenter{\hbox{\includegraphics[scale=0.275]{e-Tangle_equivalence_class.pdf}}} \quad + \quad
\vcenter{\hbox{\includegraphics[scale=0.275]{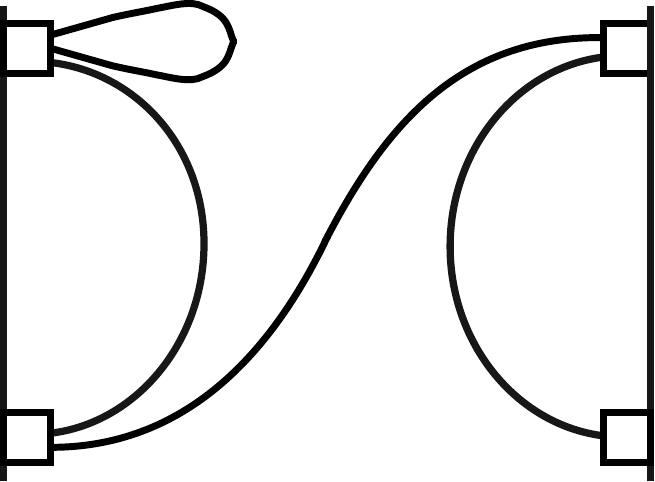}}} 
\quad \right]_\sim \\
\label{EquivClassLinkPattern}
\vcenter{\hbox{\includegraphics[scale=0.275]{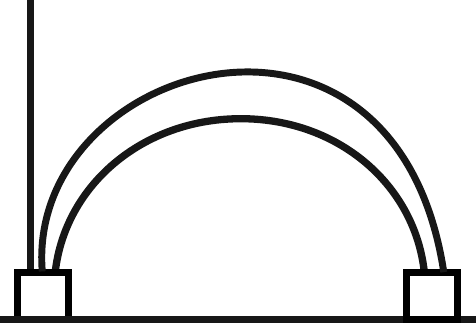}}} 
\qquad \qquad  &\longleftrightarrow \qquad \qquad
\left[ \quad
\vcenter{\hbox{\includegraphics[scale=0.275]{e-LinkState_equivalence_class.pdf}}} \quad + \quad
\vcenter{\hbox{\includegraphics[scale=0.275]{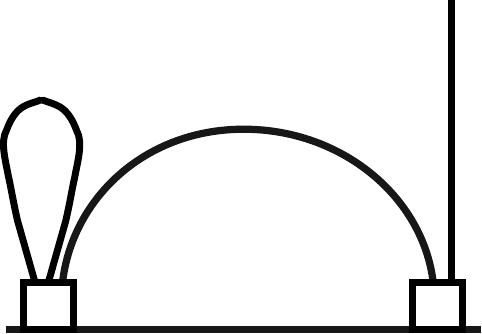}}} 
\quad \right]_\sim 
\end{align}
In particular, we identify all of the $(\multii, \multiii)$-valenced tangles and $(\multii, s)$-valenced link states which contain loop links with the zero tangle
and the zero link state, respectively:
\begin{align} \label{EquivClassTangleIsZero}
\left[ \quad \vcenter{\hbox{\includegraphics[scale=0.275]{e-Tangle_equivalence_class_looplink.pdf}}} 
\quad \right]_\sim 
\quad  \longleftrightarrow \quad 0 
\qquad \qquad \text{and} \qquad \qquad
\left[ \quad \vcenter{\hbox{\includegraphics[scale=0.275]{e-LinkState_equivalence_class_looplink.pdf}}} 
\quad \right]_\sim 
\quad \longleftrightarrow \quad 0 .
\end{align}

It is sometimes useful to distinguish tangles with given number of crossing links:
\begin{align} 
\label{LDs} \LD_\multii^{\multiii; \scaleobj{0.85}{(s)}} &:= \{ \text{all valenced link diagrams in 
$\smash{\LD_\multii^\multiii}$ with exactly $s$ crossing links} \}, \\
\label{TLs} \TL_\multii^{\multiii; \scaleobj{0.85}{(s)}} &:= \Span \LD_\multii^{\multiii; \scaleobj{0.85}{(s)}} . 
\end{align}
With these definitions, we have the $s$-grading
\begin{align}\label{WJDirSum} 
\TL_\multii^{\multiii} = \bigoplus_{s \, \in \, \DefectSet_\multii^{\multiii}} 
\smash{\TL_\multii^{\multiii; \scaleobj{0.85}{(s)}}} ,
\end{align}
where $\smash{\DefectSet_\multii^{\multiii}}$ denotes the set of all integers $s \geq 0$ such that 
$\smash{\TL_\multii^{\multiii; \scaleobj{0.85}{(s)}}}$ is not empty.
Also, with $\DefectSet_\multii$ denoting the set of all integers $s \geq 0$ such that 
the space $\smash{\LS_\multii\super{s}}$ is nontrivial, we define
\begin{align}
\label{LSDirSum2} 
\LP_\multii := \bigcup_{s \, \in \, \DefectSet_\multii} \LP_\multii\super{s} 
\qquad \text{and} \qquad 
\LS_\multii :=  \bigoplus_{s \, \in \, \DefectSet_\multii} \LS_\multii\super{s} .
\end{align}

For the respective cases that all nodes 
have valence one or there are no nodes at all, we denote 
\begin{align} 
\OneVec{n} := (\underbrace{1,1,\ldots,1}_{\text{$n$ times}}) , \quad \text{for all $n \in \bZpos$}
\qquad \qquad \text{and} \qquad \qquad \OneVec{0} := (0) , \quad \text{for $n = 0$}.
\end{align}
In both cases, we omit the arrow over the multiindex $\OneVec{n}$ whenever it appears as a superscript or subscript. 
We also notice that if $\multii =  \OneVec{n}$ and $\multiii = \OneVec{m}$ for some $n, m \in \bZnn$, then the links join the nodes pairwise, so 
\begin{align} 
\LD_n^m = \preLD_n^m, \qquad  \TL_n^m = \preTL_n^m, \qquad \LP_n\super{s} = \preLP_n\super{s}, 
\qquad \text{and} \qquad \LS_n\super{s} = \preLS_n\super{s} .
\end{align}
We call elements of $\smash{\LD_n^m}$ and $\smash{\TL_n^m}$  \emph{$(n, m)$-link diagrams} and \emph{$(n,m)$-tangles}, respectively. 
Elements of $\smash{\LP_n\super{s}}$ and $\smash{\LS_n\super{s}}$ are 
called ``$(n,s)$-link patterns" and ``$(n,s)$-link states," 
and that $(n,n)$-link diagrams and $(n,n)$-tangles are called ``$n$-link diagrams" and ``$n$-tangles,'' respectively.

\subsection{Basic combinatorial properties} \label{CombSect}

With valenced link diagrams and link patterns defined, we next study their combinatorial properties
(excluding all link diagrams and link patterns containing loop links). 
First, for all multiindices $\multii \in \smash{\{\OneVec{0}\}} \cup \smash{\bZpos^\#}$, we define
\begin{align}\label{DefSetDefn2} 
\DefectSet_\multii : = \big\{ s \in \bZnn \,\big|\, \text{there exists a $(\multii, s)$-valenced link pattern $\alpha \in \LP_\multii\super{s}$}\big\} . 
\end{align}
It is also useful to extend the definition of $\DefectSet_\multii$ to include all multiindices $\multii$ with some zero entries: denoting 
\begin{align} \label{NonNegIndexSetDef}
\bZnn^\# := \bZnn \cup \bZnn^2 \cup \bZnn^3 \cup \dotsm , 
\end{align}
we recursively define $\DefectSet_\multii$ for any multiindex $\multii \in \smash{\bZnn^\#}$ to be 
the set $\DefectSet_{\vartheta}$, where $\vartheta$ is any multiindex obtained by dropping a zero entry from $\multii$. 
In the special case of $\multii = \OneVec{n}$, 
definition~\eqref{DefSetDefn2} for $\DefectSet_n$ agrees with~\eqref{DefectSet} given in section~\ref{TLSecIntro}.

By breaking links into pairs of defects, it becomes evident that there are integers 
$\smin(\multii), \smax(\multii) \geq 0$ such that
\begin{align}\label{DefSet2} 
\DefectSet_\multii = \{ \smin(\multii), \smin(\multii) + 2, \; \ldots, \; \smax(\multii) \} . 
\end{align}
In particular, to determine this set, it suffices to find its extreme values.  
We establish this in lemmas~\ref{SpecialDefLem} and~\ref{SminLem}.

\begin{lem} \label{SpecialDefLem} 
Let $r,s,t \in \bZnn$.  We have 
\begin{align}\label{SpecialDefSet} 
\DefectSet\sub{s} = \{s\} \qquad \textnormal{and} \qquad \DefectSet\sub{r,t} = \{ |r-t| , |r-t| + 2, \ldots, r+t\} . 
\end{align}
Furthermore, we have the symmetry relations
\begin{align} \label{SameDefSet} 
s \in \DefectSet\sub{r,t} \qquad \Longleftrightarrow \qquad r \in \DefectSet\sub{t,s} \qquad \Longleftrightarrow \qquad t \in \DefectSet\sub{s,r} . 
\end{align}
\end{lem}

\begin{proof} 
The proof of~\eqref{SpecialDefSet} is a simple combinatorial exercise. 
Then~\eqref{SameDefSet} immediately follows from~\eqref{SpecialDefSet}.
\end{proof}

In the next lemma, we give a recursion formula for the set $\DefectSet_\multii$. 
To state the recursion, we use the following notation, frequently appearing throughout this article:
\begin{align}\label{hats} 
\multii = (\sIndex_1, \sIndex_2, \ldots, \sIndex_{\np_\multii}) 
\qquad  \Longrightarrow  \qquad 
\lds := (\sIndex_1, \sIndex_2, \ldots, \sIndex_{\np_\multii-1}) \quad \textnormal{and} \quad t := \sIndex_{\np_\multii} .
\end{align}

\begin{lem} \label{SetRecursLem} 
Let $\multii \in \smash{\bZnn^\#}$.
With notation~\eqref{hats}, we have the recursion
\begin{align} \label{SetRecurs} 
\DefectSet_\multii = \bigcup_{r \, \in \, \DefectSet_{\lds}} \DefectSet\sub{r,t}. 
\end{align}
\end{lem} 

\begin{proof}
By dropping all of the zero entries from $\multii$,
we may assume that $\multii \in \smash{\bZpos^\#}$.
Then, for each $s \in \DefectSet_\multii$, we may write any valenced link pattern $\alpha \in \smash{\LP_\multii\super{s}}$ in the form
\begin{align}\label{BoxLinkPatt0} 
\vcenter{\hbox{\includegraphics[scale=0.275]{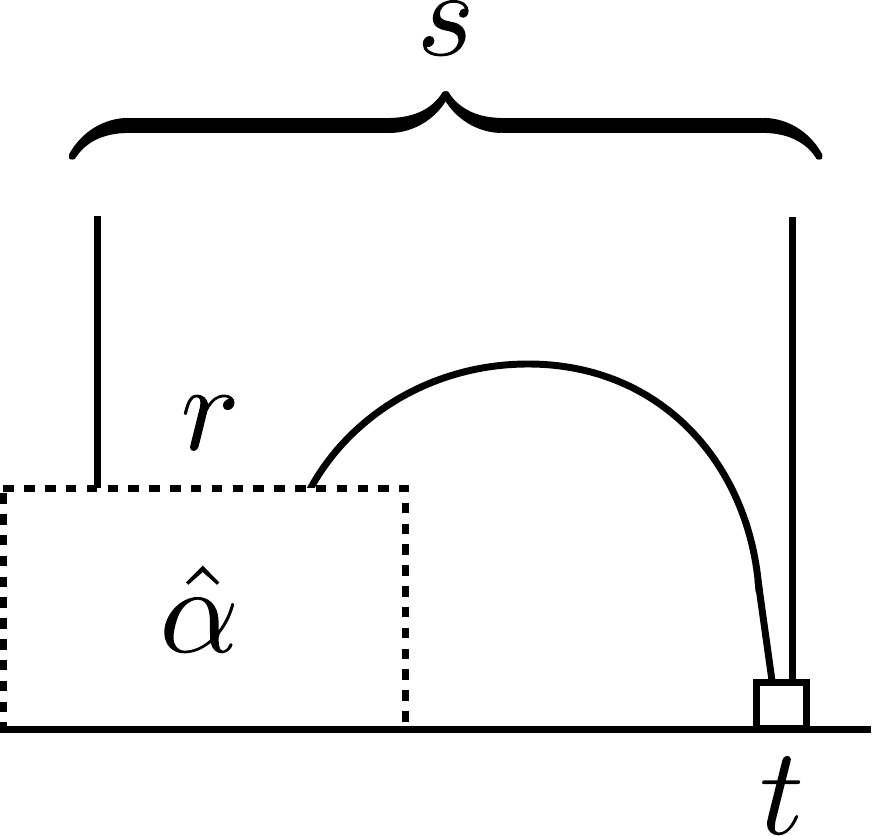} ,}}
\end{align}
for a unique valenced link pattern $\hat{\alpha} \in \smash{\LP_{\lds}}$ that depends on $\alpha$. 
With $r$ defects leaving $\hat{\alpha}$, we must have $r \in \DefectSet_{\lds}$.  
Furthermore, after removing $\hat{\alpha}$ from this valenced link pattern, we obtain the simpler $((r,t),s)$-valenced link pattern
\begin{align} 
\vcenter{\hbox{\includegraphics[scale=0.275]{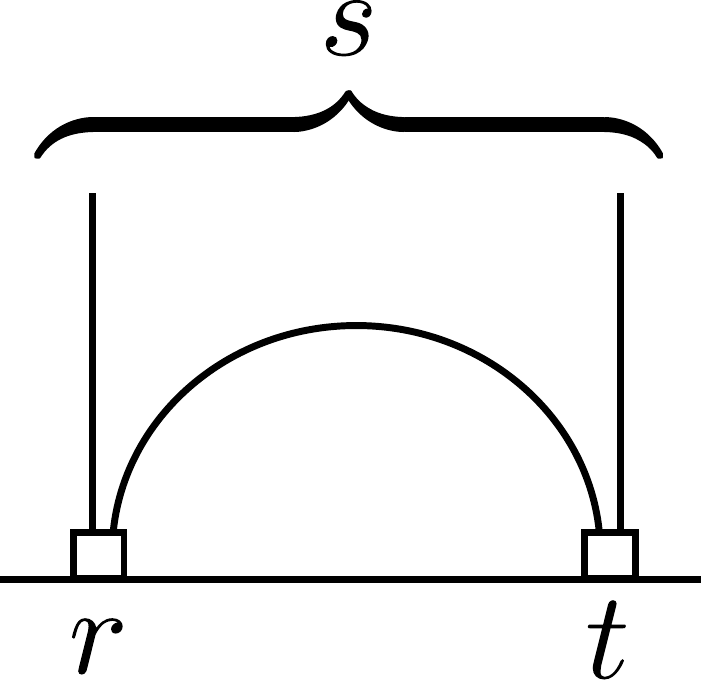} ,}}
\end{align}
whose existence implies that $s \in \DefectSet\sub{r,t}$.  Altogether, we thus have
\begin{align} 
\Big\{
s \in \DefectSet_\multii \quad \Longrightarrow \quad \text{$s \in \DefectSet\sub{r,t}$ for some $r \in \DefectSet_\lds$,} 
\quad \Longrightarrow \quad 
s \in \bigcup_{r \, \in \, \DefectSet_{\lds}} \DefectSet\sub{r,t} 
\Big\}
\qquad \Longrightarrow \qquad \DefectSet_\multii \subset \bigcup_{r \, \in \, \DefectSet_{\lds}} \DefectSet\sub{r,t}. 
\end{align}

On the other hand, let $s \in \DefectSet\sub{r,t}$ for some $r \in \DefectSet_\lds$. 
Then, insertion of a valenced link pattern $\hat{\alpha} \in \smash{\LP_{\lds}\super{r}}$
into~\eqref{BoxLinkPatt0} determines a unique valenced link pattern $\alpha \in \smash{\LP_\multii\super{s}}$. 
This shows that $s \in \DefectSet_\multii$, so we have
\begin{align} 
\Big\{
s \in \bigcup_{r \, \in \, \DefectSet_{\lds}} \DefectSet\sub{r,t} \quad \Longrightarrow \quad s \in \DefectSet_\multii 
\Big\}
\qquad \Longrightarrow \qquad 
\bigcup_{r \, \in \, \DefectSet_{\lds}} \DefectSet\sub{r,t} \subset \DefectSet_\multii . 
\end{align}
This finishes the proof.
\end{proof}

Now we are ready to determine 
$\DefectSet_\multii$. 
We split $\multii$ into two parts thus:
\begin{align}\label{cutmultii} 
\multii = (\sIndex_1, \sIndex_2, \ldots, \sIndex_{\np_\multii}) 
\qquad \qquad \Longrightarrow \qquad \qquad 
\lds_i := (\sIndex_1, \sIndex_2, \ldots, \sIndex_i) \qquad \text{and} 
\qquad \smash{\fds}_i := (\sIndex_i, \sIndex_{i+1}, \ldots, \sIndex_{\np_\multii} ). 
\end{align}

\begin{lem} \label{SminLem}
Let $\multii \in \smash{\bZnn^\#}$.
With notation~\eqref{cutmultii}, the following hold:
\begin{enumerate}
\itemcolor{red}
\item \label{itsmin1} We have 
\begin{align} \label{smaxeq} 
\smax(\multii) = \Summed_\multii \overset{\eqref{ndefn}}{:=} \sIndex_1 + \sIndex_2 + \dotsm + \sIndex_{\np_\multii}.
\end{align}
\item \label{itsmin2} For each $i \in \{1,2,\ldots, \np_\multii - 1\}$, we have the recursion
\begin{align} \label{RecuFormula} 
\smin(\lds_1) = \sIndex_1, \qquad \qquad 
\smin(\lds_{i+1}) = 
\begin{cases} \hphantom{(} \smin(\lds_i) - \sIndex_{i+1}, & \sIndex_{i+1} \leq \smin(\lds_i), \\
(\smin(\lds_i) - \sIndex_{i+1})  \; \textnormal{mod} \; 2, & \smin(\lds_i) < \sIndex_{i+1} < \smax(\lds_i), \\
\hphantom{(} \sIndex_{i+1} - \smax(\lds_i), & \smax(\lds_i) \leq \sIndex_{i+1}. 
\end{cases}
\end{align}
In particular, with 
$\smin(\multii) = \smin(\lds_{\np_\multii})$, this recursion formula with $i = \np_\multii - 1$ determines $\smin(\multii)$. 
\end{enumerate}
\end{lem}

\begin{proof} 
Item~\ref{itsmin1} immediately follows from considering the $(\multii,\Summed_\multii)$-valenced link state with only defects and no links. 
For item~\ref{itsmin2}, we observe that by lemma~\ref{SetRecursLem}, we have
\begin{align}
\smin( \lds_{i+1} ) 
\overset{\eqref{SetRecurs}}{=} \min_{ r \, \in \, \DefectSet_{\lds_i}} \min \DefectSet\sub{r, \sIndex_{i+1}} 
\overset{\eqref{SpecialDefSet}}{=} \min_{ r \, \in \, \DefectSet_{\lds_i}} | r - \sIndex_{i+1} |. 
\end{align}
Because $\DefectSet_{ \lds_i}$ has the form~\eqref{DefSet2}, this result simplifies to the right side of~\eqref{RecuFormula}.
\end{proof}

We note that using the form~\eqref{DefSet2} for $\DefectSet_\multii$, it is straightforward to show that lemma~\ref{SminLem} implies that
\begin{align} \label{subset} 
\DefectSet_\multii \subset \DefectSet_{\Summed_\multii}. 
\end{align}

We conclude our investigation of $\smin(\multii)$ with the following lemma, 
which we also need in section~\ref{rofSect32}.

\begin{lem} \label{DefectLem} 
If $\alpha \in \smash{\LP_\multii\super{s}}$ with $s = \smin(\multii)$, then 
\begin{enumerate}
\itemcolor{red}
\item \label{DefIt1} all $\smin(\multii)$ defects of $\alpha$ attach to a common node of $\alpha$, and
\item \label{DefIt2} if all defects of $\alpha$ attach to its $i$:th node, then $\smin(\multii) < \sIndex_i$ if $\np_\multii > 1$, and $\smin(\multii) = \sIndex_1$ if $\np_\multii = 1$.
\end{enumerate}
In particular, items~\ref{DefIt1} and~\ref{DefIt2} together imply that
\begin{align} \label{sminineq} 
\smin(\multii) \quad 
\begin{cases} 
< \max \multii, & \np_\multii > 1, \\ = \max \multii, & \np_\multii = 1. 
\end{cases} 
\end{align}
\end{lem} 

\begin{proof}  
To prove item~\ref{DefIt1}, we assume the contrary.  
Then, replacing two adjacent defects attached to different boxes by a link creates a valenced link pattern with fewer than $\smin(\multii)$ links, 
contradicting the definition of $\smin(\multii)$:
\begin{align}
\vcenter{\hbox{\includegraphics[scale=0.275]{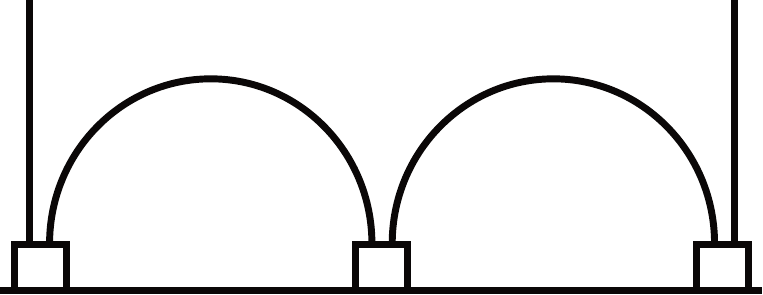}}} 
\qquad \longmapsto \qquad 
\vcenter{\hbox{\includegraphics[scale=0.275]{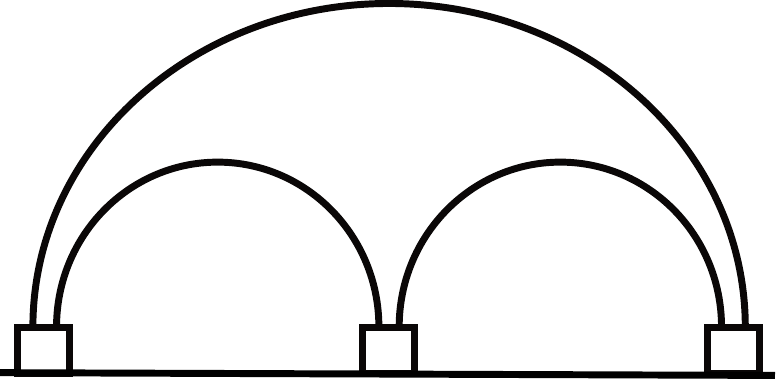} .}}
\end{align}
This proves item~\ref{DefIt1} by contradiction.

Item~\ref{DefIt2} is straightforward to prove for $\np_\multii \in \{1,2\}$.  
For $\np_\multii \geq 3$, we prove item~\ref{DefIt2} by contradiction. Thus, we
assume the contrary: $\smin(\multii) = \sIndex_i$.  With no defects attached to the other boxes, 
we make the replacement (here, $i=1$)
\begin{align}
\vcenter{\hbox{\includegraphics[scale=0.275]{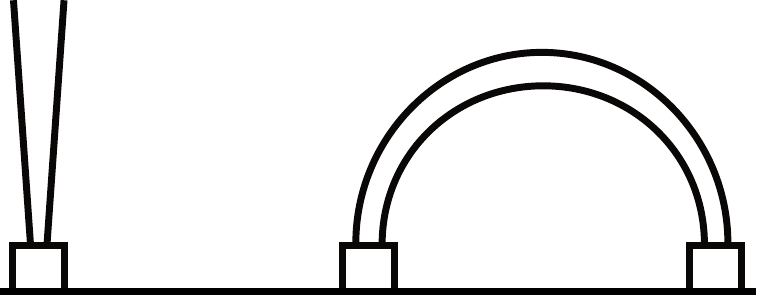}}} \qquad \longmapsto \qquad \vcenter{\hbox{\includegraphics[scale=0.275]{e-smin2.pdf} ,}} 
\end{align}
creating a valenced link pattern with $\smin(\multii)$ defects not all attached to a common box. This contradicts item~\ref{DefIt1}.
\end{proof}

Next, in lemmas~\ref{WJSandwichLem} and~\ref{LinkDiagPattLem2}, we obtain two isomorphisms relating the vector space $\TL_\multii^\multiii$
with spaces of link states. Using them, in corollary~\ref{WJDimLem1} we determine the dimension of $\TL_\multii^\multiii$.
Finally, in lemma~\ref{LSDimLem2}, we find a recursion formula for 
the total number of $(\multii,s)$-valenced link patterns.  To state these results, we first introduce notation.

For a $\multii$-valenced link state $\alpha$,
we let $\tilde{\alpha}$ denote the valenced link state obtained by reflecting $\alpha$ about a vertical axis:
\begin{align}\label{Flip} 
\alpha \quad = \quad \vcenter{\hbox{\includegraphics[scale=0.275]{e-LinkPattern3_valenced.pdf}}}
\qquad \qquad \Longrightarrow \qquad \qquad
\tilde{\alpha} \quad = \quad \vcenter{\hbox{\includegraphics[scale=0.275]{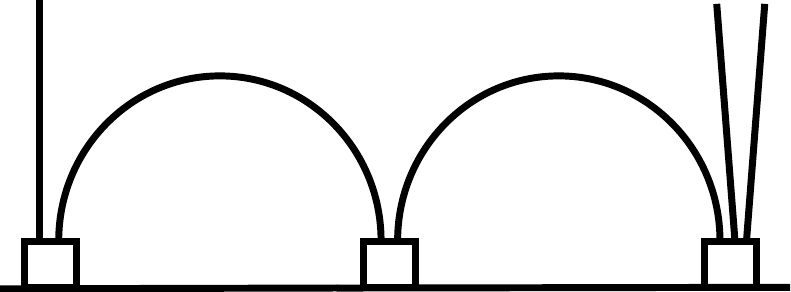} .}} 
\end{align}
Thus, we have $\tilde{\alpha} \in \LS_{\tilde{\multii}}$, where the multiindex $\tilde{\multii}$ is given by
inverting the order of the indices in $\multii$,
\begin{align}
\multii = (\sIndex_1, \sIndex_2, \ldots, \sIndex_{\np_\multii-1}, \sIndex_{\np_\multii}) 
\qquad \Longrightarrow \qquad
\tilde{\multii} := (\sIndex_{\np_\multii}, \sIndex_{\np_\multii-1}, \ldots, \sIndex_2, \sIndex_1). 
\end{align}
Now, we note that
\begin{align} \label{EqualDefSets} 
\dim \LS_{\tilde{\multii}}\super{s} = \dim \LS_\multii\super{s} 
 \qquad \text{and} \qquad \DefectSet_{\tilde{\multii}} = \DefectSet_\multii. 
\end{align}

In section~\ref{WJLinkDiagSect}, we construct valenced link patterns from valenced link diagrams.  Conversely, 
we may construct any $(\multii, \multiii)$-valenced link diagram with $s$ crossing links from 
a $(\multii,s)$-valenced link pattern $\alpha$ and a $(\multiii,s)$-valenced link pattern $\beta$ in the following way:

\begin{enumerate}[leftmargin=*]
\itemcolor{red}
\item \label{FormIt1} we flip $\beta \mapsto \smash{\tilde{\beta}}$, 

\item \label{FormIt2} we position $\alpha$ to the left of $\smash{\tilde{\beta}}$ in the plane, 

\item \label{FormIt3} we rotate $\alpha$ and $\smash{\tilde{\beta}}$ by $-\pi/2$ and $\pi/2$ radians respectively, and 

\item \label{FormIt4} we join the $s$ defects of $\alpha$ and $\smash{\tilde{\beta}}$ together pairwise top-to-bottom.
\end{enumerate}
We denote the $(\multii, \multiii)$-valenced link diagram thus obtained by $\BarAction \alpha \quad \beta \BarAction$.
For example, 
\begin{align} 
\overbrace{\vcenter{\hbox{\includegraphics[scale=0.275]{e-LinkPattern5_valenced.pdf}}}}^{\alpha \, \in \,\LP\super{3}\sub{1,1,1,2}} \qquad  \text{and} \qquad
\overbrace{\vcenter{\hbox{\includegraphics[scale=0.275]{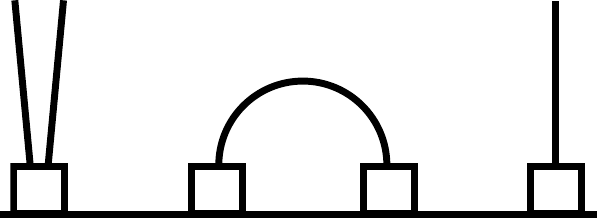}}}}^{\beta \, \in \,\LP\super{3}\sub{2,1,1,1}} 
\qquad \qquad \longmapsto \qquad \qquad
\vcenter{\hbox{\includegraphics[scale=0.275]{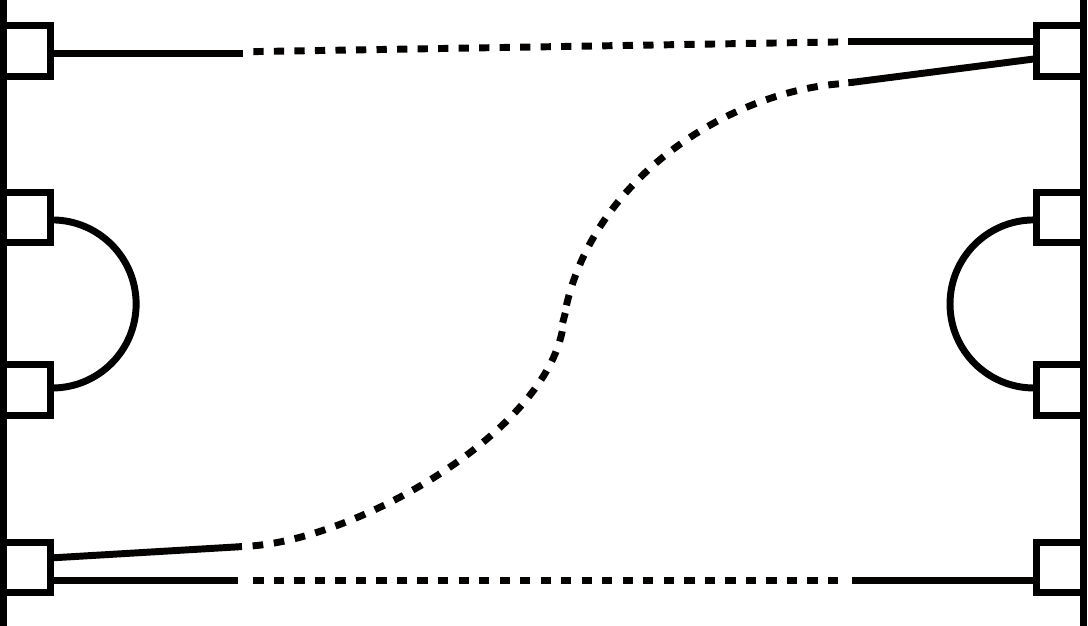}}} \quad 
= \quad \overbrace{\;\;\vcenter{\hbox{\includegraphics[scale=0.275]{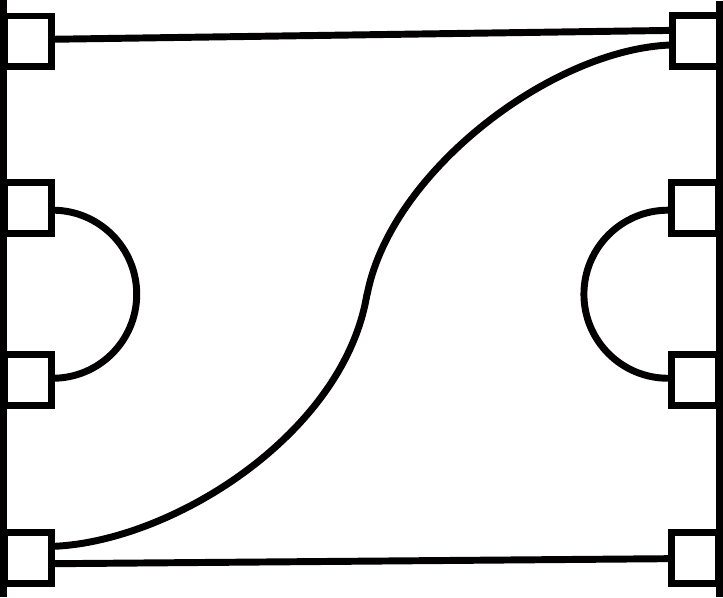} .}}}^{\BarAction \alpha \quad  \beta \BarAction \, \in \, \LD\sub{1,1,1,2}\super{2,1,1,1}}
\end{align}

%

\begin{lem} \label{WJSandwichLem}  
The map $\BarAction \, \cdot \quad \cdot \, \BarAction \colon 
\smash{\underset{s \, \in \, \DefectSet_\multii^\multiii}{\bigoplus}
\LS_\multii\super{s}} \otimes \smash{\LS_\multiii\super{s}} \longrightarrow \TL_\multii^\multiii$ 
defined by linear extension of 
\begin{align} \label{WJSandwichMap}
\alpha \otimes \beta \qquad \longmapsto \qquad \BarAction \alpha \quad \beta \BarAction 
\quad := \quad \vcenter{\hbox{\includegraphics[scale=0.275]{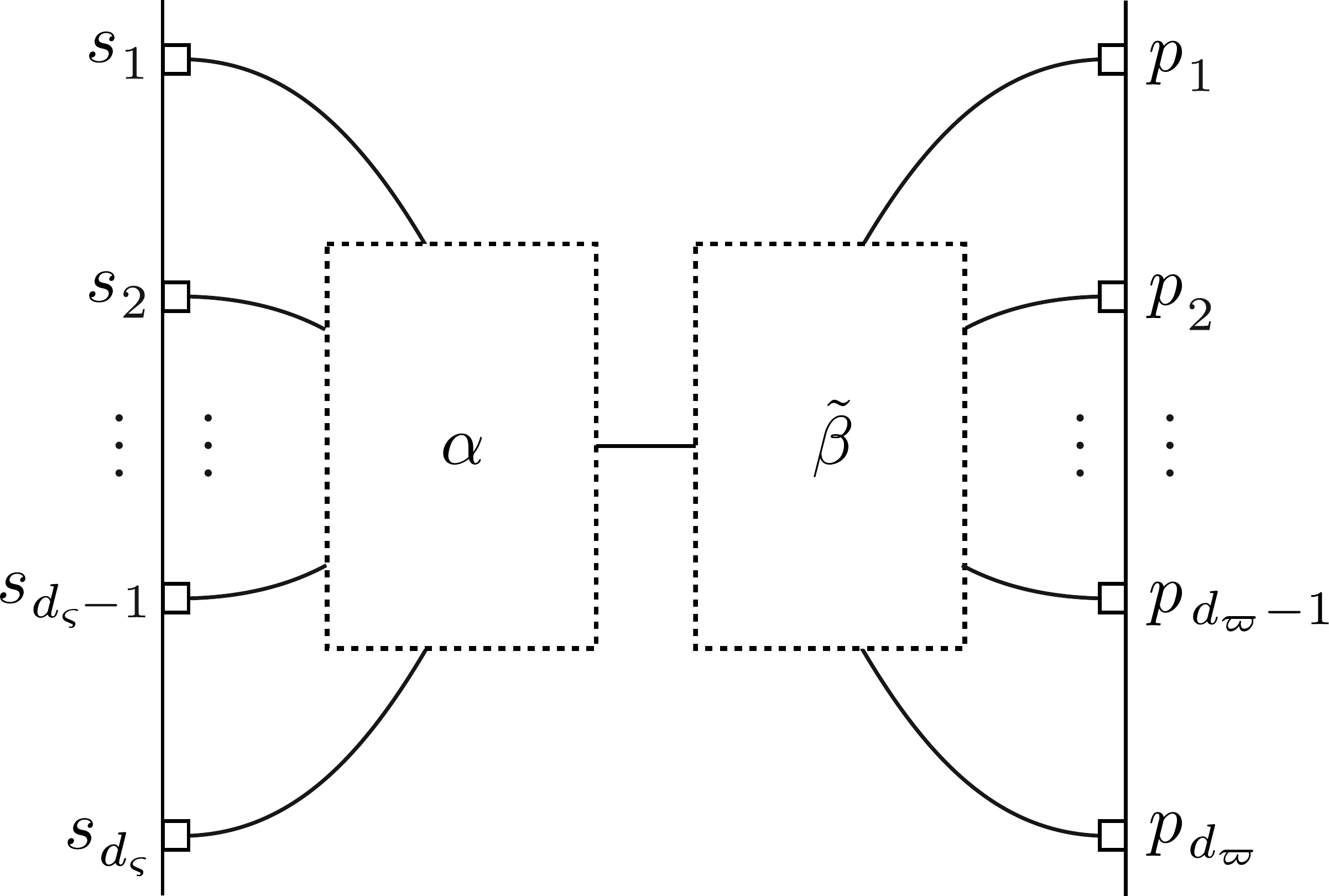} ,}}
\end{align}
for all valenced link patterns $\alpha \in \smash{\LP_\multii\super{s}}$ and $\beta \in \smash{\LP_\multiii\super{s}}$, is an isomorphism of vector spaces.
\end{lem}

\begin{proof} 
This map sends the collection 
$\smash{ \big\{ \alpha \otimes \beta \, \big| \,  s \in \DefectSet_\multii^\multiii, \; \alpha \in \LP_\multii\super{s}, \; \beta \in \LP_\multiii\super{s} \big\} }$, 
which is a basis for its domain, to 
\begin{align}
\LD_\multii^\multiii 
= \big\{ \BarAction \alpha \quad \beta \BarAction \, \big| \, s \in \DefectSet_\multii^\multiii, \; \alpha \in \LP_\multii\super{s}, \; \beta \in \LP_\multiii\super{s} \big\}, 
\end{align}
which is a basis for its codomain.  The claim follows. 
\end{proof}

With notation~\eqref{MultiindexNotation}, we let $\oplus$ denote the operation that concatenates two multiindices $\multii, \multiii \in \{ \vec{0} \} \cup \bZpos^\#$,
\begin{align} \label{concatenateMultiindices}
\left\{ \begin{array}{rl} 
\multii \oplus \vec{0} & := \multii, \\ 
\vec{0} \oplus \multiii & := \multiii, \\ 
\multii \oplus \multiii & :=  (\sIndex_1, \sIndex_2, \ldots, \sIndex_{\np_\multii}, p_1, p_2, \ldots, p_{\np_\multiii}) .
\end{array} \right.
\end{align}

Next, we present another useful isomorphism, 
from $\smash{\TL_\multii^\multiii}$ to $\smash{\LS_{\multii \oplus \smash{\tilde{\multiii}}}\super{0}}$,
by sending each valenced tangle $\BarAction \alpha \quad \beta \BarAction \in \smash{\TL_\multii^\multiii}$ 
to a valenced link state in 
$\smash{\LS_{\multii \oplus \smash{\tilde{\multiii}}}\super{0}}$ formed by ``unfolding it,'' 
\begin{align}
\vcenter{\hbox{\includegraphics[scale=0.275]{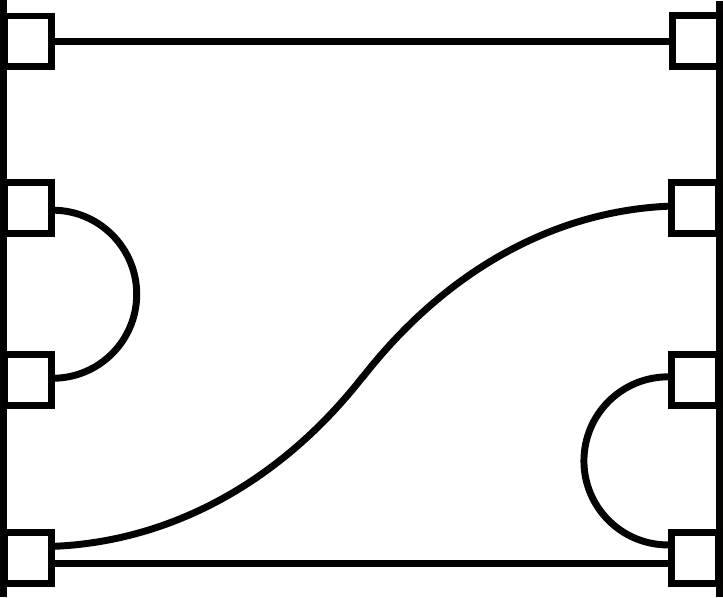}}}
\qquad \qquad \longmapsto \qquad \qquad
\vcenter{\hbox{\includegraphics[scale=0.275]{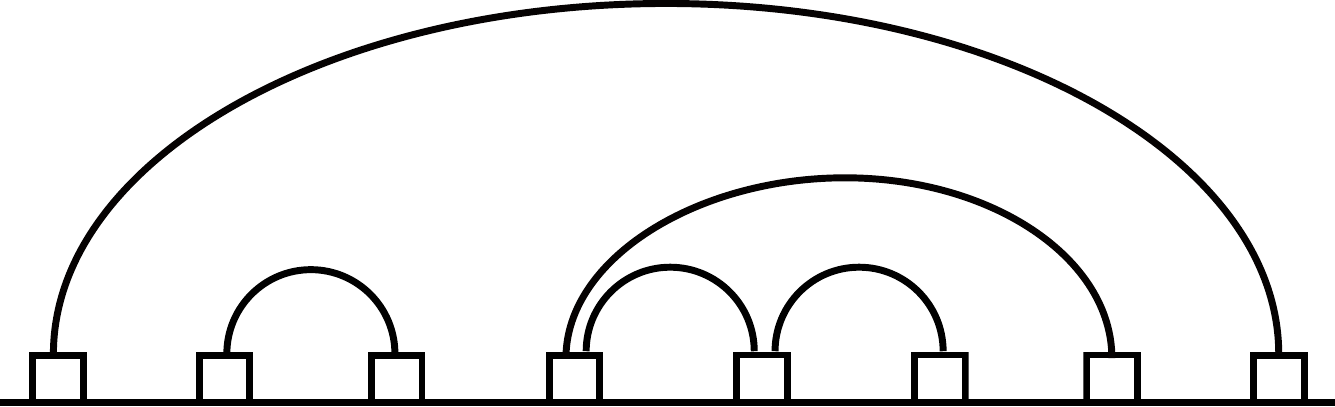} ,}}
\end{align}
that is, joining all of the $s \in \DefectSet_\multii^\multiii$
defects of $\alpha \otimes \tilde{\beta} \in \smash{\LS_\multii\super{s}} \otimes \smash{\LS_{\smash{\tilde{\multiii}}}\super{s}}$ together.

\begin{lem} \label{LinkDiagPattLem2} 
The following map is an isomorphism of vector spaces from $\smash{\TL_\multii^\multiii}$ to $\smash{\LS_{\multii \oplus \smash{\tilde{\multiii}}}\super{0}}$:
\begin{align}\label{LinkDiagPatt2} 
\BarAction \alpha \quad \beta \BarAction 
\qquad \longmapsto \qquad 
\vcenter{\hbox{\includegraphics[scale=0.275]{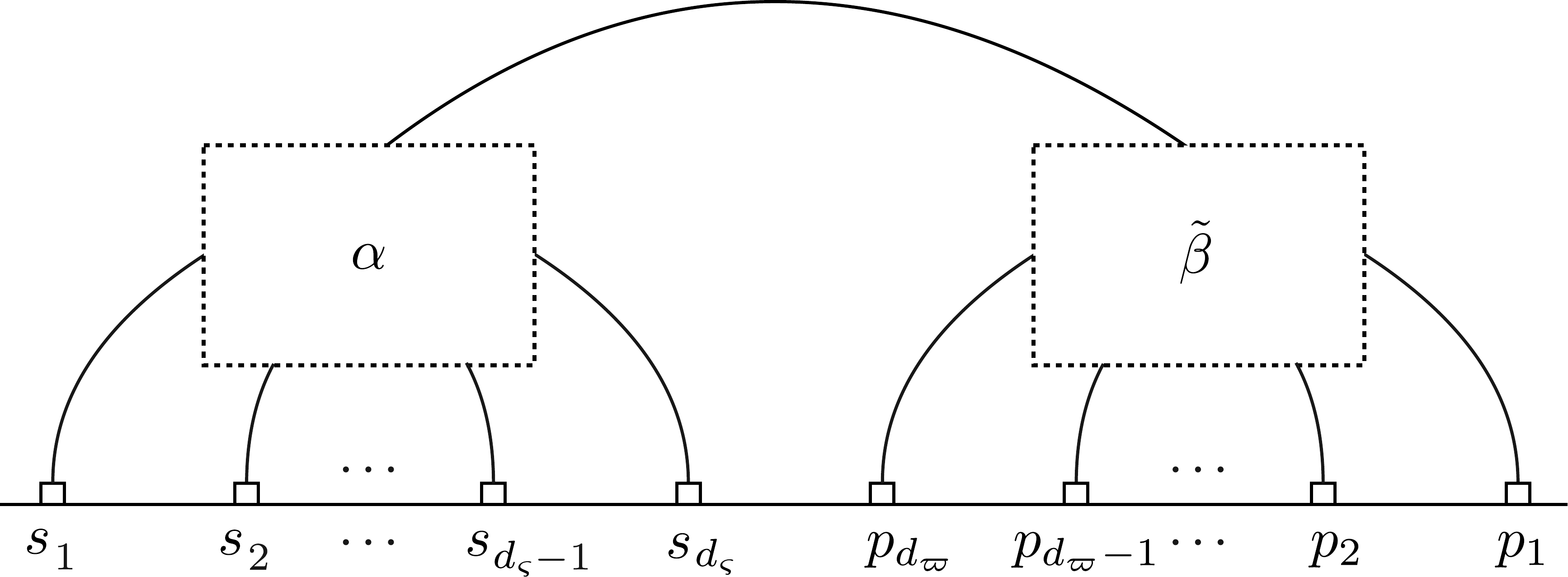} .}}
\end{align}
\end{lem}

\begin{proof} 
The map~\eqref{LinkDiagPatt2} from valenced link diagrams in $\LD_\multii^\multiii$, a basis for $\TL_\multii^\multiii$, to valenced link patterns in 
$\smash{\LP_{\multii \oplus \smash{\tilde{\multiii}}}\super{0}}$, a basis for $\smash{\LS_{\multii \oplus \smash{\tilde{\multiii}}}\super{0}}$, is a bijection. The claim follows.
\end{proof}

Among other uses in section~\ref{StdModulesSect}, we may use valenced link patterns to find the dimension of the vector space $\smash{\TL_\multii^\multiii}$.  

\begin{cor} \label{WJDimLem1} 
We have
\begin{align} \label{Dim56}
\smash{\dim \LS_{\multii \oplus \tilde{\multiii}}\super{0}} 
= \dim \smash{\TL_\multii^\multiii} = \sum_{s \, \in \, \DefectSet_\multii^\multiii} \big(\dim \LS_\multii\super{s}\big) \big(\dim \LS_\multiii\super{s} \big).
\end{align}
\end{cor}

\begin{proof} 
The first equality follows from lemma~\ref{LinkDiagPattLem2} with~\eqref{EqualDefSets}, and the second equality follows from lemma~\ref{WJSandwichLem}.
\end{proof}

In particular, to determine the dimension of $\smash{\TL_\multii^\multiii}$, 
it is sufficient to determine the dimensions of the vector spaces $\smash{\LS_\multii\super{s}}$.  
For this, using notation~\eqref{hats}, we define the numbers $\smash{\{\Dim_\multii\super{s}\}_{s \in \DefectSet_\multii}}$
as the unique solution to the recursion
\begin{align} 
\label{Recursion2} 
\Dim_\multii\super{s} \hspace{.1cm}
 = \hspace{.5cm} \sum_{\mathclap{r \, \in \, \DefectSet_{\lds} \, \cap \, \DefectSet\sub{s,t} }}  
\quad \Dim_{\lds}\super{r}, 
\qquad \text{and} \qquad \Dim \sub{s}\super{s} = 1.
\end{align}
In lemma~\ref{LSDimLem2}, we show that $\smash{\Dim_\multii\super{s}}$ equals the dimension of the vector space $\smash{\LS_\multii\super{s}}$.

In the special case that $\multii = \OneVec{n}$ for some integer $n \geq 0$, we denote these numbers by 
$\smash{\Dim_n\super{s}} := \smash{\Dim_{\OneVec{n}}\super{s}}$.
Then for all integers $s \in \DefectSet_n$ (so $n$ and $s$ have the same parity), recursion~\eqref{Recursion2} reduces to
\begin{align} 
\label{Recursion} 
\Dim_n\super{s} \hspace{.1cm}
 = \hspace{.5cm} \sum_{\mathclap{r \, \in \, \DefectSet_{n-1} \, \cap \, \DefectSet\sub{s,1} }}  
\quad \Dim_{n-1}\super{r} = 
\begin{cases} 
\Dim_{n-1}\super{1}, & s = 0, \\[5pt] 
\Dim_{n-1}\super{s-1} + \Dim_{n-1}\super{s+1}, 
& s \in \{1,2,\ldots, n-1\}, \\[5pt] 
\Dim_{n-1}\super{n-1}, & s = n , 
\end{cases} \qquad \qquad
\text{and} \qquad \qquad \Dim_1\super{1} = 1.
\end{align}
We also note that $\smash{\Dim_n\super{n}} = 1$ for all $n \in \bZnn$.
The following well-known formula gives the unique solution to~\eqref{Recursion}:
\begin{align}\label{Dns} 
\Dim_n\super{s} = \frac{2 (s + 1)}{n + s + 2} \binom{n}{\frac{n + s}{2}} . 
\end{align}
The $n$:th Catalan number $C_n$ arises in a special instance of these numbers:
\begin{align}\label{Catalan} 
C_n := \frac{1}{n + 1} \binom{2n}{n} \overset{\eqref{Dns}}{=} \Dim_{2n}\super{0}. 
\end{align}

\begin{lem} \label{LSDimLem2}
We have 
\begin{align}\label{LSDim2} 
\dim \smash{\LS_\multii\super{s}} = \# \LP_\multii\super{s} = \smash{\Dim_\multii\super{s}} . 
\end{align}
\end{lem}

\begin{proof} 
With $\dim \smash{\LS_\multii\super{s}} = \# \smash{\LP_\multii\super{s}}$ by definition, 
we only need to prove the second equality.  For this, we show that
$\# \smash{\LP_\multii\super{s}}$ satisfies recursion~\eqref{Recursion2}. 
As in the proof of lemma~\ref{SetRecursLem}, we may write any valenced link pattern $\alpha \in \smash{\LP_\multii\super{s}}$ in 
the form~\eqref{BoxLinkPatt0}, for a unique valenced link pattern $\hat{\alpha} \in \smash{\LP_{\lds}}$ that depends on $\alpha$, with
\begin{align} \label{rReq} 
r \in \DefectSet_{\lds} \cap \DefectSet\sub{s,t}. 
\end{align}

On the other hand, insertion of a valenced link pattern $\hat{\alpha} \in \smash{\LP_{\lds}\super{r}}$ 
with $r$ as in~\eqref{rReq}  into~\eqref{BoxLinkPatt0} 
determines a unique valenced link pattern $\alpha \in \smash{\LP_\multii\super{s}}$. This establishes a bijection 
\begin{align} 
\LP_\multii\super{s} \qquad \longleftrightarrow \qquad \bigcup_{r \, \in \, \DefectSet_{\lds} \, \cap \, \DefectSet\sub{s,t}} \LP_{\lds}\super{r}. 
\end{align}
Therefore, the cardinalities $\#\LP_\multii\super{s}$ satisfy recursion~\eqref{Recursion2}, including its initial condition: 
$\# \smash{\LP\sub{s}\super{s}} = 1$.
Because the solution to this recursion problem is unique, we must have 
$\smash{\#\LP_\multii\super{s}} = \smash{\Dim_\multii\super{s}}$.
\end{proof}

Corollary~\ref{WJDimLem1} and lemma~\ref{LSDimLem2} together imply the following identity: 
\begin{align}\label{CuriousID} 
\sum_{s \, \in \, \DefectSet_\multii^\multiii} \Dim_\multii\super{s} \Dim_\multiii\super{s} \overset{\eqref{Dim56}}{=} \Dim_{\multii \oplus \smash{\tilde{\multiii}}}\super{0}. 
\end{align}
If $\multii =  \OneVec{n}$ and $\multiii = \OneVec{m}$ 
then this reduces to a more explicit identity including the numbers~\eqref{Dns},
\begin{align}\label{DnsId} 
\sum_{s \, \in \, \DefectSet_n^m} \Dim_n\super{s} \Dim_m\super{s} \overset{\eqref{CuriousID}}{=} \Dim_{n + m}\super{0} \overset{\eqref{Catalan}}{=} C_{(n+m)/2}.
\end{align}


\subsection{Composition of valenced tangles and valenced link states} \label{ValencedCompositionSec}

Next, we explain how the diagrammatic multiplication of the Temperley-Lieb algebra, discussed in section~\ref{TLSecIntro},
generalizes for valenced tangles, and how the latter act naturally on valenced link states.


The Jones-Wenzl projectors are important tools when defining composition of valenced tangles.
To define the former, we recall from section~\ref{TLSecIntro} the generators of the Temperley-Lieb algebra $\TL_n(\nu)$:
\begin{align}
\hphantom{\vcenter{\hbox{\includegraphics[scale=0.275]{e-TLalgebra5.pdf}}}}
\Gen_i \quad & \overset{\eqref{ExtMe}}{:=}  
&& \vcenter{\hbox{\includegraphics[scale=0.275]{e-TLalgebra6.pdf}}} \quad \in \TL_n(\nu) .
\quad \text{for all $i \in \{1,2, \ldots, n-1\}$},
\end{align} 
The multiplication in $\TL_n(\nu)$ is determined by diagram concatenation, 
replacing each loop by a multiplicative factor of $\nu$, the fugacity parameter
\begin{align} \label{numbertwo}
\nu = -q - q^{-1} = - [2] ,
\end{align} 
where $[2]$ is an example of a \emph{quantum integer}, defined for any $k \in \bZ$ as
\begin{align}\label{Qinteger} 
[k] = [k]_q := \frac{q^{k} - q^{-k}}{q - q^{-1}}. 
\end{align}

Recalling also definition~\eqref{MinPower} of $\ppmin(q)$, 
we define the \emph{Jones-Wenzl projector of size} $s \in \{0, 1, \ldots, \ppmin(q) - 1 \}$ 
to be the unique nonzero tangle $\WJProj\sub{s} \in \TL_s(\nu)$ 
satisfying the two properties~\cite{vj, hw, kl}
\begin{enumerate}[leftmargin=*, label = P\arabic*., ref = P\arabic*]
\itemcolor{red}
\item \label{wj1} $\smash{\WJproj\sub{s}^2} = \smash{\WJproj\sub{s}}$, and
\item \label{wj2}
$\smash{\Gen_i} \smash{\WJproj\sub{s}} = \smash{\WJproj\sub{s}} \smash{\Gen_i} = 0$ 
for all $i \in \{1, 2, \ldots, s - 1\}$.
\end{enumerate}
For example, we have
\begin{align}\label{WJ2} 
\WJProj\sub{0} = \text{the empty tangle}, \qquad \qquad
\WJProj\sub{1} = \mathbf{1}_{\TL_1(\nu)}, \qquad \qquad \text{and} \qquad \qquad
\WJProj\sub{2} = \mathbf{1}_{\TL_2(\nu)} - \nu^{-1} \Gen_1 .
\end{align}

We use the following diagrammatic representation for the Jones-Wenzl projectors.
Within a tangle, we call a collection of $s$ parallel links a \emph{cable of size $s$}, and we illustrate it as one link with label ``$s$" next to it:
\begin{align} \label{cable}
s
\begin{cases}
\quad \vcenter{\hbox{\includegraphics[scale=0.275]{e-Sstrands2.pdf}}}
\end{cases}
\; = \quad \raisebox{1pt}{\includegraphics[scale=0.275]{e-Sstrands1.pdf}  .} 
\end{align} 
We represent the tangle $\WJProj\sub{s}$ as an empty \emph{projector box} with a cable of size $s$ passing into either side:
\begin{align}\label{ProjBoxDiag}
s
\begin{cases}
\quad \vcenter{\hbox{\includegraphics[scale=0.275]{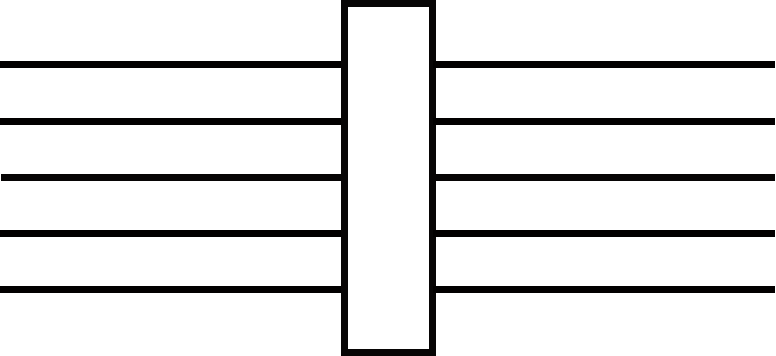}}}
\end{cases}
\; = \quad 
\vcenter{\hbox{\includegraphics[scale=0.275]{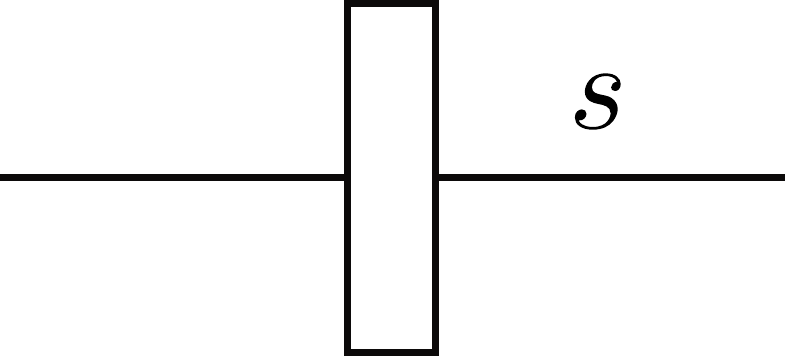}  .}}
\end{align} 
For example, the first, 
second, 
and third 
Jones-Wenzl projectors are
\begin{alignat}{2} 
\label{ExampleProj1} 
\WJProj\sub{1} \quad = \quad
\vcenter{\hbox{\includegraphics[scale=0.275]{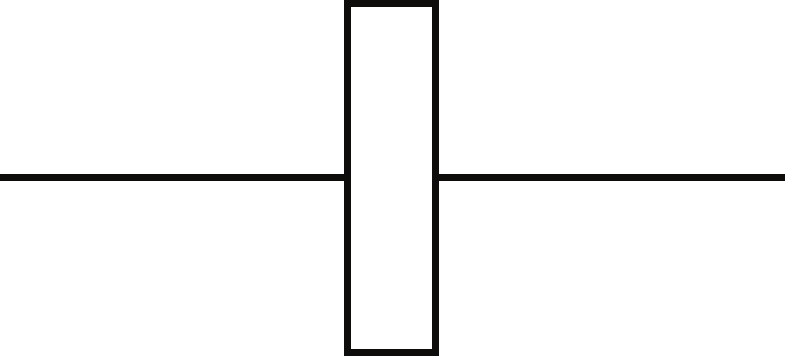}}} \quad & = \quad 
\vcenter{\hbox{\includegraphics[scale=0.275]{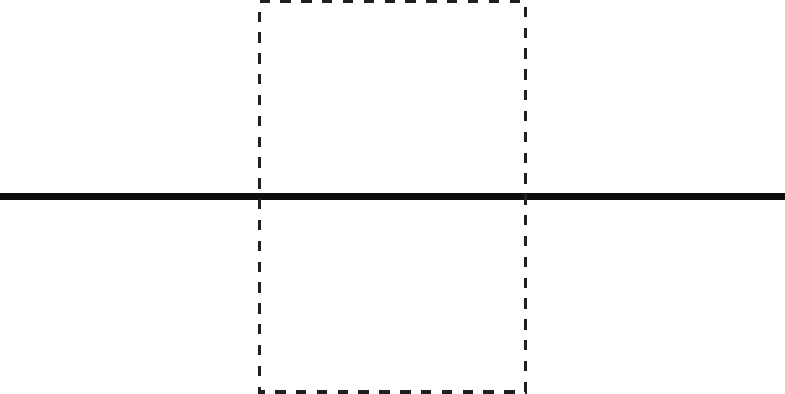} ,}} && \\[1em]
\label{ExampleProj2} 
\WJProj\sub{2} \quad = \quad
\vcenter{\hbox{\includegraphics[scale=0.275]{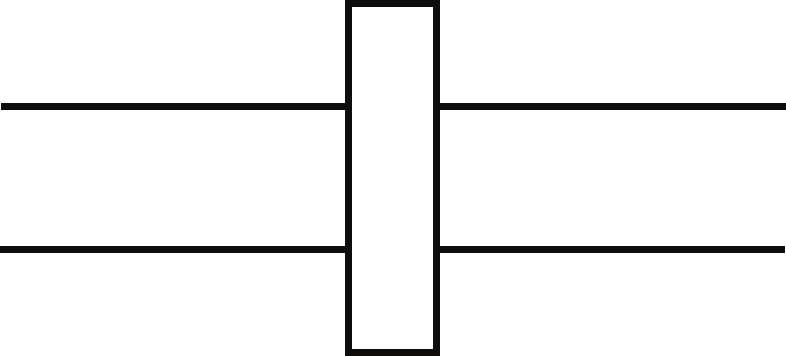}}} \quad & = \quad 
\vcenter{\hbox{\includegraphics[scale=0.275]{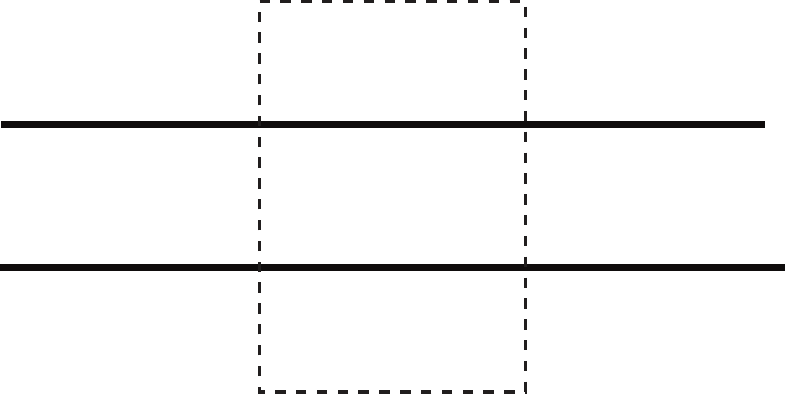}}} && \quad + \quad \frac{1}{[2]} \,\, \times \,\,
\vcenter{\hbox{\includegraphics[scale=0.275]{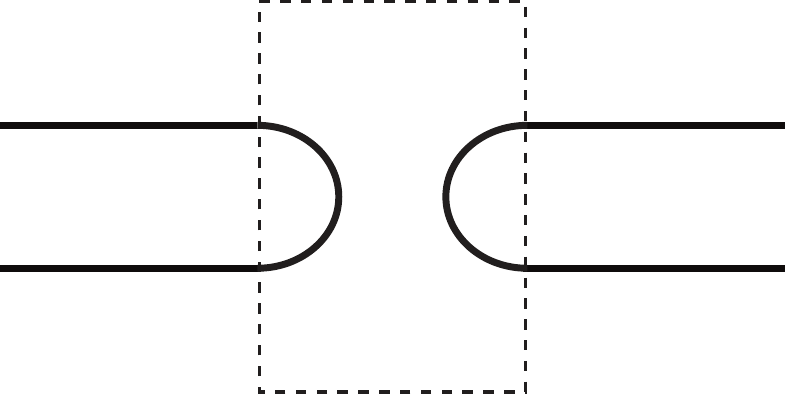} ,}} \\[1em] 
\nonumber
\WJProj\sub{3} \quad = \quad
\vcenter{\hbox{\includegraphics[scale=0.275]{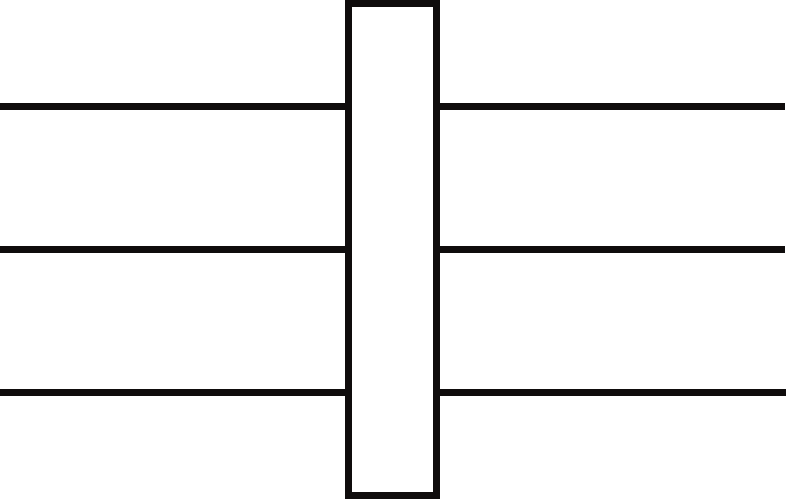}}} \quad & = \quad 
\vcenter{\hbox{\includegraphics[scale=0.275]{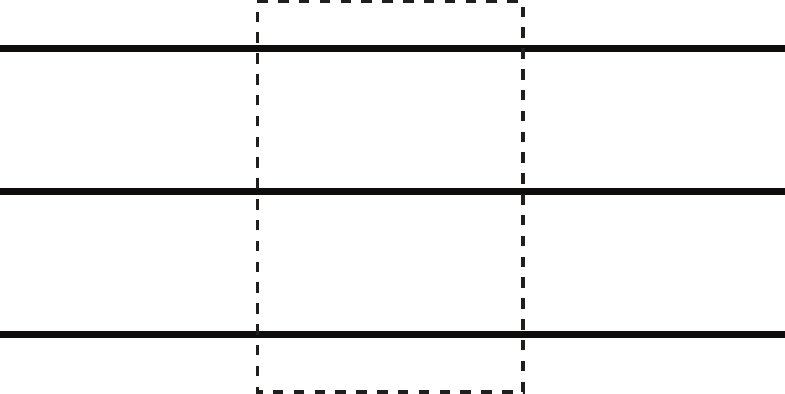}}} && \quad 
+ \quad \frac{[2]}{[3]}  \,\, \times \,\, \left( \; \vcenter{\hbox{\includegraphics[scale=0.275]{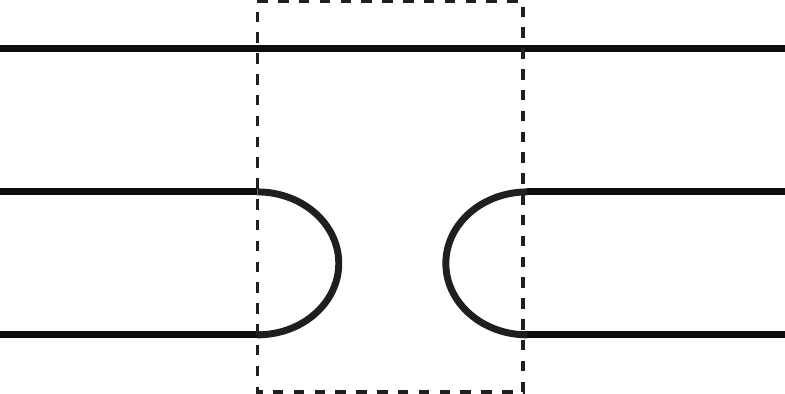}}} 
\quad + \quad \vcenter{\hbox{\includegraphics[scale=0.275]{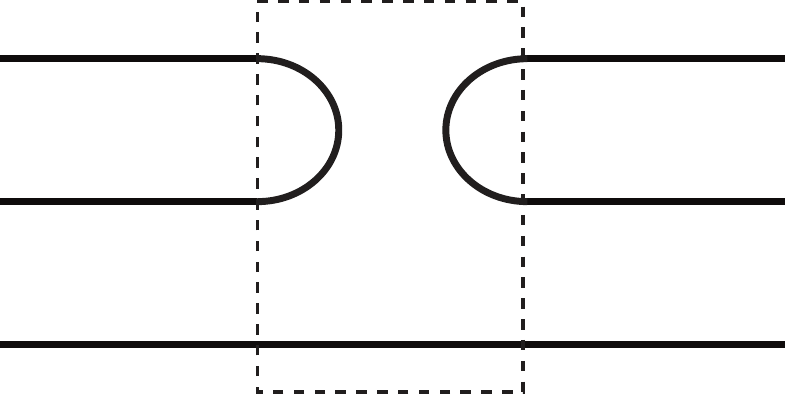}}} \; \right) \\
& && \quad + \quad \frac{1}{[3]}  \,\, \times \,\,
\left( \; \vcenter{\hbox{\includegraphics[scale=0.275]{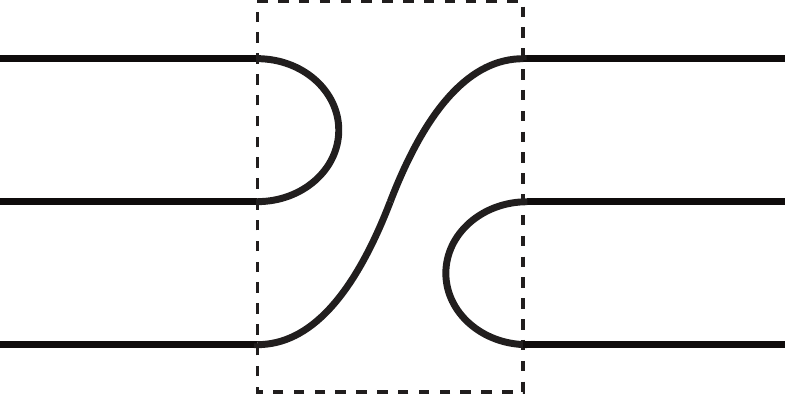}}} \quad + \quad 
\vcenter{\hbox{\includegraphics[scale=0.275]{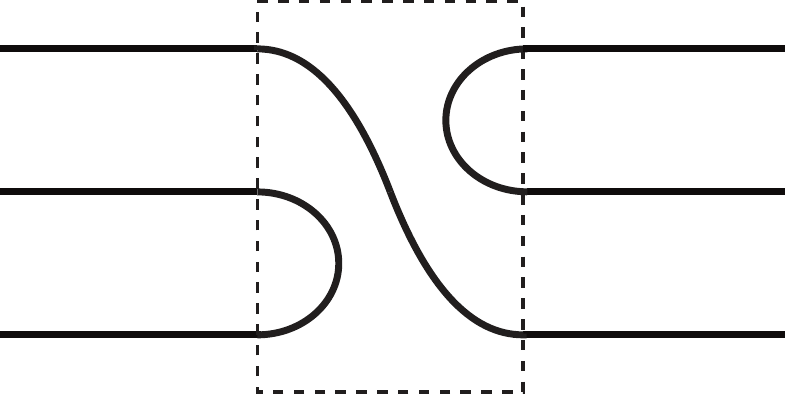}}} \; \right) .
\label{ExampleProj3} 
\end{alignat}
If $s \leq n$, then we may embed a projector box of size $s$ into a tangle in $\TL_n(\nu)$.  For example, the diagram
\begin{align}\label{CanonEmbed}
\vcenter{\hbox{\includegraphics[scale=0.275]{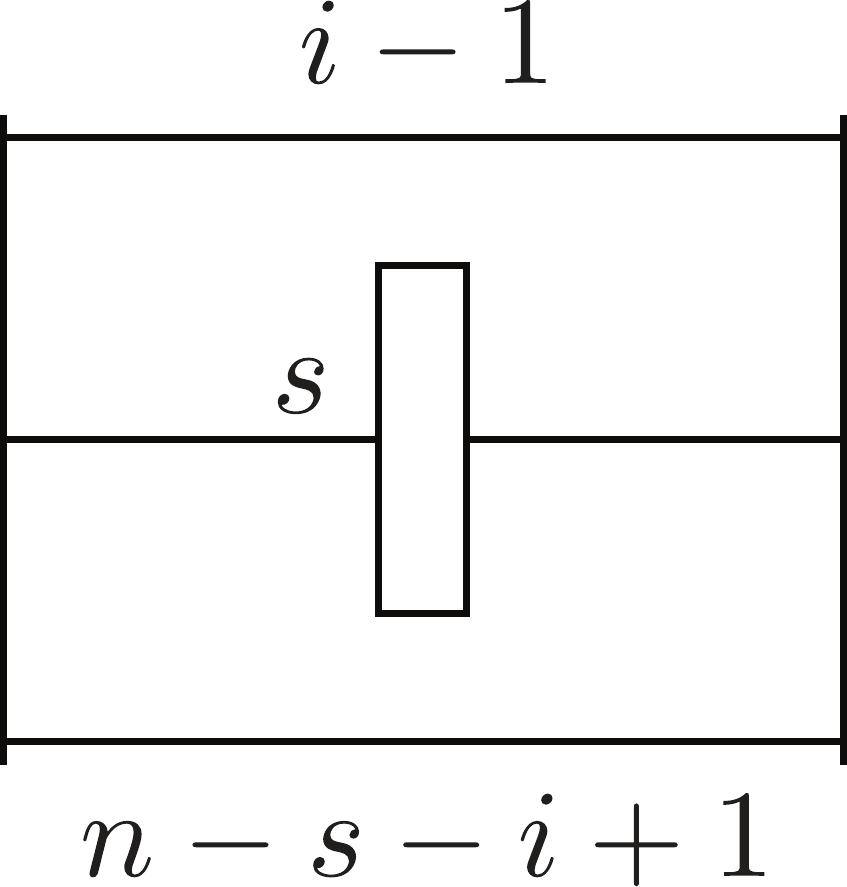}}} 
\end{align} 
represents the tangle in $\TL_n(\nu)$ obtained by replacing the box of size $s$ with the tangle $\WJProj\sub{s}$ within the larger diagram.

Abusing notation, we let the symbol $\WJProj\sub{s} \in \TL_n(\nu)$ also denote the tangle~\eqref{CanonEmbed} 
in $\TL_n(\nu)$ with $i = 1$.  Then the various projectors 
$\WJProj\sub{1}, \WJProj\sub{2}, \ldots,\WJProj\sub{n} \in \TL_n(\nu)$
satisfy the recursion relations~\cite{vj, hw, kl}
\begin{align}\label{wjrecursion} 
\smash{\WJproj\sub{1}} = \mathbf{1}_{\TL_n(\nu)} , \qquad 
\smash{\WJproj\sub{s+1}} = \smash{\WJproj\sub{s}} + \left( \frac{[s]}{[s+1]} \right) 
\smash{\WJproj\sub{s}} \smash{\Gen_s} \smash{\WJproj\sub{s}} ,
\end{align}
for all $s \in \{1, 2, \ldots, n - 1\}$. In terms of diagrams recursion relation~\eqref{wjrecursion} reads
\begin{align}\label{RecursionDiagram}
\vcenter{\hbox{\includegraphics[scale=0.275]{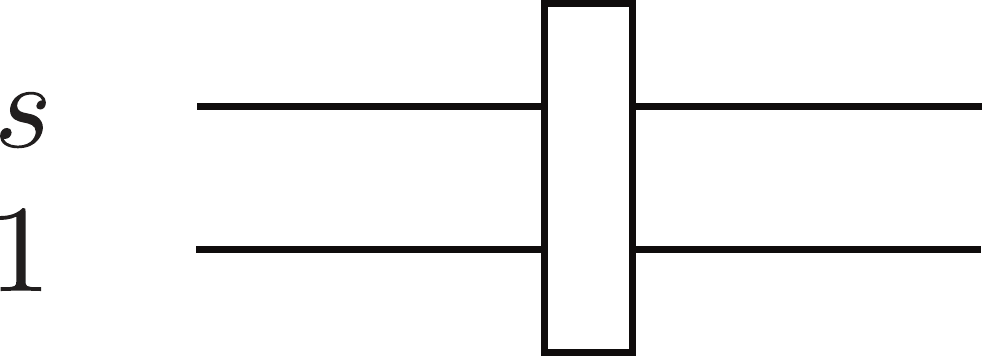}}} \quad = \quad 
\vcenter{\hbox{\includegraphics[scale=0.275]{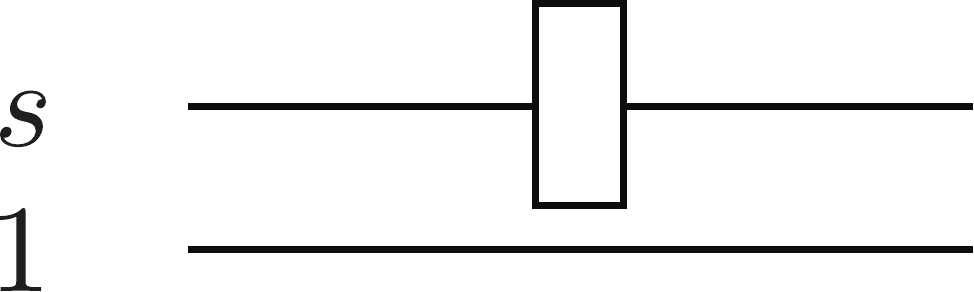}}} \quad + \quad 
\frac{[s]}{[s+1]} \,\, \times \,\, \vcenter{\hbox{\includegraphics[scale=0.275]{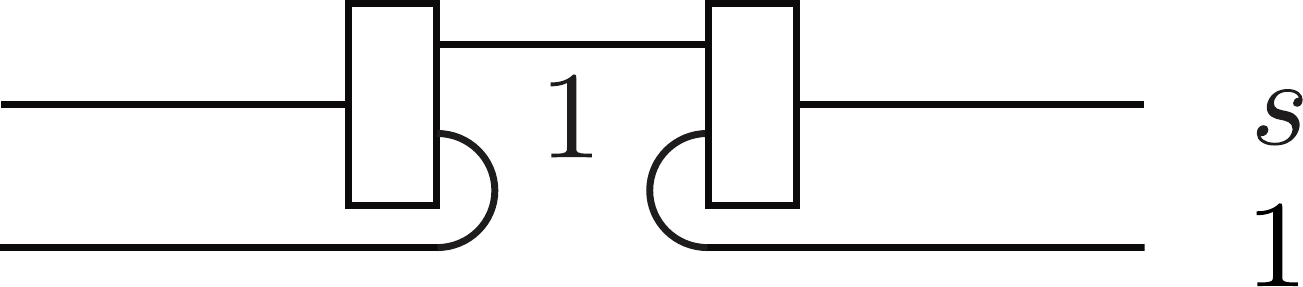} .}} 
\end{align}

Next, we state a few more facts concerning  the Jones-Wenzl projectors, for use throughout this article.  
First, we let $T^\dagger$ denote the tangle obtained by reflecting $T$ about a vertical axis.  For example,
\begin{align}\label{DaggerRefl} 
T \quad = \quad \vcenter{\hbox{\includegraphics[scale=0.275]{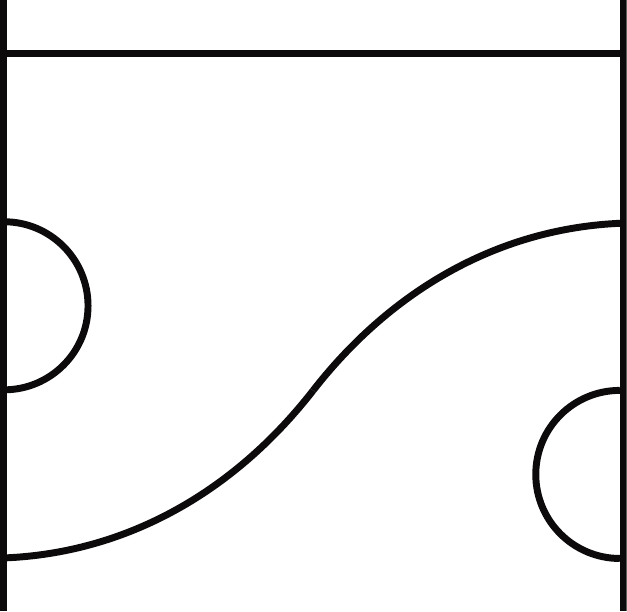}}} 
\qquad \qquad \Longrightarrow \qquad \qquad
T^\dagger \quad = \quad \vcenter{\hbox{\includegraphics[scale=0.275]{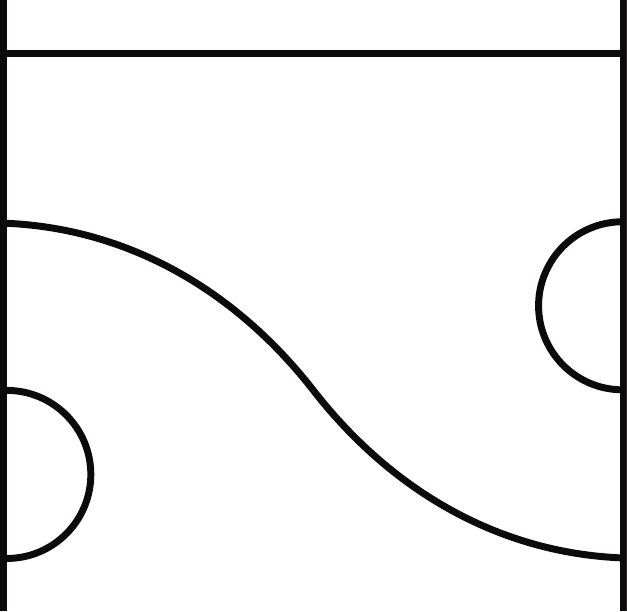} .}}
\end{align}
Rule~\eqref{wjrecursion} inductively gives the reflection symmetry 
\begin{align} \label{ReflectionSymmetryOfP}
\smash{\WJProj\sub{s}^\dagger} = \WJProj\sub{s}. 
\end{align}
With the graphical representation, properties~\ref{wj1} and~\ref{wj2} respectively translate to the diagram identities
\begin{align} 
\label{ProjectorID0} 
\tag{\ref{wj1}} 
\vcenter{\hbox{\includegraphics[scale=0.275]{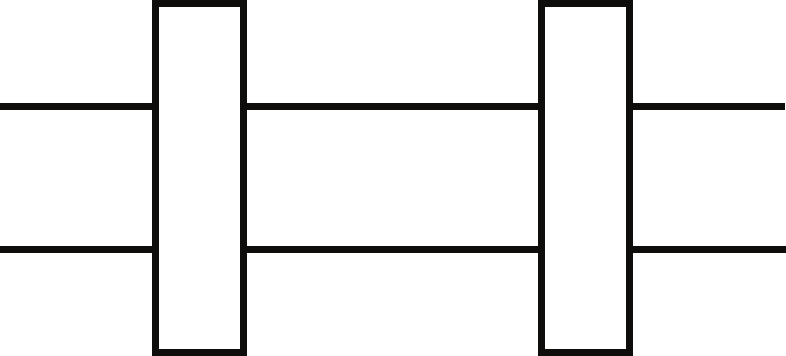}}} 
\quad & = \quad \vcenter{\hbox{\includegraphics[scale=0.275]{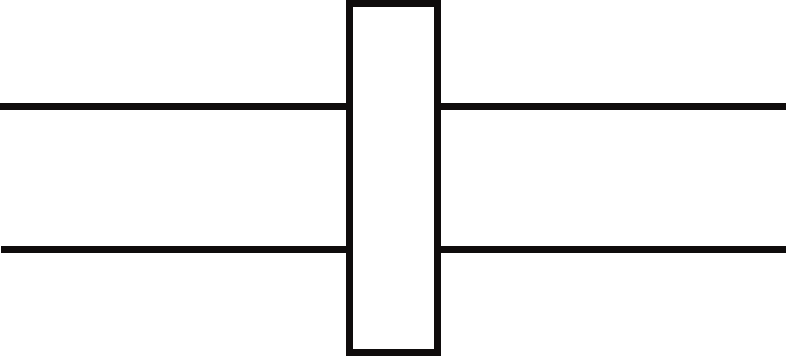}}} \\[1em]
\label{ProjectorID2} 
\tag{\ref{wj2}} 
\vcenter{\hbox{\includegraphics[scale=0.275]{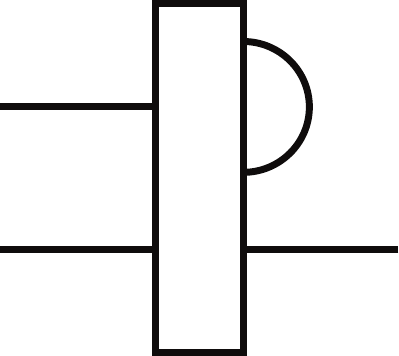}}} 
\quad = \quad \vcenter{\hbox{\includegraphics[scale=0.275]{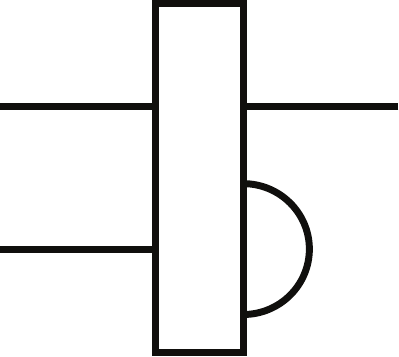}}} 
\quad & = \quad \vcenter{\hbox{\includegraphics[scale=0.275]{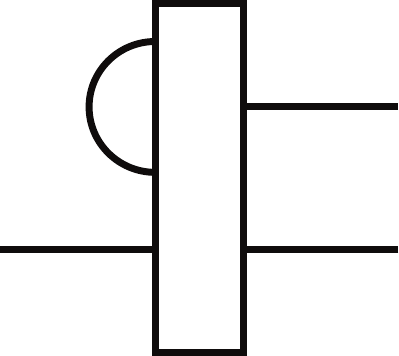}}} 
\quad = \quad \vcenter{\hbox{\includegraphics[scale=0.275]{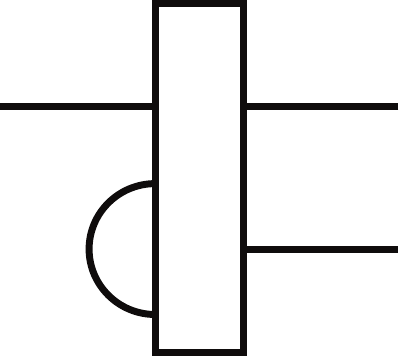}}}
\quad = \quad 0.
\end{align}
In fact, property~\ref{wj1} can be strengthened to say that $\WJProj\sub{s}\WJProj\sub{t} = \WJProj\sub{t}\WJProj\sub{s} = \WJProj\sub{s}$ whenever $t \leq s$~\cite{kl}:
\begin{align} \label{ProjectorID1}
\tag{\ref{wj1}$\red{'}$} 
\vcenter{\hbox{\includegraphics[scale=0.275]{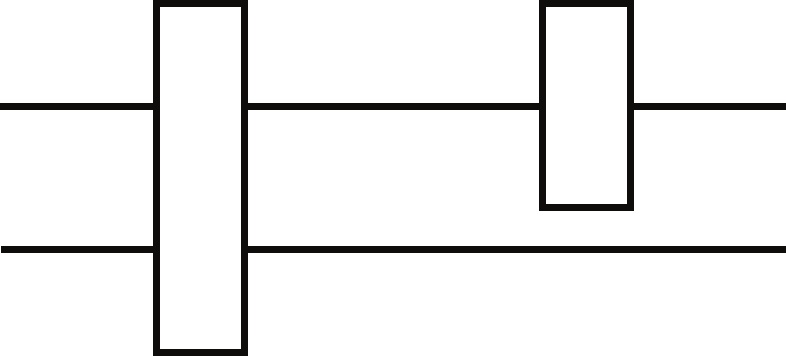}}} 
\quad = \quad \vcenter{\hbox{\includegraphics[scale=0.275]{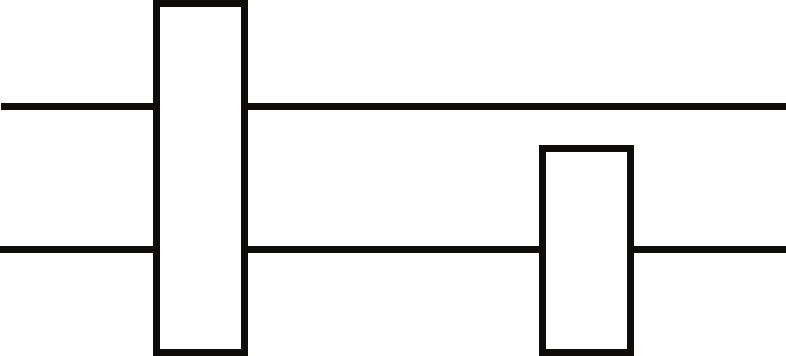}}} 
\quad = \quad \vcenter{\hbox{\includegraphics[scale=0.275]{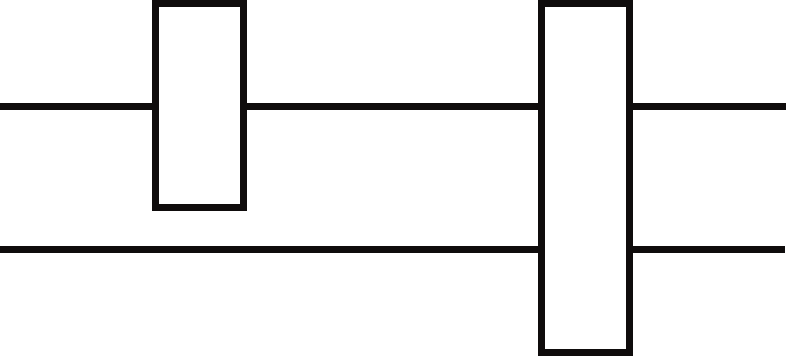}}} 
\quad = \quad \vcenter{\hbox{\includegraphics[scale=0.275]{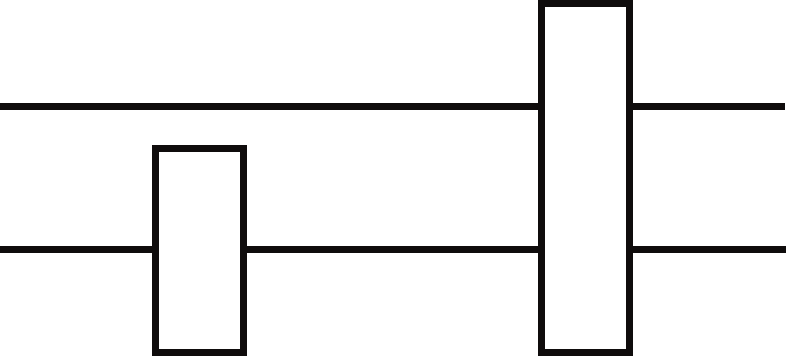}}} 
\quad = \quad \vcenter{\hbox{\includegraphics[scale=0.275]{e-ProjectorBox14.pdf} .}}
\end{align}

Finally, as a tangle in $\TL_s$, the Jones-Wenzl projector $\WJProj\sub{s}$
 equals a linear combination of the link diagrams in $\LD_s$.  
Property~\ref{wj1} implies that the coefficient of the unit $\mathbf{1}_{\TL_s}$ in this linear combination equals one.  Hence, we have
\begin{align}\label{ProjDecomp} 
\WJproj\sub{s} = \mathbf{1}_{\TL_s(\nu)} + \sum_{\substack{T \, \in \, \LD_s, \\  T \, \neq \, \mathbf{1}_{\TL_s(\nu)}}} (\text{coef}_T) \,T,
\end{align}
for some coefficients $\text{coef}_T \in \bC$ (whose values depend on $q \in \bC^\times$). 
In fact, S.~Morrison derived an explicit formula for these coefficients in~\cite{sm}. 
In~\cite[appendix~\red{A}]{fp0}, we give a new, alternative derivation of his formula.

Now we use the Jones-Wenzl projectors 
to generalize concatenation rules (\ref{TLmult1}--\ref{TLmult4}) for $n$-tangles
to concatenation rules for $(\multii, \multiii)$-valenced tangles.
For any complex number $\nu \in \bC$ parameterized by $q \in \bC^\times$ as in~\eqref{fugacity},
and for any ``intermediate'' multiindex 
$\varepsilon = (e_1, e_2, \ldots, e_{\np_\varepsilon}) \in \smash{\{ \OneVec{0} \} \cup \bZpos^\#}$ 
satisfying the condition
\begin{align}\label{SomeConds} 
\max \varepsilon < \ppmin(q) ,
\end{align}
we define a bilinear map 
$\mu_\nu \colon \smash{\TL_\multii^\varepsilon} \times \smash{\TL_\varepsilon^\multiii} \longrightarrow \smash{\TL_\multii^\multiii}$, 
by bilinear extension of the following recipe:
\begin{enumerate}[leftmargin=*, label = $\mu$\arabic*., ref = $\mu$\arabic*] 
\itemcolor{red}
\item \label{WJit1} we concatenate the $(\multii,\varepsilon)$-valenced link diagram $T$ to the $(\varepsilon,\multiii)$-valenced link diagram $U$ from the left, 

\item \label{WJit2} we replace the node of size $e_i$ with a Jones-Wenzl projector box of size $e_i$ for each $i \in \{1,2,\ldots,\np_\varepsilon\}$, and

\item \label{WJit3} we decompose each projector box and, in each resulting term, 
we replace each loop with a multiplicative factor of $\nu$,
and we set each diagram containing a loop link to zero, 
arriving with the $(\multii,\multiii)$-valenced tangle
\begin{align} 
TU := \mu_\nu(T,U) . 
\end{align}
\end{enumerate}
We write the vector space $\smash{\TL_\multii^\multiii}$ as $\TL_\multii^\multiii(\nu)$ to emphasize the chosen value of $\nu$.  
Pictorially, steps~\ref{WJit1} and~\ref{WJit2} are
\begin{align}
& \vcenter{\hbox{\includegraphics[scale=0.275]{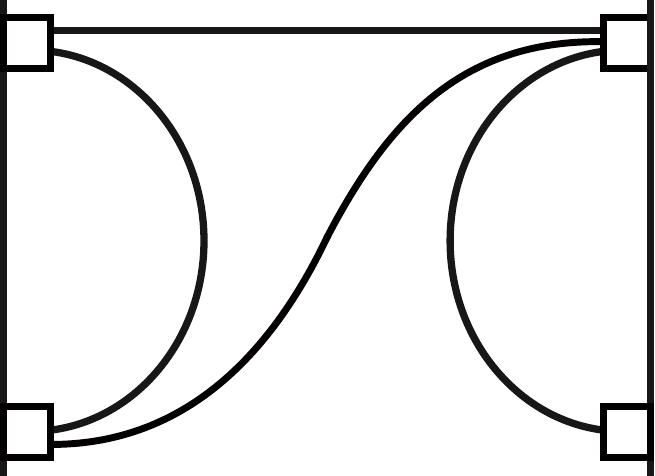}}} \quad 
\vcenter{\hbox{\includegraphics[scale=0.275]{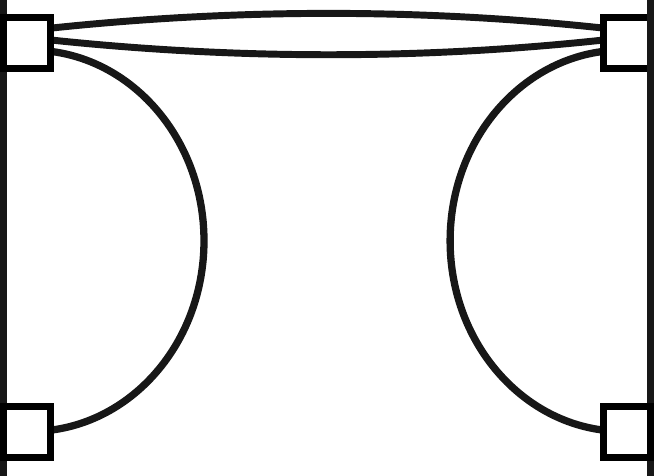}}} \quad := \quad 
\vcenter{\hbox{\includegraphics[scale=0.275]{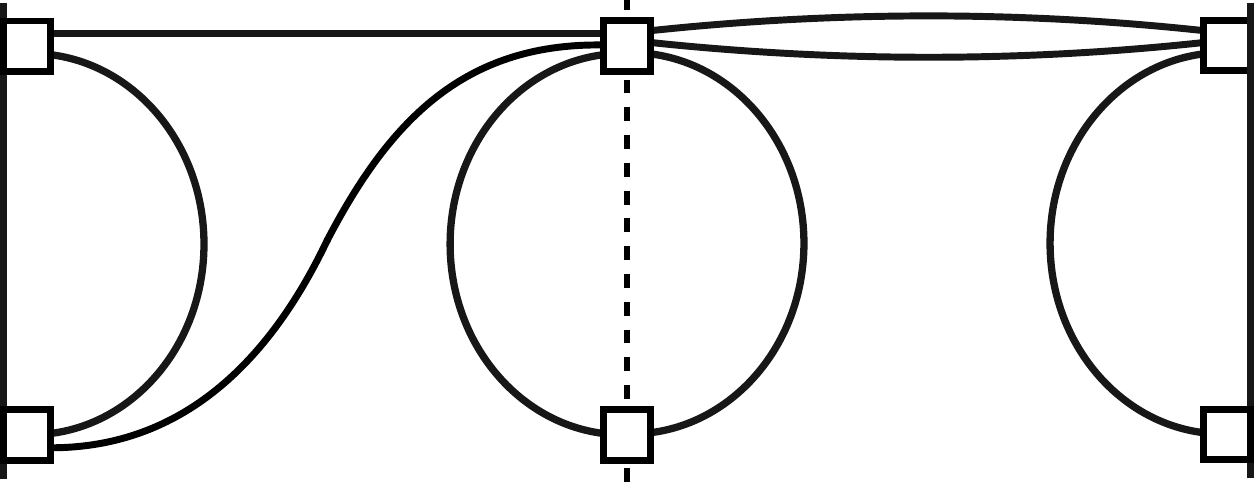}}} \quad = \quad 
\vcenter{\hbox{\includegraphics[scale=0.275]{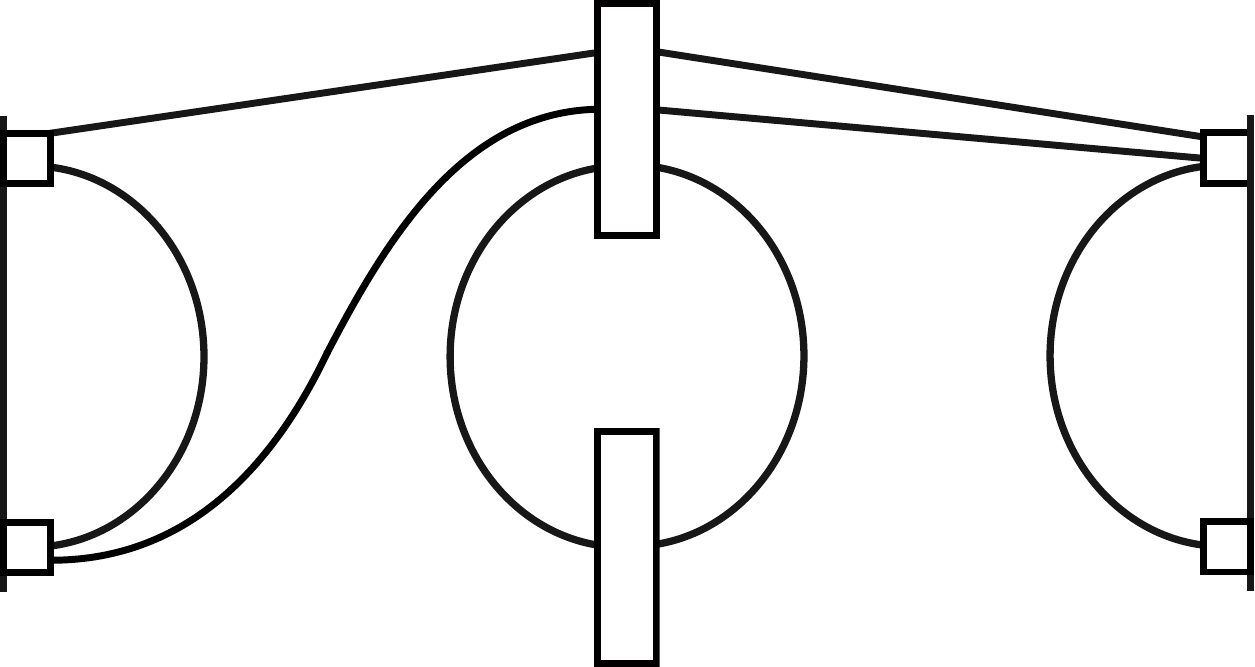} .}}
\end{align}
Followed by step~\ref{WJit3}, the above concatenation evaluates to the following:
\begin{align}
& \quad \vcenter{\hbox{\includegraphics[scale=0.275]{e-WJ_example6_valenced.pdf}}} \quad 
\vcenter{\hbox{\includegraphics[scale=0.275]{e-WJ_example7_valenced.pdf}}} \quad := \quad 
\vcenter{\hbox{\includegraphics[scale=0.275]{e-WJ_example6_and_7_concatenation_halfvalenced.pdf}}} \quad
\overset{\eqref{ProjectorID1}}{=}  \quad 
\vcenter{\hbox{\includegraphics[scale=0.275]{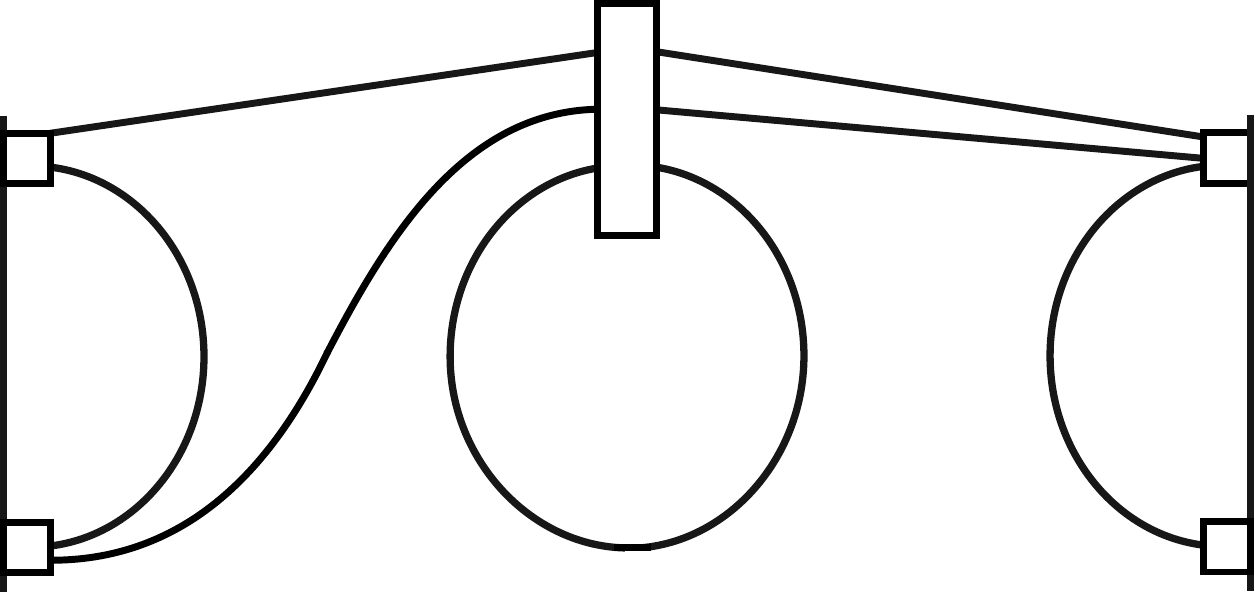}}} \\[1em]
\overset{\eqref{ExampleProj3}}{=} & \quad \vcenter{\hbox{\includegraphics[scale=0.275]{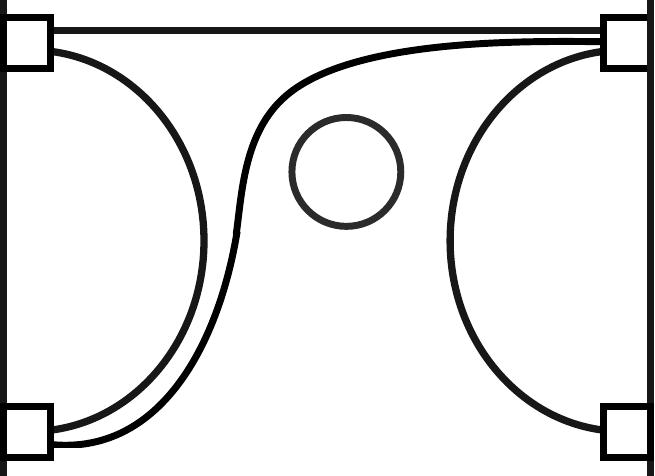}}} 
\quad + \quad \frac{[2]}{[3]}  \,\, \times \,\, \vcenter{\hbox{\includegraphics[scale=0.275]{e-WJ_example6_valenced.pdf}}}
\quad + \quad \frac{[2]}{[3]} \,\, \times \,\, \vcenter{\hbox{\includegraphics[scale=0.275]{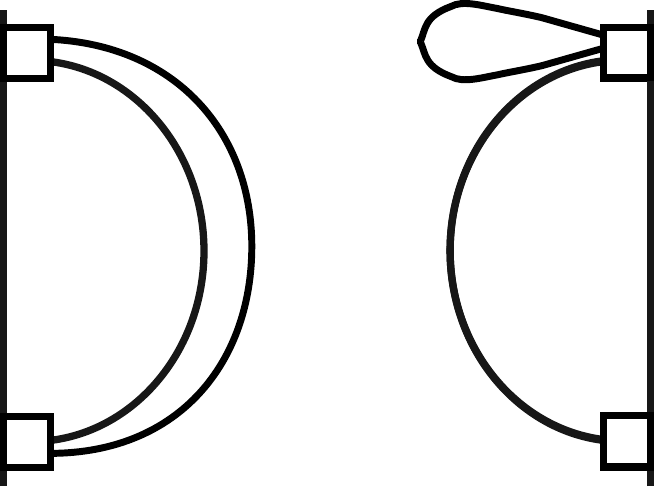}}}
\quad + \quad \frac{2}{[3]} \,\, \times \,\, \vcenter{\hbox{\includegraphics[scale=0.275]{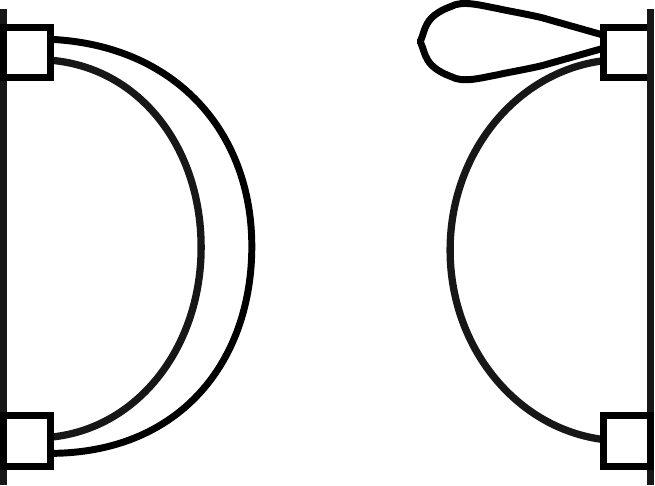}}} \\[1em]
\overset{\eqref{numbertwo}}{=} & \quad 
\frac{[2] (1 - [3])}{[3]} \,\, \times \,\, \vcenter{\hbox{\includegraphics[scale=0.275]{e-WJ_example6_valenced.pdf}}} 
\quad = \quad - \frac{[4]}{[3]} \,\, \times \,\, \vcenter{\hbox{\includegraphics[scale=0.275]{e-WJ_example6_valenced.pdf} ,}}
\end{align}
where we also used the identity $[4] = [2] ([3] - 1)$ (see, e.g.,~item~\ref{QintIDItem} of lemma~\ref{QintIDLemAndSumFormulaLem2} in appendix~\ref{TLRecouplingSect}).

Similarly, we generalize concatenation rules (\ref{loopex}--\ref{turnbackex}) for $n$-tangles
and $(n,s)$-link states to concatenation rules for $(\multii, \multiii)$-valenced tangles 
and $(\multii, s)$-valenced link states. We define a bilinear map 
$\smash{\lambda_\nu\super{s} \colon \TL_\multii^\multiii(\nu) \times \LS_\multiii\super{s} \longrightarrow \LS_\multii\super{s}}$
by bilinear extension of the following recipe: 
\begin{enumerate}[leftmargin=*, label = $\lambda$\arabic*., ref = $\lambda$\arabic*] 
\itemcolor{red}
\item \label{WJLSit1} 
we concatenate the $(\multii,\multiii)$-valenced link diagram $T$ to the $(\multiii,s)$-valenced link pattern $\alpha$
(rotated by $-\pi/2$ radians) from the left, 

\item \label{WJLSit2} 
we replace the node of size $p_j$ with a Jones-Wenzl projector box of size $p_j$ for each $j \in \{1,2,\ldots,\np_\multiii\}$, and

\item \label{WJLSit3} 
we decompose each projector box and, in each resulting term, 
we replace each loop with a multiplicative factor of $\nu$, 
we replace each turn-back path with a multiplicative factor of zero,
and we set each diagram containing a loop link to zero, 
arriving with the $(\multii,s)$-valenced link state
\begin{align} 
T\alpha := \lambda_\nu\super{s}(T,\alpha) . 
\end{align}
We note that the number $s$ of defects is preserved in this concatenation, 
and if $s \notin \DefectSet_\multii$, then we have $T\alpha = 0$.
\end{enumerate}
Pictorially, steps~\ref{WJLSit1} and~\ref{WJLSit2} are
\begin{align}
& \vcenter{\hbox{\includegraphics[scale=0.275]{e-WJ_example1_valenced.pdf}}} \quad 
\vcenter{\hbox{\includegraphics[scale=0.275]{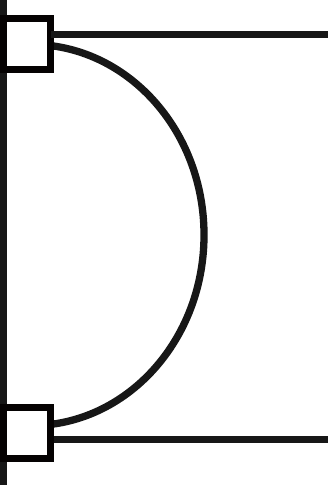}}} \quad := \quad 
\vcenter{\hbox{\includegraphics[scale=0.275]{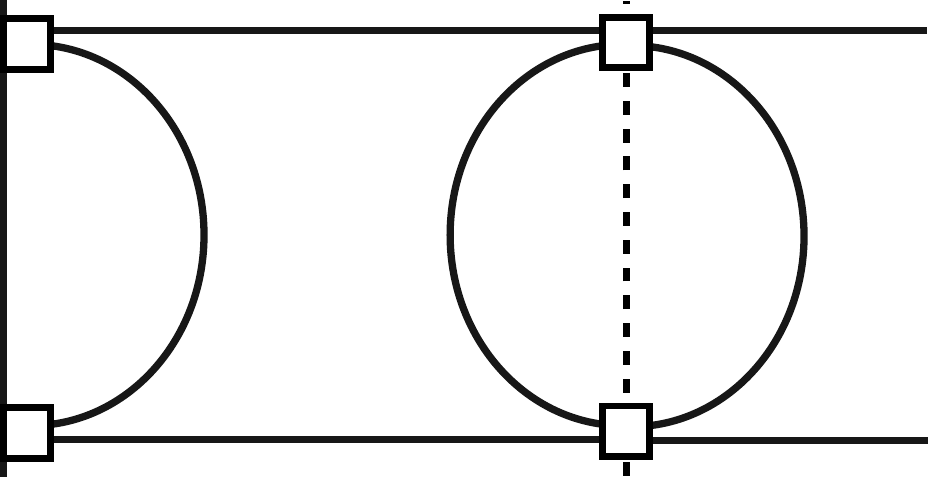}}} \quad = \quad 
\vcenter{\hbox{\includegraphics[scale=0.275]{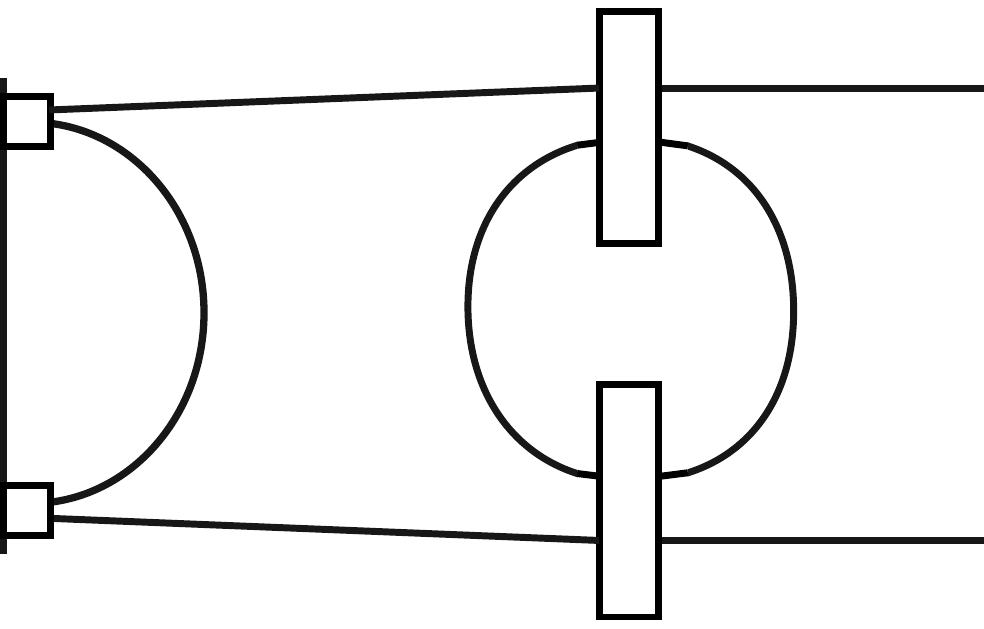} .}}
\end{align}
Followed by step~\ref{WJLSit3}, the above concatenation evaluates to the following:
\begin{align}
& \quad \vcenter{\hbox{\includegraphics[scale=0.275]{e-WJ_example1_valenced.pdf}}} \quad 
\vcenter{\hbox{\includegraphics[scale=0.275]{e-LinkPattern1_valenced_rotated.pdf}}} \quad := \quad 
\vcenter{\hbox{\includegraphics[scale=0.275]{e-WJ_example1_action_halfvalenced.pdf}}}  \\ 
\overset{\eqref{ExampleProj2}}{=} & \quad 
\raisebox{-12pt}{\includegraphics[scale=0.275]{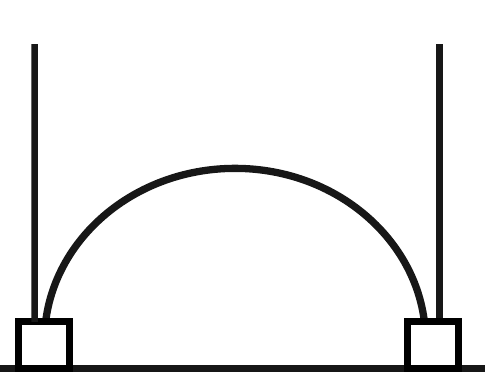}}
\quad + \quad \frac{2}{[2]} \,\, \times \,\, \vcenter{\hbox{\includegraphics[scale=0.275]{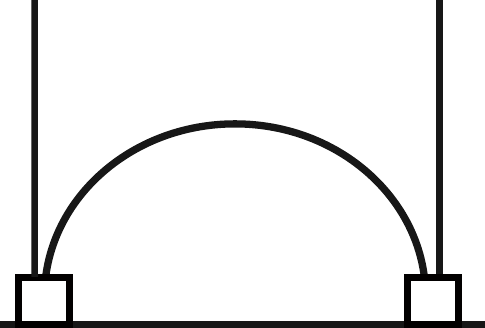}}} 
\quad + \quad \frac{1}{[2]^2} \,\, \times \,\, 
\raisebox{-12pt}{\includegraphics[scale=0.275]{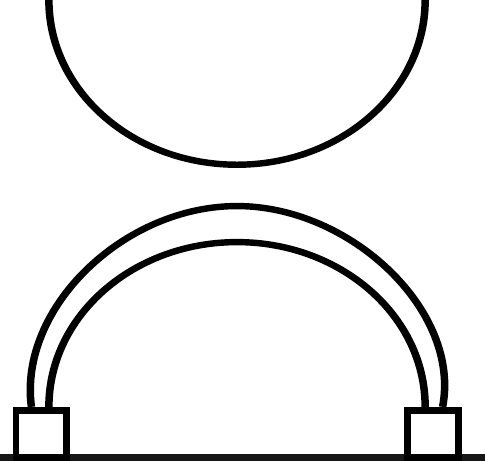}} \quad
\overset{\eqref{numbertwo}}{=} \quad \left( [2] + \frac{2}{[2]} \right) \,\, \times \,\,
\vcenter{\hbox{\includegraphics[scale=0.275]{e-LinkPattern1_valenced.pdf} .}} 
\end{align}

\bigskip

Valenced tangles arise as morphisms of a category $\TL(\nu)$, whose object class and morphism class are respectively 
\begin{align} 
\label{cateob}
\text{Ob} \, \TL(\nu) &= \big\{\multii \in \{ \vec{0} \} \cup \bZpos^\# \, \big| \, \max \multii < \ppmin(q) \big\} , \\
\label{catehom}
\Hom  \TL(\nu) &= \big\{ \TL_\multii^\multiii(\nu) \, \big| \, \text{$\multii, \multiii \in \text{Ob} \, \TL(\nu)$ with $\Summed_\multii + \Summed_\multiii = 0 \Mod 2$} \big\} .
\end{align}
The source and target associated with the tangle $T \in \smash{\TL_\multii^\multiii}(\nu)$ are the objects $\multiii$ and $\multii$ respectively, 
and the identity morphism associated with the object $\multii$ is the unit~\eqref{ValencedCompProjIntro} of 
the valenced Temperley-Lieb algebra $\TL_\multii(\nu)$.
The composition of two morphisms $T, U \in \Hom \TL(\nu)$ is given by $\mu_\nu(T,U)$, 
defined in recipe~(\ref{WJit1}--\ref{WJit3}), and thus depends on the fugacity parameter $\nu \in \bC$.
Finally, $\TL(\nu)$ is in fact a monoidal category, with identity object $\smash{\OneVec{0}}$ and tensor product 
\begin{align} 
\multii \otimes \multiii := \multii \oplus \multiii \qquad \qquad \text{and} \qquad \qquad
T \otimes U := \quad \vcenter{\hbox{\includegraphics[scale=0.275]{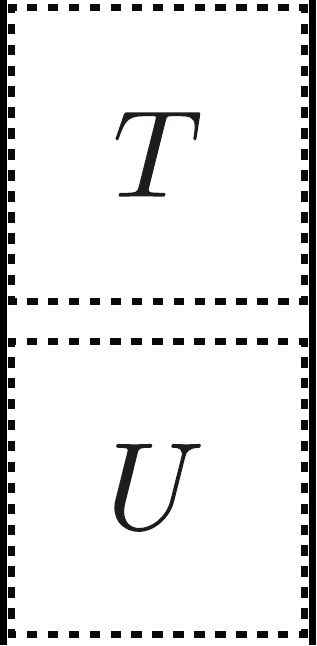} ,}}
\end{align} 
for all objects $\multii, \multiii \in \text{Ob} \, \TL(\nu)$ and morphisms $T, U \in \Hom  \TL(\nu)$,
where $\multii \oplus \multiii$ is the concatenation~\eqref{concatenateMultiindices} of multiindices. 
Analogously, it is also sometimes useful to identify the tensor product $\alpha \otimes \beta$ of 
two valenced link states $\alpha \in \smash{\LS_\multii\super{s}}$ and 
$\beta \in \smash{\LS_\multiii\super{t}}$ with the link state $\gamma \in \smash{\LS_{\multii \oplus \multiii}\super{s + t}}$ 
obtained by concatenating $\alpha$ to the left of $\beta$:
\begin{align} 
\alpha \otimes \beta 
\quad := \quad \vcenter{\hbox{\includegraphics[scale=0.275]{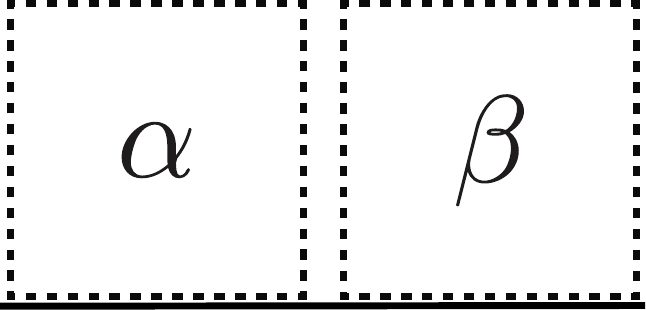} .}}
\end{align}

In appendix~\ref{CategorySect}, we consider a 
subcategory of $\TL(\nu)$, which we call the ``Temperley-Lieb category"~\cite{vt, ck, gl2}. 
We give a minimal set of ``generators" and ``relations" for morphisms of this category.

\subsection{Valenced Temperley-Lieb algebra} \label{ValTLdefSec}

%

In the special case that $\multiii = \multii$, the vector space 
$\smash{\TL_\multii^\multii}(\nu) =: \TL_\multii(\nu)$ is the \emph{valenced Temperley-Lieb algebra}, already 
appearing in section~\ref{MainResultSec}. A generic valenced tangle in $\TL_\multii(\nu)$ is of the form
\begin{align}\label{TLform} 
\vcenter{\hbox{\includegraphics[scale=0.275]{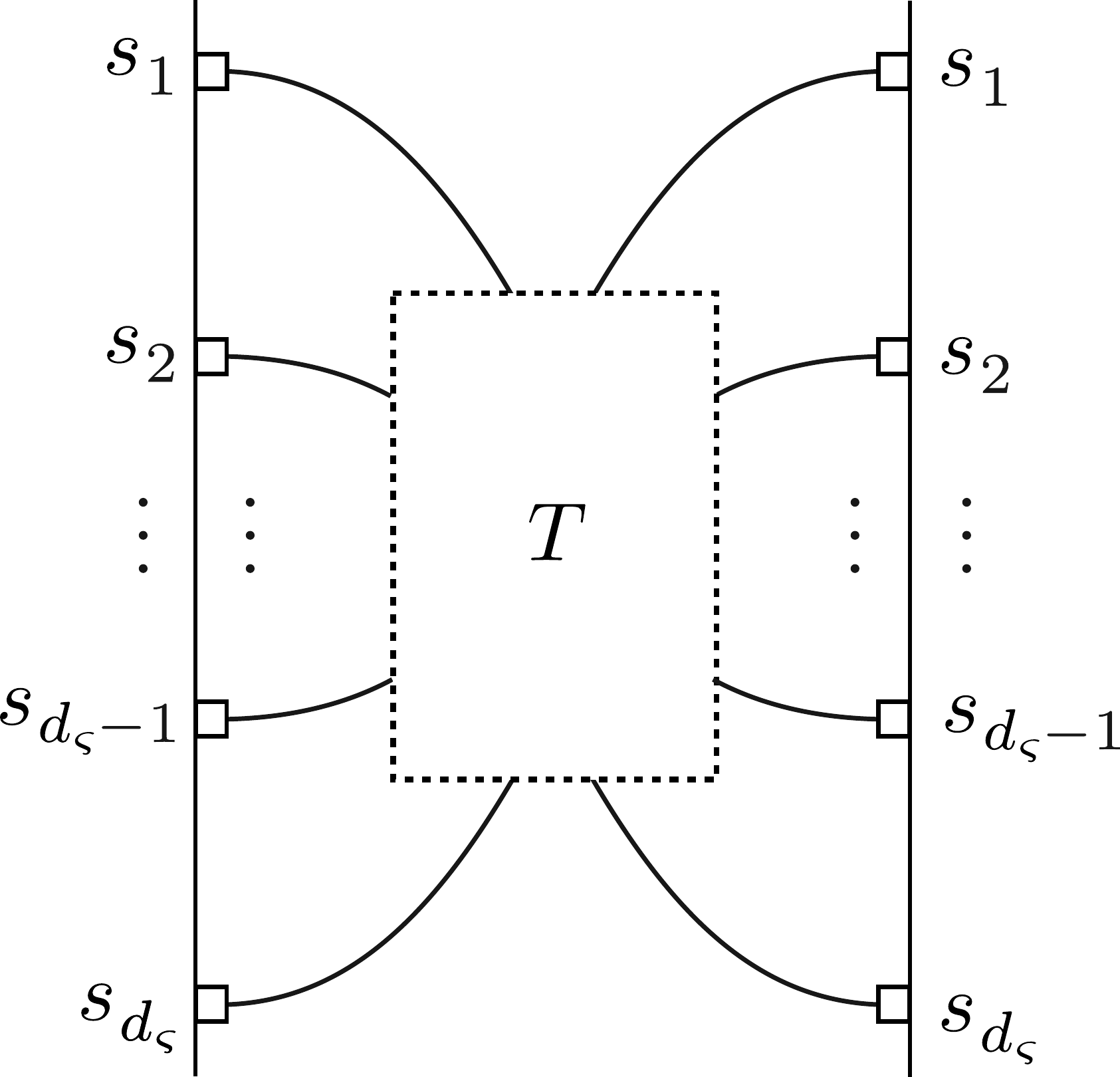} ,}}
\end{align}
for some ordinary tangle $T \in \TL_{\Summed_\multii}(\nu)$.
We call elements of $\TL_\multii(\nu)$ \emph{$\multii$-valenced tangles}. 
If $T$ is an $\Summed_\multii$-link diagram such that~\eqref{TLform} does not vanish (i.e., does not contain loop links),
then we call~\eqref{TLform} a \emph{$\multii$-valenced link diagram.}

$\TL_\multii(\nu)$ is an associative algebra, with multiplication 
given by the bilinear map $\mu_\nu$ defined in recipe~\ref{WJit1}--\ref{WJit3}, and 
\begin{align} \label{ValencedCompProj}
\mathbf{1}_{\TL_\multii} \quad = \quad 
\vcenter{\hbox{\includegraphics[scale=0.275]{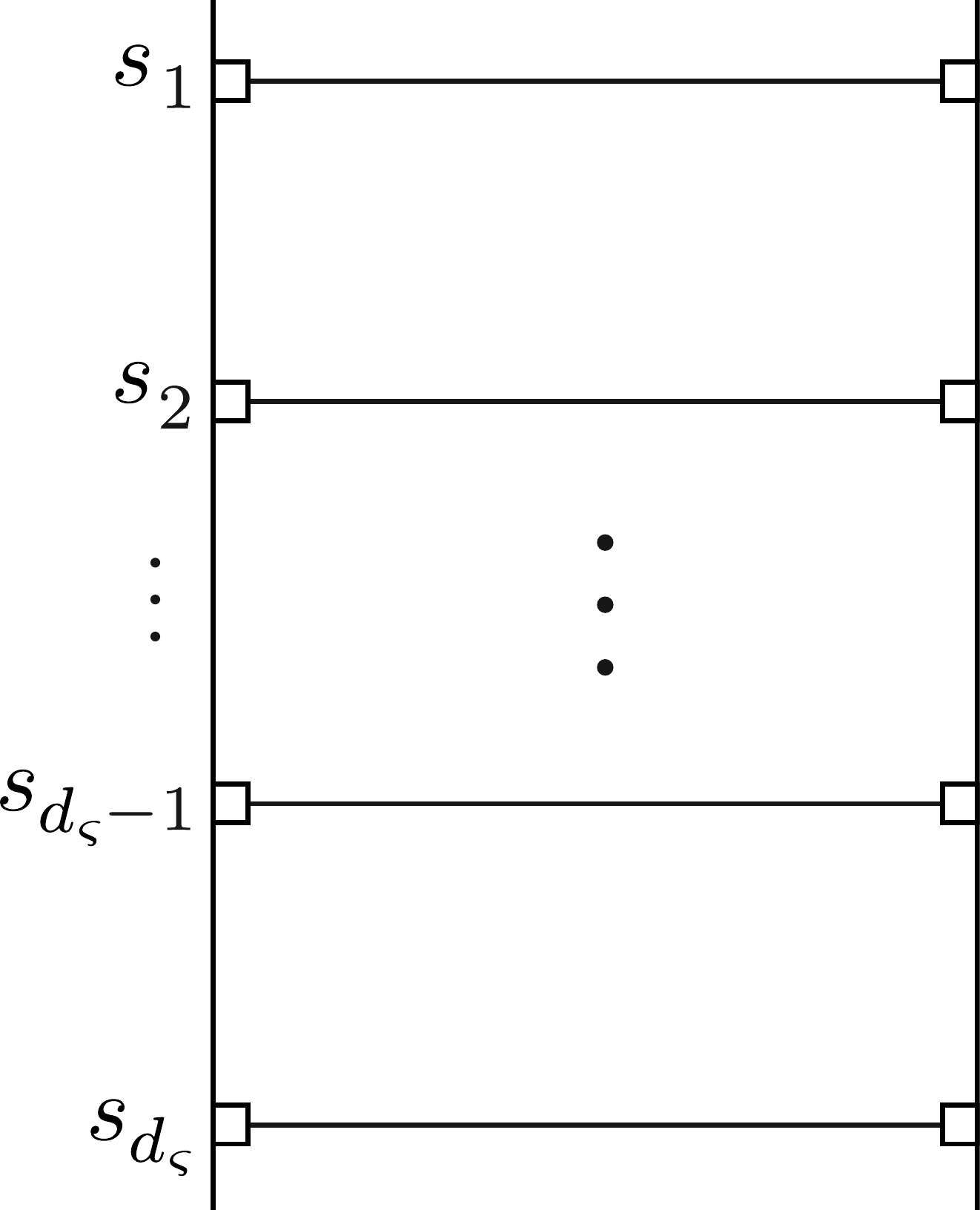}}}
\hphantom{\mathbf{1}_{\TL_\multii} \quad = \quad}
\end{align} 
is the unit of $\TL_\multii(\nu)$. From~\eqref{SomeConds}, we see that 
the multiplication map of $\TL_\multii(\nu)$ is defined only if we have
\begin{align}  \label{maxllppmin}
\max \multii < \ppmin(q),
\end{align}
because if condition~\eqref{maxllppmin} is violated, then the Jones-Wenzl projector boxes cease to be well-defined.
However, with explicit generators for $\TL_\multii(\nu)$, e.g., from proposition~\ref{GeneratorPropTwo} below,
and a complete set of relations for these generators, 
the algebra structure on $\TL_\multii(\nu)$ could be defined more generally
(see also~\cite{fp0}).

\begin{remark}
\textnormal{
One can check that the valenced Temperley-Lieb algebra $\TL_\multii(\nu)$ is a ``cellular algebra''~\cite{gl}.
This follows from~\cite[proposition~\red{2.4}]{fp0} and the isomorphism of item~\ref{MapIt1} 
of corollary~\ref{AnIsoCor} from appendix~\ref{AppWJ}.
}
\end{remark}

As a complex vector space, $\TL_\multii(\nu)$ has the basis
\begin{align} 
\LD_\multii := \LD_\multii^\multii . 
\end{align}
By setting $\multii = \multiii$ in corollary~\ref{WJDimLem1},  
we obtain the following expressions for the dimension of this algebra:
\begin{align} \label{Dim78} 
\Dim_{\multii \oplus \smash{\tilde{\multii}}}\super{0} \underset{\eqref{LSDim2}}{\overset{\eqref{Dim56}}{=}}  \dim \TL_\multii(\nu) 
\overset{\eqref{Dim56}}{=} \sum_{s \, \in \, \DefectSet_\multii} \big(\dim \LS_\multii\super{s}\big)^2.
\end{align}

%

\bigskip

To end this section, 
we present two minimal generating sets for the valenced Temperley-Lieb algebra, 
assuming that $\Summed_\multii < \ppmin(q)$. 
In the second generating set~\eqref{MasterDiagramsTL}, we use the definition of a ``closed three-vertex,'' 
given by~\eqref{3vertex1} in section~\ref{ConformalBlocksSect}.
This result is crucial in our work~\cite{fp3} for obtaining 
a generalization of the quantum Schur-Weyl duality, a result essential to determining unique 
monodromy-invariant CFT correlation functions in~\cite{fp1}.

\begin{restatable}{prop}{GeneratorPropTwo} \label{GeneratorPropTwo} 
Suppose $\Summed_\multii < \ppmin(q)$.
Then the following hold:
\begin{enumerate}
\itemcolor{red}

\item \label{GeneratorPropItem1}
The unit $\mathbf{1}_{\TL_\multii}$~\eqref{ValencedCompProj} together with the 
$\multii$-valenced link diagrams
\begin{align}\label{ValGenerators} 
\vcenter{\hbox{\includegraphics[scale=0.275]{e-Generators0_valenced.pdf} ,}} 
\end{align} 
with $i \in \{1,2,\ldots,\np_\multii-1\}$, 
forms a minimal generating set for the valenced Temperley-Lieb algebra $\TL_\multii(\nu)$.

\item \label{GeneratorPropItem2}
Alternatively $($using notation~\eqref{3vertex1}$)$, 
the collection of all $\multii$-valenced tangles of the form
\begin{align} \label{MasterDiagramsTL} 
\vcenter{\hbox{\includegraphics[scale=0.275]{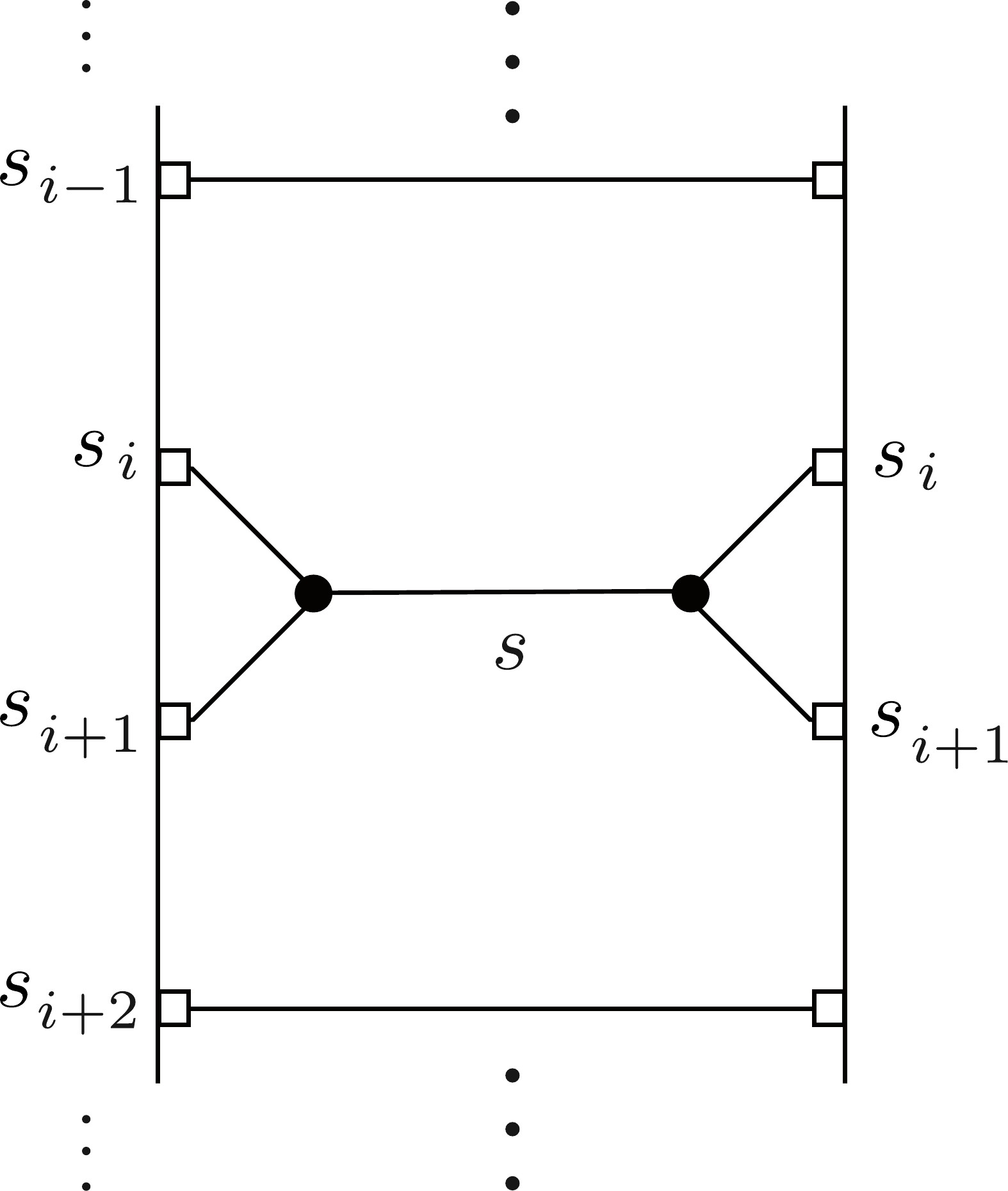} ,}}
\end{align}
with $s \in \DefectSet\sub{\sIndex_i,\sIndex_{i+1}}$ and $i \in \{ 1, 2, \ldots, \np_\multii - 1 \}$, 
forms a minimal generating set for $\TL_\multii(\nu)$.
\end{enumerate}
%
\end{restatable}

\begin{proof}
The claim follows by combining~\cite[theorem~\red{1.1}]{fp0} with item~\ref{MapIt1} of corollary~\ref{AnIsoCor} 
from appendix~\ref{AppWJ}.
\end{proof}

The unit~\eqref{ValencedCompProj} of $\TL_\multii(\nu)$
is obtained from generators~\eqref{MasterDiagramsTL} via the 
relation~\cite[equation~(\red{3.1})]{fp0},~\cite{kl}
\begin{align} \label{AllProjesSumToOneTL}
\vcenter{\hbox{\includegraphics[scale=0.275]{e-CompositeProjector_valenced.pdf}}} 
\quad = \quad 
\sum_{s \, \in \, \DefectSet\sub{\sIndex_i,\sIndex_{i+1}}} \frac{(-1)^s [s+1]}{\ThetaNet(\sIndex_i,\sIndex_{i+1},s)} \,\, \times \,\,
\vcenter{\hbox{\includegraphics[scale=0.275]{e-Generators_3Vertex_valenced.pdf} ,}}
\end{align}
for any $i \in \{ 1, 2, \ldots, \np_\multii - 1 \}$.
In~\cite{fp3}, we prove that each diagram~\eqref{MasterDiagramsTL} 
equals a nonzero multiple of a submodule projector 
in a tensor product representation of the Hopf algebra $U_q(\mathfrak{sl}_2)$ 
(with multiplicative constant as in the above sum).
Relation~\eqref{AllProjesSumToOneTL} says that summing over all of these projectors gives the identity operator.
Another way to view this relation is a decomposition of the unit~\eqref{ValencedCompProj} of $\TL_\multii(\nu)$
into a sum of orthogonal (but not central) idempotents: indeed by identity~\eqref{LoopErasure1} from appendix~\ref{TLRecouplingSect}, we have
\begin{align} \label{orthogonalidemTL}
\hspace*{-5mm}
\vcenter{\hbox{\includegraphics[scale=0.275]{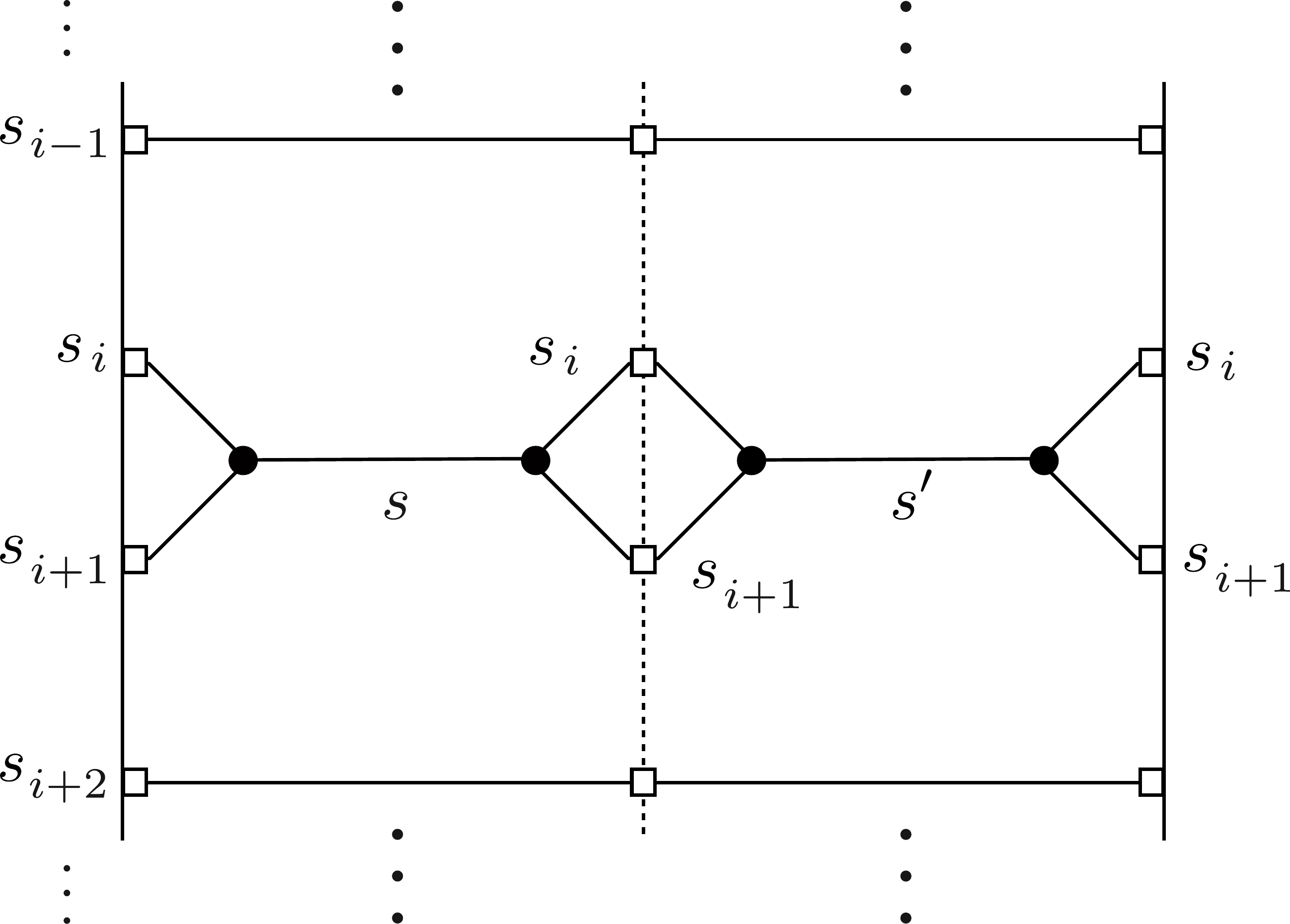}}}
\overset{\eqref{LoopErasure1}}{=} 
\; \delta_{s,s'} \frac{\ThetaNet(\sIndex_i,\sIndex_{i+1},s)}{(-1)^s [s+1]} \,\, \times \,\,
\vcenter{\hbox{\includegraphics[scale=0.275]{e-Generators_3Vertex_valenced.pdf} .}} 
\end{align}

\begin{conj} \label{GeneratorConj} 
Item~\ref{GeneratorPropItem1} in proposition~\ref{GeneratorPropTwo} holds whenever $\max \multii < \ppmin(q)$. 
\end{conj}

With a complete set of generators and relations for the valenced Temperley-Lieb algebra $\TL_\multii(\nu)$,
one could extend its definition to the range $\max \multii \geq \ppmin(q)$.
Such a more general definition might be useful in applications, e.g., for logarithmic conformal field theories 
and critical planar statistical mechanics models, 
where the assumption that $\max \multii < \ppmin(q)$ can be violated. 
We leave this generalization for future work, pointing out that a special case appears in~\cite{mrr}
and that in~\cite{fp0}, we find relations for the generators in other special cases.

\section{Standard modules} \label{StdModulesSect}

In this section, we begin to investigate the representation theory of the valenced Temperley-Lieb algebra $\TL_\multii(\nu)$.
Fixing terminology, we define a \emph{representation} of an associative algebra $\mathsf{A}$ to be a homomorphism 
$\rho \colon \mathsf{A} \longrightarrow \End \mathsf{V}$ of algebras mapping $\mathsf{A}$ in the space $\End \mathsf{V}$ of endomorphisms
of some vector space $\mathsf{V}$, and we call the pair $(\mathsf{V}, \rho)$ an $\mathsf{A}$-\emph{module}.
We shorten the notation by writing $a v := \rho(a)(v)$ for all $a \in \mathsf{A}$ and $v \in \mathsf{V}$.
Suppose $\mathsf{V}$ is not zero.
We say that $\mathsf{V}$ is \emph{simple}, and $\rho$ is \emph{irreducible}, if $\mathsf{V}$ contains 
no non-zero proper submodules. Also, we say that $\mathsf{V}$ is \emph{semisimple}, and $\rho$ is \emph{completely reducible}, 
if $\mathsf{V}$ can be decomposed into a direct sum of simple submodules. 
Finally, we say that $\mathsf{V}$ and $\rho$ are \emph{indecomposable} 
if $\mathsf{V}$ cannot be decomposed into a direct sum of two non-zero submodules. 

%

For each $s \in \DefectSet_\multii$, the bilinear map $\smash{\lambda_\nu\super{s} \colon \TL_\multii(\nu) \times \LS_\multii\super{s} \longrightarrow \LS_\multii\super{s}}$,
defined in recipe~\ref{WJLSit1}--\ref{WJLSit3} in section~\ref{ValencedCompositionSec},
endows the space $\smash{\LS_\multii\super{s}}$ of $(\multii,s)$-valenced link states
with a $\TL_\multii(\nu)$-module structure. Hence, we call $\smash{\LS_\multii\super{s}}$ a \emph{valenced standard module}, 
or simply a ``standard module."  Also, we call the direct sum $\LS_\multii$ from~\eqref{LSDirSum2}, now a $\TL_\multii(\nu)$-module too,
the \emph{valenced link state module} of $\TL_\multii(\nu)$, or simply the ``link state module."

The standard modules $\smash{\LS_\multii\super{s}}$ and their radicals play key roles in the representation theory of $\TL_\multii(\nu)$.
Following the approach of~\cite{jm79, bw,gl2,rsa}, we study 
them via a natural bilinear form
$\BiForm{\cdot}{\cdot} \colon \smash{\LS_\multii\super{s}} \times \smash{\LS_\multii\super{s}} \longrightarrow \bC$, 
that we define in section~\ref{BilinFormSec}.
In the spirit of~\cite{rsa}, we take a constructive approach not relying on the general theory 
of cellular algebras~\cite{gl2} of J.~Graham and G.~Lehrer.  
In the language of~\cite{gl2}, the standard modules are ``cell modules''
and the bilinear form can be seen to be associated to a cellular basis for $\TL_\multii(\nu)$,
see~\cite[section~\red{2}]{fp0}.

A key result in this section is proposition~\ref{GenLem2} in section~\ref{LinkStateModSect}, which
says that if the bilinear form $\BiForm{\cdot}{\cdot}$ is not identically zero on it, then
the standard module $\smash{\LS_\multii\super{s}}$ is indecomposable and the radical
\begin{align}\label{radLns}
\rad \LS_\multii\super{s} := \big\{ \alpha \in \LS_\multii\super{s} \,\big|\, \text{$\BiForm{\alpha}{\beta} = 0$ for all $\beta \in \LS_\multii\super{s}$} \big\} 
\end{align}
is the maximal proper submodule of $\smash{\LS_\multii\super{s}}$. In particular, 
$\smash{\LS_\multii\super{s}}$ is simple if and only if its radical is trivial, and 
otherwise, its quotient by its radical is simple.  
This is the first step to determining all of the simple modules of $\TL_\multii(\nu)$. 
Furthermore, 
corollary~\ref{nonisoCor2} shows that the nontrivial quotient modules are all non-isomorphic.

Another important result is corollary~\ref{PreFaithfulCor} in section~\ref{FaithfulSect3}, which says that 
the link state representation of the valenced Temperley-Lieb algebra $\TL_\multii(\nu)$ on the link state module $\LS_\multii$ 
is faithful if and only if the radical of $\LS_\multii$ is trivial.

Most of the results stated in this section depend on properties of the radical~\eqref{radLns}. 
To understand the scope of these results, we completely and explicitly determine these radicals in section~\ref{RadicalSect} (proposition~\ref{BigTailLem} 
and theorem~\ref{BigTailLem2}).

\subsection{Networks and the link state bilinear form} \label{BilinFormSec}

We first define the link state bilinear form for the special case when $\multii = \OneVec{n}$ for some $n \in \bZnn$.
For this purpose, we introduce the notion of a \emph{network}: a collection of nonintersecting, non-self-intersecting planar loops and paths within a rectangle. 
A path in a network can be a \emph{through-path}, which is a curve that respectively enters and exits the network at the bottom and top sides of the rectangle, or
a \emph{turn-back path}, which enters and exits the network at the same side of the rectangle, either top or bottom:
\begin{align}
\vcenter{\hbox{\includegraphics[scale=0.275]{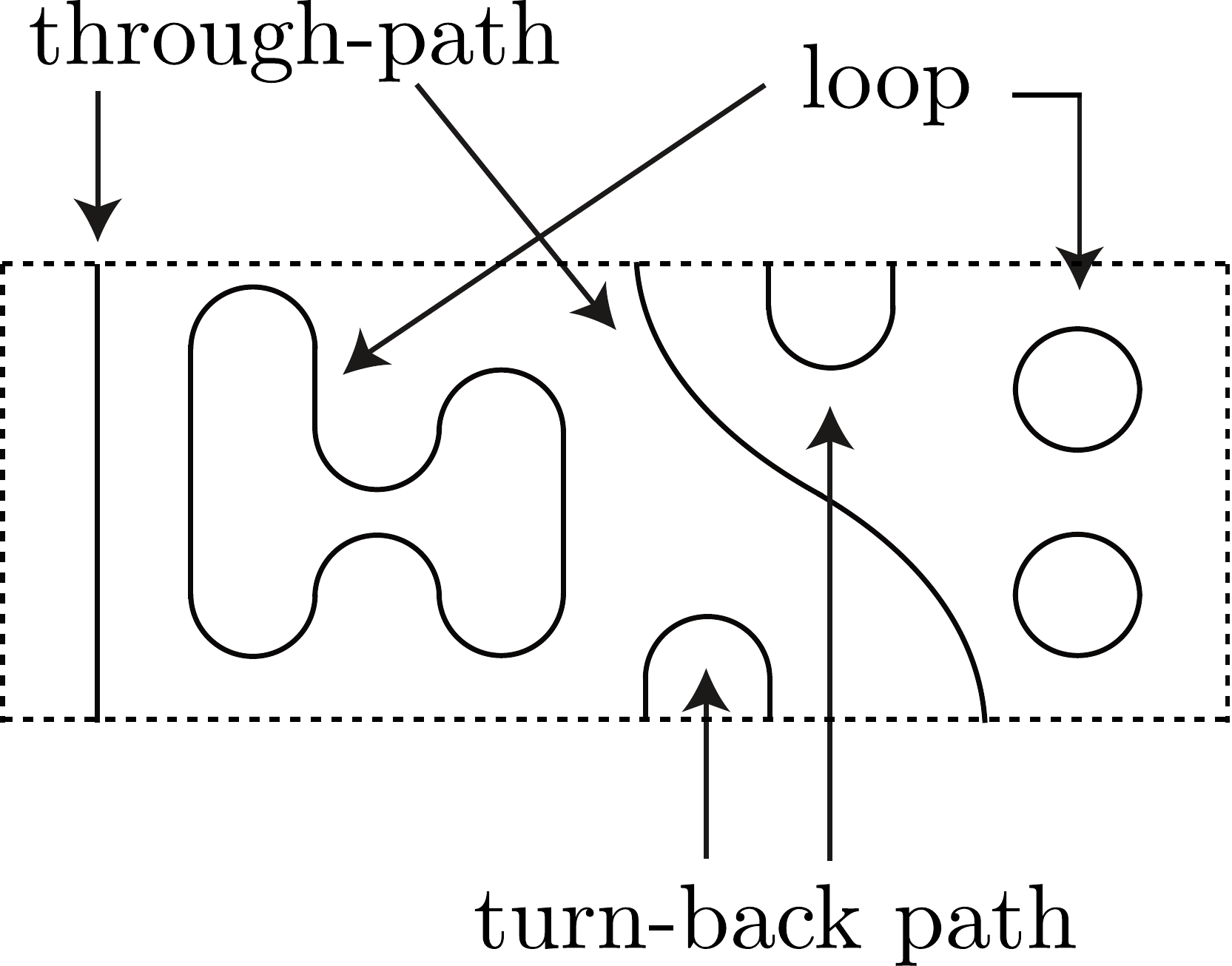} .}}
\end{align}

We assign all loops, through-paths, and turn-back paths in $T$ the following weights in $\bC$:
\begin{alignat}{7} 
\label{LoopWeight} 
& \text{loop weight (fugacity):} \qquad \qquad 
& \vcenter{\hbox{\includegraphics[scale=0.275]{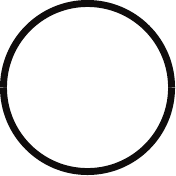}}}  \qquad & \text{and} \qquad  
& \raisebox{-18pt}{\includegraphics[scale=0.275]{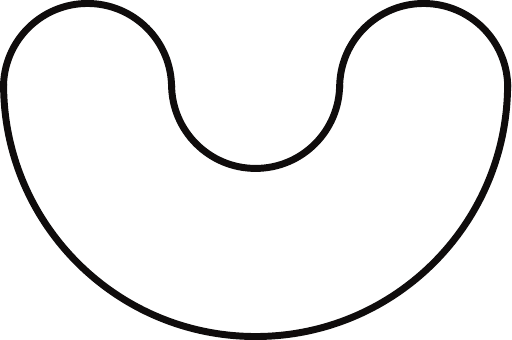}}  \qquad & \text{and} \qquad
& \raisebox{-11pt}{\includegraphics[scale=0.275]{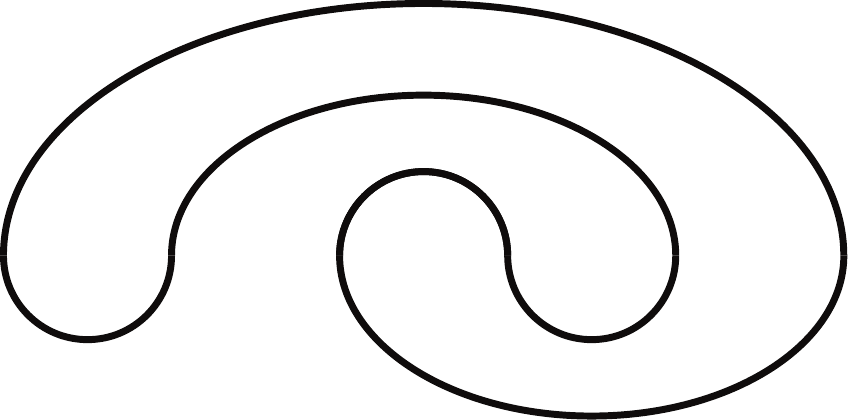}}  \qquad & \text{etc.}  
\quad = \quad \nu , \\[1em]
\label{ThroughPathWeight}
& \text{through-path weight:} \qquad \qquad
& \vcenter{\hbox{\hspace*{-4mm} \includegraphics[scale=0.275]{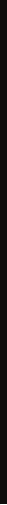}}}  \qquad & \text{and} \qquad
& \vcenter{\hbox{\includegraphics[scale=0.275]{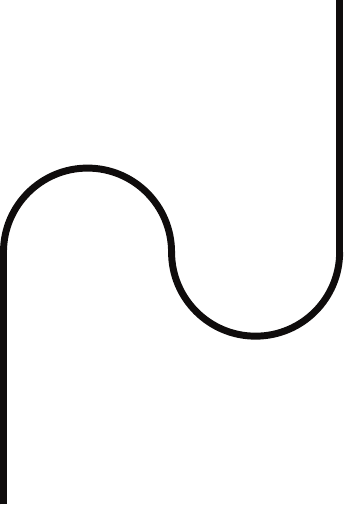} \hspace*{2mm}}}  \qquad & \text{and} \qquad
& \vcenter{\hbox{\includegraphics[scale=0.275]{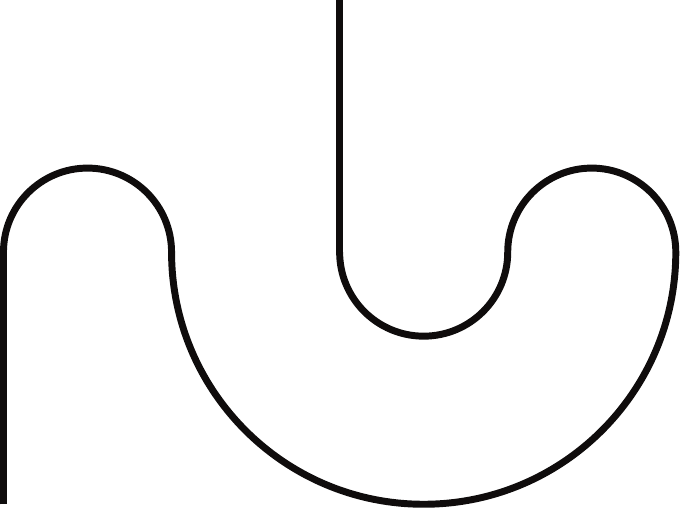}\hspace*{3mm}}}  \qquad & \text{etc.} 
\quad = \quad 1 , \\[1em]
\label{TurnBack0}
& \text{turn-back path weight:} \qquad \qquad
& \raisebox{-5pt}{\includegraphics[scale=0.275]{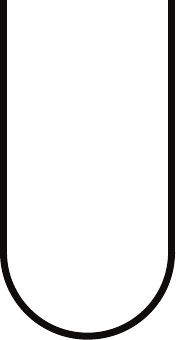}}  \qquad & \text{and} \qquad
& \raisebox{-5pt}{\includegraphics[scale=0.275]{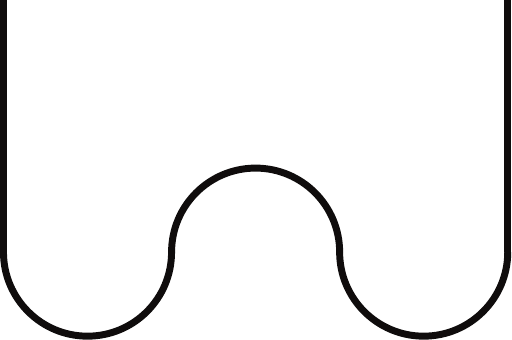}}  \qquad & \text{and} \qquad
& \raisebox{-19pt}{\includegraphics[scale=0.275]{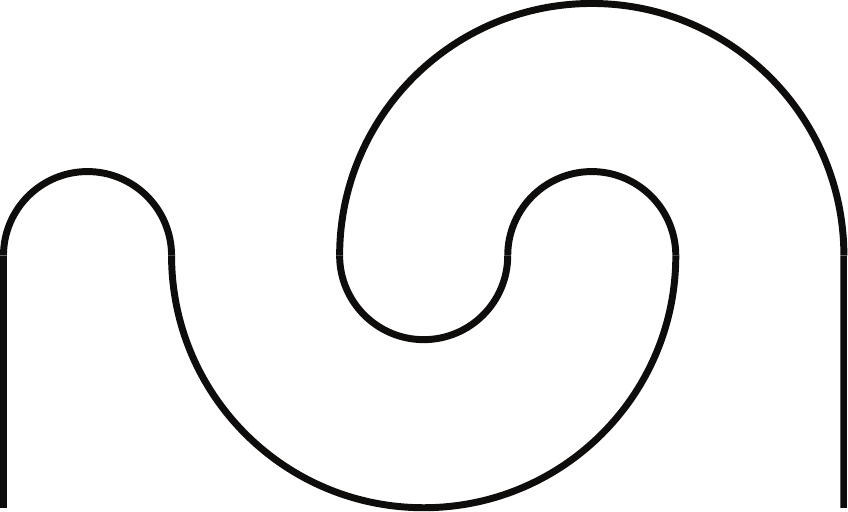}}  \qquad & \text{etc.} 
\quad = \quad 0 .
\end{alignat}
(We note that rules~(\ref{LoopWeight},~\ref{TurnBack0}) are also used in diagram concatenation, 
as discussed in section~\ref{ValencedCompositionSec}.)
With ``\# loops in $T$" equaling the number of loops in the network $T$, 
we define the \emph{evaluation of $T$} as the complex number
\begin{align} \label{evT} 
(\, T \,) &:= \prod \{ \text{the weights of all connected components in the network $T$} \} \\
\label{evT2} & = 
\begin{cases} 
\nu^{\textnormal{\# loops in $T$}} , 
& \textnormal{if the network $T$ has no turn-back path} , \\ 
0 , & \textnormal{if the network $T$ has a turn-back path.}
\end{cases}
\end{align}

Now, using the notion of a network and its evaluation, we define a bilinear form on the link state module $\LS_n$. 
For two link patterns 
$\alpha,\beta \in \LP_n$, we horizontally reflect $\alpha$ so it is upside down, we concatenate it to $\beta$ from below, and delete the overlapping 
horizontal lines of $\alpha$ and $\beta$.  The resulting diagram is a network $\alpha \BarAction \beta$.  For instance, we have
\begin{align} 
\label{LinkStateRule} 
& \alpha \; = \; \raisebox{1pt}{\includegraphics[scale=0.275]{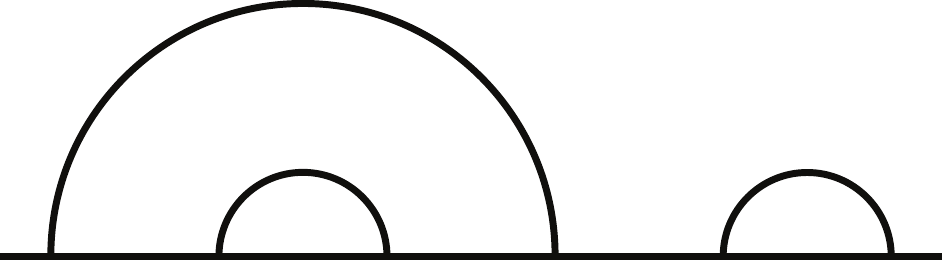}} \; , \qquad 
\beta \; = \; \raisebox{1pt}{\includegraphics[scale=0.275]{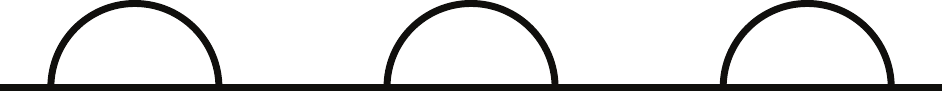}}
\qquad \qquad \Longrightarrow \qquad \qquad
\alpha \BarAction \beta \; = \;
\raisebox{-19pt}{\includegraphics[scale=0.275]{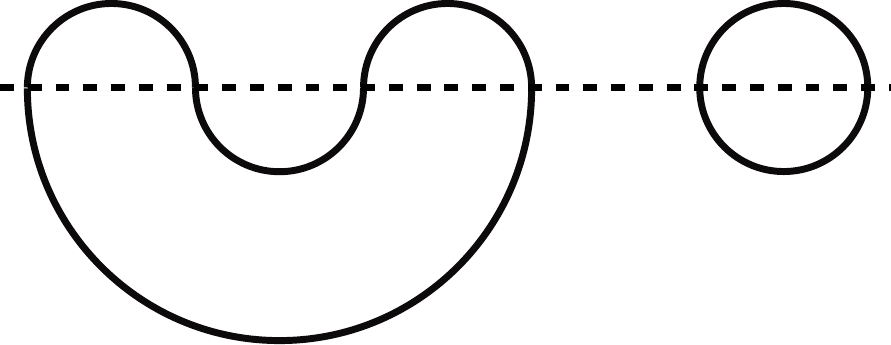}} \; , \\[.3em]
\label{LinkStateRuleDefect2}
& \alpha \; = \; \raisebox{1pt}{\includegraphics[scale=0.275]{e-Connectivities9.pdf}} \; , \qquad 
\beta \; = \; \raisebox{1pt}{\includegraphics[scale=0.275]{e-Connectivities6.pdf}}
\qquad \qquad \Longrightarrow \qquad \qquad
\alpha \BarAction \beta \; = \;
\vcenter{\hbox{\includegraphics[scale=0.275]{e-Connectivities10.pdf}}} \; , \\[1em]
\label{LinkStateRuleDefect1}
& \alpha \; = \; \raisebox{1pt}{\includegraphics[scale=0.275]{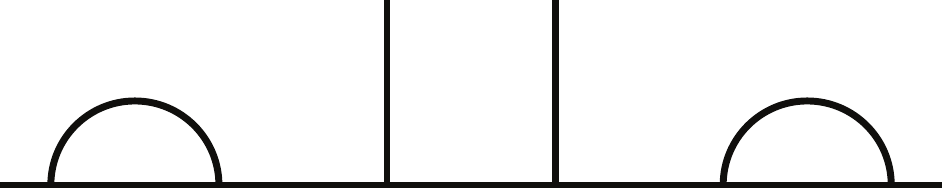}} \; , \qquad 
\beta \; = \; \raisebox{1pt}{\includegraphics[scale=0.275]{e-Connectivities6.pdf}} 
\qquad \qquad \Longrightarrow \qquad \qquad
\alpha \BarAction \beta \; = \;
\vcenter{\hbox{\includegraphics[scale=0.275]{e-Connectivities8.pdf}}} \; .
\end{align}
Then we define the \emph{link state bilinear form} 
$\BiForm{\cdot}{\cdot} \colon \LS_n \times \LS_n \longrightarrow \bC$
by bilinear extension of the rule 
\begin{align} \label{LSBiForm} 
(\alpha,\beta) \quad \longmapsto \quad \BiForm{\alpha}{\beta} ,
\end{align}
for each pair of link patterns $\alpha,\beta \in \LP_n$.
If $\alpha,\beta \in \LS_0$, then the product in~\eqref{evT} is empty, so we have $\BiForm{\alpha}{\beta} = 1$.  
For example, the bilinear forms $\BiForm{\alpha}{\beta}$ of the link patterns $\alpha$ and $\beta$ in 
(\ref{LinkStateRule},~\ref{LinkStateRuleDefect2},~\ref{LinkStateRuleDefect1}) respectively evaluate to
\begin{align}
\bigg( \; \raisebox{-18pt}{\includegraphics[scale=0.275]{e-Connectivities4.pdf}} \; \bigg) \;
= \; \nu^2 , \qquad \qquad
\bigg( \; \vcenter{\hbox{\includegraphics[scale=0.275]{e-Connectivities10.pdf}}}  \; \bigg) \;
= \; \nu , \qquad \qquad
\bigg( \; \vcenter{\hbox{\includegraphics[scale=0.275]{e-Connectivities8.pdf}}}  \; \bigg) \;
= \; 0 . 
\end{align}

In order to generalize the above definition to give a bilinear form on the valenced link state module $\LS_\multii$,
assuming $\max \multii < \ppmin(q)$, we first define the
\emph{Jones-Wenzl composite projector}
\begin{align}\label{WJCompProj} 
\WJProj_\multii \quad := \quad \vcenter{\hbox{\includegraphics[scale=0.275]{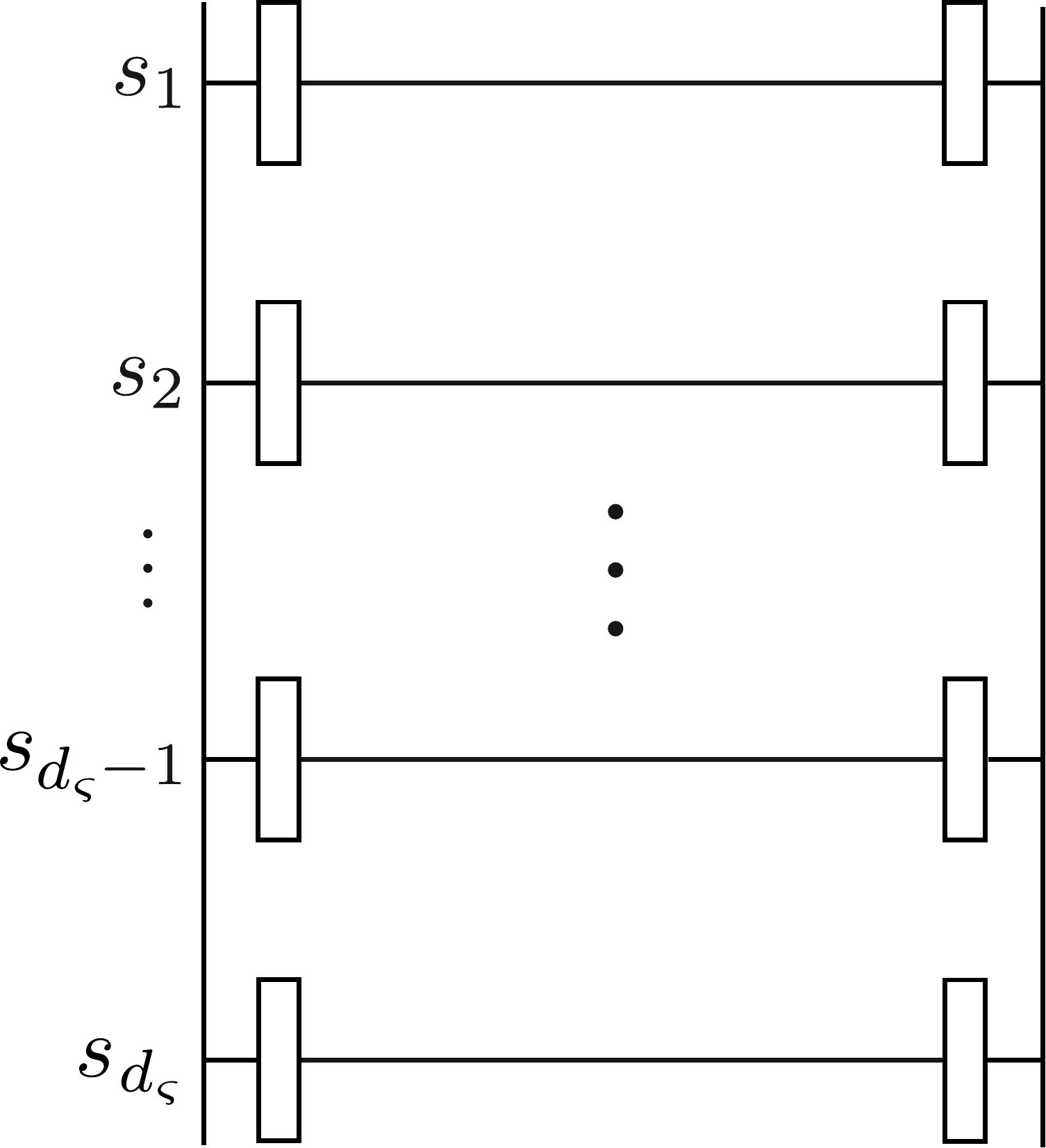} ,}} 
\hphantom{\WJProj_\multii \quad := \quad}
\end{align}
and the \emph{Jones-Wenzl composite embedder}
\begin{align}\label{WJCompEmb} 
\WJEmb_\multii \quad := \quad \vcenter{\hbox{\includegraphics[scale=0.275]{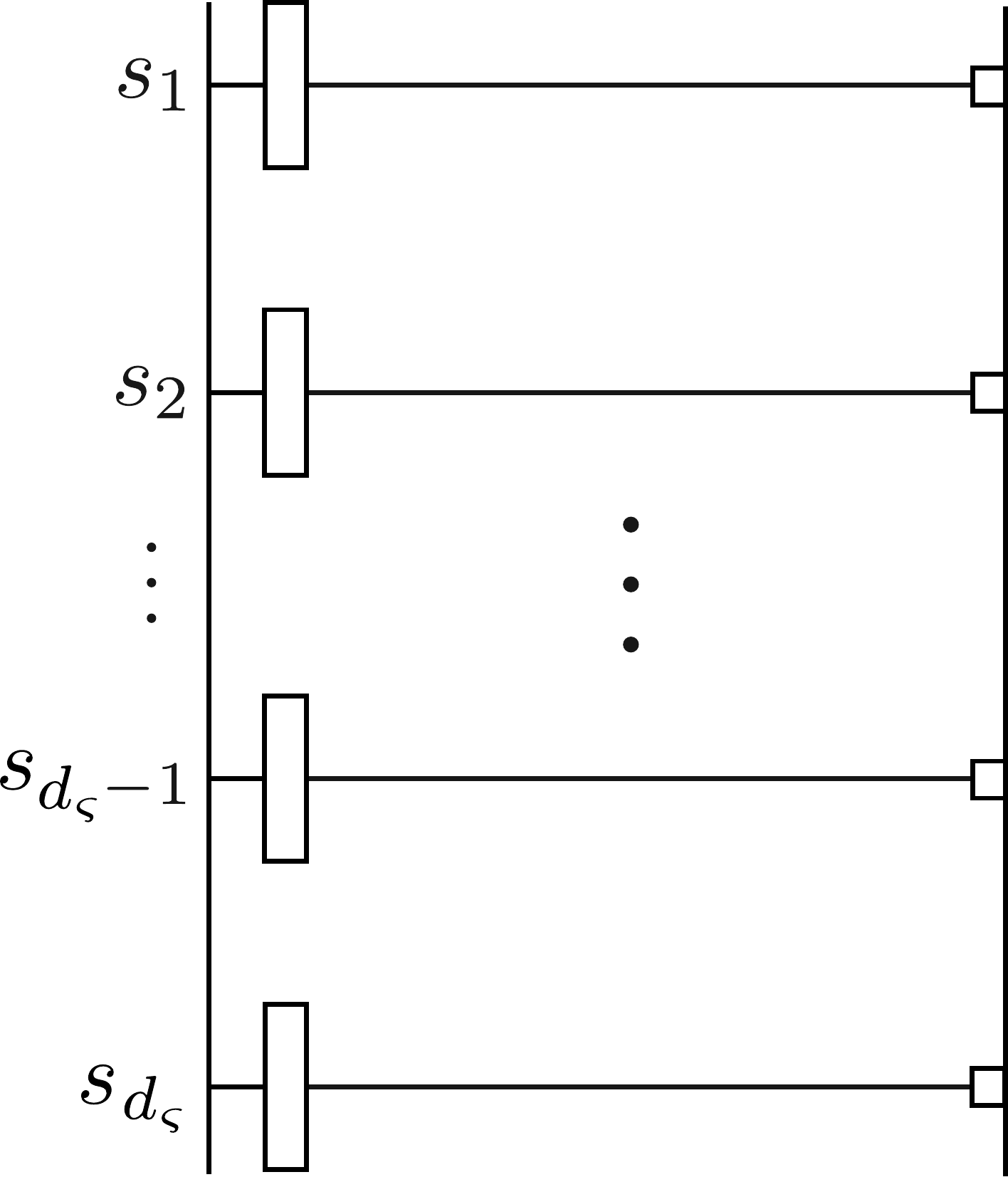} .}}
\hphantom{\WJEmb_\multii \quad := \quad}
\end{align}
In~\cite{fp3}, we relate $\WJProj_\multii$ and $\WJEmb_\multii$ respectively to a certain projection
and embedding on a type-one tensor product (``spin chain'') representation 
of the Hopf algebra $U_q(\mathfrak{sl}_2)$.
We denote the reflection $\WJEmb_\multii^\dagger$ of $\WJEmb_\multii$ about a vertical axis by
\begin{align} \label{WJProjHatEmb} 
\WJProjHat_\multii \quad := \quad \vcenter{\hbox{\includegraphics[scale=0.275]{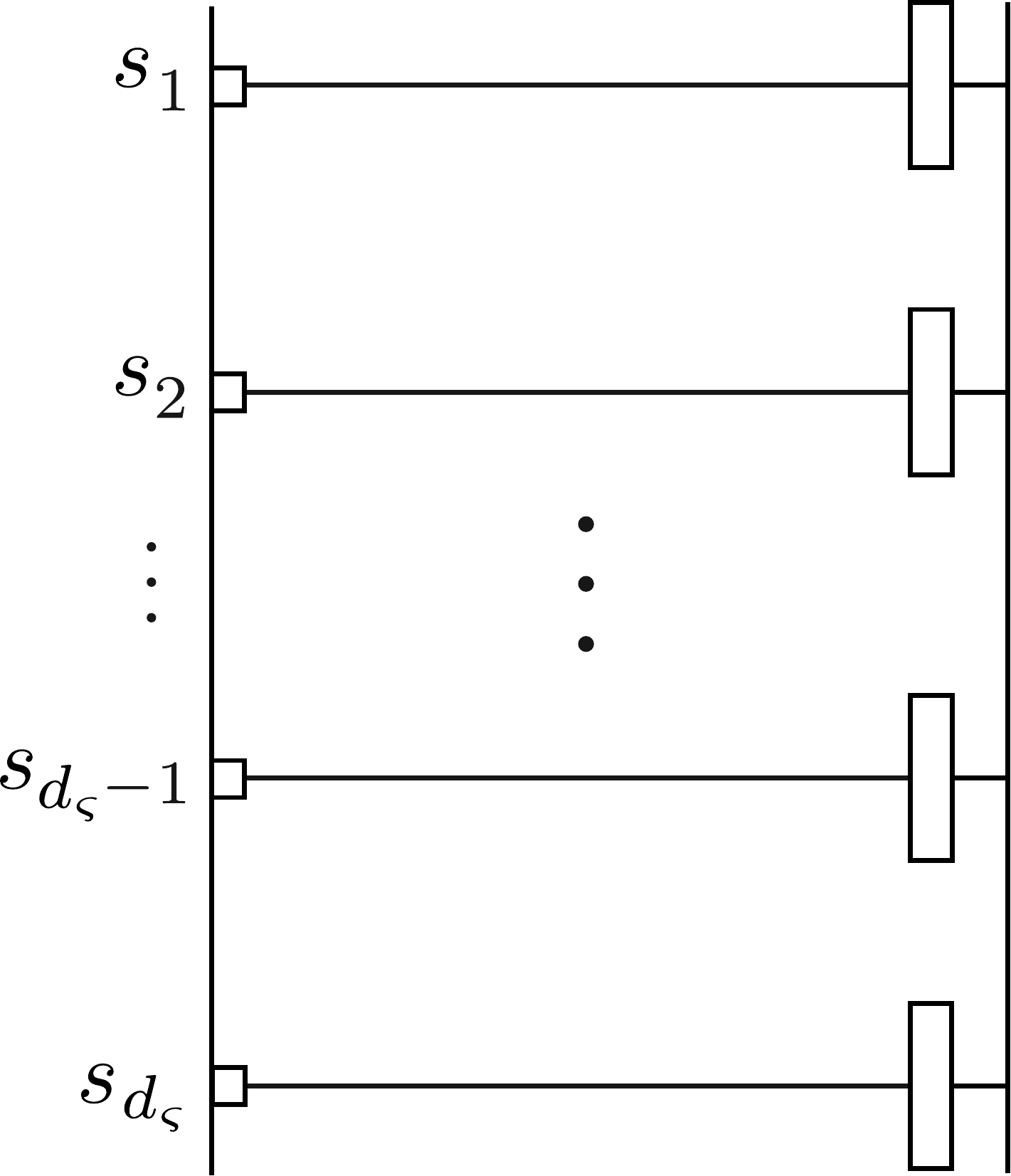} .}}
\hphantom{\WJProjHat_\multii \quad := \quad}
\end{align}
Also $\WJProjHat_\multii$ corresponds to a surjective homomorphism of 
$U_q(\mathfrak{sl}_2)$-type-one modules~\cite{fp3}.
We observe that 
\begin{align}\label{IdComp} 
\WJProjHat_\multii \WJEmb_\multii = \mathbf{1}_{\TL_\multii(\nu)} 
\qquad \qquad \text{and} \qquad \qquad
\WJEmb_\multii \WJProjHat_\multii = \WJProj_\multii ,
\end{align}
and that $\WJEmb_\multii$ defines a linear injection
$\WJEmb_\multii (\,\cdot\,) \colon \LS_\multii \longrightarrow \LS_{\Summed_\multii}$ by sending 
a valenced link state $\alpha \in \LS_\multii$ to the link state
\begin{align} \label{EmbeddingsDef1-1}
\alpha \mapsto \WJEmb_\multii\alpha .
\end{align}
Similarly, $\WJProjHat_\multii$ and $\WJProj_\multii$ define linear surjections 
$\WJProjHat_\multii (\,\cdot\,) \colon \LS_{\Summed_\multii} \longrightarrow \LS_\multii$
and $\WJProj_\multii (\,\cdot\,) \colon \LS_{\Summed_\multii} \longrightarrow \LS_{\Summed_\multii}$, 
sending $\beta \in \LS_\multii$ respectively~to
\begin{align} \label{ProjHatDef1-1} 
\beta \mapsto \WJProjHat_\multii\beta \qquad \qquad \text{and} \qquad \qquad \beta \mapsto \WJProj_\multii \beta .
\end{align}
In lemma~\ref{WJLSBasisLem} in appendix~\ref{AppWJ}, we collect salient properties of these maps.

Assuming that $\max \multii < \ppmin(q)$ and using $\WJEmb_\multii$,
we define the \emph{(valenced) link state bilinear form} 
$\BiForm{\cdot}{\cdot} \colon \LS_\multii \times \LS_\multii \longrightarrow \bC$,
\begin{align} \label{LSBiFormExt} 
\BiForm{\alpha}{\beta} := \BiForm{\WJEmb_\multii \alpha}{\WJEmb_\multii \beta}, 
\end{align}
for all valenced link states $\alpha,\beta \in \LS_\multii$, 
where the right side is the bilinear form on $\LS_{\Summed_\multii}$. 
For example, we have 
\begin{align} 
\bigg( \; 
\vcenter{\hbox{\includegraphics[scale=0.275]{e-LinkPattern4_valenced.pdf}}}  \;\;  \vcenter{\hbox{\text{\scalebox{1}[2.2]{$\BarAction$}}}} \;\;  
\vcenter{\hbox{\includegraphics[scale=0.275]{e-LinkPattern3_valenced.pdf}}} 
\; \bigg)
\quad = \quad \bigg( \; \vcenter{\hbox{\includegraphics[scale=0.275]{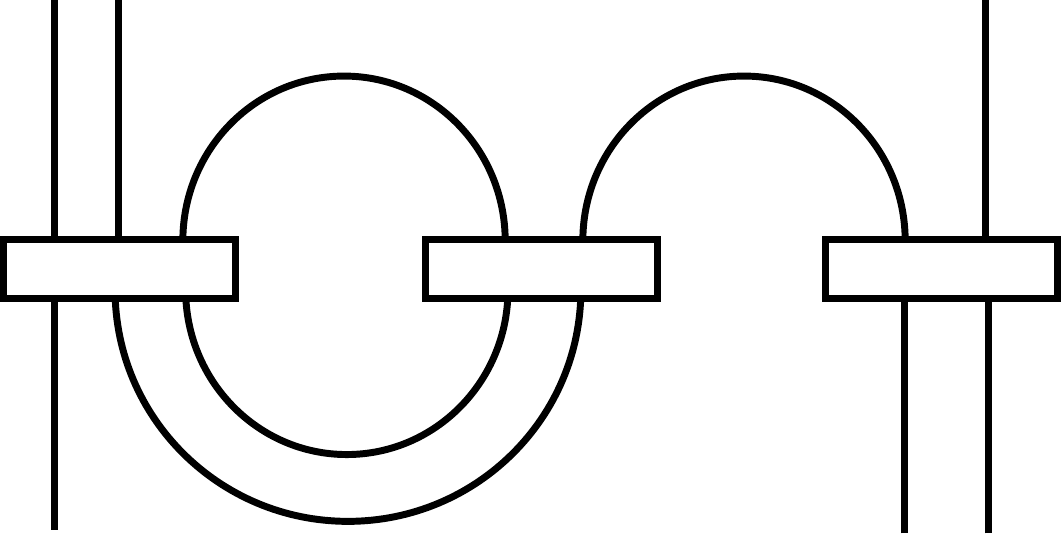}}} \; \bigg) .
\end{align}
As a consequence of rule~\eqref{TurnBack0}, the standard modules $\smash{\LS_\multii\super{s}}$ and $\smash{\LS_\multii\super{r}}$ 
are orthogonal if $s \neq r$, so we have
\begin{align} 
\rad \LS_\multii\super{s} := \; & \big\{\alpha \in \LS_\multii\super{s} \, | \, 
\text{$\BiForm{\alpha}{\beta} = 0$ for all $\beta \in \LS_\multii\super{s}$} \big\},\\
\label{RadDirSum}
\rad \LS_\multii := \; & \{\alpha \in \LS_\multii \, | \, 
\text{$\BiForm{\alpha}{\beta} = 0$ for all $\beta \in \LS_\multii$}\} =\bigoplus_{s \, \in \, \DefectSet_\multii} \rad \LS_\multii\super{s}.
\end{align}



In the next lemma, we give two basic properties of the bilinear form $\BiForm{\cdot}{\cdot}$. 

\begin{lem} \label{EasyLem2} 
Suppose $\max \multii < \ppmin(q)$. For all valenced link patterns $\alpha, \beta \in \LS_\multii$ and 
for all valenced tangles $T \in \TL_\multii(\nu)$, we have  
\begin{align}
\label{SymmProp} 
\textnormal{symmetry:} \qquad \quad
\BiForm{\alpha}{\beta} = \; & \BiForm{\beta}{\alpha}, \\
\textnormal{invariance:} \qquad \;
\label{InvarProp} 
\BiForm{\alpha}{T \beta} = \; & \BiForm{T^\dagger \alpha}{\beta}.
\end{align}
\end{lem}

\begin{proof} 
In light of definition~\eqref{LSBiFormExt}, we may assume that $\multii = \OneVec{n}$ for some 
$n \in \bZpos$. In this case, 
identity~\eqref{SymmProp} is immediate from definition~\eqref{LSBiForm}, and 
identity~\eqref{InvarProp} also follows easily from the definitions: for example, for
\begin{align}
\alpha \quad = & \quad \vcenter{\hbox{\includegraphics[scale=0.275]{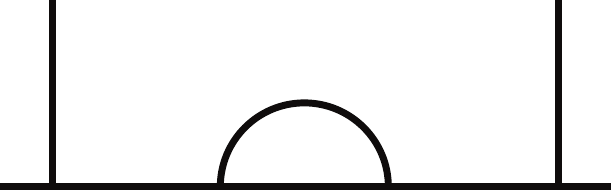} ,}} \qquad
\beta \quad = \quad \vcenter{\hbox{\includegraphics[scale=0.275]{e-LinkPattern4.pdf} ,}} \qquad \text{and} \qquad
T \quad = \quad \vcenter{\hbox{\includegraphics[scale=0.275]{e-TL_example3.pdf} ,}} 
\end{align}
the following network (rotated by $-\pi/2$ radians)
represents either quantity $\BiForm{\alpha}{T \beta}$ or $\BiForm{T^\dagger \beta}{\alpha}$:
\begin{align}
\vcenter{\hbox{\includegraphics[scale=0.275]{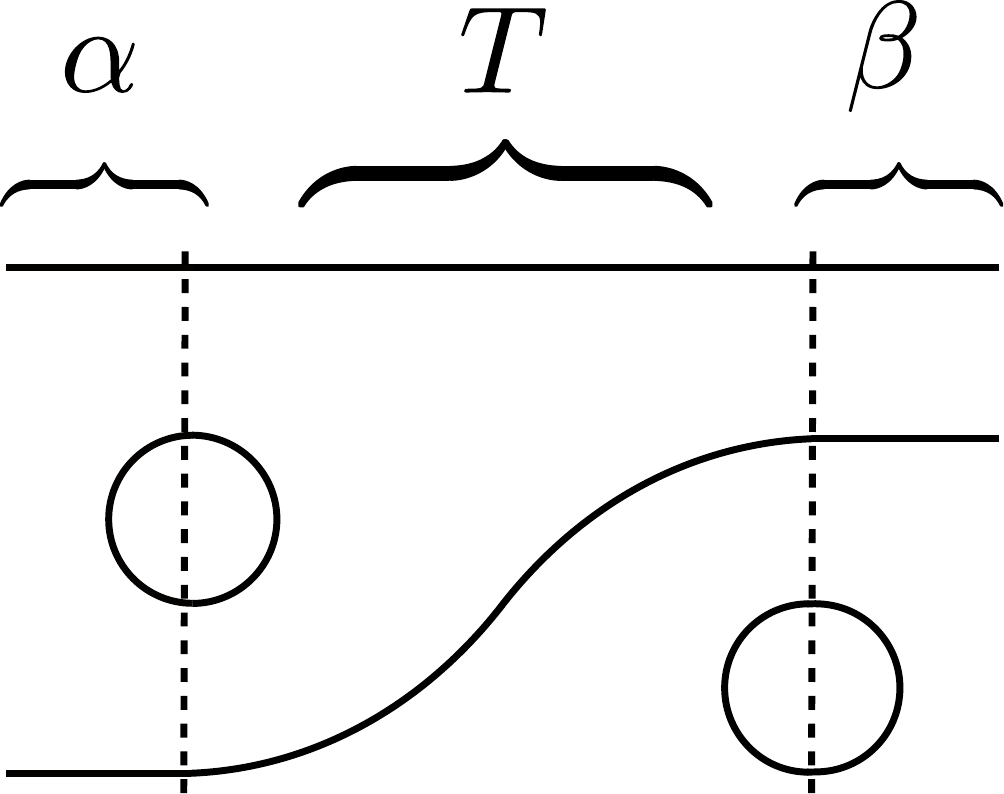} .}}
\end{align}
\end{proof}

We say that the bilinear form $\BiForm{\cdot}{\cdot}$ is \emph{symmetric} because of property~\eqref{SymmProp}
and \emph{invariant} because of property~\eqref{InvarProp}.
Invariance property~\eqref{InvarProp} guarantees that the radical~\eqref{RadDirSum} of the bilinear form is 
a $\TL_\multii(\nu)$-submodule of $\LS_\multii$, and that, for each $s \in \DefectSet_\multii$, 
the radical $\smash{\rad \LS_\multii\super{s}}$ is a $\TL_\multii(\nu)$-submodule of 
the standard module $\smash{\LS_\multii\super{s}}$.

\bigskip

In spite of its simplicity, identity~\eqref{RidoutId} in the next lemma is a powerful tool 
for determining representation-theoretic properties of the standard modules.
This lemma is a natural generalization of~\cite[lemma~\red{3.2}]{rsa}.

\begin{lem} \label{RidoutIdLem} 
Suppose $\max (\multii, \multiii) < \ppmin(q)$.  
For all valenced link states $\alpha \in \smash{\LS_\multii\super{s}}$ and $\beta, \gamma \in \smash{\LS_\multiii\super{s}}$, 
we have
\begin{align}\label{RidoutId} 
\BarAction \alpha \quad \beta \BarAction \gamma = \BiForm{\beta}{\gamma}{\alpha} . 
\end{align}
\end{lem}

\begin{proof} 
First, we prove identity~\eqref{RidoutId} for the case when 
$\multii = \OneVec{n}$ and $\multiii = \OneVec{m}$, with $\alpha$, $\beta$, and $\gamma$ 
ordinary link states.  By linearity, we may assume that $\alpha$, $\beta$, and $\gamma$ 
are link patterns. We consider two scenarios:
\begin{enumerate}[leftmargin=*]
\itemcolor{red}
\item All defects of $\beta$ join with all defects of $\gamma$. 
In this case, we readily simplify $\BarAction \alpha \quad \beta \BarAction \gamma$ to $\alpha$ multiplied by the number of loops in the diagram for $\BiForm{\beta}{\gamma}$, which equals $\BiForm{\beta}{\gamma}\alpha$.  
Therefore, identity~\eqref{RidoutId} holds for this case. 

\item Some defects of $\beta$ do not join with defects of $\gamma$. 
In this case, a link of $\beta$ must join two defects of $\gamma$ together in the diagram 
for $\BiForm{\beta}{\gamma}$.  By rule~\eqref{TurnBack0}, 
identity~\eqref{RidoutId} holds also for this case. 
\end{enumerate}
By linearity, this proves~\eqref{RidoutId} for 
all link states $\alpha \in \smash{\LS_{n}\super{s}}$ and $\beta,\gamma \in \smash{\LS_{m}\super{s}}$.

Second, we prove identity~\eqref{RidoutId} for the general case.
Now, we have $\WJEmb_\multii \alpha \in \smash{\LS_{\Summed_\multii}\super{s}}$ and $\WJEmb_\multiii \beta, \WJEmb_\multiii \gamma \in \smash{\LS_{\Summed_\multiii}\super{s}}$, 
so the already proved identity~\eqref{RidoutId} for them gives 
\begin{align} \label{RidoutIdAgain}
\BarAction \WJEmb_\multii \alpha \quad \WJEmb_\multiii \beta \BarAction  \WJEmb_\multiii \gamma 
\overset{\eqref{RidoutId}}{=} \BiForm{\WJEmb_\multiii \beta}{\WJEmb_\multiii \gamma} \WJEmb_\multii \alpha .
\end{align}
On the other hand, by drawing a picture, it is straightforward to see that
\begin{align}\label{EId} 
\WJEmb_\multii \BarAction \alpha \quad \beta \BarAction \WJProjHat_\multiii = \BarAction \WJEmb_\multii \alpha \quad \WJEmb_\multiii \beta \BarAction.
\end{align}
Using these identities, we obtain asserted identity~\eqref{RidoutId} for all valenced link states 
$\alpha \in \smash{\LS_\multii\super{s}}$ and $\beta, \gamma \in \smash{\LS_\multiii\super{s}}$:
\begin{align} 
\BarAction \alpha \quad \beta \BarAction \gamma 
& \overset{\eqref{IdComp}}{=} \WJProjHat_\multii \WJEmb_\multii \BarAction \alpha \quad \beta \BarAction \gamma 
\overset{\eqref{IdComp}}{=} 
\WJProjHat_\multii (\WJEmb_\multii \BarAction \alpha \quad \beta \BarAction \WJProjHat_\multiii) (\WJEmb_\multiii \gamma ) \\
& \overset{\eqref{EId}}{=} \WJProjHat_\multii \BarAction \WJEmb_\multii \alpha \quad \WJEmb_\multiii \beta \BarAction  \WJEmb_\multiii \gamma  \\
& \overset{\eqref{RidoutIdAgain}}{=}
\BiForm{\WJEmb_\multiii \beta}{\WJEmb_\multiii \gamma} \WJProjHat_\multii \WJEmb_\multii \alpha \\
& \underset{\eqref{IdComp}}{\overset{\eqref{LSBiForm}}{=}} \BiForm{\beta}{\gamma} \alpha .
\end{align}
This concludes the proof.
\end{proof}

\subsection{Standard modules and their radicals} \label{LinkStateModSect}

We use the link state bilinear form $\BiForm{\cdot}{\cdot}$ 
to study the structure of the standard modules $\smash{\LS_\multii\super{s}}$, 
their radicals $\rad \smash{\LS_\multii\super{s}}$, and their quotients $\smash{\Quo_\multii\super{s}}$ by these radicals. 
We say that the bilinear form on $\smash{\LS_\multii\super{s}}$ 
is \emph{totally degenerate} if $\smash{\rad \LS_\multii\super{s}} = \smash{\LS_\multii\super{s}}$. 
Throughout this section, we assume that the bilinear forms are not totally degenerate.
This is usually true, but not always. In section~\ref{rofSect32}, we determine when this assumption is violated.

The following proposition is a key result in this section. 
Its proof is a straightforward adaptation of the proof of~\cite[proposition~\red{3.3}]{rsa},
which in turn follows closely~\cite[proposition~\red{3.2}]{gl2}.

\begin{prop} \label{GenLem2} 
Suppose $\max \multii < \ppmin(q)$. If $\rad \smash{\LS_\multii\super{s}} \neq \smash{\LS_\multii\super{s}}$,
then the following hold:
\begin{enumerate}
\itemcolor{red}
\item \label{simple} 
The quotient module $\smash{\Quo_\multii\super{s}}$ is simple, and $\smash{\rad \LS_\multii\super{s}}$ is 
the unique maximal proper submodule of $\smash{\LS_\multii\super{s}}$.
\item \label{indecomposable} 
The standard module $\smash{\LS_\multii\super{s}}$ is indecomposable.
\end{enumerate}
\end{prop} 

\begin{proof} 
Using property~\eqref{InvarProp} of the bilinear form, it is straightforward to show that 
$\smash{\rad \LS_\multii\super{s}}$ is a $\TL_\multii(\nu)$-submodule of $\LS_\multii\super{s}$.
For a valenced link state $\alpha \in \smash{\LS_\multii\super{s}}$, we let $[\alpha] \in \smash{\Quo_\multii\super{s}}$
denote its equivalence class in the quotient module~\eqref{QuoMod}.
\begin{enumerate}[leftmargin=*]
\itemcolor{red}
\item Let $\gamma \in \smash{\LS_\multii\super{s}} \setminus \rad \smash{\LS_\multii\super{s}}$. 
Then, we may choose a valenced link state $\beta \in \smash{\LS_\multii\super{s}}$ such that $\BiForm{\beta}{\gamma} = 1$.
Now, identity~\eqref{RidoutId} of lemma~\ref{RidoutIdLem} shows that we can construct any $[\alpha] \in \smash{\Quo_\multii\super{s}}$
through multiplying $[\gamma]$ by a tangle in $\TL_\multii(\nu)$:
\begin{align}\label{NiceComp} 
\BarAction \alpha \quad \beta \BarAction[\gamma] 
= [\BarAction \alpha \quad \beta \BarAction\gamma \, ] 
\overset{\eqref{RidoutId}}{=} \BiForm{\beta}{\gamma}[\alpha] 
= [\alpha] .
\end{align}
Hence, $\smash{\Quo_\multii\super{s}}$ is cyclic with generator $[\gamma]$. 
Because $\gamma \in \smash{\LS_\multii\super{s}} \setminus \rad \smash{\LS_\multii\super{s}}$ can be chosen arbitrarily, 
any nonzero element of $\smash{\Quo_\multii\super{s}}$ generates $\smash{\Quo_\multii\super{s}}$.
This shows that 
$\smash{\Quo_\multii\super{s}}$ is simple,
and $\smash{\rad \LS_\multii\super{s}}$ is the unique maximal proper submodule of $\smash{\LS_\multii\super{s}}$.

\item Suppose $\smash{\LS_\multii\super{s}}$ can be decomposed as a direct sum of two $\TL_\multii(\nu)$-submodules $U$ and $V$,
i.e., $\smash{\LS_\multii\super{s}} = U \oplus V$.
To prove that $\smash{\LS_\multii\super{s}}$ is indecomposable, we need to show that one of the submodules, $U$ or $V$, is trivial. 
Now, the same argument that we used to prove item~\ref{simple} shows that $\smash{\LS_\multii\super{s}}$ is cyclic,
generated by any nonzero valenced link state $\gamma \notin \rad \smash{\LS_\multii\super{s}}$.
We choose such $\gamma$, write it as $\gamma = \gamma_U + \gamma_V$ with $\gamma_U \in U$ and $\gamma_V \in V$,
and consider two cases:
\begin{enumerate}
\itemcolor{red}
\item[(a):] Both $\gamma_U$ and $\gamma_V$ belong to the radical $\smash{\rad \LS_\multii\super{s}}$.
In this case, the $\TL_\multii(\nu)$-submodules of $\smash{\LS_\multii\super{s}}$ generated by $\gamma_U$ and $\gamma_V$ 
are also $\TL_\multii(\nu)$-submodules of $\smash{\rad \LS_\multii\super{s}} \subset \smash{\LS_\multii\super{s}}$, and so is their direct sum. We get a contradiction:
\begin{align}
\LS_\multii\super{s} = \TL_\multii(\nu) \, \gamma \subset \TL_\multii(\nu) \, \gamma_U \oplus \TL_\multii(\nu) \, \gamma_V 
\subset \rad \LS_\multii\super{s} \subset \LS_\multii\super{s} \qquad \Longrightarrow \qquad \rad \LS_\multii\super{s} = \LS_\multii\super{s} .
\end{align} 

\item[(b):] Either $\gamma_U$ or $\gamma_V$ does not belong to the radical $\smash{\rad \LS_\multii\super{s}}$. 
Without loss of generality, we assume that $\gamma_U \notin \smash{\rad \LS_\multii\super{s}}$. 
Then, $\gamma_U$ generates the whole module $\smash{\LS_\multii\super{s}}$.  
Thus, we have $U = \smash{\LS_\multii\super{s}}$ and $V = \{0\}$.
\end{enumerate}
\end{enumerate}
This concludes the proof.
\end{proof}

Next, we show in proposition~\ref{HomLem2} and corollary~\ref{nonisoCor2} that 
the $\TL_\multii(\nu)$-modules of proposition~\ref{GenLem2}
are non-isomorphic. These results are straightforward adaptations 
of~\cite[proposition~\red{3.6} and corollary~\red{3.7}]{rsa}, based on~\cite{gl2}.

\begin{prop} \label{HomLem2} 
Suppose $\max \multii < \ppmin(q)$. If $\rad \smash{\LS_\multii\super{s}} \neq \smash{\LS_\multii\super{s}}$ and 
$\mathsf{M}$ and $\mathsf{N}$ respectively are submodules of 
$\smash{\LS_\multii\super{s}}$ and $\smash{\LS_\multii\super{r}}$, with $s < r$,
then the only homomorphism sending
$\smash{\LS_\multii\super{s}} / \mathsf{M} \longrightarrow \smash{\LS_\multii\super{r}} / \mathsf{N}$ is the zero homomorphism.
\end{prop}

\begin{proof}  
We let $[\alpha] \in \smash{\LS_\multii \super{s}} / \mathsf{M}$ 
(resp.~$[\alpha] \in \smash{\LS_\multii \super{r}} / \mathsf{N}$) denote the equivalence class 
of $\alpha \in \smash{\LS_\multii \super{s}}$ (resp.~$\alpha \in \smash{\LS_\multii \super{r}}$).  
Then for any $\beta, \gamma \in \smash{\LS_\multii \super{s}}$ chosen such that 
$\BiForm{\beta}{\gamma} = 1$, 
and for any homomorphism 
$\theta \colon \smash{\LS_\multii \super{s}} / \mathsf{M} \longrightarrow \smash{\LS_\multii \super{r}} / \mathsf{N}$,
we have
\begin{align}\label{AnotherOne} 
\BarAction\alpha \quad \beta \BarAction \theta \big( [\gamma] \big) 
= \theta \big( \BarAction \alpha \quad \beta \BarAction [\gamma] \big)
\overset{\eqref{NiceComp}}{=} \theta\big( [\alpha] \big) .
\end{align}
We let $\delta \in \smash{\LS_\multii \super{r}}$ be a valenced link state such that $\theta ( [\gamma] ) = [\delta]$.
Then, because 
$\alpha, \beta \in \smash{\LS_\multii \super{s}}$ and $\delta \in \smash{\LS_\multii \super{r}}$ with $s < r$, 
the tangle $\BarAction \alpha \quad \beta \BarAction$ necessarily joins two defects of $\delta$ 
together, so $\BarAction \alpha \quad \beta \BarAction \delta = 0$ by rule~\eqref{TurnBack0}. It follows that
\begin{align}
\theta\big( [\alpha] \big) \overset{\eqref{AnotherOne}}{=} 
\BarAction \alpha \quad \beta \BarAction \theta \big( [\gamma] \big) 
= \BarAction \alpha \quad \beta \BarAction[\delta] 
= [\BarAction \alpha \quad \beta \BarAction \delta] = 0 .
\end{align}
This implies that $\theta$ is the zero homomorphism.
\end{proof}

\begin{cor} \label{nonisoCor2} 
Suppose $\max \multii < \ppmin(q)$. If $\rad \smash{\LS_\multii\super{s}} \neq \smash{\LS_\multii\super{s}}$ and 
$\rad \smash{\LS_\multii\super{r}} \neq \smash{\LS_\multii\super{r}}$, then we have
\begin{align}\label{TwoStatements} 
\LS_\multii\super{s} \cong \LS_\multii\super{r} \quad \Longleftrightarrow \quad s = r 
\qquad \qquad \textnormal{and} \qquad \qquad 
\Quo_\multii\super{s} \cong 
\Quo_\multii\super{r} \quad \Longleftrightarrow \quad s = r . 
\end{align}
\end{cor}

\begin{proof}
Proposition~\ref{HomLem2} with $\mathsf{M} = \mathsf{N} = \{0\}$ 
(resp.~$\mathsf{M} = \smash{\rad \LS_\multii \super{s}}$, $\mathsf{N} = \smash{\rad \LS_\multii \super{r}}$)
implies the first (resp.~second) equivalence.
\end{proof}

In proposition~\ref{SimpleModuleProp} in section~\ref{TLSimpleModSect}, 
we prove that the nonzero quotients $\smash{\Quo_\multii\super{s}}$ of the standard modules
constitute the complete set of non-isomorphic simple $\TL_\multii(\nu)$-modules. 
For this, the results of the next sections~\ref{GramMatrixSect} and~\ref{RadicalSect} are not needed.
Sections~\ref{GramMatrixSect} and~\ref{RadicalSect} focus on the fine structure of the standard 
modules $\smash{\LS_\multii\super{s}}$ and their radicals.

It is also worthwhile to remark on the prospect of obtaining 
analogues of the above results
when the bilinear form on $\smash{\LS_\multii\super{s}}$ is totally degenerate.
In~\cite{rsa}, this is successfully done in the case that $\multii = \OneVec{n}$ for some $n \in \bZpos$. 
To establish this, the authors first show that the bilinear form on $\smash{\LS_\multii\super{s}}$
is totally degenerate if and only if $\nu = 0$ (i.e., $\ppmin(q) = 2$ by~\eqref{MinPower}) and $s = 0$. 
Then, with $\nu = 0 = s$, they use the renormalized bilinear form
\begin{align}
\BiForm{\cdot}{\cdot}' := \lim_{\nu \to 0} \nu^{-1} \BiForm{\cdot}{\cdot} 
\end{align}
to prove an analogue~\cite[proposition~\red{3.5}]{rsa} of proposition~\ref{GenLem2}.  
They also show that the radical of $\smash{\LS_n\super{0}}$ with respect to 
the new bilinear form is trivial~\cite[proposition~\red{4.9}]{rsa}, 
concluding that $\smash{\LS_n\super{0}}$ is a simple module when $\nu = 0$.
However, this simple module is isomorphic to another simple $\TL_n(\nu)$-module,
which is a nontrivial quotient of a standard module whose bilinear form is not totally 
degenerate~\cite[theorem~\red{7.2}]{rsa}. For example, it is straightforward to verify that
\begin{align}\label{example0} 
\rad \smash{\LS_2\super{0}} = \smash{\LS_2\super{0}} 
\qquad \Longleftrightarrow \qquad \text{$\ppmin(q) = 2$ \; (i.e., $\nu = 0$)} 
\qquad \Longleftrightarrow \qquad \LS_2\super{0} \cong \LS_2\super{2}. 
\end{align}

In general, we believe that if the radical of $\smash{\LS_\multii\super{s}}$
is totally degenerate, then $\smash{\LS_\multii\super{s}}$ is simple and 
isomorphic to some simple module $\smash{\Quo_{\multii}\super{r}}$ for which the index $r \in \DefectSet_\multii$ 
could be determined as in~\cite{rsa} by analyzing ``critical lines.''

%

%
%

\begin{conj}
Suppose $\max \multii < \ppmin(q)$.
The radical $\smash{\rad \LS_\multii\super{s}}$ is either the trivial module or a simple module.
\end{conj}

\subsection{Faithfulness of the link state representations} \label{FaithfulSect3}

In this section, we investigate when the link state representation of $\TL_\multii (\nu)$ on $\LS_\multii$ is faithful.
We say that the bilinear form on $\smash{\LS_\multii\super{s}}$ (resp.~$\LS_\multii$) 
is \emph{nondegenerate} if its radical is trivial, i.e., $\smash{\rad \LS_\multii\super{s}} = \{0\}$ (resp.~$\rad \LS_\multii = \{0\}$).
In corollary~\ref{PreFaithfulCor}, we prove that the representation of $\TL_\multii (\nu)$ on $\LS_\multii$
is faithful if and only if the bilinear form on $\LS_\multii$ is nondegenerate.
We obtain it as a corollary of proposition~\ref{PreFaithfulProp} below, in the special case that $\multiii = \multii$.

In the proof of proposition~\ref{PreFaithfulProp}, we use the Gram matrix $\smash{\Gram_\multii\super{s}}$ of 
the bilinear form~\eqref{LSBiFormExt} with respect to the basis $\smash{\LP_\multii\super{s}}$ of valenced link patterns. This matrix is given by 
\begin{align} \label{GramMatrixForSec3}
[ \Gram_\multii\super{s} ]_{\alpha, \beta} := \BiForm{\alpha}{\beta} , \quad 
\text{for all $\alpha, \beta \in \LP_\multii\super{s}$.}
\end{align}
In the proof, we only use the elementary fact that the matrix $\smash{\Gram_\multii\super{s}}$ is invertible if and only if $\rad \smash{\LS_\multii\super{s}} = \{0\}$. 
We study this matrix in greater detail in the next section~\ref{GramMatrixSect}.

\begin{prop} \label{PreFaithfulProp}
Suppose $\max \multiii < \ppmin(q)$. 
The following statements are equivalent:
\begin{enumerate}
\itemcolor{red}

\item \label{FF1}
We have $\rad \smash{\LS_\multiii\super{s}} \neq \{0\}$, for some integer $s \in \DefectSet_\multii^\multiii$.

\item \label{FF2}
There exists a nonzero 
valenced tangle $T \in \smash{\TL_\multii^\multiii (\nu)}$ such that 
\begin{align} \label{desired}
T\gamma=0 , \quad 
\textnormal{for all valenced link states $\gamma \in \smash{\LS_\multiii\super{t}}$ with $t \in \smash{\DefectSet_\multii^\multiii}$.}
\end{align}
\end{enumerate}
\end{prop}

\begin{proof} 
We prove the equivalence as follows:
\begin{enumerate}[leftmargin=3.5em]
\item[\ref{FF2} $\Rightarrow$ \ref{FF1}:] 
Suppose that there exists a nonzero valenced tangle $T \in \smash{\TL_\multii^\multiii (\nu)}$ 
such that we have $T\gamma = 0$, 
for all valenced link states $\gamma \in \smash{\LS_\multiii\super{t}}$ with $t \in \smash{\DefectSet_\multii^\multiii}$.
We will show that for some integer $s \in \DefectSet_\multii^\multiii$, 
the bilinear form on $\smash{\LS_\multiii\super{s}}$ is nondegenerate. 
For this, we expand the valenced tangle $T$ according to the number $r$ of crossing links:
\begin{align} \label{TSumForm} 
T \overset{\eqref{WJDirSum}}{=} \sum_{r \, \in \, \DefectSet_\multii^\multiii} T\super{r} ,
\quad \text{where $T\super{r} \in \smash{\TL_\multii^{\multiii ;\scaleobj{0.85}{(r)}}(\nu)}$ as in~\eqref{TLs}}, 
\end{align} 
and we expand each valenced tangle $T\super{r}$ in the link diagram basis $\smash{\LD_\multii^{\multiii; \scaleobj{0.85}{(r)}}}$ 
of $\smash{\TL_\multii^{\multiii ;\scaleobj{0.85}{(r)}}(\nu)}$ given by~(\ref{LDs},~\ref{WJSandwichMap}):
\begin{align} \label{Ts} 
T\super{r} \overset{\eqref{WJSandwichMap}}{=} \sum_{\substack{\alpha  \, \in \, \LP_{\multii}\super{r} \\ \beta \, \in \, \LP_\multiii\super{r}}} 
c_{\alpha,\beta}\super{r} \BarAction \alpha \quad \beta \BarAction ,
\end{align} 
for some coefficients $\smash{c_{\alpha,\beta}\super{r}} \in \bC$. 
Now, we note that if $\gamma \in \smash{\LP_\multiii\super{t}}$ and $r < t$, 
then each term in $T\super{r} \gamma$ contains a turn-back link. 
This observation combined with identity~\eqref{TurnBack0} and linearity shows that
\begin{align} \label{SmallerRAnnihilate}
\gamma \in \smash{\LS_\multiii\super{t}} \qquad \overset{\eqref{TurnBack0}}{\Longrightarrow} \qquad
T\super{r} \gamma = 0, \text{ for all $r < t$.}
\end{align}
Also, because $T \neq 0$, we may choose $s \in \DefectSet_\multii^\multiii$ 
to be the largest number such that $\smash{c_{\alpha,\beta}\super{s}} \neq 0$ in~\eqref{Ts}, 
for some pair of valenced link patterns $\alpha \in \smash{\LP_\multii\super{s}}$ 
and $\beta \in \smash{\LP_\multiii\super{s}}$. 
Then, 
using lemma~\ref{RidoutIdLem}, we have
\begin{align} 
\label{TGammaIsZero0}
0 = T \gamma 
\overset{\eqref{TSumForm}}{=} \sum_{r \, \leq \, s} T\super{r} \gamma 
& \; \overset{\eqref{Ts}}{=} \sum_{\substack{\alpha  \, \in \, \LP_{\multii}\super{s} \\ \beta \, \in \, \LP_\multiii\super{s}}} 
c_{\alpha,\beta}\super{s} \BarAction \alpha \quad \beta \BarAction \gamma \\
\label{TGammaIsZero1}
& \; \overset{\eqref{RidoutId}}{=} \; \sum_{\alpha  \, \in \, \LP_{\multii}\super{s} } \BiForm{\Sigma_\alpha}{\gamma}\alpha, \qquad 
\text{where} \quad \Sigma_\alpha := \sum_{\beta \, \in \, \smash{\LP_\multiii\super{s}}} c_{\alpha,\beta}\super{s} \beta \; \in \; \LS_\multiii\super{s} ,
\end{align}
for all valenced link states $\gamma \in \smash{\LS_\multiii\super{s}}$.
By our choice of $s$, we have $\Sigma_\alpha \neq 0$, for some $\alpha \in \smash{\LP_\multii\super{s}}$.
Furthermore, with the set $\smash{\LP_\multii\super{s}}$ linearly independent, 
(\ref{TGammaIsZero0}--\ref{TGammaIsZero1}) implies that 
${\BiForm{\Sigma_\alpha}{\gamma} = 0}$, for all $\gamma \in \smash{\LS_\multiii\super{s}}$. 
Hence, $\Sigma_\alpha$ is a nonzero valenced tangle that 
lies in the radical of $\smash{\LS_\multiii\super{s}}$, so this radical is not trivial.
This proves that \ref{FF2} $\Rightarrow$ \ref{FF1}.

\item[\ref{FF1} $\Rightarrow$ \ref{FF2}:] 
Suppose $\smash{\rad \LS_\multiii\super{s}} \neq \{0\}$, for some $s \in \DefectSet_\multii^\multiii$. 
We will construct a nonzero valenced tangle 
$T \in \smash{\TL_\multii^\multiii(\nu)}$ with property~\eqref{desired}. 
For this, we let $s$ be the smallest number such that $\rad \smash{\LS_\multiii\super{s}} \neq \{0\}$, 
we choose arbitrary nonzero valenced link states $\delta \in \rad \smash{\LS_\multiii\super{s}}$ 
and $\epsilon \in \smash{\LS_\multii\super{s}}$, 
and we form the valenced tangle
\begin{align} \label{TConstruction}
T 
:= \BarAction \epsilon \quad \delta \BarAction \; + \;   
\sum_{ \vphantom{\LP_{\multii}\super{r}} r \, < \, s}  \; \sum_{\substack{\alpha  \, \in \, \LP_{\multii}\super{r} \\ \beta \, \in \, \LP_\multiii\super{r}}} 
c_{\alpha,\beta}\super{r} \BarAction \alpha \quad \beta \BarAction ,
\end{align} 
where the coefficients $\smash{c_{\alpha,\beta}\super{r}} \in \bC$ are to be determined later. 
We immediately note that $T \neq 0$ because $\delta, \epsilon \neq 0$.

Let $\gamma$ be an arbitrary valenced $(\multiii,t)$-link state with $t \in \smash{\DefectSet_\multii^\multiii}$.
If $t > s$, then we immediately have $T\gamma = 0$ by~\eqref{SmallerRAnnihilate}.
Also, if $t = s$, then with our choice of $\delta \in \rad \smash{\LS_\multii\super{s}}$, lemma~\ref{RidoutIdLem} gives
$T\gamma=0$ as well:
\begin{align} 
T \gamma \overset{\eqref{TConstruction}}{=}  \BarAction \epsilon \quad \delta \BarAction \gamma \; + \;   
\sum_{ \vphantom{\LP_{\multii}\super{r}} r \, < \, s}  \; \sum_{\substack{\alpha  \, \in \, \LP_{\multii}\super{r} \\ \beta \, \in \, \LP_\multiii\super{r}}} 
c_{\alpha,\beta}\super{r} \BarAction \alpha \quad \beta \BarAction \gamma 
\overset{\eqref{RidoutId}}{\underset{\eqref{SmallerRAnnihilate}}{=}} \BiForm{\delta}{\gamma} \epsilon = 0.
\end{align} 
Hence, we assume that $t < s$ from now on. 
We prove recursively that there exists a $t$-independent choice of coefficients for $T$ in~\eqref{TConstruction}
such that $T$ has property~\eqref{desired}.
For this, we fix $t < s$ and assume that we have already found the desired coefficients for $r > t$:
i.e., we assume that there exist constants $\smash{b_{\alpha,\beta}\super{r}}$ such that
\begin{align} 
\label{recassu} 
& c_{\alpha, \beta}\super{r} = b_{\alpha,\beta}\super{r}, 
&& \quad \text{for all $r \in \DefectSet_\multii^\multiii$  such that $t < r < s$, and for all
$\alpha \in \smash{\LP_\multii\super{r}}$, $ \beta \in \smash{\LP_\multiii\super{r}}$} \\
\label{recimpl}
\qquad \Longrightarrow \qquad 
& T\gamma = 0 , && \quad
\text{for all $\gamma \in \LS_\multiii\super{u}$ with $u \in \DefectSet_\multii^\multiii$ such that $u > t$}.
\end{align}
Our aim is to determine constants 
$\smash{\big\{ b_{\alpha,\beta}\super{t} \, \big| \, 
\alpha \in \LP_\multii\super{t} \, \beta \in \LP_\multiii\super{t} \big\}}$, 
such that the valenced tangle $T$~\eqref{TConstruction} whose coefficients for $r \geq t$ are chosen 
to be~\eqref{recassu} for $r > t$ and $\smash{c_{\alpha,\beta}\super{t} = b_{\alpha,\beta}\super{t}}$ 
for $r=t$ has the property
\begin{align} \label{desiredrecu}
T\gamma=0 , \quad 
\text{for all $\gamma \in \LS_\multiii\super{u}$ with $u \in \DefectSet_\multii^\multiii$ such that $u \geq t$.}
\end{align}
By linearity and~\eqref{recimpl}, it suffices to establish~\eqref{desiredrecu} for valenced link patterns
$\gamma \in \smash{\LP_\multiii\super{t}}$:
\begin{align} 
\label{LongSumExpression0}
0 = T \gamma \; \underset{\eqref{recassu}}{\overset{\eqref{TConstruction}}{=}} \; 
& \BarAction \epsilon \quad \delta \BarAction \gamma 
\; + \; \sum_{ \vphantom{\alpha  \, \in \, \LP_{\multii}\super{r}} t \, < \, r \, < \, s}  \; \sum_{\substack{\alpha  \, \in \, \LP_{\multii}\super{r} \\ \beta \, \in \, \LP_\multiii\super{r}}} 
b_{\alpha,\beta}\super{r} \BarAction \alpha \quad \beta \BarAction \gamma \\
\label{LongSumExpression1}
\; & +  \sum_{\substack{\alpha  \, \in \, \LP_{\multii}\super{t} \\ \beta \, \in \, \LP_\multiii\super{t}}} 
b_{\alpha,\beta}\super{t} \BarAction \alpha \quad \beta \BarAction \gamma
\; + \; \sum_{ \vphantom{\alpha  \, \in \, \LP_{\multii}\super{r}} r \, < \, t}  \; \sum_{\substack{\alpha  \, \in \, \LP_{\multii}\super{r} \\ \beta \, \in \, \LP_\multiii\super{r}}} 
c_{\alpha,\beta}\super{r} \BarAction \alpha \quad \beta \BarAction \gamma .
\end{align} 
Expanding the first line~\eqref{LongSumExpression0} in the valenced link pattern basis 
$\smash{\LP_\multii\super{t}}$ of the space $\smash{\LS_\multii\super{t}}$, we have
\begin{align} \label{LongSumExpression2}
\BarAction \epsilon \quad \delta \BarAction \gamma 
\; + \; \sum_{ \vphantom{\alpha  \, \in \, \LP_{\multii}\super{r}} t \, < \, r \, < \, s} \; 
\sum_{\substack{\alpha  \, \in \, \LP_{\multii}\super{r} \\ \beta \, \in \, \LP_\multiii\super{r}}} 
b_{\alpha,\beta}\super{r} \BarAction \alpha \quad \beta \BarAction \gamma 
\; = \sum_{\eta  \, \in \, \smash{\LP_{\multii}\super{t}}} a_{\eta, \gamma}\super{t} \eta ,
\end{align}
for some coefficients $\smash{a_{\eta, \gamma}\super{t}} \in \bC$ indexed by valenced link patterns 
$\eta \in \smash{ \LP_\multii\super{t}}$ and $\gamma \in \smash{ \LP_\multiii\super{t}}$.
We note that these coefficients are determined by our assumption~\eqref{recassu} together with 
our choice of $\epsilon$ and $\delta$ in the beginning.

Then, after inserting~\eqref{LongSumExpression2} into~\eqref{LongSumExpression0}, 
applying identity~\eqref{RidoutId} from lemma~\ref{RidoutIdLem} to the first sum in~\eqref{LongSumExpression1}, 
and observing that the second sum 
in~\eqref{LongSumExpression1} equals zero by~\eqref{SmallerRAnnihilate}, we obtain
\begin{align} \label{LongSumExpression3}
T \gamma \; \overset{\eqref{RidoutId}}{\underset{(\ref{SmallerRAnnihilate},~\ref{LongSumExpression2})}{=}}  
\sum_{\eta  \, \in \, \smash{\LP_{\multii}\super{t}}} a_{\eta, \gamma}\super{t} \eta 
\; + \sum_{\substack{\alpha  \, \in \, \LP_{\multii}\super{t} \\ \beta \, \in \, \LP_\multiii\super{t}}} b_{\alpha,\beta}\super{t} \BiForm{\beta}{\gamma}{\alpha} 
\; =  \sum_{\alpha  \, \in \, \LP_{\multii}\super{t}}
\Big( a_{\alpha, \gamma}\super{t}  \; + \sum_{\beta \, \in \, \LP_\multiii\super{t}} b_{\alpha,\beta}\super{t} \BiForm{\beta}{\gamma} \Big) \alpha .
\end{align}

Now, the requirement that~\eqref{LongSumExpression3} equals zero for all valenced link patterns 
$\gamma \in \smash{\LP_\multiii\super{t}}$  
determines a linear system of equations for each valenced link pattern $\alpha \in \smash{\LP_\multii\super{t}}$. 
Each system is of the form
\begin{align} \label{LinearSystemFaithful}
\Gram_\multiii\super{t} v_\alpha = - a_\alpha ,
\end{align}
where $\smash{\Gram_\multiii\super{t}}$ is the Gram matrix~\eqref{GramMatrixForSec3},
and $v_\alpha$ and $a_\alpha$  are respectively the vectors with components 
$\smash{b_{\alpha,\beta}\super{t}}$ and
$\smash{a_{\alpha, \beta}\super{t}}$ indexed by $\beta \in \smash{\LP_\multiii\super{t}}$. 
Now, because $t < s$ and $s$ is the smallest integer such that $\rad \smash{\LS_\multiii\super{s}} \neq \{0\}$, 
we have 
\begin{align}
\rad \smash{\LS_\multiii\super{t}} = \{0\} \qquad \Longrightarrow \qquad \det \smash{\Gram_\multiii\super{t}} \neq 0.
\end{align}
Therefore, for each $\alpha \in \smash{\LP_\multii\super{t}}$,
the corresponding linear system~\eqref{LinearSystemFaithful} has a unique solution
$v_\alpha = \smash{\big( b_{\alpha,\beta}\super{t} \big)_{\beta \in \LP_\multiii\super{t}}}$.
Assumption~(\ref{recassu}--\ref{recimpl}) now expands to the result
\begin{align} 
& c_{\alpha, \beta}\super{r} = b_{\alpha,\beta}\super{r}, 
&& \quad \text{for all $r \in \DefectSet_\multii^\multiii$  such that $t \leq r < s$, and for all
$\alpha \in \smash{\LP_\multii\super{r}}$, $ \beta \in \smash{\LP_\multiii\super{r}}$} \\
\qquad \Longrightarrow \qquad 
& T\gamma = 0 , && \quad
\text{for all $\gamma \in \LS_\multiii\super{u}$ with $u \in \DefectSet_\multii^\multiii$ such that $u \geq t$}.
\end{align}
Hence, $T$ with these coefficients indeed has property~\eqref{desiredrecu}.
This shows that we may recursively determine all coefficients for $T$ in~\eqref{TConstruction}
in such a way that $T$ has property~\eqref{desired}.
In other words, we can construct a nonzero valenced tangle 
$T \in \smash{\TL_\multii^\multiii(\nu)}$ with property~\eqref{desired}.
This proves that \ref{FF1} $\Rightarrow$ \ref{FF2}.
\end{enumerate}
\end{proof}

\begin{cor} \label{PreFaithfulCor}
Suppose $\max \multii < \ppmin(q)$. 
The link state representation induced by the action of $\TL_\multii (\nu)$ on $\LS_\multii$
is faithful if and only if $\rad \LS_\multii = \{0\}$.
\end{cor}
\begin{proof}
This follows from decomposition~\eqref{RadDirSum}
and proposition~\ref{PreFaithfulProp} specialized to the case $\multiii = \multii$.
\end{proof}

We remark that corollary~\ref{PreFaithfulCor} does not hold if we replace $\LS_\multii$ with $\smash{\LS_\multii\super{s}}$ in it. 
Indeed, to paraphrase observation~\eqref{SmallerRAnnihilate}, if $T \in \smash{\TL_\multii^{\multiii ;\scaleobj{0.85}{(r)}}(\nu)}$ for some $r < s$, 
then we have $T\alpha = 0$ for all $\alpha \in \smash{\LS_\multii\super{s}}$.
\section{Trivalent link states and Gram matrix} \label{GramMatrixSect}

Proposition~\ref{GenLem2} shows that the quotient $\smash{\Quo_\multii\super{s}}$, 
if not trivial,
is a simple $\TL_\multii(\nu)$-module. It is natural to ask what is the dimension of the simple module $\smash{\Quo_\multii\super{s}}$
and when do we have the equality $\smash{\Quo_\multii\super{s}} = \smash{\LS_\multii\super{s}}$, or equivalently, $\rad \smash{\LS_\multii\super{s}} = \{0\}$.
The answer hinges on a complete understanding of the radical of $\smash{\LS_\multii\super{s}}$.

In section~\ref{RadicalSect}, we completely determine the radical of $\smash{\LS_\multii\super{s}}$,
for all multiindices $\multii \in \{\OneVec{0}\} \cup \smash{\bZpos^\#}$ 
and integers $s \in \DefectSet_\multii$.  That is, we determine bases for and the dimensions of all of these radicals.  
In the present section, we introduce tools that we use in section~\ref{RadicalSect} to complete these tasks.  
A key tool is the Gram matrix $\smash{\Gram_\multii\super{s}}$ of the bilinear form~\eqref{LSBiFormExt},
\begin{align}\label{GramMatrix2}
[ \Gram_\multii\super{s} ]_{\alpha, \beta} := \BiForm{\alpha}{\beta} , \quad 
\text{for all $\alpha, \beta \in \LP_\multii\super{s}$.}
\end{align}
The Gram matrix encodes valuable information about the radical of $\smash{\LS_\multii\super{s}}$.  Specifically, the dimension of $\rad \smash{\LS_\multii\super{s}}$ 
equals the nullity of $\smash{\Gram_\multii\super{s}}$, so $\rad \smash{\LS_\multii\super{s}}$ is trivial 
if and only if $\det \smash{\Gram_\multii\super{s}} \neq 0$.  Hence, we are interested in zeros of the determinant of $\smash{\Gram_\multii\super{s}}$.  
In proposition~\ref{GramDetLem}, we give an explicit formula for $\smash{\det \Gram_\multii\super{s}}$. 
We use this formula to prove the main result of this section, proposition~\ref{VanishDetLem2}, 
which says that $\smash{\det \Gram_\multii\super{s}} \neq 0$ 
(thus, $\rad \smash{\LS}_\multii\super{s}$ is trivial), for all $s \in \DefectSet_\multii$ if $\Summed_\multii < \ppmin(q)$.

However, this key result is not enough to determine all parameters $q \in \bC^\times$ 
such that the radical of $\smash{\LS_\multii\super{s}}$ trivial, let alone to find a basis for it. 
Because of its complexity, the formula appearing in proposition~\ref{GramDetLem} for 
$\smash{\det \Gram_\multii\super{s}}$ is difficult to use 
when $\Summed_\multii \geq \ppmin(q)$.  For example, if $\ppmin(q) \neq 2$, then proposition~\ref{VanishDetLem2} in fact holds 
as an if-and-only-if statement (corollary~\ref{RadicalCor4}), but one direction of it
does not follow easily by using only the formula for the determinant.

Thus, in section~\ref{RadicalSect} we turn to other methods for determining the radical of $\smash{\LS_\multii\super{s}}$: 
we use special types of valenced link states that we call ``trivalent link states."
They form an orthogonal basis for the standard module $\smash{\LS_\multii\super{s}}$ if $\Summed_\multii < \ppmin(q)$,
as we prove in proposition~\ref{IndOrthBasisLem}.
We define and study these link states in sections~\ref{ConformalBlocksSect} and~\ref{ConformalBlocksProp}.
First, we give a simple definition for the trivalent link states, assuming $\Summed_\multii < \ppmin(q)$.
Then, for cases where this inequality does not hold, we give an alternative, more complicated definition, 
which reduces to the simpler one if $\Summed_\multii < \ppmin(q)$.

In section~\ref{DetSect}, we use the basis of trivalent link states to diagonalize the Gram matrix $\smash{\Gram_\multii\super{s}}$, 
compute its determinant (proposition~\ref{GramDetLem}), and prove that this determinant does not vanish if $\Summed_\multii < \ppmin(q)$ (proposition~\ref{VanishDetLem2}). 
In section~\ref{RecursSect}, we derive various recursion formulas for the deteminant of $\smash{\Gram_\multii\super{s}}$, to be used
in section~\ref{RadicalSect} and in subsequent work~\cite{fp2}.

\subsection{Definition of the trivalent link states} \label{ConformalBlocksSect}

In this section, we introduce trivalent link states.  In our forthcoming article~\cite{fp3}, we identify 
them with conformal blocks of CFT via relations that we call ``the link-state spin-chain correspondence"~\cite{fp2}
and the ``spin-chain -- Coulomb gas correspondence"~\cite{kp2}.

\begin{figure}
\includegraphics[scale=0.275]{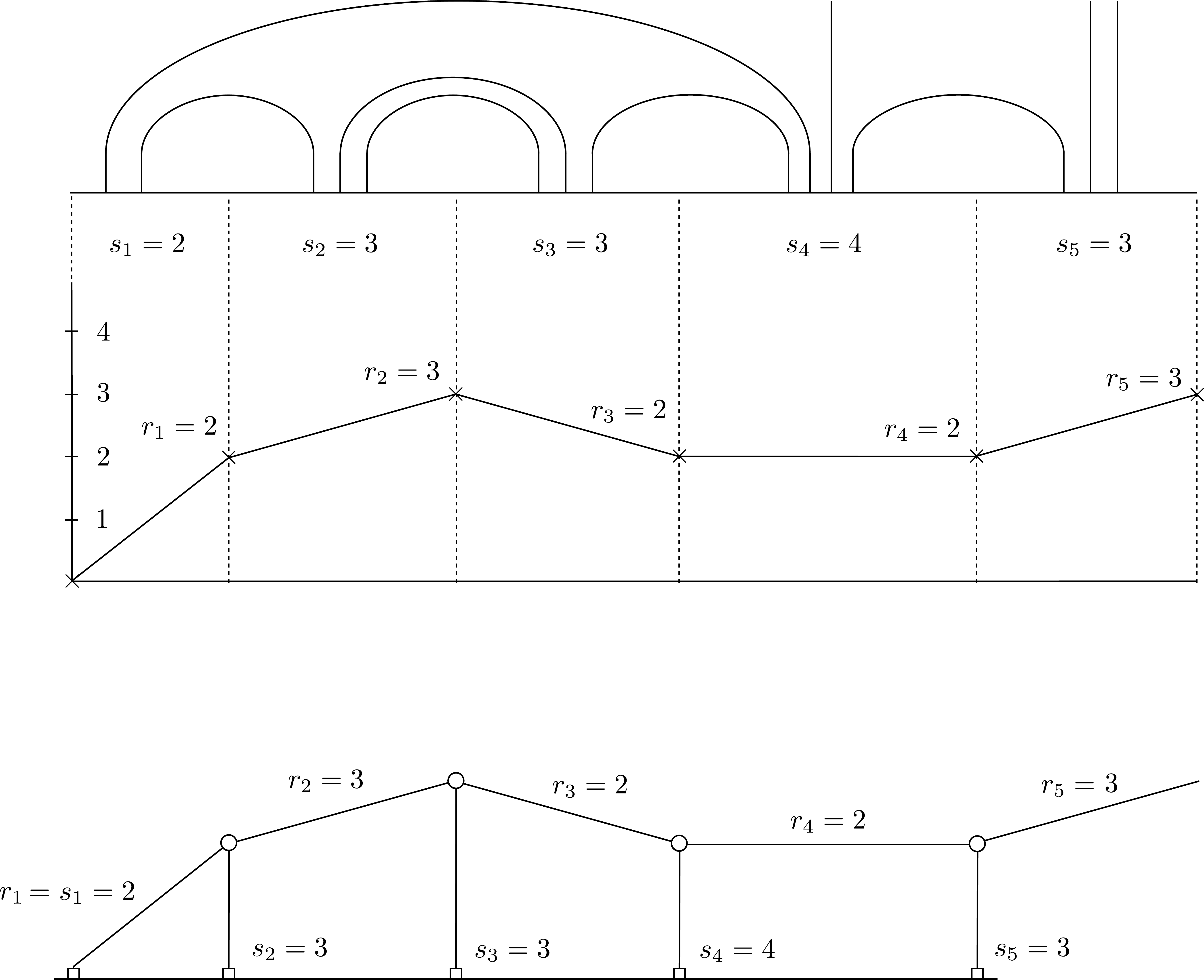} \vspace*{1cm}
\caption{\label{fig1}
A walk over $\multii = (2, 3, 3, 4, 3)$, representing a link pattern in $\mathrm{LP}_{(2, 3, 3, 4, 3)}$,
and the walk representation of the latter.}
\end{figure}

To begin, we define the trivalent link states under the assumption that $\Summed_\multii < \ppmin(q)$.  
In section~\ref{DetSect}, we use them to diagonalize the Gram matrix $\smash{\Gram_\multii\super{s}}$ and compute its determinant.
We present the definition first and explain later 
why the assumption that $\Summed_\multii < \ppmin(q)$ is needed.  Using the \emph{open three-vertex} notation~\cite{kl}
\begin{align}\label{3vertex2} 
\text{for $s \in \DefectSet\sub{r,t}$} , \qquad 
\vcenter{\hbox{\includegraphics[scale=0.275]{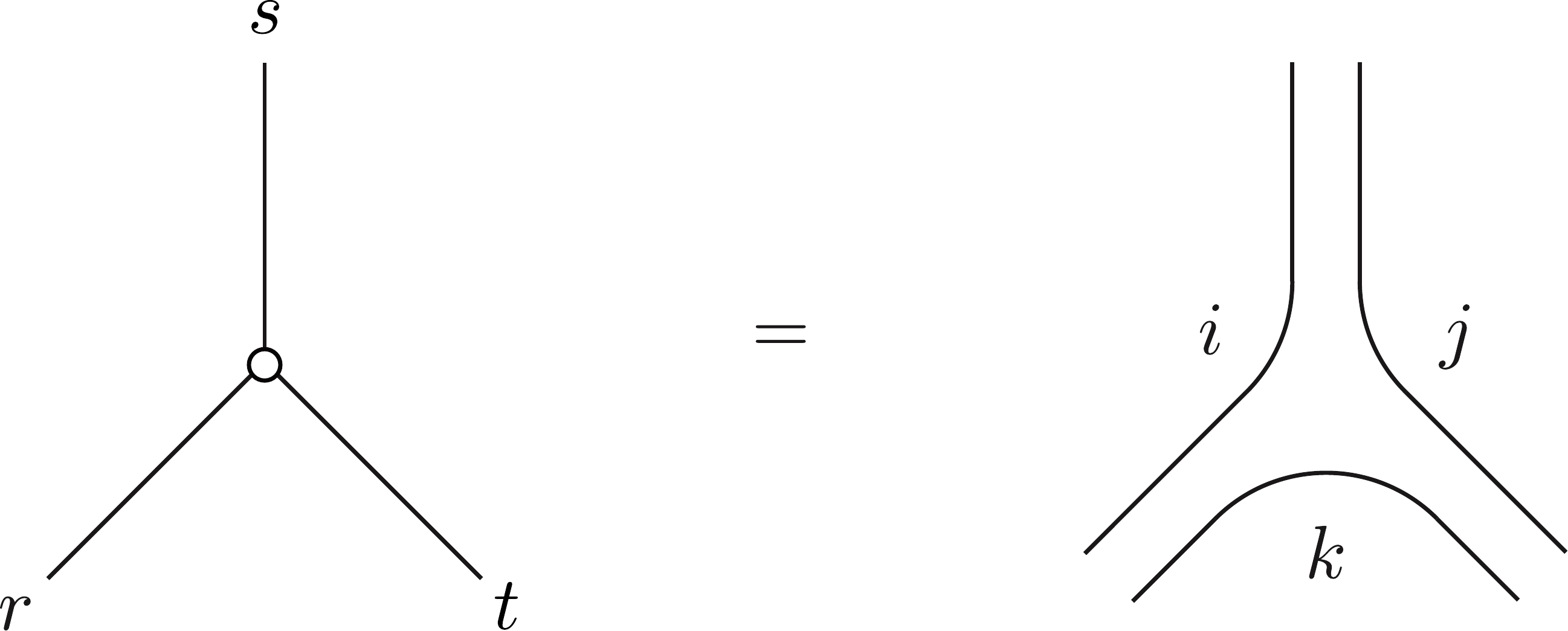} ,}} \qquad  \qquad \qquad
\begin{aligned} 
i & = \frac{r + s - t}{2} , \\[.7em] 
j & = \frac{s + t - r}{2} , \\[.7em] 
k & = \frac{t + r - s}{2} ,
\end{aligned}
\end{align}
we write each $\multii$-valenced link pattern $\alpha$ in the generic form of a trivalent graph with open vertices,
\begin{align}\label{Generic} 
\alpha \quad = \quad 
\vcenter{\hbox{\includegraphics[scale=0.275]{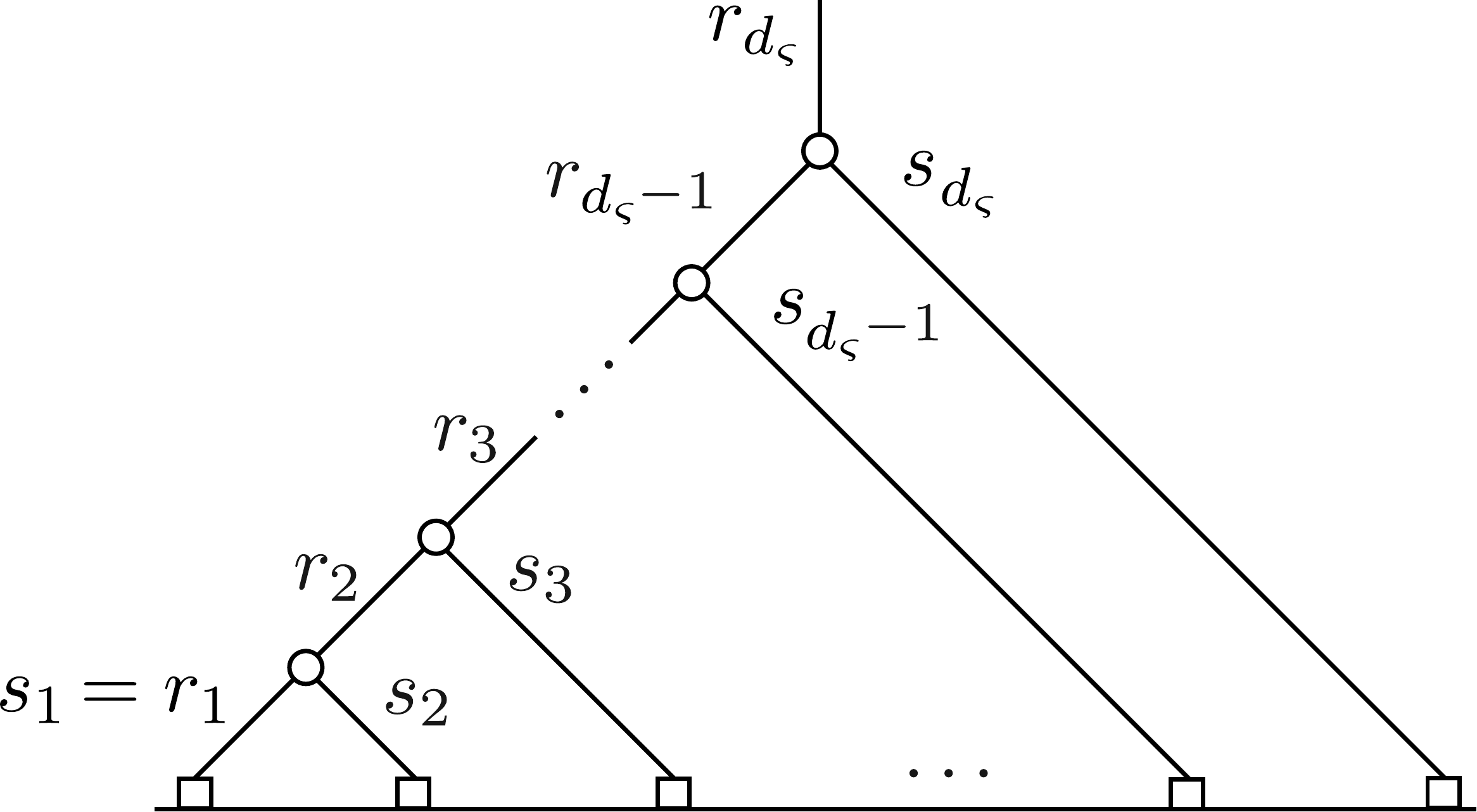} ,}}
\end{align}
which we call the \emph{walk representation} of $\alpha$, and we let 
\begin{align}\label{alphaWalk} 
\varrho_\alpha := (r_1, r_2, \ldots, r_{\np_\multii}) 
\end{align}
denote the multiindex of cable sizes in~\eqref{Generic}.  
The multiindex $\varrho_\alpha$ is a \emph{walk over $\multii = (\sIndex_1, \sIndex_2, \ldots, \sIndex_{\np_\multii})$}, 
that is, a multiindex $\varrho = (r_1, r_2, \ldots, r_{\np_\multii})$
whose entries, each called a \emph{height}, satisfy the following two conditions relative to $\multii$:
\begin{align}\label{WalkHeights} 
r_0 = 0 , \qquad r_{i+1} \in \DefectSet\sub{r_i,\sIndex_{i+1}} 
\overset{\eqref{SpecialDefSet}}{=} 
\{ |r_i - \sIndex_{i+1}|, |r_i - \sIndex_{i+1}| + 2, \ldots, r_i + \sIndex_{i+1} \}, \quad \text{for all $i \in \{1,2, \ldots, \np_{\multii}-1\}$} . 
\end{align}
As a notation convention, we do not explicitly show the zeroth height $r_0 = 0$ of the walk 
$\varrho = (r_1, r_2, \ldots, r_{\np_\multii})$ so the length of $\varrho$ matches that of $\multii$.  
However, it is convenient to implicitly include this entry for later use.  We observe that
\begin{align}
r_0 = 0 \quad \text{and} \quad r_1 \in \DefectSet\sub{r_0, \sIndex_1}
\qquad \overset{\eqref{SpecialDefSet}}{\Longrightarrow} \qquad \quad r_1 = \sIndex_1.
\end{align}

We visualize a walk by joining the points 
$(j, r_j)$ and $(j+1, r_{j+1})$ with a line segment, for each $j \in \{0,1,\ldots,\np_{\multii}-1\}$, as exemplified in figure~\ref{fig1}.
We refer to the $j$:th vertex as the $j$:th ``step" of the walk. 
For each walk $\varrho$ over $\multii$, there are two related walks over $\OneVec{n}_\multii$  that are useful to consider:
\begin{align}
\label{highest} 
\varrho\superscr{\, \uparrow} &:= \text{the unique highest walk over $\OneVec{n}_\multii$ that touches the walk $\varrho$ over $\multii$ at all steps of the latter}, \\
\label{lowest} 
\varrho\superscr{\, \downarrow} &:= \text{the unique lowest walk over $\OneVec{n}_\multii$ that touches the walk $\varrho$ over $\multii$ at all steps of the latter}. 
\end{align}
For example, figure~\ref{fig2} shows an illustration of the walks $\varrho\superscr{\, \uparrow}$ 
and $\varrho\superscr{\, \downarrow}$ relative to the walk $\varrho$ of figure~\ref{fig1}. 
For each walk $\varrho = (r_1, r_2, \ldots, r_{\np_\multii})$ over $\multii$, we also define the following quantities:
\begin{align}
\label{minmaxh} 
h_{\min, j}(\varrho) &:= 
\begin{cases} 
\dfrac{r_j + r_{j+1} - \sIndex_{j+1}}{2}, & j \in \{0, 1, \ldots, \np_\multii - 1\}, \\ r_{\np_\multii}, & j = \np_\multii, 
\end{cases} \\  
\label{minmaxh2} 
h_{\max, j}(\varrho) &:= 
\begin{cases} 
\dfrac{r_j + r_{j+1} + \sIndex_{j+1}}{2}, & j \in \{0, 1, \ldots, \np_\multii - 1\}, \\ r_{\np_\multii}, & j = \np_\multii ,
\end{cases} 
\end{align}
and for each $j \in \{0,1,\ldots,\np_\multii - 1\}$, we call $h_{\max, j}(\varrho)$ the \emph{apex} of the $(j+1)$:st step of $\varrho$, 
and we call the last apex
$h_{\max, \np_\multii}(\varrho) = r_{\np_\multii}$ the \emph{defect} of the walk $\varrho$.  
By definition, we have
\begin{align} \label{EasyCompare} 
0 \overset{\eqref{WalkHeights}}{\leq} h_{\min,j}(\varrho) \underset{\eqref{minmaxh2}}{\overset{\eqref{minmaxh}}{\leq} } h_{\max,j}(\varrho) ,
\end{align}
for all $j \in \{0,1,\ldots,\np_\multii\}$. Finally, for each $j \in \{1, 2, \ldots, \np_\multii\}$, using notation from~\eqref{cutmultii}, we define
\begin{align} 
\label{mudefn} 
\mu_j(\multii) : & = \max \{ \smin( \smash{\lds}_j ), \smin( \smash{\fds}_j ) \} , \\ 
\label{mumaxdefn} 
{\rm M}_j(\multii) : &= \min \{ \smax( \smash{\lds}_j ), \smax( \smash{\fds}_j ) \} 
\overset{\eqref{smaxeq}}{=} \min \{ \sIndex_1 + \sIndex_2 + \dotsm + \sIndex_j, \sIndex_j + \sIndex_{j+1} + \dotsm + \sIndex_{\np_\multii}\} .
\end{align}

\begin{figure}
\includegraphics[scale=0.275]{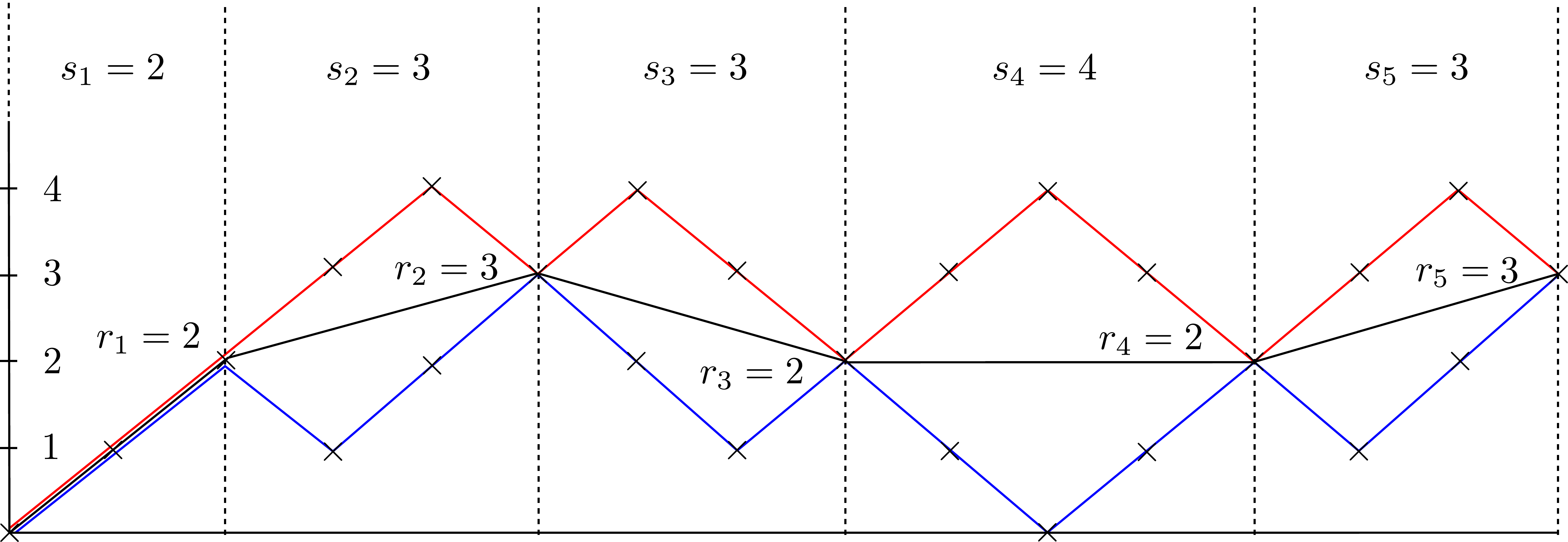}
\caption{\label{fig2}
The black walk is $\varrho$, the red walk is $\varrho^{\, \uparrow}$ and the blue walk is $\varrho^{\, \downarrow}$.}
\end{figure}

\begin{lem} \label{WalkMultiiLem} 
The following hold:
\begin{enumerate}
\itemcolor{red}
\item \label{wmlIt1} 
For each valenced link pattern $\alpha \in \LP_\multii$, the multiindex $\varrho_\alpha$, defined in~\eqref{alphaWalk}, is a walk over $\multii$.

\item \label{wmlIt2} 
The map $\alpha \mapsto \varrho_\alpha$ is a bijection from $\LP_\multii$ to the set of all walks over $\multii$.

\item \label{wmlIt3} 
For each valenced link pattern $\alpha \in \LP_\multii$, the defect of $\alpha$ equals the defect of $\varrho_\alpha$.

\item \label{wmlIt4} 
The number $\smash{\Dim_\multii\super{s}}$ and the set $\smash{\DefectSet_\multii}$, defined respectively 
in~\eqref{Recursion2} and~\eqref{DefSetDefn2}, satisfy
\begin{align} 
\label{CountLP} 
\Dim_\multii\super{s} = \; & \# \{ \textnormal{walks over $\multii$ with defect $s$} \}, \\
\label{AltDefectSet}
\DefectSet_\multii = \; & \{ s \in \bZnn \, | \, 
\textnormal{there exists a walk over $\multii$ with defect $s$} \} .
\end{align}

\item \label{wmlIt7} 
If there exists a walk $\varrho$ over $\multii$ with defect zero \textnormal{(}i.e., $0 \in \DefectSet_\multii$ by~\eqref{AltDefectSet}\textnormal{)}, 
then the set $\smash{\DefectSet_{\smash{\lds}_j} \cap \DefectSet_{\smash{\fds}_j}}$ is nonempty and equals
\begin{align}\label{pre-rj2} 
\DefectSet_{\smash{\lds}_j} \cap \DefectSet_{\smash{\fds}_j} = \{ \mu_j(\multii), \mu_j(\multii) + 2, \ldots, {\rm M}_j(\multii) \},
\end{align}
for each $j \in \{1,2,\ldots,\np_\multii\}$.
Furthermore, the $j$:th height of any walk $\varrho$ over $\multii$ with defect zero is an element of this set, and conversely, 
every element of this set equals the $j$:th height of some walk $\varrho$ over $\multii$ with defect zero.

\item \label{wmlIt6} 
For each $j \in \{1,2,\ldots,\np_\multii\}$,
at the $j$:th step, the minimum height over all walks $\varrho$ over $\multii$ that have defect zero equals
\begin{align}\label{minrj} 
\min_\varrho r_j = \mu_j(\multii) := \max \{ \smin( \smash{\lds}_j ), \smin( \smash{\fds}_j ) \}. 
\end{align}

\item \label{wmlIt5} 
For each walk $\varrho$ over $\multii$ and for each $j \in \{1,2,\ldots,\np_{\multii}-1\}$, 
\begin{enumerate}
\itemcolor{red}
\item[(a):]  \label{wmlIt5a} the maximal height of $\varrho\superscr{\, \uparrow}$ between the heights $r_j$ and $r_{j+1}$ of $\varrho$ equals $h_{\max, j}(\varrho)$, and

\item[(b):]  \label{wmlIt5b} the minimal height of $\varrho\superscr{\, \downarrow}$ between the heights $r_j$ and $r_{j+1}$ of $\varrho$ equals $h_{\min, j}(\varrho)$.
\end{enumerate}

\item \label{wmlIt9} 
For each $j \in \{0,1, \ldots, \np_\multii-1\}$, at the $(j+1)$:st step, the minimum apex over all walks $\varrho$ over $\multii$ with defect zero is
\begin{align}\label{minapex} 
\min_\varrho h_{\max, j}(\varrho) = \max \bigg\{ \frac{\mu_j(\multii) + \mu_{j+1}(\multii) + \sIndex_{j+1}}{2}, \, \sIndex_{j+1} \bigg\}. 
\end{align}
\end{enumerate}
\end{lem}

\begin{proof} 
We prove items~\ref{wmlIt1}--\ref{wmlIt9} as follows:
\begin{enumerate}[leftmargin=*]
\itemcolor{red}

\item From the definition~\eqref{3vertex2} of the open three-vertex, it is clear that the entries of $\varrho_\alpha = (r_1, r_2, \ldots, r_{\np_\multii})$ satisfy 
the defining conditions~\eqref{WalkHeights} of a walk over $\multii$.  

\item By item~\ref{wmlIt1}, for each valenced link pattern $\alpha \in \LP_\multii$,
the multiindex $\varrho_\alpha$~\eqref{alphaWalk} is a walk over $\multii$. 
On the other hand, we may substitute any walk $\varrho = (r_1, r_2, \ldots, r_{\np_\multii})$ over $\multii$ into~\eqref{Generic} 
to obtain a valenced link pattern $\alpha \in \LP_\multii$.  

\item Each link in $\alpha$ contributes no net gain to the defect of $\varrho_\alpha$, but each defect of $\alpha$ contributes a gain of one to $\varrho_\alpha$.  

\item Lemma~\ref{LSDimLem2} implies that $\#\smash{\LP_\multii\super{s}} = \smash{\Dim_\multii\super{s}}$, and items~\ref{wmlIt2} and~\ref{wmlIt3} imply that $\#\smash{\LP_\multii\super{s}}$ equals the number of walks over $\multii$ with defect $s$.  
Now,~\eqref{CountLP} follows from these two facts, and~\eqref{AltDefectSet} follows from this with definition~\eqref{DefSetDefn2} of $\DefectSet_\multii$.

\item If a walk $\varrho$ over $\multii$ has defect zero, then we split it into two pieces: $(r_1, r_2, \ldots, r_j)$, which is a walk over $\smash{\lds}_j$,
and $(r_{\np_\multii-1}, r_{\np_\multii-2}, \ldots, r_j)$, which is a walk over $\smash{ \tilde{\fds}_j }$.
As $r_j$ is the defect of either walk, item~\ref{wmlIt4} implies that
\begin{align}\label{pre-rj} 
r_j \overset{\eqref{AltDefectSet}}{\in} \DefectSet_{\smash{\lds}_j} \cap \DefectSet_{\tilde{\smash{\fds}}_j} 
\overset{\eqref{EqualDefSets}}{=} \DefectSet_{\smash{\lds}_j} \cap \DefectSet_{\smash{\fds}_j}. 
\end{align}
This shows that the set $\smash{\DefectSet_{\smash{\lds}_j} \cap \DefectSet_{\smash{\fds}_j}}$ is nonempty and, by~\eqref{DefSet2}, it has the form of~\eqref{pre-rj2}.  
Furthermore, for each element $s$ of this set,~\eqref{AltDefectSet} from 
item~\ref{wmlIt4} and observation~\eqref{EqualDefSets} implies that there exists a walk over 
$\smash{\lds}_j$ with defect $s$ and another walk over $\smash{ \tilde{\fds}_j }$ also with defect $s$. After reflecting the latter about a vertical axis and joining it to the former 
from the right, we obtain a walk $\varrho$ over $\multii = \smash{\lds}_j \oplus \smash{\fds}_j$ with defect zero and with its $j$:th height $r_j$ equaling $s$.

\item After combining (\ref{pre-rj2},~\ref{pre-rj}) and taking the minimum over all walks $\varrho$ over $\multii$ with defect zero, we infer that
\begin{align}\label{pre-rj3} 
\mu_j(\multii) := \max \{ \smin( \smash{\lds}_j ), \smin( \smash{\fds}_j ) \} \underset{\eqref{pre-rj2}}{\overset{\eqref{pre-rj}}{\leq}} \min_\varrho r_j. 
\end{align}
Because the left side is an element of $\smash{\DefectSet_{\smash{\lds}_j} \cap \DefectSet_{\smash{\fds}_j}}$, it follows from item~\ref{wmlIt7} that there exists 
a walk $\varrho$ over $\multii$ with defect zero and with its $j$:th height $r_j$ equaling this left side. Hence,~\eqref{pre-rj3} is really an equality, which gives~\eqref{minrj}.

\item Item~\ref{wmlIt5} can be proven by simple geometry.  

\item 
To begin, we observe that by definition~\eqref{WalkHeights}, for any walk $\varrho = (r_1, r_2, \ldots, r_{\np_\multii})$ over $\multii$, we have
$r_{j+1} \in \DefectSet\sub{r_j,\sIndex_{j+1}}$, which by  lemma~\ref{SpecialDefLem} translates to 
\begin{align} \label{sindexInSet}
\sIndex_{j+1} \in \DefectSet\sub{r_j,r_{j+1}} \overset{\eqref{SpecialDefSet}}{=} \{ |r_j - r_{j+1}|,  |r_j - r_{j+1}| + 2, \ldots,  r_j + r_{j+1} \} .
\end{align}
Therefore, we have 
\begin{align} \label{lowerbound2}
\min_\varrho h_{\max, j}(\varrho) 
\overset{\eqref{minmaxh2}}{=} \; &  \min_\varrho \left( \frac{r_j + r_{j+1} + \sIndex_{j+1}}{2} \right)
\overset{\eqref{sindexInSet}}{\geq}  \sIndex_{j+1} .
\end{align}
On the other hand, by item~\ref{wmlIt6}, we also have 
\begin{align} \label{lowerbound1}
\min_\varrho \left( \frac{r_j + r_{j+1} + \sIndex_{j+1}}{2} \right)
\geq \frac{( \min_\varrho r_j ) + ( \min_\varrho r_{j+1} ) + \sIndex_{j+1}}{2}
\overset{\eqref{minrj}}{=}  \frac{\mu_j(\multii) + \mu_{j+1}(\multii) + \sIndex_{j+1}}{2} ,
\end{align}
so altogether,~(\ref{lowerbound2}--\ref{lowerbound1}) imply
\begin{align} \label{lowerbound}
\min_\varrho h_{\max, j}(\varrho) \geq \max \bigg\{ \frac{\mu_j(\multii) + \mu_{j+1}(\multii) + \sIndex_{j+1}}{2}, \, \sIndex_{j+1} \bigg\}. 
\end{align}
Hence, to prove~\eqref{minapex}, it remains to prove the reverse inequality of~\eqref{lowerbound}. For this, we consider two cases:

\begin{enumerate}
\itemcolor{red}

\item[(a):]  $\sIndex_{j+1} \leq \mu_j(\multii) + \mu_{j+1}(\multii)$: 
Item~\ref{wmlIt7} shows that there exist two walks $\varrho = (r_1, r_2, \ldots, r_{\np_\multii})$ and
$\varrho' = (r_1', r_2', \ldots, r_{\np_\multii}')$ over $\multii$, both with defect zero and such that we have 
\begin{align}\label{r_rprime}
r_j = \mu_j(\multii) \qquad \text{and} \qquad r_{j+1}' = \mu_{j+1}(\multii) .
\end{align}
Then,~\eqref{sindexInSet} shows that $\sIndex_{j+1} \in \DefectSet\sub{r_j',r_{j+1}'}$, so item~\ref{wmlIt6} gives
\begin{align}\label{sindexsmall}
\mu_j(\multii) - \mu_{j+1}(\multii) \underset{\eqref{r_rprime}}{\overset{\eqref{minrj}}{\leq}} r_j' - r_{j+1}' \leq \sIndex_{j+1} 
\qquad \text{and} \qquad
\mu_{j+1}(\multii) - \mu_j(\multii) \underset{\eqref{r_rprime}}{\overset{\eqref{minrj}}{\leq}} r_{j+1} - r_j \leq \sIndex_{j+1} .
\end{align}
On the other hand, we observe that concatenating pieces of the walks $\varrho$ and $\varrho'$ we obtain a new walk
$(r_1, r_2, \ldots, r_j, r_{j+1}', r_{j+2}', \ldots, r_{\np_\multii}')$ over $\multii$ with defect zero: indeed, lemma~\ref{SpecialDefLem} shows that
\begin{align}
& \qquad |\mu_j(\multii) - \mu_{j+1}(\multii)| \overset{\eqref{sindexsmall}}{\leq} 
\sIndex_{j+1} \;  \leq \;  \mu_j(\multii) + \mu_{j+1}(\multii) \\
\overset{\eqref{SpecialDefSet}}{\Longrightarrow} & \qquad 
\sIndex_{j+1} \in \DefectSet\sub{\mu_j(\multii),\mu_{j+1}(\multii)} \overset{\eqref{SpecialDefSet}}{=}
\{ |\mu_j(\multii) - \mu_{j+1}(\multii)|, \ldots, \mu_j(\multii) + \mu_{j+1}(\multii) \} \\
\overset{\eqref{SameDefSet}}{\Longleftrightarrow} & \qquad
\mu_{j+1}(\multii) \in \DefectSet\sub{\mu_j(\multii),\sIndex_{j+1}} \\
\overset{\eqref{r_rprime}}{\Longleftrightarrow} & \qquad
r_{j+1}' \in \DefectSet\sub{r_j,\sIndex_{j+1}}.
\end{align}
Therefore, we obtain 
\begin{align}\label{upperbound}
\min_\varrho h_{\max, j}(\varrho) 
\overset{\eqref{minmaxh2}}{\leq} \frac{r_j + r_{j+1}' + \sIndex_{j+1}}{2} 
\overset{\eqref{r_rprime}}{=}  \frac{\mu_j(\multii) + \mu_{j+1}(\multii) + \sIndex_{j+1}}{2} .
\end{align}
By our assumption $\sIndex_{j+1} \leq \mu_j(\multii) + \mu_{j+1}(\multii)$, the right side of~\eqref{upperbound} equals 
the right side of~\eqref{minapex}. Thus, combining~\eqref{lowerbound} and~\eqref{upperbound} gives asserted equality~\eqref{minapex} in this case.

\item[(b):]  $\mu_j(\multii) + \mu_{j+1}(\multii) < \sIndex_{j+1}$: We observe that the sets
\begin{align}
\mathsf{E} := \DefectSet_{\smash{\lds}_j} \cap \DefectSet_{\smash{\fds}_j} \overset{\eqref{pre-rj2}}{=} \; & 
\{ \mu_j(\multii), \mu_j(\multii) + 2, \ldots, {\rm M}_j(\multii) \} , \\
\mathsf{F} := \sIndex_{j+1} -  (\DefectSet_{\smash{\lds}_{j+1}} \cap \DefectSet_{\smash{\fds}_{j+1}}) \overset{\eqref{pre-rj2}}{=} \; & 
\{ \sIndex_{j+1} - {\rm M}_{j+1}(\multii) , \sIndex_{j+1} - {\rm M}_{j+1}(\multii) +2 , \ldots, \sIndex_{j+1} - \mu_{j+1}(\multii) \}
\end{align}
satisfy $\mathsf{E} \cap \mathsf{F} \neq \emptyset$,
because we have $0 \leq \mu_j(\multii) \leq \sIndex_{j+1} - \mu_{j+1}(\multii)$ by assumption, and 
${\rm M}_j(\multii) \geq \sIndex_{j+1} - {\rm M}_{j+1}(\multii)$ by definition~\eqref{mumaxdefn}. 
(In fact, we have ${\rm M}_{j+1}(\multii) \geq \sIndex_{j+1}$.)
Hence, item~\ref{wmlIt7} shows that there exist two walks $\varrho = (r_1, r_2, \ldots, r_{\np_\multii})$ and
$\varrho' = (r_1', r_2', \ldots, r_{\np_\multii}')$ over $\multii$, both with defect zero and such that we have 
\begin{align} \label{r_rprime2}
r_j =  \sIndex_{j+1} -  r_{j+1}' .
\end{align}
Furthermore, the walk $(r_1, r_2, \ldots, r_j, r_{j+1}', r_{j+2}', \ldots, r_{\np_\multii}')$ obtained 
by concatenating pieces of the walks $\varrho$ and $\varrho'$ 
is a walk over $\multii$ with defect zero.
Therefore, we have 
\begin{align}\label{upperbound2}
\min_\varrho h_{\max, j}(\varrho) 
\overset{\eqref{minmaxh2}}{\leq} \frac{r_j + r_{j+1}' + \sIndex_{j+1}}{2} 
\overset{\eqref{r_rprime2}}{=}  \sIndex_{j+1} .
\end{align}
By our assumption $\mu_j(\multii) + \mu_{j+1}(\multii) < \sIndex_{j+1}$, the right side of~\eqref{upperbound2} equals 
the right side of~\eqref{minapex}. Thus, combining~\eqref{lowerbound} and~\eqref{upperbound2} gives asserted equality~\eqref{minapex} in this case.
\end{enumerate}
\end{enumerate}
This concludes the proof.
\end{proof}

Items~\ref{wmlIt7} and~\ref{wmlIt6} of lemma~\ref{WalkMultiiLem} may seem to be useful only in the 
case that $0 \in \DefectSet_\multii$. However, by lemma~\ref{SminLem}, we have
\begin{align} 
\begin{cases} 
s \in \DefectSet_\vartheta \\ \multii = \vartheta \oplus (s) 
\end{cases} 
\qquad \overset{\eqref{RecuFormula}}{\Longrightarrow} \qquad 
0 = \smin(\multii) \overset{\eqref{DefSet2}}{\in} \DefectSet_\multii, 
\end{align}
which allows us to adapt items~\ref{wmlIt7} and~\ref{wmlIt6} to useful statements when $0 \not \in \DefectSet_\vartheta$. We leave the details to the reader.

Now we construct an orthogonal basis of $\LS_\multii$, assuming that $\Summed_\multii < \ppmin(q)$.  
A key element of our construction is the following \emph{closed three-vertex} notation~\cite{kl, mv, cfs}:
\begin{align}\label{3vertex1} 
\text{for $s \in \DefectSet\sub{r,t}$} , \qquad 
\vcenter{\hbox{\includegraphics[scale=0.275]{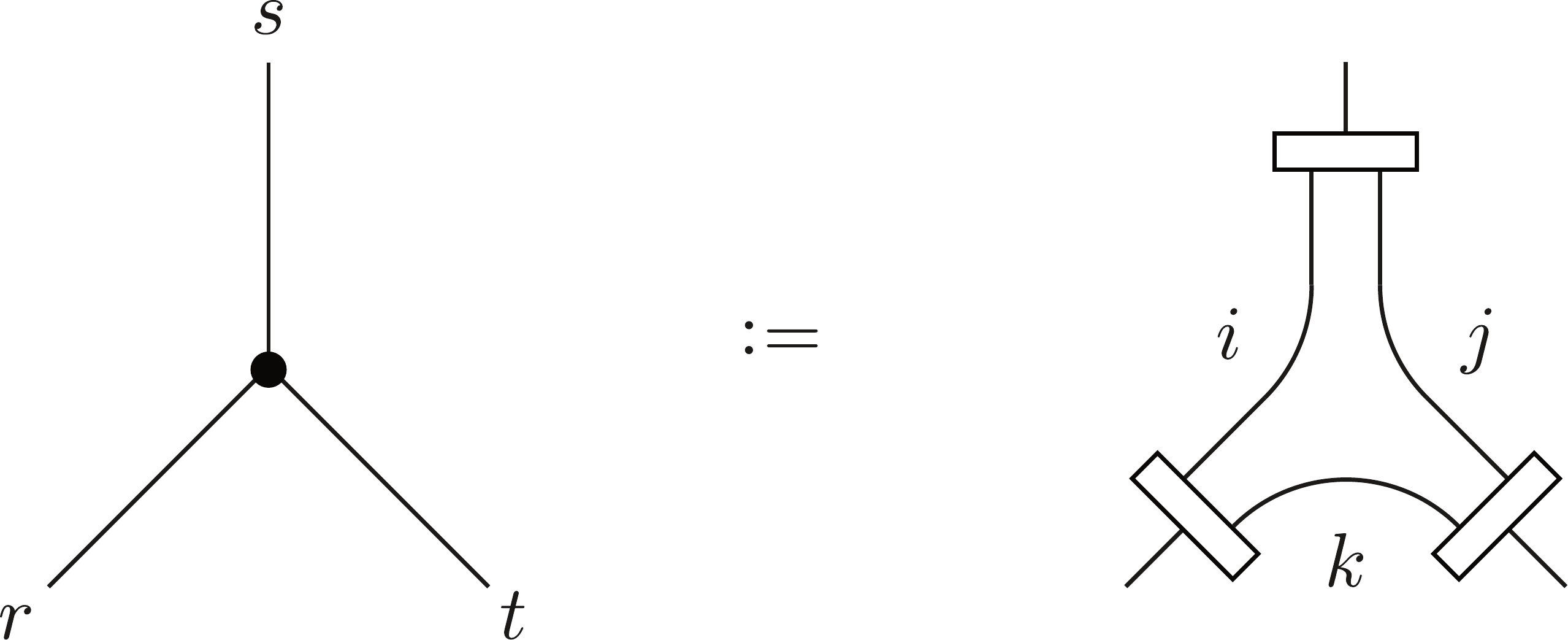} ,}} \qquad  \qquad \qquad
\begin{aligned} 
i & = \frac{r + s - t}{2} , \\[.7em] 
j & = \frac{s + t - r}{2} , \\[.7em] 
k & = \frac{t + r - s}{2}.
\end{aligned}
\end{align}
For each $\multii$-valenced link pattern $\alpha$, 
we define the \emph{trivalent link state} $\hcancel{\alpha}$ to be the following $\multii$-valenced link state:
\begin{align}\label{cbs} 
\hcancel{\alpha} \quad := \quad 
\vcenter{\hbox{\includegraphics[scale=0.275]{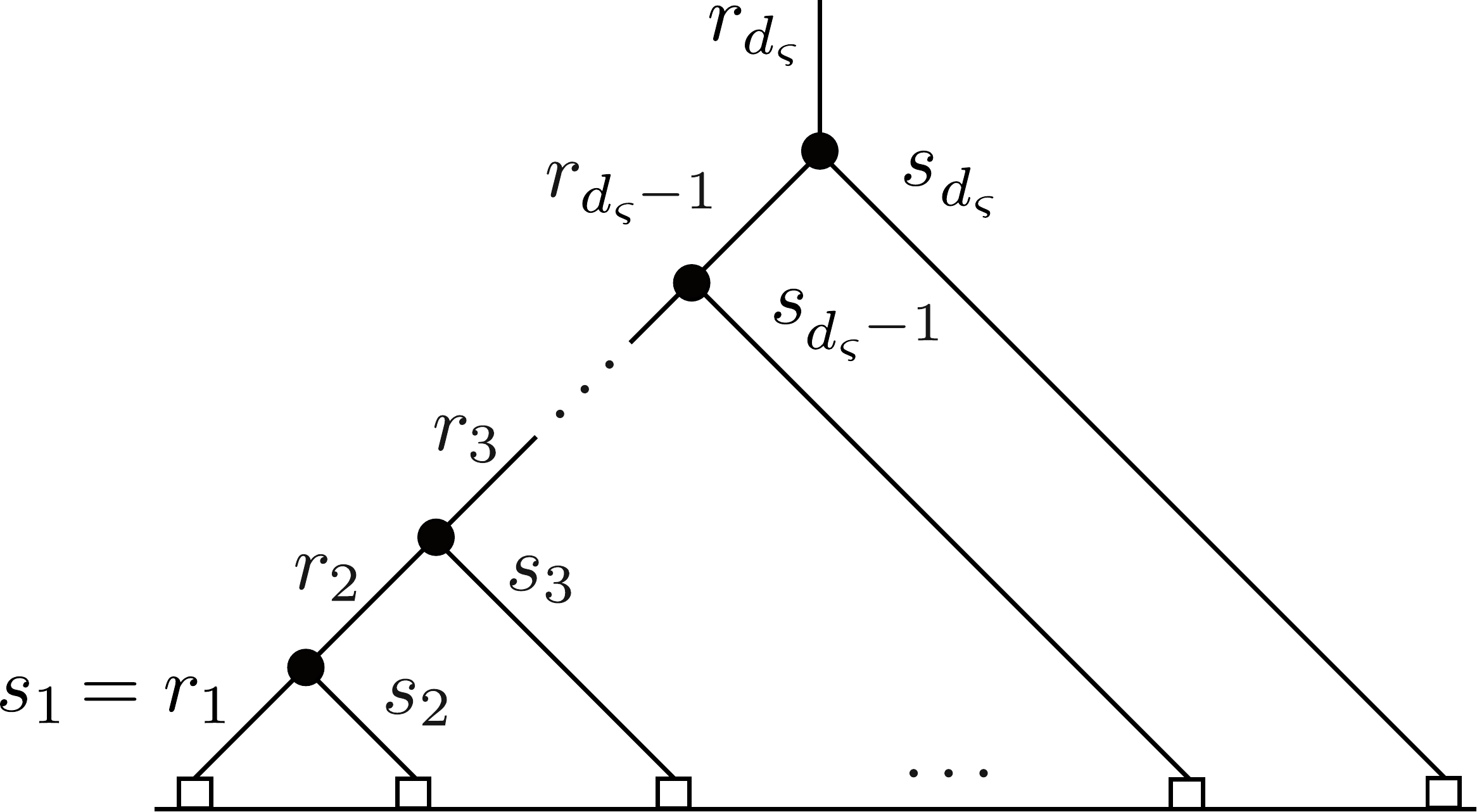} ,}} 
\end{align}
that is, to obtain $\hcancel{\alpha}$ from $\alpha$, we replace the $j$:th open vertex in the walk representation~\eqref{Generic} 
of the link pattern $\alpha$ with a closed vertex for each step $j \in \{1,2,\ldots,\np_{\multii}-1\}$ of the walk.
Replacing the $j$:th open vertex with a closed vertex inserts three projector boxes on the cables of respective sizes $r_j$, $\sIndex_{j+1}$, and $r_{j+1}$. 
Because we have
\begin{align}\label{rLessThanN} 
0 \leq r_j \underset{\eqref{WalkHeights}}{\overset{\eqref{SpecialDefSet}}{\leq}} r_{j-1} + \sIndex_j \underset{\eqref{WalkHeights}}{\overset{\eqref{SpecialDefSet}}{\leq}} r_{j-2} + \sIndex_{j-1} + \sIndex_j \underset{\eqref{WalkHeights}}{\overset{\eqref{SpecialDefSet}}{\leq}} \cdots \\
\cdots \underset{\eqref{WalkHeights}}{\overset{\eqref{SpecialDefSet}}{\leq}} \sIndex_1 + \sIndex_2 + \dotsm + \sIndex_j \overset{\eqref{ndefn}}{<} \Summed_\multii < \ppmin(q) ,
\end{align}
for all $j \in \{1,2,\ldots,\np_\multii-1\}$, these projector boxes do exist. 
According to rule~\eqref{TurnBack0}, when decomposing the projector boxes in the closed vertices,
we give weight zero to turn-back paths. Thus, it is evident that
\begin{align}\label{PreserveDef} 
\alpha \in \LP_\multii\super{s} \qquad \Longrightarrow \qquad \hcancel{\alpha} \in \LS_\multii\super{s}. 
\end{align}
Furthermore, we may freely omit the projector box across the $r_{\np_\multii} = s$ defects
(c.f. lemma~\ref{InsProjBoxLem} of appendix~\ref{TLRecouplingSect}).
Finally, by idempotent property~\eqref{ProjectorID0} of the Jones-Wenzl projector, the boxes of sizes $\sIndex_{j+1}$ are redundant.
In summary, the map $\alpha \mapsto \hcancel{\alpha}$ amounts to the insertion of projector boxes of sizes $r_2, r_3, \ldots, r_{\np_\multii-1}$
to the walk representation~\eqref{Generic} of $\alpha$.

\begin{example} \label{HcancelExample}
As a very simple example, we consider all $4$-link patterns:
\begin{align} \label{HcancelExampleLinPatt}
\includegraphics[scale=0.275]{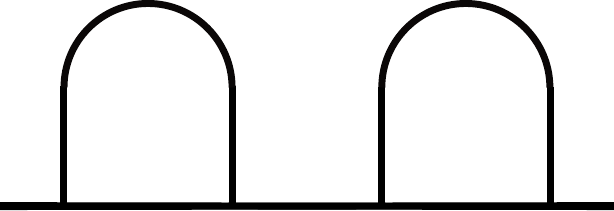},
\qquad
\raisebox{.2pt}{\includegraphics[scale=0.275]{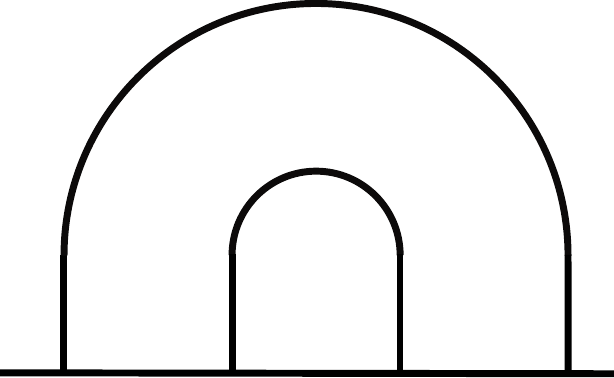},}
\qquad
\raisebox{.1pt}{\includegraphics[scale=0.275]{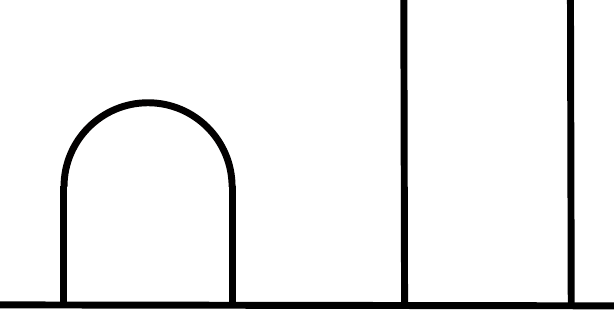},} 
\qquad
\raisebox{.1pt}{\includegraphics[scale=0.275]{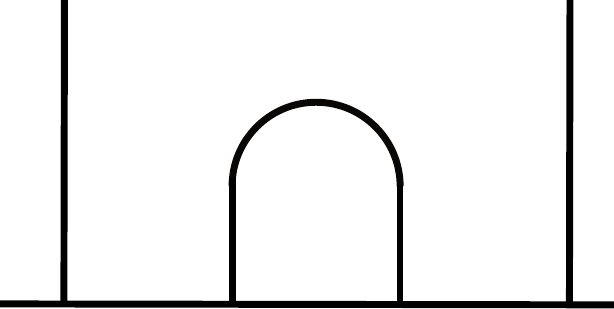},} 
\qquad
\raisebox{.1pt}{\includegraphics[scale=0.275]{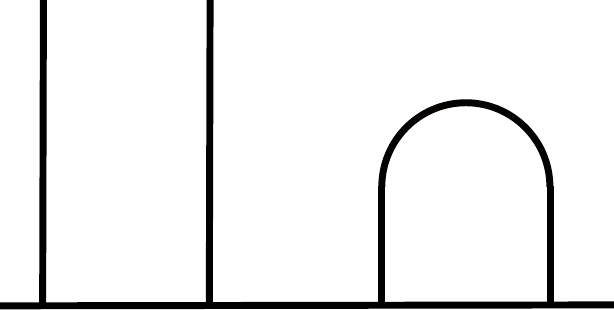},} 
\qquad
\raisebox{.1pt}{\includegraphics[scale=0.275]{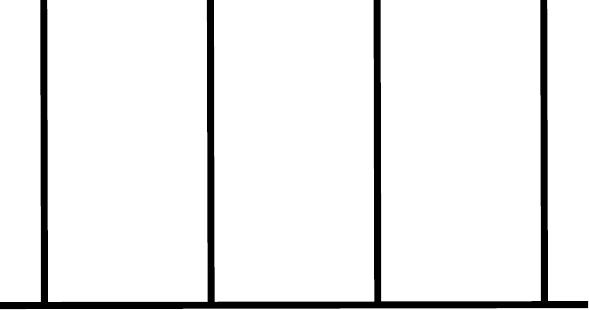} ,} 
\end{align}
whose walk representations are respectively 
$(1,0,1,0)$, $(1,2,1,0)$, $(1,0,1,2)$, $(1,2,1,2)$, $(1,2,3,2)$, and $(1,2,3,4)$.
The map $\alpha \mapsto \hcancel{\alpha}$ sends each of these link patterns respectively to the following link states:
\begin{align} \label{HcancelExampleLinPattHcan}
\includegraphics[scale=0.275]{e-pre_hcancel6.pdf},
\qquad
\raisebox{.2pt}{\includegraphics[scale=0.275]{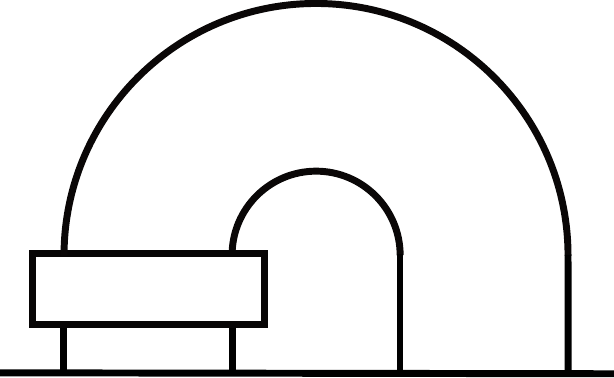},}
\qquad
\raisebox{.1pt}{\includegraphics[scale=0.275]{e-pre_hcancel2.pdf},}
\qquad
\raisebox{.1pt}{\includegraphics[scale=0.275]{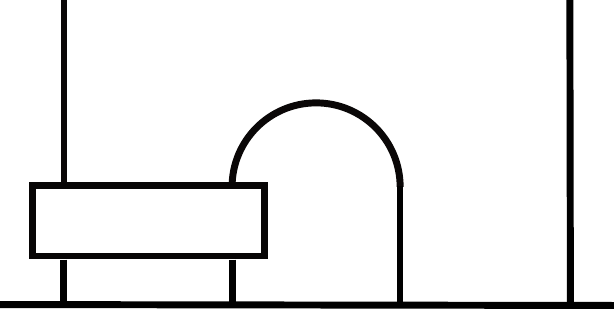},}
\qquad
\raisebox{.1pt}{\includegraphics[scale=0.275]{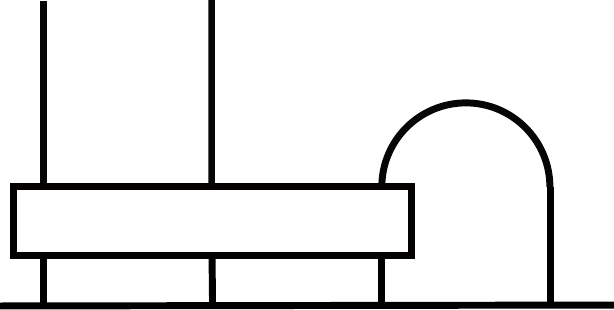},}
\qquad
\raisebox{.1pt}{\includegraphics[scale=0.275]{e-pre_hcancel4.pdf} .}
\end{align}
\end{example}

If $\Summed_\multii \geq \ppmin(q)$, then our definition of the trivalent link state $\hcancel{\,\alpha}$ 
may be invalid because it may use a projector box with size 
greater than $\ppmin(q)-1$, which does not exist.  
However, we can overcome this problem and give a useful definition for $\hcancel{\alpha}$ whenever $\max \multii < \ppmin(q)$.  
We use this definition to investigate the radical of $\smash{\LS_\multii\super{s}}$ in section~\ref{RadicalSect}.

To define trivalent link states $\hcancel{\,\alpha}$ that make sense under the weaker condition $\max \multii < \ppmin(q)$, we need some more terminology. 
First, we recall definitions (\ref{DeltaDefn},~\ref{skDefn}) of the symbols $\Delta_k$ and $R_s$ from section~\ref{Intro} and observe that
\begin{align} \label{skSet} 
s \overset{\eqref{DeltaDefn}}{\in} \{\Delta_{k_s}, \Delta_{k_s}+1, \ldots, \Delta_{k_s+1}-1 \}. 
\end{align}
Next, for each walk $\varrho = (r_1, r_2, \ldots, r_{\np_\multii})$ over $\multii$ with defect $s$, we let $J$ denote the special index 
\begin{align}\label{Jindex0} 
J = J_\varrho(q) := \sup \big\{ j \in \bZnn \; \, \big| \; \, 
\text{either} \quad h_{\min,j}(\varrho) \leq \Delta_{k_s}, \quad \text{or} \quad \Delta_{k_s+1} \leq h_{\max, j}(\varrho) \big\},
\end{align}
with the convention that $\sup \emptyset = -\infty$.  If $J \geq 0$, then 
we divide $\varrho$ into two pieces, called the \emph{head} and \emph{tail} of $\varrho$:
\begin{align}\label{TailDef} 
\varrho = (r_1, r_2, \ldots, r_{\np_\multii}) \qquad \Longrightarrow \qquad 
\left\{ 
\begin{aligned} 
\text{head}(\varrho) &:= (r_1, r_2, \ldots, r_{J-1}), \\ \text{tail}(\varrho) &:= (r_J, r_{J+1}, \ldots, r_{\np_\multii}), 
\end{aligned} 
\right. 
\qquad \text{where $J = J_\varrho(q)$.}
 \end{align}
We also define the head and tail of a link pattern $\alpha$ to be the head and tail of its corresponding walk $\varrho_\alpha$, and we write
\begin{align}\label{Jindex2} 
J = J_\alpha(q) := J_{\varrho_\alpha}(q), \qquad \text{head}(\alpha) := \text{head}(\varrho_\alpha), \qquad \text{and} \qquad \text{tail}(\alpha) := \text{tail}(\varrho_\alpha). 
\end{align}

\begin{figure}
\vspace*{1cm} \includegraphics[scale=0.275]{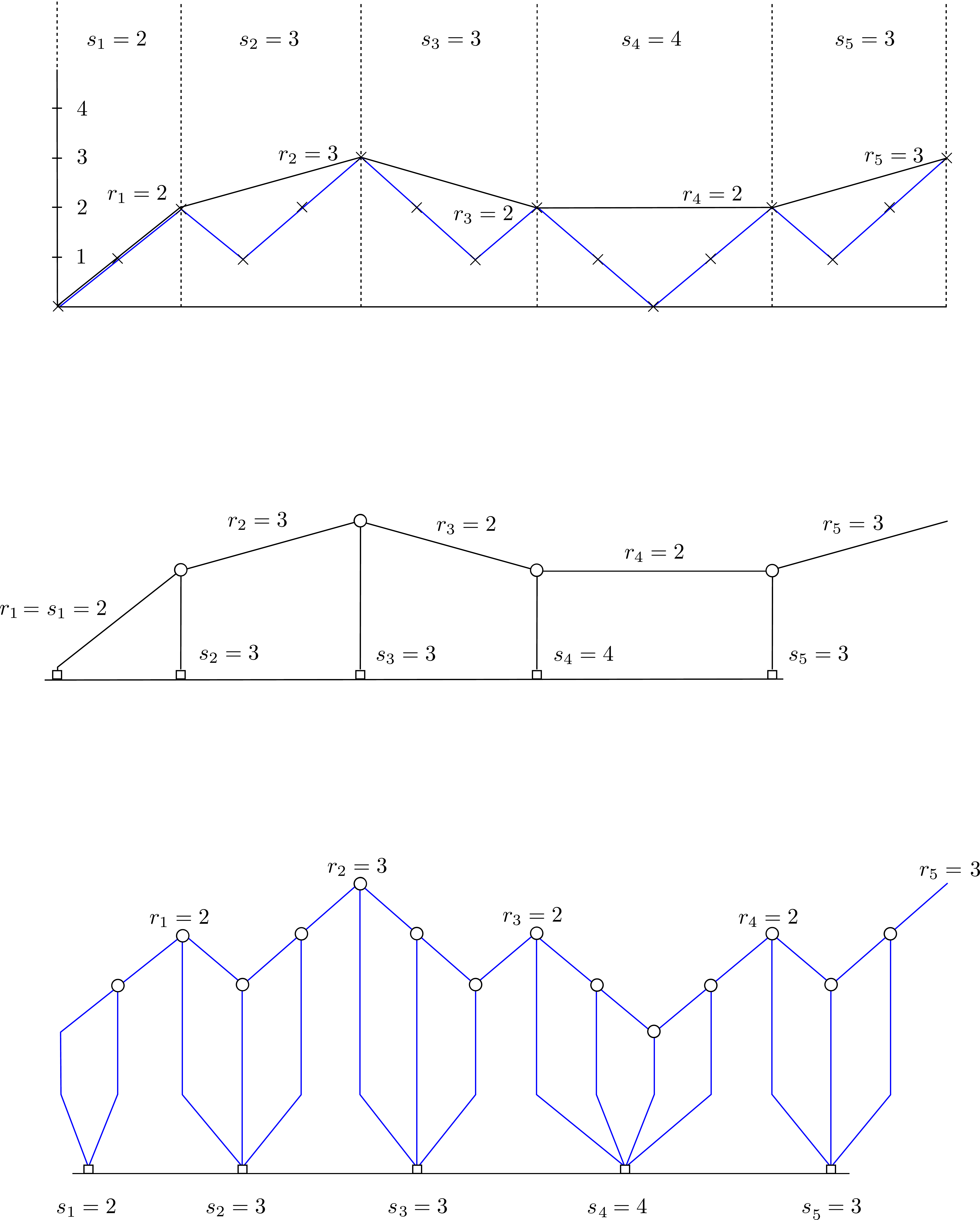} \vspace*{2cm}
\caption{\label{fig3}
Two walk representations of a link pattern $\alpha \in \mathrm{LP}_{(2, 3, 3, 4, 3)}$ and the associated walks:
the black walk is $\varrho_\alpha$ over $\multii = (2, 3, 3, 4, 3)$, associated to the walk representation of type~\eqref{Generic},
and the blue walk $\varrho^{\, \downarrow}$ over $\OneVec{\Summed_\multii} = \OneVec{15}$ gives an alternative walk representation
of $\alpha$ via replacements~\eqref{EquivPaths}.}
\end{figure}

Now we are ready to define the trivalent link states $\hcancel{\,\alpha}$, 
for all $\multii$-valenced link patterns $\alpha$ with $\max \multii < \ppmin(q)$.  

\begin{defn} \label{TrivalentLinkStateDef} 
\textnormal{
Suppose $\max \multii < \ppmin(q)$. 
For each $\multii$-valenced link pattern $\alpha$, we define the trivalent link state $\hcancel{\alpha}$ as follows.
First, we write the walk representation of 
$\alpha$ in a different way: for each $j \in \{1,2,\ldots,\np_{\multii}-1\}$, we replace the $j$:th open vertex of $\alpha$ with $\sIndex_{j+1}$ adjacent vertices as follows:
\begin{align}\label{EquivPaths} 
\vcenter{\hbox{\includegraphics[scale=0.275]{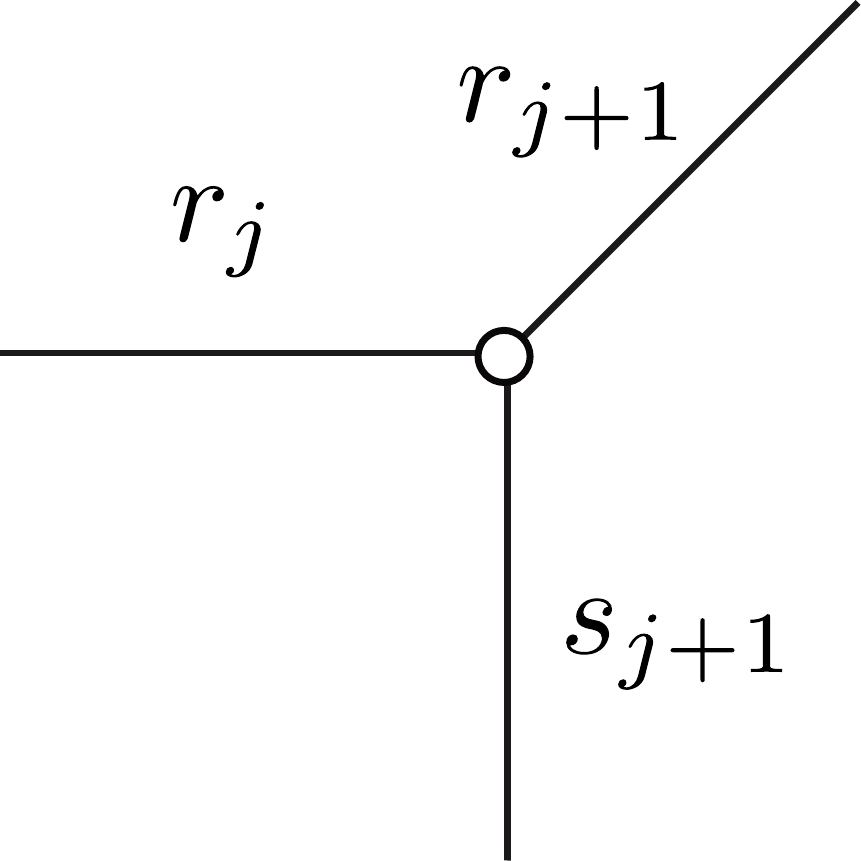}}} \qquad \qquad \longmapsto \qquad \qquad
\vcenter{\hbox{\includegraphics[scale=0.275]{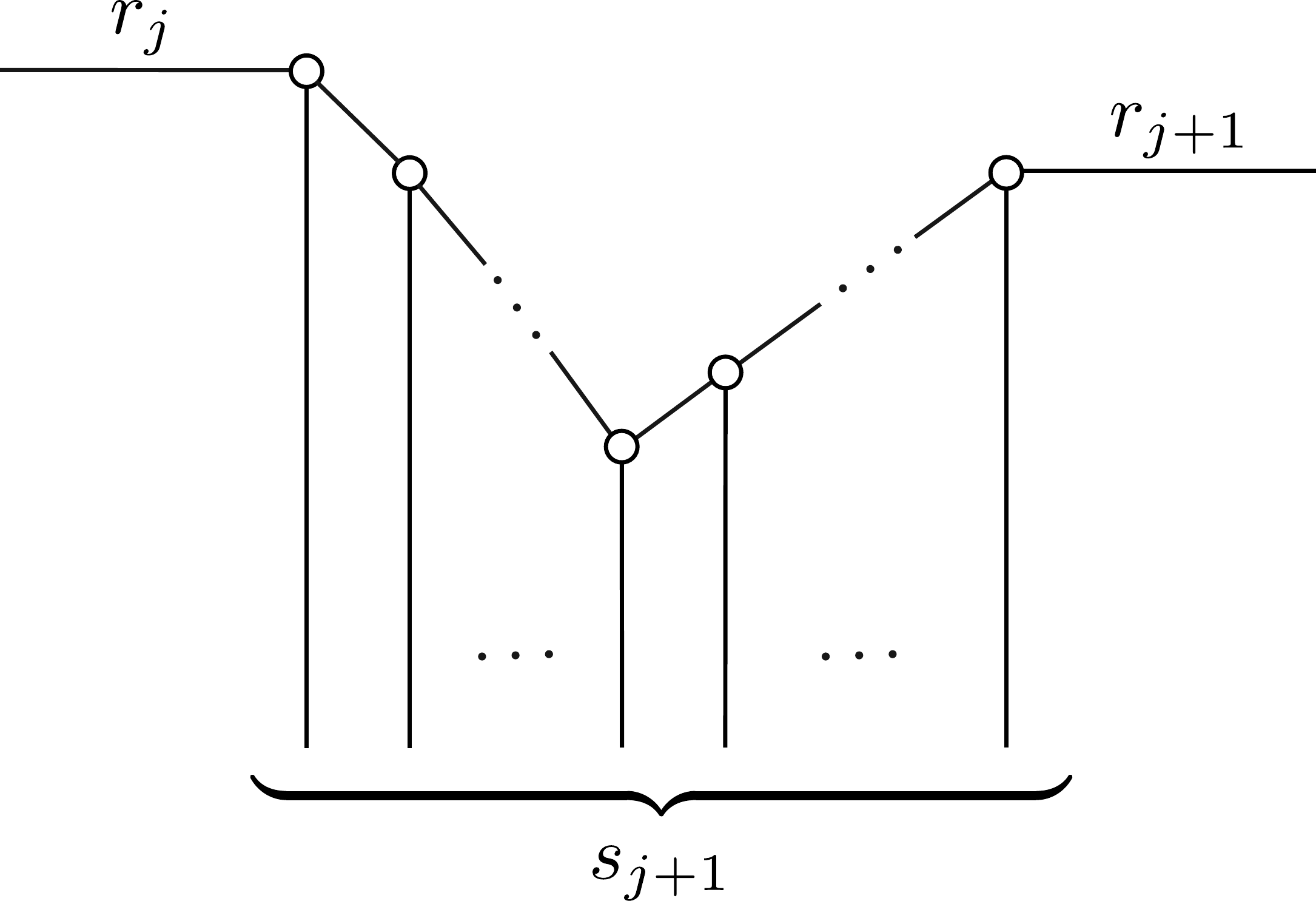} .}} 
\end{align}
As shown, each link from the cable of size $\sIndex_{j+1}$ enters its own open vertex.  
Making the replacement~\eqref{EquivPaths} at each open vertex in the original walk representation~\eqref{Generic} 
of $\alpha$ gives a new walk representation for $\alpha$.
Now, for each $\multii$-valenced link pattern $\alpha$, 
the lowest walk $\varrho\superscr{\, \downarrow}_\alpha$ is the path connecting all of the open vertices 
in this new walk representation of $\alpha$.  
Figure~\ref{fig3} shows an example of these two walk representations of a link pattern $\alpha$.
}

\textnormal{
Second, starting from the rightmost vertex and proceeding leftwards,
we replace each open vertex in the walk $\varrho\superscr{\, \downarrow}_\alpha$ of 
the new walk representation of $\alpha$ with a closed vertex.  
If $J = -\infty$, then we make this replacement at all vertices.  
Otherwise, the last vertex replacement occurs between the $J$:th and $(J+1)$:st steps of $\varrho_\alpha$ at the first time that
\begin{enumerate}
\itemcolor{red}
\item \label{StopIt1} we arrive at a step of the walk $\varrho\superscr{\, \downarrow}_\alpha$ whose height equals $\Delta_{k_s+1}$, or
\item \label{StopIt2} we arrive at the $J$:th step of the walk $\varrho_\alpha$, with height $r_J$ satisfying $\Delta_{k_s} < r_J < \Delta_{k_s+1}$, or
\item \label{StopIt3} we arrive at a step of the walk $\varrho\superscr{\, \downarrow}_\alpha$ whose height equals $\Delta_{k_s}$.
\end{enumerate}
After making the last vertex replacement, we arrive with $\hcancel{\alpha}$, which has thus been defined.
Figures~\ref{fig4-1},~\ref{fig4-2}, and~\ref{fig4-3} show examples of trivalent link states $\hcancel{\alpha}$ derived from each of these stopping conditions.
}
\end{defn}

If $J = 0$, then item~\ref{StopIt3} always gives the stopping condition. 
On the other hand, if $\Summed_\multii < \ppmin(q)$, then we have $J_\alpha(q) = -\infty$, 
for all $\multii$-valenced link patterns $\alpha$.
In this case, definition~\ref{TrivalentLinkStateDef} of $\hcancel{\alpha}$ reduces to the old definition~\eqref{cbs}, 
because definition~\eqref{3vertex1} of the closed three-vertex and property~\eqref{ProjectorID1} of the Jones-Wenzl projectors show that 
\begin{align}\label{EquivPathsClosed} 
\vcenter{\hbox{\includegraphics[scale=0.275]{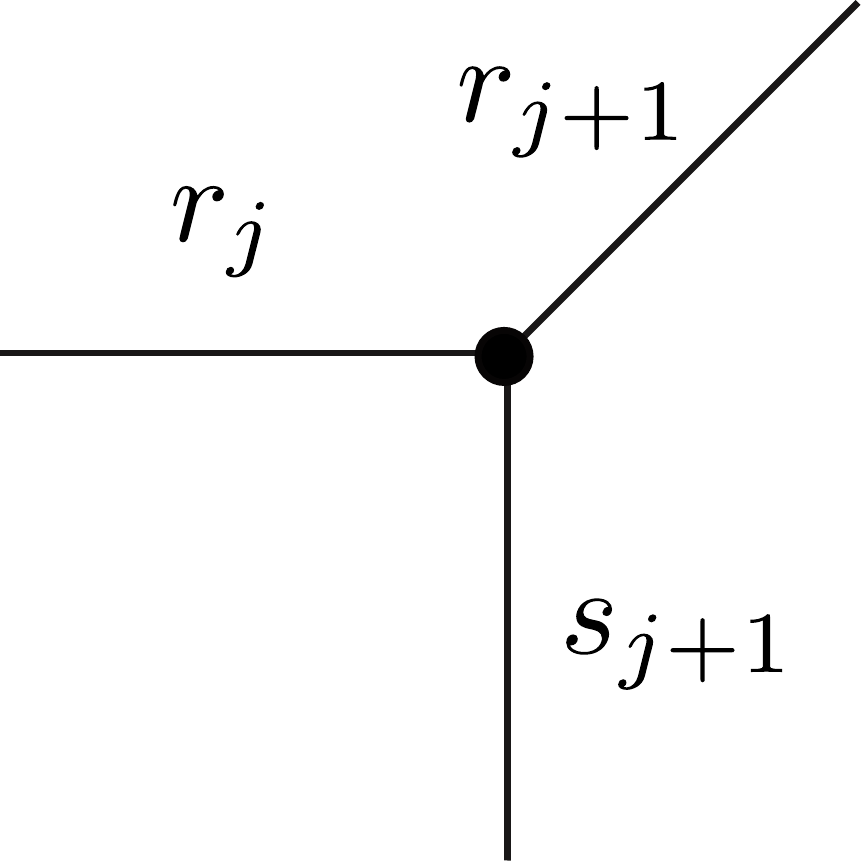}}} \qquad \qquad = \qquad \qquad
\vcenter{\hbox{\includegraphics[scale=0.275]{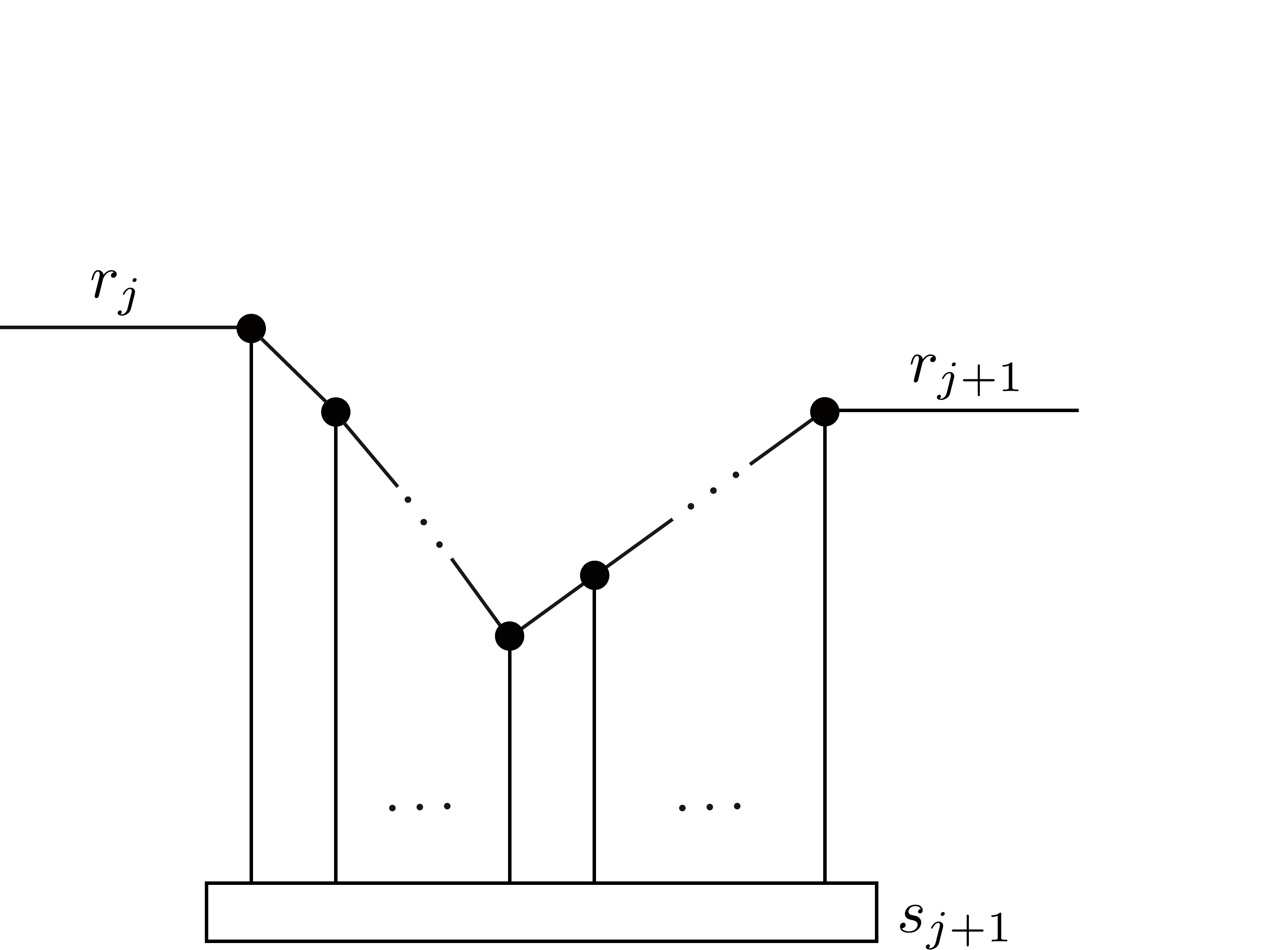} .}} 
\end{align}

\begin{figure}
\includegraphics[scale=0.275]{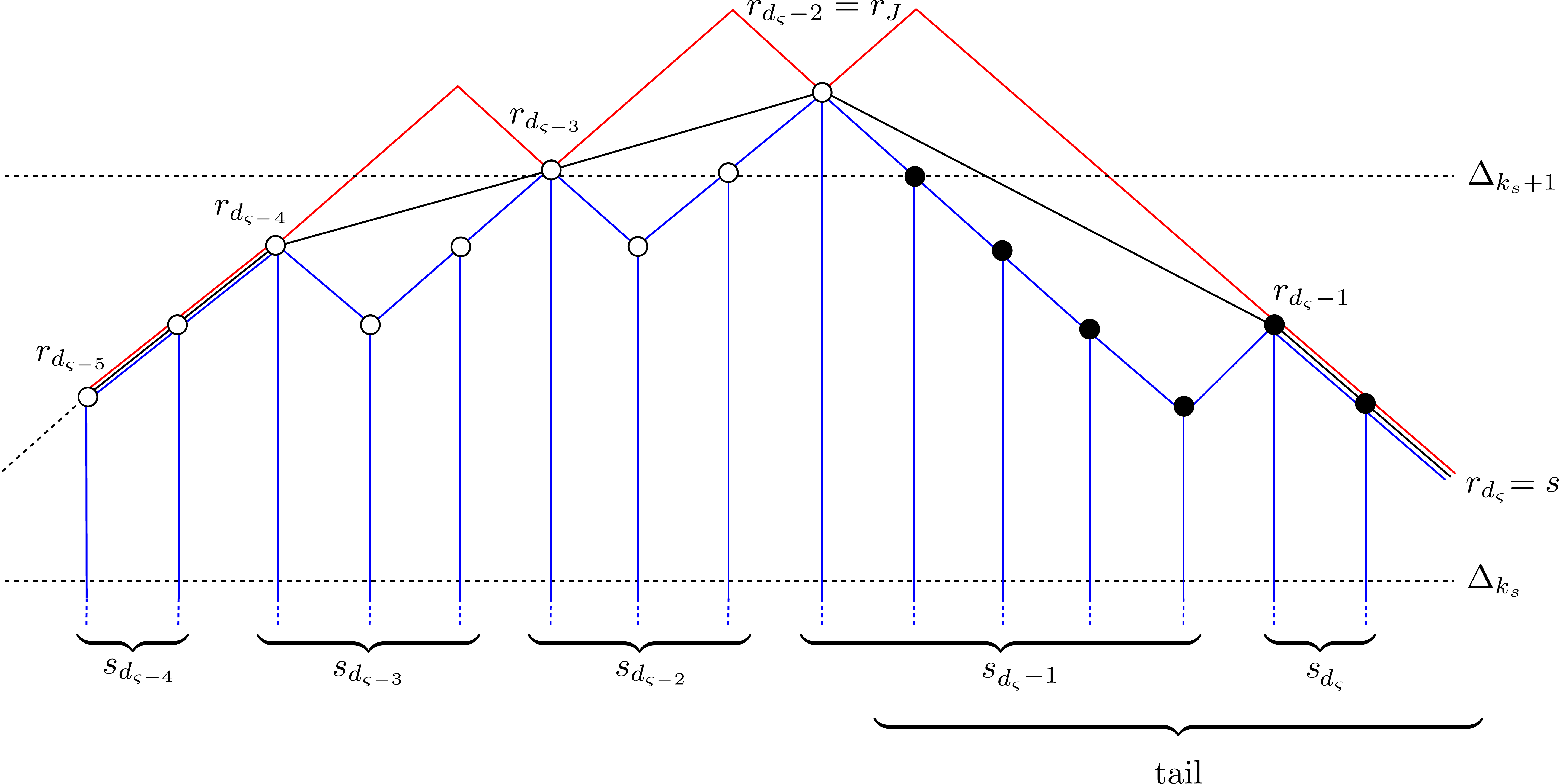} 
\caption{\label{fig4-1}
Tail of a trivalent link state $\hcancel{\alpha}$ associated to a $(\multii,s)$-valenced link pattern $\alpha$ when stopping condition~\ref{StopIt1} occurs.
The lowest walk $\varrho^{\, \downarrow}_\alpha$ and the associated walk representation is depicted in blue, the highest walk $\varrho^{\, \uparrow}_\alpha$ 
in red, and the walk $\varrho_\alpha$ in black.
}
\end{figure}

\begin{figure}
\includegraphics[scale=0.275]{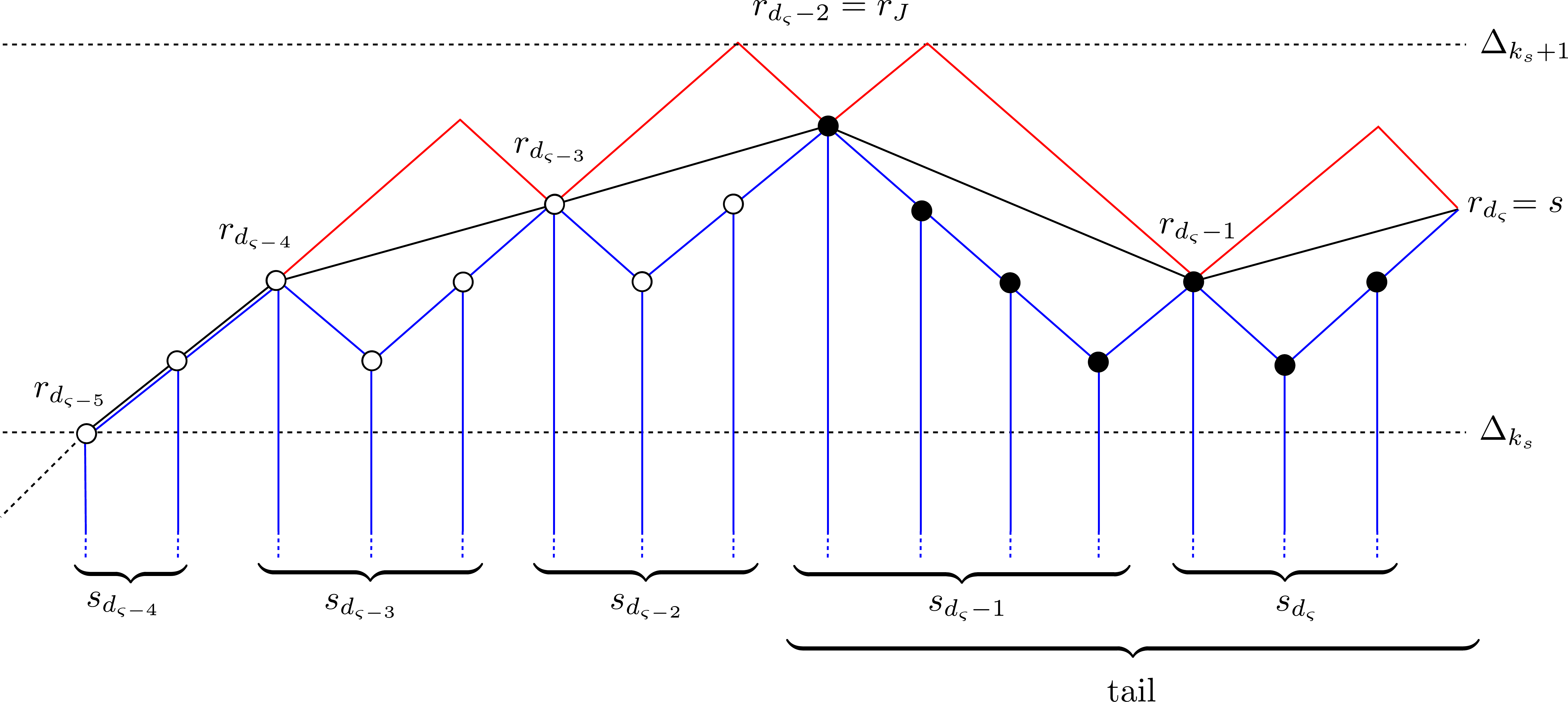} 
\caption{\label{fig4-2}
Tail of a trivalent link state $\hcancel{\alpha}$ associated to a $(\multii,s)$-valenced link pattern $\alpha$ when stopping condition~\ref{StopIt2} occurs.
The lowest walk $\varrho^{\, \downarrow}_\alpha$ and the associated walk representation is depicted in blue, the highest walk $\varrho^{\, \uparrow}_\alpha$ 
in red, and the walk $\varrho_\alpha$ in black.
}
\end{figure}

\begin{figure}
\includegraphics[scale=0.275]{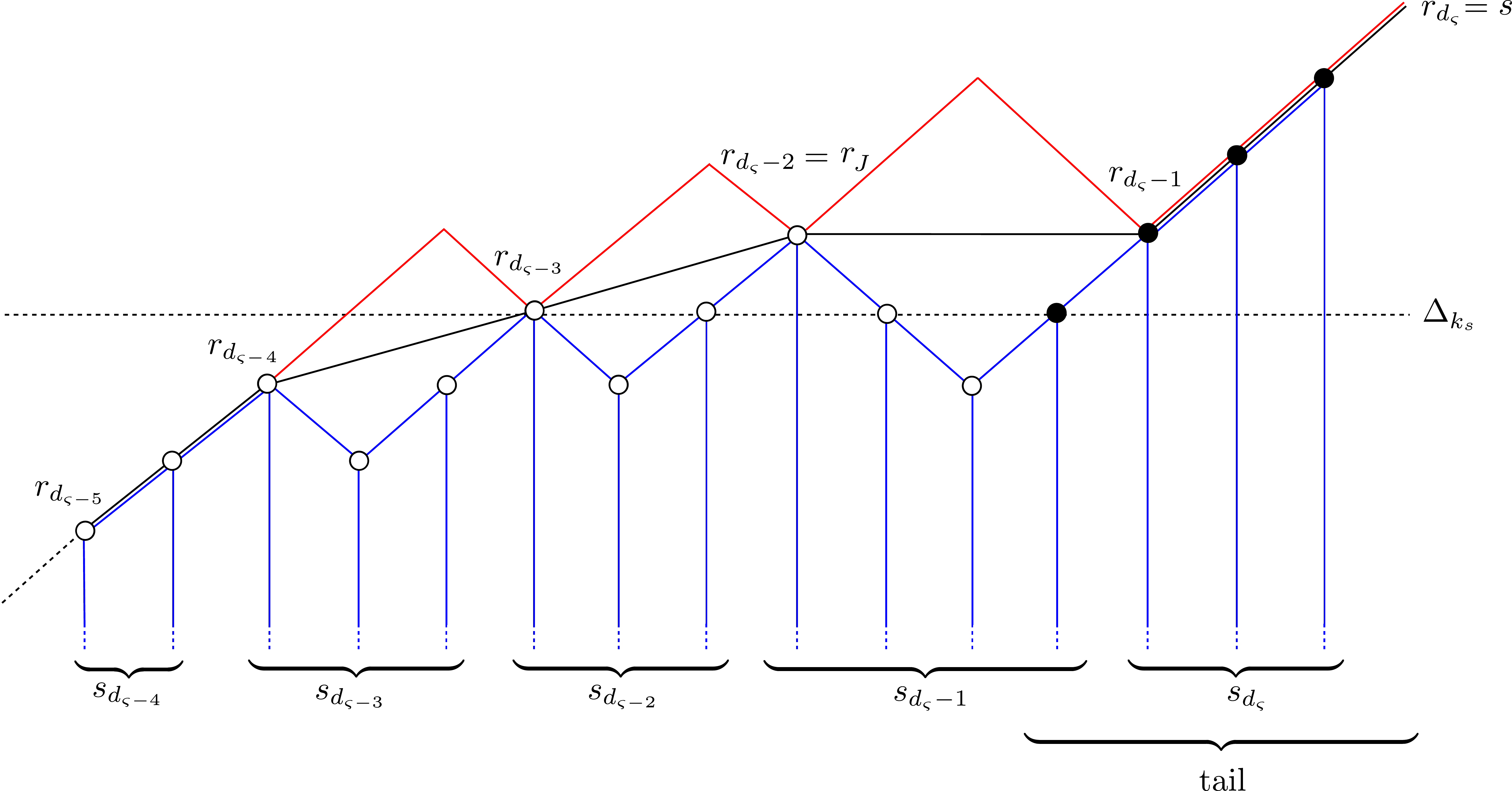} 
\caption{\label{fig4-3}
Tail of a trivalent link state $\hcancel{\alpha}$ associated to a $(\multii,s)$-valenced link pattern $\alpha$ when stopping condition~\ref{StopIt3} occurs.
The lowest walk $\varrho^{\, \downarrow}_\alpha$ and the associated walk representation is depicted in blue, the highest walk $\varrho^{\, \uparrow}_\alpha$ 
in red, and the walk $\varrho_\alpha$ in black.
}
\end{figure}

\begin{remark} \label{TrivLinkStateRem}
In definition~\ref{TrivalentLinkStateDef} of $\hcancel{\alpha}$, we do not require $\Summed_\multii < \ppmin(q)$, 
but only $\max \multii < \ppmin(q)$.
However, we cautiously note that if $\Summed_\multii \geq \ppmin(q)$, then the trivalent link state $\hcancel{\alpha}$ 
may contain projector boxes of size at least $\ppmin(q)$, which are undefined.  
Nevertheless, in this situation, we can define $\hcancel{\alpha}$ by analytic continuation.
We let $\hcancel{\alpha}_{q'}$ denote 
the trivalent link state $\hcancel{\alpha}$ with $q$ perturbed to some value $q'  \in \bC^\times$ with $\ppmin(q') = \infty$ (so projector boxes of all sizes exist) while holding 
$J$ fixed (so $J = J_\alpha(q) \neq J_\alpha(q')$).  Then we 
define $\hcancel{\alpha}$ to be the limit of $\hcancel{\alpha}_{q'}$ as $q' \to q$ along a sequence not containing roots of unity. 
We show that this limit exists in appendix~\ref{RadicalAppendix}, lemma~\ref{LimitLem}, justifying the validity of 
definition~\ref{TrivalentLinkStateDef}. 
\end{remark}

\subsection{Properties of the trivalent link states} \label{ConformalBlocksProp}

In this section, we determine salient properties of the trivalent link states.  
Most importantly, we show in item~\ref{IndOrthBasisLemIt3} of proposition~\ref{IndOrthBasisLem} that 
the set $\smash{\{ \hcancel{\alpha} \, | \, \alpha \in \smash{\LP_\multii\super{s}} \}}$ is a basis for $\smash{\LS_\multii\super{s}}$, 
and if $\Summed_\multii < \ppmin(q)$,  then this basis is orthogonal.
For this purpose, we first show 
in lemmas~\ref{ChangeOfBasisLem} and~\ref{InsProjAlphaLem} how operations of adding projector boxes to valenced link states
can be viewed as linear maps with upper-triangular matrix representations whose diagonal entries equal one (upper-unitriangular). 
As a simple example, in figure~\ref{table} we illustrate the map $\alpha \mapsto \hcancel{\alpha}$ associated to example~\ref{HcancelExample}.

\begin{figure}
\begin{displaymath}
\begin{tabular}{c c c c c c c}
\qquad & \qquad
\includegraphics[scale=0.275]{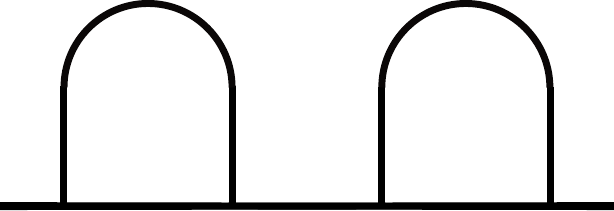} \; & \;
\raisebox{.2pt}{\includegraphics[scale=0.275]{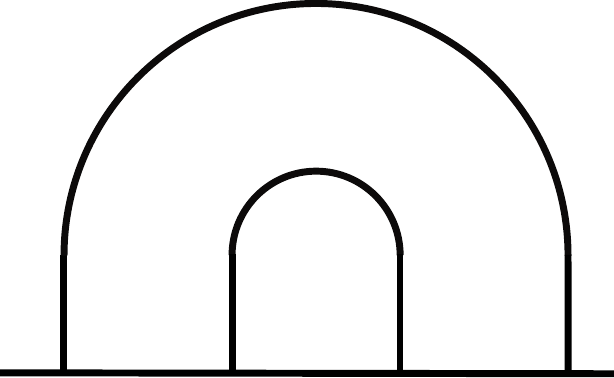}} \; & \;
\raisebox{.1pt}{\includegraphics[scale=0.275]{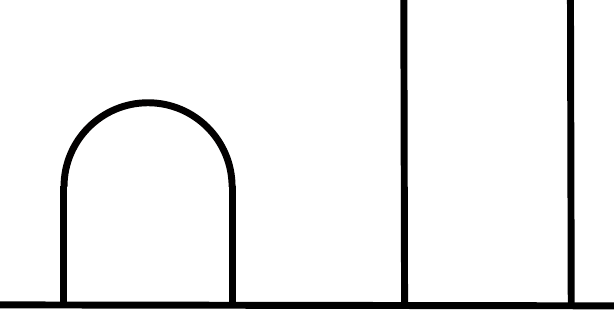}}  \; & \;
\raisebox{.1pt}{\includegraphics[scale=0.275]{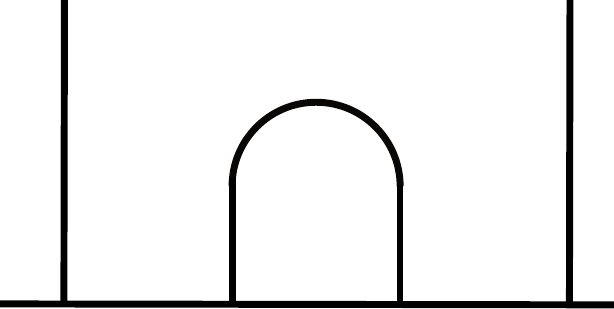}}  \; & \;
\raisebox{.1pt}{\includegraphics[scale=0.275]{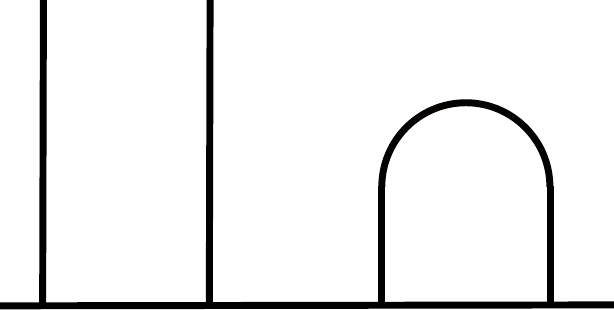}}  \; & \;
\raisebox{.1pt}{\includegraphics[scale=0.275]{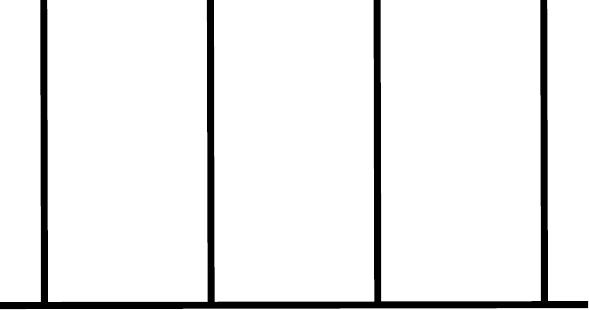}}  \\[1em]
\includegraphics[scale=0.275]{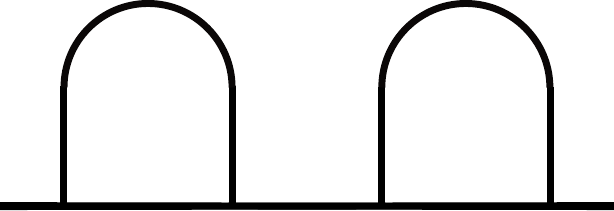} & {\Large{1}} & {\Large{$\frac{1}{[2]}$}} & \large{0} & \large{0} & \large{0} & \large{0} \\[1em]
\raisebox{.2pt}{\includegraphics[scale=0.275]{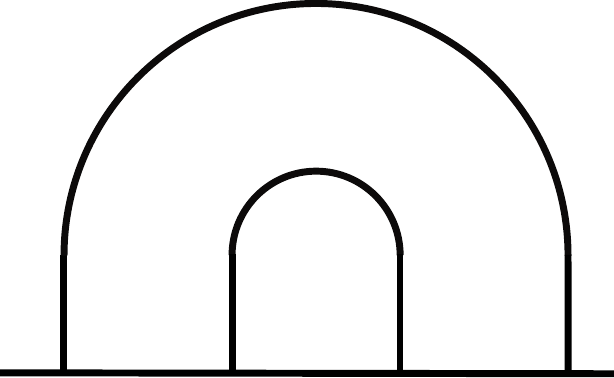}}  & \large{0} & {\Large{1}} & \large{0} & \large{0} & \large{0} & \large{0} \\[1em]
\raisebox{.1pt}{\includegraphics[scale=0.275]{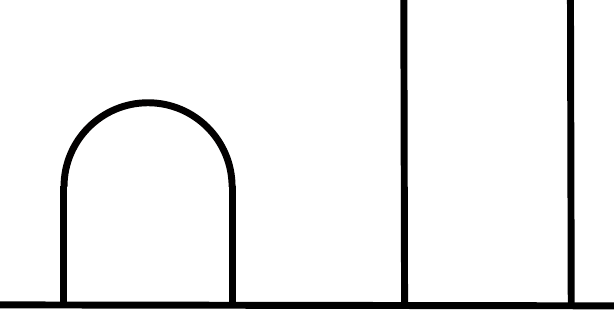}} & \large{0} & \large{0} & {\Large{1}} & {\Large{$\frac{1}{[2]}$}} & {\Large{$\frac{1}{[3]}$}} & \large{0} \\[1em]
\raisebox{.1pt}{\includegraphics[scale=0.275]{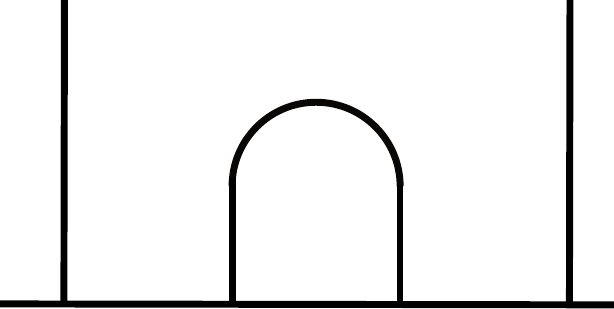}} & \large{0} & \large{0} & \large{0} & {\Large{1}} & {\Large{$\frac{[2]}{[3]}$}} & \large{0} \\[1em]
\raisebox{.1pt}{\includegraphics[scale=0.275]{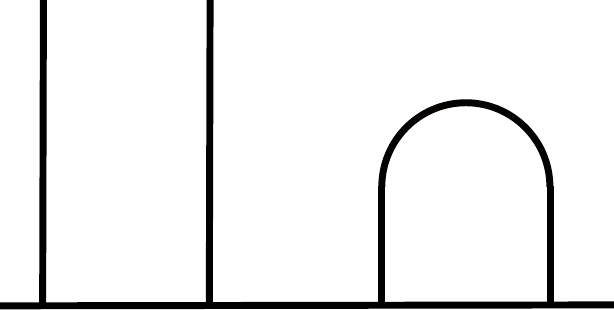}} & \large{0} & \large{0} & \large{0} & \large{0} & {\Large{1}} & \large{0} \\[1em]
\raisebox{.1pt}{\includegraphics[scale=0.275]{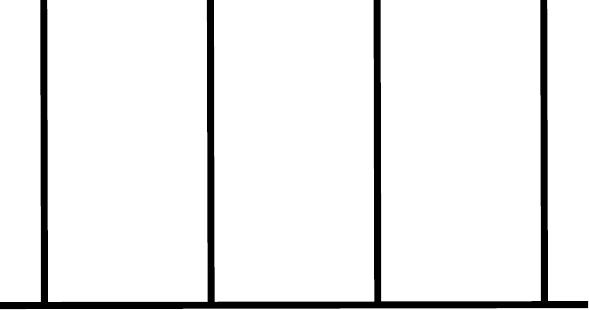}} & \large{0} & \large{0} & \large{0} & \large{0} & \large{0} & {\Large{1}} \\
\end{tabular}
\end{displaymath}
\caption{\label{table}
Matrix representation of the map $\alpha \mapsto \hcancel{\alpha}$ in example~\ref{HcancelExample}.
We can order the $4$-link patterns~\eqref{HcancelExampleLinPatt} is such a way that the matrix representation of the map
$\alpha \mapsto \hcancel{\alpha}$ is 
upper-unitriangular.
Furthermore, this matrix is block-diagonal, and its blocks correspond to the numbers $s \in \{0,2,4\}$
of defects of $\alpha \in \LP_4$.}
\end{figure}

\begin{lem} \label{ChangeOfBasisLem}
Suppose $\max \multii < \ppmin(q)$. Let $\mathsf{B}_\multii$ be a basis for $\LS_\multii$, 
all of whose elements $\alpha$ may be written in the form
\begin{align}\label{LPrep} 
\vcenter{\hbox{\includegraphics[scale=0.275]{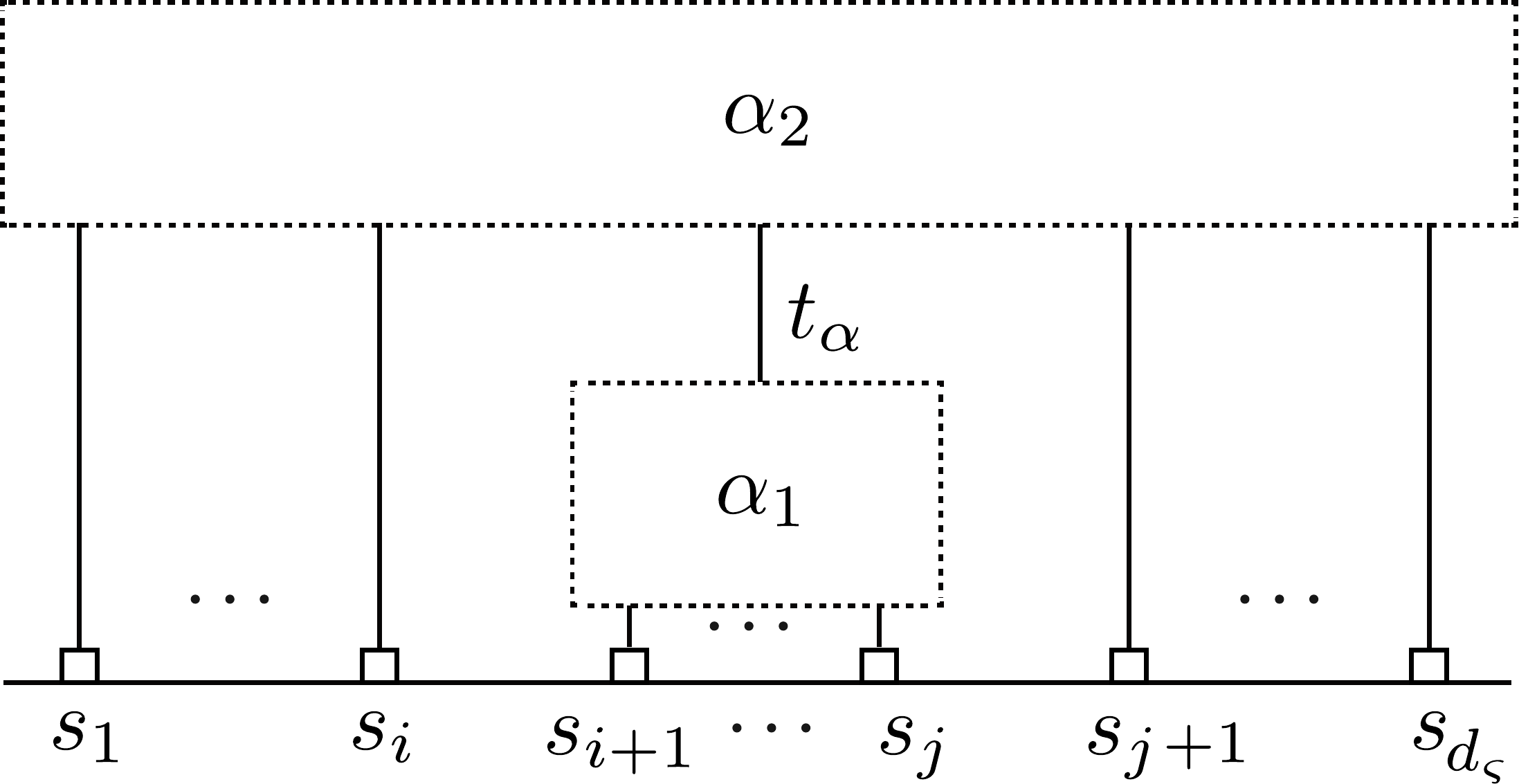} ,}} 
\end{align}
for some integers $i \in \{1,2,\ldots,\np_\multii-1\}$ and $j \in \{i+1, i+2, \ldots, \np_\multii\}$ common to all elements of $\mathsf{B}_\multii$, and 
\textnormal{(}with $\np = \np_\multii$\textnormal{)}
\begin{align} 
t_\alpha \in \DefectSet_{(\sIndex_{i+1}, \sIndex_{i+2}, \ldots, \sIndex_j)}, 
\qquad \alpha_1 \in \LS_{(\sIndex_{i+1}, \sIndex_{i+2}, \ldots, \sIndex_j)}^{(t_\alpha)}, 
\qquad \alpha_2 \in \LS_{(\sIndex_1, \sIndex_2, \ldots, \sIndex_i, t_\alpha, \sIndex_{j+1}, \sIndex_{j+2}, \ldots, \sIndex_{\np})} .
\end{align}
Also, let $T \colon \LS_\multii \longrightarrow \LS_\multii$ be the linear extension of the map sending each element $\alpha \in \mathsf{B}_\multii$, 
represented as~\eqref{LPrep}, to
\begin{align}\label{LPrepIns} 
\vcenter{\hbox{\includegraphics[scale=0.275]{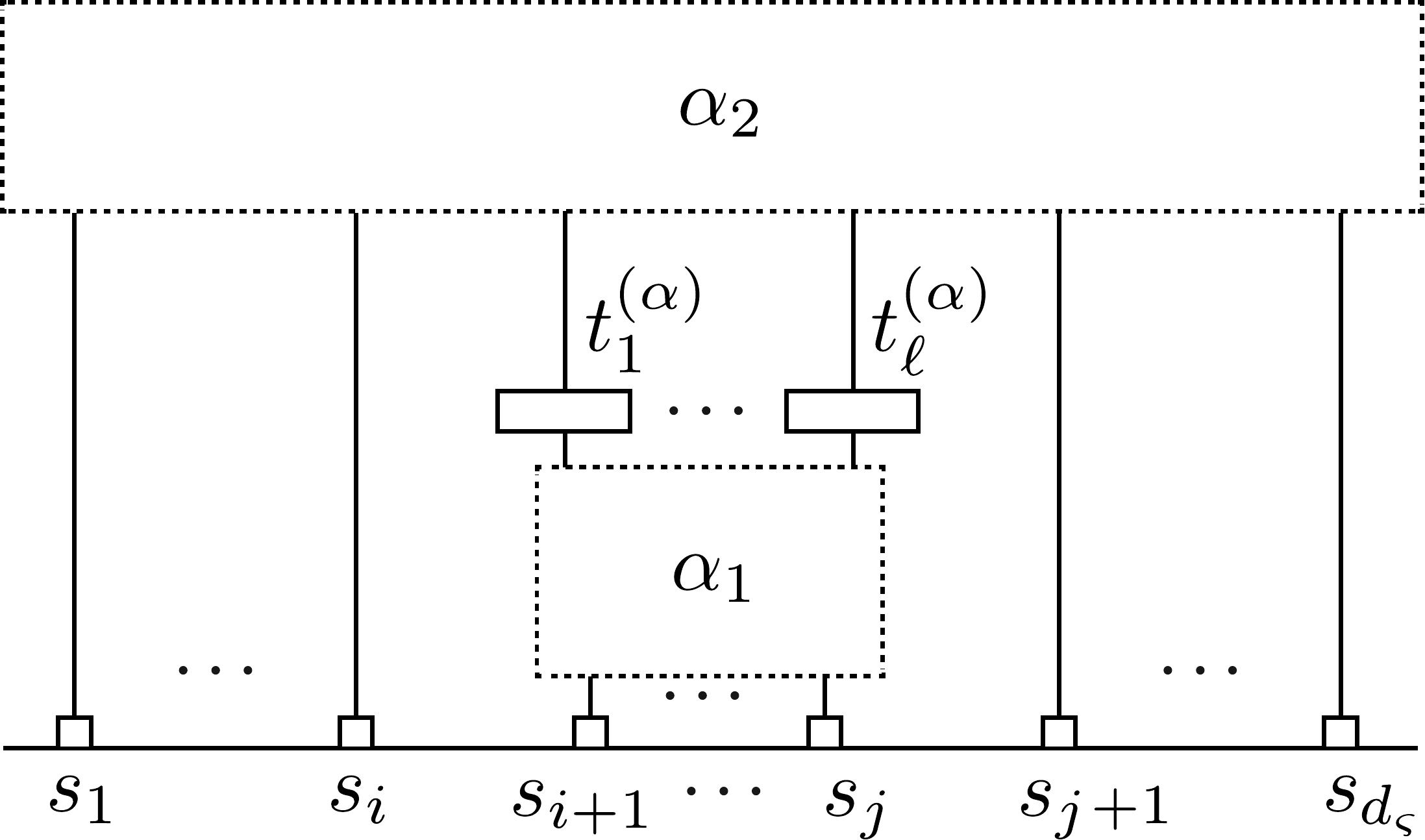} ,}} 
\end{align}
for some integers $\ell, \smash{t_1\super{\alpha}}$, $\smash{t_2\super{\alpha}, \ldots, t_\ell\super{\alpha}} \in \bZnn$ 
depending on $\alpha$ and such that $\smash{t_1\super{\alpha} + t_2\super{\alpha} + \dotsm + t_\ell\super{\alpha} = t_\alpha}$,
with $\ell = \ell_\alpha$ vanishing only if $t_\alpha = 0$. 
Then $T$ has an upper-unitriangular matrix representation.
\end{lem}

\begin{proof} 
The following relation endows the basis $\mathsf{B}_\multii$ with a strict partial order:
\begin{align} \label{StrictPO}
t_\alpha < t_\beta \qquad \Longleftrightarrow \qquad \alpha < \beta. 
\end{align}
Now, for $\alpha \in \mathsf{B}_\multii$ represented as~\eqref{LPrep}, we decompose the $\ell_\alpha$ projector boxes 
in $T(\alpha)$ depicted in~\eqref{LPrepIns} over their internal link diagrams. 
Because the coefficient of the identity term in~\eqref{ProjDecomp} equals one, 
and all of the other terms contain at least one turn-back link according to recursion relation~\eqref{wjrecursion}, we arrive with a sum of the form
\begin{align} \label{UpperTriMat}
T(\alpha) 
\underset{\eqref{ProjDecomp}}{\overset{\eqref{wjrecursion}}{=}}
\sum_{\beta \, \in \, \mathsf{B}_\multii} T_{\beta, \alpha} \, \beta , \qquad \text{with} \quad
T_{\beta, \alpha}
\begin{cases} 
\in \bC, & \alpha > \beta , \\ 
= 1 , & \alpha = \beta , \\ 
= 0 , & \text{otherwise} .
\end{cases} 
\end{align}
The coefficients $T_{\beta,\alpha}$ form a matrix representation 
of the linear operator $T$ with respect to the basis $\mathsf{B}_\multii$.
If we arrange the elements of $\mathsf{B}_\multii$ in such a way that the 
column for $\beta$ is left of the column for $\alpha$ if $\beta < \alpha$, then
$(T_{\beta, \alpha})$ is an upper-unitriangular matrix.
\end{proof}

\begin{lem} \label{InsProjAlphaLem} 
Suppose $\max \multii < \ppmin(q)$.  
The self-map of $\LS_\multii$ given by linear extension of $\alpha \mapsto \hcancel{\alpha}$ is an automorphism of vector spaces 
with an upper-unitriangular matrix representation.
\end{lem}

\begin{proof} 
The map $\alpha \mapsto \hcancel{\alpha}$ amounts to the insertion of projector boxes of sizes 
$\smash{r_2\super{\alpha}, r_3\super{\alpha}, \ldots, r_{\np_\multii-1}\super{\alpha}}$ to the walk representation 
$\smash{\varrho_\alpha = (r_1\super{\alpha}, r_2\super{\alpha}, \ldots, r_{\np_\multii}\super{\alpha})}$ of $\alpha$,
in order to convert the open vertices in~\eqref{Generic} to closed vertices as in~\eqref{cbs}.
This map is a composition of linear maps of type~\eqref{LPrepIns} in lemma~\ref{ChangeOfBasisLem}, 
obtained by inserting the projector boxes one by one,
and each such map has an upper-unitriangular matrix representation.
We must show that we can order the elements $\alpha \in \LP_\multii$ so that 
the whole composition also has an upper-unitriangular matrix representation.

The set of walks over $\multii$ has a natural strict partial order defined by
\begin{align} \label{PartialOrd} 
\varrho \DPle \varrho' \quad \text{if and only if} \quad
r_i < r_i' , \text{ for all $i \in \{0,1,\ldots,\np_\multii\}$,}
\end{align}
for any two walks $\varrho = (r_1, r_2, \ldots, r_{\np_\multii})$ and $\varrho' = (r_1', r_2', \ldots, r_{\np_\multii}')$.
Using this, we endow the set $\LP_\multii$ with a strict partial order $\DPle$ 
by considering the walk representations $\{ \varrho_\alpha \, | \, \alpha \in \LP_\multii \}$.
We show recursively that the partial order $\DPle$ 
preserves the upper-unitriangular matrix structure of the box insertions in $\alpha \mapsto \hcancel{\alpha}$.

We define $\smash{\hcancel{\alpha}\super{1} }:= \alpha$ and $\smash{\hcancel{\alpha}\super{k+1} := T_{k+1}(\hcancel{\alpha}\super{k})}$,
for each $k \in \{1,2,\ldots, \np_\multii-2\}$, where $\smash{\hcancel{\alpha}\super{k}}$
is a network of type~\eqref{Generic} but with the first $k-1$ vertices closed, 
and $T_{k+1}$ converts the $k$:th vertex of $\smash{\hcancel{\alpha}\super{k}}$ from open to closed by 
inserting a projector box of size $\smash{r_{k+1}\super{\alpha}}$ into $\smash{\hcancel{\alpha}\super{k}}$.
(In particular, $\smash{\hcancel{\alpha} = \hcancel{\alpha}\super{\np_\multii-1}}$.)
For example, when $k = 1$, the map $T_2$ replaces the first open vertex with a closed vertex in
the walk representation~\eqref{Generic} of $\alpha$, by inserting a projector box on the cable of size $r_2$.
After decomposing this projector box via recursion relation~\eqref{wjrecursion}, 
analogously to~\eqref{UpperTriMat} in the proof of lemma~\ref{ChangeOfBasisLem},
we arrive with $\alpha$ plus a linear combination of valenced link patterns $\beta$ which all satisfy $\beta \DPle \alpha$.
In particular, any ordering such that the column for $\beta$ is left of the column for $\alpha$ if $\beta \DPle \alpha$
yields an upper-unitriangular matrix representation for $T_2$.
Iterating this argument, we see that with such an ordering, all maps $T_{k+1}$ for $k \in \{1,2,\ldots, \np_\multii-2\}$
have an upper-unitriangular matrix representation. 
In particular, the composition $\alpha \mapsto \hcancel{\alpha}$ of these maps has such a matrix representation, which is what we sought to prove. 
\end{proof}


Now we use lemma~\ref{InsProjAlphaLem} to prove that the set 
$\{ \hcancel{\alpha} \, | \, \alpha \in \LP_\multii, \max \varrho_\alpha <\ppmin(q) \}$ 
is orthogonal and linearly independent. 
We use lemmas~\ref{ExtractLem} and~\ref{LoopErasureLem} of appendix~\ref{TLRecouplingSect}, the latter containing 
the evaluation of the Theta network from lemma~\ref{ThetaLem}: 
\begin{align} \label{ThetaFormula0} 
\ThetaNet(r,s,t) 
= \frac{(-1)^{\frac{r + s + t}{2}} \left[ \frac{r + s + t}{2} + 1 \right]! \left[ \frac{ r + s - t }{2} \right]! \left[ \frac{ s + t - r}{2} \right]! \left[ \frac{t + r - s}{2} \right]! }{[ r ]! [s ]! [ t ]!} ,
\end{align}
where $[k]!$ denotes the \emph{$k$:th quantum factorial}, defined as
\begin{align}
[k]! = [k]_q! := \prod_{\ell=1}^k [\ell] . 
\end{align}

\begin{prop} \label{IndOrthBasisLem} 
Suppose $\max \multii < \ppmin(q)$.  The following hold:
\begin{enumerate}
\itemcolor{red}
\item \label{IndOrthBasisLemIt1} 
For any valenced link patterns $\alpha,\beta \in \smash{\LP_\multii}$ with $J_\alpha(q) = J_\beta(q) = -\infty$, 
writing $\varrho_\alpha = (r_1, r_2, \ldots, r_{\np_\multii})$, we have
\begin{align}\label{WalkBiForm} 
\BiForm{\hcancel{\alpha}}{\hcancel{\beta}} = \delta_{\alpha, \beta} \prod_{j = 1}^{\np_{\multii} - 1}  
\frac{ \ThetaNet( r_j, r_{j+1}, \sIndex_{j+1} )}{ (-1)^{r_{j + 1}} [r_{j + 1}+1]} .
\end{align}
\item \label{IndOrthBasisLemIt2} 
The collection $\{ \hcancel{\alpha} \, | \, \alpha \in \LP_\multii, J_\alpha(q) = -\infty \}$ is orthogonal and linearly independent.
\item \label{IndOrthBasisLemIt3} 
The collection $\big\{ \hcancel{\alpha} \, | \, \alpha \in \smash{\LP_\multii\super{s}} \big\}$ is a basis for $\smash{\LS_\multii\super{s}}$, 
and if $\Summed_\multii < \ppmin(q)$, then this basis is orthogonal.
\end{enumerate}
\end{prop}
\begin{proof} 
We prove items~\ref{IndOrthBasisLemIt1}--\ref{IndOrthBasisLemIt3} as follows:
\begin{enumerate}[leftmargin=*]
\itemcolor{red}
\item 
For $\alpha,\beta \in \smash{\LP_\multii}$ with $J_\alpha(q) = J_\beta(q) = -\infty$, the bilinear form 
$\BiForm{\hcancel{\alpha}}{\hcancel{\beta}}$ equals the evaluation of the following network: 
\begin{align}\label{ManyLoops} 
\vcenter{\hbox{\includegraphics[scale=0.275]{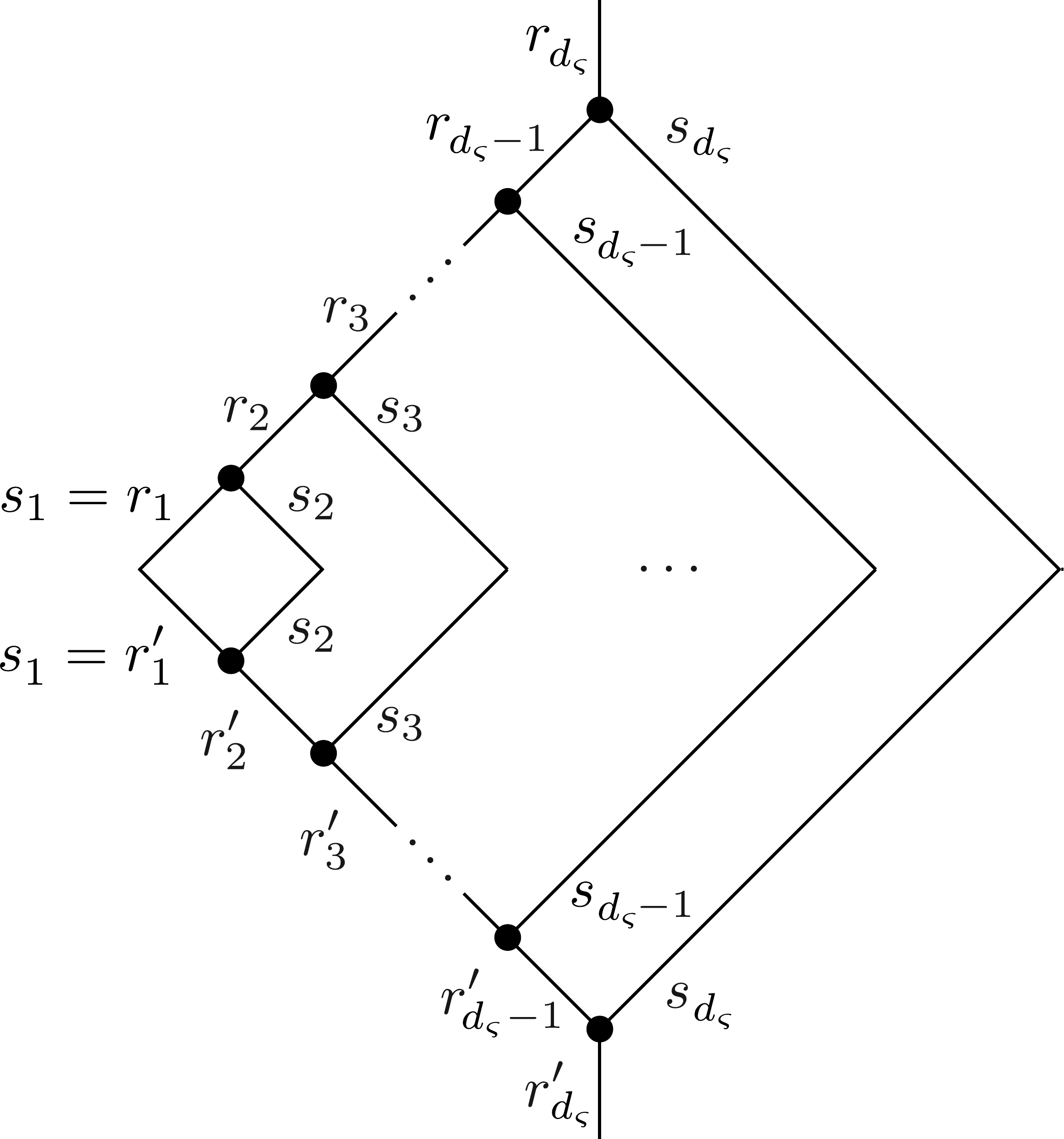} .}} 
\end{align}
To evaluate this network, we use lemmas~\ref{ExtractLem} and~\ref{LoopErasureLem}
of appendix~\ref{TLRecouplingSect} to recursively erase the smallest 
(i.e., the leftmost) loop in it, eventually arriving with~\eqref{WalkBiForm}. 

\item 
Orthogonality of the collection $\{ \hcancel{\alpha} \, | \, \alpha \in \LP_\multii, J_\alpha(q) = -\infty \}$ follows from~\eqref{WalkBiForm}.  
Linear independence of this collection follows from 
lemma~\ref{InsProjAlphaLem} and 
the fact that $\LP_\multii$ is a basis for $\LS_\multii$. 

\item That the set $\big\{ \hcancel{\alpha} \, | \, \alpha \in \smash{\LP_\multii\super{s}} \big\}$ is a basis for $\smash{\LS_\multii\super{s}}$ follows 
from lemma~\ref{InsProjAlphaLem} and the fact that $\smash{\LP_\multii\super{s}}$ is a basis for it.  
Orthogonality follows from item~\ref{IndOrthBasisLemIt1} with the fact that $J_\alpha(q) = -\infty$, 
for all $\multii$-valenced link patterns $\alpha$ if $\Summed_\multii < \ppmin(q)$.
\end{enumerate}
This concludes the proof.
\end{proof}

Next, we study the product~\eqref{WalkBiForm} in item~\ref{IndOrthBasisLemIt1} of proposition~\ref{IndOrthBasisLem} 
when the assumption that $J_\alpha(q) = -\infty$ may not hold (but we still have $\max \multii < \ppmin(q)$).  
In the next lemma, we show that certain factors in the product~\eqref{WalkBiForm} are finite and nonzero. 
We use this result to prove proposition~\ref{VanishDetLem2} below, and to determine the radical of $\smash{\LS_\multii\super{s}}$ in section~\ref{RadicalSect}.

\begin{lem} \label{ThetaInFiniteAndNonzeroLem} 
Suppose $\max \multii < \ppmin(q)$. 
Let $\varrho = (r_1, r_2, \ldots, r_{\np_\multii})$ be a walk over $\multii$, and let $\bar{J}:= \max(J,0)$.   
Then, for all $j \in \{\bar{J}+1, \bar{J}+2, \ldots, \np_\multii-1\}$, 
we have
\begin{align} \label{ThetaShouldBeFiniteAndNonzero0} 
0 < \bigg| \frac{ \ThetaNet( r_j, r_{j+1}, \sIndex_{j+1} )}{ (-1)^{r_{j+1}} [r_{j + 1}+1]} \bigg| < \infty. 
\end{align}
\end{lem}
\begin{proof}
We fix $j \in \{\bar{J}+1, \bar{J}+2, \ldots, \np_{\multii}-1\}$ and consider the different factors in 
\begin{align} \label{ThetaShouldBeFiniteAndNonzero} 
\bigg| \frac{ \ThetaNet( r_j, r_{j+1}, \sIndex_{j+1} )}{ (-1)^{r_{j + 1}} [r_{j + 1}+1]} \bigg| \overset{\eqref{ThetaFormula0}}{=}
\bigg| \frac{ \big[\frac{r_{j} + r_{j + 1} + \sIndex_{j + 1}}{2} + 1 \big]! 
\big[\frac{r_j + r_{j+1} - \sIndex_{j+1} }{2} \big]! \big[\frac{r_{j+1} + \sIndex_{j+1} - r_j }{2} \big]! \big[\frac{\sIndex_{j+1} + r_j - r_{j+1} }{2} \big]!}{ [r_j]! [r_{j + 1}+1]! [\sIndex_{j+1}]! } \bigg| .
\end{align}
Definition~\eqref{WalkHeights} gives $r_{i+1} \in \DefectSet\sub{r_i, \sIndex_{i+1}}$, for all $i \in \{ 0, 1, \ldots, \np_{\multii} - 1 \}$. 
Combining this with (\ref{SpecialDefSet},~\ref{SameDefSet}) from lemma~\ref{SpecialDefLem}, we obtain
\begin{align} 
\label{EidsGeneral} 
& |r_i - r_{i+1}| \leq \sIndex_{i+1}, \qquad r_{i+1} \leq r_i + \sIndex_{i+1} , \qquad \text{and} \qquad r_i \leq r_{i+1} + \sIndex_{i+1},
\end{align}
for all $i \in \{ 0, 1, \ldots, \np_{\multii} - 1 \}$.  Thus, with $\max \multii < \ppmin(q)$, for any index $i$ in this set, we have
\begin{align} 
\label{condi1General} 
0 \overset{\eqref{EidsGeneral}}{\leq} \max \left( \frac{\sIndex_{i+1} - (r_j - r_{i+1}) }{2}, \frac{\sIndex_{i+1} - (r_{i+1} - r_i) }{2} \right) 
\overset{\eqref{EidsGeneral}}{\leq}  \sIndex_{i+1} < \max \multii < \ppmin(q). 
\end{align}
By definition~\eqref{Qinteger}, the quantum integer $[k]$ does not vanish if $k \in \{ 0, 1, \ldots, \ppmin(q) - 1 \}$, 
so~\eqref{condi1General} shows that the quantum factorials $\left[ \sIndex_{j+1} \right]!$, $\left[ \frac{1}{2}(r_{j+1} + \sIndex_{j+1} - r_j) \right]!$, and 
$\left[ \frac{1}{2} (\sIndex_{j+1} + r_j - r_{j+1}) \right]!$ in~\eqref{ThetaShouldBeFiniteAndNonzero} are nonzero.

To finish, we show that the zeros of the remaining factors 
in~\eqref{ThetaShouldBeFiniteAndNonzero} cancel, so their ratio is also finite and nonzero:
\begin{align} \label{ShouldBeFiniteAndNonzero}
0 < \bigg| \frac{\big[ \frac{ r_j + r_{j+1} + \sIndex_{j+1}}{2} + 1  \big]! \big[ \frac{r_j + r_{j+1} - \sIndex_{j+1}}{2} \big]! }{[ r_j ]! [ r_{j+1} + 1 ]! } \bigg| < \infty .
\end{align}
To see this, we observe that, for any $j \in \{\bar{J}+1, \bar{J}+2, \ldots, \np_{\multii}-1\}$, we have
\begin{align} 
\label{condi2General0}
k_s \pmin(q) \overset{\eqref{DeltaDefn}}{=} \Delta_{k_s} + 1 & \overset{\eqref{Jindex0}}{\leq} h_{\min,j}(\varrho) 
\overset{\eqref{minmaxh}}{=} \frac{r_j + r_{j+1} - \sIndex_{j+1}}{2} \\
& \overset{\eqref{EidsGeneral}}{\leq} \min(r_j, r_{j+1}) 
< \max(r_j, r_{j+1}) + 1 \\
& \overset{\eqref{EidsGeneral}}{<} 
\frac{r_j + r_{j+1} + \sIndex_{j+1}}{2} + 1 
\overset{\eqref{minmaxh2}}{=}  h_{\max,j}(\varrho) + 1 \\
\label{condi2General1} 
& \overset{\eqref{Jindex0}}{\leq} \Delta_{k_s+1} \overset{\eqref{DeltaDefn}}{=} (k_s +1)\pmin(q) - 1. 
\end{align}
From definition~\eqref{Qinteger}, we see that, for any $k \in \{ 0, 1, \ldots, \ppmin(q) - 1 \}$, we have
$[k]_q = 0$ if and only if $\ppmin(q) \,|\, k$, and these zeros of $q \mapsto [k]_q$ 
are of first order. 
Therefore,~(\ref{condi2General0}--\ref{condi2General1}) imply~\eqref{ShouldBeFiniteAndNonzero}, 
which implies~\eqref{ThetaShouldBeFiniteAndNonzero0}. 
\end{proof}

\subsection{Determinant of the Gram matrix} \label{DetSect}

Next we use proposition~\ref{IndOrthBasisLem} to find a formula for the determinant of the Gram matrix 
$\smash{\Gram_\multii\super{s}}$, defined in~\eqref{GramMatrix2}. We use this formula to prove that if $\Summed_\multii < \ppmin(q)$, 
then $\det \smash{\Gram_\multii\super{s}} \neq 0$, i.e., the radical of $\smash{\LS_\multii\super{s}}$ trivial, for all $s \in \DefectSet_\multii$.

\begin{prop} \label{GramDetLem} 
Suppose $\max \multii < \ppmin(q)$. We have
\begin{align}\label{DetFormula} 
\det \smash{\Gram_\multii\super{s}}
= \prod_{\varrho} \prod_{j=1}^{\np_{\multii} - 1}  
\frac{\ThetaNet(r_j, r_{j+1}, \sIndex_{j+1})}{(-1)^{r_{j+1}} [ r_{j+1} + 1 ]}, 
\end{align}
where the first product is over all walks $\WalkMultii = (r_1, r_2, \ldots, r_{\np_\multii})$ over $\multii$ with defect $r_{\np_\multii} = s$.
\end{prop}
\begin{proof} 
For all $\alpha,\beta \in \smash{\LP_\multii\super{s}}$, the bilinear form $\BiForm{\alpha}{\beta}$ is an analytic function of 
$q \in \{q \in \bC^\times  \, | \,  \max \multii < \ppmin(q)\}$.
Thus, the determinant  of the Gram matrix $\smash{\Gram_\multii\super{s}}$,
defined in terms of this bilinear form in~\eqref{GramMatrix2},
is also analytic on this set. 
Hence, we may assume that $\Summed_\multii < \ppmin(q)$.  
Then by lemma~\ref{InsProjAlphaLem} and items~\ref{IndOrthBasisLemIt1} and~\ref{IndOrthBasisLemIt3} of proposition~\ref{IndOrthBasisLem}, we have
\begin{align} \label{detdet}  
\det \smash{\Gram_\multii\super{s}} 
\underset{\text{lem.}~\ref{InsProjAlphaLem}}{\overset{\eqref{GramMatrix2}}{=}} \det [\BiForm{\hcancel{\alpha}}{\hcancel{\beta}}]_{\alpha,\beta \in \LP_\multii\super{s}} 
\overset{\eqref{WalkBiForm}}{=} \prod_{\alpha \, \in \, \smash{\LP_\multii\super{s}}} 
\prod_{j=1}^{\np_{\multii} - 1}  \frac{\ThetaNet(r_j, r_{j+1}, \sIndex_{j+1})}{(-1)^{r_{j+1}} [ r_{j+1} + 1 ]} . 
\end{align}
To obtain determinant formula~\eqref{DetFormula} from this expression, we use the bijection of item~\ref{wmlIt2} of lemma~\ref{WalkMultiiLem}, 
sending $\alpha \mapsto \varrho_\alpha$, for all $\alpha \in \smash{\LP_\multii\super{s}}$, to index the product by all walks over 
$\multii$ with defect $s$. 
\end{proof}

In some cases, an alternative formula for the determinant of $\smash{\Gram_\multii\super{s}}$ is known.  
For example, in the case with $\multii = \OneVec{n}$ and $s = 0$ (so $n$ is necessarily 
even by~\eqref{DefectSet}),~\cite[equation~(\red{5.6})]{fgg} 
gives an explicit formula for the determinant of the Gram matrix $\smash{\Gram_n\super{0}}$, now called the \emph{meander matrix}.  
Next, in the case with $\multii = \OneVec{n}$,~\cite[theorem~\red{4.7}]{rsa}
 gives an alternative formula for the determinant of $\smash{\Gram_n\super{s}}$;
see also~\cite[corollary~\red{4.7}]{gl2} for a more general case. 
We state this formula as lemma~\ref{RidoutDetLem} in section~\ref{RecursSect} below, 
and we derive it from our formula~\eqref{DetFormula} for use 
in section~\ref{RadicalSect}.  The last special case that we know of has $\multii = (1,1,\ldots,1,k)$ for some $k \in \bZpos$.  
In this case,~\cite[proposition~\red{D.4}]{mrr} 
gives an alternative formula for the determinant of $\smash{\Gram\sub{1,1,\ldots,1,k}\super{s}}$.

Now, using the explicit formula~\eqref{DetFormula} for the determinant of $\smash{\Gram_\multii\super{s}}$ from proposition~\ref{GramDetLem}, 
we deduce that this determinant is not zero if $\Summed_\multii < \ppmin(q)$. 
This is the main result of the present section.

\begin{prop} \label{VanishDetLem2} 
Suppose $\Summed_\multii < \ppmin(q)$.  Then we have $\det \Gram_\multii\super{s} \neq 0$, for all $s \in \DefectSet_\multii$.
\end{prop}
\begin{proof} 
If $\pmin(q) = 1$ (so $q \in \{\pm1\}$ and $\ppmin(q) = \infty$ by~\eqref{MinPower}), 
then definition~\eqref{Qinteger} shows that we have $[k] \neq 0$, for all 
$k \in \bZpos$.  
Then the claim follows from the explicit formula~\eqref{DetFormula} for 
the determinant of $\smash{\Gram_\multii\super{s}}$.
Throughout the rest of the proof, we therefore assume that $\pmin(q) \neq 1$, so $\ppmin(q) = \pmin(q)$ by~\eqref{MinPower}, 
and we let $s$ be an arbitrary integer in $\DefectSet_\multii$.  Then we have
\begin{align}\label{ks0} 
\begin{cases} s \in \DefectSet_\multii \\ \Summed_\multii < \ppmin(q) = \pmin(q) \end{cases} 
\qquad \overset{\eqref{DefSet2}}{\Longrightarrow} \qquad s \leq \Summed_\multii < \pmin(q) 
\qquad \underset{\eqref{skDefn}}{\overset{\eqref{DeltaDefn}}{\Longrightarrow}} \qquad k_s =0. 
\end{align}
We divide our analysis into two cases:
\begin{enumerate}[leftmargin=*]
\itemcolor{red}
\item $s < \pmin(q) - 1$:  
For any walk $\varrho = (r_1, r_2, \ldots, r_{\np_\multii})$ over $\multii$ with defect $s$, and for all $j \in \{0,1,\ldots,\np_\multii-2\}$, we have
\begin{align} 
\label{Jinf1-0} 
\Delta_{k_s} \underset{\eqref{ks0}}{\overset{\eqref{DeltaDefn}}{=}} -1 < 0 
& \overset{\eqref{EasyCompare}}{\leq} h_{\min,j}(\varrho) \\
\label{Jinf1-1} 
& \overset{\eqref{EasyCompare}}{\leq} h_{\max,j}(\varrho) \underset{\eqref{minmaxh2}}{\overset{\eqref{WalkHeights}}{\leq}} 
r_j + \sIndex_{j+1} \overset{\eqref{rLessThanN}}{<} \Summed_\multii \leq \pmin(q) - 1 \underset{\eqref{ks0}}{\overset{\eqref{DeltaDefn}}{=}} \Delta_{k_s+1}. 
\end{align}
Moreover, for $j = \np_\multii - 1$, we use the fact that $s < \pmin(q) - 1$ to obtain
\begin{align} 
\label{Jinf2-0} 
\Delta_{k_s} \underset{\eqref{ks0}}{\overset{\eqref{DeltaDefn}}{=}} -1 < 0 
& \overset{\eqref{EasyCompare}}{\leq} h_{\min,\np_\multii - 1}(\varrho) \\
\label{Jinf2-1} 
& \overset{\eqref{EasyCompare}}{\leq} h_{\max,\np_\multii - 1}(\varrho) \underset{\eqref{rLessThanN}}{\overset{\eqref{minmaxh2}}{\leq}} 
\frac{\Summed_\multii + s}{2} \overset{\eqref{rLessThanN}}{<} \pmin(q) - 1 \underset{\eqref{ks0}}{\overset{\eqref{DeltaDefn}}{=}} \Delta_{k_s+1}, 
\end{align}
and finally, for $j = \np_\multii$, we again use the fact that $s < \pmin(q) - 1$ to obtain
\begin{align} 
\label{Jinf3-0} 
\Delta_{k_s} \underset{\eqref{ks0}}{\overset{\eqref{DeltaDefn}}{=}} -1 < 0 \overset{\eqref{DefSet2}}{\leq} s 
& \overset{\eqref{minmaxh}}{=} h_{\min,\np_\multii}(\varrho) \\
\label{Jinf3-1} 
& \overset{\eqref{minmaxh}}{=} h_{\max,\np_\multii}(\varrho) \overset{\eqref{minmaxh2}}{=} s 
< \pmin(q) - 1 \underset{\eqref{ks0}}{\overset{\eqref{DeltaDefn}}{=}} \Delta_{k_s+1}. 
\end{align}
Altogether,~(\ref{Jinf1-0}--\ref{Jinf3-1}) imply that $J_\varrho(q) = -\infty$ for any walk $\varrho$ over $\multii$ with defect $s$.  
Lemma~\ref{ThetaInFiniteAndNonzeroLem} implies that the product in the formula~\eqref{DetFormula} for the determinant of $\smash{\Gram_\multii\super{s}}$ does not vanish.

\item $s = \pmin(q) - 1$: We have
\begin{align}
\begin{cases} 
s \leq \Summed_\multii \\ s = \pmin(q) - 1 \\ \Summed_\multii < \pmin(q) 
\end{cases} 
\qquad \Longrightarrow \qquad s = \pmin(q) - 1 = \Summed_\multii . 
\end{align}
Now, there exists only one walk over $\multii$ with defect $s = \Summed_\multii$: the one 
with heights $r_j = \sIndex_1 + \sIndex_2 + \dotsm + \sIndex_j$, for all $j \in \{1,2,\ldots,\np_\multii\}$.  
Moreover, it follows from~\eqref{ThetaFormula0} 
that for any pair of integers $r,t \in \bZnn$, we have
\begin{align}
\ThetaNet(r,t,r+t) \overset{\eqref{ThetaFormula0}}{=} 1.
\end{align}
Combining these facts with~\eqref{DetFormula}, we see that the determinant of $\smash{\Gram_\multii\super{s}}$ equals one 
(especially, not zero) in this case.
\end{enumerate}
This shows that  $\det \Gram_\multii\super{s} \neq 0$.
\end{proof}

\subsection{Recursion formulas for the Gram determinant} \label{RecursSect}

In this section, we gather additional results on the determinant of the Gram matrix $\smash{\Gram_\multii\super{s}}$. 
First, in lemma~\ref{RecurseLem} we give a general recursion formula for the determinant of $\smash{\Gram_\multii\super{s}}$,
and in lemma~\ref{RidoutDetLem}, we employ it to derive an alternative formula for this determinant 
when $\multii = \OneVec{n}$ for some $n \in \bZnn$. This formula already appears, e.g., in~\cite{bw, rsa}.
In lemma~\ref{ExtGramLem}, we derive a formula for the determinant of another 
Gram matrix $\smash{\Gram_{\multii;\ExtraDefect}\super{s}}$, 
related to the original $\smash{\Gram_\multii\super{s}}$ and defined below. 
We use these results in section~\ref{RadicalSect}.
Finally, in lemma~\ref{ForBSALem} we obtain yet another recursion formula for the determinant of 
$ \smash{\Gram_n\super{s}}$ in the case that $\multii = \OneVec{n}$, which we use in our forthcoming article~\cite{fp2}.

%

For the next lemma, we use notation~\eqref{hats} for $\lds$ and $t$. 

\begin{lem} \label{RecurseLem} 
Suppose $\max \multii < \ppmin(q)$. We have the recursion
\begin{align}\label{DetRecurse} 
\det \Gram_\multii\super{s} = \prod_{r \, \in \, \DefectSet_{\lds} \, \cap \, \DefectSet\sub{t,s}} 
\big(\det \Gram\sub{r,t}\super{s} \big)^{\Dim_{\lds}\super{r}} \det \Gram_{\lds}\super{r} . 
\end{align}
\end{lem} 
\begin{proof} 
The following fact is evident from the definition~\eqref{WalkHeights} of a walk $\varrho$ over $\multii$:
\begin{align} 
\text{$\varrho = (r_1, r_2, \ldots, r_{\np_\multii})$ is a walk over $\multii$} 
\qquad \qquad \overset{\eqref{WalkHeights}}{\Longrightarrow} \qquad \qquad 
\text{$\hat{\varrho} = (r_1, r_2, \ldots, r_{\np_{\multii}-1})$ is a walk over $\lds$} . 
\end{align}
Therefore, by item~\ref{wmlIt4} of lemma~\ref{WalkMultiiLem}, the penultimate height $r_{\np_{\multii}-1}$ of a walk $\varrho$ over 
$\multii$ is an element of $\DefectSet_{\lds}$.  Moreover, 
\begin{align} 
\text{$\varrho = (r_1, r_2, \ldots, r_{\np_\multii})$ is a walk over $\multii$ with defect $r_{\np_\multii} = s$} 
\qquad \qquad \overset{\eqref{WalkHeights}}{\Longrightarrow} \qquad \qquad  r_{\np_{\multii}-1} \in \DefectSet\sub{t,s} . 
\end{align}
We conclude that if $\varrho = (r_1, r_2, \ldots, r_{\np_\multii})$ is a walk over $\multii$ with defect $s$, then 
$r := r_{\np_{\multii}-1} \in \DefectSet_{\lds} \cap \DefectSet\sub{t,s}$. 
Hence, with $\varrho_r$ denoting a walk over $\multii$ with penultimate height $r$ and defect $s$, we may write  
determinant formula~\eqref{DetFormula} as
\begin{align} \label{DetFormula2} 
\det \Gram_\multii\super{s}
\overset{\eqref{DetFormula}}{=} \prod_{r \, \in \, \DefectSet_{\lds} \, \cap \, \DefectSet\sub{t,s}} \prod_{\varrho_r} \prod_{j=1}^{\np_{\multii} - 1} 
\bigg( \frac{\ThetaNet(r_j, r_{j+1}, \sIndex_{j+1})}{(-1)^{r_{j+1}} [ r_{j+1} + 1 ]} \bigg).
\end{align}
The factor in the product over $j \in \{1,2,\ldots,\np_{\multii}-1\}$ with 
$j = \np_{\multii}-1$ depends only on the last two heights $r$ and $s$ of a walk over $\multii$, which are the same for all 
walks $\varrho_r$.  Hence, we may factor it out of the product, obtaining
\begin{align}\label{FinalDet} 
\det \Gram_\multii\super{s} 
\overset{\eqref{DetFormula2}}{=} \prod_{r \, \in \, \DefectSet_{\lds} \, \cap \, \DefectSet\sub{t,s}} 
\bigg(\frac{\ThetaNet(r, s, t)}{(-1)^s [ s + 1 ]} \bigg)^{\Dim_{\lds}\super{r}} \prod_{\varrho_r} \prod_{j=1}^{\np_{\multii} - 2} \bigg( 
\frac{\ThetaNet(r_j, r_{j+1}, \sIndex_{j+1})}{(-1)^{r_{j+1}} [ r_{j+1} + 1 ]} \bigg),
\end{align}
where the power $\smash{\Dim_{\lds}\super{r}}$ follows from the fact that there are this many distinct walks $\hat{\varrho}_r$ over $\lds$ with defect $r$.  
Finally, recalling the formula~\eqref{DetFormula} for $\smash{\det \Gram\sub{r,t}\super{s}}$ and $\smash{\det \Gram_\lds\super{r}}$, 
we obtain sought recursion formula~\eqref{DetRecurse}.
\end{proof}

In the next lemma, we give another formula for the determinant of the Gram matrix $\smash{\Gram_n\super{s}}$.
This formula appears in~\cite{bw}, and a proof for it appears in~\cite[theorem~\red{4.7}]{rsa}; ee also~\cite[corollary~\red{4.7}]{gl2}.
We use this formula to prove lemma~\ref{EasyRadLem} in section~\ref{RadicalSect}, a crucial ingredient for completely 
and explicitly characterizing the radical of 
$\smash{\LS_\multii\super{s}}$ for any multiindex $\multii \in \{\OneVec{0}\} \cup \smash{\bZpos^\#}$ 
and for any integer $s \in \DefectSet_\multii$.

\begin{lem} \label{RidoutDetLem} \textnormal{\cite[theorem~\red{4.7}]{rsa}} 
We have 
\begin{align}\label{RidoutDet} 
\det \Gram_n\super{s} = \prod_{j=1}^{\frac{n-s}{2}} \left(\frac{[s+j+1]}{(-1)^{s+1}[j]}\right)^{\Dim_n\super{s+2j}} . 
\end{align}
\end{lem}  
\begin{proof} 
As in the proof of proposition~\ref{GramDetLem}, we may assume that $n < \ppmin(q)$ throughout.  
For convenience, we also substitute $q \mapsto -q$ throughout. 
With $[s]_{-q} = (-1)^{s-1}[s]_q$ by~\eqref{Qinteger}, 
we may write~\eqref{RidoutDet} as
\begin{align}\label{RidoutDet2} 
\det \Gram_n\super{s} = \prod_{j=1}^{\frac{n-s}{2}} \left(\frac{[s+j+1]_{-q}}{[j]_{-q}}\right)^{\Dim_n\super{s+2j}} . 
\end{align} 
We prove formula~\eqref{RidoutDet2} by induction on $n \in \bZnn$.
First, it trivially holds for $n = 0$. 
Next, assuming that~\eqref{RidoutDet2} holds with $n = m-1$ for some $m \geq 2$, we prove that it holds with $n = m$.  
Lemma~\ref{RecurseLem} with $\multii = \OneVec{m}$ gives
\begin{align}\label{AllOnesRecurs} 
\det \Gram_m\super{s} \overset{\eqref{DetRecurse}}{=} \det \Gram_{m-1}\super{s-1} \big(\det \Gram\sub{s-1,1}\super{s} \big)^{\Dim_{n-1}\super{s-1}} \det \Gram_{m-1}\super{s+1} \big(\det \Gram\sub{s+1,1}\super{s} \big)^{\Dim_{n-1}\super{s+1}}. \end{align}
Using proposition~\ref{GramDetLem}, we have 
\begin{align}\label{AllOnesRecursSpec}
\det \Gram\sub{s-1,1}\super{s} \overset{\eqref{DetFormula}}{=} 1 \qquad \qquad \text{and} \qquad \qquad
\det \Gram\sub{s+1,1}\super{s} \overset{\eqref{DetFormula}}{=} \frac{[s+2]_{-q}}{[s+1]_{-q}} . 
\end{align}
After inserting these formulas into~\eqref{AllOnesRecurs} and applying the induction hypothesis, we arrive with
\begin{align} 
\label{NumDenomProd0} 
\det \Gram_m\super{s} & 
\underset{\eqref{AllOnesRecursSpec}}{\overset{\eqref{AllOnesRecurs}}{=}} \left( \frac{[s+2]_{-q}}{[s+1]_{-q}} \right)^{\Dim_{m-1}\super{s+1}} 
\Bigg( \prod_{j=1}^{\frac{m-s}{2}} \left(\frac{[s+j]_{-q}}{[j]_{-q}}\right)^{\Dim_m\super{s+2j-1}} \Bigg) 
\Bigg( \prod_{k=1}^{\frac{m-s}{2}-1} \left(\frac{[s+k+2]_{-q}}{[k]_{-q}}\right)^{\Dim_m\super{s+2k+1}} \Bigg)\\
\label{NumDenomProd1} 
& \underset{\hphantom{\eqref{AllOnesRecursSpec}}}{\overset{\hphantom{\eqref{AllOnesRecurs}}}{=}}
\left( \frac{[s+2]_{-q}}{[s+1]_{-q}} \right)^{\Dim_{m-1}\super{s+1}} \frac{\Big( \prod_{j=0}^{\frac{m-s}{2}-1} [s+j+1]_{-q}^{\Dim_{m-1}\super{s+2j+1}} \Big) 
\Big( \prod_{k=2}^{\frac{m-s}{2}} [s+k+1]_{-q}^{\Dim_{m-1}\super{s+2k-1}} \Big)}{\left[ \frac{m-s}{2} \right]_{-q} 
\prod_{j=1}^{\frac{m-s}{2}-1} [j]_{-q}^{\Dim_{m-1}\super{s+2j-1} + \Dim_{m-1}\super{s+2j+1}}}.
\end{align}
Using the properties $\smash{\Dim_{m-1}\super{s-1} + \Dim_{m-1}\super{s+1} = \Dim_m\super{s}}$ for $s \in \DefectSet_m \cap \{1,2,\ldots,m-1\}$ and $\smash{\Dim_m\super{m}} = 1$ from~\eqref{Recursion}, 
this simplifies to
\begin{align} 
\det \Gram_m\super{s} & \underset{(\ref{NumDenomProd0}-\ref{NumDenomProd1})}{\overset{\eqref{Recursion}}{=}} 
\left( \frac{[s+2]_{-q}}{[s+1]_{-q}} \right)^{\Dim_{m-1}\super{s+1}} \left(\frac{ [s+1]_{-q}^{\Dim_{m-1}\super{s+1}}\left[s+\frac{m-s}{2}+1\right]_{-q}^{\Dim_m\super{m}} 
\prod_{j=1}^{\frac{m-s}{2}-1}[s+j+1]_{-q}^{\Dim_m\super{s+2j}}}{ [s+2]_{-q}^{\Dim_{m-1}\super{s+1}} \left[ \frac{m-s}{2} \right]_{-q}^{\Dim_m\super{m}} 
\prod_{j=1}^{\frac{m-s}{2}-1} [j]_{-q}^{\Dim_m\super{s+2j}} }\right), 
\end{align}
which further simplifies to~\eqref{RidoutDet2} with $n = m$.  This completes the induction step and finishes the proof.
\end{proof}

Next, in lemma~\ref{ExtGramLem} we prove another recursion formula, for use in section~\ref{rofSect2}. 
To state it, we define the set $\smash{\LP\super{s}_{\multii;\ExtraDefect}}$ 
to be the collection of valenced link patterns obtained by increasing the size $\sIndex_{\np_\multii}$ 
of the rightmost valenced node in every $(\multii,s)$-valenced link pattern $\alpha$ to size 
$\sIndex_{\np_\multii} + \ExtraDefect$, and attaching $\ExtraDefect$ defects to the extended box.  For example,
\begin{align}\label{LSd} 
\vcenter{\hbox{\includegraphics[scale=0.275]{e-RecursiveDecomposition2.pdf}}} 
\qquad \in \LP\super{s}\sub{r,t}
\qquad \longmapsto \qquad 
\vcenter{\hbox{\includegraphics[scale=0.275]{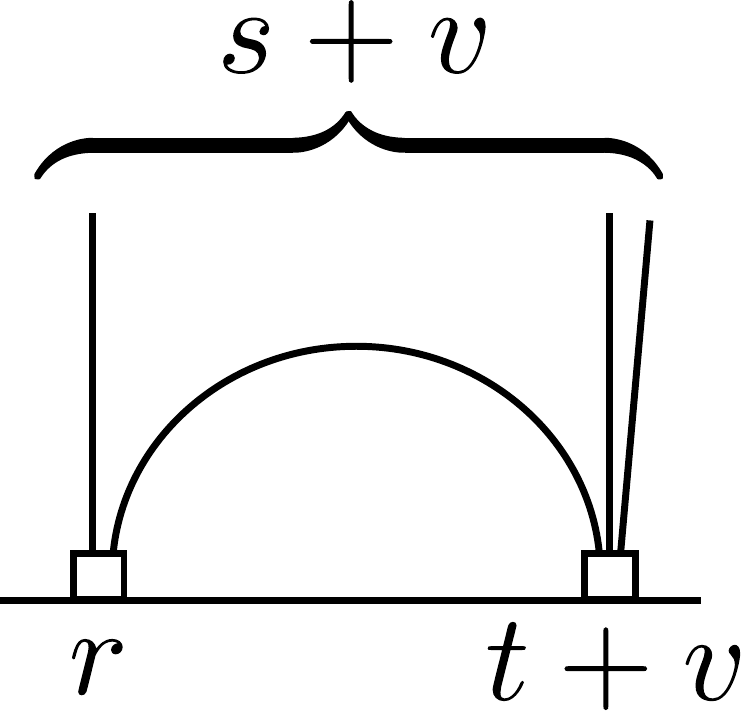}}}  
\qquad \in \LP\super{s}\sub{r,t}{}_{;\ExtraDefect} \; \subset \; \LP\super{s+\ExtraDefect}\sub{r,t+\ExtraDefect} . 
\end{align}
We also let $\smash{\LS\super{s}_{\multii;\ExtraDefect}}$ denote the complex vector space with 
basis $\smash{\LP\super{s}_{\multii;\ExtraDefect}}$.

\begin{lem} \label{AlternativeBasisWith_d_ExtraDefects} 
Suppose $\max \multii < \ppmin(q)$. We have 
\begin{align}
\alpha \in \smash{\LP\super{s}_{\multii;\ExtraDefect}} 
\qquad \Longrightarrow \qquad \hcancel{\alpha} \in \smash{\LS\super{s}_{\multii;\ExtraDefect}},
\end{align}
and the set $\smash{\big\{ \hcancel{\alpha} \,\big|\, \alpha \in \smash{\LP\super{s}_{\multii;\ExtraDefect}} \big\}}$ 
is a basis of $\smash{\LS\super{s}_{\multii;\ExtraDefect}}$.
\end{lem}

\begin{proof}
Lemma~\ref{InsProjBoxLem} of appendix~\ref{TLRecouplingSect} implies that, 
for each valenced link pattern $\alpha \in \smash{\LP\super{s}_{\multii;\ExtraDefect}}$, discarding 
the projector box set across the $s + \ExtraDefect$ defects of the associated trivalent link state $\hcancel{\alpha}$  
does not alter $\hcancel{\alpha}$.  
Therefore, the collection $\smash{\big\{ \hcancel{\alpha} \,\big|\, \alpha \in \smash{\LP\super{s}_{\multii;\ExtraDefect}} \big\}}$ 
is a subset of $\smash{\LS\super{s}_{\multii;\ExtraDefect}}$.  
Furthermore, this collection is linearly independent by lemma~\ref{InsProjAlphaLem}, and its cardinality obviously equals 
the dimension of $\smash{\LS\super{s}_{\multii;\ExtraDefect}}$.  
Thus, the set $\smash{\big\{ \hcancel{\alpha} \,\big|\, \alpha \in \smash{\LP\super{s}_{\multii;\ExtraDefect}} \big\}}$ 
is actually a basis of $\smash{\LS\super{s}_{\multii;\ExtraDefect}}$.
\end{proof}

Now, we find the determinant of the Gram matrix
of the link state bilinear form with respect to the basis $\smash{\LP\super{s}_{\multii;\ExtraDefect}}$,
\begin{align}\label{GramMatrix3}
[ \Gram\super{s}_{\multii;\ExtraDefect} ]_{\alpha, \beta} := \BiForm{\alpha}{\beta} , \quad 
\text{for all $\alpha, \beta \in \LP\super{s}_{\multii;\ExtraDefect}$.}
\end{align}

\begin{lem} \label{ExtGramLem} 
Suppose $\max \multii < \ppmin(q)$. We have 
\begin{align} \label{ExtGram} 
\det \Gram\super{s}_{\multii;\ExtraDefect} = \det \Gram_\multii\super{s} 
\prod_{\varrho} \frac{ \big[ \frac{1}{2}(r_{\np_{\multii}-1} + \sIndex_{\np_\multii} + s) + \ExtraDefect + 1 \big]! 
\big[ \frac{1}{2}(\sIndex_{\np_\multii} + s - r_{\np_\multii-1}) + \ExtraDefect \big]! [\sIndex_{\np_{\multii}}]! [s + 1]! }{ \big[\frac{1}{2}(r_{\np_{\multii}-1} + \sIndex_{\np_\multii} + s) + 1 \big]! \big[ \frac{1}{2}(\sIndex_{\np_\multii} + s - r_{\np_{\multii}-1}) \big]! [\sIndex_{\np_\multii} + \ExtraDefect]! [s + \ExtraDefect + 1]!} , 
\end{align}
where the product is over all walks $\WalkMultii = (r_1, r_2, \ldots, r_{\np_\multii})$ over $\multii$ with defect $r_{\np_\multii} = s$.
\end{lem}

\begin{proof}
As in the proof of proposition~\ref{GramDetLem}, we may assume that $\Summed_\multii < \ppmin(q)$ throughout. 
Then, lemma~\ref{AlternativeBasisWith_d_ExtraDefects} shows that
$\smash{\big\{ \hcancel{\alpha} \,\big|\, \alpha \in \smash{\LP\super{s}_{\multii;\ExtraDefect}} \big\}}$ is a basis 
for $\smash{\LS\super{s}_{\multii;\ExtraDefect}}$, so 
lemma~\ref{InsProjAlphaLem} with item~\ref{IndOrthBasisLemIt1} of proposition~\ref{IndOrthBasisLem} imply
\begin{align} 
\label{detdet3-0} 
\det \smash{\Gram\super{s}_{\multii;\ExtraDefect}} & 
\underset{\text{lem.}~\ref{InsProjAlphaLem}}{\overset{\eqref{GramMatrix2}}{=}} 
\det [\BiForm{\hcancel{\alpha}}{\hcancel{\beta}}]_{\alpha,\beta \in \LP\super{s}_{\multii;\ExtraDefect}} \\
\label{detdet3-1} 
& \overset{\eqref{WalkBiForm}}{=} \prod_{\varrho_\alpha}  
\bigg( \prod_{j=1}^{\np_{\multii} - 2}  
\frac{\ThetaNet(r_j, r_{j+1}, \sIndex_{j+1})}{(-1)^{r_{j+1}} [ r_{j+1} + 1 ]} \bigg) \frac{\ThetaNet(r_{\np_{\multii}-1}, s + \ExtraDefect, \sIndex_{\np_{\multii}} + \ExtraDefect)}{(-1)^{s + \ExtraDefect} [ s + \ExtraDefect + 1 ]},
\end{align}
where the product is over all walks $\varrho_\alpha = (r_1, r_2, \ldots, r_{\np_{\multii}-1}, s+\ExtraDefect)$ corresponding to some valenced link pattern 
$\alpha \in \smash{\LP\super{s}_{\multii;\ExtraDefect}}$.  Now, 
in light of item~\ref{wmlIt2} in lemma~\ref{WalkMultiiLem} and the definition of the set $\smash{\LP\super{s}_{\multii;\ExtraDefect}}$, 
the map 
\begin{align} 
\varrho_\alpha = (r_1, r_2, \ldots, r_{\np_{\multii}-1}, s+\ExtraDefect) \qquad \longmapsto \qquad 
\varrho = (r_1, r_2, \ldots, r_{\np_{\multii}-1}, s) 
\end{align}
is a bijection from the set $\smash{\LP\super{s}_{\multii;\ExtraDefect}}$ of valenced link patterns 
to the set of all walks 
$\WalkMultii = (r_1, r_2, \ldots, r_{\np_{\multii}-1}, r_{\np_\multii})$
over $\multii$ with defect $r_{\np_\multii} = s$.  Hence, using proposition~\ref{GramDetLem}, we can write
\begin{align} 
\label{detdet4-0} 
\det \smash{\Gram\super{s}_{\multii;\ExtraDefect}} & \overset{(\ref{detdet3-0}-\ref{detdet3-1})}{=} 
\prod_\varrho \bigg( \prod_{j=1}^{\np_{\multii} - 2}  
\frac{\ThetaNet(r_j, r_{j+1}, \sIndex_{j+1})}{(-1)^{r_{j+1}} [ r_{j+1} + 1 ]} \bigg) 
\frac{\ThetaNet(r_{\np_{\multii}-1}, s + \ExtraDefect, \sIndex_{\np_{\multii}} + \ExtraDefect)}{(-1)^{s + \ExtraDefect} [ s + \ExtraDefect + 1 ]} \\
\label{detdet4-1} 
& \overset{\eqref{DetFormula}}{=} \det \Gram_\multii\super{s} 
\prod_\varrho \frac{(-1)^s [ s + 1 ] \ThetaNet(r_{\np_{\multii}-1}, 
s + \ExtraDefect, \sIndex_{\np_{\multii}} + \ExtraDefect)}{(-1)^{s + \ExtraDefect} [ s + \ExtraDefect + 1 ] \ThetaNet(r_{\np_{\multii}-1}, s, \sIndex_{\np_{\multii}})} .
\end{align}
After inserting~\eqref{ThetaFormula0} into~(\ref{detdet4-0}--\ref{detdet4-1}) and simplifying, we arrive with sought identity~\eqref{ExtGram}.
\end{proof}

Lastly, we derive a recursion formula for $\det \smash{\Gram_n\super{s}}$ 
different from the one obtained from lemma~\ref{RecurseLem}, for use in~\cite{fp2}.  
To obtain this recursion, we group the nodes of each link pattern in $\smash{\LP_n\super{s}}$ into $\np_{\multii}$ adjacent bins of sizes 
$\sIndex_1$, $\sIndex_2,\ldots,\sIndex_{\np_\multii}$, and we partition the link pattern into a unique collection of 
$\np_{\multii}+1$ sub-link patterns thus:
\begin{align}\label{PartitionMe} 
\alpha \quad = \quad \vcenter{\hbox{\includegraphics[scale=0.275]{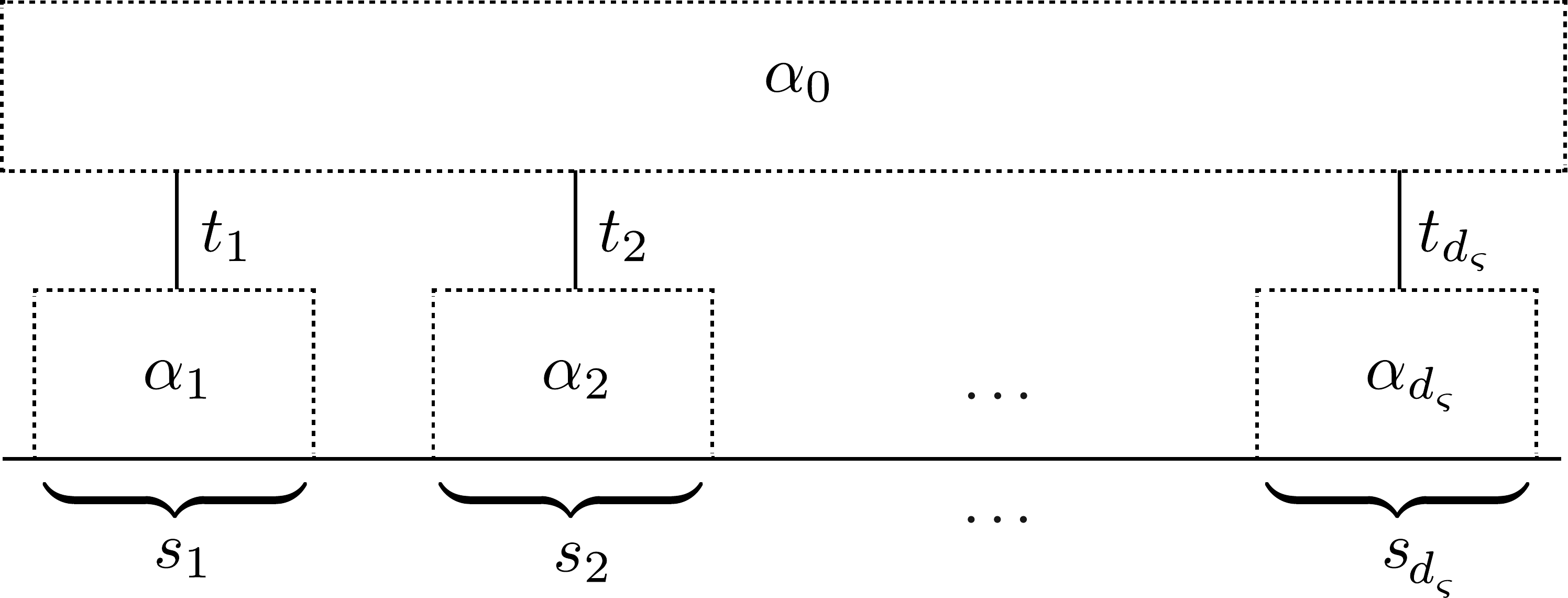} ,}}
\end{align}
where, for all  $i \in \{1,2,\ldots,\np_{\multii}\}$, we have 
\begin{align}\label{DefOfVarTheta} 
\alpha_i \in \LP_{\sIndex_i}\super{t_i} , \quad
\text{for some } t_i\in\DefectSet_{\sIndex_i} , 
\qquad \qquad \text{and} \qquad \qquad 
\alpha_0 \in \SpecialPattern_\vartheta\super{s} \subset \LP_{\Summed_\vartheta}\super{s} , \quad 
\text{ with $\vartheta=(t_1,t_2,\ldots,t_{\np_\multii})$} ,
\end{align}
and where $\SpecialPattern_\vartheta\super{s}$ is the set~\eqref{SpecialPatterns} 
of link patterns that do not have a turn-back link joining two nodes in a common group in~\eqref{PartitionMe} 
(discussed in appendix~\ref{AppWJ}).

Let $\multii = (\sIndex_1,\sIndex_2,\ldots,\sIndex_{\np_\multii})$ be the multiindex consisting of the sizes of the bins. 
To avoid too many subscripts, we denote $\np = \np_\multii$.  
The above type of partitioning of link patterns into link sub-patterns implies the existence of an isomorphism of vector spaces that sends 
$\smash{\LS_n\super{s}}$ onto the vector space 
\begin{align}\label{target} 
\bigoplus_{t_1 \, \in \, \DefectSet_{\sIndex_1}} 
\bigoplus_{t_2 \, \in \, \DefectSet_{\sIndex_2}} \dotsm \bigoplus_{t_{\np}  \, \in  \, \DefectSet_{\sIndex_\np}} 
\big(\Span \SpecialPattern_\vartheta\super{s}\big) \otimes \LS_{\sIndex_1}\super{t_1} \otimes \LS_{\sIndex_2}\super{t_2} \otimes \dotsm \otimes \LS_{\sIndex_{\np}}\super{t_{\np}} , \quad \text{with $\vartheta = (t_1,t_2,\ldots,t_{\np})$,}
\end{align}
by mapping $\alpha$ to the tensor product 
$\alpha_0 \otimes \alpha_1 \otimes \dotsm \otimes \alpha_{\np}$, where $\alpha_i$
are defined relative to $\alpha$ in~\eqref{PartitionMe}.  Next, assuming that $\max \multii < \ppmin(q)$, 
we consider the linear self-map of $\smash{\LS_n\super{s}}$ sending the link pattern $\alpha \in \smash{\LP_n\super{s}}$ 
to the link state
\begin{align}
\munderbar{\alpha} \quad = \quad  \vcenter{\hbox{\includegraphics[scale=0.275]{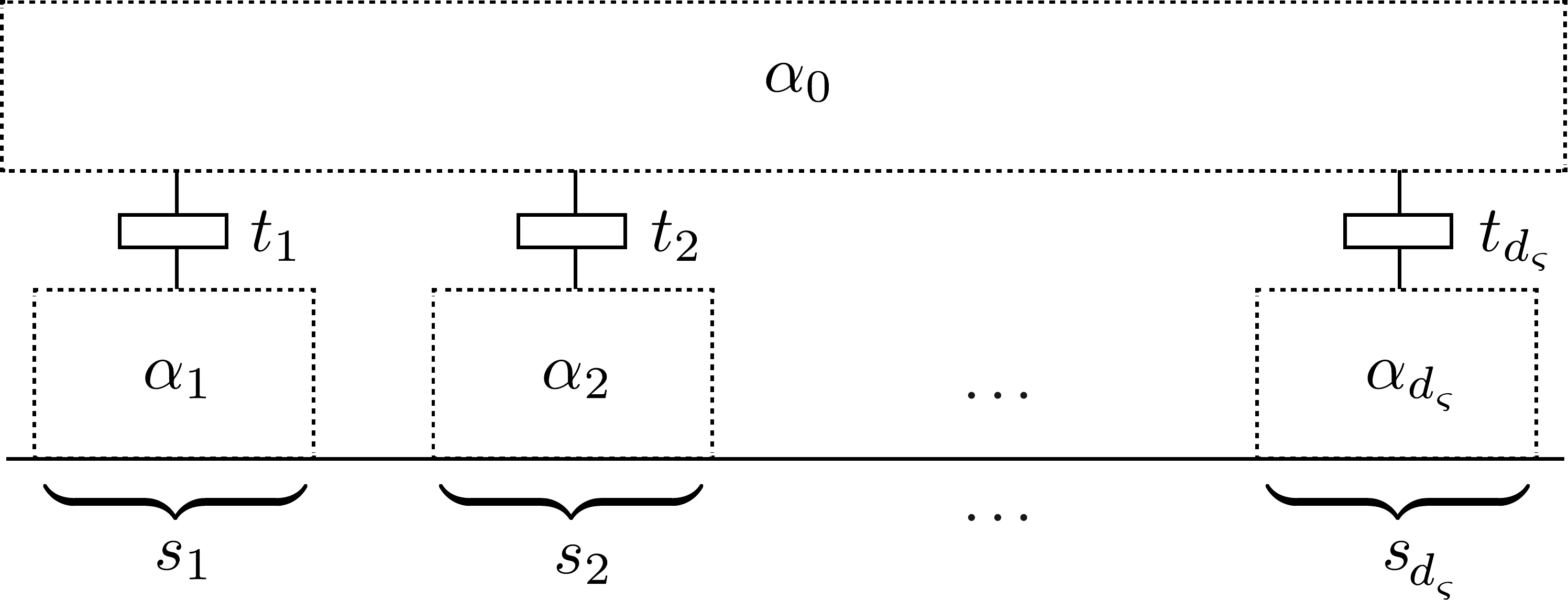} ,}}
\end{align}
with a projector box between 
$\alpha_0$ and each $\alpha_i$. Lemma~\ref{ChangeOfBasisLem} implies that this map is an automorphism of vector spaces with 
an upper-unitriangular matrix representation. 
Hence, we have
\begin{align}\label{detdet2}
\det \Gram_n\super{s} = \det \munderbar{\Gram}_n\super{s} , \qquad \text{where} \quad 
\munderbar{\Gram}_n\super{s}:= \big[ \BiForm{\munderbar{\alpha}}{\munderbar{\beta}} \big]_{\alpha,\beta \in \LP_n\super{s}}. 
\end{align}

With these observations, we are ready to state the new recursion for $\det \smash{\Gram_n\super{s}}$.  

\begin{lem} \label{ForBSALem}
Suppose $\max \multii < \ppmin(q)$, and denote $n = \Summed_\multii$ and $\np = \np_\multii$.  We have the recursion
\begin{align} \label{DetFactored}
\det\Gram_n\super{s} = 
\prod_{t_1 \, \in \, \DefectSet_{\sIndex_1}} 
\prod_{t_2 \, \in \, \DefectSet_{\sIndex_2}} \dotsm 
\prod_{t_\np \, \in \,\DefectSet_{\sIndex_\np}} 
\big(\det\Gram_\vartheta\super{s}\big)^{\Dim_{\sIndex_1}\super{t_1}\Dim_{\sIndex_2}\super{t_2}\dotsm\Dim_{\sIndex_\np}\super{t_\np}} 
\prod_{i=1}^{\np}\big(\det\Gram_{\sIndex_i}\super{t_i}\big)^{p_i}, 
\end{align}
where
\begin{align}
\vartheta := (t_1,t_2,\ldots,t_\np) , 
\qquad \textnormal{and} \qquad  
p_i := \Dim_\vartheta\super{s}\prod_{j \, \neq \, i}^{\np}\Dim_{\sIndex_i}\super{t_i} ,
\end{align}
and $\Dim_\vartheta\super{s}$, $\Dim_{\sIndex_i}\super{t_i}$ are the numbers determined 
respectively by recursions~\eqref{Recursion2} and~\eqref{Recursion}.
\end{lem}

\begin{proof}  
To prove the lemma, we show that the determinant of $\smash{\munderbar{\Gram}_n\super{s}}$ 
equals the right side of~\eqref{DetFactored} and use observation~\eqref{detdet2}. 
To determine $\smash{\det \munderbar{\Gram}_n\super{s}}$, we consider the generic form of the network 
\begin{align} 
\munderbar{\alpha} \BarAction \munderbar{\beta} \quad = \quad 
\vcenter{\hbox{\includegraphics[scale=0.275]{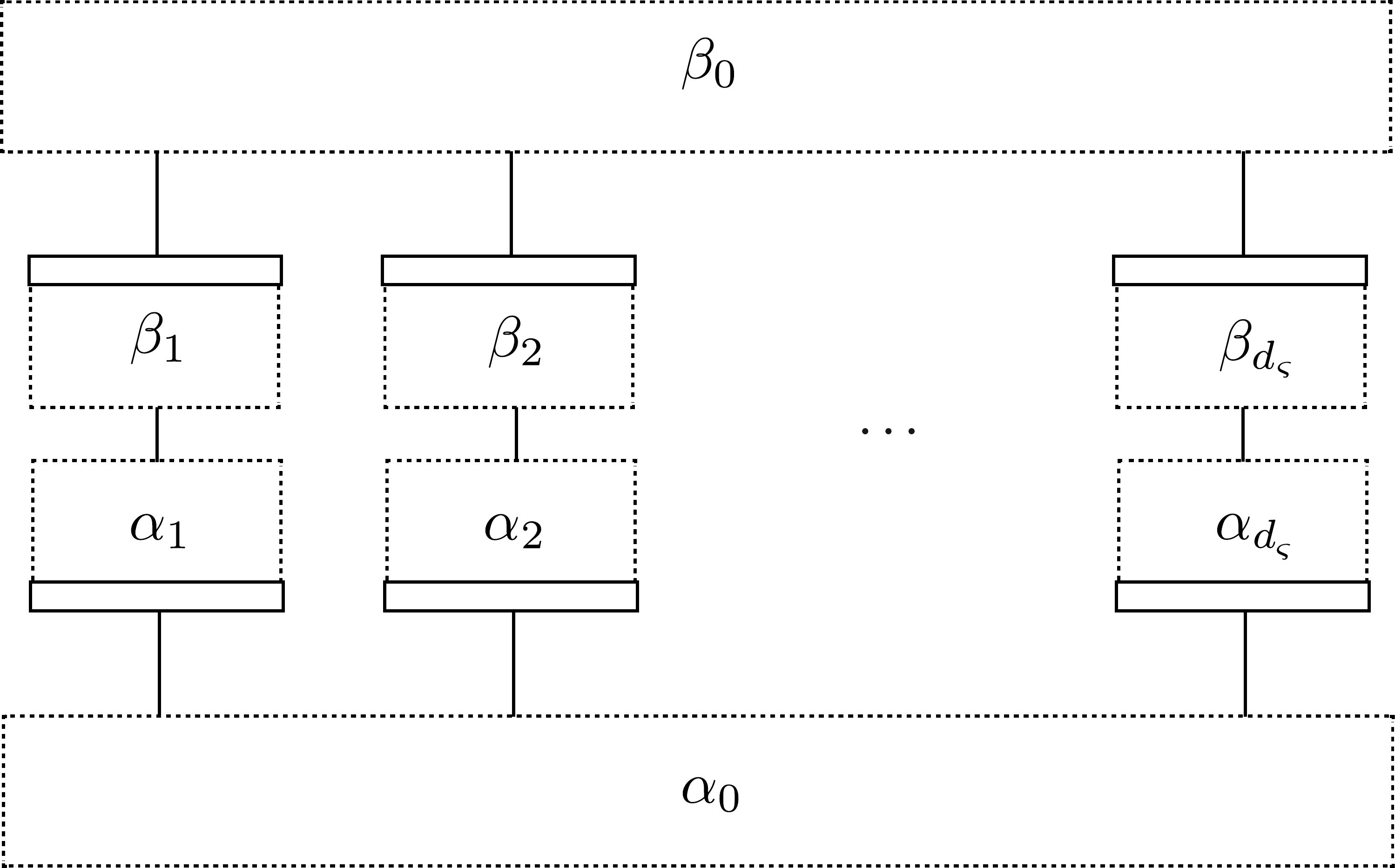} .}}
\end{align}
After using lemma~\ref{ExtractLem} of appendix~\ref{TLRecouplingSect} to factor out the evaluation of the networks
$\alpha_i \BarAction \beta_i$ 
sandwiched between the $i$:th left and right projector boxes, for each $i \in \{1,2,\ldots,d\}$, we find the factorization
\begin{align} \label{FactoredForm} 
\BiForm{\munderbar{\alpha}}{\munderbar{\beta}} 
= \BiForm{\WJProj_\vartheta\alpha_0}{\WJProj_\vartheta\beta_0} \BiForm{\alpha_1}{\beta_1} \BiForm{\alpha_2}{\beta_2} \dotsm \BiForm{\alpha_{\np}}{\beta_{\np}} .
\end{align}
Because the spaces $\smash{\LS_{\sIndex_i}\super{r_i}}$ and $\smash{\LS_{\sIndex_i}\super{t_i}}$ are orthogonal if $r_i\neq t_i$, 
and the linear map 
\begin{align}
\alpha \quad \longmapsto \quad \alpha_0 \otimes \alpha_1 \otimes \dotsm \otimes \alpha_\np ,
\end{align}
for all $(n,s)$-link patterns $\alpha$ is an isomorphism of vector spaces from $\smash{\LS_n\super{s}}$ to the space~\eqref{target}, we infer 
from factorization~\eqref{FactoredForm}
that the matrix $\smash{\munderbar{\Gram}_n\super{s}}$ equals the following direct sum of  
tensor products of Gram matrices:
\begin{align}
\munderbar{\Gram}_n\super{s} = 
\bigoplus_{t_1 \, \in \, \DefectSet_{\sIndex_1}} 
\bigoplus_{t_2 \, \in \, \DefectSet_{\sIndex_2}} \dotsm 
\bigoplus_{t_\np \, \in \, \DefectSet_{\sIndex_\np}} \Gram_\vartheta\super{s} \otimes 
\Gram_{\sIndex_1}\super{t_1}\otimes\Gram_{\sIndex_2}\super{t_2} \otimes \dotsm \otimes \Gram_{\sIndex_\np}\super{t_\np} . 
\end{align}
After taking the determinant of both sides, using the well-known formula for the determinant of the tensor product of matrices, and recalling 
from~\eqref{detdet2} that $\det \smash{\munderbar{\Gram}_n\super{s}} = \det \smash{\Gram_n\super{s}}$, we finally arrive with~\eqref{DetFactored}.
\end{proof} 

\section{Radical of the link state bilinear form} \label{RadicalSect}

In this section, we determine the dimension of and a basis for the radical of the standard module 
$\smash{\LS_\multii\super{s}}$ for all multiindices 
$\multii \in \{\OneVec{0}\} \cup \smash{\bZpos^\#}$ and integers $s \in \DefectSet_\multii$.  
Corollaries~\ref{RadicalCor2} and~\ref{RadicalCor3} below treat the trivial cases 
with $\Summed_\multii < \ppmin(q)$.  For the nontrivial cases, 
we divide the problem into two parts: the specific case with $\multii = \OneVec{n}$ for some $n \in \bZnn$, 
treated in section~\ref{rofSect1}, and the case of a general multiindex 
$\multii \in \{\OneVec{0}\} \cup \smash{\bZpos^\#}$, treated in section~\ref{rofSect2}. 
In section~\ref{rofSect31},
we use these results to prove the remarkable fact that $\rad \smash{\LS_\multii\super{s}}$ is trivial if and only if 
$\rad \smash{\LS_{\Summed_\multii}\super{s}}$ is trivial, 
and we then use this fact with our other results to determine all pairs $(\multii,s)$ such that $\rad \smash{\LS_\multii\super{s}}$ is trivial. 
Finally, in section~\ref{rofSect32}, we study cases in which and conditions under which $\rad \smash{\LS_\multii\super{s}}$ equals the entire module $\smash{\LS_\multii\super{s}}$.

\begin{cor} \label{RadicalCor2} 
Suppose $\Summed_\multii < \ppmin(q)$.  
Then we have $\rad \smash{\LS_\multii\super{s}} = \{0\}$, for all $s \in \DefectSet_\multii$, and thus, $\rad \LS_\multii = \{0\}$.
\end{cor}

\begin{proof} 
The claim immediately follows from proposition~\ref{VanishDetLem2} and direct-sum decomposition~\eqref{RadDirSum}.
\end{proof}

\begin{cor} \label{RadicalCor3} 
Suppose $\ppmin(q) = \infty$.  
Then we have $\rad \smash{\LS_\multii} = \{0\}$, for all $\multii \in \{\OneVec{0}\} \cup \smash{\bZpos^\#}$.
\end{cor}

\begin{proof} 
The claim immediately follows from corollary~\ref{RadicalCor2}.
\end{proof}

Corollary~\ref{RadicalCor3} settles the case that $\ppmin(q) = \infty$: all radicals of $\smash{\LS_\multii\super{s}}$ are trivial.  
On the other hand, if $\ppmin(q) < \infty$, then the radical of these standard modules is not trivial for certain 
$\multii \in \smash{\bZpos^\#}$ and $s \in \DefectSet_\multii$.  
In the remainder of this section, we completely determine these radicals, for all values of $\ppmin(q)$, 
all multiindices $\multii \in \smash{\bZpos^\#}$, and all integers $s \in \DefectSet_\multii$.  
These forthcoming results also include the case $\ppmin(q) = \infty$, already settled above, as a special instance.

\subsection{Radical at roots of unity} \label{rofSect1}

In this section, we find the dimension of and a basis for the radical of $\LS_n\super{s}$ for all integers $n \in \bZnn$ and $s \in \DefectSet_n$. 
Throughout, we use the integers $k_s$ and $R_s$ defined in~\eqref{skDefn}. 
First, we treat the case of $R_s = 0$ (i.e., $s = \Delta_{k_s}$).

\begin{lem} \label{EasyRadLem} 
\textnormal{\cite[corollary~\red{4.8}]{rsa}} 
Suppose $\pmin(q) \,|\, (s+1)$.  Then we have $\rad \LS_n\super{s} = \{0\}$.
\end{lem}

\begin{proof} 
First, we assume that $q \neq \pm1$.  
By definition~\eqref{Qinteger}, for all nonzero integers $k$, 
we have $[k]_q = 0$ if and only if $\pmin(q) \,|\, k$, and these zeros of $q \mapsto [k]_q$ are of first order. 
Also, our assumption that $\pmin(q) \,|\, (s+1)$ implies that 
$j$ is a multiple of $\pmin(q)$ if and only if $s+j+1$ is a multiple of $\pmin(q)$.  
Hence, we have
\begin{align} \label{zeroscancel}
\Big\{
[j]_q =0 \quad \Longleftrightarrow \quad [s+j+1]_q = 0 
\Big \}
\qquad \qquad \overset{\eqref{Qinteger}}{\Longrightarrow} \qquad \qquad 
0 < \bigg| \frac{[s+j+1]_{q}}{[j]_{q}} \bigg| < \infty .
\end{align}
Thus, from~\eqref{zeroscancel} and formula~\eqref{RidoutDet} for $\det \smash{\Gram_n\super{s}}$ from lemma~\ref{RidoutDetLem}, 
we see that $\det \smash{\Gram_n\super{s}} \neq 0$, so $\rad \smash{\LS_n\super{s}} = \{0\}$.

Last, we assume that $q = \pm 1$, in which case we have $\pmin(q) \,|\, (r+1)$ for all $r \geq 0$.  
If $q = \pm1$, then clearly no quantum integer~\eqref{Qinteger} 
except $[0]$ vanishes.  Hence, it is evident from~\eqref{RidoutDet} that $\det \smash{\Gram_n\super{s}} \neq 0$, so $\rad \smash{\LS_n\super{s}} = \{0\}$.
\end{proof}

Now we use lemma~\ref{EasyRadLem} to determine the radical of $ \LS_n\super{s}$ when $\ppmin(q) < \infty$.  
First, we recall from section~\ref{ConformalBlocksSect} the definition (\ref{Jindex0},~\ref{Jindex2}) 
of the tail of a link pattern $\alpha$ and definition~\ref{TrivalentLinkStateDef}
of the corresponding trivalent link state $\hcancel{\,\alpha}$. 
In the special case that $\multii = \OneVec{n}$ for some $n \in \bZnn$, stopping condition~\ref{StopIt2} 
in definition~\ref{TrivalentLinkStateDef} for forming the tail of $\alpha$ cannot occur.  
Hence, the definition of $\hcancel{\,\alpha}$ reduces to the following: to obtain $\hcancel{\alpha}$ from the link pattern $\alpha$, 
we replace each open three-vertex in the tail of 
the walk representation of $\alpha$ with a closed three-vertex.

We let $\smash{\mathsf{T}_n\super{s}}$ denote the collection of all tails pertaining to $(n,s)$-link patterns:
\begin{align} \label{Alltails}
\mathsf{T}_n\super{s} := \big\{ \text{tail}(\alpha) \,|\, \alpha \in \LP_n\super{s} \big\}. 
\end{align}
Any tail in $\smash{\mathsf{T}_n\super{s}}$ is exactly one of two possible types:
a \emph{radical tail} has $r_J = \Delta_{k_s+1}$ 
(i.e., stopping condition~\ref{StopIt1} in definition~\ref{TrivalentLinkStateDef} occurs), and
a \emph{moderate tail} has $r_J = \Delta_{k_s}$ (i.e., stopping condition~\ref{StopIt3} 
in definition~\ref{TrivalentLinkStateDef} occurs), where $k_s$ is given in~\eqref{skDefn}.  
We let $\smash{\mathsf{R}_n\super{s}}$ and $\smash{\mathsf{M}_n\super{s}}$ 
respectively denote the collection of all radical tails and all moderate tails of $(n,s)$-link patterns.  
Because stopping condition~\ref{StopIt2} 
in definition~\ref{TrivalentLinkStateDef} cannot occur if $\multii = \OneVec{n}$, we have  
\begin{align} 
\smash{\mathsf{R}_n\super{s}} \cup \smash{\mathsf{M}_n\super{s}} = \smash{\mathsf{T}_n\super{s}}.
\end{align}
We also note that, for each tail $\tau = (r_{J+1}, r_{J+2}, \ldots, r_n) \in \smash{\mathsf{T}_n\super{s}}$, the following properties hold:
\begin{align} \label{rJ+1} 
r_J = \begin{cases} \Delta_{k_s+1}, & \tau \in \smash{\mathsf{R}_n\super{s}}, \\ 
\Delta_{k_s}, & \tau \in \smash{\mathsf{M}_n\super{s}}, 
\end{cases} 
\qquad \qquad \text{and} \qquad  \qquad
r_{J+1} = 
\begin{cases} 
r_J-1, & \tau \in \smash{\mathsf{R}_n\super{s}}, 
\\ r_J+1, & \tau \in \smash{\mathsf{M}_n\super{s}}.
\end{cases} 
\end{align}

Or next goal, proposition~\ref{BigTailLem}, is to prove that
$\smash{\big\{ \hcancel{\,\alpha} \,\big|\, \alpha \in \smash{\LP_\multii\super{s}}, \, \textnormal{tail}(\alpha) \in \mathsf{R}_\multii\super{s} \big\} }$ 
is a basis for the radical of $\smash{\LS_n\super{s}}$. 
For this, we first observe that link patterns with different tails span orthogonal subspaces of $\smash{\LS_n\super{s}}$.

\begin{lem} \label{OrthTailsLem} 
For all link patterns $\alpha, \beta \in \smash{\LP_n\super{s}}$, we have 
\begin{align} \label{OrthTails} 
\textnormal{tail}(\alpha) \neq \textnormal{tail}(\beta) \qquad \Longrightarrow \qquad \BiForm{\hcancel{\,\alpha}}{\hcancel{\,\beta}} = 0.
\end{align}
\end{lem}

\begin{proof} 
Assuming that $\alpha, \beta \in \smash{\LP_n\super{s}}$ satisfy $\textnormal{tail}(\alpha) \neq \textnormal{tail}(\beta)$, with 
$\varrho_\alpha = (r_1, r_2, \ldots, r_n)$ and $\varrho_\beta = (r_1', r_2', \ldots, r_n')$, we set $I := \max \{ j \in \bZnn \,|\, r_j \neq r_j' \}$.  
Then we have $\max( J_\alpha(q), J_\beta(q)) \leq I$, so the network $\hcancel{\,\alpha} \BarAction \hcancel{\,\beta}$ has the form
\begin{align} \label{AlphBet}
\hcancel{\,\alpha} \BarAction \hcancel{\,\beta} \quad = \quad
\vcenter{\hbox{\includegraphics[scale=0.275]{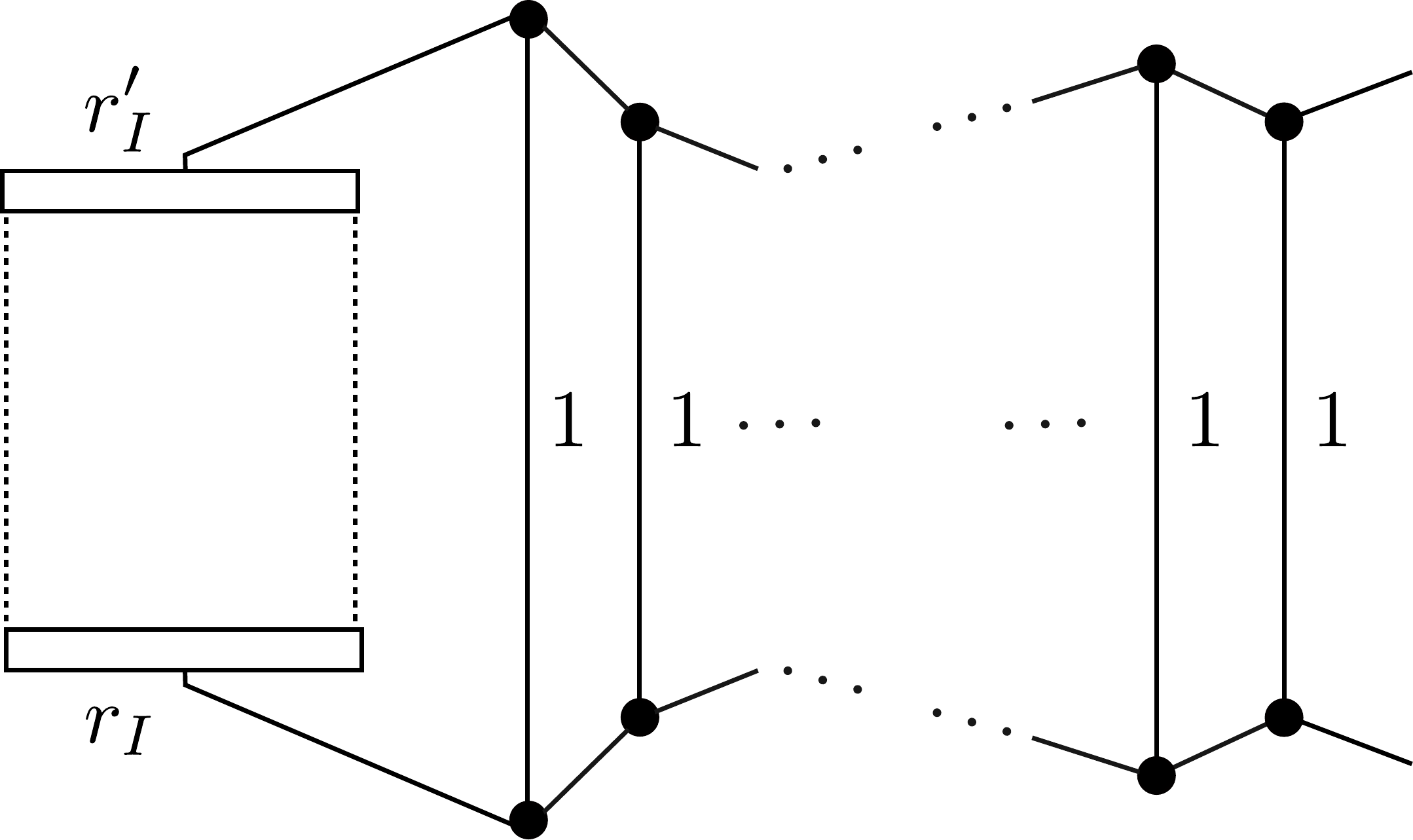} .}} 
\end{align}
The sizes $r_I$ and $r_I'$ of the two leftmost projector boxes are different by definition of $I$.  
Hence, there exists a link with both of its endpoints touching the larger of these two boxes, so we have $\BiForm{\hcancel{\,\alpha}}{\hcancel{\,\beta}} = 0$.
\end{proof}

Orthogonality of link patterns with different tails immediately gives a direct-sum decomposition of the radical:

\begin{cor} \label{RadDirectCor} 
We have the direct-sum decomposition
\begin{align} \label{radDirect} 
\rad \LS_n\super{s} = \bigoplus_{\tau \, \in \, \smash{\mathsf{T}_n\super{s}}} \rad \Span \big\{ \hcancel{\,\alpha} \,\big|\, \alpha \in \smash{\LP_n\super{s}}, \, \textnormal{tail}(\alpha) = \tau \big\}.
\end{align}
\end{cor}

\begin{proof} 
Item~\ref{IndOrthBasisLemIt3} of proposition~\ref{IndOrthBasisLem} implies the direct-sum decomposition
\begin{align} \label{DirectLS1} 
\LS_n\super{s} =  \bigoplus_{\tau \, \in \, \smash{\mathsf{T}_n\super{s}}} \Span \big\{ \hcancel{\,\alpha} \,\big|\, 
\alpha \in \smash{\LP_n\super{s}}, \, \textnormal{tail}(\alpha) = \tau \big\} .
\end{align}
Also, lemma~\ref{OrthTailsLem} implies that the spans in the direct sum~\eqref{DirectLS1} are orthogonal.  
Hence,~\eqref{radDirect} follows from~\eqref{DirectLS1}. 
\end{proof}

Now, to determine the radical of $\smash{\LS_n\super{s}}$, we only need to understand the summands of~\eqref{radDirect}.  

\begin{lem} \label{RadCasesLem} 
Suppose $\tau \in \smash{\mathsf{T}_n\super{s}}$.  Then we have
\begin{align} \label{radCases} \rad \Span \big\{ \hcancel{\,\alpha} \,\big|\, \alpha \in \smash{\LP_n\super{s}}, \,\textnormal{tail}(\alpha) = \tau \big\} = \begin{cases}    \Span \big\{ \hcancel{\,\alpha} \,\big|\, \alpha \in \smash{\LP_n\super{s}}, \, \textnormal{tail}(\alpha) = \tau \big\}, & \textnormal{$\tau \in \smash{\mathsf{R}_n\super{s}}$}, \\ \{0\}, & \textnormal{$\tau \in \smash{\mathsf{M}_n\super{s}}$}. \end{cases} 
\end{align}
\end{lem}

\begin{proof} 
Let $\beta$ and $\gamma$ be two link states in the span of 
$\smash{ \big\{ \hcancel{\,\alpha} \,\big|\, \alpha \in \smash{\LP_n\super{s}}, \, \textnormal{tail}(\alpha) = \tau \big\}}$,
and denote their common tail by $\tau = (r_J, r_{J+1}, \ldots, r_n)$, 
with~\eqref{rJ+1}.
Then, the bilinear form $\BiForm{\beta}{\gamma}$ equals the evaluation of the network
\begin{align} \label{AlphBetNet}
\beta \BarAction \gamma \quad = \quad \vcenter{\hbox{\includegraphics[scale=0.275]{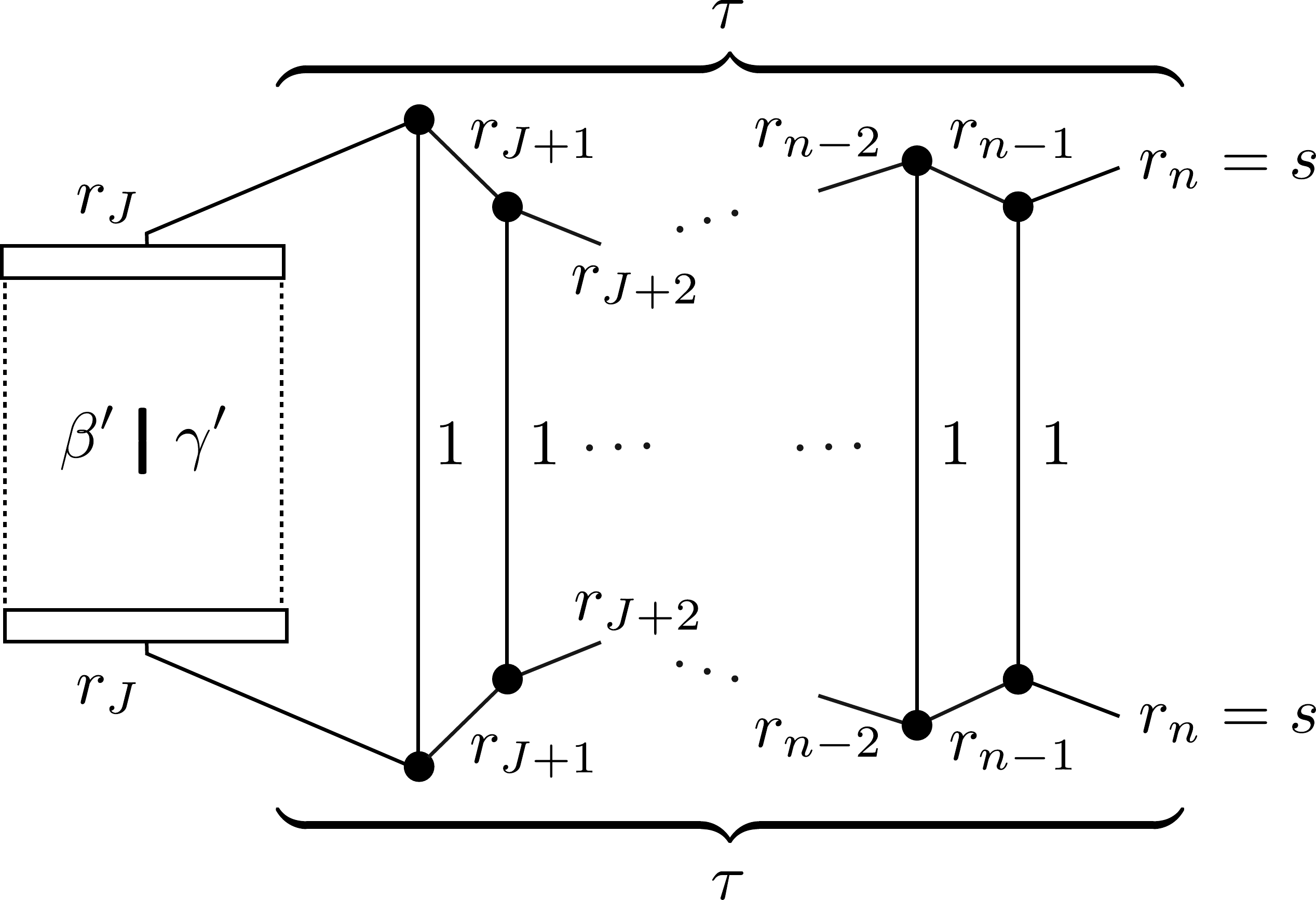} ,}} 
\end{align}
where $\beta'$ and $\gamma'$ are $(J,r_J)$-link states.  
We use lemmas~\ref{ExtractLem} and~\ref{LoopErasureLem} of appendix~\ref{TLRecouplingSect} to evaluate this network, obtaining
\begin{align} 
\label{DotBiIt0} 
\BiForm{\beta}{\gamma} & \quad \overset{\eqref{ExtractID}}{=} \quad \BiForm{\beta'}{\gamma'} \,\, \times 
\,\, 
\left( \;
\vcenter{\hbox{\includegraphics[scale=0.275]{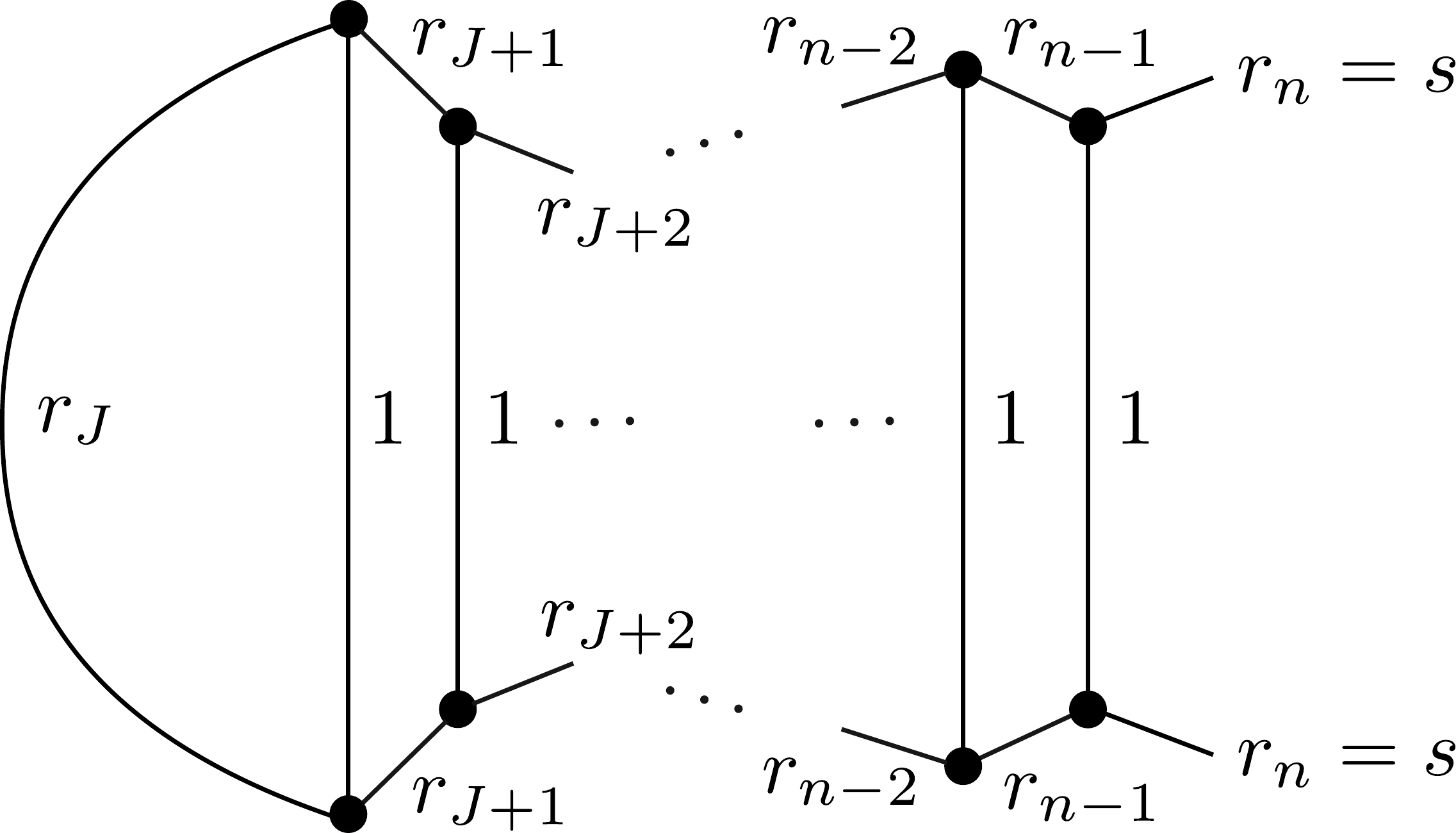}}} \; \right) \\[1em]
\label{DotBiIt} 
& \quad \overset{\eqref{LoopErasure1}}{=} \quad \BiForm{\beta'}{\gamma'} \,\, \prod_{j = J}^{n-1} \frac{ \ThetaNet( r_j, r_{j+1}, 1 )}{ (-1)^{r_{j + 1}} [r_{j + 1}+1]}.
\end{align}
Now, lemma~\ref{ThetaInFiniteAndNonzeroLem} shows that the product of the factors in~\eqref{DotBiIt} with $j > J$ is finite and does not vanish,
\begin{align} \label{NonVan} 
0 < \bigg| \prod_{j = J+1}^{n-1} \frac{ \ThetaNet( r_j, r_{j+1}, 1 )}{ (-1)^{r_{j + 1}} [r_{j + 1}+1]} \bigg| < \infty .
\end{align}
Using (\ref{ThetaFormula0},~\ref{rJ+1}), we find that the factor with $j = J$ equals
\begin{align} \label{ThetaCases} 
\frac{\ThetaNet(r_J,r_{J+1},1)}{(-1)^{r_{J+1}}[r_{J+1}+1]} \underset{\eqref{rJ+1}}{\overset{\eqref{ThetaFormula0}}{=}} 
\begin{cases} 
0, 
& \tau \in \smash{\mathsf{R}_n\super{s}}, \\ 1, & \tau \in \smash{\mathsf{M}_n\super{s}}. 
\end{cases}
\end{align}

Next, according to lemma~\ref{EasyRadLem} together with~\eqref{rJ+1}, we have $\rad \smash{\LS_J\super{r_J}} =\{0\}.$  
Hence, for each nonzero link state $\beta' \in \smash{\LS_J\super{r_J}}$, there exists a companion link state $\gamma_\beta' \in \smash{\LS_J\super{r_J}}$ such that 
\begin{align} \label{alphaprime} 
\BiForm{\beta'}{\gamma_\beta'} \neq 0. 
\end{align}
We define $\gamma_\beta$ to be the link state obtained by setting $\gamma' = \gamma_\beta'$ in
\begin{align}
\gamma \quad = \quad \vcenter{\hbox{\includegraphics[scale=0.275]{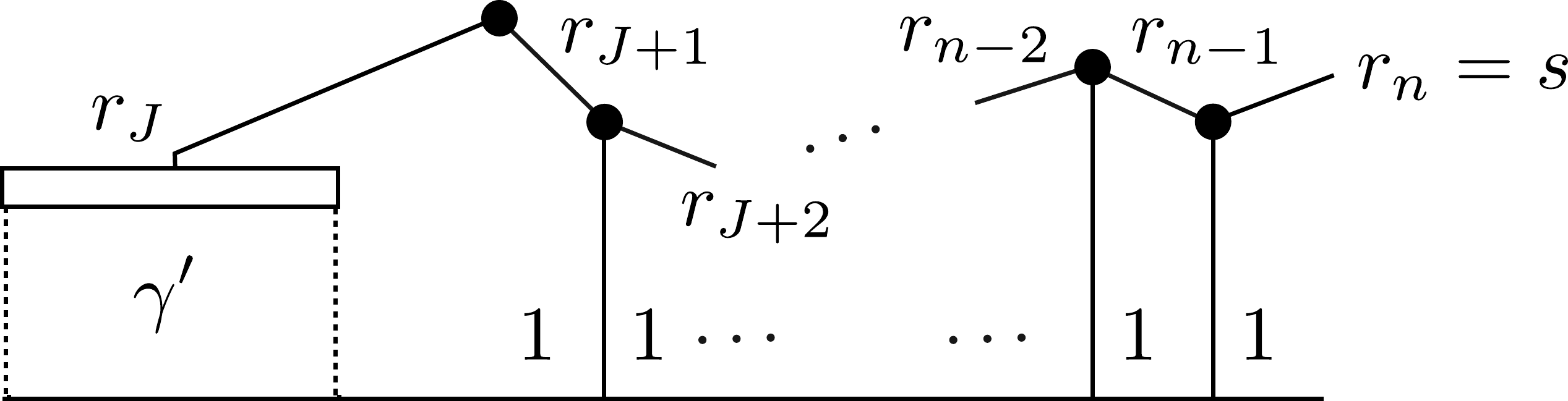} .}} 
\end{align} 
In conclusion, we may combine (\ref{DotBiIt0}--\ref{alphaprime}) to arrive with the following result: for every link state $\beta$ in the span of 
$\smash{ \big\{ \hcancel{\,\alpha} \,\big|\, \alpha \in \smash{\LP_n\super{s}}, \, \textnormal{tail}(\alpha) = \tau \big\}}$, we have
\begin{align} 
\BiForm{\beta}{\gamma} \quad \begin{cases} \text{ $=0 \,\,$ for all $\gamma \in \Span \big\{ \hcancel{\,\alpha} \,\big|\, \alpha \in \LP_n\super{s}, 
\textnormal{tail}(\alpha) = \tau \big\}$}, & \tau \in \smash{\mathsf{R}_n\super{s}}, \\ \text{ $\neq 0 \,\,$ if $\gamma = \gamma_\beta$}, 
& \tau \in \smash{\mathsf{M}_n\super{s}}.
\end{cases} 
\end{align}
This final result is equivalent to assertion~\eqref{radCases}.
\end{proof}

\begin{prop} \label{BigTailLem} 
The collection
\begin{align} \label{BigTail} 
\big\{ \hcancel{\,\alpha} \,\big|\, \alpha \in \smash{\LP_n\super{s}}, \, \textnormal{tail}(\alpha) \in \smash{\mathsf{R}_n\super{s}} \big\} 
\end{align}
is a basis for $\smash{\rad \LS_n\super{s}}$.
\end{prop}

\begin{proof}  
Combining corollary~\ref{RadDirectCor} with lemma~\ref{RadCasesLem}, we obtain the direct-sum decomposition
\begin{align} 
\rad \LS_n\super{s} \underset{\eqref{radCases}}{\overset{\eqref{radDirect}}{=}} \bigoplus_{\tau \, \in \, \smash{\mathsf{R}_n\super{s}}} 
\Span \{ \hcancel{\,\alpha} \,|\, \alpha \in \smash{\LP_n\super{s}}, \, \textnormal{tail}(\alpha) = \tau \} 
= \Span \big\{ \hcancel{\,\alpha} \,\big|\, \alpha \in \smash{\LP_n\super{s}}, \, \textnormal{tail}(\alpha) \in \smash{\mathsf{R}_n\super{s}} \big\}. 
\end{align}
Item~\ref{IndOrthBasisLemIt3} of proposition~\ref{IndOrthBasisLem} implies that the collection~\eqref{BigTail} is linearly independent.  
Hence, it is basis for $\smash{\rad \LS_n\super{s}}$.
\end{proof}

Using proposition~\ref{BigTailLem}, we next determine the dimension of $\rad \smash{\LS_n\super{s}}$. 
We recall from lemma~\ref{LSDimLem2} that the dimension of the standard module $\smash{\LS_n\super{s}}$ is $\smash{\Dim_n\super{s}}$, 
the unique solution to recursion problem~\eqref{Recursion}.
Then, we define the numbers $\smash{\hcancel{\Dim}_n\super{s}}$ for all integers $n \geq 0$ and
$s \in \DefectSet_n$ to be the unique solution to the recursion 
\begin{align} \label{RadRecurs} 
\hspace*{-3mm}
\hcancel{\Dim}_n\super{s} = \begin{cases} 0, & R_s = 0, \\ 
\hcancel{\Dim}_{n-1}\super{s-1} + \Dim_{n-1}\super{s+1}, & R_s = \pmin(q) - 1, \\ 
\hcancel{\Dim}_{n-1}\super{s-1} + \hcancel{\Dim}_{n-1}\super{s+1}, & R_s \in \{1, 2, \ldots, \pmin(q) -2\},
\end{cases}  
\quad \quad \text{and} \quad \qquad \hcancel{\Dim}_1\super{1} = 0 ,
\end{align}
with the convention that $\hcancel{\Dim}_{n-1}\super{-1} = 0$.
This recursion is equivalent to the recursion problem in~\cite[proposition~\red{4.5}]{rsa}.

The following lemma is similar to observation~\eqref{CountLP} in item~\ref{wmlIt4} of lemma~\ref{WalkMultiiLem} with $\multii = \OneVec{n}$.

\begin{lem} \label{RadWalkLem} 
We have 
\begin{align} \label{RadWalk} 
\hcancel{\Dim}_n\super{s} = \#\left\{\parbox{8.3cm}{\textnormal{walks $\varrho$ over $\OneVec{n}$ with defect $s$ and such that, when} \\ 
\textnormal{followed backward, hit height $\Delta_{k_s+1}$ before height $\Delta_{k_s}$}}\right\}.
\end{align} 
\end{lem}

\begin{proof} 
We recall that by~\eqref{CountLP}, the quantity $\smash{\Dim_n\super{s}}$ equals the number of walks over $\OneVec{n}$ with defect $s$.  
By considering the last step of an arbitrary walk in the set appearing on the right side of~\eqref{RadWalk}, we see that
the cardinality of this set satisfies recursion~\eqref{RadRecurs}. If $\multii = (1)$ and $s=1$, then there are no such walks, 
so the initial condition also holds.
We conclude that the right side of asserted equation~\eqref{RadWalk} equals 
the unique solution $\smash{\hcancel{\Dim}_n\super{s}}$ to recursion problem~\eqref{RadRecurs}.
\end{proof}

By (\ref{DeltaDefn},~\ref{skDefn}) we have $\Delta_{k_s} = \Delta_0 = -1$ if $s + 1 < \pmin(q)$. Because a walk over $\multii$ cannot have negative height, it follows that 
if $s + 1 < \pmin(q)$, then $\smash{\hcancel{\Dim}_n\super{s}}$ equals the number of walks $\varrho$ over $\multii$ with defect $s$ and that hit height $\pmin(q) - 1$.

\begin{cor} \label{DnsLemAndRadDimCor} 
We have
\begin{align} \label{DnsDotDefn} 
 \dim \rad \LS_n\super{s} = 
\#\mathsf{R}_n\super{s} = \hcancel{\Dim}_n\super{s}.
\end{align} 
\end{cor}
\begin{proof} 
The first equality in~\eqref{DnsDotDefn} immediately follows from proposition~\ref{BigTailLem}.
The second equality follows from identity~\eqref{RadWalk} of lemma~\ref{RadWalkLem} and the definition of a radical tail.
\end{proof}

\subsection{Valenced radical at roots of unity} \label{rofSect2}

In this section, we find the dimension of and a basis for the radical of $\smash{\LS_\multii\super{s}}$ 
for all multiindices $\multii \in \{\OneVec{0}\} \cup \smash{\bZpos^\#}$ and integers 
$s \in \DefectSet_\multii$. 
For convenience, we assume that $\ppmin(q) < \infty$, although this condition is not necessary for the results in this section to be true. 
As before, we use the integers $k_s$ and $R_s$ defined in~\eqref{skDefn}.

To begin, we prove in corollary~\ref{EmbProjCor} in appendix~\ref{AppWJ} that
the radical of $\smash{\LS_\multii\super{s}}$ is given by a projection of
the corresponding radical of $\smash{\LS_{\Summed_\multii}\super{s}}$ 
under the map defined via~(\ref{WJProjHatEmb},~\ref{ProjHatDef1-1}),
\begin{align} \label{EmbProj2}
\rad \LS_\multii\super{s} \overset{\eqref{EmbProj22}}{=} \WJProjHat_\multii \rad \LS_{\Summed_\multii}\super{s} .
\end{align}

\begin{cor} \label{EasyRadCor} 
Suppose $\max \multii < \ppmin(q)$. If $\pmin(q) \,|\, (s+1)$, then $\rad \smash{\LS_\multii\super{s}} = \{0\}$.
\end{cor}

\begin{proof} 
This immediately follows from lemma~\ref{EasyRadLem} with~\eqref{EmbProj2}.
\end{proof}

We first recall from section~\ref{ConformalBlocksSect} the definition 
(\ref{TailDef},~\ref{Jindex2}) of the tail of a valenced  link pattern $\alpha$, and we denote
\begin{align} \label{tails} 
\mathsf{T}_\multii\super{s} := \big\{ \text{tail}(\alpha) \,\big|\, \alpha \in \smash{\LP_\multii\super{s}} \big\}. 
\end{align} 
Recalling definition~\ref{TrivalentLinkStateDef} of the trivalent link state $\hcancel{\,\alpha}$ 
for any $(\multii,s)$-valenced link pattern $\alpha$, we say that the tail of $\alpha$ is 
\begin{enumerate}
\itemcolor{red}
\item \emph{type-one radical tail} if condition~\ref{StopIt1} 
in definition~\ref{TrivalentLinkStateDef} occurs in the determination of the trivalent link state $\hcancel{\alpha}$,
\item \emph{type-two radical tail} if condition~\ref{StopIt2} 
in definition~\ref{TrivalentLinkStateDef} occurs in the determination of the trivalent link state $\hcancel{\alpha}$,
\item \emph{moderate tail} if condition~\ref{StopIt3} 
in definition~\ref{TrivalentLinkStateDef} occurs in the determination of the trivalent link state $\hcancel{\alpha}$.
\end{enumerate}
Figures~\ref{fig4-1},~\ref{fig4-2}, and~\ref{fig4-3} show examples of these tails.
We also say that the tail of $\alpha$ is a moderate tail if none of these stopping condition occurs, i.e., if $J_\alpha(q) = -\infty$.  Finally, we set
\begin{align} 
\mathsf{R}_{\multii,1}\super{s} &:= \big\{ \text{tail}(\alpha) \,\big|\, \alpha \in \smash{\LP_\multii\super{s}}, \, \text{$\text{tail}(\alpha)$ is a type-one radical tail} \big\}, \\ 
\mathsf{R}_{\multii,2}\super{s} &:= \big\{ \text{tail}(\alpha) \,\big|\, \alpha \in \smash{\LP_\multii\super{s}}, \, \text{$\text{tail}(\alpha)$ is a type-two radical tail} \big\}, \\ 
\mathsf{M}_\multii\super{s} &:= \big\{ \text{tail}(\alpha) \,\big|\, \alpha \in \smash{\LP_\multii\super{s}}, \, \text{$\text{tail}(\alpha)$ is a moderate tail} \big\} ,
\end{align}
we denote
\begin{align} \label{R12} 
\mathsf{R}_\multii\super{s} := \mathsf{R}_{\multii,1}\super{s} \cup \mathsf{R}_{\multii,2}\super{s} , 
\end{align} 
and we call an element of this set a ``radical tail." 
The union of these sets equals the collection of all tails $\smash{\mathsf{T}_\multii\super{s}}$:
\begin{align} \label{AllTails}
\smash{\mathsf{T}_\multii\super{s}} = \smash{\mathsf{R}_\multii\super{s}} \cup  \smash{\mathsf{M}_\multii\super{s}}.
\end{align}

Our next goal, theorem~\ref{BigTailLem2}, is to prove that the set 
$\smash{\big\{ \hcancel{\,\alpha} \,\big|\, \alpha \in \smash{\LP_\multii\super{s}}, \, \textnormal{tail}(\alpha) \in \mathsf{R}_\multii\super{s} \big\} }$ 
is a basis for the radical of $\smash{\LS_\multii\super{s}}$.
The logic of this work is similar to the proof of proposition~\ref{BigTailLem}: 
we decompose the radical of $\smash{\LS_\multii\super{s}}$ into a direct sum of 
certain subspaces labeled by either type-one radical tails, type-two radical tails, or moderate tails, 
and we explicitly determine these subspaces.  
To establish this, we employ proposition~\ref{BigTailLem} from section~\ref{rofSect1}, 
which already gives the radical of the related standard module $\smash{\LS_{\Summed_\multii}\super{s}}$.  

\begin{lem} \label{FirstOrthLem} 
Suppose $\max \multii < \ppmin(q)$.  We have
\begin{align} \label{FirstOrth} 
\big\{ \hcancel{\,\alpha} \,\big|\, \alpha \in \smash{\LP_\multii\super{s}}, \, \textnormal{tail}(\alpha) \in \mathsf{R}_{\multii,1}\super{s} \big\} 
\subset \rad \smash{\LS_\multii\super{s}}. 
\end{align} 
\end{lem}

\begin{proof} 
Let $\beta \in \smash{\LP_{\Summed_\multii}\super{s}}$ be the link pattern obtained by separating the $i$:th node of $\alpha \in \smash{\LP_\multii\super{s}}$ into $\sIndex_i$ adjacent nodes, for each $i \in \{1,2,\ldots,\np_\multii\}$, 
without changing the connectivities of the links in $\alpha$.  Then we have
\begin{align} \label{alphID} 
\begin{cases} 
\alpha = \WJProjHat_\multii \beta \\ \text{tail}(\alpha) = \mathsf{R}_{\multii,1}\super{s} 
\end{cases} 
\qquad \Longrightarrow \qquad 
\hcancel{\alpha} = \WJProjHat_\multii \hcancel{\beta} ,
\end{align} 
where $\WJProjHat_\multii$ is the map defined via~(\ref{WJProjHatEmb},~\ref{ProjHatDef1-1}).
Now, if the valenced link pattern $\alpha$ has a type-one radical tail, then the link pattern $\beta$ necessarily has a radical tail.  
Hence, by proposition~\ref{BigTailLem} and corollary~\ref{EmbProjCor} in appendix~\ref{AppWJ}, we have  
\begin{align} 
\hcancel{\beta} \in \rad \smash{\LS_{\Summed_\multii}\super{s}} \qquad 
\underset{\eqref{EmbProj22}}{\overset{\eqref{EmbProj2}}{\Longrightarrow}} \qquad \hcancel{\alpha} 
\overset{\eqref{alphID}}{=} \WJProjHat_\multii \hcancel{\beta} \in \rad \smash{\LS_\multii\super{s}}, 
\end{align} 
so any trivalent link state $\hcancel{\alpha}$ derived from a $(\multii,s)$-valenced link state $\alpha$ with type-one radical tail 
belongs to $\rad \smash{\LS_\multii\super{s}}$.
\end{proof}

\begin{lem} \label{JDeltaLem} 
Suppose $\max \multii < \ppmin(q)$.  The following statements are equivalent: 
\begin{enumerate}
\itemcolor{red}
\item \label{JDeltaIt1} We have $s = \Delta_{k_s}$ \textnormal{(}i.e., $R_s = 0$\textnormal{)}.
\item \label{JDeltaIt2} We have $J_\alpha(q) = \np_\multii$, for all valenced link patterns $\alpha \in \smash{\LP_\multii\super{s}}$.
\item \label{JDeltaIt3} We have $J_\alpha(q) = \np_\multii$, for some valenced link pattern $\alpha \in \smash{\LP_\multii\super{s}}$.
\end{enumerate}
Furthermore, if any one of statements~\ref{JDeltaIt1}--\ref{JDeltaIt3} holds, then we have 
\begin{align} \label{BreakDown} 
\mathsf{T}_\multii\super{s} = \mathsf{M}_\multii\super{s}, \qquad \qquad 
\mathsf{R}_{\multii,1}\super{s} = \mathsf{R}_{\multii,2}\super{s} = \emptyset, \qquad \qquad \textnormal{and} \qquad \qquad 
\rad \smash{\LS_\multii\super{s}} = \{0\}. 
\end{align} 
\end{lem}

\begin{proof} 
First, it is evident from definition (\ref{Jindex0},~\ref{Jindex2}) of the index $J_\alpha(q)$ and definitions~(\ref{minmaxh},~\ref{minmaxh2}) of the heights 
$h_{\min,\np_\multii}(\varrho_\alpha)$ and $h_{\max,\np_\multii}(\varrho_\alpha)$ that item~\ref{JDeltaIt1} implies item~\ref{JDeltaIt2}.  
Also, it is obvious that item~\ref{JDeltaIt2} implies item~\ref{JDeltaIt3}.

Next, we prove that item~\ref{JDeltaIt3} implies item~\ref{JDeltaIt1}.  By definition (\ref{Jindex0},~\ref{Jindex2}) of the index $J_\alpha(q)$ and the fact that $\np_\multii$ is 
the maximal value that this quantity may equal, we have
\begin{align} \label{Jconds} 
J_\alpha(q) \overset{\eqref{Jindex2}}{=} J_{\varrho_\alpha}(q) = \np_\multii \qquad 
\overset{\eqref{Jindex0}}{\Longrightarrow} \qquad 
\text{$s \overset{\eqref{minmaxh}}{=} h_{\min,\np_\multii}(\varrho_\alpha) \overset{\eqref{Jindex0}}{\leq} \Delta_{k_s}$ \quad \text{or} \quad $\Delta_{k_s+1} 
\overset{\eqref{Jindex0}}{\leq} h_{\max, \np_\multii}(\varrho_\alpha) \overset{\eqref{minmaxh2}}{=} s.$} 
\end{align} 
By~\eqref{skSet}, it is impossible to have $\Delta_{k_s+1} \leq s$ or $s < \Delta_{k_s}$.  Hence,~\eqref{Jconds} implies item~\ref{JDeltaIt1}.

Finally, with $s = \Delta_{k_s}$, we have by definition that $\text{tail}(\alpha) \in \smash{\mathsf{M}_\multii\super{s}}$ for all $(\multii,s)$-valenced link patterns $\alpha$.  Furthermore, the equality $s = \Delta_{k_s}$ implies that $\pmin(q) \,|\, (s+1)$, so $\rad \smash{\LS_\multii\super{s}} = \{0\}$ by corollary~\ref{EasyRadCor}.
\end{proof}

Now we 
decompose the radical of $\smash{\LS_\multii\super{s}}$ into a direct sum of radicals of three subspaces, specified by 
the three types of tails according to decomposition~(\ref{R12},~\ref{AllTails}).

\begin{lem} \label{radDir2Lem} 
Suppose $\max \multii < \ppmin(q)$.  We have the direct-sum decomposition
\begin{align} \label{radDir2} 
\rad \LS_\multii\super{s}
\; = \; 
\rad \Span \big\{ \hcancel{\,\alpha} \,\big|\, \alpha \in \smash{\LP_\multii\super{s}}, \, \textnormal{tail}(\alpha) \in \mathsf{M}_\multii\super{s} \big\} 
\bigoplus_{i=1}^2 \rad \Span \big\{ \hcancel{\,\alpha} \,\big|\, \alpha \in \smash{\LP_\multii\super{s}}, \, \textnormal{tail}(\alpha) \in \mathsf{R}_{\multii,i}\super{s} \big\}. 
\end{align} 
\end{lem}

\begin{proof} 
To begin, we prove the lemma 
when $\pmin(q) \,|\, (s+1)$, or equivalently, $\Delta_{k_s} = s$.  In this case, all tails in $\smash{\mathsf{T}_\multii\super{s}}$ are moderate by lemma~\ref{JDeltaLem}. 
This fact with item~\ref{IndOrthBasisLemIt3} of proposition~\ref{IndOrthBasisLem} shows that 
\begin{align} \label{TrivLS} 
\smash{\LS_\multii\super{s}} = \Span \big\{ \hcancel{\alpha} \,|\, \alpha \in \smash{\LP_\multii\super{s}} \big\} \overset{\eqref{BreakDown}}{=} \Span 
\big\{ \hcancel{\,\alpha} \,\big|\, \alpha \in \smash{\LP_\multii\super{s}}, \, \textnormal{tail}(\alpha) \in \mathsf{M}_\multii\super{s} \big\}, \qquad \qquad \text{if $\Delta_{k_s} = s$.}
\end{align}
After taking the radical of both sides and invoking~\eqref{BreakDown}, we arrive with asserted formula~\eqref{radDir2}.

Now we prove the lemma when $\pmin(q) \nmid (s+1)$, or equivalently, $\Delta_{k_s} \neq s$.
In this case, lemma~\ref{FirstOrthLem} readily implies that the following subspaces are orthogonal to one another:
\begin{align}
\rad \Span \big\{ \hcancel{\,\alpha} \,\big|\, \alpha \in \smash{\LP_\multii\super{s}}, \, \textnormal{tail}(\alpha) \in \mathsf{R}_{\multii,1}\super{s} \big\} 
&\perp \rad \Span \big\{ \hcancel{\,\alpha} \,\big|\, \alpha \in \smash{\LP_\multii\super{s}}, \, \textnormal{tail}(\alpha) \in \mathsf{M}_\multii\super{s} \big\}, \\
\rad \Span \big\{ \hcancel{\,\alpha} \,\big|\, \alpha \in \smash{\LP_\multii\super{s}}, \, \textnormal{tail}(\alpha) \in \mathsf{R}_{\multii,1}\super{s} \big\} 
&\perp \rad \Span \big\{ \hcancel{\,\alpha} \,\big|\, \alpha \in \smash{\LP_\multii\super{s}}, \, \textnormal{tail}(\alpha) \in \mathsf{R}_{\multii,2}\super{s} \big\}.
\end{align}
Thus, it remains to prove that
\begin{align} \label{OrthProve} 
\rad \Span \big\{ \hcancel{\,\alpha} \,\big|\, \alpha \in \smash{\LP_\multii\super{s}}, \, \textnormal{tail}(\alpha) \in \mathsf{M}_\multii\super{s} \big\} 
\perp \rad \Span \big\{ \hcancel{\,\alpha} \,\big|\, \alpha \in \smash{\LP_\multii\super{s}}, \, \textnormal{tail}(\alpha) \in \mathsf{R}_{\multii,2}\super{s} \big\} . 
\end{align} 
For this, we let $\alpha$ and $\beta$ be arbitrary $(\multii,s)$-valenced link patterns with moderate and type-two radical tails, 
\begin{align} \label{tails2} 
\text{tail}(\alpha) \in \mathsf{M}_\multii\super{s} \qquad \qquad \text{and} \qquad \qquad \text{tail}(\beta) \in \mathsf{R}_{\multii,2}\super{s} , 
\end{align} 
and we write $\varrho_\alpha = (r_1, r_2, \ldots, \smash{r_{\np_\multii}})$ and $\varrho_\beta = (r_1', r_2', \ldots, \smash{r_{\np_\multii}'})$ for their respective walks over the multiindex $\multii$. We set 
\begin{align} \label{shorthandI} 
I := \max \{ j \in \bZnn \,|\, r_j \neq r_j' \} \qquad \qquad \text{and} \qquad \qquad K :=  \max(J_\alpha(q), J_\beta(q)) . 
\end{align} 
Because the tail of $\beta$ is not moderate, we have $J_\beta(q) \geq 0$ by definition.  We also have  
$J_\alpha(q), J_\beta(q), K \leq \np_\multii$ by definition.  
However, if $K = \np_\multii$, then $s = \Delta_{k_s}$ by lemma~\ref{JDeltaLem}, which contradicts our initial assumption.  Thus, we have
\begin{align} \label{Kineq} 
0 \leq K \leq \np_\multii - 1 . 
\end{align} 
In the two respective cases that $J_\alpha < J_\beta = K$ or $J_\beta \leq J_\alpha = K$, 
the network $\hcancel{\,\alpha} \BarAction \hcancel{\,\beta}$ has the following form:
\begin{align}
\hcancel{\,\alpha} \BarAction \hcancel{\,\beta} \quad = \quad \vcenter{\hbox{\includegraphics[scale=0.275]{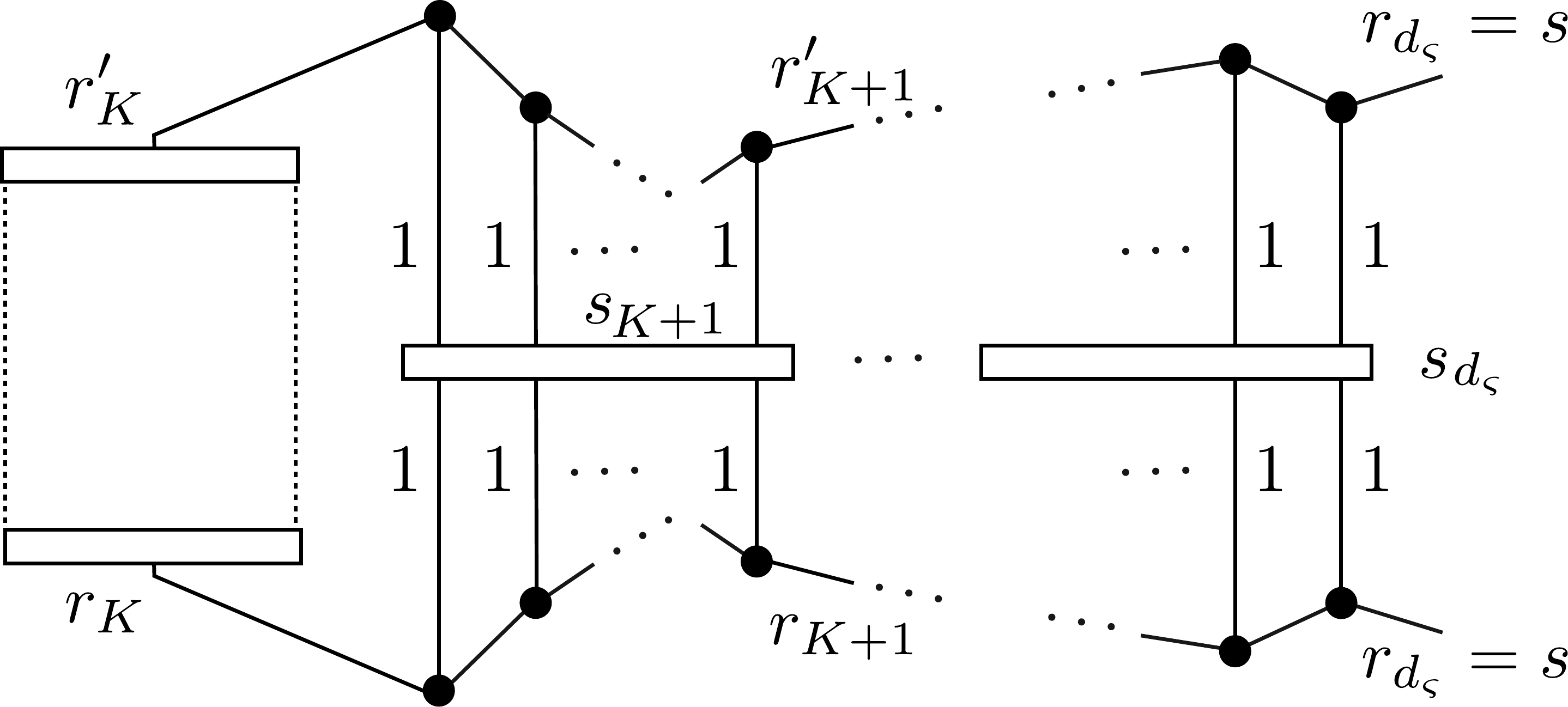} ,}} 
\qquad \qquad  \text{when $J_\alpha < J_\beta$} , \\[2em]
\hcancel{\,\alpha} \BarAction \hcancel{\,\beta} \quad = \quad \vcenter{\hbox{\includegraphics[scale=0.275]{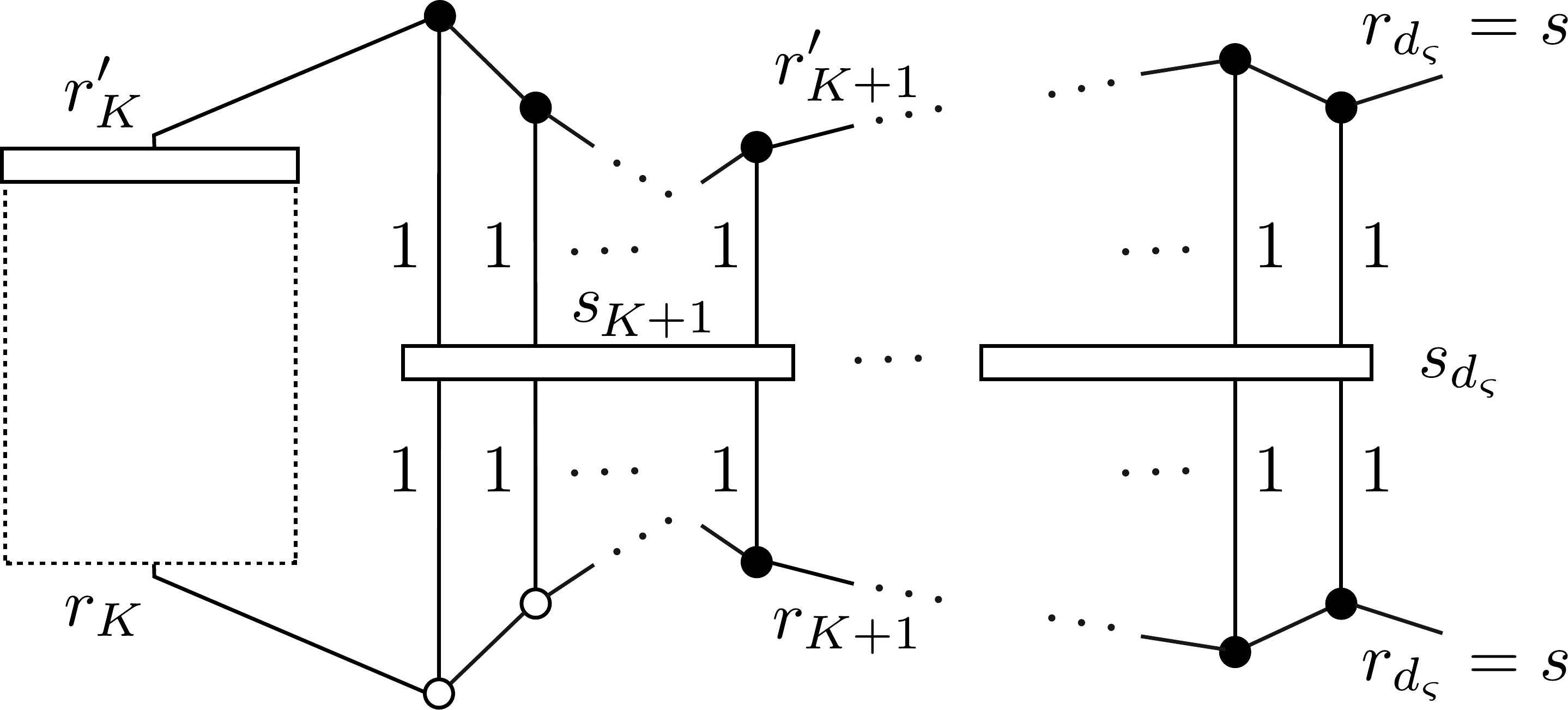} ,}} 
\qquad \qquad  \text{when $J_\alpha \geq J_\beta$} .
\end{align}
Because the tails of $\alpha$ and $\beta$ are different, we also have $K \leq I$.  If $K < I$, then we arrive with $\BiForm{\hcancel{\,\alpha}}{\hcancel{\,\beta}} = 0$ 
by reusing the arguments in the proof of lemma~\ref{OrthTailsLem}.  
Hence, we assume $K = I$ throughout and, to lighten notation, we write
\begin{align} \label{shorthand} 
r := r_K, \qquad r' := r_K', \qquad t := \sIndex_{K+1}, \qquad u:= r_{K+1} = r_{K+1}' . 
\end{align} 
Inequality~\eqref{Kineq} guarantees that these quantities exist.
Now, there are three cases to consider: either $J_\alpha < J_\beta \; (= K)$, or $J_\alpha = J_\beta \; (= K)$, or $J_\beta < J_\alpha \; (= K)$. 
We illustrate the part of the walk $\varrho_\alpha$ that goes from height $r$ to height $u$ and 
the part of the walk $\varrho_\beta$ that goes from height $r'$ to height $u$ respectively as
\begin{align} \label{Illust1}
\vcenter{\hbox{\includegraphics[scale=0.275]{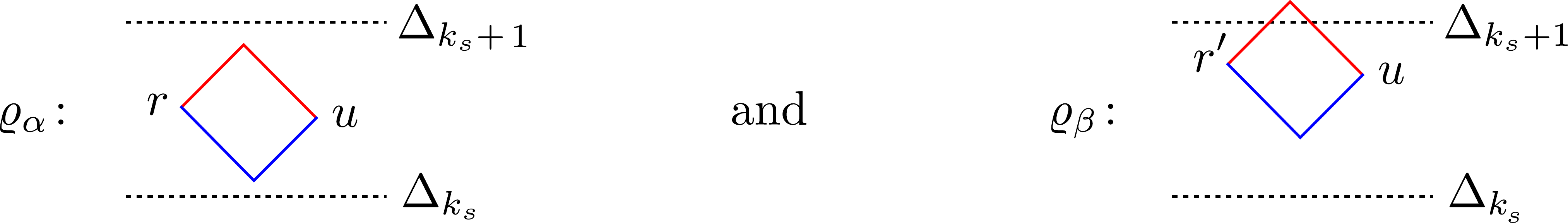} ,}} 
\qquad \qquad \text{when $J_\alpha < J_\beta$},
\end{align} 
\begin{align} \label{Illust2}
\vcenter{\hbox{\includegraphics[scale=0.275]{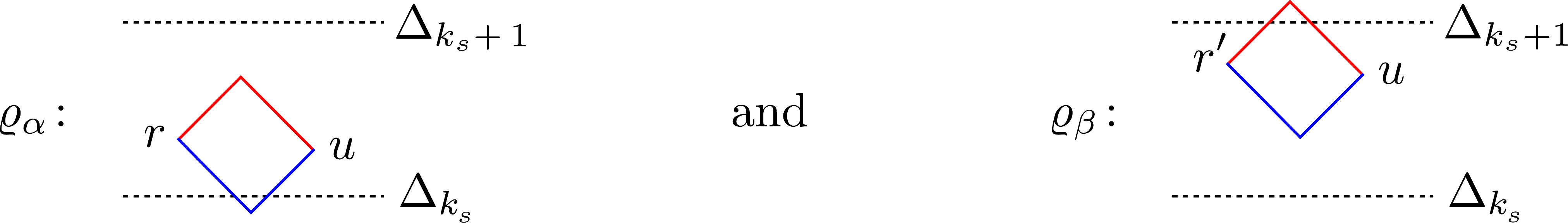} ,}}
\qquad \qquad  \text{when $J_\alpha = J_\beta$},
\end{align} 
\begin{align} \label{Illust3}
\vcenter{\hbox{\includegraphics[scale=0.275]{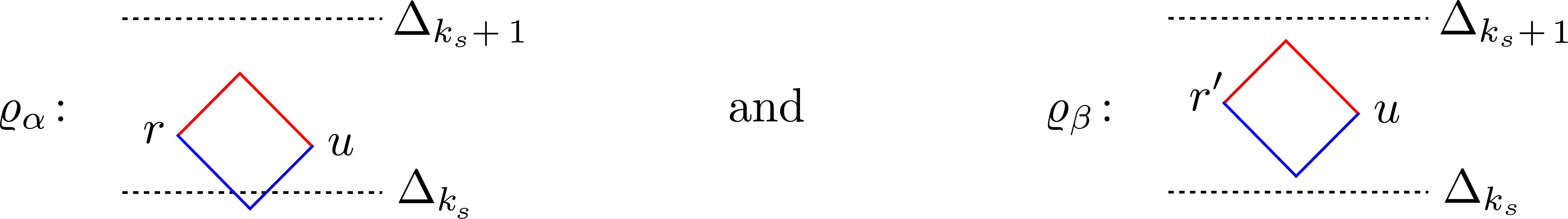} ,}} 
\qquad \qquad  \text{when $J_\alpha > J_\beta$} ,
\end{align} 
where the top sides and bottom sides of the rhombus are respectively parts of the walks 
$\smash{\varrho_\alpha\superscr{\uparrow}}$ or $\smash{\varrho_\beta\superscr{\uparrow}}$ and 
$\smash{\varrho_\alpha\superscr{\downarrow}}$ or $\smash{\varrho_\beta\superscr{\downarrow}}$. 
We show that in all of the cases, the following inequality holds:
\begin{align} \label{r<r'} 
r < r' :
\end{align}  
\begin{enumerate}[leftmargin=*]
\itemcolor{red}
\item \label{J1} $J_\alpha \leq J_\beta = K$: From illustrations~(\ref{Illust1},~\ref{Illust2}) we see that
\begin{align}
\frac{r + u + t}{2} \overset{\eqref{minmaxh2}}{=} h_{\max,K}(\varrho_\alpha)  \overset{\eqref{tails2}}{\leq} \Delta_{k_s+1} - 1
< \Delta_{k_s+1} \overset{\eqref{tails2}}{\leq} h_{\max,K}(\varrho_\beta) \overset{\eqref{minmaxh2}}{=} \frac{r' + u + t}{2} ,
\end{align}
which implies that $r < r'$, so~\eqref{r<r'} holds in this case.

\item \label{J3}  $K = J_\alpha > J_\beta$: Similarly, from illustration~\eqref{Illust3}, we see that
\begin{align}
\frac{r + u - t}{2} \overset{\eqref{minmaxh}}{=} h_{\min,K}(\varrho_\alpha)  \overset{\eqref{tails2}}{\leq} \Delta_{k_s}
< \Delta_{k_s} + 1 \overset{\eqref{tails2}}{\leq} h_{\min,K}(\varrho_\beta) \overset{\eqref{minmaxh}}{=} \frac{r' + u - t}{2} ,
\end{align}
which implies that $r < r'$, so~\eqref{r<r'} holds also in this case.
\end{enumerate}
Now with $r < r'$, the same argument that we used in the proof of lemma~\ref{OrthTailsLem} shows that
in the network $\hcancel{\alpha} \BarAction \hcancel{\beta}$, 
there must exist a turn-back link with both endpoints touching the projector box of size $r'$.  
Hence, we have $\BiForm{\hcancel{\alpha}}{\hcancel{\beta}} = 0$, which proves~\eqref{OrthProve} and implies the claim~\eqref{radDir2}.
\end{proof}

Our next task is to determine all three of the radicals appearing in direct-sum decomposition~\eqref{radDir2}. 
We begin with moderate tails, in which case the radical is trivial.

\begin{lem} \label{M0Lem} 
Suppose $\max \multii < \ppmin(q)$.  We have 
\begin{align} \label{M0} \rad \Span \big\{ \hcancel{\,\alpha} \,\big|\, \alpha \in \smash{\LP_\multii\super{s}}, \, \textnormal{tail}(\alpha) \in \mathsf{M}_\multii\super{s} \big\} = \{0\}. \end{align} 
\end{lem}

\begin{proof}  
To begin, we prove the lemma when $\pmin(q) \,|\, (s+1)$, or equivalently, $\Delta_{k_s} = s$.  In this case, we recall the work in the first paragraph in 
the proof of lemma~\ref{radDir2Lem}.  After taking the radical of both sides of~\eqref{TrivLS} and invoking corollary~\ref{EasyRadCor}, we arrive with~\eqref{M0}.

Now we prove the lemma when $\pmin(q) \nmid (s+1)$, or equivalently, $\Delta_{k_s} \neq s$.  
In this case, lemma~\ref{JDeltaLem} shows that 
for any $(\multii, s)$-valenced link pattern $\alpha$, we have
\begin{align} \label{Jinq} 
J := J_\alpha(q) < \np_\multii . 
\end{align} 
Therefore, with inequality~\eqref{Jinq} satisfied, we may denote
\begin{align} 
\textnormal{tail}(\alpha) = (r_J, r_{J+1}, r_{J+2}, \ldots, r_{\np_\multii}) 
\qquad \qquad \Longrightarrow \qquad \qquad \textnormal{tail}\smash{\und{\,}}(\alpha) := (r_{J+1}, r_{j+2}, \ldots, r_{\np_\multii}),
\end{align} 
and with this notation, we write
\begin{align} 
\Span \big\{ \hcancel{\,\alpha} \,\big|\, \alpha \in \smash{\LP_\multii\super{s}}, \, \textnormal{tail}(\alpha) \in \mathsf{M}_\multii\super{s} \big\} 
= \bigoplus_{\smash{\und{\tau}} \, \in \, \bZnn^\#} \Span \big\{ \hcancel{\,\alpha} \,\big|\, \alpha \in \smash{\LP_\multii\super{s}}, \, \textnormal{tail}(\alpha) \in \mathsf{M}_\multii\super{s}, \, \textnormal{tail}\smash{\und{\,}}(\alpha) = \smash{\und{\tau}} \big\} .
\end{align} 
Reusing the arguments from the proof of lemma~\ref{OrthTailsLem}, we see that the subspaces in this direct sum are orthogonal, so 
\begin{align} \label{radDir} 
\rad \Span \big\{ \hcancel{\,\alpha} \,\big|\, \alpha \in \smash{\LP_\multii\super{s}}, \, \textnormal{tail}(\alpha) \in \mathsf{M}_\multii\super{s} \big\} 
= \bigoplus_{\smash{\und{\tau}} \, \in \, \bZnn^\#} \rad \Span \big\{ \hcancel{\,\alpha} \,\big|\, \alpha \in \smash{\LP_\multii\super{s}}, \, 
\textnormal{tail}(\alpha) \in \mathsf{M}_\multii\super{s}, \, \textnormal{tail}\smash{\und{\,}}(\alpha) = \smash{\und{\tau}} \big\}.
\end{align}

Now we show that each summand in the direct sum~\eqref{radDir} is trivial.  
Selecting an arbitrary summand, let $\beta$ and $\gamma$ be two link states in the span of 
$\smash{ \big\{ \hcancel{\,\alpha} \,\big|\, \alpha \in \smash{\LP_\multii\super{s}}, \, 
\textnormal{tail}(\alpha) \in \mathsf{M}_\multii\super{s}, \, \textnormal{tail}\smash{\und{\,}}(\alpha) = \smash{\und{\tau}} \big\} }$,
and $\smash{\und{\tau}} = (r_{J+1}, r_{J+2}, \ldots, r_{\np_\multii})$, with 
\begin{align} \label{Jindex3} 
J = J_\beta(q) = J_\gamma(q) .
\end{align} 
Then, the bilinear form $\BiForm{\beta}{\gamma}$ equals the evaluation of the network
\begin{align} \label{AlphBetNet2} 
\beta \BarAction \gamma \quad & = \quad \vcenter{\hbox{\includegraphics[scale=0.275]{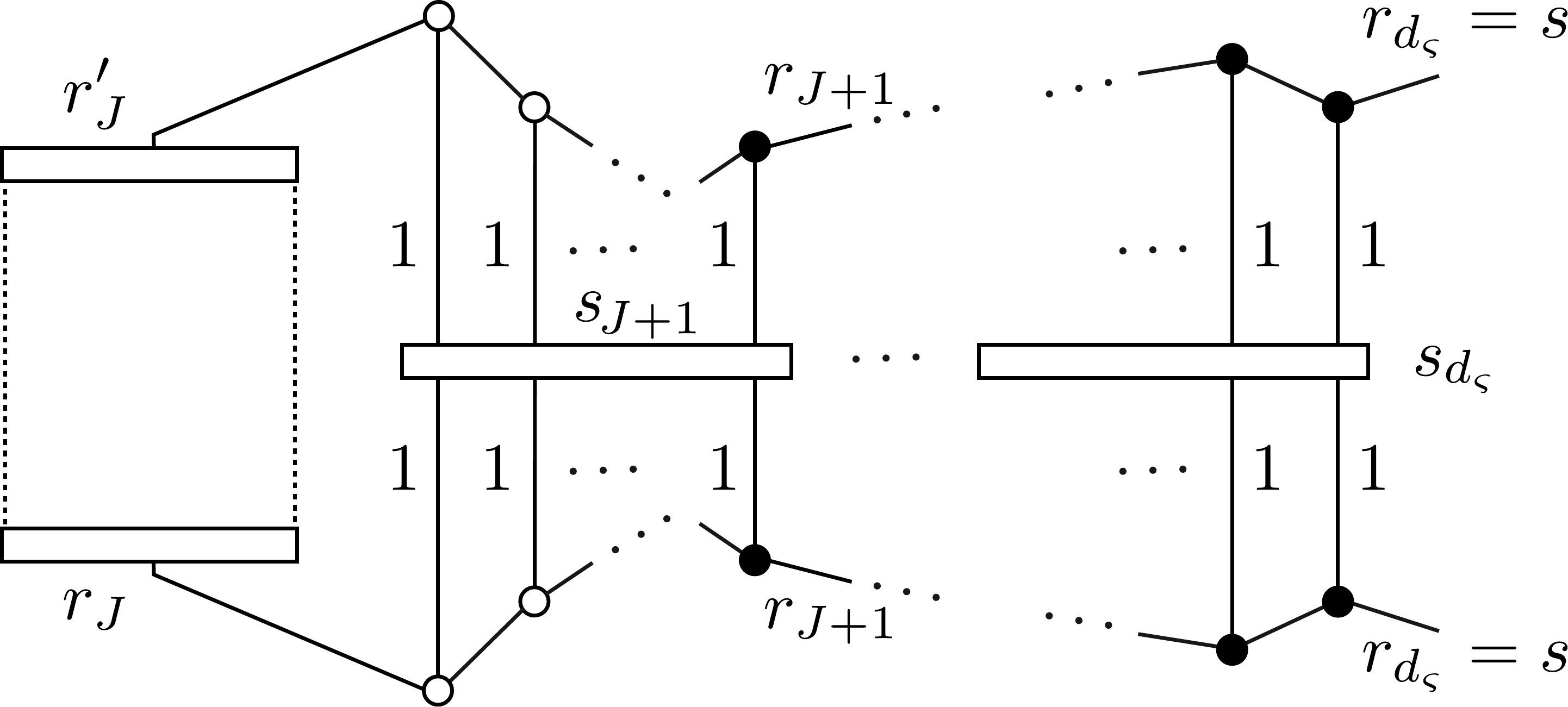}}}  \\[1em]
& \overset{\eqref{EquivPathsClosed}}{=} \quad \vcenter{\hbox{\includegraphics[scale=0.275]{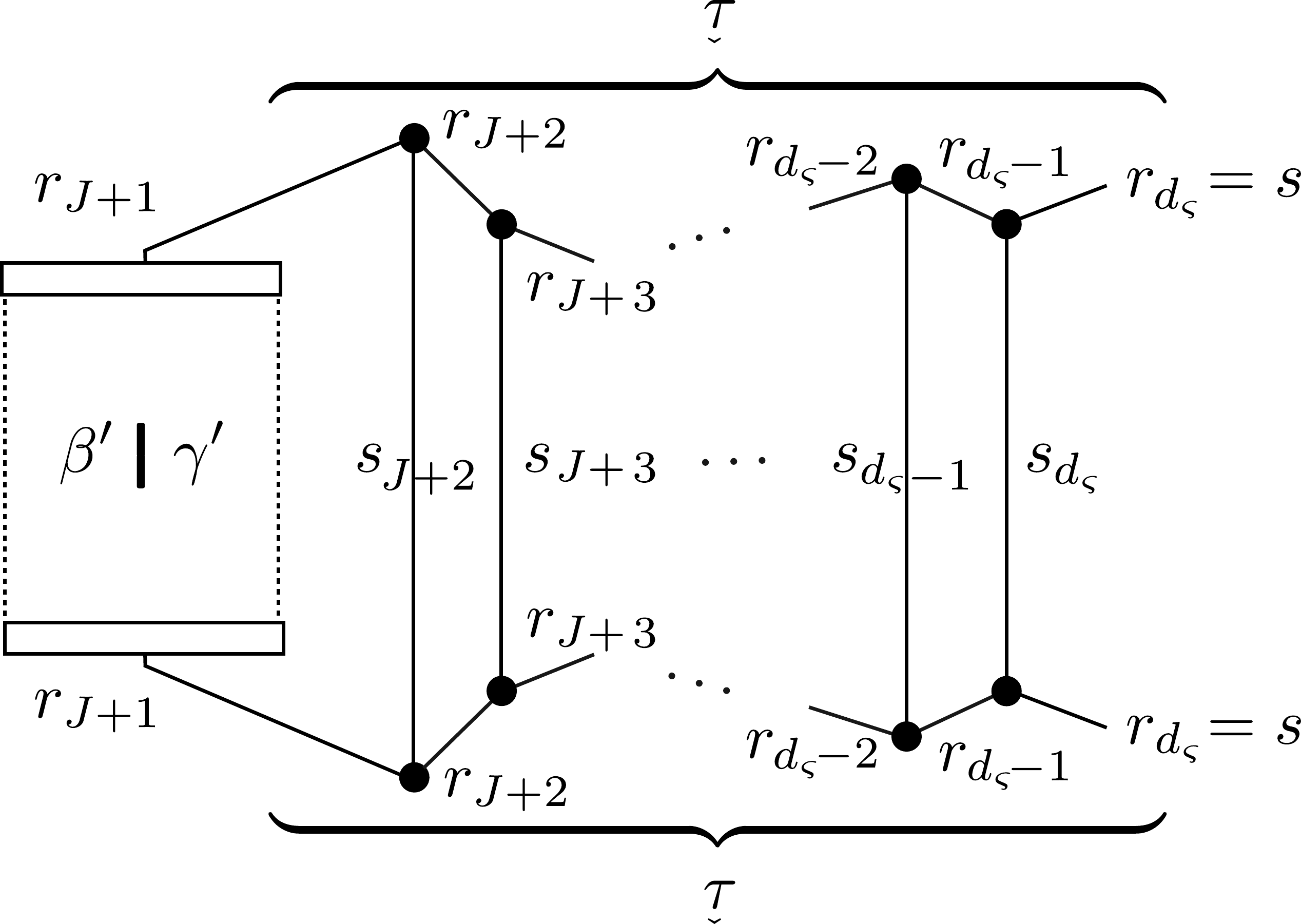} ,}}
\end{align} 
where $\beta'$ and $\gamma'$ are are appropriate valenced link states, described in greater detail below. 
We use lemmas~\ref{ExtractLem} and~\ref{LoopErasureLem} of appendix~\ref{TLRecouplingSect} to evaluate this network, obtaining
\begin{align} 
\label{DotBiIt20} 
\BiForm{\beta}{\gamma} \quad & \overset{\eqref{ExtractID}}{=} \BiForm{\beta'}{\gamma'} \,\, \times  \,\,
\left( \; \vcenter{\hbox{\includegraphics[scale=0.275]{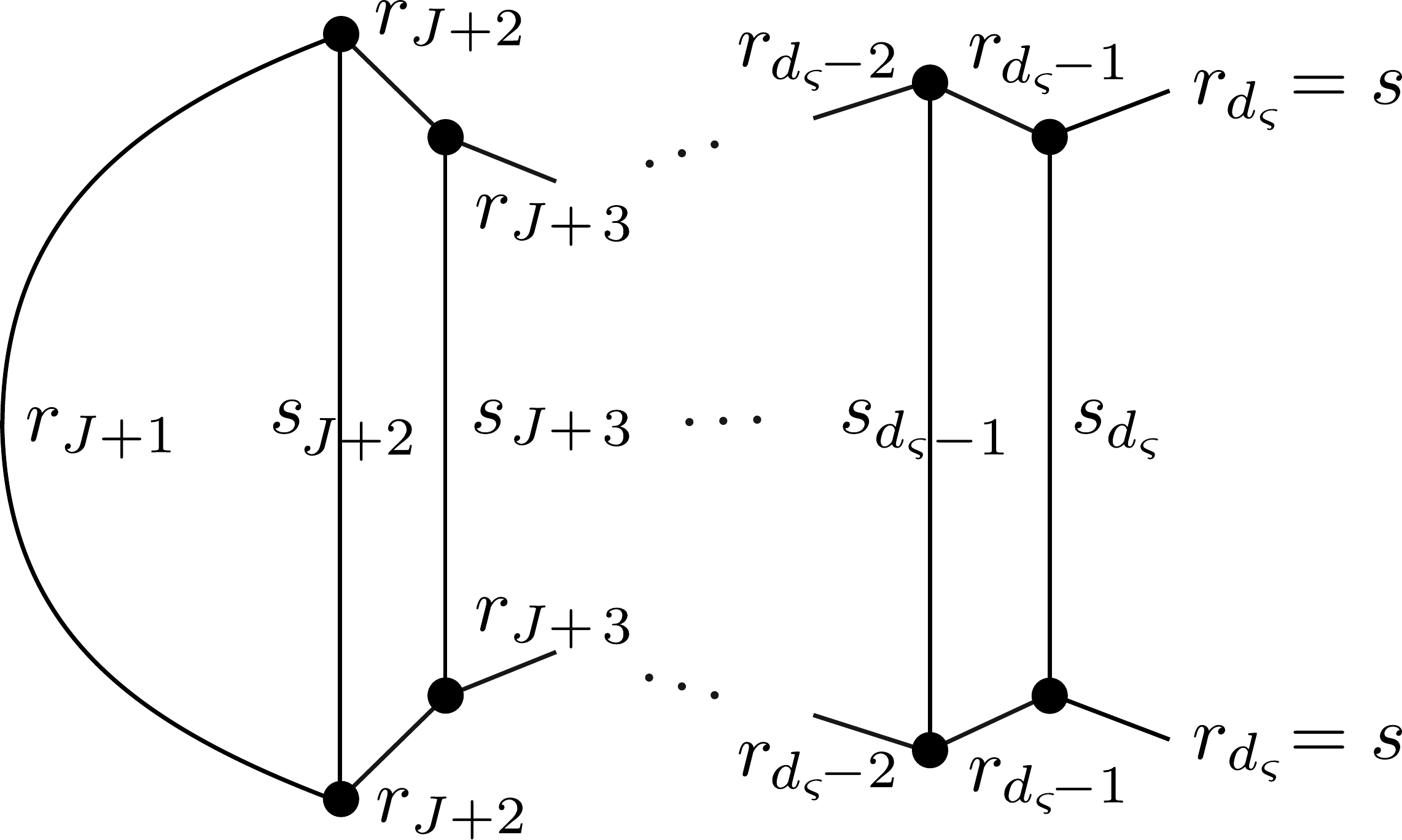}}} \; \right) \\[1em]
\label{DotBiIt2} 
& \overset{\eqref{LoopErasure1}}{=} \BiForm{\beta'}{\gamma'} \prod_{j = J+1}^{\np_\multii-1} \frac{ \ThetaNet( r_j, r_{j+1}, \sIndex_{j+1} )}{ (-1)^{r_{j + 1}} [r_{j + 1}+1]} .
\end{align} 
Lemma~\ref{ThetaInFiniteAndNonzeroLem} shows that the product in~\eqref{DotBiIt2} is finite and nonzero,
\begin{align} \label{NonVan2} 
0 < \bigg| \prod_{j = J+1}^{\np_\multii-1} \frac{ \ThetaNet( r_j, r_{j+1}, \sIndex_{j+1} )}{ (-1)^{r_{j + 1}} [r_{j + 1}+1]} \bigg| < \infty .
\end{align}

Next, we focus on the bilinear form $\BiForm{\beta'}{\gamma'}$ in~\eqref{DotBiIt2}.  To understand this quantity, we first define
\begin{align} \label{ShortH} 
u := \Delta_{k_s}, \qquad\qquad v := r_{J+1} - u \geq 0, \qquad\qquad t := \sIndex_{J+1} - v \geq 0, \qquad\qquad \varpi := (\sIndex_1, \sIndex_2, \ldots, \sIndex_J, t) . 
\end{align} 
By definition (\ref{Jindex0},~\ref{Jindex2}) of $J$, the valenced link states $\beta', \gamma'$ are elements of 
$\smash{\LS_{\varpi; v}\super{u}}$
(defined beneath~\eqref{LSd}), with $v$ defects anchored to the rightmost valenced node of size $\sIndex_{J+1}$:
\begin{align} \label{LSform3ab} 
\vcenter{\hbox{\includegraphics[scale=0.275]{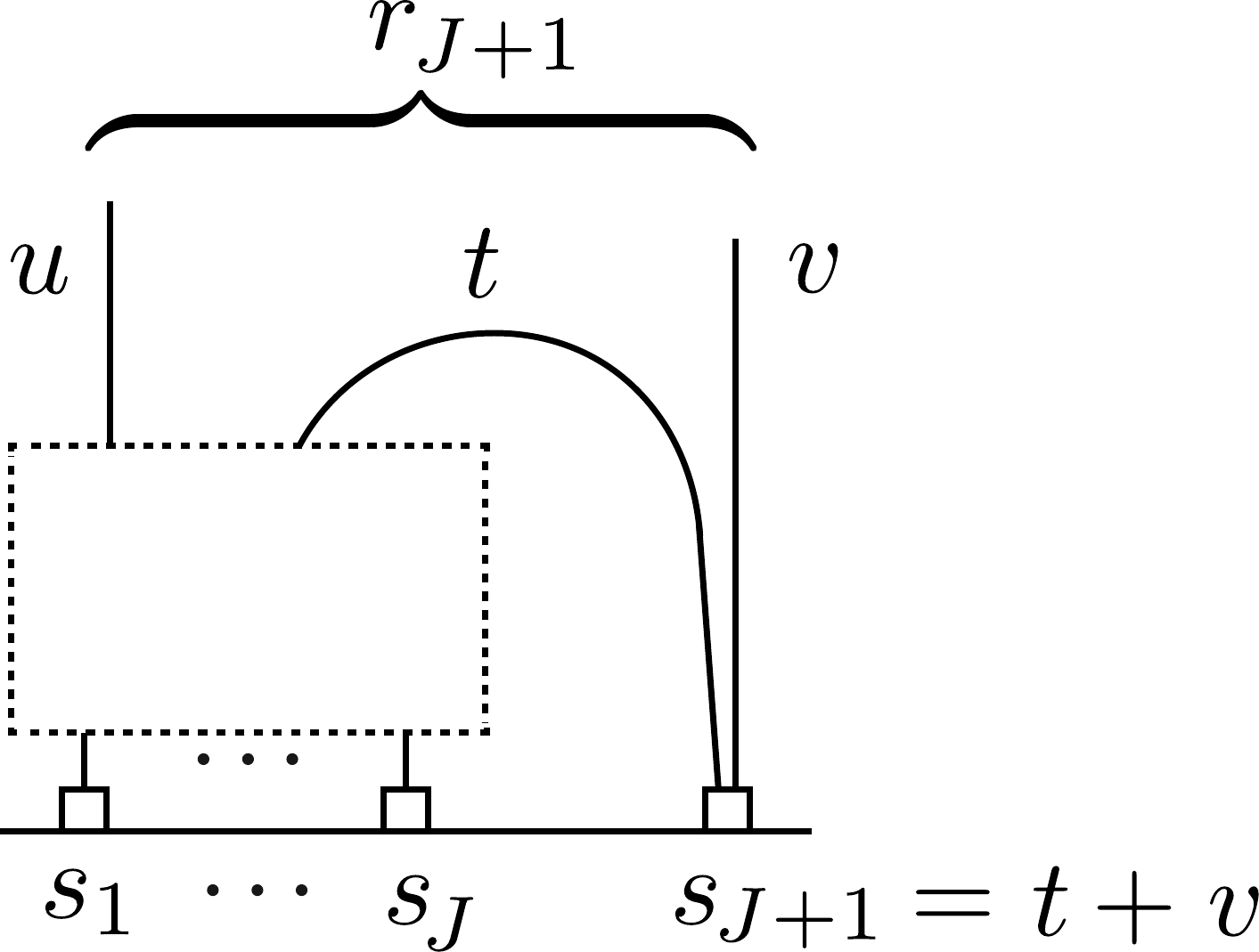} .}} 
\end{align}
We will prove that $\rad \smash{\LS_{\varpi; v}\super{u}} =\{0\}$ by showing that the determinant $\det \smash{\Gram_{\varpi; v}\super{u}}$ 
does not vanish.  
For this, we use lemma~\ref{ExtGramLem}, with replacements $\np_\multii \mapsto J+1$, $\multii \mapsto \varpi$, 
and $s \mapsto u$, to write
\begin{align} \label{ExtGram2} 
\det \Gram_{\varpi; v}\super{u} \overset{\eqref{ExtGram}}{=} 
\det \Gram_\varpi\super{u} \prod_{\varrho} 
\frac{ \big[ \frac{r + t + u}{2} + v + 1 \big]! \big[ \frac{t + u - r}{2} + v \big]! [t]! [u + 1]! }{ \big[\frac{r + t + u}{2} + 1 \big]! \big[ \frac{t + u - r}{2} \big]! [t + v]! [u + v + 1]!},
\end{align} 
where $u, v, t$ are fixed in~\eqref{ShortH} and 
the product is over all walks $\varrho$ over $\varpi$ with defect $u$, and where $r$ denotes the penultimate height of $\varrho$.  
Next, we prove that $\det \smash{\Gram_{\varpi; v}\super{u}} \neq 0$. 
According to lemma~\ref{EasyRadCor}, we have
\begin{align} \label{consq3} 
\text{$\pmin(q) \,|\, (u+1)$ by (\ref{DeltaDefn},~\ref{ShortH})} \qquad \Longrightarrow \qquad \rad \smash{\LS_\multiii\super{u}} = \{0\} 
\qquad \Longrightarrow \qquad \det \smash{\Gram_\multiii\super{u}} \neq 0. 
\end{align}

To show that the product over $\varrho$ in~\eqref{ExtGram2} does not vanish either, we gather some inequalities.  
First, we have
\begin{align} \label{ineqs1} 
0 \overset{\eqref{ShortH}}{\leq} t \overset{\eqref{ShortH}}{\leq} t + v \overset{\eqref{ShortH}}{=} \sIndex_{J+1} \leq \max \multii \leq \ppmin(q) - 1 
\qquad \Longrightarrow \qquad \frac{[t]!}{[t + v]!} \neq 0.  
\end{align} 
Furthermore, by definition~\eqref{Jindex0} of $J = J_\beta(q) = J_\gamma(q)$, we have $r_{J+1} < \Delta_{k_s+1}$.  
Hence, we have
\begin{align} 
k_s \pmin(q) \overset{\eqref{DeltaDefn}}{=} \Delta_{k_s} + 1 & \overset{\eqref{ShortH}}{=} u + 1 \overset{\eqref{ShortH}}{\leq} u + v + 1  \\ 
& \overset{\eqref{ShortH}}{=} r_{J+1} + 1 
\overset{\eqref{LSform3ab}}{\leq} \Delta_{k_s+1} 
\overset{\eqref{DeltaDefn}}{=} (k_s + 1)\pmin(q) - 1 \qquad \Longrightarrow \qquad 
0 < \frac{ [ u + 1 ]! }{ [ u + v + 1 ]! } < \infty .
\end{align}
Next, with $u \in \DefectSet\sub{r,t}$, we have
\begin{align} 
k_s \pmin(q) & \overset{\eqref{DeltaDefn}}{=} \Delta_{k_s} + 1 
\overset{\eqref{ShortH}}{=} u + 1 \overset{\eqref{SpecialDefSet}}{\leq} \frac{r + t + u}{2} + 1 \overset{\eqref{ShortH}}{\leq} 
\frac{r - t + u}{2} + t + v + 1 \\
&  \overset{\eqref{SpecialDefSet}}{\leq}  u + t + v + 1 \underset{\eqref{ineqs1}}{\overset{\eqref{ShortH}}{\leq}} \Delta_{k_s} + \pmin(q) 
\overset{\eqref{DeltaDefn}}{=} (k_s + 1)\pmin(q) - 1 \qquad \Longrightarrow \qquad 
0 <  \frac{ \big[ \frac{r + t + u}{2} + v + 1 \big]! }{ \big[\frac{r + t + u}{2} + 1 \big]! } < \infty ,
\end{align}
and finally, with $t \in \DefectSet\sub{r,u}$, we have
\begin{align} 
0 \overset{\eqref{SpecialDefSet}}{\leq} \frac{t-(r-u)}{2} & \overset{\eqref{ShortH}}{\leq} \frac{t-(r-u)}{2} + v \\
\label{ineqs4} 
& \underset{\eqref{ineqs1}}{\overset{\eqref{SpecialDefSet}}{\leq}} t + v \leq \ppmin(q) - 1 \qquad \Longrightarrow \qquad 
0 < \frac{ \big[ \frac{t + u - r}{2} + v \big]! }{ \big[ \frac{t + u - r}{2} \big]! } < \infty .
\end{align}
Combining (\ref{consq3}--\ref{ineqs4}), we conclude from~\eqref{ExtGram2} that $\det \smash{\Gram_{\varpi; v}\super{u}} \neq 0$.  
Hence, we have $\rad \smash{\LS_{\varpi; v}\super{u}} =\{0\}$.

Now we are ready to finish the proof.  
Because the radical of $\smash{\LS_{\varpi; v}\super{u}}$ is trivial, it follows that each nonzero valenced link state 
$\beta' \in \smash{\LS_{\varpi; v}\super{u}}$ has a companion link state $\gamma_\beta' \in \smash{\LS_{\varpi; v}\super{u}}$ such that 
\begin{align} \label{alphaprime2} \BiForm{\beta'}{\gamma_\beta'} \neq 0. 
\end{align} 
We define $\gamma_\beta$ to be the link state obtained by setting $\gamma' = \gamma_\beta'$ in 
\begin{align}
\gamma \quad = \quad \vcenter{\hbox{\includegraphics[scale=0.275]{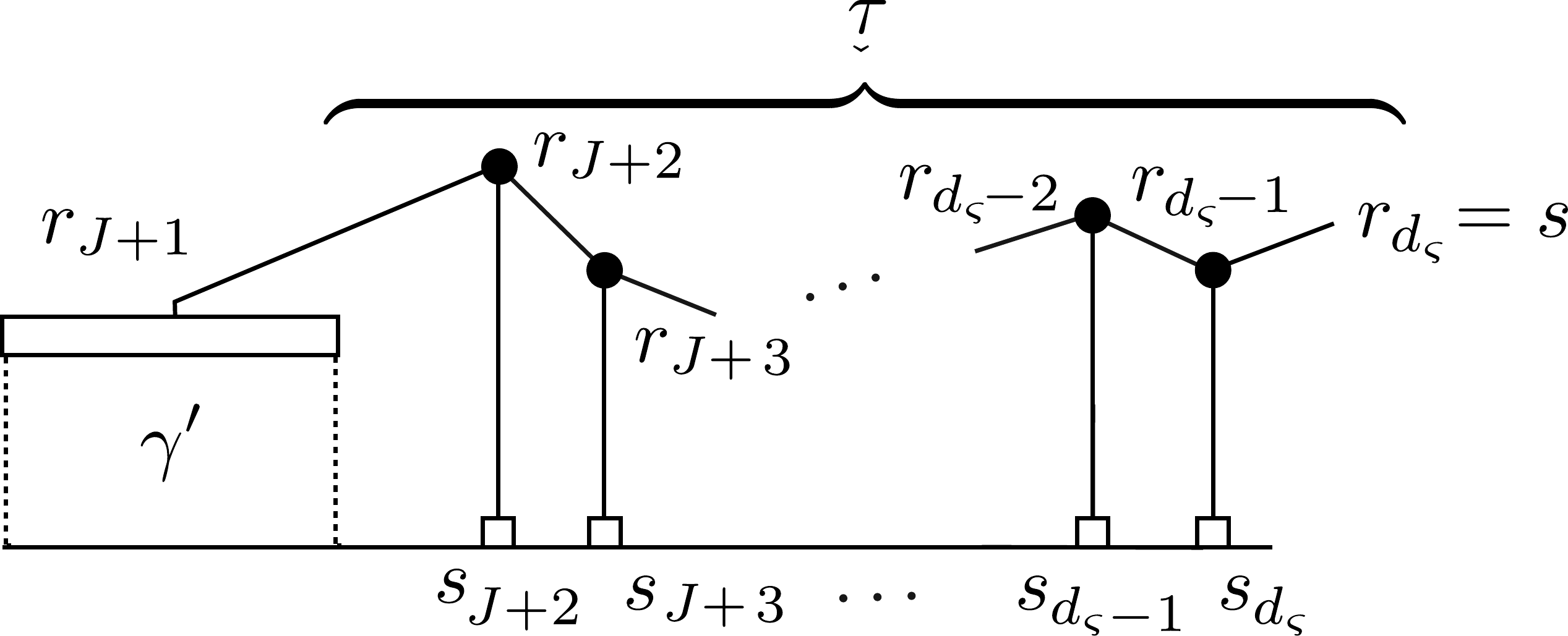} .}} 
\end{align} 
Finally, we combine~(\ref{DotBiIt20},~\ref{DotBiIt2},~\ref{NonVan2},~\ref{alphaprime2}) to 
conclude that for every nonzero link state $\beta$ in the span of 
$\smash{ \big\{ \hcancel{\,\alpha} \,\big|\, \alpha \in \smash{\LP_\multii\super{s}}, \, \textnormal{tail}(\alpha) \in \mathsf{M}_\multii\super{s}, \, \textnormal{tail}\smash{\und{\,}}(\alpha) = \smash{\und{\tau}} \big\} }$, we have
\begin{align} 
\BiForm{\beta}{\gamma_\beta} \overset{\eqref{DotBiIt2}}{=} \BiForm{\beta'}{\gamma_\beta'} 
\prod_{j = J+1}^{\np_\multii-1} \frac{ \ThetaNet( r_j, r_{j+1}, \sIndex_{j+1} )}{ (-1)^{r_{j + 1}} [r_{j + 1}+1]} \underset{\eqref{alphaprime2}}{\overset{\eqref{NonVan2}}{\neq}} 0. 
\end{align} 
Therefore, we have
\begin{align} \label{Comp0} 
\rad \Span \big\{ \hcancel{\,\alpha} \,\big|\, \alpha \in \smash{\LP_\multii\super{s}}, \, \textnormal{tail}(\alpha) \in \mathsf{M}_\multii\super{s}, \, \textnormal{tail}\smash{\und{\,}}(\alpha) = \smash{\und{\tau}} \big\} = \{0\} . 
\end{align} 
Because $\smash{\und{\tau}}$ is arbitrary, (\ref{radDir},~\ref{Comp0}) combine to give~\eqref{M0}.  This finishes the proof.
\end{proof}

We continue the determination of the radical~\eqref{radDir2}, now addressing the case of radical tails.

\begin{lem} \label{MiLem} 
Suppose $\max \multii < \ppmin(q)$.  For $i \in \{1,2\}$, we have
\begin{align} \label{Mi} 
\rad \Span \big\{ \hcancel{\,\alpha} \,\big|\, \alpha \in \smash{\LP_\multii\super{s}}, \, \textnormal{tail}(\alpha) \in \mathsf{R}_{\multii,i}\super{s} \big\} 
= \Span \big\{ \hcancel{\,\alpha} \,\big|\, \alpha \in \smash{\LP_\multii\super{s}}, \, \textnormal{tail}(\alpha) \in \mathsf{R}_{\multii,i}\super{s} \big\}. \end{align} 
\end{lem}

\begin{proof} 
Lemma~\ref{FirstOrthLem} already gives~\eqref{Mi} for $i = 1$,
so it only remains to prove the case $i = 2$. For this, we let $\beta$ or $\gamma$ be any valenced link states in the span of 
$\smash{ \big\{ \hcancel{\,\alpha} \,\big|\, \alpha \in \smash{\LP_\multii\super{s}}, \, \textnormal{tail}(\alpha) \in \mathsf{R}_{\multii,2}\super{s} \big\}}$.  
Without loss of generality, we assume that
\begin{align} 
J_\gamma(q) \leq J:= J_\beta(q). 
\end{align} 
By definition and lemma~\ref{JDeltaLem}, we have  
\begin{align} \label{Kinq} 
0 \leq J < \np_\multii ,
\end{align} 
because the tails of $\beta$ and $\gamma$ are not moderate.  
The bilinear form $\BiForm{\beta}{\gamma}$ equals the evaluation of the network
\begin{align} 
\label{AlphBetNet3} 
\beta \BarAction \gamma \quad & = \quad \vcenter{\hbox{\includegraphics[scale=0.275]{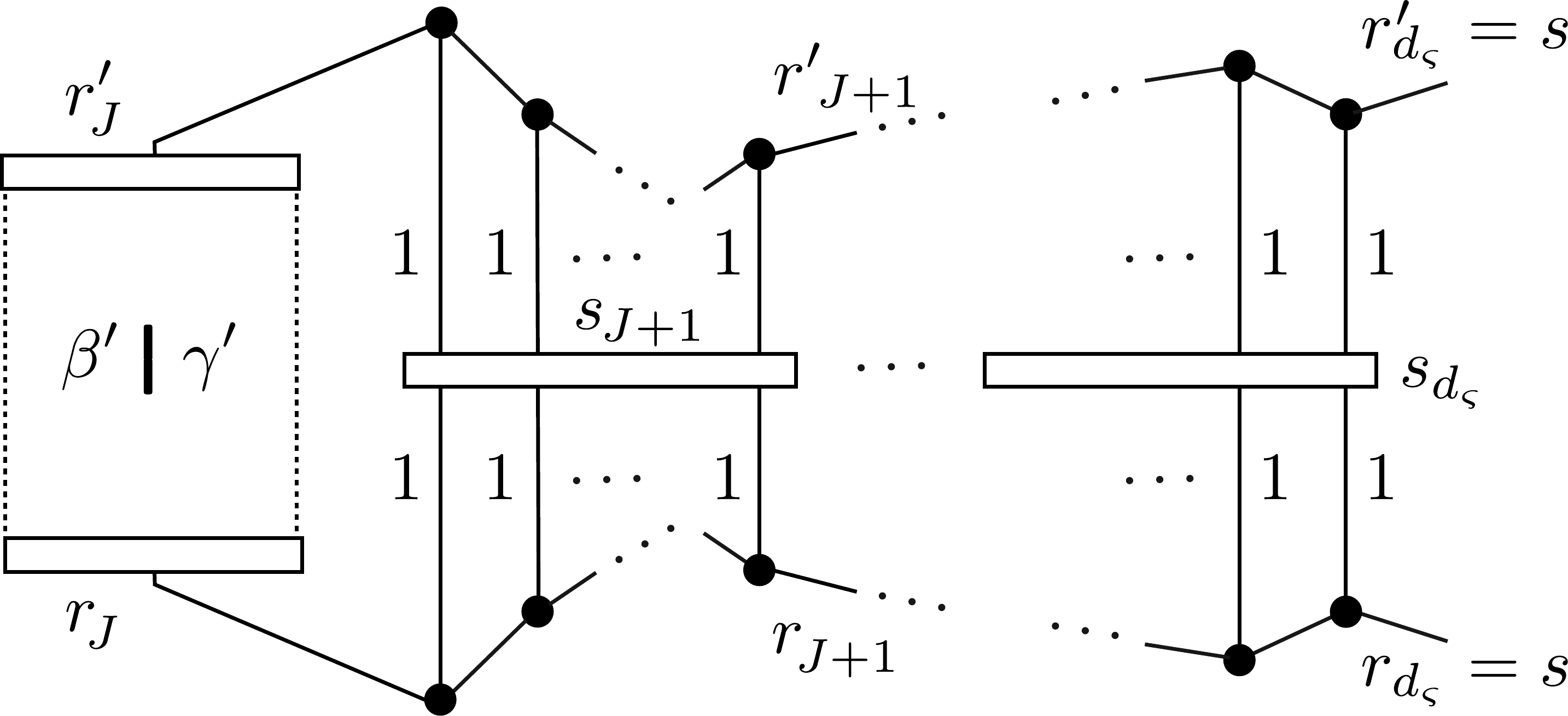} }} \\[2em]
\label{AlphBetNet30} 
& \overset{\eqref{EquivPathsClosed}}{=} \quad \vcenter{\hbox{\includegraphics[scale=0.275]{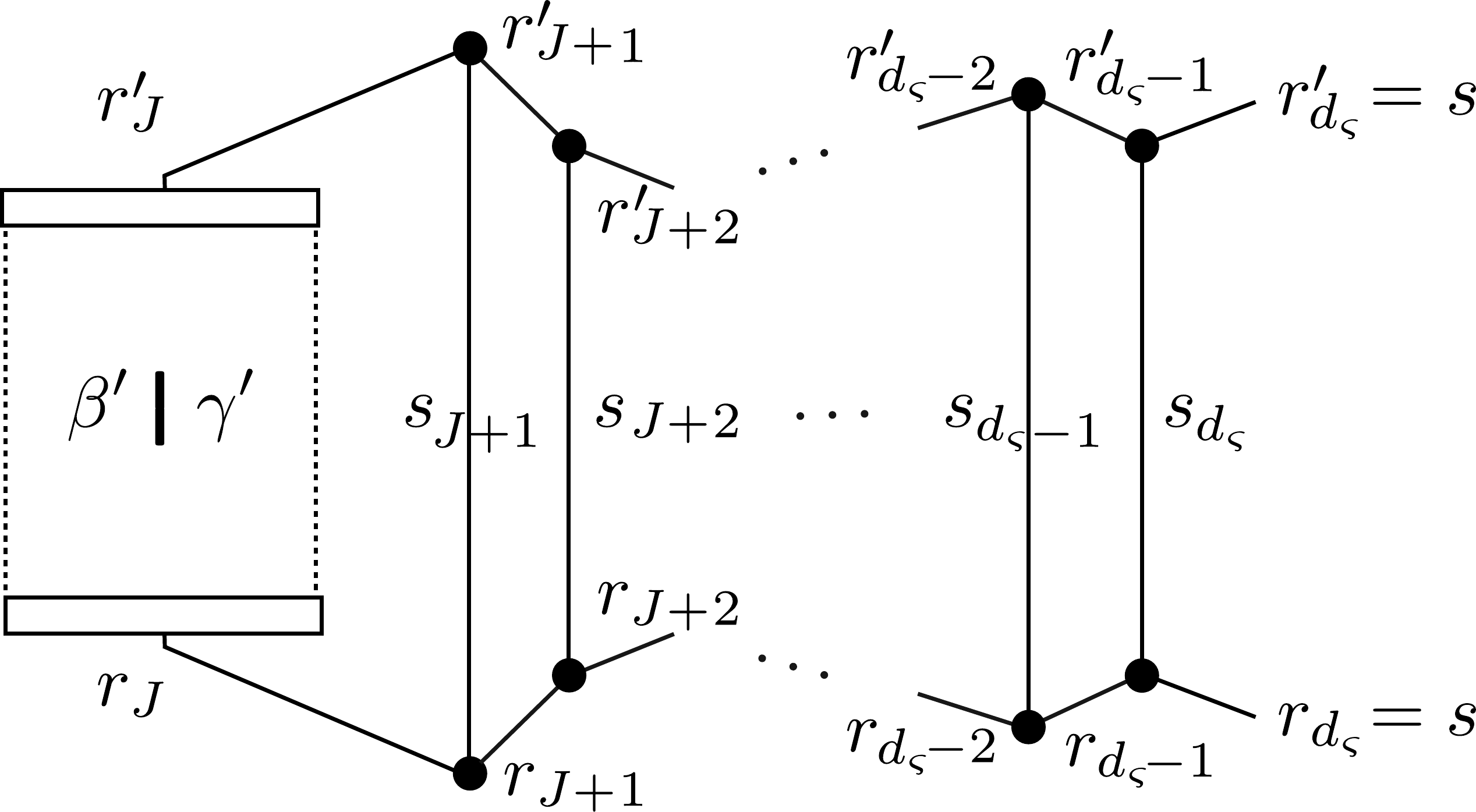} ,}}
\end{align} 
where $\beta'$ and $\gamma'$ are are appropriate valenced link states. 
We use lemmas~\ref{ExtractLem} and~\ref{LoopErasureLem} of appendix~\ref{TLRecouplingSect} to evaluate~\eqref{AlphBetNet30},
\begin{align} 
\BiForm{\beta}{\gamma} \quad 
& \overset{\eqref{ExtractID}}{=} \BiForm{\beta'}{\gamma'} \, \times \quad 
\left( \; \vcenter{\hbox{\includegraphics[scale=0.275]{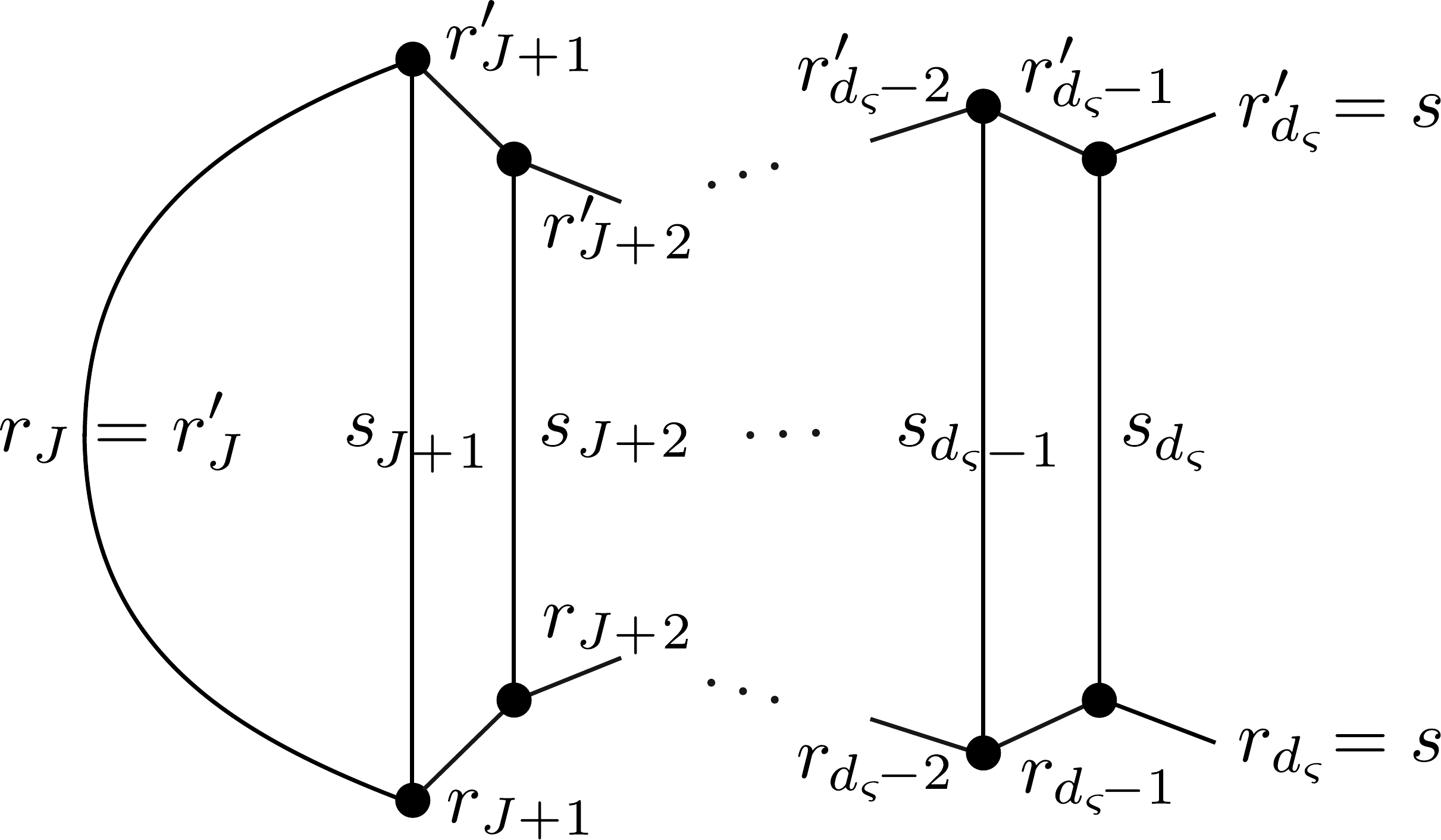}}} \; \right) \\[1em]
\label{DotBiIt4} 
& \overset{\eqref{LoopErasure1}}{=} \BiForm{\beta'}{\gamma'} 
\prod_{j = J}^{\np_\multii-1} \delta_{r_j,r_j'} \frac{ \ThetaNet( r_j, r_{j+1}, \sIndex_{j+1} )}{ (-1)^{r_{j + 1}} [r_{j + 1}+1]}.
\end{align}

First, we consider the factor of the product in~\eqref{DotBiIt4} with $j = J$.  
By definition~(\ref{Jindex0},~\ref{Jindex2}) of $J = J_\beta(q)$, we have
\begin{align} \label{InequalityForType2RadicalTail}
\max(r_J, r_{J+1}) < \Delta_{k_s+1} .
\end{align} 
This gives 
\begin{align} 
\label{condi4-0}  
k_s \pmin(q) \overset{\eqref{DeltaDefn}}{=} \Delta_{k_s} + 1 & \overset{\eqref{Jindex0}}{\leq} h_{\min,J}(\varrho) 
\overset{\eqref{minmaxh}}{=} \frac{r_J + r_{J+1} - \sIndex_{J+1}}{2} \\
\label{condi4-1}  
& \overset{\eqref{EidsGeneral}}{\leq} 
\min(r_J, r_{J+1}) < \max(r_J, r_{J+1}) + 1 \\
\label{condi4-2}  
& \overset{\eqref{Jindex0}}{\underset{\eqref{InequalityForType2RadicalTail}}{<}}
\Delta_{k_s+1} \overset{\eqref{DeltaDefn}}{=} (k_s +1)\pmin(q) - 1 \\
\label{condi4-3}  
& \overset{\eqref{Jindex0}}{<} h_{\max, J}(\varrho) + 1 \overset{\eqref{minmaxh2}}{=} \frac{r_J + r_{J+1} + \sIndex_{J+1}}{2} + 1 . 
\end{align}
Similarly to our reasoning in the proof of lemma~\ref{ThetaInFiniteAndNonzeroLem}, formula~\eqref{ThetaFormula0} of the Theta network with~(\ref{condi4-0}--\ref{condi4-3}) gives
\begin{align} \label{Van}
\frac{ \ThetaNet( r_J, r_{J+1}, \sIndex_{J+1} )}{ (-1)^{r_{J + 1}} [r_{J + 1}+1]} = 0 . 
\end{align}

On the other hand, lemma~\ref{ThetaInFiniteAndNonzeroLem} says that the factors in~\eqref{DotBiIt4} with $j \in \{J+1, J+2, \ldots, \np_\multii - 1\}$ are finite, so
\begin{align} \label{NonVan3} 
0 < \bigg| \prod_{j = J+1}^{\np_\multii-1} \frac{ \ThetaNet( r_j, r_{j+1}, \sIndex_{j+1} )}{ (-1)^{r_{j + 1}} [r_{j + 1}+1]} \bigg| < \infty . 
\end{align} 
After inserting 
(\ref{Van},~\ref{NonVan3}) into~\eqref{DotBiIt4}, we arrive with $\BiForm{\beta}{\gamma} = 0$.   
Because $\beta$ and $\gamma$ were arbitrary valenced link states in the span of the collection 
$\smash{ \big\{ \hcancel{\,\alpha} \,\big|\, \alpha \in \smash{\LP_\multii\super{s}}, \, \textnormal{tail}(\alpha) \in \mathsf{R}_{\multii,2}\super{s} \big\}}$, 
we conclude that~\eqref{Mi} holds for $i = 2$.
\end{proof}

Now we are finally ready to collect our results and finish the complete determination of the radical $\smash{\rad \LS_\multii\super{s}}$.

\begin{theorem} \label{BigTailLem2} 
Suppose $\max \multii < \ppmin(q)$.  The collection
\begin{align} \label{BigTail2} 
\big\{ \hcancel{\,\alpha} \,\big|\, \alpha \in \smash{\LP_\multii\super{s}}, \, \textnormal{tail}(\alpha) \in \mathsf{R}_\multii\super{s} \big\} 
\end{align} 
is a basis for $\smash{\rad \LS_\multii\super{s}}$.
\end{theorem} 
\begin{proof}
Combining lemma~\ref{radDir2Lem} with lemmas~\ref{M0Lem} and~\ref{MiLem}, we obtain
\begin{align} 
\rad \smash{\LS_\multii\super{s}} 
& \underset{\eqref{M0}}{\overset{\eqref{radDir2}}{=}} 
\bigoplus_{i =1}^2 \rad \Span \{ \hcancel{\,\alpha} \,|\,  \alpha \in \smash{\LP_\multii\super{s}}, \, \textnormal{tail}(\alpha) \in \mathsf{R}_{\multii,i}\super{s} \} \\
& \underset{\eqref{R12}}{\overset{\eqref{Mi}}{=}} \Span \big\{ \hcancel{\,\alpha} \,\big|\, \alpha \in \smash{\LP_\multii\super{s}}, \, \textnormal{tail}(\alpha) \in \mathsf{R}_\multii\super{s} \big\}. 
\end{align}
Also, item~\ref{IndOrthBasisLemIt3} of proposition~\ref{IndOrthBasisLem} implies that the set~\eqref{BigTail2} is linearly independent.  Thus, it is basis for $\smash{\rad \LS_\multii\super{s}}$.
\end{proof}

To end this section, we determine the dimension of $\rad \smash{\LS_\multii\super{s}}$.  
We recall from lemma~\ref{LSDimLem2} that the dimension of the standard module $\smash{\LS_\multii\super{s}}$ is $\smash{\Dim_\multii\super{s}}$, 
the unique solution to recursion problem~\eqref{Recursion2}.
Then, analogously to~\eqref{RadRecurs}, with $\Delta_k$ defined in~\eqref{DeltaDefn} and 
denoting $\multii = (\sIndex_1, \sIndex_2, \ldots, \sIndex_{\np_\multii})$, $\hat{\multii} := (\sIndex_1, \sIndex_2, \ldots, \sIndex_{\np_\multii-1})$, and $t := \sIndex_{\np_\multii}$,
we define the numbers $\smash{\hcancel{\Dim}_\multii\super{s}}$ to be the unique solution to the recursion 
\begin{align} 
\label{RadRecurs2}
\hcancel{\Dim}_\multii\super{s} = 
\sum_{r \, \in \, \DefectSet_{\hat{\multii}} \, \cap \, \DefectSet\sub{s,t}} 
\Big(\one{ \Big\{\Delta_{k_s} < \, \frac{r+s-t}{2} \Big\} } & \one{ \Big\{ \frac{r+s+t}{2} \, < \, \Delta_{k_s+1} \Big\} } \hcancel{\Dim}_{\hat{\multii}}\super{r} \\
\nonumber
 + & \one{ \Big\{ \Delta_{k_s+1} \, \leq \, \frac{r+s+t}{2} \Big\}} \Dim_{\hat{\multii}}\super{r} \Big), \qquad \qquad \text{and} \quad \hcancel{\Dim}\sub{r}\super{r} = 0.
\end{align}

The following lemma is similar to~\eqref{CountLP} in item~\ref{wmlIt4} of lemma~\ref{WalkMultiiLem}. 
Before stating it, it is useful to make the following observation: for any walk $\varrho$ over $\multii$ with defect $s$, 
and for each $j \in \{0, 1, \ldots, \np_\multii - 1\}$, we have
\begin{align}\label{NotSimul} 
0 \overset{\eqref{minmaxh}}{\underset{\eqref{minmaxh2}}{\leq}} h_{\max, j}(\varrho) - h_{\min, j}(\varrho) 
\overset{\eqref{minmaxh}}{\underset{\eqref{minmaxh2}}{=}} \sIndex_{j+1} \leq \max \multii < \pmin(q) = \Delta_{k_s+1} - \Delta_{k_s}. 
\end{align}
Thus, the walks $\varrho\superscr{\,\uparrow}$ and $\varrho\superscr{\,\downarrow}$ cannot simultaneously hit the 
heights $\Delta_{k_s+1}$ and $\Delta_{k_s}$ respectively at the same step of $\varrho$.

\begin{lem} \label{RadWalkLem2} We have 
\begin{align} \label{RadWalk2} 
\hcancel{\Dim}_\multii\super{s} = \#\left\{\parbox{8.9cm}{\textnormal{walks $\varrho$ over $\multii$ with defect $s$ and such that, when followed} \\ 
\textnormal{backward, $\varrho\superscr{\,\uparrow}$ hits height $\Delta_{k_s+1}$ before $\varrho\superscr{\,\downarrow}$ hits height $\Delta_{k_s}$}}\right\}. 
\end{align} 
\end{lem}

\begin{proof} 
This lemma can be proven similarly as lemma~\ref{RadWalkLem}, 
except that $\multii$ may be any multiindex in $\{\OneVec{0}\} \cup \smash{\bZnn^\#}$.  
\end{proof}

Extending the comment following lemma~\ref{RadWalkLem}, if $s + 1 < \pmin(q)$, then $\smash{\hcancel{\Dim}_\multii\super{s}}$ equals the number of walks $\varrho$ 
over $\multii$ with defect $s$ and such that $\varrho\superscr{\,\uparrow}$ hits height $\pmin(q) - 1$ (or equivalently, with maximum apex at or above $\pmin(q) - 1$).

\begin{cor} \label{DnsLemAndRadDimCor2} 
Suppose $\max \multii < \ppmin(q)$.  We have
\begin{align} \label{DnsDotDefn2} 
 \dim \rad \smash{\LS_\multii\super{s}} = 
\#\mathsf{R}_\multii\super{s} = \hcancel{\Dim}_\multii\super{s}.
\end{align} 
\end{cor}
\begin{proof} 
The first equality in~\eqref{DnsDotDefn2} immediately follows from theorem~\ref{BigTailLem2}, and the second 
equality from~\eqref{RadWalk2} of proposition~\ref{RadWalkLem2} and the definition of a radical tail.
\end{proof}

Lemma~\ref{RadWalkLem2} and corollary~\ref{DnsLemAndRadDimCor2} reduce to lemma~\ref{RadWalkLem} and 
corollary~\ref{DnsLemAndRadDimCor} respectively if $\multii = \OneVec{n}$.

\subsection{Nondegenerate cases} \label{rofSect31}

We recall that the bilinear form~\eqref{LSBiFormExt} on the standard module $\smash{\LS_\multii\super{s}}$ is said to be ``nondegenerate'' if its radical is trivial, i.e., $\rad \smash{\LS_\multii\super{s}} = \{0\}$. 
In section~\ref{LinkStateModSect}, proposition~\ref{GenLem2} implies that the standard module $\smash{\LS_\multii\super{s}}$ is simple if and only if $\rad \smash{\LS_\multii\super{s}} = \{0\}$. 
Thus, for the purpose of classifying all simple $\TL_\multii(\nu)$-modules, it is worthwhile to determine all $q \in \bC^\times$ for which the bilinear form on 
a given standard module is nondegenerate. To this end, we begin with the following lemma. To state it, it is first helpful to recall the containment 
$\DefectSet_\multii \subset \DefectSet_{\Summed_\multii}$ from~\eqref{subset}.

\begin{lem} \label{IfonlyIfLem} 
Suppose $\max \multii < \ppmin(q)$.  For each $s \in \DefectSet_\multii$, we have
\begin{align} \label{IfonlyIf} 
\hcancel{\Dim}_\multii\super{s} = 0 \qquad \Longleftrightarrow \qquad \hcancel{\Dim}_{\Summed_\multii}\super{s} = 0 .
\end{align} 
\end{lem}

\begin{proof} 
We consider the tallest walk $\varrho_{\max}$ over $\multii$ with defect $s$, and the tallest walk $\varrho\superscr{\,\uparrow}_{\max}$ 
over $\OneVec{n}_\multii$ with defect $s$: e.g.,
\begin{align}
\vcenter{\hbox{\includegraphics[scale=0.275]{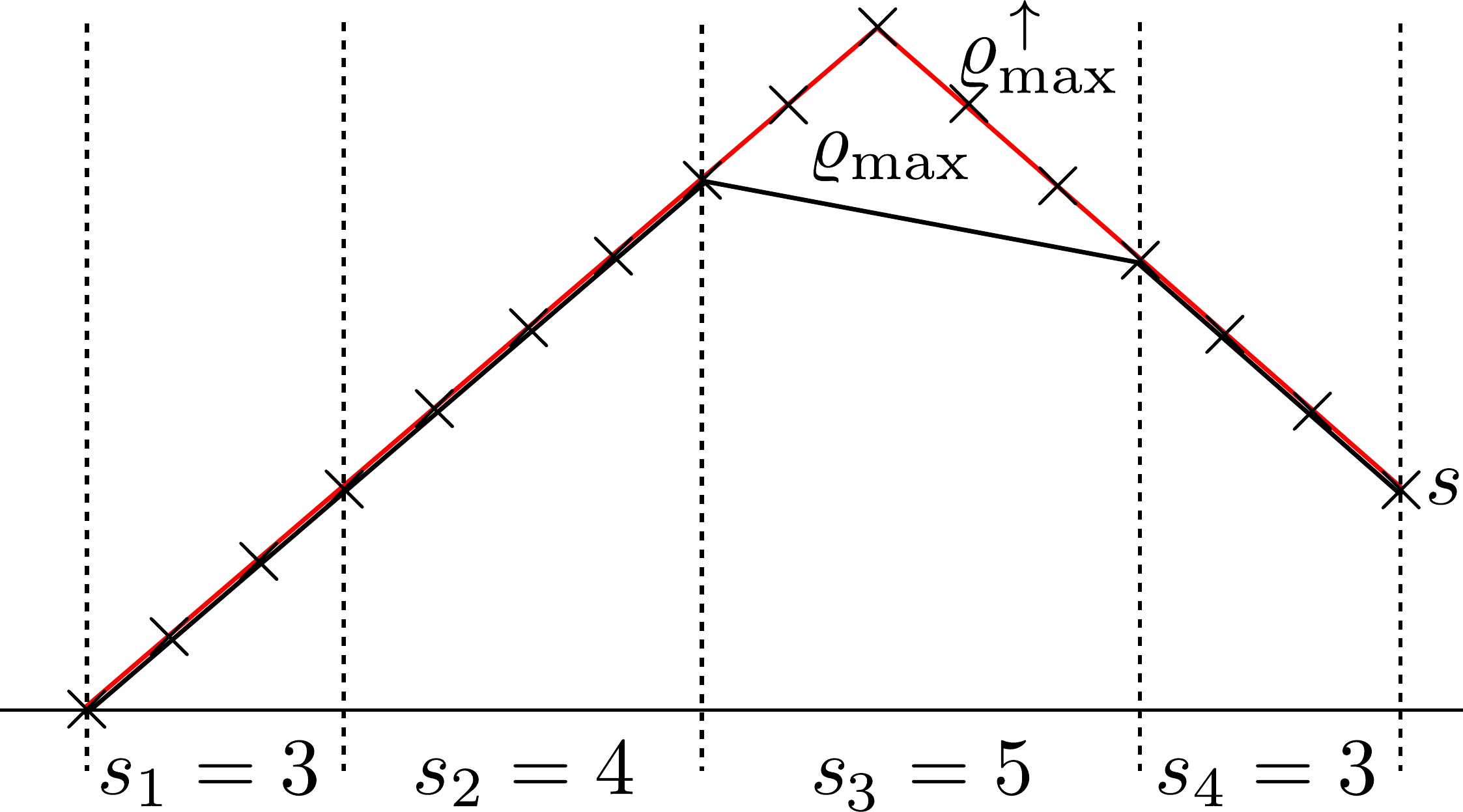}}} \qquad \qquad \qquad
\vcenter{\hbox{\includegraphics[scale=0.275]{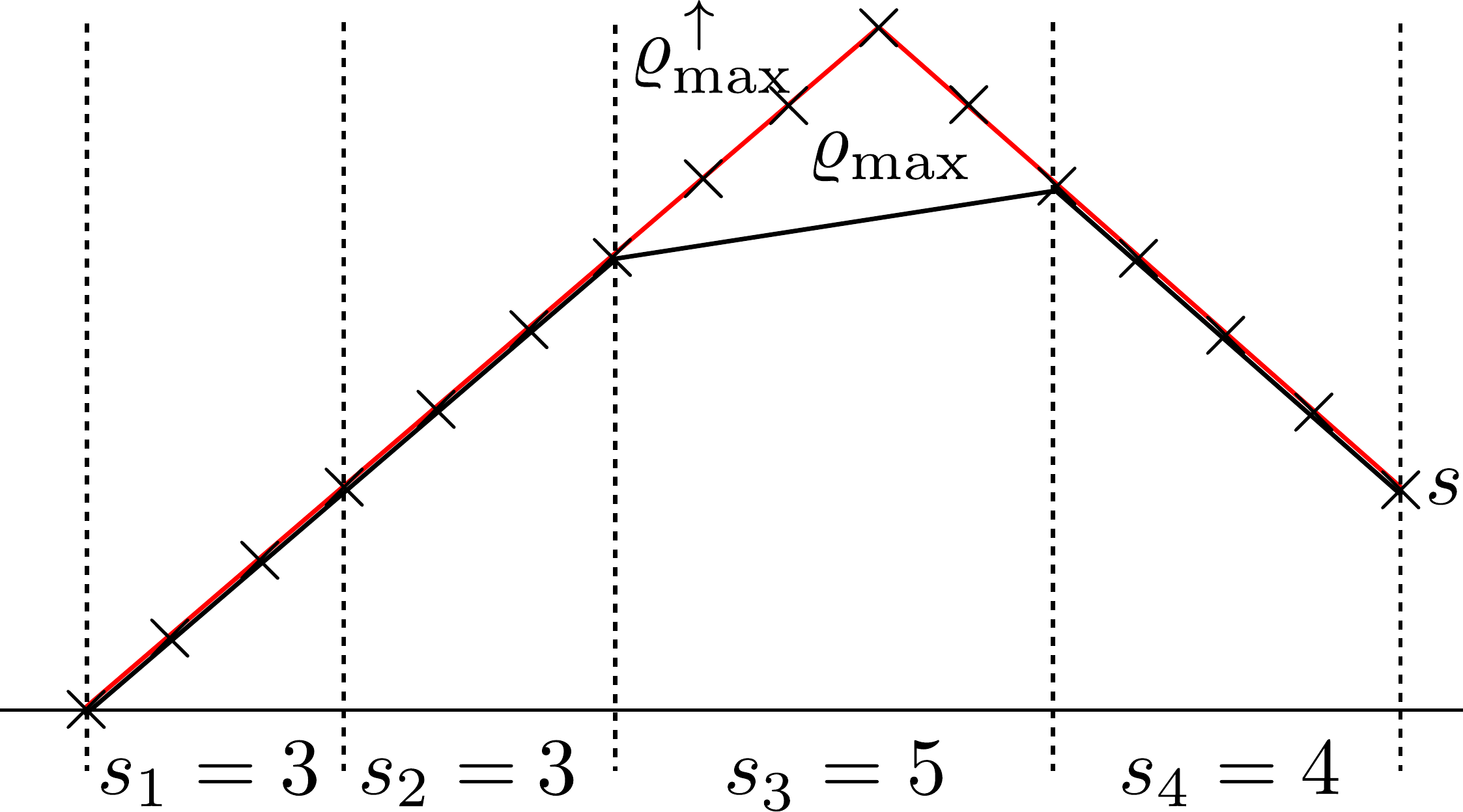} .}}
\end{align}
Because $\varrho_{\max}$ is the tallest walk over $\multii$ with defect $s$, if there exists a walk over $\multii$ with defect $s$ and with radical tail, then the tail of 
$\varrho_{\max}$ is also radical.  
A similar fact holds for $\varrho_{\max}\superscr{\,\uparrow}$ and walks over $\OneVec{n}_\multii$ with defect $s$.  
Hence, we have 
\begin{align} \label{Iff}
& \text{the tail of $\varrho_{\max}$ is radical} \qquad \overset{\eqref{RadWalk2}}{\Longleftrightarrow}  
\qquad \hcancel{\Dim}_\multii\super{s} \neq 0 \\
& \text{the tail of $\varrho_{\max}\superscr{\,\uparrow}$ is radical} \qquad 
\overset{\eqref{RadWalk}}{\Longleftrightarrow} 
\qquad \hcancel{\Dim}_{\Summed_\multii}\super{s} \neq 0 .
\end{align}
Now from the definition of a radical tail, we see that the tail of $\varrho_{\max}$ is radical if and only if the tail of $\varrho_{\max}\superscr{\,\uparrow}$ is radical.  
Combined with~\eqref{Iff}, this last fact implies that $\smash{\hcancel{\Dim}_\multii\super{s}} \neq 0$ if and only if $\smash{\hcancel{\Dim}_{\Summed_\multii}\super{s}} \neq 0$.
\end{proof}

\begin{cor} \label{RadDimCor3} 
Suppose $\max \multii  < \ppmin(q)$.  For each $s \in \DefectSet_\multii$, we have
\begin{align} \label{IfonlyIf2} 
\rad \smash{\LS_\multii\super{s}} = \{0\} \qquad \Longleftrightarrow \qquad \rad \LS_{\Summed_\multii}\super{s} = \{0\} . 
\end{align} 
\end{cor}

\begin{proof} 
This immediately follows from corollaries~\ref{DnsLemAndRadDimCor} and~\ref{DnsLemAndRadDimCor2} 
with lemma~\ref{IfonlyIfLem}.
\end{proof}

\begin{figure}
\includegraphics[scale=0.275]{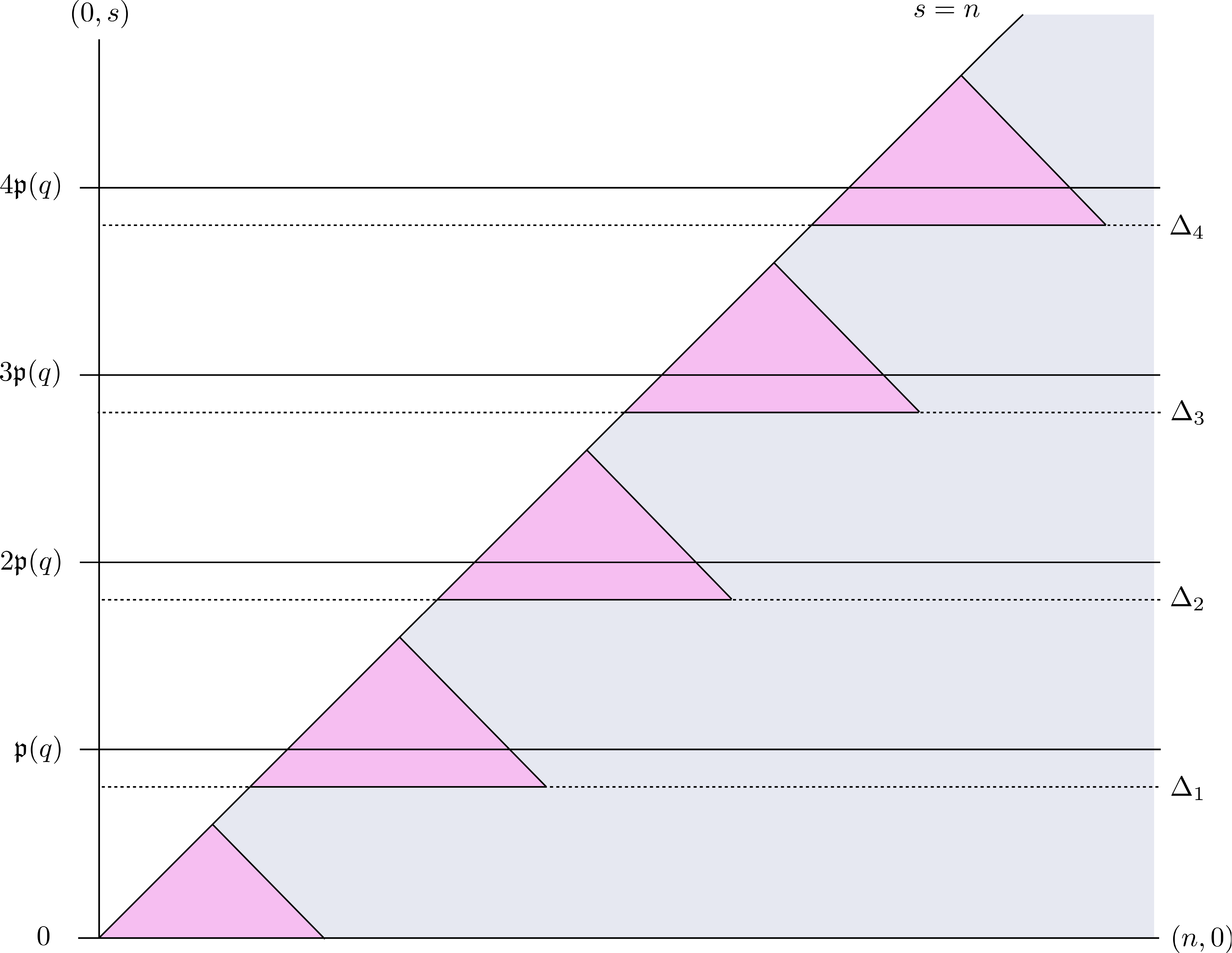} 
\caption{\label{GridFigure}
Illustration of the pairs $(n,s) \in \bZnn \times \DefectSet_n$.  
The set of $q \in \bC^\times$ such that $(n,s)$ is on either a pink triangle or 
a horizontal dashed line is denoted by $\smash{\Dom_n^{(s)}}$.}
\end{figure}

Now we determine all $q \in \bC^\times$ such that the bilinear form  on $\smash{\LS_\multii\super{s}}$ is nondegenerate. 
In light of corollary~\ref{RadDimCor3}, we only need to consider the case $\multii = \OneVec{n}$, for $n \in \bZnn$. 
As illustrated in figure~\ref{GridFigure}, the pairs $(n,s) \in \bZnn \times \DefectSet_n$ live on the square lattice, 
within a semi-infinite triangle bound between the lines $s = 0$ and $s = n$.
Certain points of the lattice shown in figure~\ref{GridFigure} are of special interest:
\begin{enumerate}
\itemcolor{red}
\item points $(n,s) = (n,\Delta_k)$ on a dashed line, each at height $\Delta_k = \Delta_k(q)$ for some $k \in \bZpos$, 

\item points $(n,s)$ on the pink
triangle with corners at $(0,0)$, $(2\Delta_1-2,0)$, and $(\Delta_1-1,\Delta_1-1)$, and

\item points $(n,s)$ on pink
triangles with corners at $(\Delta_k,\Delta_k)$, $(2\Delta_{k+1}-2,\Delta_k)$, and $(\Delta_{k+1}-1,\Delta_{k+1}-1)$, for $k \in \bZpos$.
\end{enumerate}
We define
\begin{align} 
\Dom_n\super{s} 
&:= \Big\{ q \in \bC^\times \,\Big|\, \parbox{4.7cm}{$(n,s)$ lies on a 
pink triangle \\ or on a dashed line in figure~\ref{GridFigure}} \Big\} \\
& \hphantom{:}= \Big\{ q \in \bC^\times \,\big|\, \text{either $R_s = 0$, or $\frac{n - s}{2} \in \{0, 1, \ldots, \pmin(q) - 1 - R_s \}$}  \Big\}.
\label{Domns2} 
\end{align}
We note that the complement of this set within $\bC$ has Lebesgue measure zero. We also define
\begin{align} \label{Domn2} 
\Dom_n := \bigcap_{s \, \in \, \DefectSet_n} \Dom_n\super{s} 
= \big\{ q \in \bC^\times \,\big|\, \textnormal{either $n < \ppmin(q),$ or if $n$ is odd, $q = \pm \ii$} \big\}.
\end{align}
Finally, for each $s \in \DefectSet_\multii$, we define
\begin{align}\label{Dommultii} 
\Dom_\multii\super{s} := \Dom_{\Summed_\multii}\super{s} \qquad \qquad \text{and} \qquad \qquad \Dom_\multii := \bigcap_{s \, \in \, \DefectSet_\multii} \Dom_\multii\super{s}. 
\end{align}
We consider these sets in more detail in the end of this section, lemmas~\ref{ContainLem} and~\ref{qtonulem}.

\begin{cor} \label{GridCor} 
Suppose $\max \multii  < \ppmin(q)$.  
We have $\rad \smash{\LS_\multii\super{s}} = \{0\}$ if and only if $q \in \smash{\Dom_\multii\super{s}}$.
\end{cor}
\begin{proof} 
It is evident that $\smash{\hcancel{\Dim}_{\Summed_\multii}\super{s}} = 0$ if and only if $(\Summed_\multii,s)$ lies in the closure of a pink triangle 
or on a dashed line in the lattice in figure~\ref{GridFigure}. 
Hence, the claim follows from corollaries~\ref{DnsLemAndRadDimCor} 
and~\ref{RadDimCor3} and the definition of $\smash{\Dom_\multii\super{s}}$.
\end{proof}

We can use corollary~\ref{GridCor} to strengthen corollaries~\ref{RadicalCor2} and~\ref{RadicalCor3} to if-and-only-if statements:

\begin{cor} \label{RadicalCor4} 
Suppose $\max \multii < \ppmin(q)$.
We have $\rad \LS_\multii = \{0\}$ if and only if $q \in \Dom_\multii$. 
\end{cor}
\begin{proof} 
Corollary~\ref{GridCor} with corollary~\ref{RadDimCor3} implies that 
$\rad \smash{\LS_\multii\super{s}} = \{0\}$
for all $s \in \DefectSet_\multii$, or equivalently by direct-sum decomposition~\eqref{RadDirSum} 
that $\rad \LS_{\Summed_\multii} = \{0\} = \rad \LS_\multii$, 
if and only if all of the points $(\Summed_\multii,s)$ with $s \in \DefectSet_\multii$ 
lie in the closures of the pink triangles or on the dashed lines in figure~\ref{GridFigure}. 
This happens if and only if $q \in \Dom_\multii$.  
\end{proof}

\begin{cor} \label{RadicalCor5} 
Suppose $\max \multii < \ppmin(q)$.
We have $\rad \smash{\LS_\multii} = \{0\}$, for all $\multii \in \{\OneVec{0}\} \cup \smash{\bZnn^\#}$, if and only if $\ppmin(q) = \infty$.
\end{cor}

\begin{proof} 
This immediately follows from corollary~\ref{RadicalCor4}.
\end{proof}

The containment $\DefectSet_\multii \subset \DefectSet_{\Summed_\multii}$ 
implies that $\Dom_{\Summed_\multii} \subset \Dom_\multii$.
In fact, this containment becomes an equality when intersected with the set $ \{ q \in \bC^\times \,|\, \max \multii < \ppmin(q)\}$:

\begin{lem} \label{ContainLem} We have 
\begin{align} 
\label{Contain0} 
\Dom_\multii \cap \{ q \in \bC \,|\, \max \multii < \ppmin(q) \} 
&= \Dom_{\Summed_\multii} \cap \{ q \in \bC \,|\, \max \multii < \ppmin(q) \} \\
\label{Contain}
&= \big\{ q \in \bC^\times \,\big|\, \textnormal{either $\Summed_\multii < \ppmin(q),$ or if 
$\Summed_\multii$ is odd and $\multii = \OneVec{\Summed}_\multii$, $q = \pm \ii$} \big\}. 
\end{align}

\end{lem}
\begin{proof}
To prove the lemma, we show that each $q \in \bC^\times$ either belongs to both sets on either side of the equality~\eqref{Contain0}, 
or belongs to neither. This approach will indirectly yield the explicit form~\eqref{Contain} for these two sets.

Throughout this proof, we assume that $q \in \bC^\times$ is such that $\max \multii < \ppmin(q)$. Proving~\eqref{Contain0} first, we also initially assume that $q \neq \pm 1$, so $2 \leq \pmin(q) = \ppmin(q)$ by~\eqref{MinPower}. Under these assumptions, we consider two cases:

\begin{enumerate}[leftmargin=*]
\itemcolor{red}
\item  \label{ItemA} $\Summed_\multii \geq \ppmin(q)$: In this case, we first observe that
\begin{align}
\np_\multii = 1 \qquad \Longrightarrow \qquad 
\text{$\qquad \multii = (s)$, for some $s \in \bZnn$} \qquad \Longrightarrow \qquad \Summed_\multii = s = \max \multii < \ppmin(q), 
\end{align}
a contradiction. Hence, we must have $\np_\multii > 1$ whenever $\Summed_\multii \geq \ppmin(q)$. In light of this observation, we may invoke lemma~\ref{DefectLem} to say that the minimum value $\smin(\multii)$ of the set $\DefectSet_\multii$~\eqref{DefSet2} satisfies
\begin{align} \label{sminSmallerThanDelta}
\smin(\multii) \overset{\eqref{sminineq}}{<} \max \multii \leq \ppmin(q) - 1 \overset{\eqref{MinPower}}{=} \pmin(q) - 1 \overset{\eqref{DeltaDefn}}{=} \Delta_1 .
\end{align}
Furthermore, lemma~\ref{DefectLem} implies that the maximum value $\smax(\multii)$ of the set $\DefectSet_\multii$ satisfies
\begin{align} \label{twoine} 
\Delta_1 \overset{\eqref{DeltaDefn}}{<} \ppmin(q) \leq \Summed_\multii \overset{\eqref{smaxeq}}{=}\smax(\multii). 
\end{align}
Assuming that $q \neq \pm \ii$, so $\ppmin(q) \geq 3$, it is straightforward to see that, under (\ref{sminSmallerThanDelta},~\ref{twoine}), there is a lattice point 
$(\Summed_\multii, s)$ off the pink triangles and dashed lines in figure~\ref{GridFigure} and with $s \in \DefectSet_\multii \subset \DefectSet_{\Summed_\multii}$. 
Thus, we have $q \notin \Dom_\multii \cup \Dom_{\Summed_\multii}$.

On the other hand, if $q = \pm \ii$, then $\ppmin(q) = 2$, and $\max \multii < \ppmin(q) = 2$ 
implies that $\multii = \OneVec{\Summed_\multii}$.
Also, if $\Summed_\multii$ is odd, then by~\eqref{DefSet2}, $\ppmin(q) = \pmin(q) = 2$ divides $s+1$, 
for each $s \in \DefectSet_{\Summed_\multii}$. 
By the containment $\DefectSet_\multii \subset \DefectSet_{\Summed_\multii}$, the same holds for each $s \in \DefectSet_\multii$. 
Hence, $\pm \ii \in \Dom_\multii \cap \Dom_{\Summed_\multii}$ if $\Summed_\multii$ is odd.  
On the other hand, if $\Summed_\multii$ is even, then $\pmin(q) = 2$ 
divides no element in the set $\DefectSet_{\Summed_\multii}$ nor in $\DefectSet_\multii$.  
Reasoning as in the previous paragraph, we then see that 
$\pm \ii \notin \Dom_\multii \cup \Dom_{\Summed_\multii}$.

\item  \label{ItemB} $\Summed_\multii < \ppmin(q)$: 
By~\eqref{DefSet2}, it is evident that for each $s \in \DefectSet_{\Summed_\multii}$, the lattice point $(\Summed_\multii, s)$ is on the bottommost pink triangle in figure~\ref{GridFigure}, with one exception: if $s = \Summed_\multii =  \pmin(q)-1$, then $(\Summed_\multii, s)$ lies on the lowest dashed line, at height $\Delta_1$. By the containment $\DefectSet_\multii \subset \DefectSet_{\Summed_\multii}$, the same holds for every $s \in \DefectSet_\multii$. Therefore, we have $q \in \Dom_\multii \cap \Dom_{\Summed_\multii}$. 
\end{enumerate}
Finally, from (\ref{MinPower},~\ref{Domns2},~\ref{Domn2}), it is straightforward to see that $\pm1$ is an element of the sets in (\ref{Contain0},~\ref{Contain}). From this and items~\ref{ItemA} and~\ref{ItemB} above, we conclude that equality~\eqref{Contain0} holds, and we infer~\eqref{Contain}.
\end{proof}

It is sometimes useful to understand the domain $\Dom_n$, 
determined mainly by the condition $\Summed < \ppmin(q)$, in terms of the fugacity $\nu$. 

\begin{lem} \label{qtonulem} 
Suppose $\nu = -q-q^{-1}$ and $\ppmin(q)$ is given by~\eqref{MinPower}. The following hold:
\begin{enumerate}
\itemcolor{red}

\item \label{AltItem1}
We have $q = \pm \ii$ if and only if $\nu = 0$.

\item \label{AltItem2}
We have
\begin{align}\label{Alt} 
n < \ppmin(q) \qquad \qquad \Longleftrightarrow \qquad \qquad
\nu^2 \neq 4\cos^2\left(\frac{\pi p'}{p}\right) \quad \parbox{5cm}{\textnormal{for any $p',p \in \bZpos$ coprime \\ and satisfying $0 < p' < p \leq n.$}}
\end{align} 
\end{enumerate}
\end{lem}

\begin{proof}
Item~\ref{AltItem1} is obvious.  
For item~\ref{AltItem2}, we note that with $p'$ any positive integer coprime with and less than $p$, we have
\begin{align} 
p:= \ppmin(q) \leq n \qquad \Longleftrightarrow \qquad q = \pm e^{\pi \ii p'/p} .
\end{align} 
Relation~\eqref{Alt} follows from this and our chosen parameterization $\nu = -q-q^{-1}$.
\end{proof}

\subsection{Totally degenerate cases} \label{rofSect32}

We recall that the bilinear form  on the standard module $\smash{\LS_\multii\super{s}}$ 
is said to be ``totally degenerate'' if $\rad \smash{\LS_\multii\super{s}} = \smash{\LS_\multii\super{s}}$. 
In section~\ref{LinkStateModSect}, propositions~\ref{GenLem2} and~\ref{HomLem2}, 
and corollary~\ref{nonisoCor2} all assume that 
this is not the case. Because these results are fundamental to 
understanding the structure of the standard modules, it is worthwhile to determine all $q \in \bC^\times$ for which the bilinear form  of a given standard module is totally degenerate. 
We establish this in proposition~\ref{WholeRadicalImpliesSsmallLem}.

\begin{lem} \label{AuxiliaryForWholeRadicalImpliesSsmallLem}
Suppose $\max \multii < \ppmin(q)$. If  $\ppmin(q) \leq s + 1$,
then there exists a walk $\varrho$ over $\multii$ with defect $s$ such that, when followed
backward, $\varrho\superscr{\,\downarrow}$ hits height $\Delta_{k_s}$ before $\varrho\superscr{\,\uparrow}$ hits height $\Delta_{k_s+1}$.
\end{lem}

\begin{proof}
Because $s$ is finite and $\ppmin(\pm1) = \infty$ by~\eqref{MinPower}, we must assume that $q \neq \pm1$ throughout, 
so $\ppmin(q) = \pmin(q) $. We prove the lemma by induction on the length $\np_\multii \in \bZpos$ of the multiindex $\multii$. 
Assuming first that $\np_\multii = 1$, we have
\begin{align}\label{mulmax} 
\multii = (\sIndex_1) \qquad \Longrightarrow \qquad \sIndex_1 = \max \multii < \pmin(q) ,
\end{align}
by the assumption in the lemma. Furthermore, there is exactly one walk $\varrho = (\sIndex_1)$ over 
$\multii = (\sIndex_1)$, trivially with defect $s = \sIndex_1$. 
Thus, with $\pmin(q) < s + 1$ by assumption, we have
\begin{align}
\pmin(q) -1 \leq s + 1= \sIndex_1 + 1 \qquad \overset{\eqref{mulmax}}{\Longrightarrow} \qquad \pmin(q) = \sIndex_1 + 1 = s + 1 \qquad 
\underset{\eqref{skDefn}}{\overset{\eqref{DeltaDefn}}{\Longrightarrow}} \qquad 
\begin{cases} s = \Delta_{k_s}, \\ k_s = 1. 
\end{cases} 
\end{align}
Thus, it is trivially true that when followed
backward, $\varrho\superscr{\,\downarrow}$ hits height $\Delta_{k_s}$ before $\varrho\superscr{\,\uparrow}$ hits height $\Delta_{k_s+1}$.

Next, we prove that if the lemma holds for all multiindices in $\smash{\bZpos^{\np-1}}$ for some $\np \in \{2,3, \ldots\}$, then it holds for all 
multiindices $\multii \in \smash{\bZpos^\np}$. In light of the comment immediately beneath the proof of lemma~\ref{WalkMultiiLem}, 
item~\ref{wmlIt7} of that lemma implies that there exists a walk $\varrho$ over $\multii$ with defect $s$ whose penultimate height equals
\begin{align}\label{rDefn} 
r:= r_{\np-1} = \min (\DefectSet_{\hat{\multii}} \cap \DefectSet\sub{\sIndex_{\np},s}) 
\underset{\eqref{SpecialDefSet}}{\overset{\eqref{pre-rj2}}{=}} \max(\smin(\hat{\multii}),|s - \sIndex_{\np}|). 
\end{align}
Now, there are two scenarios to consider:

\begin{enumerate}[leftmargin=*]
\itemcolor{red}
\item $r \leq \Delta_{k_s}$: With the penultimate height of $\varrho$ equaling $r$, it is evident that 
$\varrho\superscr{\,\downarrow}$ hits height $\Delta_{k_s}$ at the last step of $\varrho$ while, as we observed in~\eqref{NotSimul}, 
$\varrho\superscr{\,\uparrow}$ cannot simultaneously hit height $\Delta_{k_s+1}$.

\item $r > \Delta_{k_s}$:  
First, for the last step of $\varrho$,  by the assumptions of this lemma and by lemma~\ref{DefectLem}, we have
\begin{align}\label{FirstOb} 
\begin{cases} \pmin(q) \leq s + 1 \\ \sIndex_\np \leq \max \multii < \pmin(q) 
\end{cases} 
\qquad \Longrightarrow \qquad |s - \sIndex_\np| 
= s - \sIndex_\np \leq s \quad \text{and} \quad 0 \overset{\eqref{skDefn}}{<} k_s, 
\end{align}
and 
\begin{align}\label{SecondOb} 
\smin(\smash{\lds}) \overset{\eqref{sminineq}}{\leq} \max \smash{\lds}
\leq \max \multii \leq \pmin(q) - 1 \underset{\eqref{FirstOb}}{\overset{\eqref{DeltaDefn}}{\leq}} \Delta_{k_s} < r, 
\end{align}
which together show that
\begin{align}\label{whatisr} 
r \overset{\eqref{rDefn}}{=} \max(\smin(\hat{\multii}),|s - \sIndex_{\np}|) \underset{\eqref{SecondOb}}{\overset{\eqref{FirstOb}}{=}} s - \sIndex_\np. 
\end{align}
This implies that the apex of $\varrho$ at its last step is less than $\Delta_{k_s+1}$,
\begin{align}\label{lastapex} 
h_{\max,\np-1}(\varrho) \overset{\eqref{minmaxh2}}{=} \frac{r + \sIndex_\np + s}{2} 
\overset{\eqref{whatisr}}{=} s \overset{\eqref{skDefn}}{<} \Delta_{k_s + 1}. 
\end{align}
Second, we consider $\varrho$ from its first to its penultimate step, or equivalently, we consider $\hat{\varrho}$. 
We observe that
\begin{align} 
\pmin(q) - 1 & \overset{\eqref{DeltaDefn}}{=} \Delta_1 
\label{krks} 
\, \leq \, \Delta_{k_s} < r \overset{\eqref{whatisr}}{\leq} s \overset{\eqref{skDefn}}{<} \Delta_{k_s + 1},
\end{align}
which implies two facts. First, by~\eqref{skDefn}, we have $k_r = k_s$. Second, with $\pmin(q) < r + 1$, the induction hypothesis says that we can choose the walk 
$\varrho = (\hat{\varrho}, s)$ such that $\hat{\varrho}$ is a walk over $\smash{\lds}$ with defect $r$ and with the property that, when followed backward, 
$\hat{\varrho}\superscr{\,\downarrow}$ hits height $\Delta_{k_r} = \Delta_{k_s}$ before $\hat{\varrho}\superscr{\,\uparrow}$ hits height $\Delta_{k_r+1} = \Delta_{k_s + 1}$. 
In light of~\eqref{krks}, this implies that $\varrho$, a walk over $\multii$ with defect $s$, has this same property when followed backward from its penultimate step.  
It follows from this fact and~\eqref{lastapex} that $\varrho\superscr{\,\downarrow}$ hits height $\Delta_{k_s}$ before $\varrho\superscr{\,\uparrow}$ hits height 
$\Delta_{k_s + 1}$ as we follow $\varrho$ backward.
\end{enumerate}
This concludes the proof.
\end{proof}

Next, we state a condition that is both necessary and sufficient for the bilinear form on $\smash{\LS_\multii\super{s}}$ to be totally degenerate. 
For this purpose, we define the following set, with Lebesgue-measure zero in $\bC$:
\begin{align}\label{TotDefn} 
\Tot_\multii\super{s} := \big\{ q \in \bC^\times \, \big|\, s + 1 < \ppmin(q) 
\leq \displaystyle{\min_\varrho} \, \max _{ 0 \, \leq \, j \, < \, \np_\multii} h_{\max,j}(\varrho) + 1 \big\} ,
\end{align}
where the maximum is taken over all walks $\varrho$ over $\multii$ with defect $s$. 
Stated in other words, we have $q \in \smash{\Tot_\multii\super{s}}$ if and only if 
$q \in \bC^\times$ and the maximum apex of each walk $\varrho$ over $\multii$, with defect $s < \ppmin(q) - 1$, 
is at or above height $\ppmin(q) - 1$.

\begin{prop} \label{WholeRadicalImpliesSsmallLem} 
Suppose $\max \multii  < \ppmin(q)$.  
We have $\rad \smash{\LS_\multii\super{s}} = \smash{\LS_\multii\super{s}}$ if and only if $q \in \smash{\Tot_\multii\super{s}}$.
\end{prop}

\begin{proof} 
If $q \in \{\pm1\}$, then $q \not \in \smash{\Tot_\multii\super{s}}$ and 
$\rad \smash{\LS_\multii\super{s}} = \{0\}$ by~\eqref{MinPower} and corollary~\ref{RadicalCor2}. 
Thus, we take $q \not\in \{\pm1\}$ throughout the proof. As such, we have $\ppmin(q) = \pmin(q) $ throughout.

First, we assume that $q \in \smash{\Tot_\multii\super{s}}$. Then by the comment beneath~\eqref{TotDefn}, 
the comment beneath corollary~\ref{RadWalkLem2}, and 
item~\ref{wmlIt4} of lemma~\ref{WalkMultiiLem}, this implies that 
$\smash{\hcancel{\Dim}_\multii\super{s}} = \smash{\Dim_\multii\super{s}}$, or equivalently by 
lemma~\ref{LSDimLem2} and corollary~\ref{DnsLemAndRadDimCor2}, that $\rad \smash{\LS_\multii\super{s}} = \smash{\LS_\multii\super{s}}$.

Next, we assume that $\rad \smash{\LS_\multii\super{s}} = \smash{\LS_\multii\super{s}}$, or equivalently by lemma~\ref{LSDimLem2} and 
corollary~\ref{DnsLemAndRadDimCor2}, that $\smash{\hcancel{\Dim}_\multii\super{s}} = \smash{\Dim_\multii\super{s}}$. 
Now, item~\ref{wmlIt4} of lemma~\ref{WalkMultiiLem}, 
lemma~\ref{RadWalkLem2}, and lemma~\ref{AuxiliaryForWholeRadicalImpliesSsmallLem} combine to show that if $\pmin(q) \leq s + 1$, 
then 
$\smash{\hcancel{\Dim}_\multii\super{s}} < \smash{\Dim_\multii\super{s}}$, a contradiction. Hence, we have $s + 1 < \pmin(q)$. 
From this inequality, the equality 
$\smash{\hcancel{\Dim}_\multii\super{s}} = \smash{\Dim_\multii\super{s}}$, 
the comment following corollary~\ref{RadWalkLem2}, and the comment following~\eqref{TotDefn}, 
we conclude that $q \in \smash{\Tot_\multii\super{s}}$.
\end{proof}

Now we consider the special case that $\multii = \OneVec{n}$, for some $n \in \bZpos$.

\begin{lem} \label{TnsLem} 
We have
\begin{align} \label{TnsLemID}
\Tot_n\super{s} = 
\begin{cases} 
\emptyset, & s \neq 0, \\ \{\pm \ii \}, & s = 0. 
\end{cases}
\end{align}
\end{lem}

\begin{proof}
Suppose first $s \in \DefectSet_n \setminus \{0\}$, and let $\varrho$ be the walk over $\OneVec{n}$ which, 
when followed backward, descends from height 
$s$ until it hits height zero and then jumps back and forth between heights one and zero for the rest of its length. 
For this walk, the maximum apex $h_{\max,j}(\varrho)$ over $j \in \{0,1,\ldots,\np_\multii-1\}$ equals $s$. 
Therefore, by~\eqref{TotDefn}, we have
\begin{align}
\Tot_n\super{s} = \{q \in \bC^\times \,|\, s + 1 < \ppmin(q) \leq s + 1 \} = \emptyset .
\end{align}
On the other hand, if $s = 0$ and $\varrho$ is the walk over $\OneVec{n}$ that jumps back and forth between heights one and zero, then 
the maximum apex $h_{\max,j}(\varrho)$ over $j \in \{0,1,\ldots,\np_\multii-1\}$ equals one and we have
\begin{align} 
\Tot_n\super{0} = \{q \in \bC^\times \,|\, 1 < \ppmin(q) \leq 2 \} \overset{\eqref{MinPower}}{=} \{\pm \ii \} .
\end{align}
This proves asserted identity~\eqref{TnsLemID}.
\end{proof}

From the above lemma, we recover a result of D.~Ridout and Y.~Saint-Aubin: 

\begin{cor} \label{FullRadProp}
\textnormal{\cite[proposition~\red{3.5}]{rsa}} 
We have $\rad \smash{\LS_n\super{s}} = \smash{\LS_n\super{s}}$ if and only if $s = 0$ and $\ppmin(q) = 2$ 
\textnormal{(}i.e., $q \in \{\pm \ii\}$\textnormal{)}.
\end{cor}

\begin{proof} 
This immediately follows from proposition~\ref{WholeRadicalImpliesSsmallLem} and lemma~\ref{TnsLem}. 
\end{proof}

We finish by some further observations concerning a totally degenerate bilinear form.
First, thanks to the condition $s + 1 < \ppmin(q)$, we can ``zero-out" the defect height in all cases where 
$\smash{\rad \LS_\multii\super{s}} = \smash{\LS_\multii\super{s}}$.

\begin{lem} \label{ZeroOutLem} 
Suppose $\max \multii < \ppmin(q)$. We have
\begin{align} 
\rad \LS_\multii\super{s} = \LS_\multii\super{s} \qquad \Longrightarrow \qquad
\rad \LS_{\multii \oplus \scaleobj{0.85}{(s)}}\super{0} 
= \LS_{\multii \oplus \scaleobj{0.85}{(s)}}\super{0}. 
\end{align}
\end{lem}

\begin{proof} 
If $q \in \{\pm1\}$, then we have $\rad \smash{\LS_\multii\super{s}} = \{0\}$ by~\eqref{MinPower} 
and corollary~\ref{RadicalCor2}. 
Thus, we may assume that $q \not\in \{\pm1\}$ throughout the proof. As such, we have $\ppmin(q) = \pmin(q)$ throughout.

Now, if $\rad \smash{\LS_\multii\super{s}} = \smash{\LS_\multii\super{s}}$, then we have $s + 1 < \pmin(q)$ 
by lemma~\ref{WholeRadicalImpliesSsmallLem} and~\eqref{TotDefn}. 
In light of this fact, lemma~\ref{TieOffLem} of appendix~\ref{TLRecouplingSect} says that the evaluations of 
the following networks are equal for any 
two valenced link states $\smash{\alpha, \beta \in \LS_\multii\super{s}}$:
\begin{align}\label{CloseIt0} 
\vcenter{\hbox{\includegraphics[scale=0.275]{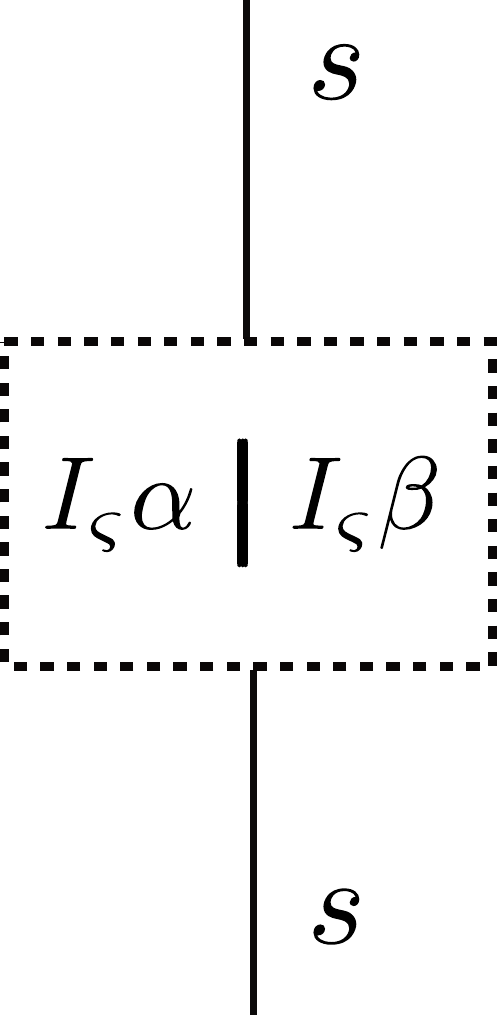}}} 
\qquad \qquad  \text{and} \qquad \qquad
\frac{(-1)^s}{[s+1]} \,\, \times \,\,
\vcenter{\hbox{\includegraphics[scale=0.275]{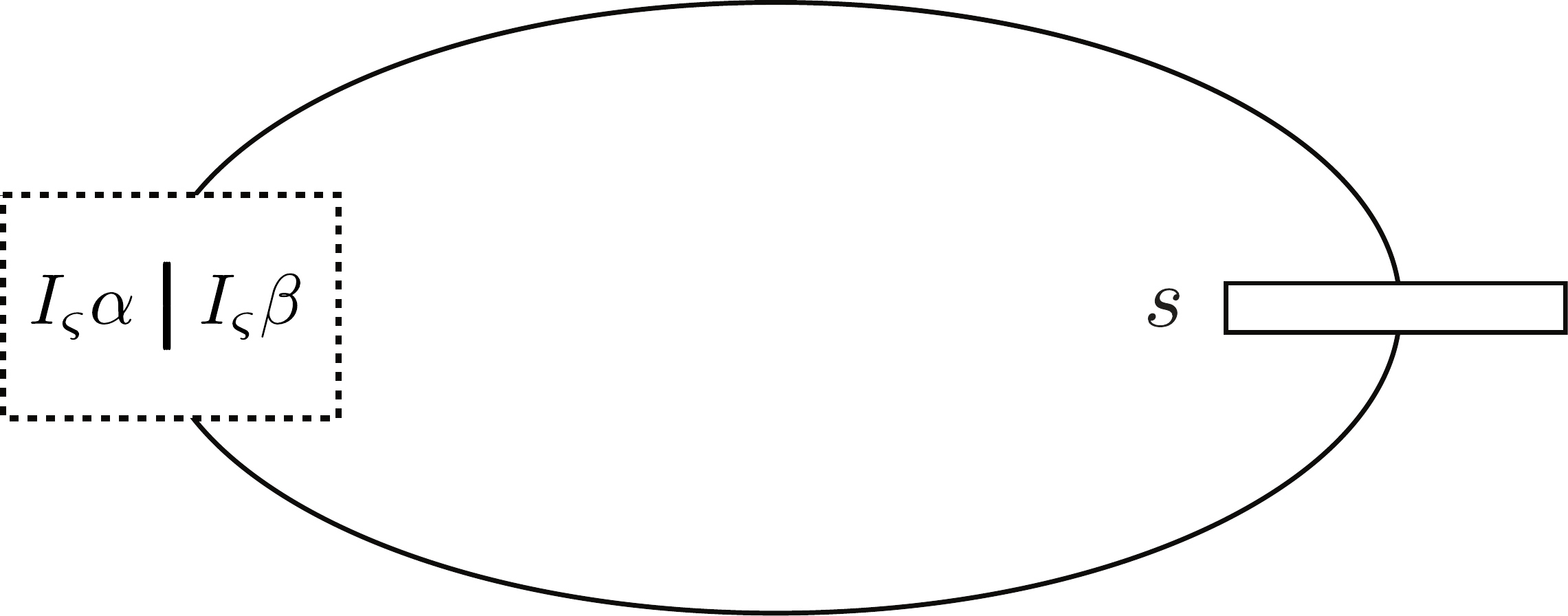} .}} 
\end{align}
For any pair of valenced link states $\smash{\gamma, \delta \in \LS_{\multii \oplus \scaleobj{0.85}{(s)}}\super{0}}$, 
we may write the network $\gamma \BarAction \delta$ in the form on the right side of~\eqref{CloseIt0} for 
some corresponding pair of valenced link states $\smash{\alpha, \beta \in \LS_\multii\super{s}}$. 
Thus, if $\smash{\rad \LS_\multii\super{s} = \LS_\multii\super{s}}$, then the network on the left side of~\eqref{CloseIt0} 
vanishes for any pair of valenced link states 
$\smash{\alpha, \beta \in \LS_\multii\super{s}}$, and with $s + 1 < \ppmin(q)$, we have $[s+1] \neq 0$. 
Thus, the network $\gamma \BarAction \delta$ on 
the right side of~\eqref{CloseIt0} also vanishes for any pair of valenced link states 
$\smash{\gamma, \delta \in \LS_{\multii \oplus \scaleobj{0.85}{(s)}}\super{0}}$.
\end{proof}

\begin{cor}
Suppose $\max \multii < \ppmin(q)$. We have
\begin{align}
\begin{cases} 
s + 1 < \pmin(q) \\ 
\rad \LS_{\multii \oplus \scaleobj{0.85}{(s)}}\super{0} = \LS_{\multii \oplus \scaleobj{0.85}{(s)}}\super{0} 
\end{cases} 
\qquad \Longrightarrow \qquad \rad \LS_\multii\super{s} = \LS_\multii\super{s}. 
\end{align}
\end{cor}

\begin{proof}
Arguing as in the proof of lemma~\ref{ZeroOutLem}, we see that
if $s + 1 < \pmin(q)$, the evaluations of the two networks in~\eqref{CloseIt0} are equal, so
$\smash{\rad \LS_{\multii \oplus \scaleobj{0.85}{(s)}}\super{0} = \LS_{\multii \oplus \scaleobj{0.85}{(s)}}\super{0}}$
implies $\rad \LS_\multii\super{s} = \LS_\multii\super{s}$.
\end{proof}

In the next two lemmas, we assume that $s = 0$. However,
using lemma~\ref{ZeroOutLem} and the comment beneath the proof of lemma~\ref{WalkMultiiLem}, 
it is straightforward to extend them to all cases in which this assumption is not true.

\begin{lem} \label{maxmaxlem} 
Suppose $\max \multii < \ppmin(q)$.  
If there exists a walk over $\multii$ with defect zero \textnormal{(}i.e., $0 \in \DefectSet_\multii$~\eqref{AltDefectSet}\textnormal{)}, then 
\begin{align}\label{thistothat} 
\ppmin(q) \leq \max_{0 \, \leq \, j \, < \, \np_\multii} \max \bigg\{ \frac{\mu_j(\multii) + \mu_{j+1}(\multii) + \sIndex_{j+1}}{2}, \, \sIndex_{j+1} \bigg\} + 1 
\qquad \qquad \Longrightarrow \qquad \qquad  \rad \LS_\multii\super{0} = \LS_\multii\super{0}. 
\end{align}
\end{lem}

\begin{proof} 
For each $j \in \{0, 1, \ldots, \np_\multii-1\}$ and each walk $\varrho$ over $\multii$ with defect zero, we have 
\begin{align}\label{observ0} 
\min_{\varrho'} h_{\max, j}(\varrho') \,  \leq \, 
h_{\max, j}(\varrho) \, \leq \, 
\max_{0 \, \leq \, k \, < \, \np_\multii} h_{\max, k}(\varrho),
\end{align}
where the minimum is over all walks $\varrho'$ over $\multii$ with defect zero. 
Taking the minimum and maximum of~\eqref{observ0} over all $j \in \{0, 1, \ldots, \np_\multii-1\}$ 
and all walks $\varrho$ over $\multii$ with defect zero, respectively, 
and using item~\ref{wmlIt9} of lemma~\ref{WalkMultiiLem}, we find
\begin{align} 
\ppmin(q) 
& \overset{\eqref{thistothat}}{\leq} \max_{0 \, \leq \, j \, < \, \np_\multii} \max \bigg\{ \frac{\mu_j(\multii) + \mu_{j+1}(\multii) + \sIndex_{j+1}}{2}, \, \sIndex_{j+1} \bigg\} \\
&\overset{\eqref{minapex}}{=} \max_{0 \, \leq \, j \, < \, \np_\multii} \min_\varrho h_{\max, j}(\varrho) 
\overset{\eqref{observ0}}{\leq} \min_\varrho \max_{0 \, \leq \, j \, < \, \np_\multii} h_{\max, j}(\varrho) . 
\end{align}
Taking this together with~\eqref{TotDefn} and corollary~\ref{WholeRadicalImpliesSsmallLem} with 
$s = 0$ and the fact that $1 < \ppmin(q)$ for all $q \in \bC^\times$ by~\eqref{MinPower}, we arrive with~\eqref{thistothat}.
\end{proof}

\begin{cor} \label{WholeRadLem} 
Suppose $\max \multii < \ppmin(q)$. 
If $\ppmin(q) = \sIndex_i + 1$ for some $i \in \{1,2,\ldots,\np_\multii\}$ and there exists a walk over $\multii$ with defect zero 
\textnormal{(}i.e., $0 \in \DefectSet_\multii$~\eqref{AltDefectSet}\textnormal{)}, then we have $\rad \smash{\LS_\multii\super{0}} = \smash{\LS_\multii\super{0}}$.
\end{cor}

\begin{proof} 
This follows from lemma~\ref{maxmaxlem} with the fact that
\begin{align}
\ppmin(q) = \sIndex_i + 1 \leq \max_{0 \, \leq \, j \, < \, \np_\multii} \max \bigg\{ \frac{\mu_j(\multii) + \mu_{j+1}(\multii) + \sIndex_{j+1}}{2}, \, \sIndex_{j+1} \bigg\} + 1. 
\end{align}
(Alternatively, one may use lemma~\ref{LoopLemGen} of appendix~\ref{TLRecouplingSect} with $s =  \sIndex_i = \ppmin(q) - 1$ and $r=0$.)
\end{proof}

\section{Semisimplicity of the valenced Temperley-Lieb algebra} \label{FinalResultSect}

In this section, we give several equivalent criteria for the valenced Temperley-Lieb algebra to be semisimple.
We also classify the simple and indecomposable $\TL_\multii(\nu)$-modules, and determine the Jacobson radical of $\TL_\multii(\nu)$.

The proofs that we present in this section are explicit and self-contained, 
relying on results from sections~\ref{DiagramAlgebraSect} and~\ref{StdModulesSect} and basic representation-theoretical facts. 
We remark that, admitting the fact that the valenced Temperley-Lieb algebra is cellular, 
by~\cite[proposition~\red{2.4}]{fp0} and corollary~\ref{AnIsoCor} from appendix~\ref{AppWJ},
some of these results could also be obtained using the formalism of cellular algebras from J.~Graham and G.~Lehrer~\cite{gl, gl2}.

\subsection{Perspective from general representation theory of algebras} \label{RecapSec}

To begin, we briefly recall basic notions on the representation theory of associative algebras 
and collect salient facts in proposition~\ref{RepRecapProp}.
We recommend~\cite[appendix~\red{A}]{am} and~\cite[chapters~\red{III},~\red{IV}, and~\red{VIII}]{cr} for background.
We also use the standard terminology introduced in the beginning of section~\ref{StdModulesSect}.

Throughout, we let $\mathsf{A}$ be a finite-dimensional associative unital $\bC$-algebra.
An element $e \in \mathsf{A}$ is called an \emph{idempotent} if $e^2 = e$.
If the following further properties hold, then $e$ is called a \emph{primitive idempotent}:
first, $e \neq 0$, and second, if $e = e_1 + e_2$ for some idempotents $e_1,e_2 \in \mathsf{A}$
such that $e_1 e_2 = 0 = e_2 e_1$, then either $e_1 = 0$ or $e_2 = 0$. 
%

We may view $\mathsf{A}$ as a (left) $\mathsf{A}$-module, with the left action given by its multiplication.
The associated representation is called the \emph{regular representation} of $\mathsf{A}$.
There exists a finite set $\{\mathsf{P}_\lambda\}_\lambda$ of indecomposable $\mathsf{A}$-modules such that
\begin{align} \label{WeddDecPre}
\mathsf{A} \; \cong \; \bigoplus_\lambda  \mathsf{P}_\lambda .
\end{align}
These modules $\mathsf{P}_\lambda$ are called \emph{principal indecomposable} $\mathsf{A}$-modules.

Some of the modules $\mathsf{P}_\lambda$ in decomposition~\eqref{WeddDecPre} might be isomorphic,
and the multiplicities of the non-isomorphic ones are
given by the dimensions of all simple $\mathsf{A}$-modules (by item~\ref{RepRecap4} of proposition~\ref{RepRecapProp} below).
By the Krull-Schmidt theorem~\cite[theorem~\red{A6}]{am}, decomposition~\eqref{WeddDecPre} is unique up to permutation of the components.

The \emph{Jacobson radical} of $\mathsf{A}$ is the intersection of all of the maximal ideals in $\mathsf{A}$.
By~\cite[corollary~\red{A11}]{am}, it is equal to the nilradical of $\mathsf{A}$ (the intersection of all nilpotent ideals).
We denote this radical by $\rad \mathsf{A}$. 
Equivalently, the Jacobson radical of $\mathsf{A}$ is 
the intersection of all annihilators of simple $\mathsf{A}$-modules~\cite[chapter~\red{2}]{lam},
\begin{align} \label{Jacobson}
\rad \mathsf{A} 
\quad = \bigcap_{\substack{\text{maximal} \\ \text{ideals } \mathsf{J} \, \subset \,  \mathsf{A}}} \mathsf{J}
\quad  = \bigcap_{\substack{\text{simple} \\ \text{$\mathsf{A}$-modules } \mathsf{M}}} 
\big\{a \in \mathsf{A} \, \big| \, a.v = 0 \text{ for all } v \in \mathsf{M} \big\} .
\end{align}

There are numerous equivalent notions of ``semisimplicity'' of the algebra $\mathsf{A}$.
We say that $\mathsf{A}$ is \emph{semisimple} if the Jacobson radical of $\mathsf{A}$ is trivial: $\rad \mathsf{A} = \{0\}$.
The quotient algebra $\mathsf{A} / \rad \mathsf{A}$ is always semisimple.
We give some alternative conditions in the next lemma.
\begin{lem} \label{SSLem}
\textnormal{\cite[corollary~\red{A13}]{am} and \cite[(\red{24.5},~\red{25.8}), chapter~\red{25}]{cr}} \\
The algebra $\mathsf{A}$ is semisimple if and only if one of the following equivalent conditions hold:
\begin{enumerate}
\itemcolor{red}
\item \label{SS1} All principal indecomposable $\mathsf{A}$-modules are simple.

\item \label{SS2} The regular representation of $\mathsf{A}$ is completely reducible.


\item \label{SS4} All finite-dimensional $\mathsf{A}$-modules are semisimple. 

\end{enumerate}
\end{lem}


The next proposition and corollary collect salient facts about the representation theory of $\mathsf{A}$.

\begin{prop} \label{RepRecapProp}
The following hold for any finite-dimensional associative unital $\bC$-algebra $\mathsf{A}$:
\begin{enumerate}
\itemcolor{red}

\item \label{RepRecap1} 
\textnormal{\cite[corollary~\red{A8}]{am}}
Every principal indecomposable $\mathsf{A}$-module has the form $\mathsf{A} e$,
where $e$ is some primitive idempotent.

\item \label{RepRecap2} 
\textnormal{\cite[theorem~\red{A10}]{am}}
Every principal indecomposable $\mathsf{A}$-module $\mathsf{P}$ 
has a unique maximal proper submodule $\mathsf{N}$, 
and its quotient $\mathsf{P} / \mathsf{N}$ 
with respect to this submodule is simple.

\item \label{RepRecap3} 
\textnormal{\cite[corollary~\red{A12}]{am}}
There is a one-to-one correspondence between the non-isomorphic principal indecomposable $\mathsf{A}$-modules 
and the non-isomorphic simple $\mathsf{A}$-modules, given by 
$\mathsf{P} \leftrightarrow \mathsf{P} / \mathsf{N}$.

\item \label{RepRecap4} 
\textnormal{\cite[corollary~\red{A22}]{am}}
Let $\{\mathsf{M}_\lambda\}_\lambda$ and $\{\mathsf{P}_\lambda\}_\lambda$ be respectively the complete sets of non-isomorphic
simple and principal indecomposable $\mathsf{A}$-modules.
Under the regular representation, we have the direct-sum decomposition 
\begin{align} \label{WeddDec}
\mathsf{A} \; \cong \; \bigoplus_\lambda \, ( \dim \mathsf{M}_\lambda  ) \, \mathsf{P}_\lambda .
\end{align}
\end{enumerate}
\end{prop}

Item~\ref{RepRecap4} is a generalization of the well-known Wedderburn decomposition for semisimple algebras:
in the semisimple case, all of the modules $\mathsf{P}_\lambda = \mathsf{M}_\lambda$ are simple,
and we have the \emph{sum-of-squares formula}~\cite[corollary~\red{A21}]{am},
\begin{align}
\dim \mathsf{A} = \sum_\lambda \, ( \dim \mathsf{M}_\lambda  )^2 .
\end{align}


\subsection{Simple modules of the valenced Temperley-Lieb algebra} \label{TLSimpleModSect}

Next we determine the complete set of non-isomorphic simple $\TL_\multii(\nu)$-modules (proposition~\ref{SimpleModuleProp}).
All of them are quotients of standard modules $\smash{\LS_{\multii}\super{s}}$, indexed by those $s \in \DefectSet_\multii$
for which the radical of $\smash{\LS_{\multii}\super{s}}$ is not totally degenerate.
We begin with investigating idempotent elements in $\TL_\multii(\nu)$ and the corresponding principal indecomposable modules.

\begin{lem} \label{idempotentFormLem}
Suppose $\max \multii < \ppmin(q)$. Any nonzero idempotent $\idem \in \TL_\multii(\nu)$ has the form
\begin{align} \label{idempotent}
\idem \; = 
\sum_{\alpha , \beta \, \in \, \LP_{\multii}\super{s_\idem}} c_{\alpha,\beta}\super{s_\idem} \BarAction \alpha \quad \beta \BarAction
+ \sum_{\substack{\gamma, \delta  \, \in \, \LP_{\multii}\super{r} \\ r \, < \, s_\idem}} 
c_{\gamma,\delta}\super{r} \BarAction \gamma \quad \delta \BarAction , 
\end{align}
where $s_\idem \in \DefectSet_\multii$ is the largest number such that $\smash{c_{\alpha,\beta}\super{s_\idem}} \neq 0$, 
for some pair of valenced link patterns $\alpha, \beta \in \smash{\LP_\multii\super{s_\idem}}$,
and where the coefficients $\smash{c_{\alpha,\beta}\super{s_\idem}} \in \bC$ satisfy
\begin{align} \label{idempotentCoef}
c_{\alpha,\beta}\super{s_\idem}  \; = \; & \sum_{ \gamma , \delta \, \in \, \LP_{\multii}\super{s_\idem}} 
c_{\alpha,\gamma}\super{s_\idem} c_{\delta,\beta}\super{s_\idem} \BiForm{\gamma}{\delta} ,
\end{align}
for all valenced link patterns $\alpha , \beta \in \smash{\LP_{\multii}\super{s_\idem}}$.
\end{lem}
\begin{proof}
Similarly as in the proof of proposition~\ref{PreFaithfulProp}, we write the idempotent $\idem$ in the form
\begin{align}  \label{Idemp}
\idem \overset{\eqref{WJDirSum}}{=} \sum_{r \, \in \, \DefectSet_\multii} \sum_{\alpha , \beta \, \in \, \LP_{\multii}\super{r}} 
c_{\alpha,\beta}\super{r} \BarAction \alpha \quad \beta \BarAction ,
\end{align}
for some coefficients $\smash{c_{\alpha,\beta}\super{r}} \in \bC$. 
Then the square of $\idem$ reads
\begin{align} \label{IdempSquared}
\idem^2 \overset{\eqref{Idemp}}{=}
\sum_{r, r' \, \in \, \DefectSet_\multii} \sum_{\alpha , \beta \, \in \, \LP_{\multii}\super{r}} \sum_{\gamma , \delta \, \in \, \LP_{\multii}\super{r'}} 
c_{\alpha,\beta}\super{r}  c_{\gamma,\delta}\super{r'} \BarAction \alpha \quad \beta \BarAction \BarAction \gamma \quad \delta \BarAction .
\end{align}
A tangle $\BarAction \alpha \quad \beta \BarAction \BarAction \gamma \quad \delta \BarAction$ in~\eqref{IdempSquared}
is a linear combination of valenced link diagrams, each of which has a number $s \leq \min(r,r')$ of crossing links.
In particular, valenced link diagrams in~\eqref{IdempSquared} with maximal number $s_\idem$ of crossing links have 
both $r$ and $r'$ equal to $s_\idem$. Similarly as in the proof of lemma~\ref{RidoutIdLem}, 
we see that these diagrams arise from terms of the form 
\begin{align} \label{HighestTerms}
\BarAction \alpha \quad \beta \BarAction \BarAction \gamma \quad \delta \BarAction
= \BiForm{\beta}{\gamma} \BarAction \alpha \quad \delta \BarAction 
+ \sum_{s \, < \, s_\idem} T_{\alpha, \beta, \gamma, \delta}\super{s} ,
\end{align}
for some $\alpha , \beta, \gamma, \delta \in \smash{\LP_{\multii}\super{s_\idem}}$,
where the tangles $\smash{T_{\alpha, \beta, \gamma, \delta}\super{s} \in \TL_\multii^{\multii ;\scaleobj{0.85}{(s)}}(\nu)}$ 
have $s < s_\idem$ crossing links.

On the other hand, by the idempotent property $\idem^2 = \idem$, the terms in~\eqref{Idemp} and~\eqref{IdempSquared} 
with maximal number $s_\idem$ of crossing links must agree.  
These terms give rise to asserted system of equations~\eqref{idempotentCoef}:
\begin{align}
\sum_{\alpha , \beta \, \in \, \LP_{\multii}\super{s_\idem}} c_{\alpha,\beta}\super{s_\idem} \BarAction \alpha \quad \beta \BarAction
 \; \overset{(\ref{Idemp}-\ref{HighestTerms})}{=} 
 & \sum_{\alpha , \beta, \gamma , \delta \, \in \, \LP_{\multii}\super{s_\idem}} c_{\alpha,\beta}\super{s_\idem} c_{\gamma,\delta}\super{s_\idem}
\BiForm{\beta}{\gamma} \BarAction \alpha \quad \delta \BarAction \\
 \; \overset{\hphantom{(\ref{Idemp}-\ref{HighestTerms})}}{=} 
 & \sum_{\alpha , \beta, \gamma , \delta \, \in \, \LP_{\multii}\super{s_\idem}} 
 c_{\alpha,\gamma}\super{s_\idem} c_{\delta,\beta}\super{s_\idem}
\BiForm{\gamma}{\delta} \BarAction \alpha \quad \beta \BarAction ,
\end{align}
from which~\eqref{idempotentCoef} follows
because all of the valenced link diagrams $\BarAction \alpha \quad \beta \BarAction$ are linearly independent.
\end{proof}

To each nonzero idempotent in $\TL_\multii(\nu)$, we associate 
the maximal number $s_\idem \in \DefectSet_\multii$ of crossing links from lemma~\ref{idempotentFormLem}.

\begin{cor} \label{idempotentCor}
Suppose $\max \multii < \ppmin(q)$. If $\idem \in \TL_\multii(\nu)$ is a nonzero idempotent,
then $\smash{\rad \LS_\multii\super{s_\idem} \neq \LS_\multii\super{s_\idem}}$.
\end{cor}
\begin{proof}
If $\smash{\rad \LS_\multii\super{s_\idem} = \LS_\multii\super{s_\idem}}$, then 
all of the coefficients $\smash{c_{\alpha,\beta}\super{s_\idem}}$ in~\eqref{idempotentCoef} are zero.
This contradicts the choice of $s_\idem$.
\end{proof}

\begin{lem} \label{PreProjModuleLem}
Suppose $\max \multii < \ppmin(q)$. If $\idem \in \TL_\multii(\nu)$ is a nonzero idempotent,
then there exists a non-trivial surjective homomorphism of $\TL_\multii(\nu)$-modules
from $\TL_\multii(\nu) \idem$ onto the standard module $\smash{\LS_{\multii}\super{s_\idem}}$.
\end{lem}
\begin{proof}
%
We write $\idem$ in the form~\eqref{idempotent} of lemma~\ref{idempotentFormLem},
and choose $\alpha, \beta \in \smash{\LS_\multii\super{s_\idem}}$ such that $\smash{c_{\alpha,\beta}\super{s_\idem}} \neq 0$.
Then we define
\begin{align} \label{Gamma}
\gamma := 
\sum_{\eta \, \in \, \LP_{\multii}\super{s_\idem}} c_{\eta,\beta}\super{s_\idem} \eta 
\quad \in \; \LS_\multii\super{s_\idem} ,
\end{align}
and we note that $\gamma \notin \rad \smash{\LS_\multii\super{s_\idem}}$:
\begin{align}
\delta =
\sum_{\eta \, \in \, \LP_{\multii}\super{s_\idem}} c_{\alpha,\eta}\super{s_\idem} \eta 
\qquad \qquad \Longrightarrow \qquad \qquad 
\BiForm{\delta}{\gamma} = & \;
\sum_{\eta, \eta' \, \in \, \LP_{\multii}\super{s_\idem}} 
c_{\alpha,\eta}\super{s_\idem} c_{\eta',\beta}\super{s_\idem} \BiForm{\eta}{\eta'} 
\overset{\eqref{idempotentCoef}}{=}  c_{\alpha,\beta}\super{s_\idem} \neq 0 .
\end{align}
The next crucial observation is that both $\TL_\multii(\nu)$-modules are cyclic: $\idem$ generates $\TL_\multii(\nu) \idem$ and,
by the proof of proposition~\ref{GenLem2}, the valenced link state $\gamma$
generates $\smash{\LS_\multii\super{s_\idem}}$. 
Hence, our goal is to define a map $\theta \colon \TL_\multii(\nu) \idem \longrightarrow \smash{\LS_{\multii}\super{s_\idem}}$ 
by homomorphic extension of 
its image on the generator tangle $E$:
\begin{align}
\theta \colon \idem \mapsto \gamma \qquad \qquad \Longrightarrow \qquad \qquad \theta (T \idem) := T \theta(\idem) := T \gamma .
\end{align}
By construction, such a map is a homomorphism of $\TL_\multii$-modules from $\TL_\multii(\nu) \idem$ 
to $\smash{\LS_\multii\super{s_\idem}}$.
Furthermore, because the valenced link state $\gamma$ generates $\smash{\LS_\multii\super{s_\idem}}$, 
the map $\theta$ is a surjection. 
However, we need to verify that $\theta$ is well-defined, i.e.,
\begin{align} \label{weldef1}
& T_1 \idem =  T_2 \idem \qquad \text{for some valenced tangles $T_1, T_2 \in \TL_\multii(\nu)$} \\
\qquad \qquad \Longrightarrow \qquad \qquad
& \theta (T_1 \idem ) = \theta (T_2 \idem)
\qquad \Longrightarrow \qquad
T_1 \theta(\idem) = T_2 \theta(\idem)
\qquad \Longrightarrow \qquad
T_1 \gamma = T_2 \gamma .
\end{align}
In other words, $\theta$ is well-defined if and only if 
\begin{align} \label{weldef2}
T \idem = 0  \qquad \qquad \Longrightarrow \qquad \qquad T \gamma = 0  .
\end{align}
%
For this purpose, it suffices to show that $\idem \gamma = \gamma$, because 
\begin{align} \label{weldef3}
\begin{cases} 
\idem \gamma = \gamma \\ 
T \idem = 0 
\end{cases}
\qquad \qquad \Longrightarrow \qquad \qquad T \gamma =  T \idem \gamma = 0 .
\end{align}
From~\eqref{SmallerRAnnihilate} in the proof of proposition~\ref{PreFaithfulProp}, we see that 
only those terms in $\idem$ which have the maximal number $s_\idem$ of crossing links 
may give a nonzero contribution when acting on $\gamma$.
Hence, using lemmas~\ref{RidoutIdLem} and~\ref{idempotentFormLem}, we calculate
\begin{align}
\idem \gamma 
\underset{\eqref{Gamma}}{\overset{\eqref{idempotent}}{=}} & \;
\sum_{\mu, \nu , \eta \, \in \, \LP_{\multii}\super{s_\idem}} c_{\mu,\nu}\super{s_\idem} 
c_{\eta,\beta}\super{s_\idem} \BarAction \mu \quad \nu \BarAction \eta
\overset{\eqref{RidoutId}}{=}
\sum_{\mu, \nu , \eta \, \in \, \LP_{\multii}\super{s_\idem}} 
c_{\mu,\nu}\super{s_\idem} c_{\eta,\beta}\super{s_\idem} \BiForm{\nu}{\eta} \mu 
\overset{\eqref{idempotentCoef}}{=} 
\sum_{\mu \, \in \, \LP_{\multii}\super{s_\idem}} 
c_{\mu,\beta}\super{s_\idem} \mu 
\overset{\eqref{Gamma}}{=} \gamma .
\end{align}
It now follows from~(\ref{weldef1}--\ref{weldef3}) that the map $\theta$ is indeed well-defined, 
which is what we sought to prove.
\end{proof}

Next we find a connection between the principal indecomposable $\TL_\multii(\nu)$-modules and the standard modules.

\begin{lem} \label{ProjModuleLem} 
Suppose $\max \multii < \ppmin(q)$. The following hold for 
any principal indecomposable $\TL_\multii(\nu)$-module $\mathsf{P}$:
\begin{enumerate}
\itemcolor{red}
\item \label{ProjModuleProp1} 
There exists a non-trivial surjective homomorphism of $\TL_\multii(\nu)$-modules from $\mathsf{P}$ 
onto some standard module~$\smash{\LS_{\multii}\super{s}}$.

\item \label{ProjModuleProp2} 
The simple quotient of $\mathsf{P}$ by its maximal proper submodule is isomorphic to $\smash{\Quo_{\multii}\super{s}}$.
\end{enumerate}
%
\end{lem}
\begin{proof}
By item~\ref{RepRecap1} of proposition~\ref{RepRecapProp}, we have
$\mathsf{P} = \TL_\multii(\nu) \idem$, where $\idem$ is a primitive idempotent.
Thus, lemma~\ref{PreProjModuleLem} gives a non-trivial surjective homomorphism 
$\theta \colon \mathsf{P} \longrightarrow \smash{\LS_{\multii}\super{s}}$ of $\TL_\multii(\nu)$-modules,
where $s = s_\idem$ is given by lemma~\ref{idempotentFormLem}.
This proves item~\ref{ProjModuleProp1}.
Now, corollary~\ref{idempotentCor} and proposition~\ref{GenLem2} imply that 
$\rad \smash{\LS_{\multii}\super{s}}$ is the maximal proper submodule of $\smash{\LS_{\multii}\super{s}}$.
Its pre-image under the homomorphism $\theta$ is the maximal proper submodule $\mathsf{N}$ of $\mathsf{P}$.
Therefore, Schur's lemma shows that the simple quotient modules
$\mathsf{P} / \mathsf{N}$ and $\smash{\Quo_{\multii}\super{s}}$ are isomorphic.
This proves item~\ref{ProjModuleProp2}.
%
%
\end{proof}

We are now ready to conclude with the classification of the simple $\TL_\multii(\nu)$-modules. 
These modules are indexed by those $s \in \DefectSet_\multii$ for which the radical of $\smash{\LS_{\multii}\super{s}}$ is not totally degenerate.
We denote the set of such indices by
\begin{align} \label{WeddDecForTLIndex}
\DefectSet_\multii' := 
\big\{s \in \DefectSet_\multii \, \big| \, \dim \Quo_\multii\super{s} > 0 \big\}
\overset{\eqref{QuoMod}}{=}  \big\{s \in \DefectSet_\multii \, \big| \, \rad \LS_\multii\super{s} \neq \LS_\multii\super{s} \big\} 
\underset{\ref{WholeRadicalImpliesSsmallLem}}{\overset{\text{prop.}}{=}}
\big\{s \in \DefectSet_\multii \, \big| \, q \notin \Tot_\multii\super{s} \big\} .
\end{align}

\begin{prop} \label{SimpleModuleProp}
Suppose $\max \multii < \ppmin(q)$.  The collection 
$\smash{\big\{ \Quo_\multii\super{s} \,\big| \, s \in \DefectSet_{\multii}' \big\}}$
is the complete set of non-isomorphic simple $\TL_\multii(\nu)$-modules.
\end{prop}
\begin{proof}
By item~\ref{RepRecap3} of proposition~\ref{RepRecapProp}, the non-isomorphic 
simple and principal indecomposable $\TL_\multii(\nu)$-modules are in one-to-one correspondence with each other.
Therefore, 
item~\ref{ProjModuleProp2} of lemma~\ref{ProjModuleLem} shows that
the collection $\smash{\big\{ \Quo_\multii\super{s} \,\big| \, s \in \DefectSet_{\multii}' \big\}}$ 
contains all of the simple $\TL_\multii(\nu)$-modules.
On the other hand, proposition~\ref{GenLem2} and
corollary~\ref{nonisoCor2} show that all of these $\TL_\multii(\nu)$-modules are simple and non-isomorphic.
This concludes the proof.
\end{proof}

With proposition~\ref{SimpleModuleProp}, item~\ref{RepRecap4} of proposition~\ref{RepRecapProp} gives a direct-sum decomposition 
for 
$\TL_\multii(\nu)$ under the regular representation,
in terms of the complete sets of non-isomorphic simple and principal indecomposable $\TL_\multii(\nu)$-modules.

\begin{cor} \label{WeddDecForTLCor}
Suppose $\max \multii < \ppmin(q)$.  We have the direct-sum decomposition of $\TL_\multii(\nu)$-modules
\begin{align} \label{WeddDecForTL}
\TL_\multii(\nu) \; \cong \; 
\bigoplus_{s \, \in \, \DefectSet_\multii'} ( \dim \Quo_\multii\super{s}  ) \, \mathsf{P}_\multii\super{s} ,
\end{align}
where $\smash{\big\{ \mathsf{P}_\multii\super{s} \,\big| \, s \in \DefectSet_{\multii}' \big\}}$
is the complete set of non-isomorphic principal indecomposable $\TL_\multii(\nu)$-modules.
\end{cor}
\begin{proof}
Item~\ref{RepRecap4} of proposition~\ref{RepRecapProp} gives a direct-sum decomposition for $\TL_\multii(\nu)$ as 
a $\TL_\multii(\nu)$-module:
with $\{\mathsf{M}_\lambda\}_\lambda$ and $\{\mathsf{P}_\lambda\}_\lambda$ respectively the complete sets of non-isomorphic
simple and principal indecomposable $\TL_\multii(\nu)$-modules, we have
\begin{align} \label{WeddDecForTLPre}
\TL_\multii(\nu) \; \cong \; \bigoplus_\lambda \, ( \dim \mathsf{M}_\lambda  ) \, \mathsf{P}_\lambda .
\end{align}
Proposition~\ref{SimpleModuleProp}  now says that 
$\{\mathsf{M}_\lambda\}_\lambda = \smash{\big\{ \Quo_\multii\super{s} \,\big| \, s \in \DefectSet_{\multii}' \big\}}$,
and item~\ref{RepRecap3} of proposition~\ref{RepRecapProp} shows that the principal indecomposable modules 
share the same index set. This concludes the proof.
\end{proof}

\subsection{Semisimplicity of the valenced Temperley-Lieb algebra}\label{SemiSect}

Now we give several equivalent criteria for the valenced Temperley-Lieb algebra to be semisimple.

\begin{theorem} \label{BigSSTHM}
Suppose $\max \multii < \ppmin(q)$. The following statements are equivalent:
\begin{enumerate}
\itemcolor{red}
\item \label{SSitem0}
The valenced Temperley-Lieb algebra $\TL_\multii(\nu)$ is semisimple, i.e., $\rad\TL_\multii(\nu) = \{0\}$.

%

\item \label{SSitem4}
We have $\rad \LS_\multii = \{0\}$.

\item \label{SSitem5}
The link state representation induced by the action of $\TL_\multii (\nu)$ on $\LS_\multii$ is faithful.

\item \label{SSitem6}
The link state representation induces an isomorphism of algebras from $\TL_\multii (\nu)$ to 
$\smash{\underset{s \, \in \, \DefectSet_\multii}{\bigoplus} \End \LS_\multii\super{s}}$.

\item \label{SSitem2}
The collection $\smash{\big\{ \LS_\multii\super{s} \,\big| \, s \in \DefectSet_\multii \big\}}$ 
is the complete set of non-isomorphic simple $\TL_\multii(\nu)$-modules.

\item \label{SSitem3}
We have $q \in \Dom_\multii$.
\end{enumerate}
\end{theorem}

\begin{proof}
We prove the equivalences as follows:
\begin{enumerate}[leftmargin=3.5em]
\item[\ref{SSitem4} $\Leftrightarrow$ \ref{SSitem5}:] This is the content of corollary~\ref{PreFaithfulCor}.

\item[\ref{SSitem5} $\Leftrightarrow$ \ref{SSitem6}:] 
By definition~\eqref{LSDirSum2}, $\LS_\multii$ is the direct sum of its submodules $\smash{\LS_\multii\super{s}}$.
Because each of these submodules is closed under the $\TL_\multii(\nu)$-action, 
the image of the link state representation is contained in 
$\smash{\underset{s \, \in \, \DefectSet_\multii}{\bigoplus} \End \LS_\multii\super{s}}$.
Also, 
\begin{align} \label{SumSquaresEnd}
\dim \TL_\multii(\nu) 
\overset{\eqref{Dim56}}{=} \sum_{s \, \in \, \DefectSet_\multii} \big( \dim \LS_\multii\super{s} \big)^2 
= \dim \Big( \bigoplus_{s \, \in \, \DefectSet_\multii} \End \LS_\multii\super{s} \Big) ,
\end{align}
by corollary~\ref{WJDimLem1} with $\multiii = \multii$.
Therefore, by the dimension theorem, we have \ref{SSitem5} $\Leftrightarrow$ \ref{SSitem6}.

\item[\ref{SSitem4} $\Leftrightarrow$ \ref{SSitem3}:] 
This is the content of corollary~\ref{RadicalCor4}.
\end{enumerate}
Thus, items~\ref{SSitem4},~\ref{SSitem5},~\ref{SSitem6}, and~\ref{SSitem3} are equivalent.
We then prove the remaining equivalences:
\begin{enumerate}[leftmargin=3.5em]
\item[\ref{SSitem4} $\Rightarrow$ \ref{SSitem2}:]
Suppose $\rad \LS_\multii = \{0\}$.  
Then by~(\ref{RadDirSum},~\ref{WeddDecForTLIndex}), we have $\DefectSet_{\multii}' = \DefectSet_{\multii}$ and
$\smash{\LS_\multii\super{s} = \Quo_\multii\super{s}}$ for all $s \in \DefectSet_{\multii}$,
so $\smash{\big\{ \LS_\multii\super{s} \,\big| \, s \in \DefectSet_\multii \big\}}$ 
is the complete set of non-isomorphic simple $\TL_\multii(\nu)$-modules by proposition~\ref{SimpleModuleProp}.
Hence, we have \ref{SSitem4} $\Rightarrow$ \ref{SSitem2}.

\item[\ref{SSitem2} $\Rightarrow$ \ref{SSitem0}:]
Suppose $\smash{\big\{ \LS_\multii\super{s} \,\big| \, s \in \DefectSet_\multii \big\}}$ 
are all of the non-isomorphic simple $\TL_\multii(\nu)$-modules. 
Then, proposition~\ref{SimpleModuleProp} shows that $\DefectSet_{\multii}' = \DefectSet_{\multii}$,
and definition~\eqref{WeddDecForTLIndex} implies that
$\smash{\LS_\multii\super{s} = \Quo_\multii\super{s}}$ for all $s \in \DefectSet_{\multii}$. Thus, 
corollary~\ref{WeddDecForTLCor} gives 
\begin{align} \label{DSD}
\TL_\multii(\nu) \; \overset{\eqref{WeddDecForTL}}{\cong} \; 
\bigoplus_{s \, \in \, \DefectSet_\multii} ( \dim \LS_\multii\super{s}  ) \, \mathsf{P}_\multii\super{s}  .
\end{align}
On the other hand, by item~\ref{ProjModuleProp1} of lemma~\ref{ProjModuleLem},
we have $\smash{\dim \LS_\multii\super{s} \leq \dim \mathsf{P}_\multii\super{s}}$,
and combining this with~(\ref{SumSquaresEnd},~\ref{DSD}), we have  
$\smash{\dim \LS_\multii\super{s} =\dim \mathsf{P}_\multii\super{s}}$.
Therefore, by item~\ref{ProjModuleProp1} of lemma~\ref{ProjModuleLem} and the dimension theorem, we have
$\smash{\mathsf{P}_\multii\super{s}} \cong \smash{\LS_\multii\super{s}}$ for all $s \in \DefectSet_\multii$.
Hence, all of the principal indecomposable $\TL_\multii(\nu)$-modules are simple. 
Therefore, $\TL_\multii(\nu)$ is semisimple by item~\ref{SS1} of lemma~\ref{SSLem},
so we have \ref{SSitem2} $\Rightarrow$ \ref{SSitem0}.

\item[\ref{SSitem0} $\Rightarrow$ \ref{SSitem4}:] 
Suppose $\TL_\multii(\nu)$ is semisimple. Then item~\ref{SS1} of lemma~\ref{SSLem}
and item~\ref{ProjModuleProp2} of lemma~\ref{ProjModuleLem} together imply that we have
$\smash{\mathsf{P}_\multii\super{s} = \Quo_\multii\super{s}}$ for all $s \in \DefectSet_\multii'$.
Thus, corollary~\ref{WeddDecForTLCor} gives the direct-sum decomposition
\begin{align}  \label{WeddDecForTLSS}
\TL_\multii(\nu) \; \overset{\eqref{WeddDecForTL}}{\cong} \; 
\bigoplus_{s \, \in \, \DefectSet_\multii'} ( \dim \Quo_\multii\super{s}  ) \,  \Quo_\multii\super{s} .
\end{align}
On the other hand, we have
\begin{align} \label{SumSquares} 
\dim \TL_\multii(\nu)  
\overset{\eqref{WeddDecForTLSS}}{=}
\sum_{s \, \in \, \DefectSet_\multii'} \big( \dim \Quo_\multii\super{s} \big)^2 
\overset{\eqref{WeddDecForTLIndex}}{=} 
\sum_{s \, \in \, \DefectSet_\multii} \big( \dim \Quo_\multii\super{s} \big)^2 
\overset{\eqref{QuoMod}}{\leq} 
\sum_{s \, \in \, \DefectSet_\multii} \big( \dim \LS_\multii\super{s} \big)^2 
\overset{\eqref{Dim56}}{=} 
\dim \TL_\multii(\nu) .
\end{align}
Therefore, by~\eqref{SumSquares}, we have 
$\smash{\dim \Quo_\multii\super{s} = \dim \LS_\multii\super{s}}$, so 
$\smash{\dim \rad \LS_\multii\super{s}} = 0$, for all $s \, \in \, \DefectSet_\multii$.
This with~\eqref{RadDirSum} implies that $\smash{\rad \LS_\multii} = \{0\}$.
Hence, we have \ref{SSitem0} $\Rightarrow$ \ref{SSitem4}.

\end{enumerate}
This proves the asserted equivalences.
\end{proof}

In the special case that $\multii = \OneVec{n}$, for some $n \in \bZnn$,
theorem~\ref{BigSSTHM} gives equivalent criteria for the Temperley-Lieb algebra $\TL_n(\nu)$ to be semisimple,
many of which are already well-known~\cite{gl2,rsa}. For instance, we have:

\begin{cor} \label{TLSemiSimpCor} 
The Temperley-Lieb algebra $\TL_n(\nu)$ is semisimple if and only if 
\begin{align}  \label{TLSemiSimpQ}
q \in \Dom_n = \bigcap_{s \, \in \, \DefectSet_n} \Dom_n\super{s} 
= \big\{ q \in \bC^\times \,\big|\, \textnormal{either $n < \ppmin(q),$ or if $n$ is odd, $q = \pm \ii$} \big\}.
\end{align}
Lemma~\ref{qtonulem} phrases this condition~\eqref{TLSemiSimpQ} in terms of $\nu$. 
\end{cor}
\begin{proof}
This follows from the equivalence of items~\ref{SSitem0} and~\ref{SSitem3} in theorem~\ref{BigSSTHM} and definition~\eqref{Domn2}.
\end{proof}

The necessary and sufficient condition for the semisimplicity of $\TL_n(\nu)$ in terms of $q$ (or $\nu$)
seems to be rarely stated and even misstated in the literature. For example, in~\cite[theorem~\red{B8.4}]{br}, 
it is stated, without proof, 
to be
\begin{align}
\frac{1}{\nu^2} \neq 4 \cos^2 (\pi / p) \quad \text{for all $p \in \{2,3,\ldots,n\}$.}
\end{align}
This condition is similar but not identical to our condition~\eqref{Alt} discussed in lemma~\ref{qtonulem}.

%

\bigskip

Lastly, we identify the Jacobson radical of $\TL_\multii(\nu)$ as the kernel of 
its representation on the quotient space $\LS_\multii / \rad \LS_\multii$.
By propositions~\ref{GenLem2} and~\ref{SimpleModuleProp},
this quotient module is the direct sum of all simple $\TL_\multii(\nu)$-modules,
\begin{align} \label{QuotientL}
\LS_\multii / \rad \LS_\multii 
\overset{\eqref{LSDirSum2}}{=} 
\bigoplus_{s \, \in \, \DefectSet_\multii} \LS_\multii\super{s} / \rad \LS_\multii\super{s}
\overset{\eqref{QuoMod}}{=}  \bigoplus_{s \, \in \, \DefectSet_\multii'} \Quo_\multii\super{s} .
\end{align}

\begin{prop} \label{Jacobsonprop}
Suppose $\max \multii < \ppmin(q)$. 
The Jacobson radical of $\TL_\multii(\nu)$ equals 
the kernel of the representation of $\TL_\multii(\nu)$ on $\smash{\LS_\multii / \rad \LS_\multii}$.
\end{prop}
\begin{proof}
By proposition~\ref{SimpleModuleProp}, the Jacobson radical of $\TL_\multii(\nu)$ equals
\begin{align} \label{JabQ}
\rad \TL_\multii(\nu) \overset{\eqref{Jacobson}}{=} 
\bigcap_{s \, \in \, \DefectSet_\multii'} \big\{T \in \TL_\multii(\nu) \, \big| \, T.\alpha = 0 \text{ for all } \alpha \in \Quo_\multii\super{s} \big\} .
\end{align}
On the other hand, direct-sum decomposition~\eqref{QuotientL} shows that 
\begin{align}
\bigcap_{s \, \in \, \DefectSet_\multii'} \big\{T \in \TL_\multii(\nu) \, \big| \, T.\alpha = 0 \text{ for all } \alpha \in \Quo_\multii\super{s} \big\} 
\overset{\eqref{WeddDecForTLIndex}}{=} \; & 
\bigcap_{s \, \in \, \DefectSet_\multii} \big\{T \in \TL_\multii(\nu) \, \big| \, T.\alpha = 0 \text{ for all } \alpha \in \LS_\multii\super{s} / \rad \LS_\multii\super{s} \big\} \\
 \label{kernel}
\overset{\eqref{QuotientL}}{=} \; & \big\{T \in \TL_\multii(\nu) \, \big| \, T.\alpha = 0 \text{ for all } \alpha \in \LS_\multii / \rad \LS_\multii  \big\} .
\end{align}
By definition, the right side of~\eqref{kernel} 
is the kernel of the representation of $\TL_\multii(\nu)$ on $\smash{\LS_\multii / \rad \LS_\multii}$.
\end{proof}


\bigskip
\bigskip

\appendixpage
\begin{appendices}
\renewcommand{\thesection}{\Alph{section}}
\renewcommand{\thesubsection}{\arabic{subsection}}
\renewcommand{\thesubsubsection}{\Alph{subsubsection}}

%
\section{Diagram simplifications} \label{TLRecouplingSect}

The purpose of this appendix is to collect auxiliary results needed in this article, 
using diagram calculus known as Temperley-Lieb recoupling theory~\cite{pen, kl, cfs}.
We include the proofs for convenience of the reader, but all results of this appendix already appear in some forms in the literature.
We recall that the evaluation $(T)$ of a network $T$ is defined in~\eqref{evT} in section~\ref{StdModulesSect} as the product of weights 
(\ref{LoopWeight}--\ref{TurnBack0}) of all connected components in $T$.

We begin with some diagram simplifications. 
We use the following extraction rule in section~\ref{StdModulesSect}.

\begin{lem} \label{LoopLemGen} 
Suppose $s + r < \ppmin(q)$.  Then we have the following extraction rule:
\begin{align} \label{DeltaTangleGen} 
\vcenter{\hbox{\includegraphics[scale=0.275]{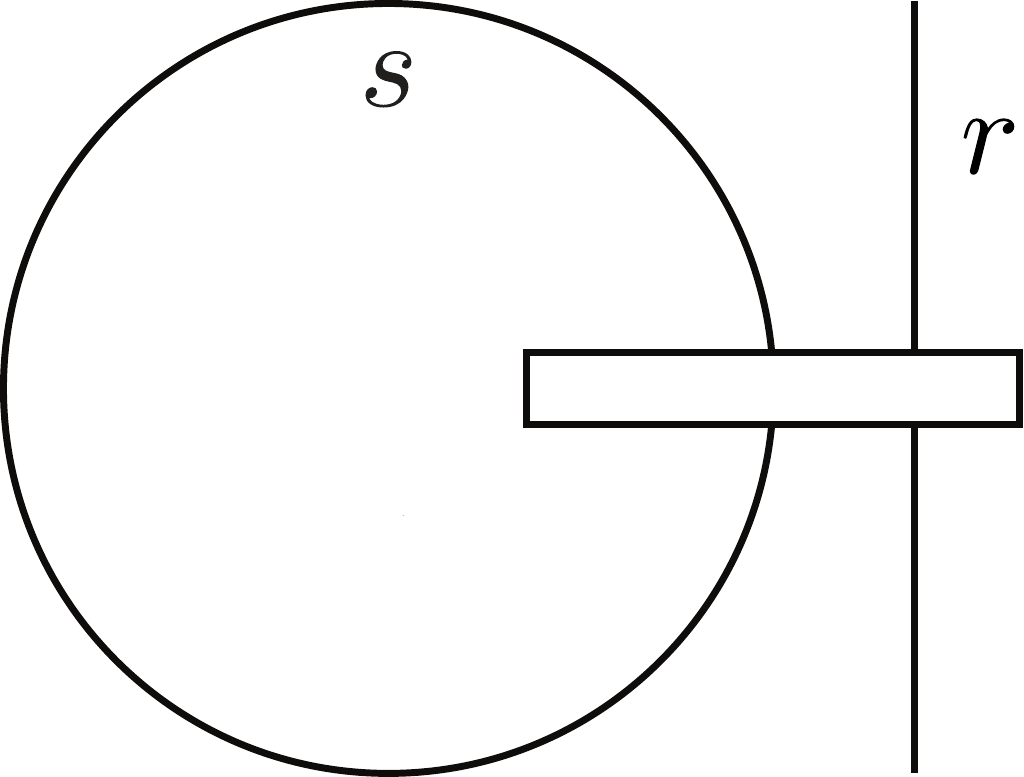}}}
\quad = \quad (-1)^s \frac{[r+s+1]}{[r+1]} 
 \,\, \times \,\,  \vcenter{\hbox{\includegraphics[scale=0.275]{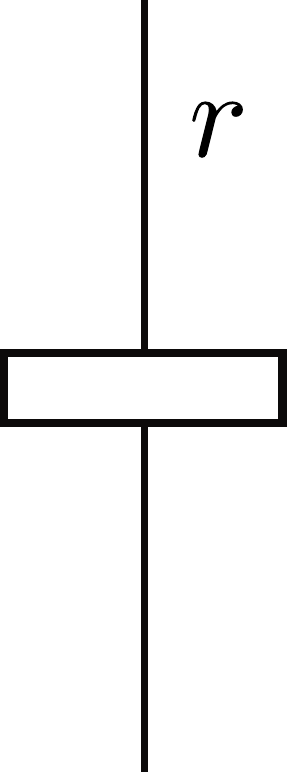} .}} 
\end{align}
\end{lem}

\begin{proof} 
We prove formula~\eqref{DeltaTangleGen} by induction on $s \in \bZnn$.
It is obvious for $s = 0$. Now, assuming that~\eqref{DeltaTangleGen} holds for all $s \leq t - 1$ for some integer $t \geq 2$, 
using the induction hypothesis first for $s=1$ and then for $s=t-1$, we get
\begin{align} 
\vcenter{\hbox{\includegraphics[scale=0.275]{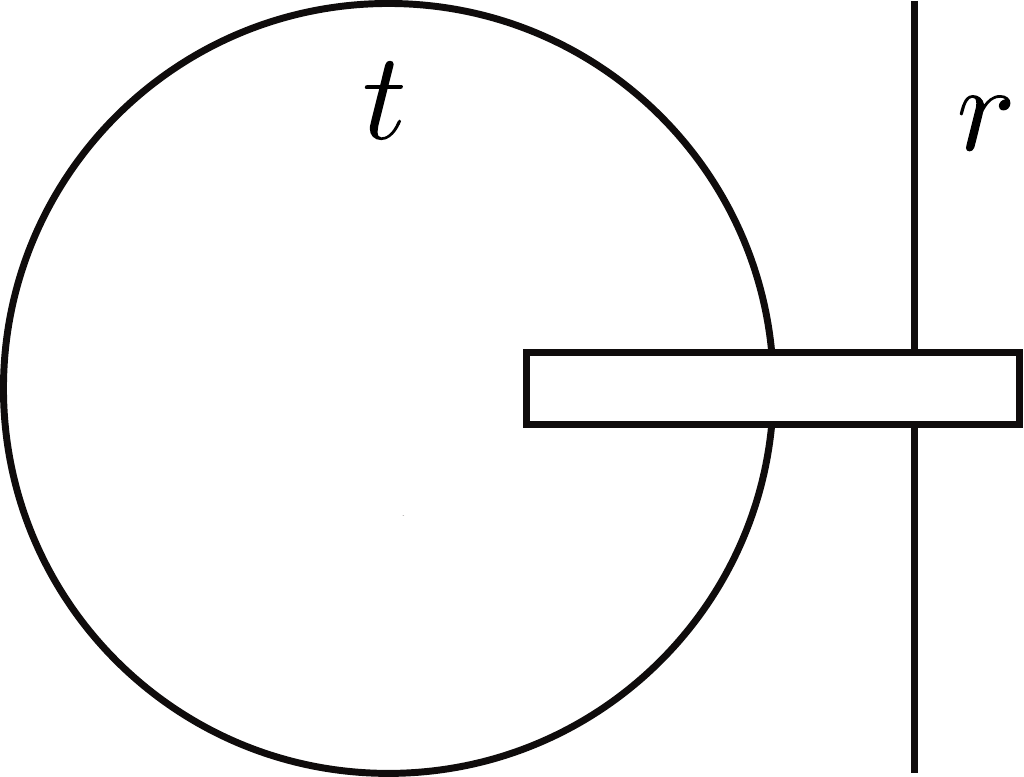}}}
\quad = \; & \quad 
\vcenter{\hbox{\includegraphics[scale=0.275]{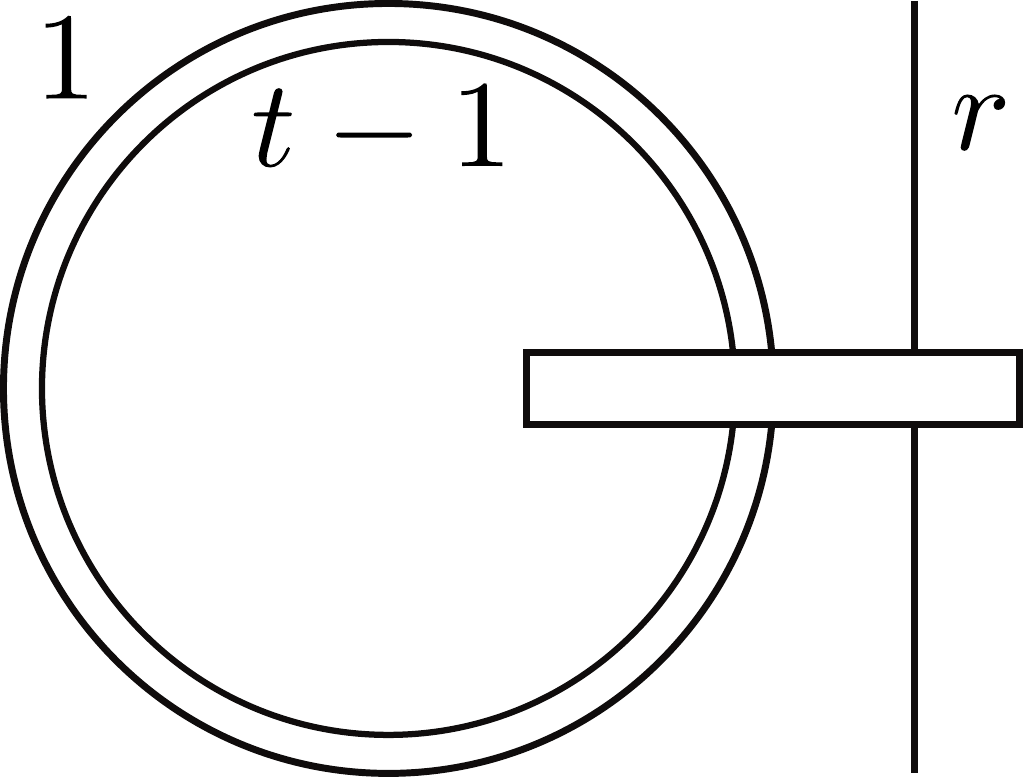}}} \quad
= \quad - \frac{[r + t + 1]}{[r + t]}  \,\, \times \,\,   \vcenter{\hbox{\includegraphics[scale=0.275]{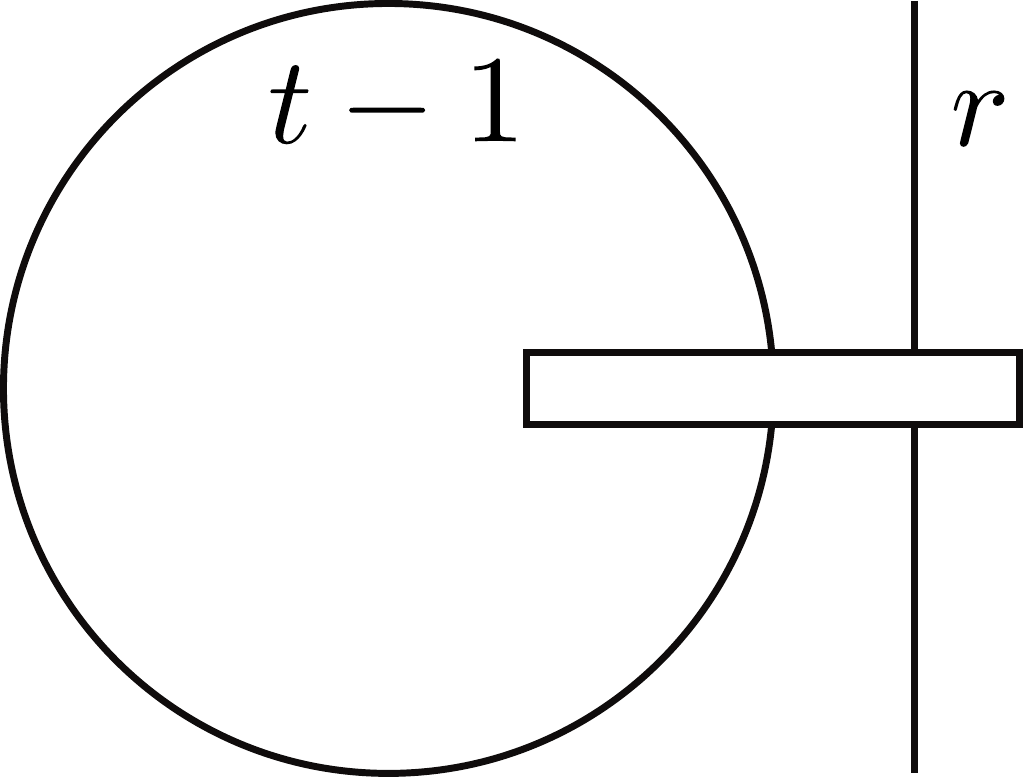}}} \\
= \; & \quad  - \frac{[r + t + 1]}{[r + t]} (-1)^{t-1} \frac{[r+t]}{[r+1]} 
 \,\, \times \,\,  \vcenter{\hbox{\includegraphics[scale=0.275]{e-ProjectorBox_tall_r.pdf}}} \quad
= \quad (-1)^t \frac{[r+t+1]}{[r+1]}  \,\, \times \,\,  \vcenter{\hbox{\includegraphics[scale=0.275]{e-ProjectorBox_tall_r.pdf} .}}
\end{align}
This proves~\eqref{DeltaTangleGen} with $s = t$, finishing the induction step.
\end{proof}

We use the next simple observation in section~\ref{GramMatrixSect} and later in this appendix.
\begin{lem} \label{InsProjBoxLem}  
We may insert a projector box above and/or below any network as follows: 
\begin{align}\label{InsProjBox} 
\vcenter{\hbox{\includegraphics[scale=0.275]{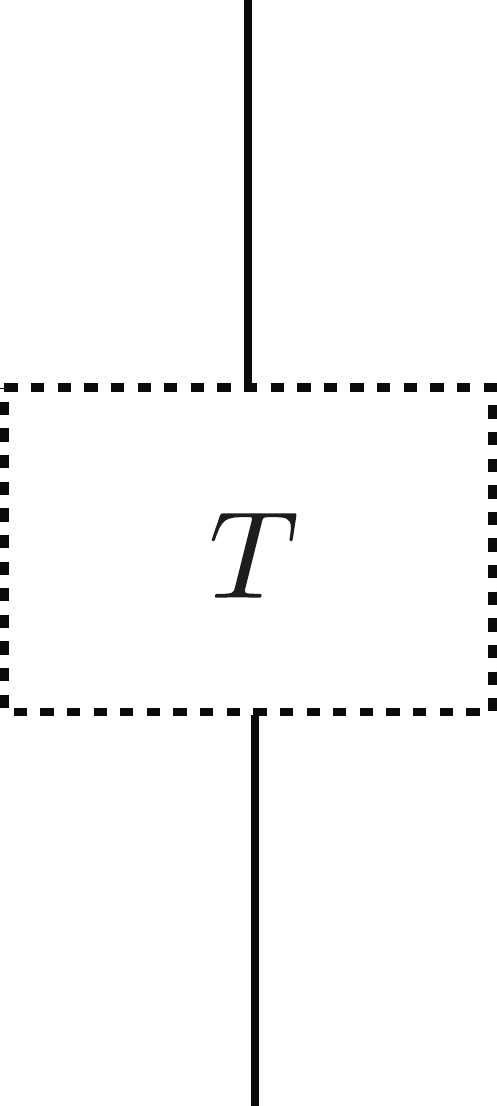}}} 
\quad =  \quad \vcenter{\hbox{\includegraphics[scale=0.275]{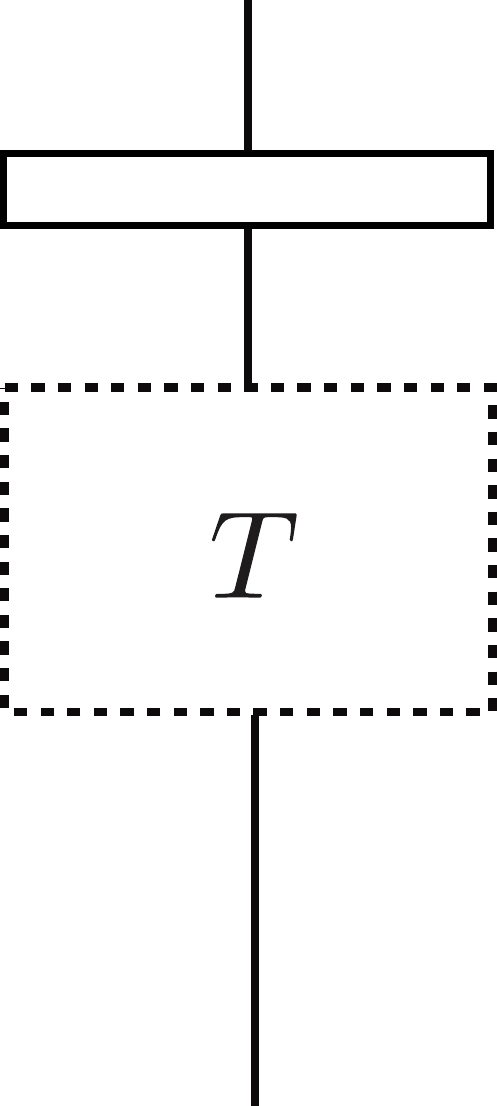}}}
\quad =  \quad \vcenter{\hbox{\includegraphics[scale=0.275]{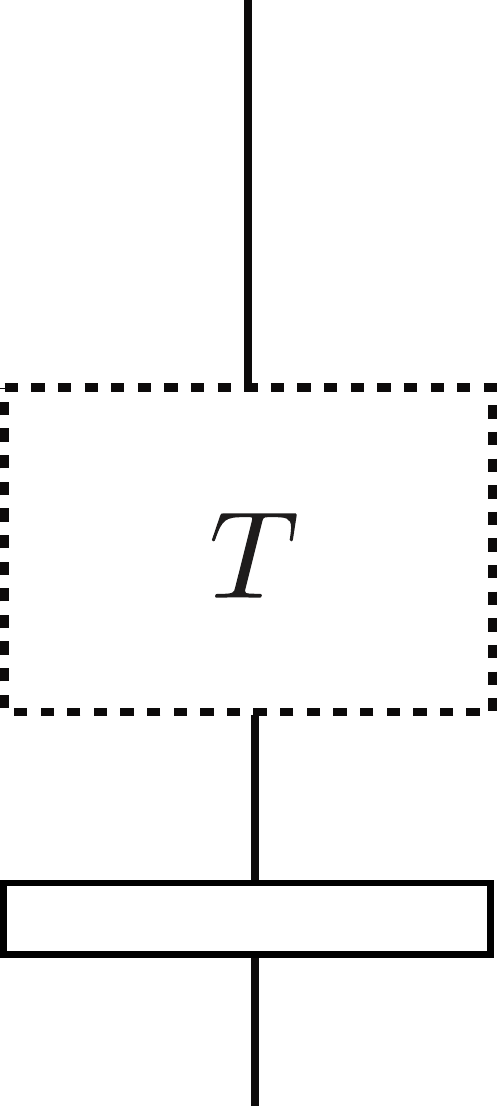}}}
\quad =  \quad \vcenter{\hbox{\includegraphics[scale=0.275]{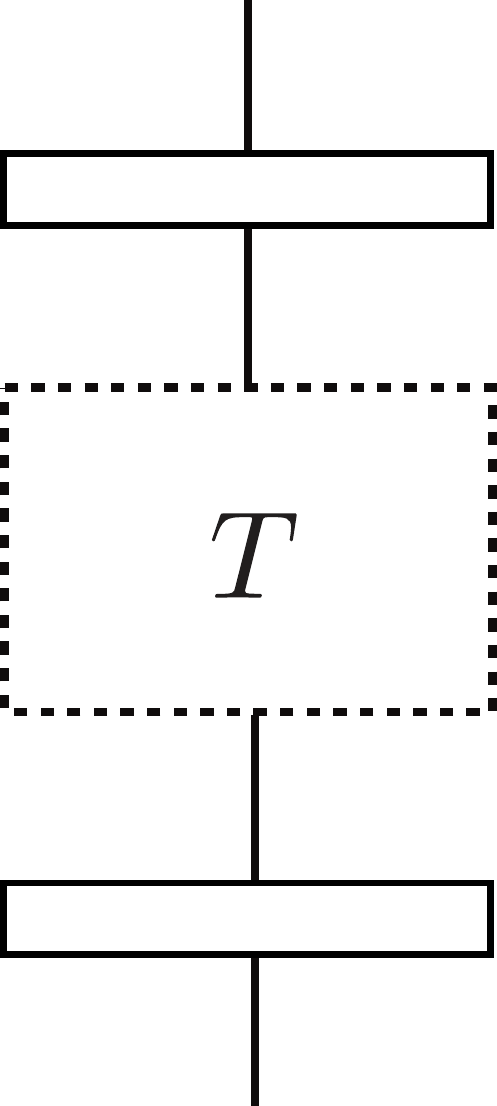} .}} 
\end{align} 
In particular, we have
\begin{align} \label{WeightOne}
\left( \; \vcenter{\hbox{\includegraphics[scale=0.275]{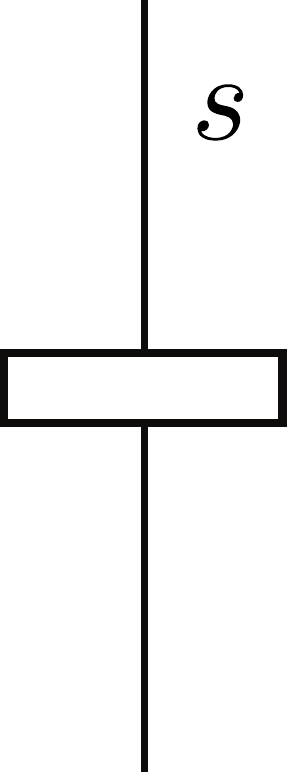}}} \; \right) 
\quad = \quad 1 . 
\end{align} 
\end{lem} 

\begin{proof}  
On the right side of~\eqref{InsProjBox}, each internal link diagram of the upper (resp.~lower) projector box with a turn-back link 
has weight zero by~\eqref{TurnBack0}. Thus, only the unit link diagram~\eqref{Units} contributes,  
and replacing the projector box with only this diagram is the same as removing the box altogether.  
Identity~\eqref{WeightOne} then follows from~\eqref{ThroughPathWeight}. 
\end{proof}

We remark that because only the unit link diagram~\eqref{Units} 
contributes to the right side of~\eqref{InsProjBox} and its coefficient equals one in~\eqref{ProjDecomp}, 
we do not need to restrict the size of the inserted projector boxes to less than $\ppmin(q)$.  
Therefore, we need not include this condition in the statement of lemma~\ref{InsProjBoxLem}.

We use the following network evaluation rule in section~\ref{StdModulesSect} and later in this appendix.

\begin{lem} \label{TieOffLem} 
Suppose $s < \ppmin(q)$.
Let $T$ be a network with $s$ links passing through the top side of the rectangle and $s$ links passing through the bottom side.  
Then the evaluations of the following networks are equal:
\begin{align}\label{ThetaExtraction3} 
(-1)^s [s+1] \quad \vcenter{\hbox{\includegraphics[scale=0.275]{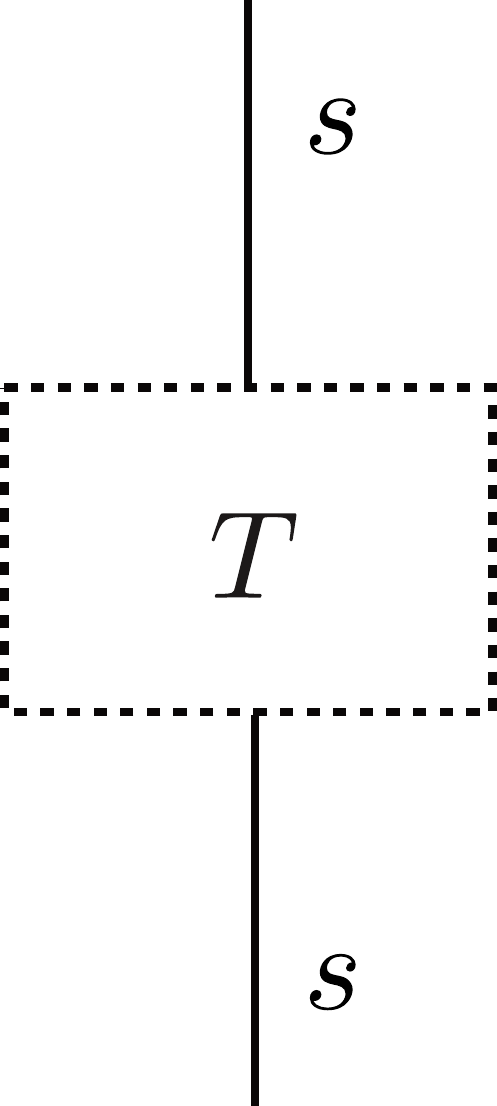}}} 
\qquad \qquad  \textnormal{and} \qquad \qquad
\vcenter{\hbox{\includegraphics[scale=0.275]{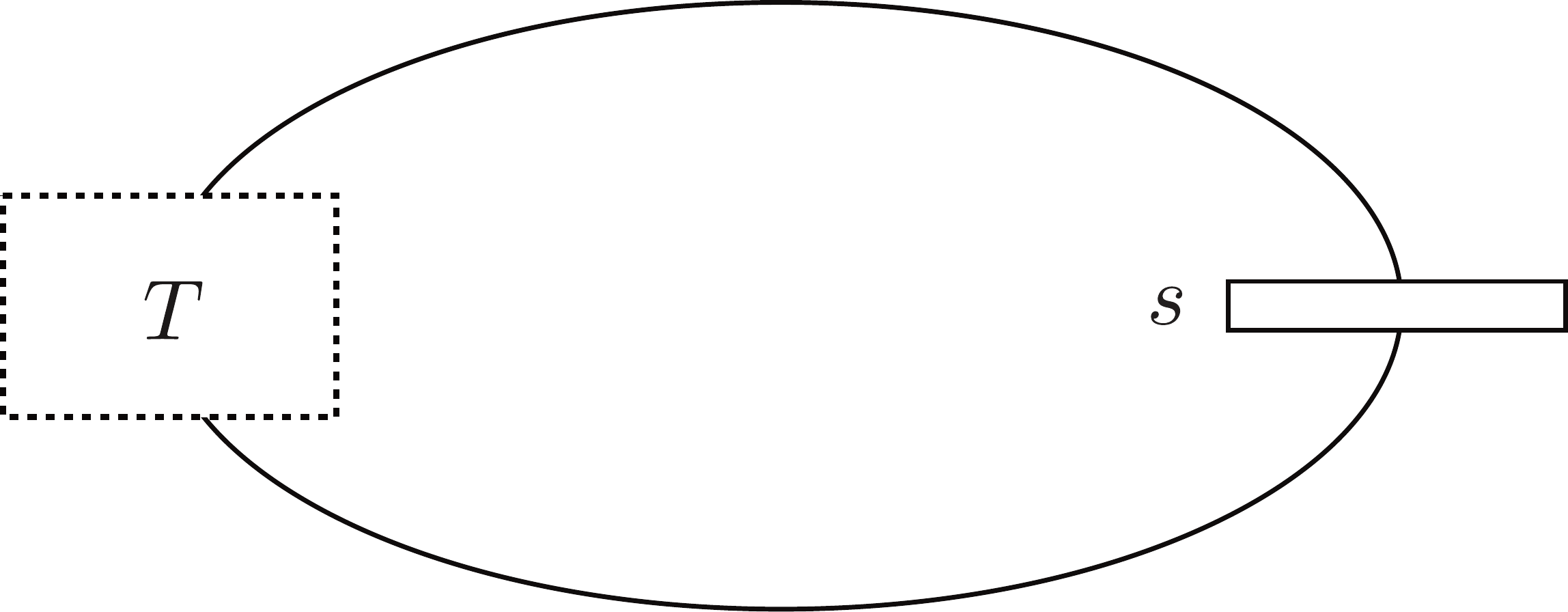} .}} 
\end{align} 
\end{lem}
\begin{proof} 
Assuming without loss of generality that $T$ is a link diagram, there are two scenarios to consider:
\begin{enumerate}[leftmargin=*]
\itemcolor{red}
\item \label{sce1} $T$ contains a turn-back link.
Then, the left side of~\eqref{ThetaExtraction3} contains a turn-back link, which vanishes by rule~\eqref{TurnBack0},
and the right side of~\eqref{ThetaExtraction3} has a link attached to two nodes of the 
projector box, which also vanishes by property~\eqref{ProjectorID2}.

\item \label{sce2} $T$ does not contain a turn-back link.
Then, all of the $s$ links passing through the top and bottom side of $T$ are through-links.
We recall from~\eqref{ThroughPathWeight} that through-links have weight one.
Thus, lemma~\ref{LoopLemGen} with $r=0$ shows that the evaluations of the link diagrams in~\eqref{ThetaExtraction3} are equal.
\end{enumerate}
This concludes the proof.
\end{proof}

The next extraction rule, even though simple, is very useful in sections~\ref{GramMatrixSect} and~\ref{RadicalSect}. 

\begin{lem} \label{ExtractLem} 
Suppose $s < \ppmin(q)$.  
Then, for any network $T$ contained between two projector boxes of size $s$
within a larger network, we have the following extraction rule:
\begin{align} \label{ExtractID} 
\vcenter{\hbox{\includegraphics[scale=0.275]{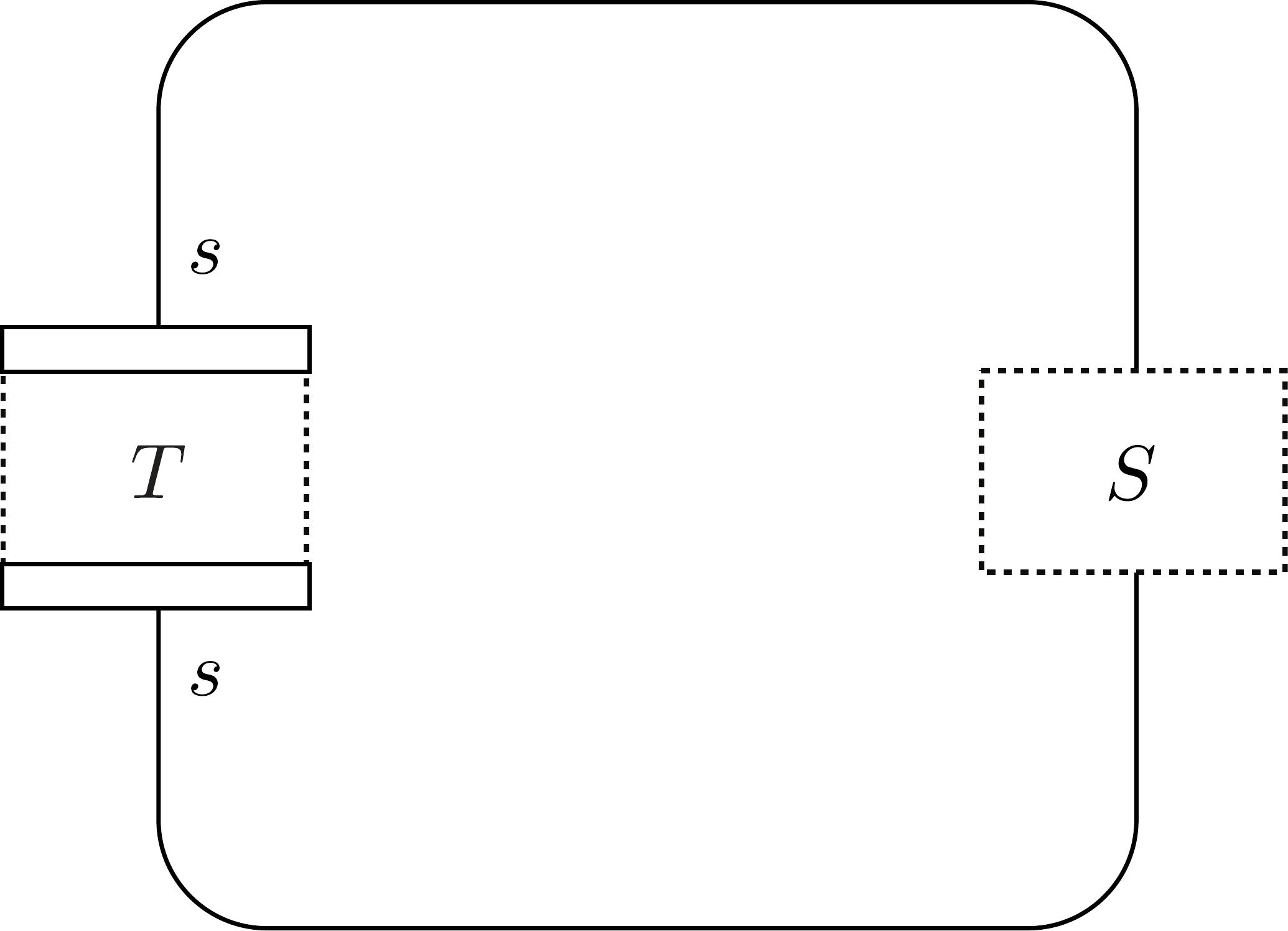}}} 
\quad = \quad (T) \,\, \times \,\, \vcenter{\hbox{\includegraphics[scale=0.275]{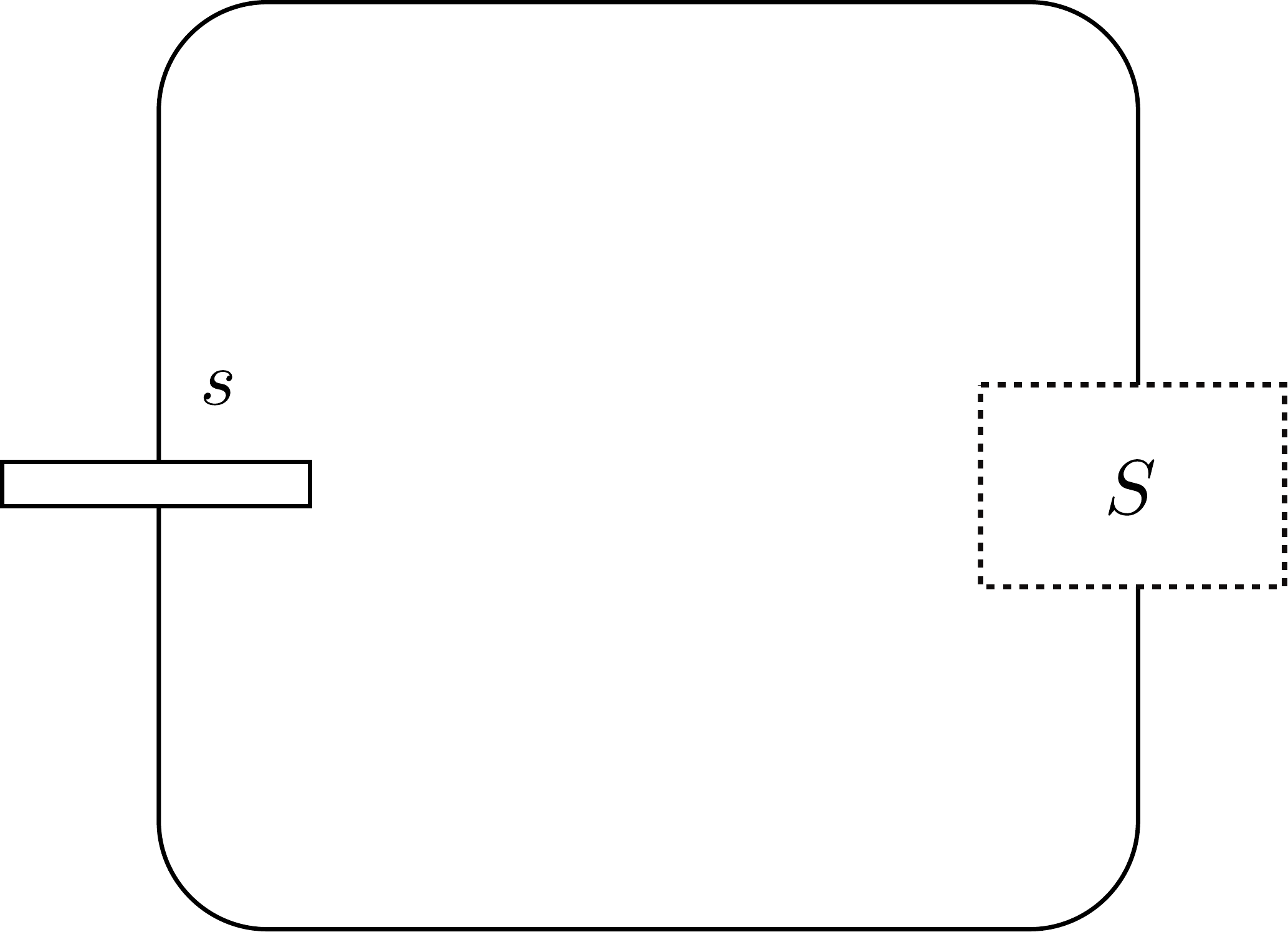} .}}
\end{align} 
\end{lem} 

\begin{proof} 
We consider two cases:
\begin{enumerate}[leftmargin=*]
\itemcolor{red}
\item $T$ contains a turn-back link.
Then, the left side of~\eqref{ExtractID} vanishes by property~\eqref{ProjectorID2}, and the right side of~\eqref{ExtractID} vanishes 
because $(T) = 0$ by rule~\eqref{TurnBack0}.  Hence, equality in~\eqref{ExtractID} holds for this case, with both sides equaling zero.  

\item $T$ does not contain a turn-back link.
Then, the network $T$ comprises only through-links and a number $k$ of loops. After replacing each loop by a factor $\nu$ on the left side of~\eqref{ExtractID}, we obtain
\begin{align} \label{ExtractProof} 
\vcenter{\hbox{\includegraphics[scale=0.275]{e-LoopNetwork2.pdf}}} 
\quad = \quad \nu^k \,\, \times \,\, \vcenter{\hbox{\includegraphics[scale=0.275]{e-LoopNetwork1.pdf} .}}
\end{align} 
However, by~\eqref{evT2}, the factor $\nu^k$ on the right side equals the evaluation $(T)$ of $T$,
so the right side of~\eqref{ExtractProof} equals the right side of~\eqref{ExtractID}.  
Hence,~\eqref{ExtractID} holds also for this case.
\end{enumerate}
This concludes the proof. 
\end{proof}

We define the \emph{Theta network}~\cite{kl} to be the tangle
\begin{align} \label{ThetaDefinition} 
\vcenter{\hbox{\includegraphics[scale=0.275]{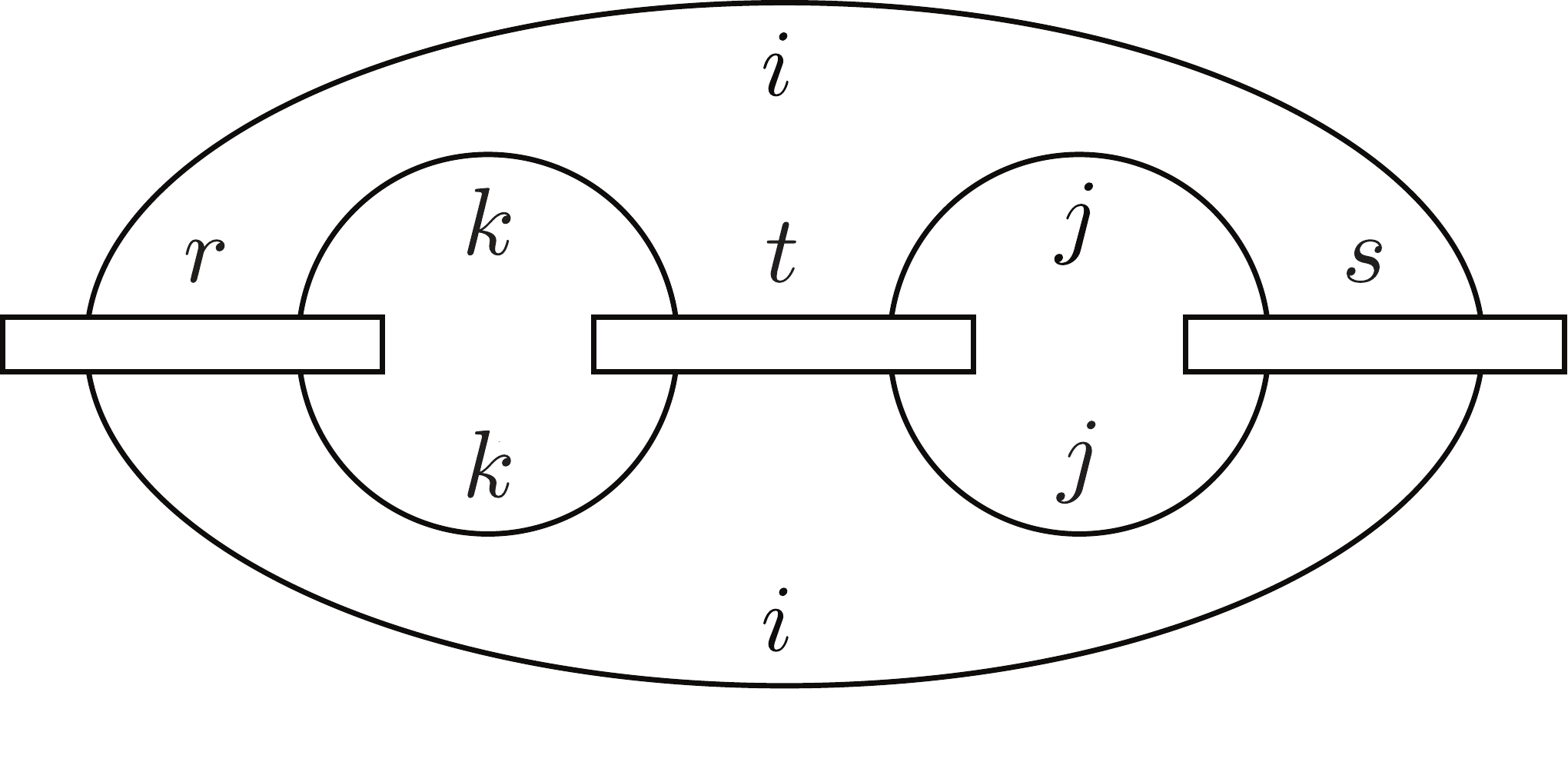}}} \quad 
& = \quad \vcenter{\hbox{\includegraphics[scale=0.275]{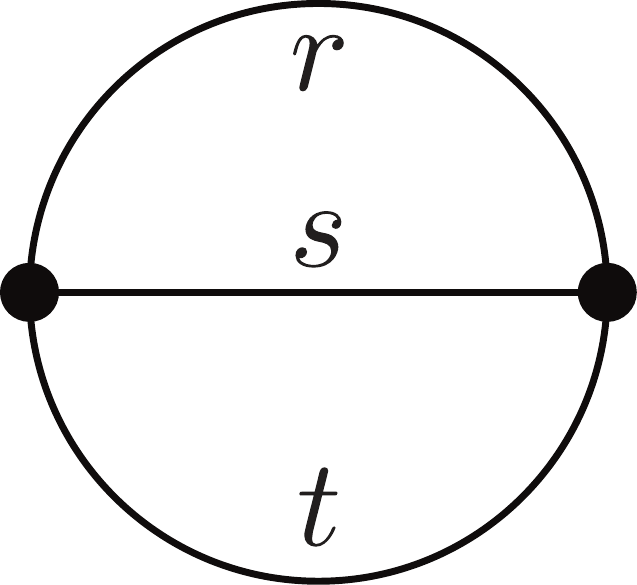} ,}} \qquad \qquad
\begin{aligned} 
i & = \frac{r + s - t}{2} , \\[.7em] 
j & = \frac{s + t - r}{2} , \\[.7em] 
k & = \frac{t + r - s}{2} .
\end{aligned}
\end{align} 
We denote the evaluation of the Theta network by $\ThetaNet(r,s,t)$.

Together with lemma~\ref{ExtractLem}, the following lemma~\ref{LoopErasureLem} 
is a crucial tool in section~\ref{GramMatrixSect} and, in particular, in section~\ref{RadicalSect}.

\begin{lem} \label{LoopErasureLem} 
Suppose $\max(r,s,s',t) < \ppmin(q)$.  Then we have 
\begin{align}\label{LoopErasure1} 
 \left( \; \vcenter{\hbox{\includegraphics[scale=0.275]{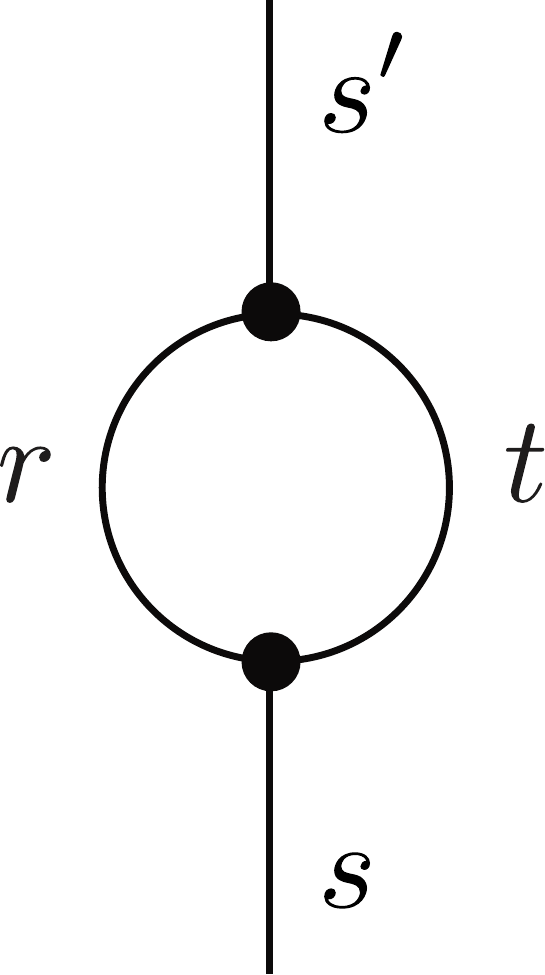}}} \;  \right) \quad 
= \quad \delta_{s, s'} \frac{ \ThetaNet(r,s,t) }{(-1)^s [s+1]} . 
\end{align} 
\end{lem} 

\begin{proof} 
First, we show that if $s \neq s'$ in~\eqref{LoopErasure1}, then both sides vanish. The right side of~\eqref{LoopErasure1} is clearly zero if $s \neq s'$.
Also, if $s' < s$ (resp.~$s < s'$), then on the left side of~\eqref{LoopErasure1}, a turn-back link touches two nodes of the projector box with size $s$ 
in the lower (resp.~upper) vertex, so the network vanishes by property~\eqref{ProjectorID2}.

Thus, we may assume $s = s'$. Then, after substituting~\eqref{3vertex1}, 
simplifying via~\eqref{ProjectorID0}, and applying lemmas~\ref{InsProjBoxLem} and~\ref{TieOffLem}, 
the network on the left side of~\eqref{LoopErasure1} becomes
\begin{align} \label{LoopErasure3}
\vcenter{\hbox{\includegraphics[scale=0.275]{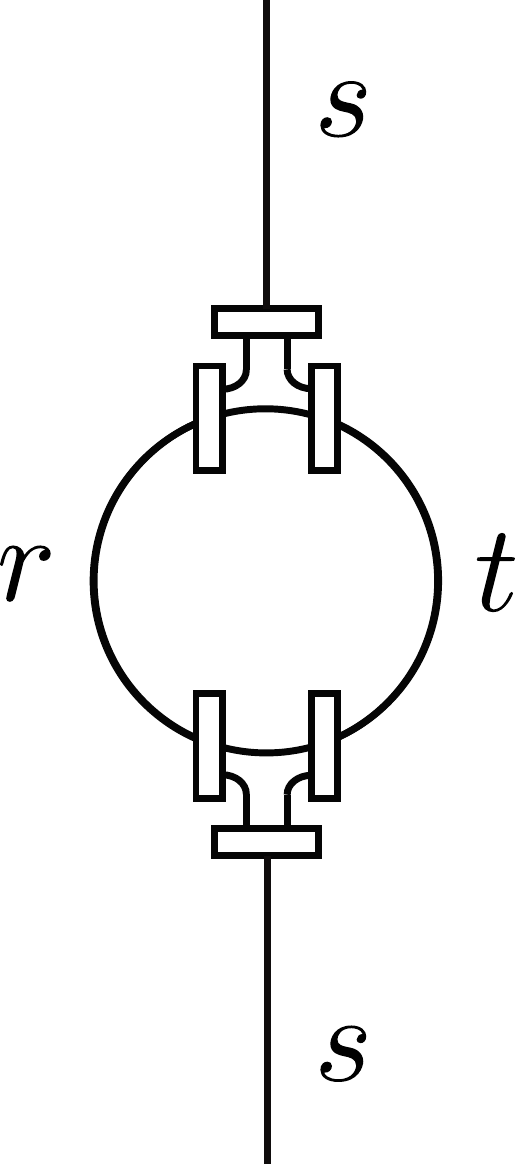}}} \quad 
\underset{\eqref{ProjectorID0}}{\overset{\eqref{InsProjBox}}{=}}
\quad \vcenter{\hbox{\includegraphics[scale=0.275]{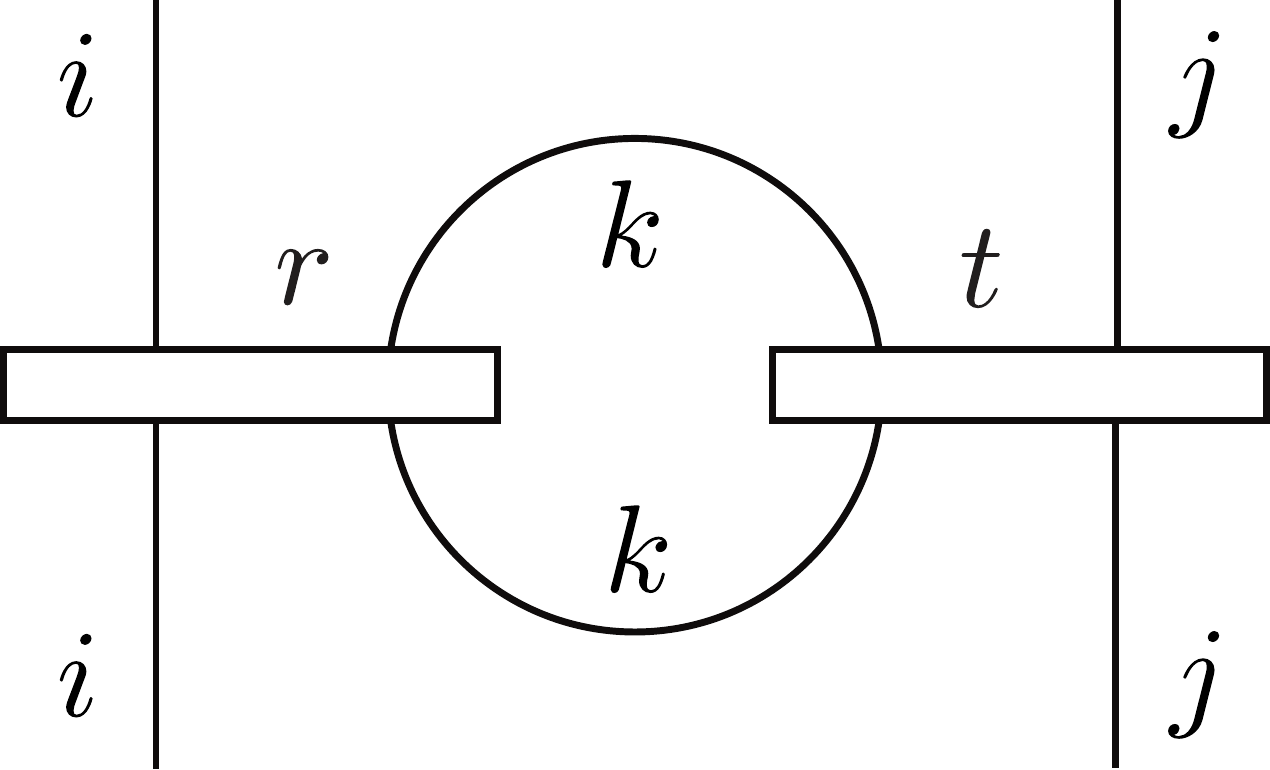}}} \quad
\underset{\eqref{ThetaDefinition}}{\overset{\eqref{ThetaExtraction3}}{=}}
\quad \frac{1}{(-1)^s [s+1]} \,\, \times \,\, 
\vcenter{\hbox{\includegraphics[scale=0.275]{e-ThetaNet2.pdf}}} \quad 
\vcenter{\hbox{\includegraphics[scale=0.275]{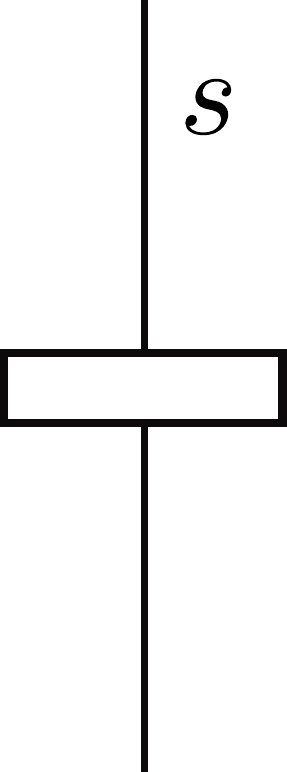} .}} 
\end{align} 
Taking evaluations of both sides and recalling identity~\eqref{WeightOne} from lemma~\ref{InsProjBoxLem} finishes the proof.
\end{proof}

We find an explicit formula for the evaluation of the Theta network in lemma~\ref{ThetaLem}.

\begin{lem} \label{QintIDLemAndSumFormulaLem2}
\
\begin{enumerate}
\itemcolor{red}
\item \label{QintIDItem}
The following identity holds, for all $i,j,k \in \bZ$:
\begin{align} \label{QintID} 
[i] [ j - k]+[j] [ k - i ]+[k] [ i - j ] = 0. 
\end{align} 

\item \label{SumFormulaItem}
The following identity holds, for all $i,k \in \bZnn$, and $j \in \bZ$: 
\begin{align} \label{SumFormula2}
\sum_{m = 0}^{\min(i,k)} \frac{(-1)^{m}}{[j+m+1]} \frac{[i+k-m]!}{[i-m]![k-m]![m]!} 
= \frac{[j]! [i+j+k+1]!}{[i+j+1]! [j+k+1]!} .
\end{align}
\end{enumerate}
\end{lem}
\begin{proof}
The proof of identity~\eqref{QintID} in item~\ref{QintIDItem} is a straightforward exercise, using definition~\eqref{Qinteger} of the $q$-integers.

For item~\ref{SumFormulaItem}, we first observe by a straightforward calculation using identity~\eqref{QintID} 
that both sides of~\eqref{SumFormula2} satisfy the recursion 
\begin{align} \label{RecursionForI}
[k] A_{i-1,k}^j - [i] A_{i,k-1}^j = \; & [k-i] A_{i-1,k-1}^{j+1} , \qquad A_{0,k}^j =  A_{i,0}^j = \frac{1}{[j+1]} , 
\qquad \text{for all $i,k \in \bZnn$, and $j \in \bZ$}
\end{align}
with the convention that $A_{-1,k}^j = 0 =A_{i,-1}^j$.  
It follows immediately from~\eqref{RecursionForI} that if $i=0$, then identity~\eqref{SumFormula2} holds, for all $k \in \bZnn$, and $j \in \bZ$. 
Then, assuming that identity~\eqref{SumFormula2} holds for all $k \in \bZnn$, $j \in \bZ$, and $i \in \{0, 1, \ldots, n-1\}$, for some $n \in \bZpos$,
it remains to conclude that by~\eqref{RecursionForI} and induction, it also holds when $i=n$.
\end{proof}

\begin{lem} \label{ThetaLem} 
Suppose $\max(r,s,t) < \ppmin(q)$. Then we have 
\begin{align} \label{ThetaFormula1}
\ThetaNet(r,s,t)
= \frac{(-1)^{\frac{r + s + t}{2}} \left[ \frac{r + s + t}{2} + 1 \right]! \left[ \frac{ r + s - t }{2} \right]! \left[ \frac{ s + t - r}{2} \right]! \left[ \frac{t + r - s}{2} \right]! }{[ r ]! [s ]! [ t ]!} . 
\end{align}
\end{lem}

\begin{proof} 
According to 
the proof of lemma~\ref{LoopErasureLem}, we write the evaluation of the Theta network in the following form:
\begin{align} \label{ThetaPrePre}
\ThetaNet(r,s,t) \overset{\eqref{LoopErasure3}}{=}
(-1)^s [s+1] \,\, \times \,\,
\left( \;  \vcenter{\hbox{\includegraphics[scale=0.275]{e-NeatTangle10.pdf}}} \;  \right) \; ,
\end{align} 
where $i$, $j$ and $k$ are given in~\eqref{ThetaDefinition}. 
Decomposing the projector box of size $r$ as in~\eqref{ProjDecomp}
and using the formula from~\cite[proposition~\red{A.9}]{fp0}
for the coefficients of this decomposition gives the sum formula
\begin{align} \label{ThetaPreSum}
\eqref{ThetaPrePre} = 
(-1)^{i+j} [i+j+1] \frac{[i]![k]!}{[i+k]!} \sum_{m = 0}^{\min(i,k)} \frac{[i+k-m]!}{[i-m]![k-m]![m]!} \,\, \times \,\,
 \left( \; \vcenter{\hbox{\includegraphics[scale=0.275]{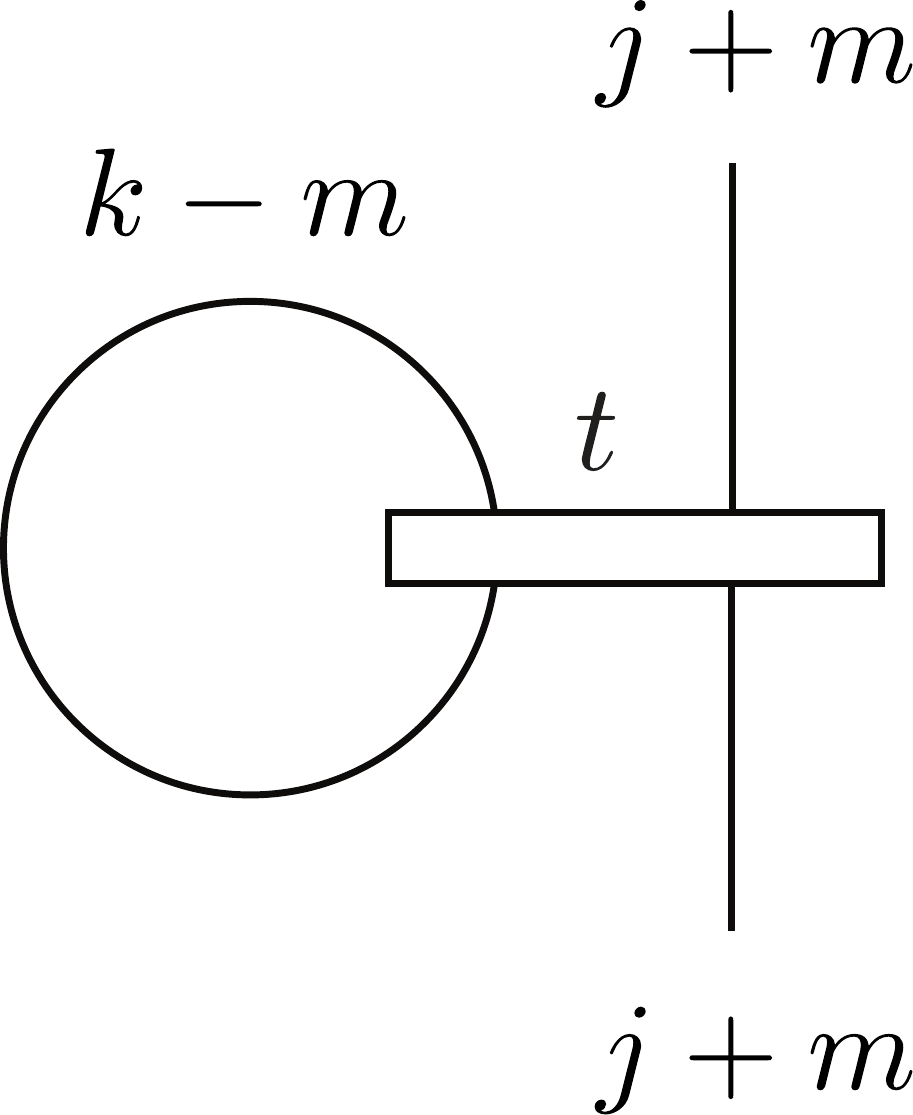}}} \; \; \right) \; .
\end{align} 
Now, we evaluate the networks on the right side of~\eqref{ThetaPreSum} using lemmas~\ref{LoopLemGen} and~\ref{InsProjBoxLem}:
\begin{align} \label{ThetaSum}
 \left( \; \vcenter{\hbox{\includegraphics[scale=0.275]{e-NeatTangle11.pdf}}} \; \; \right) \quad = \quad
\frac{(-1)^{k-m} [j+k+1]}{[j+m+1]} .
\end{align} 
Inserting~\eqref{ThetaSum} into~\eqref{ThetaPreSum}, we obtain
\begin{align}
\ThetaNet(r,s,t) \overset{(\textnormal{\ref{ThetaPrePre}--\ref{ThetaSum}})}{=}
(-1)^{i+j+k} [i+j+1] [j+k+1] \frac{[i]![k]!}{[i+k]!} \sum_{m = 0}^{\min(i,k)} \frac{[i+k-m]!}{[i-m]![k-m]![m]!} \frac{(-1)^{m}}{[j+m+1]} .
\end{align} 
Using identity~\eqref{SumFormula2} from lemma~\ref{QintIDLemAndSumFormulaLem2} 
and plugging in the values of $i$, $j$ and $k$ from~\eqref{ThetaDefinition} 
gives formula~\eqref{ThetaFormula1}.
\end{proof}

\section{Jones-Wenzl algebra} \label{AppWJ}

In this appendix, we detail the relationship of the valenced Temperley-Lieb algebra $\TL_\multii(\nu)$ 
with a certain subalgebra of the Temperley-Lieb algebra $\TL_{\Summed_\multii}(\nu)$, that we call the ``Jones-Wenzl algebra.''
In particular, we show in corollary~\ref{AnIsoCor} that $\TL_\multii(\nu)$ is isomorphic to this subalgebra.
Throughout, we assume that $\max (\multii) < \ppmin(q)$.

To begin, we define the \emph{Jones-Wenzl algebra} $\WJ_\multii(\nu)$ to be
\begin{align} \label{WJAdef}
\WJ_\multii(\nu) := \WJProj_\multii \TL_{\Summed_\multii}(\nu) \WJProj_\multii 
= \big\{ \WJProj_\multii T \WJProj_\multii \,|\, T \in \TL_{\Summed_\multii}(\nu) \big\} .
\end{align}
using the Jones-Wenzl composite projector from~\eqref{WJCompProj},
\begin{align} \label{WJCompProj00} 
\WJProj_\multii \quad := \quad \vcenter{\hbox{\includegraphics[scale=0.275]{e-CompositeProjector.pdf} .}} 
\hphantom{\WJProj_\multii \quad := \quad}
\end{align} 
In other words, $\WJ_\multii(\nu)$ is the collection of all tangles in $\TL_{\Summed_\multii}(\nu)$ of the form
\begin{align} \label{JWform2-1} 
\vcenter{\hbox{\includegraphics[scale=0.275]{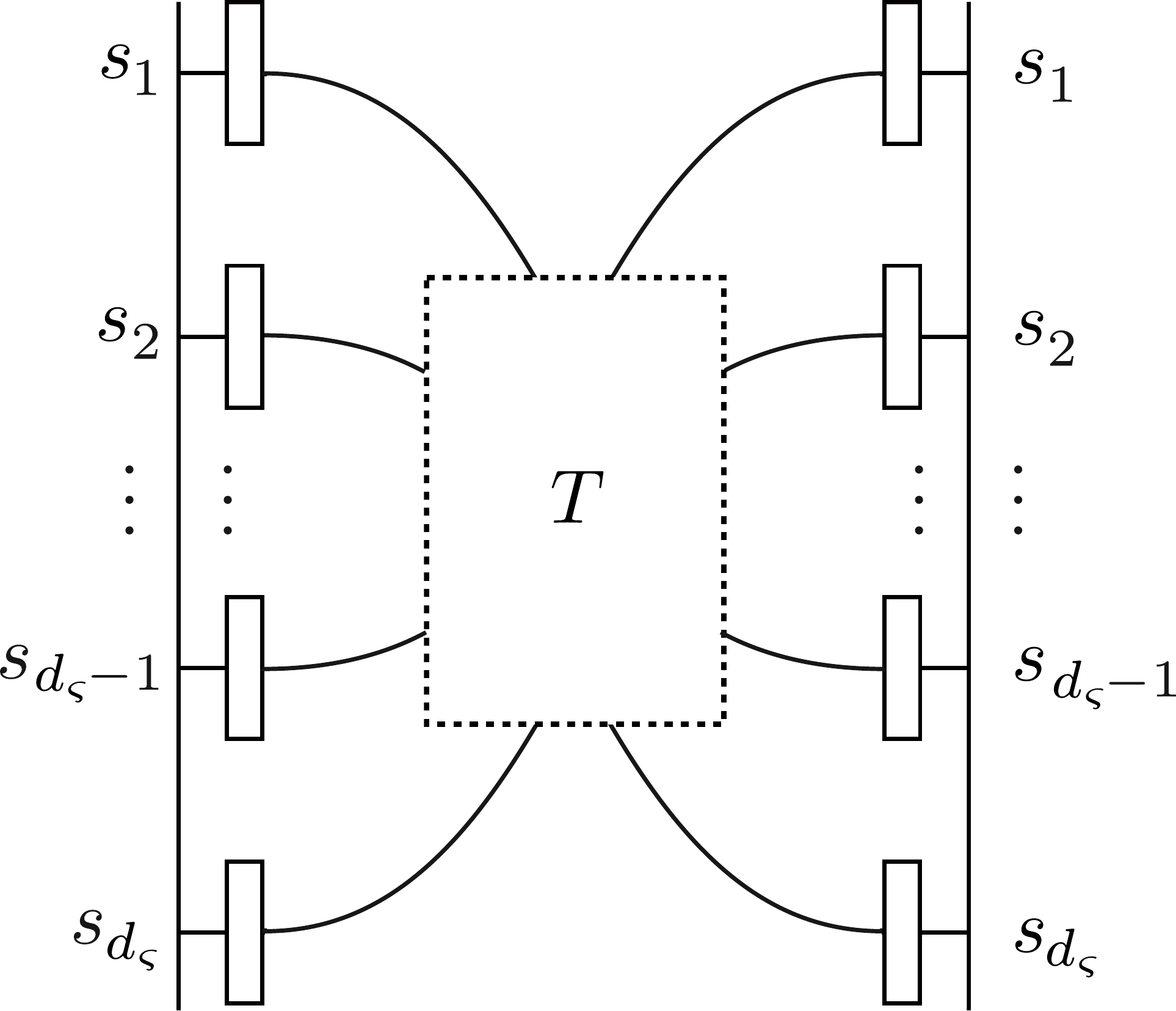} ,}}
\end{align} 
where $T \in \TL_{\Summed_\multii}(\nu)$. 
We note that 
by property~\eqref{ProjectorID2}, those tangles~\eqref{JWform2-1} that have a link with both endpoints at the same projector box are zero.
We call an element of $\WJ_\multii(\nu)$ a \emph{$\multii$-Jones-Wenzl tangle}. 
If $T$ is an $\Summed_\multii$-link diagram such that~\eqref{JWform2-1} does not vanish, then we call~\eqref{JWform2-1} 
a \emph{$\multii$-Jones-Wenzl link diagram.}
For example,
\begin{align}
\vcenter{\hbox{\includegraphics[scale=0.275]{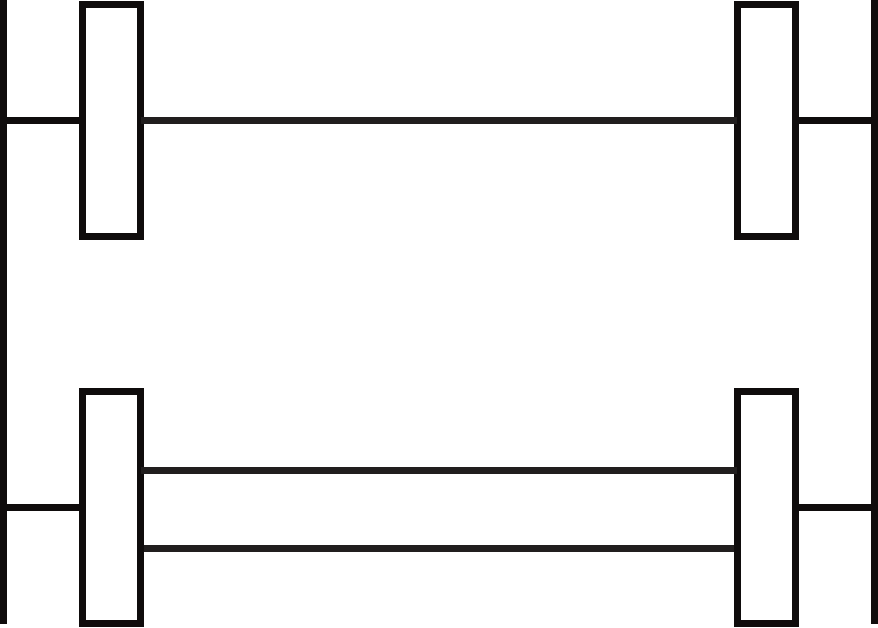}}} \qquad \qquad \text{and} \qquad \qquad
\vcenter{\hbox{\includegraphics[scale=0.275]{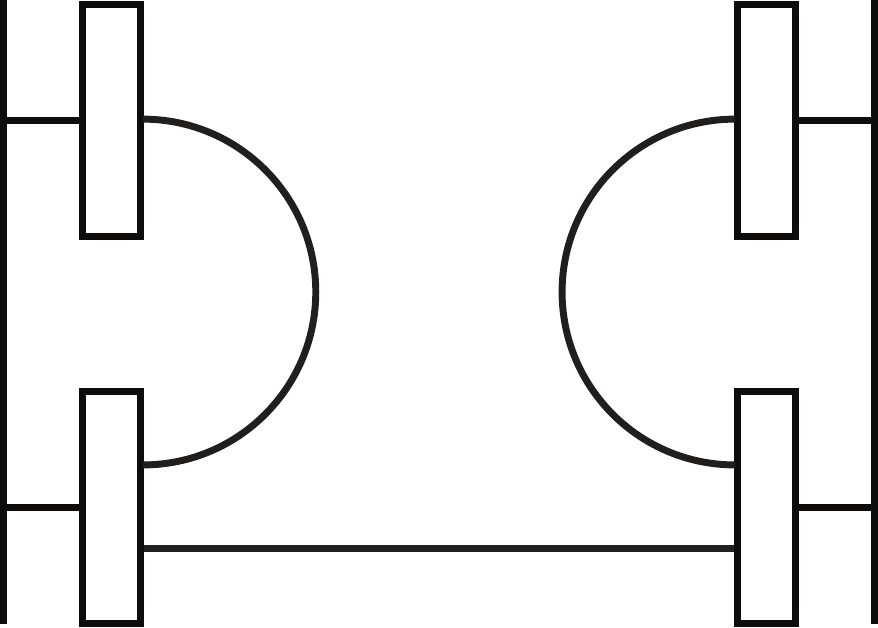}}} 
\end{align} 
are $(1,2)$-Jones-Wenzl link diagrams, and the following $(1,2)$-Jones-Wenzl tangle 
is not a Jones-Wenzl link diagram because it vanishes by property~\eqref{ProjectorID2}:
\begin{align}
\vcenter{\hbox{\includegraphics[scale=0.275]{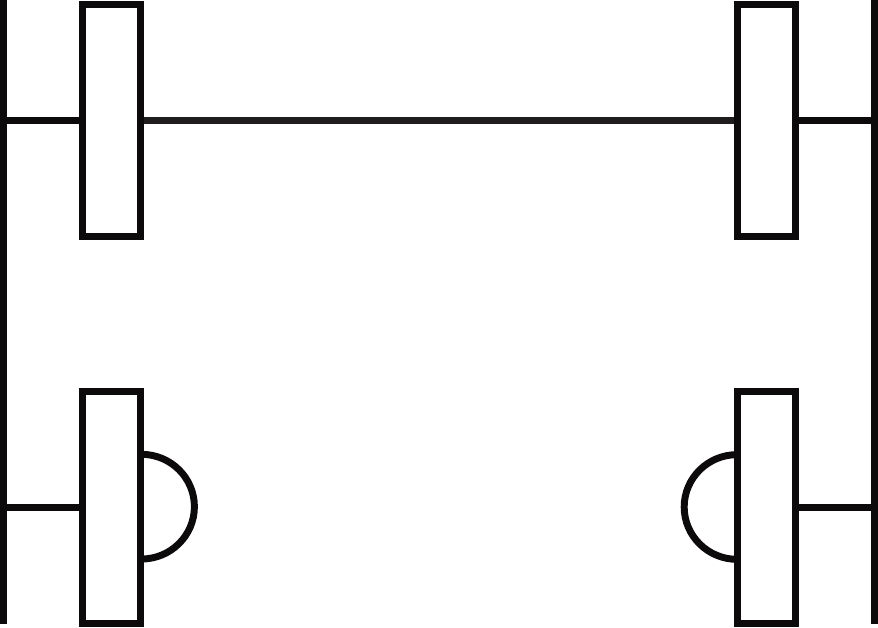}}} \quad \overset{\eqref{ProjectorID2}}{=} \quad 0 .
\end{align} 
By definition, the set of all $\multii$-Jones-Wenzl link diagrams 
forms a spanning set for the Jones-Wenzl algebra $\WJ_\multii(\nu)$,
and in fact, items~\ref{wjIt3}--\ref{wjIt4} of lemma~\ref{WJLSBasisLem} 
below imply that this spanning set is also a basis for $\WJ_\multii(\nu)$.


The Jones-Wenzl algebra is a unital, associative algebra: indeed, 
it inherits the associative multiplication from the Temperley-Lieb algebra $\TL_{\Summed_\multii}(\nu)$, 
and property~\eqref{ProjectorID0} of the Jones-Wenzl projectors implies that $\WJProj_\multii$ is its unit: 
\begin{align} \label{WJunit}
\WJProj_\multii T = T \WJProj_\multii = T ,
\end{align}
for all tangles $T \in \TL_{\Summed_\multii}(\nu)$.  

As an analogue of proposition~\ref{GeneratorPropTwo}, we prove 
in~\cite[theorem~\red{1.1}, item~\red{1}]{fp0} that, 
when $\Summed_\multii < \ppmin(q)$, 
the algebra $\WJ_\multii(\nu)$ is generated by its unit~\eqref{WJCompProj00} together with the following $\multii$-Jones-Wenzl link diagrams:
\begin{align} \label{EGenerators} 
\WJProj_\multii \Gen_{\sIndex_1 + \sIndex_2 + \dotsm + \sIndex_i} \WJProj_\multii 
\quad = \quad \vcenter{\hbox{\includegraphics[scale=0.275]{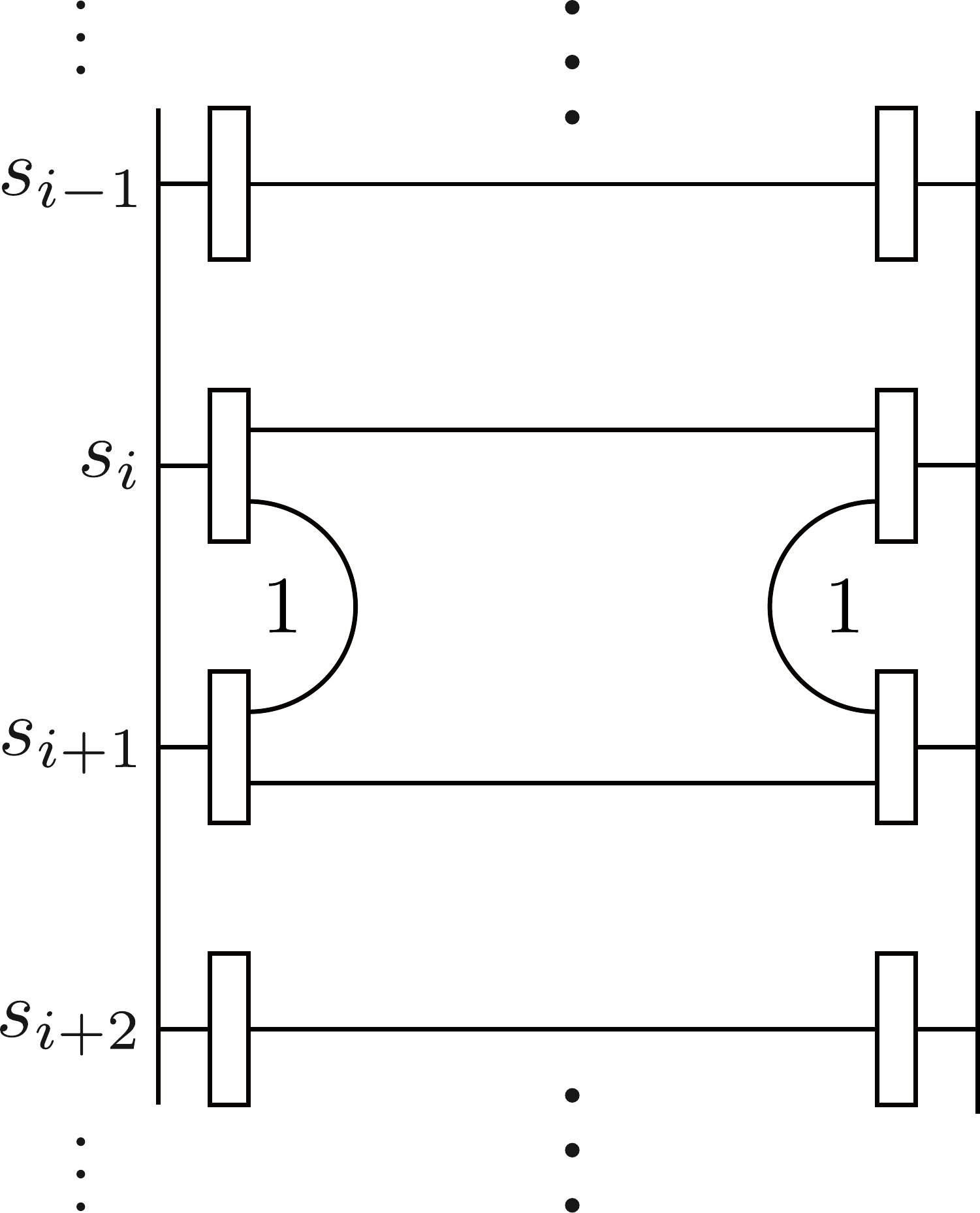} ,}}
\hphantom{\WJProj_\multii \Gen_{\sIndex_1 + \sIndex_2 + \dotsm + \sIndex_i} \WJProj_\multii 
\quad = \quad}
\end{align}
with $i \in \{1,2,\ldots,\np_\multii-1\}$, and this is a minimal generating set.
In fact, we expect this fact to hold whenever $\max (\multii) < \ppmin(q)$~\cite[conjecture~\red{1.2}]{fp0}.
We also prove in~\cite[theorem~\red{1.1}, item~\red{2}]{fp0} that, 
when $\Summed_\multii < \ppmin(q)$, all $\multii$-Jones-Wenzl tangles of the form
\begin{align} \label{MasterDiagramsWJ} 
\vcenter{\hbox{\includegraphics[scale=0.275]{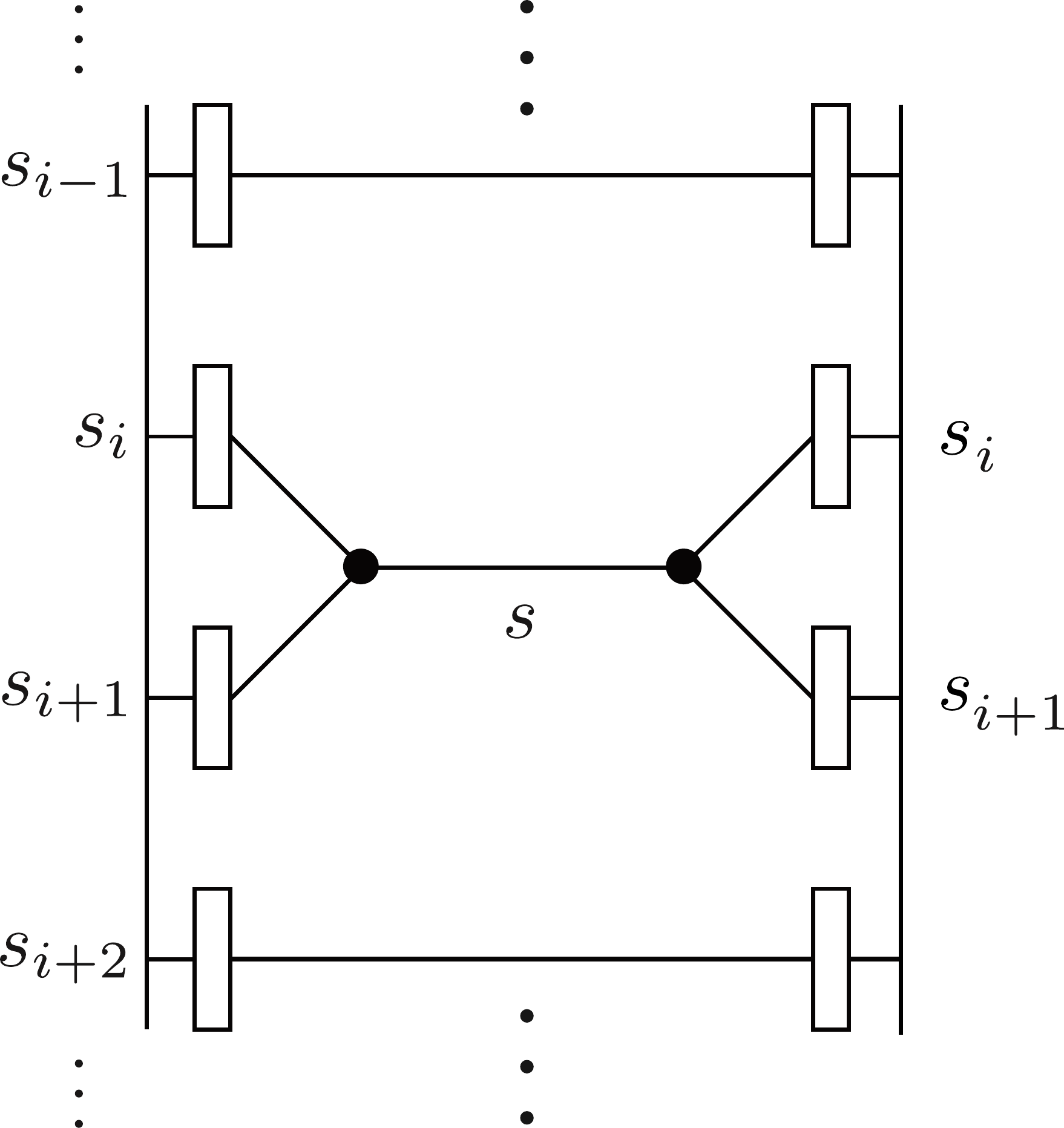} ,}}
\end{align}
with $s \in \DefectSet\sub{\sIndex_i,\sIndex_{i+1}}$ and $i \in \{ 1, 2, \ldots, \np_\multii - 1 \}$, 
form an alternative minimal generating set for $\WJ_\multii(\nu)$.
Furthermore, in~\cite[section~\red{4}]{fp0} we investigate relations satisfied by these generators.

\bigskip

The representation theory of the Jones-Wenzl algebra $\WJ_\multii(\nu)$ is analogous to that of 
the valenced Temperley-Lieb algebra $\TL_\multii(\nu)$. It has \emph{(Jones-Wenzl) standard modules} 
\begin{align} \label{ProjsSpDefns}  
\PS_\multii\super{s} := 
\WJProj_\multii  \LS_{\Summed_\multii}\super{s} = \big\{ \WJProj_\multii \alpha \,|\, \alpha \in \LS_{\Summed_\multii}\super{s} \big\}, 
\end{align}
whose direct sum we call the \emph{(Jones-Wenzl) link state module},
\begin{align}
\PS_\multii &:= \WJProj_\multii  \LS_{\Summed_\multii} = \big\{ \WJProj_\multii \alpha \,|\, \alpha \in \LS_{\Summed_\multii} \big\} 
= \bigoplus_{s \, \in \, \DefectSet_\multii} \PS_\multii\super{s}.
\end{align}
We call a generic element of $\PS_\multii\super{s}$ a \emph{$(\multii,s)$-Jones-Wenzl link state}, having the form
\begin{align}\label{JWLinkState3} 
\vcenter{\hbox{\includegraphics[scale=0.275]{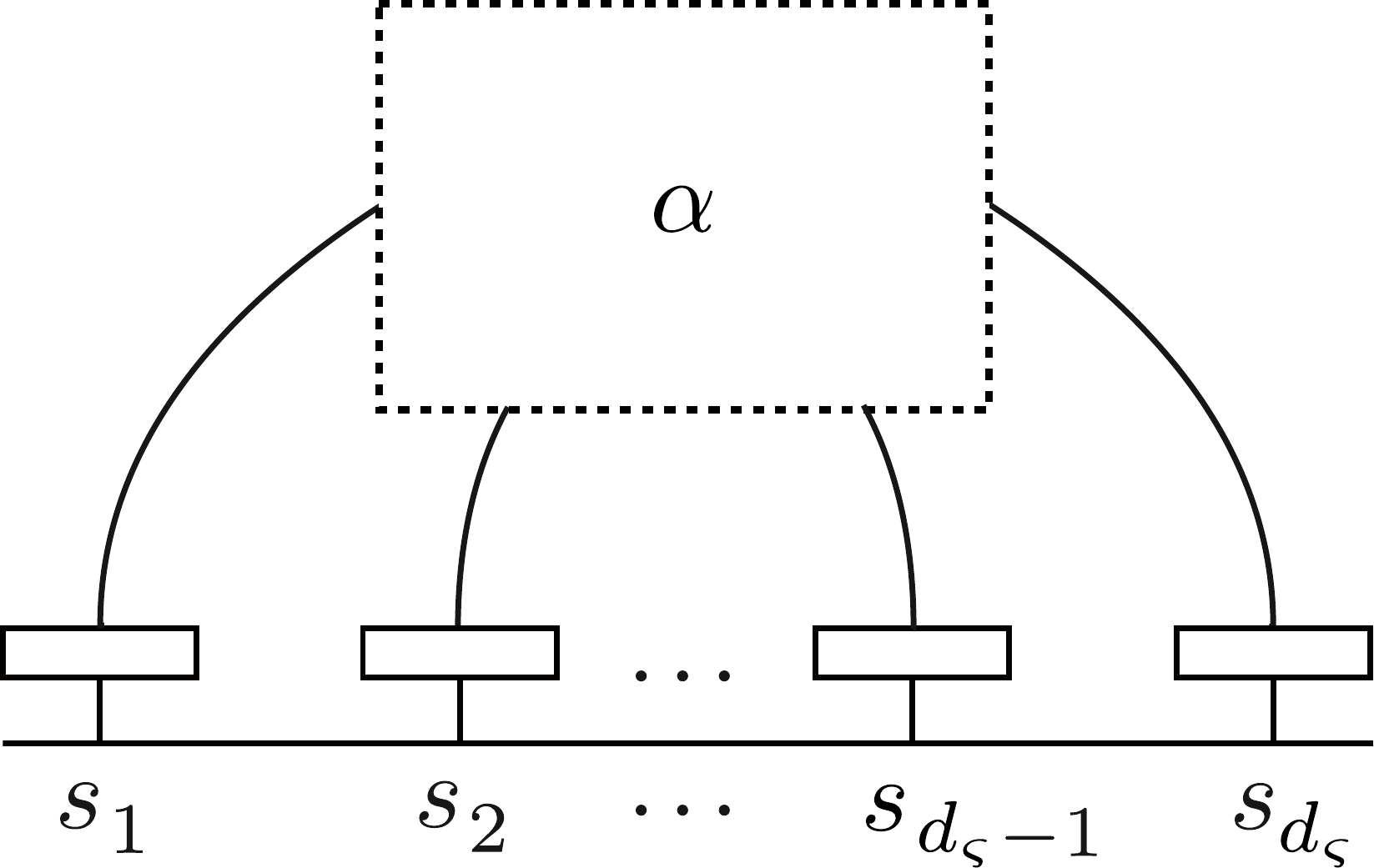} ,}}
\end{align}
for some ordinary link state $\alpha \in \smash{\LS_{\Summed_\multii}\super{s}}$.
If $\alpha$ is a $(\Summed_\multii, s)$-link pattern such that~\eqref{JWLinkState3} does not vanish, 
then we also call~\eqref{JWLinkState3} a \emph{$(\multii,s)$-Jones-Wenzl link pattern}.
Examples of $((3,2,2),3)$-Jones-Wenzl link patterns are
\begin{align}
\vcenter{\hbox{\includegraphics[scale=0.275]{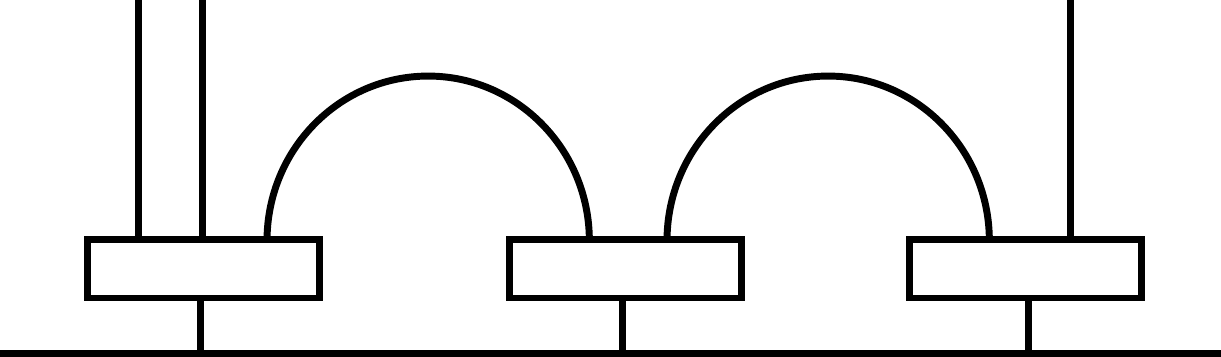}}}
\qquad \qquad \text{and} \qquad \qquad
\vcenter{\hbox{\includegraphics[scale=0.275]{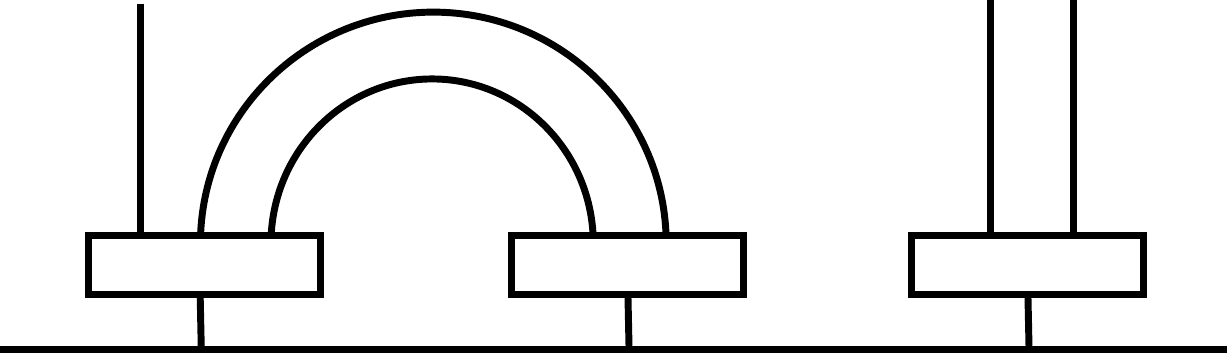}  .}}
\end{align}

%

\begin{center}
\bf Relation of the two algebras $\TL_\multii(\nu)$ and $\WJ_\multii(\nu)$
\end{center}

The valenced Temperley-Lieb algebra $\TL_\multii(\nu)$ is isomorphic to the Jones-Wenzl algebra.
The link state modules of these two algebras are isomorphic too.
We formalize this in lemma~\ref{WJLSBasisLem} and explicate it in corollary~\ref{AnIsoCor}.

With $\alpha$ and $\beta$ respectively denoting an arbitrary $\multii$-valenced link state and $\Summed_\multii$-link state, 
we define the maps
\begin{alignat}{3}
\label{EmbeddingsDef1} 
& \WJEmb_\multii (\,\cdot\,) && \colon \LS_\multii \longrightarrow \LS_{\Summed_\multii}, \qquad 
&& \alpha \mapsto \WJEmb_\multii\alpha, \\
\label{ProjHatDef1} 
& \WJProjHat_\multii (\,\cdot\,) && \colon \LS_{\Summed_\multii} \longrightarrow \LS_\multii, \qquad 
&& \beta \mapsto \WJProjHat_\multii\beta, \\
\label{ProjectionDef1} 
& \WJProj_\multii (\,\cdot\,) && \colon \LS_{\Summed_\multii} \longrightarrow \LS_{\Summed_\multii}, \qquad 
&& \beta \mapsto \WJProj_\multii \beta,
\end{alignat}
where $\WJEmb_\multii$, $\smash{\WJProjHat}_\multii$, and $\WJProj_\multii$ are respectively 
the tangles~\eqref{WJCompEmb},~\eqref{WJProjHatEmb}, and~\eqref{WJCompProj}.
Next, for another multiindex $\multiii \in \smash{\{ \OneVec{0} \} \cup \bZpos^\#}$ such that $\max \multiii < \ppmin(q)$,
with $T$ and $U$ respectively denoting an arbitrary $(\multii, \multiii)$-valenced tangle 
and $(\Summed_\multii,\Summed_\multiii)$-tangle, we define the maps
\begin{alignat}{3}
\label{EmbeddingsDef2} 
& \WJEmb_\multii (\,\cdot\,) \WJProjHat_\multiii 
&& \colon \TL_\multii^\multiii(\nu) \longrightarrow \TL_{\Summed_\multii}^{\Summed_\multiii}(\nu), 
\qquad && T \mapsto \WJEmb_\multii T \WJProjHat_\multiii, \\
\label{ProjHatDef2} 
& \WJProjHat_\multii (\,\cdot\,) \WJEmb_\multiii 
&& \colon \TL_{\Summed_\multii}^{\Summed_\multiii}(\nu) \longrightarrow \TL_\multii^\multiii(\nu), \qquad 
&& U \mapsto \WJProjHat_\multii U \WJEmb_\multiii, \\
\label{ProjectionDef2} & \WJProj_\multii (\,\cdot\,) \WJProj_\multiii 
&& \colon \TL_{\Summed_\multii}^{\Summed_\multiii}(\nu) \longrightarrow \TL_{\Summed_\multii}^{\Summed_\multiii}(\nu) , 
\qquad && U \mapsto \WJProj_\multii U \WJProj_\multiii.
\end{alignat}

In lemma~\ref{WJLSBasisLem}, we give a commuting diagram that relates these maps together and states elementary properties about them, 
including their images and kernels.  
To explicate them, we need some further definitions. We group the the $\Summed_\multii$ 
left nodes and $\Summed_\multiii$ right nodes of 
an $(\Summed_\multii,\Summed_\multiii)$-link diagram into the respective left and right bins of nodes
\begin{alignat}{4} \label{lblocks} 
&\text{left:} \quad &&\{ 1, 2, \ldots, \sIndex_1 \} , \quad && \{ \sIndex_1 + 1, \sIndex_1 + 2, \ldots, \sIndex_1 + \sIndex_2 \} , \quad 
&& \{ \sIndex_1 + \sIndex_2 + 1, \sIndex_1 + \sIndex_2 + 2, \ldots, \sIndex_1 + \sIndex_2 + \sIndex_3 \}, \quad \text{etc.} , \\
&\text{right:} \label{rblocks}  \quad &&\{ 1, 2, \ldots, p_1 \} , \quad && \{ p_1 + 1, p_1 + 2, \ldots, p_1 + p_2 \} , \quad 
&& \{ p_1 + p_2 + 1, p_1 + p_2 + 2, \ldots, p_1 + p_2 + p_3 \}, \quad \text{etc.} 
\end{alignat}
Then, we define a \emph{special link diagram} to be a link diagram in $\smash{\LD_{\Summed_\multii}^{\Summed_\multiii}}$ 
that lacks a turn-back link joining two left nodes or two right nodes in a common bin of (\ref{lblocks},~\ref{rblocks}), and we denote
\begin{align}\label{SpecialDiagrams} 
\SpecialDiagram_\multii^\multiii = \{ \text{special link diagrams in $\smash{\LD_{\Summed_\multii}^{\Summed_\multiii}}$} \}. 
\end{align}
For example, below, the left figure is a special link diagram in 
$\SpecialDiagram_\multii^\multiii$ with $\multii = (2,2,3,2)$ and $\multiii = (2,3,2)$, but
the right figure is not such a link diagram: 
\begin{align} 
\vcenter{\hbox{\includegraphics[scale=0.275]{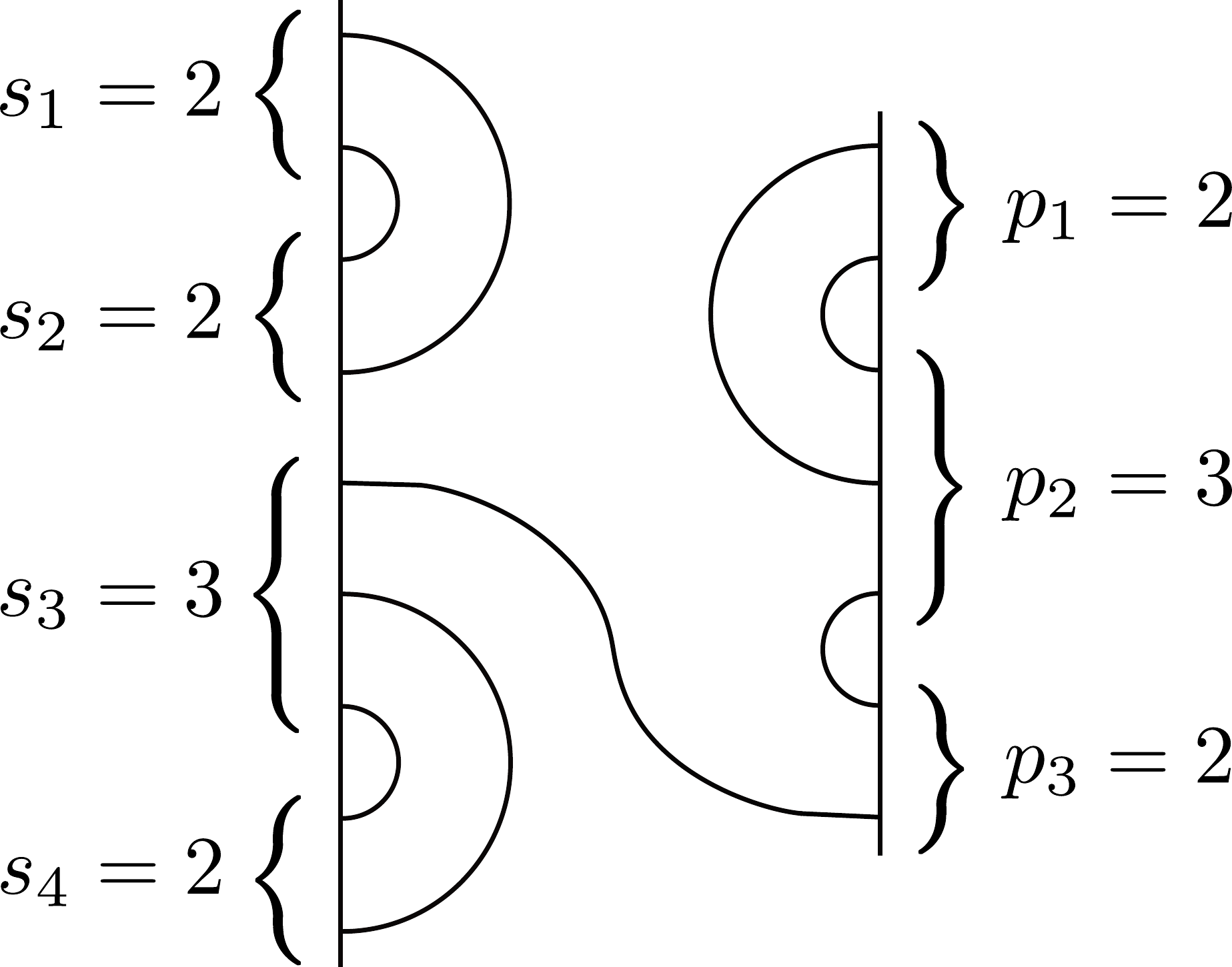}}} 
\qquad \qquad \qquad \qquad
\vcenter{\hbox{\includegraphics[scale=0.275]{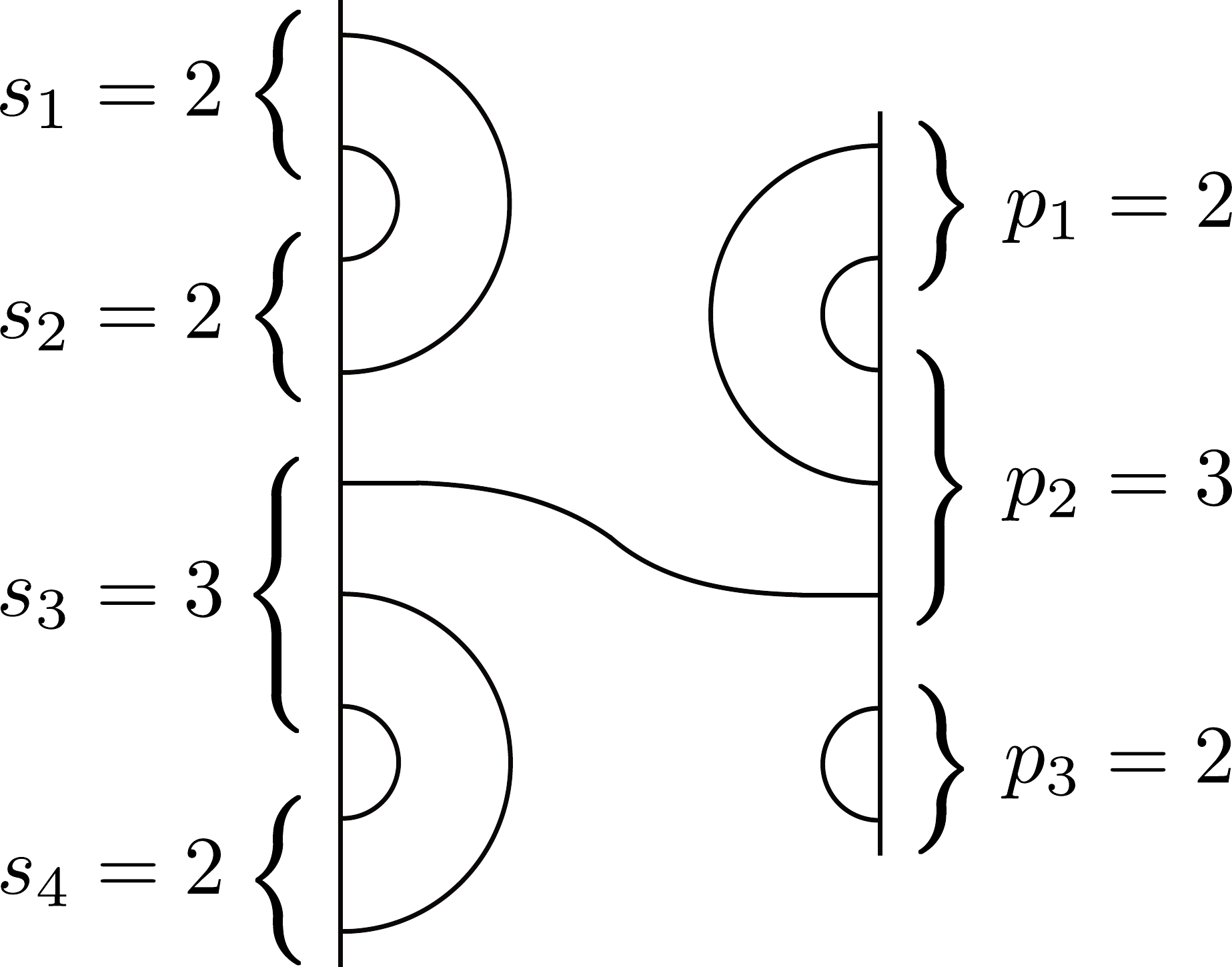} .}} 
\end{align}
We also define a \emph{special link pattern} to be a link pattern in $\LP_{\Summed_\multii}$ that lacks a turn-back link joining two nodes 
in a common bin of~\eqref{lblocks}, and we denote
\begin{align} \label{SpecialPatterns}
\smash{\SpecialPattern_\multii\super{s}} 
:= \big\{ \text{special link patterns in $\LP_{\Summed_\multii}\super{s}$} \big\} , \qquad \qquad
\SpecialPattern_\multii := \bigcup_{s \, \in  \, \DefectSet_{\Summed_\multii}} \SpecialPattern_\multii\super{s}
= \big\{ \text{special link patterns in $\LP_{\Summed_\multii}$} \big\}  .
\end{align}

\begin{lem} \label{WJLSBasisLem} 
Suppose $\max(\multii, \multiii) < \ppmin(q)$.  The following hold:
\begin{enumerate}
\itemcolor{red}
\item \label{wjIt1} \textnormal{Commuting diagrams:} 
\begin{displaymath} 
\begin{tikzcd}[column sep=2cm, row sep=1.5cm]
& \arrow{ld}[swap]{\WJProjHat_\multii (\, \cdot \,)} \arrow{d}{\WJProj_\multii (\, \cdot \,)}
\LS_{\Summed_\multii} \\ 
\LS_\multii 
\arrow{r}{\WJEmb_\multii (\, \cdot \,)}
& \LS_{\Summed_\multii}
\end{tikzcd} 
\qquad \qquad \qquad
\begin{tikzcd}[column sep=2cm, row sep=1.5cm]
& \arrow{ld}[swap]{\WJProjHat_\multii (\, \cdot \,) \WJEmb_\multiii} \arrow{d}{\WJProj_\multii (\, \cdot \,) \WJProj_\multiii}
\TL_{\Summed_\multii}^{\Summed_\multiii}(\nu) \\ 
\TL_\multii^\multiii(\nu) 
\arrow{r}{\WJEmb_\multii (\, \cdot \,) \WJProjHat_\multiii}
& \TL_{\Summed_\multii}^{\Summed_\multiii}(\nu)
\end{tikzcd} 
\end{displaymath}

\item \label{wjIt2} \emph{Basic properties:} 
\begin{enumerate}
\itemcolor{red}
\item[(a):] \label{wjIt2a} The maps $\WJEmb_\multii (\,\cdot\,) \colon \LS_\multii \longrightarrow \LS_{\Summed_\multii}$~\eqref{EmbeddingsDef1} and 
$\WJEmb_\multii (\,\cdot\,) \smash{\WJProjHat}_\multiii \colon \smash{\TL_\multii^\multiii(\nu)} \longrightarrow \smash{\TL_{\Summed_\multii}^{\Summed_\multiii}(\nu)}$~\eqref{EmbeddingsDef2} are linear injections.
\item[(b):] \label{wjIt2b} The maps $\smash{\WJProjHat}_\multii (\,\cdot\,) \colon \LS_{\Summed_\multii} \longrightarrow \LS_\multii$~\eqref{ProjHatDef1} and 
$\smash{\WJProjHat}_\multii (\,\cdot\,) \smash{\WJEmb}_\multiii \colon \smash{\TL_{\Summed_\multii}^{\Summed_\multiii}(\nu)} \longrightarrow \smash{\TL_\multii^\multiii(\nu)}$~\eqref{ProjHatDef2} are linear surjections.
\item[(c):] \label{MapPropIt3} The maps $\WJProj_\multii (\,\cdot\,) \colon \LS_{\Summed_\multii} \longrightarrow \LS_{\Summed_\multii}$~\eqref{ProjectionDef1} and 
$\WJProj_\multii (\,\cdot\,) \WJProj_\multiii \colon \smash{\TL_{\Summed_\multii}^{\Summed_\multiii}(\nu)} \longrightarrow \smash{\TL_{\Summed_\multii}^{\Summed_\multiii}(\nu)}$~\eqref{ProjectionDef2} are linear projections.
\end{enumerate}

\item \label{wjIt3} \emph{Images:}
\begin{enumerate}
\itemcolor{red}
\item[(a):]  
We have 
\begin{align} 
\im \WJEmb_\multii (\,\cdot\,) = \im \WJProj_\multii (\,\cdot\,) \qquad \textnormal{and} \qquad 
\WJEmb_\multii \LP_\multii = \WJProj_\multii \SpecialPattern_\multii, 
\end{align}
and the latter is a basis for the former.
\item[(b):]  
We have 
\begin{align}\im \WJEmb_\multii (\,\cdot\,) \WJProjHat_\multiii  = \im \WJProj_\multii (\,\cdot\,) \WJProj_\multiii \qquad \textnormal{and} \qquad 
\WJEmb_\multii \LD_\multii^\multiii \WJProjHat_\multiii = \WJProj_\multii \SpecialDiagram_\multii^\multiii \WJProj_\multiii, 
\end{align}
and the latter is a basis for the former.
\end{enumerate}

\item \label{wjIt4} \emph{Kernels:}
\begin{enumerate}
\itemcolor{red}
\item[(a):]  
We have 
\begin{align}\ker \WJProjHat_\multii (\,\cdot\,) = \ker \WJProj_\multii (\,\cdot\,) 
\end{align}
and the set $\LP_{\Summed_\multii} \setminus \SpecialPattern_\multii$ is a basis for this kernel.
\item[(b):]  
We have 
\begin{align}\ker \WJProjHat_\multii (\,\cdot\,) \WJEmb_\multiii  = \ker \WJProj_\multii (\,\cdot\,) \WJProj_\multiii
\end{align}
and the set $\smash{\LD_{\Summed_\multii}^{\Summed_\multiii}} \setminus \smash{\SpecialDiagram_\multii^\multiii}$ is a basis for this kernel.
\end{enumerate}

\item \label{wjIt5} \emph{Homomorphism properties:}
\begin{enumerate}
\itemcolor{red}
\item[(a):]  \label{wjIt5a} For all valenced tangles $T \in \TL_\multii^\multiii(\nu)$ and for all valenced link patterns $\alpha \in \LS_\multiii$, we have
\begin{align} \label{HomoProp1} 
\WJEmb_\multii(T\alpha) = (\WJEmb_\multii T \WJProjHat_\multiii) (\WJEmb_\multiii\alpha).
\end{align}
\item[(b):] \label{wjIt5b} For all valenced tangles $T \in \TL_\multii^\varepsilon(\nu)$ and $U \in \TL_\varepsilon^\multiii(\nu)$ with $\max \varepsilon < \ppmin(q)$, we have
\begin{align} \label{HomoProp2} 
\WJEmb_\multii (TU) \WJProjHat_\multiii = (\WJEmb_\multii T \WJProjHat_\varepsilon) (\WJEmb_\varepsilon U \WJProjHat_\multiii).
\end{align}
\end{enumerate}

\item \label{wjIt6} \emph{$s$-grading preservation:}
\begin{enumerate}
\itemcolor{red}
\item[(a):]  \label{wjIt6a} 
The maps~(\ref{EmbeddingsDef1}--\ref{ProjectionDef1}) respect the $s$-grading of their domains: 
\begin{align}\label{sResp} 
\WJEmb_\multii \LS_\multii\super{s} \subset \LS_{\Summed_\multii}\super{s}, \qquad 
\WJProjHat_\multii \LS_{\Summed_\multii}\super{s} = \LS_\multii\super{s}, \qquad 
\WJProj_\multii \LS_{\Summed_\multii}\super{s} \subset \LS_{\Summed_\multii}\super{s}.
\end{align}
\item[(b):]  \label{wjIt6b} 
The maps~(\ref{EmbeddingsDef2}--\ref{ProjectionDef2}) respect the $s$-grading of their domains: 
\begin{align}\label{sResp2} 
\hspace*{-5mm}
\WJEmb_\multii \TL_\multii^{\multiii; \scaleobj{0.85}{(s)}}(\nu) \WJProjHat_\multiii 
\subset \TL_{\Summed_\multii}^{\Summed_\multiii; \scaleobj{0.85}{(s)}}(\nu), \qquad 
\WJProjHat_\multii \TL_{\Summed_\multii}^{\Summed_\multiii; \scaleobj{0.85}{(s)}}(\nu) \WJEmb_\multiii 
= \TL_\multii^{\multiii; \scaleobj{0.85}{(s)}}(\nu), \qquad 
\WJProj_\multii \TL_{\Summed_\multii}^{\Summed_\multiii; \scaleobj{0.85}{(s)}}(\nu) \WJProj_\multiii 
\subset \TL_{\Summed_\multii}^{\Summed_\multiii; \scaleobj{0.85}{(s)}}(\nu).
\end{align}
\end{enumerate}
\end{enumerate}
\end{lem}

\begin{proof} 
We prove items~\ref{wjIt1}--\ref{wjIt6} as follows:
\begin{enumerate}[leftmargin=*]
\itemcolor{red}
\item That the diagrams commute immediately follows from the property $\WJEmb_\multii \smash{\WJProjHat}_\multii = \WJProj_\multii$ 
observed in~\eqref{IdComp}.
\item All of the maps in the assertion are clearly linear.  Furthermore,
\begin{enumerate}
\itemcolor{red}
\item[(a):] with $\smash{\WJProjHat}_\multii \WJEmb_\multii = \mathbf{1}_{\TL_\multii}$ in~\eqref{IdComp}, 
the map $\WJEmb_\multii (\,\cdot\,) \colon \LS_\multii \longrightarrow \LS_{\Summed_\multii}$ is invertible, 

\item[(c):] with $\WJProj_\multii^2 = \WJProj_\multii$ due to~\eqref{WJunit}, the map 
$\WJProj_\multii (\,\cdot\,) \colon \LS_{\Summed_\multii} \longrightarrow \LS_{\Summed_\multii}$ is a projection, and 

\item[(b):] for each valenced link pattern $\alpha \in \LP_\multii$, 
we have $\WJProjHat_\multii \beta = \alpha$, where the link pattern $\beta \in \LP_{\Summed_\multii}$ is
created by separating the $i$:th node of $\alpha$ into $\sIndex_i$ adjacent nodes, for each $i \in \{1,2,\ldots,\np_\multii\}$.  
Hence, $\smash{\WJProjHat}_\multii \colon \LS_{\Summed_\multii} \longrightarrow \LS_\multii$ is a surjection.

\end{enumerate}
The proofs of the asserted properties for 
$\WJEmb_\multii (\,\cdot\,) \smash{\WJProjHat}_\multiii$, $\smash{\WJProjHat}_\multii (\,\cdot\,) \smash{\WJEmb}_\multiii$, 
and $\WJProj_\multii (\,\cdot\,) \WJProj_\multiii$ 
are nearly identical to the above.

\item Because $\smash{\WJProjHat}_\multii$ is a surjection and the left diagram in item~\ref{wjIt1} commutes, 
we immediately have $\im \WJEmb_\multii (\,\cdot\,) = \im \WJProj_\multii (\,\cdot\,)$.  
Moreover, it is straightforward to verify that $\WJEmb_\multii \LP_\multii = \WJProj_\multii \SpecialPattern_\multii$. 
Finally, because $\LP_\multii$ is a basis for $\LS_\multii$ and $\WJEmb_\multii (\,\cdot\,)$ is an injection, 
it follows that $\WJEmb_\multii \LP_\multii$ 
is a basis for $\im \WJEmb_\multii (\,\cdot\,)$.  This proves part~\red{a}. The proof of part~\red{b} is similar.

\item Because $\smash{\WJEmb}_\multii$ is an injection and the left diagram in item~\ref{wjIt1} commutes, we immediately have 
$\ker \smash{\WJProjHat}_\multii (\,\cdot\,) = \ker \WJProj_\multii (\,\cdot\,)$.  
Also, because the set $\LP_{\Summed_\multii} = \SpecialPattern_\multii \cup (\LP_{\Summed_\multii} \setminus \SpecialPattern_\multii)$ is a basis 
for $\LS_{\Summed_\multii}$ and the set $\WJProj_\multii \SpecialPattern_\multii$ is a basis for $\im \WJProj_\multii$ by item~\ref{wjIt3}, 
it follows that the set $\LP_{\Summed_\multii} \setminus \SpecialPattern_\multii$ is a basis for $\ker \WJProj_\multii$.  
This proves part~\red{a}. The proof of part~\red{b} is similar.

\item By idempotent property~\eqref{ProjectorID0} for $\WJProj_\multiii$ 
and by~\eqref{IdComp}, for any valenced link state $\alpha \in \LP_\multiii$ 
and for any valenced tangle $T \in \smash{\TL_\multii^\multiii(\nu)}$, we have 
\begin{align}
T\alpha \overset{\eqref{WJunit}}{=} T \WJProj_\multiii \alpha \qquad \Longrightarrow \qquad 
\WJEmb_\multii (T\alpha) = \WJEmb_\multii (T \WJProj_\multiii \alpha) \overset{\eqref{IdComp}}{=} (\WJEmb_\multii T \WJProjHat_\multiii) (\WJEmb_\multiii \alpha) . 
\end{align}
This proves part~\red{a}, and the proof of part~\red{b} is similar.

\item Item~\ref{wjIt6} is immediate.
\end{enumerate}
This concludes the proof.
\end{proof}

In summary, we may write the commuting diagrams in item~\ref{wjIt1} as
\begin{displaymath} 
\begin{tikzcd}[column sep=2cm, row sep=1.5cm]
& \arrow{ld}[swap]{\WJProjHat_\multii (\, \cdot \,)} \arrow{d}{\WJProj_\multii (\, \cdot \,)}
\LS_{\Summed_\multii}\super{s} \\ 
\LS_\multii\super{s} 
\arrow{r}{\WJEmb_\multii (\, \cdot \,)}
& \PS_\multii\super{s} = \WJProj_\multii  \LS_{\Summed_\multii}\super{s}
\end{tikzcd} 
\qquad \qquad \qquad
\begin{tikzcd}[column sep=2cm, row sep=1.5cm]
& \arrow{ld}[swap]{\WJProjHat_\multii (\, \cdot \,) \WJEmb_\multiii} \arrow{d}{\WJProj_\multii (\, \cdot \,) \WJProj_\multiii}
\TL_{\Summed_\multii}^{\Summed_\multiii; \scaleobj{0.85}{(s)}}(\nu) \\ 
\TL_\multii^{\multiii; \scaleobj{0.85}{(s)}}(\nu)
\arrow{r}{\WJEmb_\multii (\, \cdot \,) \WJProjHat_\multiii}
& \WJ_\multii^{\multiii; \scaleobj{0.85}{(s)}}(\nu) = \WJEmb_\multii \TL_\multii^{\multiii; \scaleobj{0.85}{(s)}} \WJProjHat_\multiii 
\end{tikzcd} 
\end{displaymath}

\begin{cor} \label{AnIsoCor}
Suppose $\max \multii < \ppmin(q)$. The following hold:
\begin{enumerate}
\itemcolor{red}
\item \label{MapIt1} 
The linear map $\WJEmb_\multii (\, \cdot \,) \WJProjHat_\multii$
sending $\TL_\multii(\nu) \longrightarrow \WJ_\multii(\nu)$ via
\begin{align} \label{map1}
\vcenter{\hbox{\includegraphics[scale=0.275]{e-GenericTangle_valenced.pdf}}} 
\qquad \qquad \longmapsto \qquad \qquad 
\vcenter{\hbox{\includegraphics[scale=0.275]{e-GenericTangle_WJ.pdf} ,}}
\end{align}
where $T \in \smash{\TL_{\Summed_\multii}}(\nu)$, is an isomorphism of unital, associative algebras.
\item \label{MapIt2} 
The linear map $\WJEmb_\multii (\, \cdot \,)$
sending $\smash{\LS_\multii\super{s}} \longrightarrow \smash{\PS_\multii\super{s}}$ via
\begin{align}\label{map2} 
\vcenter{\hbox{\includegraphics[scale=0.275]{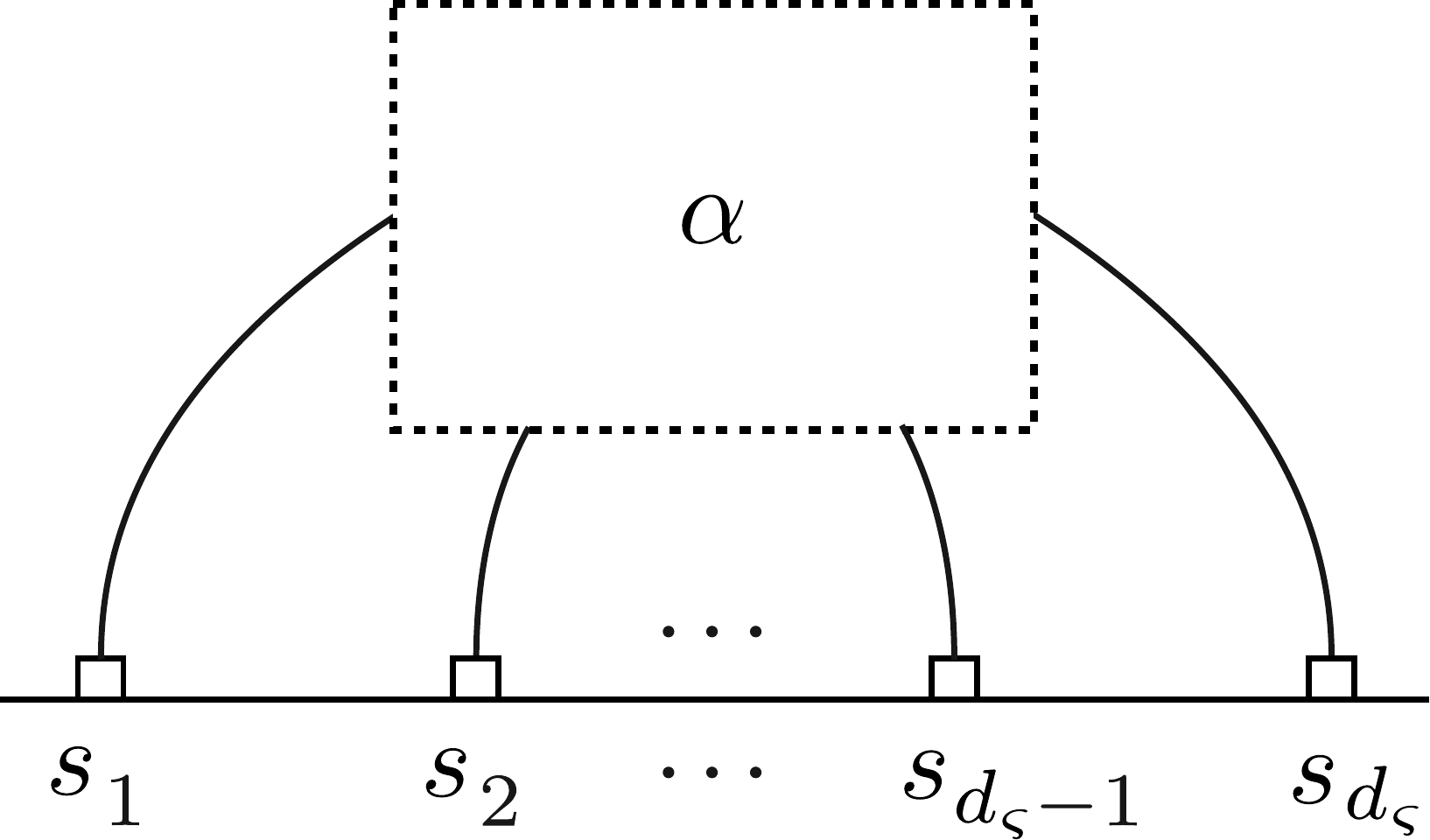}}} 
\qquad \qquad \longmapsto \qquad \qquad 
\vcenter{\hbox{\includegraphics[scale=0.275]{e-GenericLinkState_WJ.pdf} ,}}
\end{align} 
where $\alpha \in \smash{\LS_{\Summed_\multii}\super{s}}$, is an isomorphism of modules 
$($from a $\TL_\multii(\nu)$-module to a $\WJ_\multii(\nu)$ module$)$.
\end{enumerate}
\end{cor}
\begin{proof}
By lemma~\ref{WJLSBasisLem}, 
the map $\WJEmb_\multii (\,\cdot\,) \WJProjHat_\multii \colon \TL_\multii(\nu) \longrightarrow \TL_{\Summed_\multii}(\nu)$ 
is a linear injection with image $\smash{\WJ_\multii(\nu)}$, and by property~\eqref{HomoProp2}, 
this map is also a homomorphism of algebras. This proves item~\ref{MapIt1}. For item~\ref{MapIt2},
the map $\WJEmb_\multii (\, \cdot \,) \colon \smash{\LS_\multii\super{s}} \longrightarrow \smash{\LS_{\Summed_\multii}\super{s}}$
is a linear injection with image $\smash{\PS_\multii\super{s}} := \WJProj_\multii  \LS_{\Summed_\multii}\super{s}$,
and by property~\eqref{HomoProp1} and item~\ref{MapIt1},
this map defines an isomorphism from the $\TL_\multii(\nu)$-module $\smash{\LS_\multii\super{s}}$ 
to the $\WJ_\multii(\nu)$ module $\smash{\PS_\multii\super{s}}$. This proves item~\ref{MapIt2} and finishes the proof.
\end{proof}

The isomorphism from $\WJ_\multii(\nu)$ to $\TL_\multii(\nu)$ is only a cosmetic change. 
As such, the reader may wonder why do we introduce two notations for what are morally identical algebras. 
Here are some partial answers to this question:

\begin{itemize}[leftmargin=*] 

\item 
$\TL_\multii(\nu)$ is well-defined as a set for all values of $\nu \in \bC$, whereas this is not the case for $\WJ_\multii(\nu)$. 
We conjecture that the valenced Temperley-Lieb algebra $\TL_\multii(\nu)$
can be defined as an abstract algebra with generators and relations, for any $\nu \in \bC$.
We pertain to such a definition in~\cite{fp0}.


\item We may think of $\WJ_\multii(\nu)$ and $\TL_\multii(\nu)$ as the collections of all intertwiners 
of two isomorphic but otherwise different $U_q(\mathfrak{sl}_2)$-modules. 
Because the modules are different, we distinguish their algebras of intertwiners.

\item 
Our work in this article is motivated by a problem in conformal field theory, as we discuss in section~\ref{MotivationSec}.
In our application, elements of certain modules of the two algebras $\WJ_\multii(\nu)$ and $\TL_\multii(\nu)$ 
are viewed as different correlation functions, ones of which are certain limits of the other ones. 
We investigate such functions in detail in~\cite{fp2}.


\end{itemize}

\begin{center}
\bf Radical of the link state bilinear form on Jones-Wenzl link states
\end{center}

In this section, we prove results concerning the radicals of $\WJ_\multii(\nu)$-modules needed in section~\ref{rofSect2}.

The $\WJ_\multii(\nu)$-module $\PS_\multii$ 
has a natural bilinear form, given by restricting the bilinear form of its parent module $\LS_{\Summed_\multii}$ to this subspace.  
We define the radical of this bilinear form to be the vector space
\begin{align} 
\rad \PS_\multii := \; & \big\{\alpha \in \PS_\multii \, \big| \, \text{$\BiForm{\alpha}{\beta} = 0$ for all $\beta \in \PS_\multii$} \big\}.
\end{align}
The radical $\rad \PS_\multii$ is a $\WJ_\multii(\nu)$-submodule of $\PS_\multii$, 
it equals a direct sum of the radicals of its submodules,
\begin{align} 
\rad \PS_\multii = \bigoplus_{s \, \in \, \DefectSet_\multii} \rad \PS_\multii\super{s}, \qquad \text{where} \quad
\rad\smash{\PS_\multii\super{s}} := \; & 
\big\{\alpha \in\smash{\PS_\multii\super{s}} \, \big| \, \text{$\BiForm{\alpha}{\beta} = 0$ for all $\beta \in \smash{\PS_\multii\super{s}}$} \big\} ,
\end{align}
and 
$\rad \smash{\PS_\multii\super{s}}$ is a $\WJ_\multii(\nu)$-submodule of $\smash{\PS_\multii\super{s}}$.
The map $\WJEmb_\multii (\, \cdot \,) \colon \LS_\multii \longrightarrow \LS_{\Summed_\multii}$ 
preserves the bilinear form, so 
\begin{align}
\WJEmb_\multii \rad \smash{\LS_\multii\super{s}} = \rad \PS_\multii\super{s} . 
\end{align} 
In particular, corollary~\ref{AnIsoCor} induces an isomorphism of modules between the radicals:
\begin{cor}
Suppose $\max \multii < \ppmin(q)$. 
The linear map $\WJEmb_\multii (\, \cdot \,)$
sending $\smash{\rad \LS_\multii\super{s}} \longrightarrow \smash{\rad \PS_\multii\super{s}}$ via rule~\eqref{map2}
is an isomorphism of modules $($from a $\TL_\multii(\nu)$-module to a $\WJ_\multii(\nu)$ module$)$.
\end{cor}
\begin{proof}
This immediately follows from item~\ref{MapIt2} of corollary~\ref{AnIsoCor} and the fact that $\WJEmb_\multii (\, \cdot \,)$ preserves the bilinear form.
\end{proof}

On the other hand, by definition~\eqref{ProjsSpDefns}, we have 
\begin{align} 
\PS_\multii\super{s} = \WJProj_\multii \LS_{\Summed_\multii}\super{s} . 
\end{align} 
In fact, a similar property holds after we replace $\PS_\multii\super{s}$ and $\smash{\LS_{\Summed_\multii}\super{s}}$ in this equation by their radicals:

\begin{lem}\label{EmbProjLem} 
Suppose $\max \multii < \ppmin(q)$.  We have 
\begin{align} \label{EmbProj} 
\rad \PS_\multii\super{s} = \PS_\multii\super{s} \cap \rad \LS_{\Summed_\multii}\super{s} 
= \WJProj_\multii \rad \LS_{\Summed_\multii}\super{s} . 
\end{align} 
\end{lem}

\begin{proof}  
To begin, we prove the containment 
\begin{align} \label{Conc1} 
\rad \PS_\multii\super{s} \subset \PS_\multii\super{s} \cap \rad \LS_{\Summed_\multii}\super{s}. 
\end{align}  
Indeed, using invariance property~\eqref{InvarProp} of the bilinear form from~\eqref{InvarProp} 
of lemma~\ref{EasyLem2}, the property $\WJProj_\multii^\dagger = \WJProj_\multii$ from~\eqref{ReflectionSymmetryOfP},
and the idempotent property $\WJProj_\multii^2 = \WJProj_\multii$ from~\eqref{ProjectorID0}, we obtain~\eqref{Conc1}: 
\begin{align} 
\alpha \in \rad \smash{\PS_\multii\super{s}} \qquad & 
\overset{\hphantom{\eqref{InvarProp}}}{\Longrightarrow} \qquad \text{$\alpha \in \smash{\PS_\multii\super{s}}$, and 
$\BiForm{\alpha}{\gamma} = 0$, for all $\gamma \in \smash{\PS_\multii\super{s}}$}, \\[.5em]
& \overset{\hphantom{\eqref{InvarProp}}}{\Longrightarrow} \qquad \text{$\alpha = \WJProj_\multii\beta$, for some $\beta \in \smash{\LS_{\Summed_\multii}\super{s}}$, and $\BiForm{\WJProj_\multii\beta}{\WJProj_\multii\delta} = 0$, for all $\delta \in \smash{\LS_{\Summed_\multii}\super{s}}$}, \\
& \overset{\eqref{InvarProp}}{\Longrightarrow} \qquad \text{$\alpha = \WJProj_\multii\beta$, for some $\beta \in \smash{\LS_{\Summed_\multii}\super{s}}$ and 
$\BiForm{\alpha}{\delta} = \BiForm{\WJProj_\multii\beta}{\delta} = 0$, for all $\delta \in \smash{\LS_{\Summed_\multii}\super{s}}$}, \\[.5em]
& \overset{\hphantom{\eqref{InvarProp}}}{\Longrightarrow} \qquad \alpha \in \smash{\PS_\multii\super{s}} \cap \rad \smash{\LS_{\Summed_\multii}\super{s}} .
\end{align}

Next, we prove the containment
\begin{align} \label{Conc2} 
\PS_\multii\super{s} \cap \rad \LS_{\Summed_\multii}\super{s} \subset \WJProj_\multii \rad \LS_{\Summed_\multii}\super{s} . 
\end{align} 
Indeed,~\eqref{Conc2} follows from the idempotent property $\WJProj_\multii^2 = \WJProj_\multii$ from~\eqref{ProjectorID0}:
\begin{align} 
\alpha \in \PS_\multii\super{s} \cap \rad \LS_{\Summed_\multii}\super{s} \qquad 
& \Longrightarrow \qquad 
\text{$\alpha = \WJProj_\multii \beta \in \rad \smash{\LS_{\Summed_\multii}\super{s}}$, for some $\beta \in \smash{\LS_{\Summed_\multii}\super{s}}$}, \\
& \Longrightarrow \qquad \alpha = \WJProj_\multii \beta = \WJProj_\multii^2 \beta = \WJProj_\multii \alpha \in \WJProj_\multii \rad \LS_{\Summed_\multii}\super{s} .
\end{align}

To finish, we prove the containment 
\begin{align} \label{Conc3} 
\WJProj_\multii \rad \LS_{\Summed_\multii}\super{s} \subset \rad \PS_\multii\super{s} . 
\end{align} 
Indeed, we obtain~\eqref{Conc3} by
using invariance property~\eqref{InvarProp} of the bilinear form from 
lemma~\ref{EasyLem2}:
\begin{align} 
\alpha \in \WJProj_\multii \rad \LS_{\Summed_\multii}\super{s} \qquad & \Longrightarrow \qquad \text{$\alpha = \WJProj_\multii \beta$, for some $\beta \in \smash{\rad \LS_{\Summed_\multii}\super{s}}$}, \\
\vspace*{-2mm}
& \Longrightarrow \qquad \text{$\alpha \in \PS_\multii\super{s}$ and $\BiForm{\alpha}{\gamma} = \BiForm{\WJProj_\multii\beta}{\gamma} \overset{\eqref{InvarProp}}{=} \BiForm{\beta}{\WJProj_\multii^\dagger\gamma} = 0$, for all $\gamma \in \PS_\multii\super{s},$} \\[.5em]
& \Longrightarrow \qquad \alpha \in \rad \PS_\multii\super{s} .
\end{align}
Finally, combining (\ref{Conc1},~\ref{Conc2},~\ref{Conc3}) gives the sought equalities~\eqref{EmbProj}.
\end{proof}

\begin{cor} \label{EmbProjCor} 
Suppose $\max \multii < \ppmin(q)$.  We have 
\begin{align} \label{EmbProj22} 
\rad \LS_\multii\super{s} = \WJProjHat_\multii \rad \LS_{\Summed_\multii}\super{s} .
\end{align} 
\end{cor} 

\begin{proof} 
By lemma~\ref{WJLSBasisLem}, the map $\WJEmb_\multii (\, \cdot \,)$
is a linear injection from $\LS_\multii$ to $\PS_\multii$ that respects 
the $s$-grading~\eqref{LSDirSum2} of its domain and the bilinear form~\eqref{LSBiFormExt} of $\LS_\multii$. 
Applying its inverse, i.e., $\smash{\WJProjHat}_\multii$, to equalities~\eqref{EmbProj} of lemma~\ref{EmbProjLem} gives
\begin{align} \label{EmbProj3} 
\rad \smash{\LS_\multii\super{s}} = \WJProjHat_\multii \rad \PS_\multii\super{s} = \WJProjHat_\multii \WJProj_\multii \rad \LS_{\Summed_\multii}\super{s} . 
\end{align} 
By the idempotent property~\eqref{ProjectorID0} for Jones-Wenzl projector, 
we evidently have $\smash{\WJProjHat_\multii} \WJProj_\multii = \smash{\WJProjHat_\multii}$.  
After inserting this last simplification into~\eqref{EmbProj3}, we arrive with~\eqref{EmbProj22}.
\end{proof}

\section{Trivalent link states at roots of unity} \label{RadicalAppendix}

We recall from section~\ref{ConformalBlocksSect} that for each valenced link pattern $\alpha \in \LP_\multii$, 
the trivalent link state $\hcancel{\alpha} \in \LS_\multii$ is defined 
by replacing open vertices by closed ones, beginning from the rightmost vertex and proceeding leftwards, until 
encountering one of the situations in definition~\ref{TrivalentLinkStateDef}.
The procedure terminates at a special index $J = J_\alpha(q)$ defined in~\eqref{Jindex2}.
In this appendix, we prove that $\hcancel{\alpha}$ is well-defined when 
$\max \multii < \ppmin(q)$ but $\Summed_\multii \geq \ppmin(q)$, as we state in remark~\ref{TrivLinkStateRem}.

We fix $q \in \bC^\times$ and $\alpha \in \LP_\multii$ throughout.
For $q'  \in \bC^\times$ with $\ppmin(q') = \infty$, we let $\hcancel{\alpha}_{q'}$ 
denote the trivalent link state $\hcancel{\alpha} \in \LS_\multii$ with $q$ perturbed to $q'$
but the special index $J$ fixed as in~\eqref{Jindex2}, so that $J = J_\alpha(q) \neq J_\alpha(q') = - \infty$.
We note that because $\ppmin(q') = \infty$, projector boxes of all sizes exist,
so $\hcancel{\alpha}_{q'}$ is well-defined.
Our goal is to show that we may define $\hcancel{\alpha}$ by analytic continuation,
to be the limit of $\hcancel{\alpha}_{q'}$ as $q' \to q$ along a sequence not containing roots of unity. 
For this purpose, we endow $\LS_\multii$ with the normed topology induced by the sup norm
\begin{align}  \label{supnorm}
\alpha = \sum_{\beta \, \in \, \LP_\multii} c_\beta \, \beta \in \LS_\multii , \quad \text{for some $c_\beta \in \bC$} 
\qquad \qquad \Longrightarrow \qquad \qquad 
\| \alpha \| := \max_{\beta \, \in \, \LP_\multii} c_\beta . 
\end{align} 
Thus, given a sequence $(\alpha_{q'})$ of $\multii$-valenced link states, we have 
$\smash{\underset{q' \to q}{\lim}} \alpha_{q'} = \alpha$ if and only if
$\smash{\underset{q' \to q}{\lim}} \| \alpha_{q'} - \alpha \| = 0$ 
in $\bC$.  
We also note that, by definition~\eqref{supnorm} of the sup norm, if 
\begin{align} \label{alphaCoeff} 
\alpha_{q'} = \sum_{\beta \, \in \, \LP_\multii} c_\beta(q') \beta \quad \qquad \text{and} \qquad \quad
\alpha = \sum_{\beta \, \in \, \LP_\multii} c_\beta \beta, 
\end{align} 
for some constants $c_\beta(q'), c_\beta \in \bC$, then we have 
\begin{align} \label{EquivLim} 
\lim_{q' \to q} \alpha_{q'} = \alpha 
\qquad \overset{\eqref{supnorm}}{\Longleftrightarrow} \qquad 
\lim_{q' \to q} c_\beta(q') = c_\beta , \quad \text{for all $\beta \in \LP_\multii$.} 
\end{align}

\begin{lem} \label{LimitLem} 
The limit $\underset{{q_k' \to q}}{\lim} \hcancel{\alpha}_{q_k'}$ exists, for any sequence $(q_k')_{k\in\bN}$ tending to $q$
such that $\ppmin(q_k') = \infty$, for all $k\in\bN$.  
\end{lem} 

\begin{proof} 
For each $j \in \{J,J+1,\ldots,\Summed_\multii - 1\}$, starting with $j = J$, 
we decompose the projector box of size $r_j$ in the $j$:th closed vertex
of the valenced link state $\hcancel{\alpha}_{q_k'}$ over its internal link diagrams.
The $j$:th closed vertex is one of two types:
\begin{align} 
\label{Case1} 
r_{j+1} = r_j + 1: \qquad \qquad  
\quad  & \vcenter{\hbox{\includegraphics[scale=0.275]{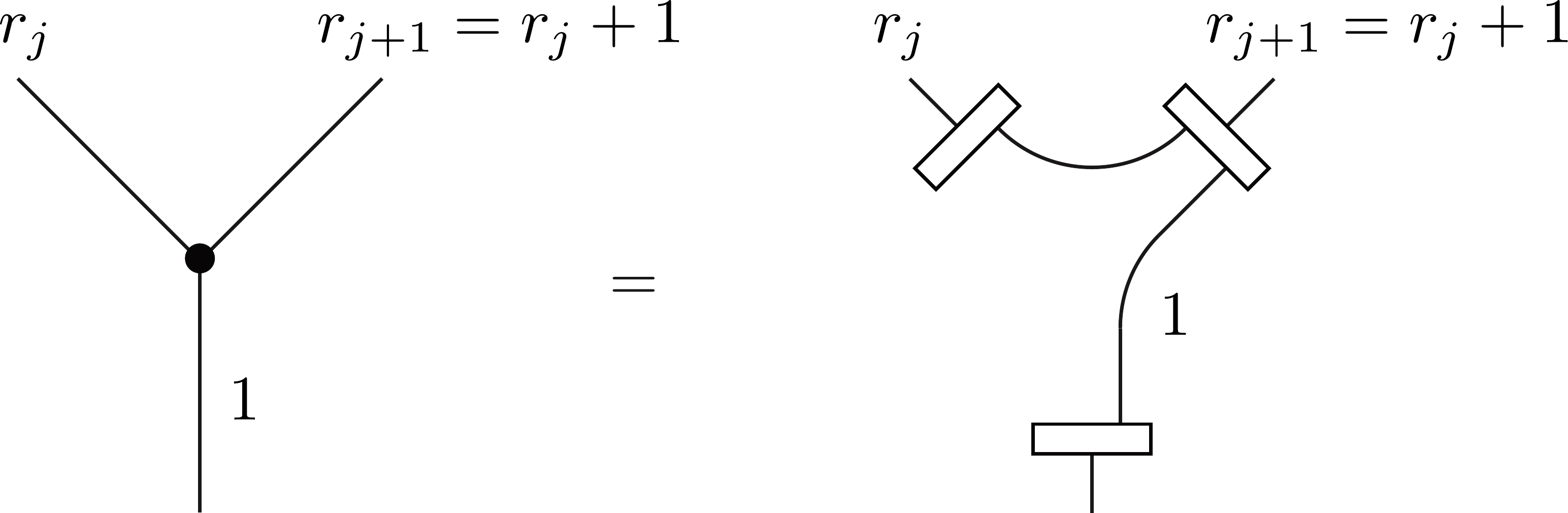} ,}}  \\[1.5em]
\label{Case2} 
r_{j+1} = r_j - 1: \qquad \qquad
\quad  & \vcenter{\hbox{\includegraphics[scale=0.275]{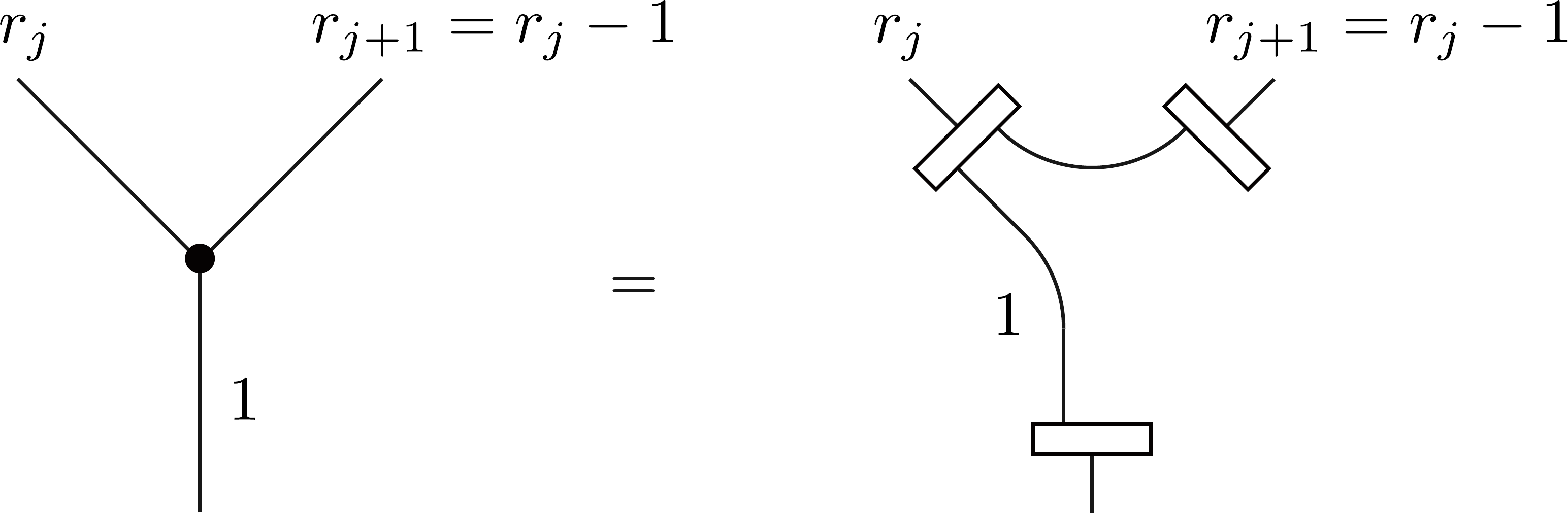} .}} 
\end{align} 
By properties~(\ref{ProjectorID1},~\ref{ProjectorID2}) of the Jones-Wezl projector, we have
\begin{align} 
\label{Expans2-1} 
\eqref{Case1}
\quad  & \overset{\eqref{ProjectorID1}}{=} \quad \vcenter{\hbox{\includegraphics[scale=0.275]{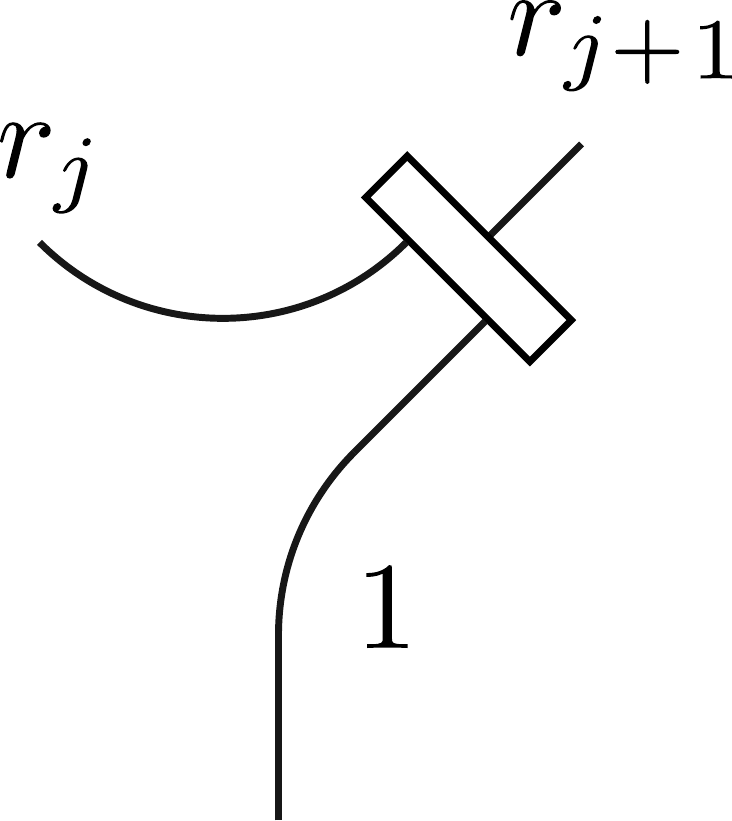} ,}}  \\[1.5em]
\label{Expans2-2} 
\eqref{Case2}
\quad  & \overset{\eqref{ProjectorID1}}{=} \quad \vcenter{\hbox{\includegraphics[scale=0.275]{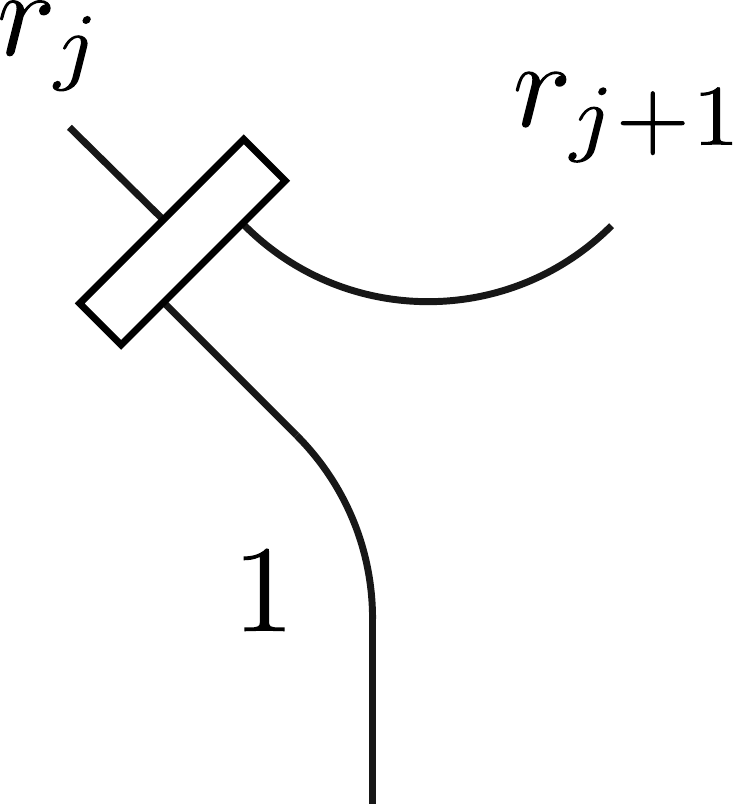}}}  
\quad \overset{\eqref{ProjDecomp}}{=} \quad 
\vcenter{\hbox{\includegraphics[scale=0.275]{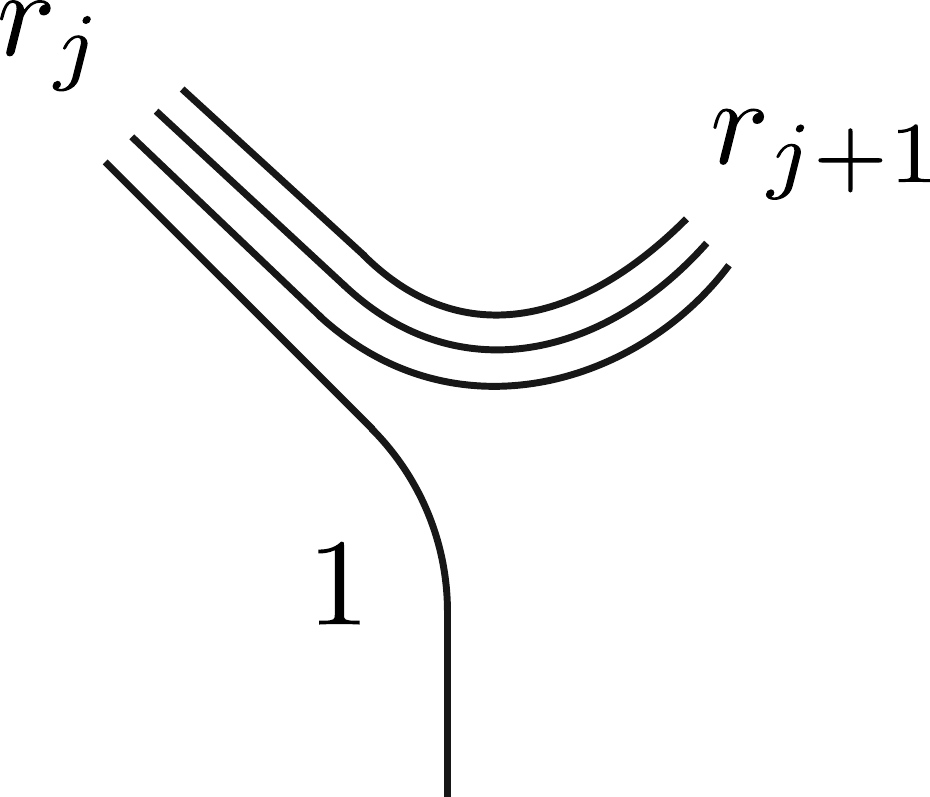}}}  
\quad + \quad \sum_{\substack{T \, \in \, \LD_{r_j}, \\  T \, \neq \, \mathbf{1}_{\TL_{r_j}}}} (\text{coef}_T) 
\,\, \times \,\, \vcenter{\hbox{\includegraphics[scale=0.275]{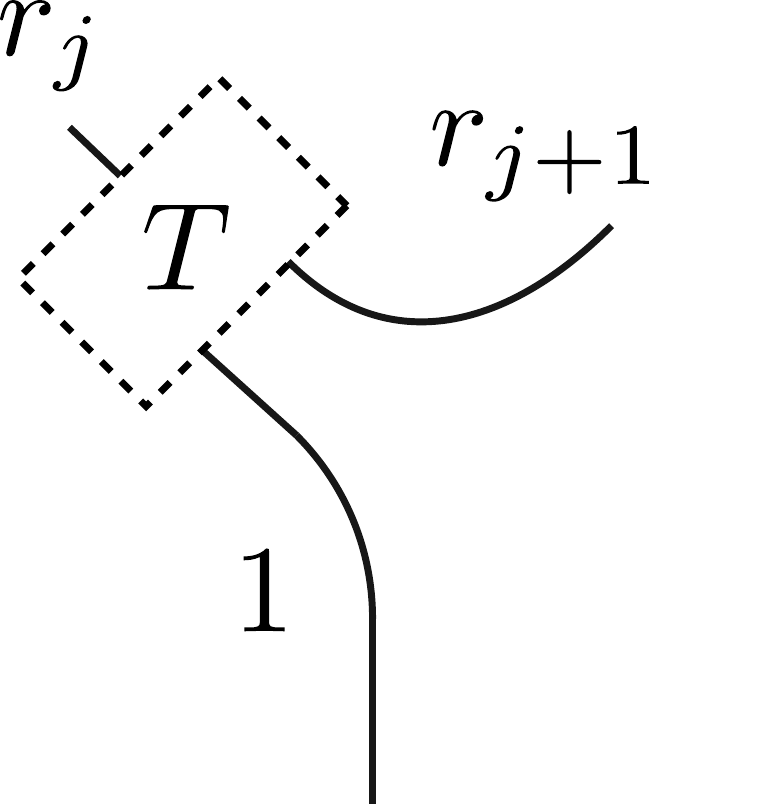} ,}}  
\end{align}  
where the $r_j$-link diagrams $T$ in~\eqref{Expans2-2} have exactly one turn-back link.
Using the formula from~\cite[proposition~\red{A.9}]{fp0}, we find the coefficients in~\eqref{Expans2-2}:
\begin{align} \label{Tcoefs} 
T & \;  =  \; T_i \quad = \quad \vcenter{\hbox{\includegraphics[scale=0.275]{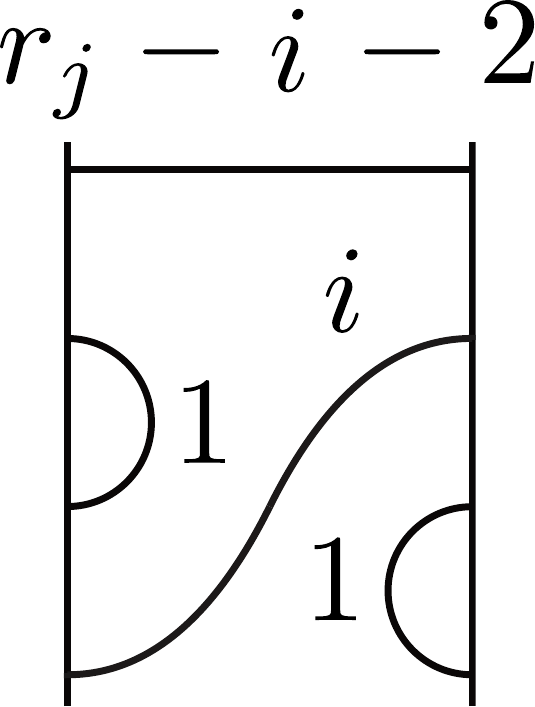} ,}} \qquad \text{with $i \in \{0,1,\ldots,r_j-2\}$} \\[1em]
\Longrightarrow \qquad 
\text{coef}_T & = \text{coef}_{T_i} 
= \frac{[i]_{q_k'}}{[r_j]_{q_k'}} \overset{\eqref{Case2}}{=} \frac{[i]_{q_k'}}{[r_{j+1} + 1]_{q_k'}}. 
\end{align} 
Now, for all $j \in \{J+1, J+2, \ldots, \Summed_\multii - 1\}$, we have $[r_{j+1} + 1]_q \neq 0$ because $\ppmin(q) \nmid ( r_{j+1} + 1 )$.  
Hence, the limit as $q_k' \to q$ of each coefficient in~\eqref{Tcoefs} exists.  
From this fact with~\eqref{EquivLim}, it follows that the limit 
$\smash{\underset{q_k' \to q}{\lim} \hcancel{\alpha}_{q_k'}}$ exists.
\end{proof} 

\section{Temperley-Lieb category} \label{CategorySect}

In this appendix, we discuss a subcategory $\smash{\TL^1(\nu)}$
of the valenced tangle category $\TL(\nu)$~(\ref{cateob},~\ref{catehom}), known as the \emph{Temperley-Lieb category}.
Its object class comprises the special multiindices with all entries equal to one:
\begin{align} 
\smash{\text{Ob} \, \TL^1(\nu) = \big\{ \OneVec{n} \,\big| \, n \in \bZnn \big\}, \qquad \text{where } \qquad
\OneVec{0} := (0) \qquad \text{and} \qquad \OneVec{n} := (\underbrace{1,1,\ldots,1}_{\text{$n$ times}}) \quad \text{for $n \in \bZpos$} ,}
\end{align}
and its morphisms are the $(n,m)$-tangles,
\begin{align} 
\Hom  \TL^1(\nu) = \big\{ \TL_n^m(\nu) \,\big|\, \text{$n, m \in \bZnn$ with $n + m = 0 \Mod 2$} \big\}.
\end{align}
The composition of two morphisms $T, U \in \Hom \smash{\TL^1(\nu)}$ is given by diagram concatenation,
which depends on the fugacity parameter $\nu \in \bC$.
The identity morphism associated with the object $\OneVec{n}$ is the unit~\eqref{Units} of 
the corresponding Temperley-Lieb algebra $\TL_n(\nu) = \TL_n^n(\nu)$.

For later use in~\cite{fp3}, we determine a minimal collection of generators for the morphism class $\Hom \smash{\TL^1(\nu)}$.  
Together with the unit objects~\eqref{Units}, these constitute the \emph{left and right generators}, defined as
\begin{align} \label{LgenForm} 
\Lgen_i \quad := \quad \vcenter{\hbox{\includegraphics[scale=0.275]{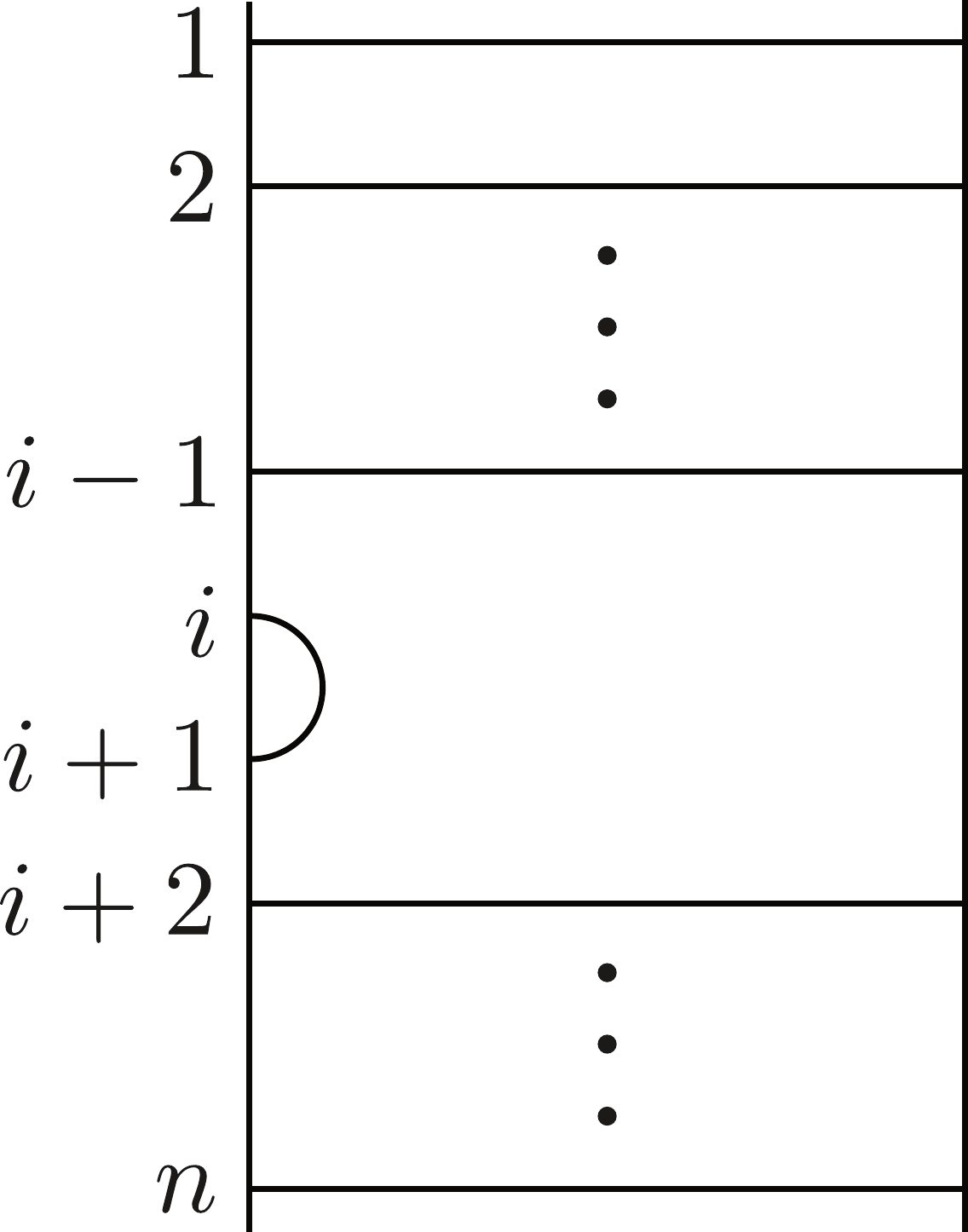}}} \;\; \in \TL_n^{n-2}(\nu) ,
\qquad \qquad  
\Rgen_j\quad := \quad \vcenter{\hbox{\includegraphics[scale=0.275]{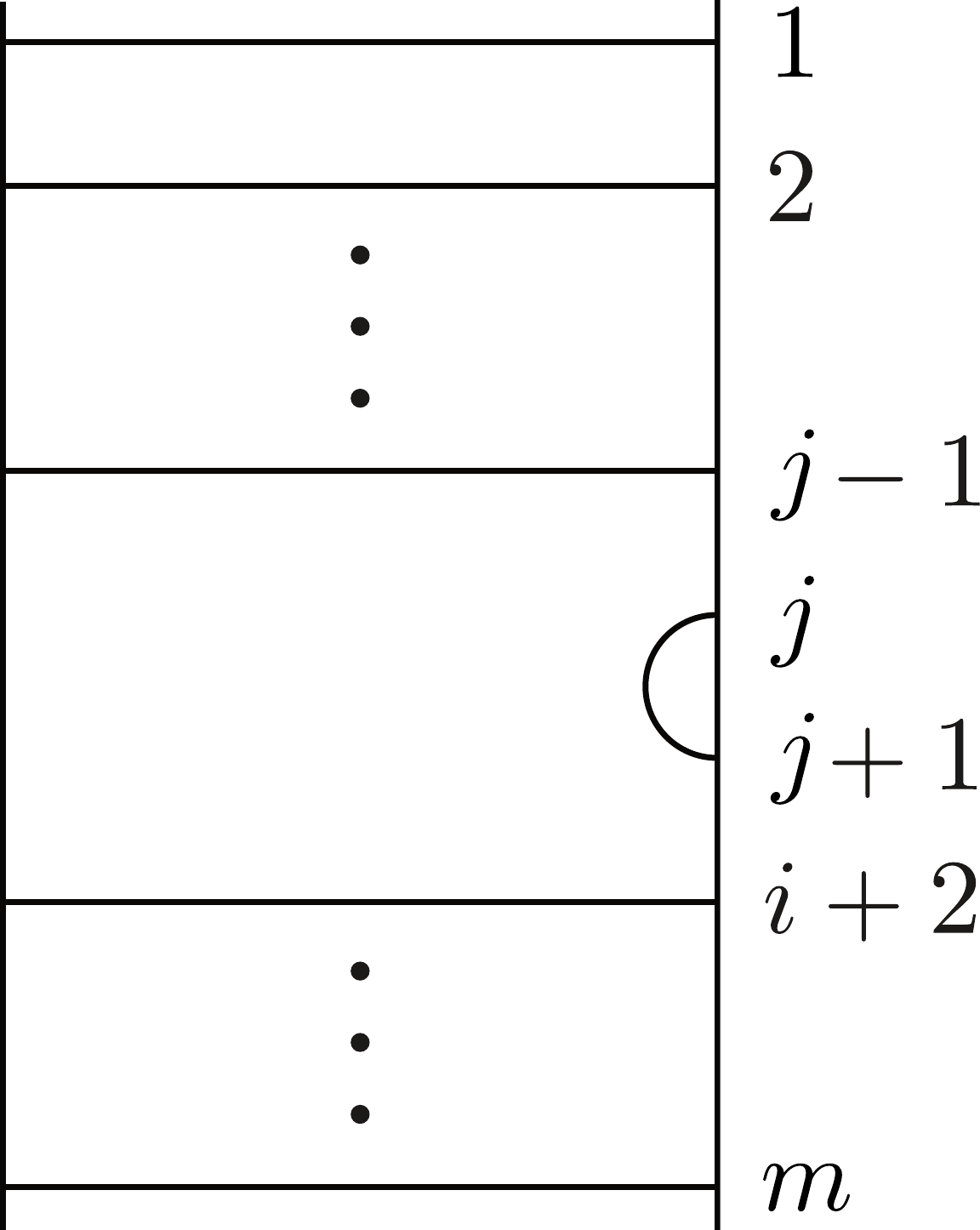}}} \;\; \in \TL_{m-2}^{m}(\nu) ,
\end{align} 
for each $i \in \{ 1, 2, \ldots, n-1 \}$ and $j \in \{ 1, 2, \ldots, m-1 \}$.
In the literature, these are also known as the ``evaluation'' and ``coevaluation'' maps.

Let $T$ be an arbitrary $(n,m)$-link diagram with $s$ crossing links. Then, we can construct $T$
by an insertion of all $(n-s)/2$ left links of $T$ into the unit diagram $\mathbf{1}_{\TL_s}$
by repeated application of the left generators $\Lgen_i$, followed by an insertion of all $(m-s)/2$ right links of $T$
by repeated application of the right generators $\Rgen_j$, that is,
\begin{align} \label{Tword} 
T = \Lgen_{i_{(n-s)/2}} \Lgen_{i_{(n-s)/2 - 1}} \dotsm \Lgen_{i_2} \Lgen_{i_1} \mathbf{1}_{\TL_s} \Rgen_{j_1} \Rgen_{j_2} 
\dotsm \Rgen_{j_{(m-s)/2 - 1}} \Rgen_{j_{(m-s)/2}}. 
\end{align} 
For example, 
\begin{align}
\vcenter{\hbox{\includegraphics[scale=0.275]{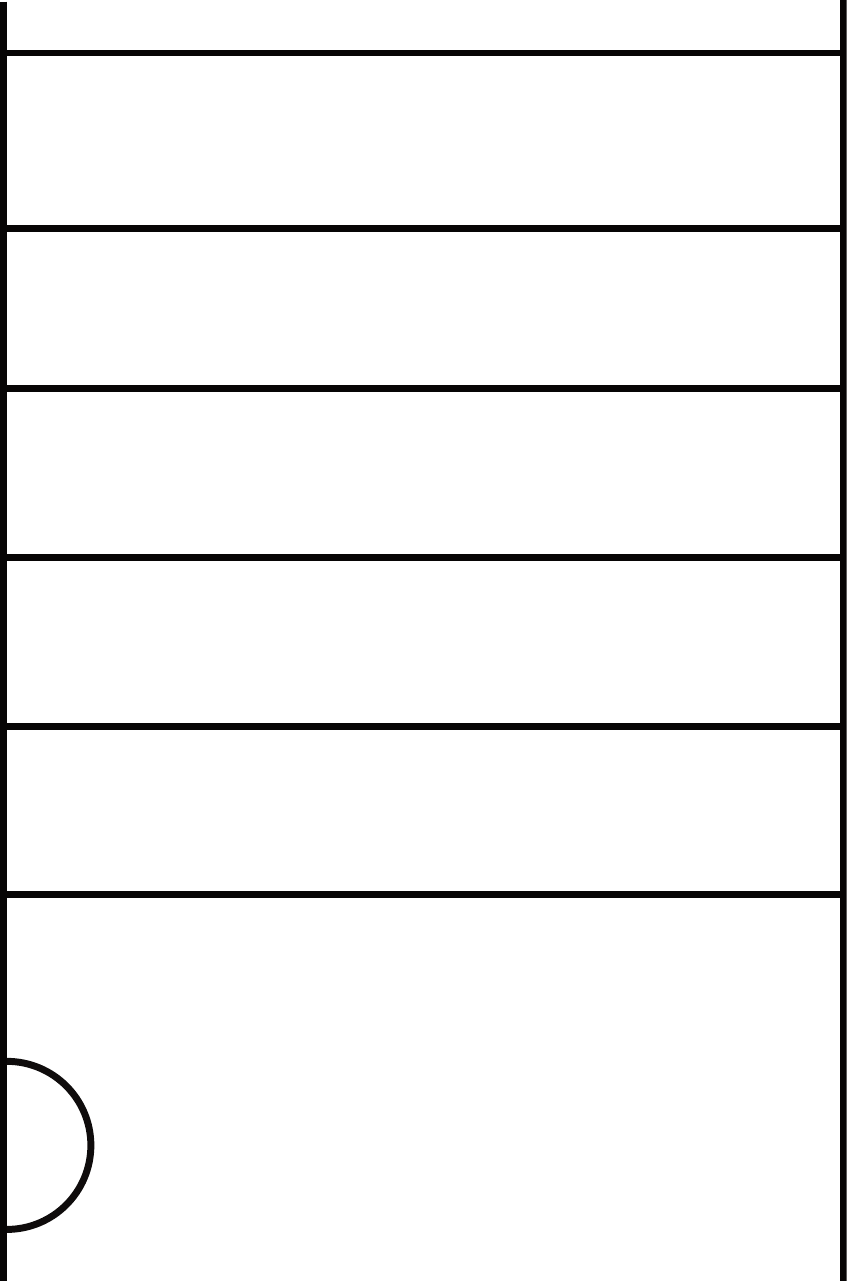}}} \;
\vcenter{\hbox{\includegraphics[scale=0.275]{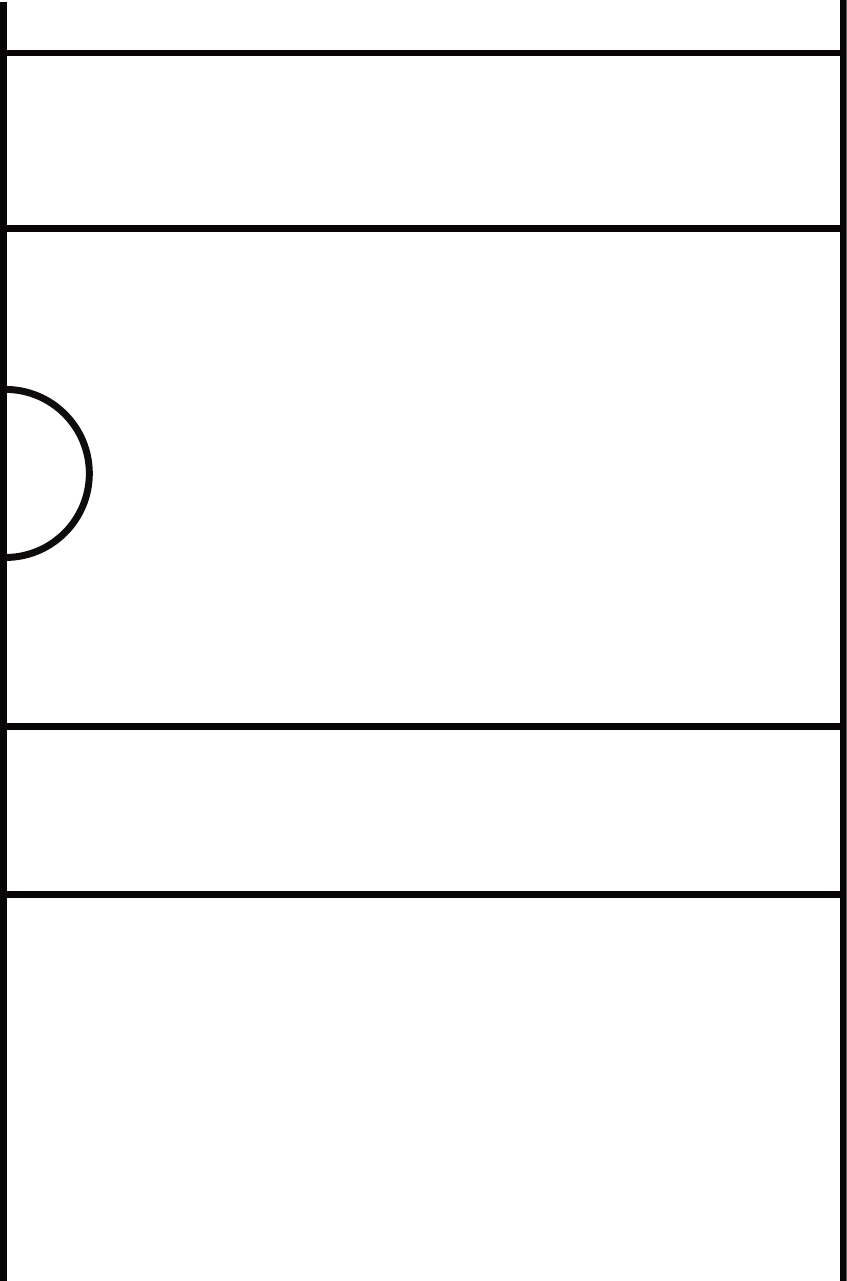}}} \;
\vcenter{\hbox{\includegraphics[scale=0.275]{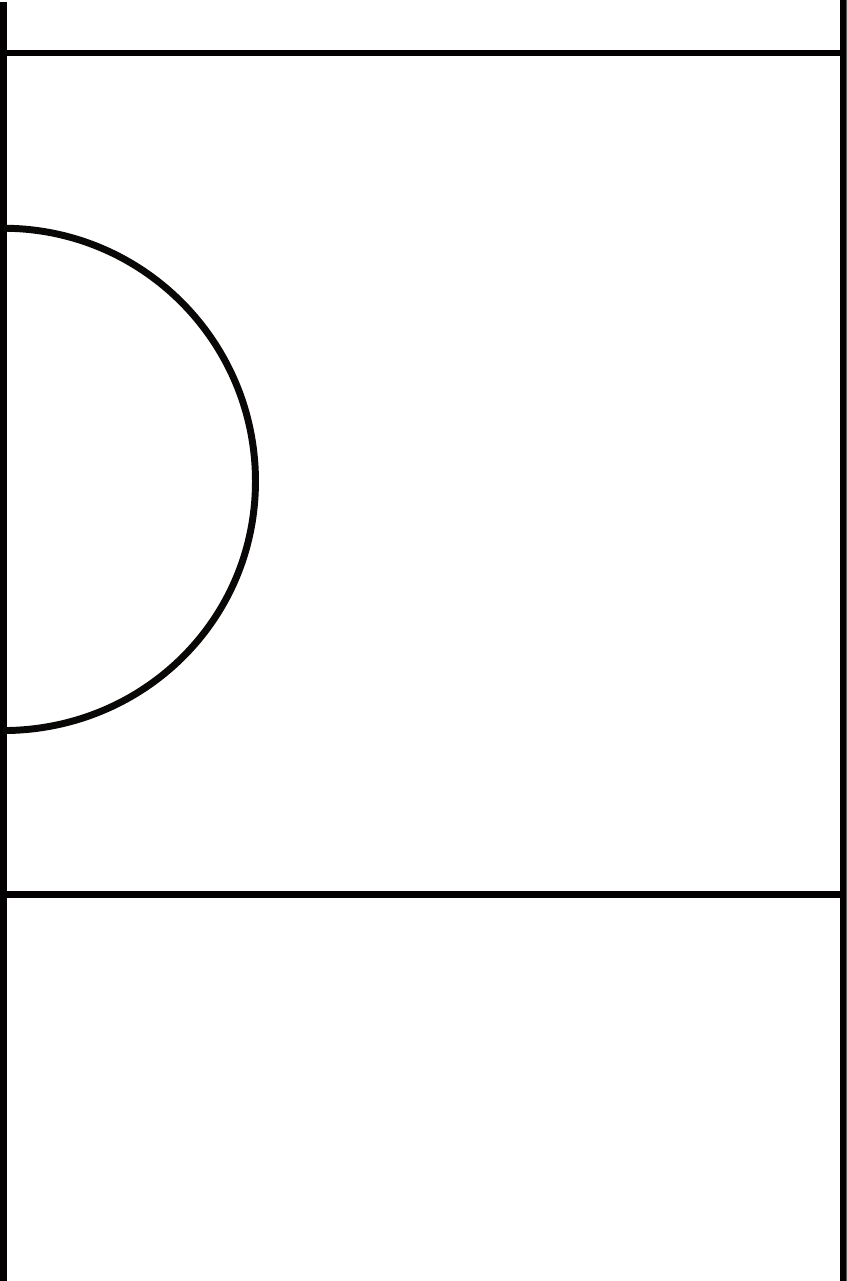}}} \;
\vcenter{\hbox{\includegraphics[scale=0.275]{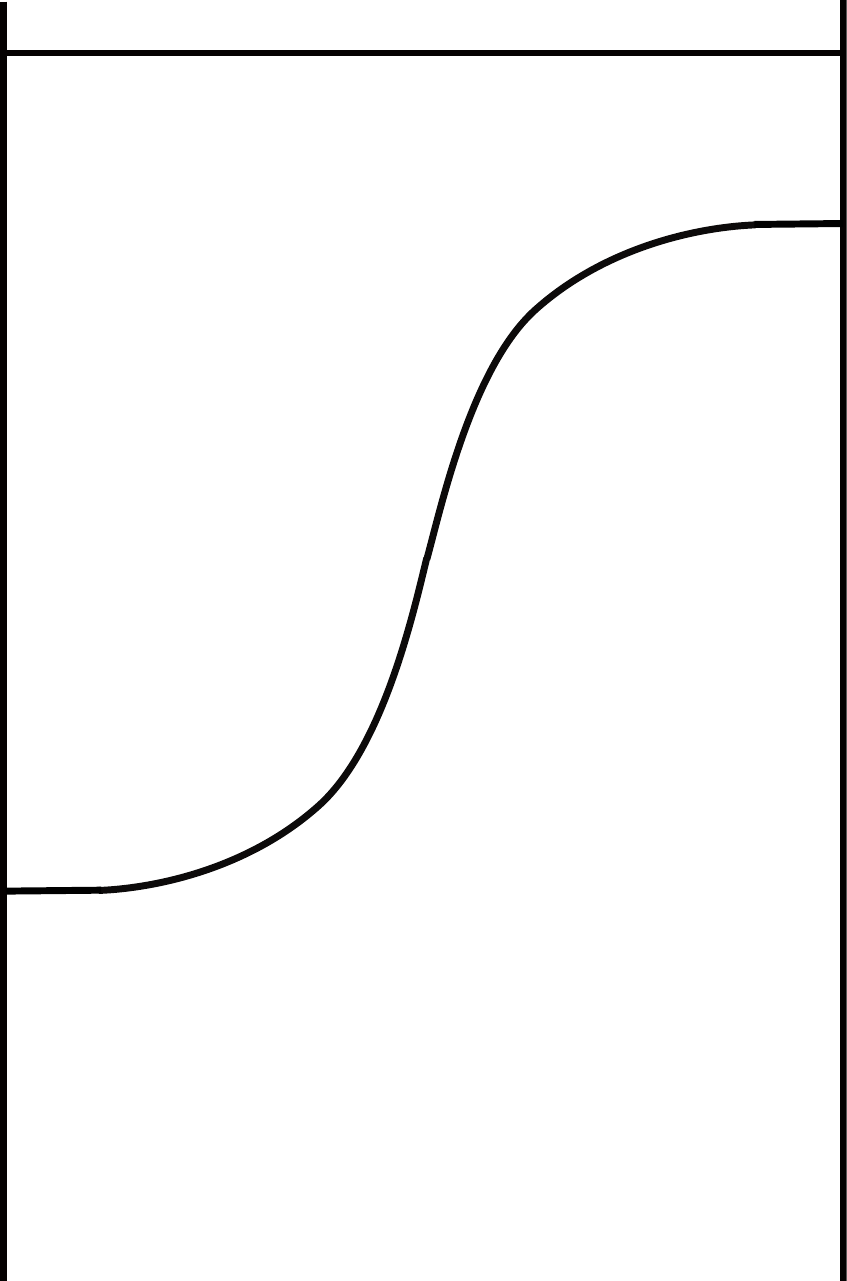}}} \;
\vcenter{\hbox{\includegraphics[scale=0.275]{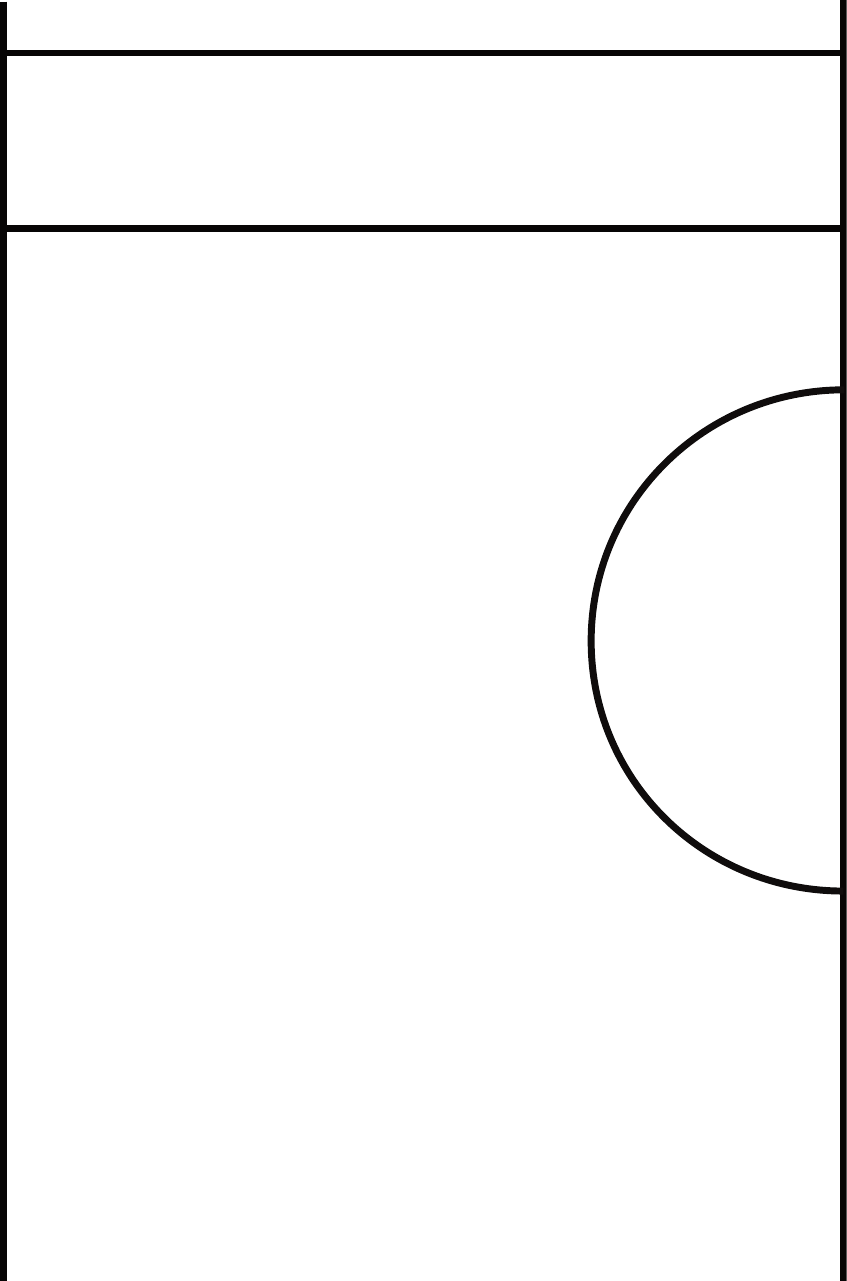}}} \;
\vcenter{\hbox{\includegraphics[scale=0.275]{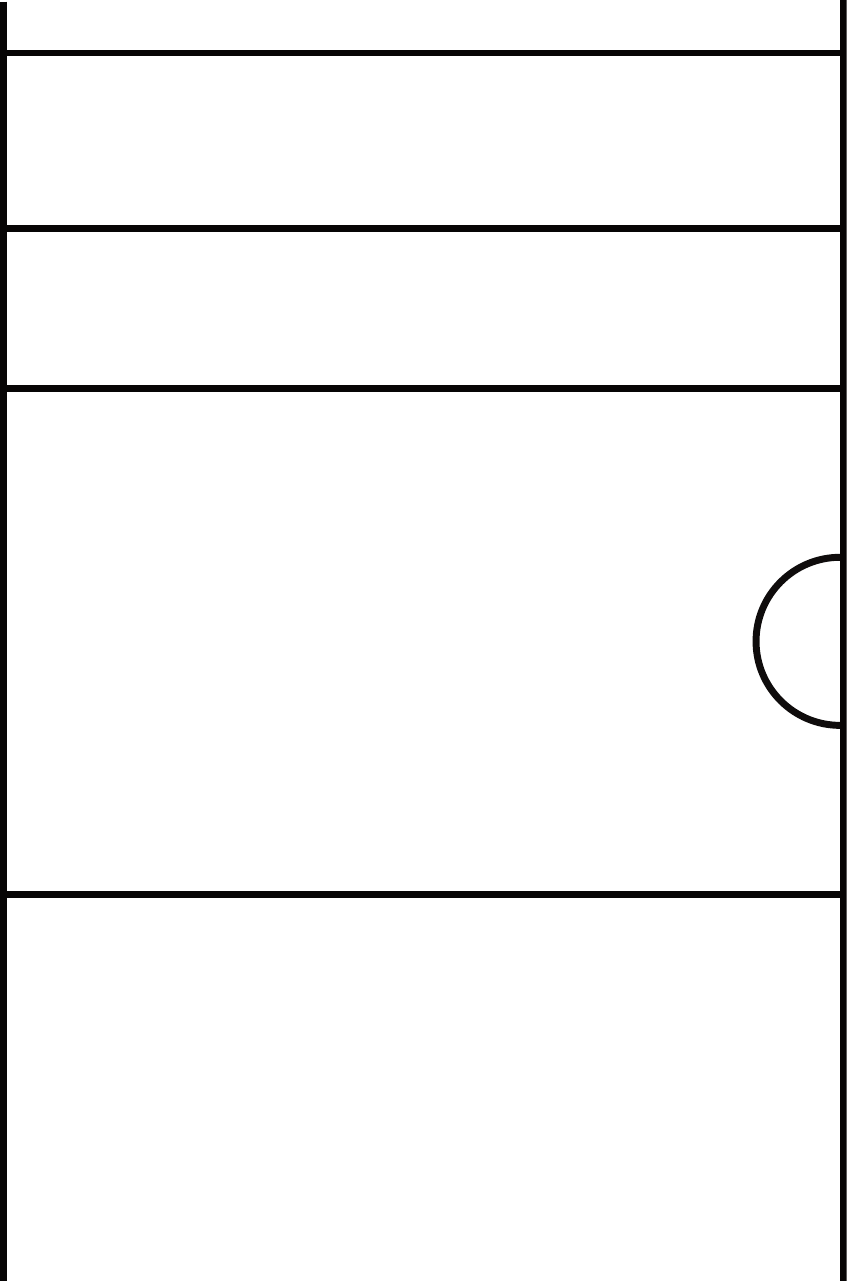}}} 
\end{align} 
gives the tangle 
\begin{align}
\vcenter{\hbox{\includegraphics[scale=0.275]{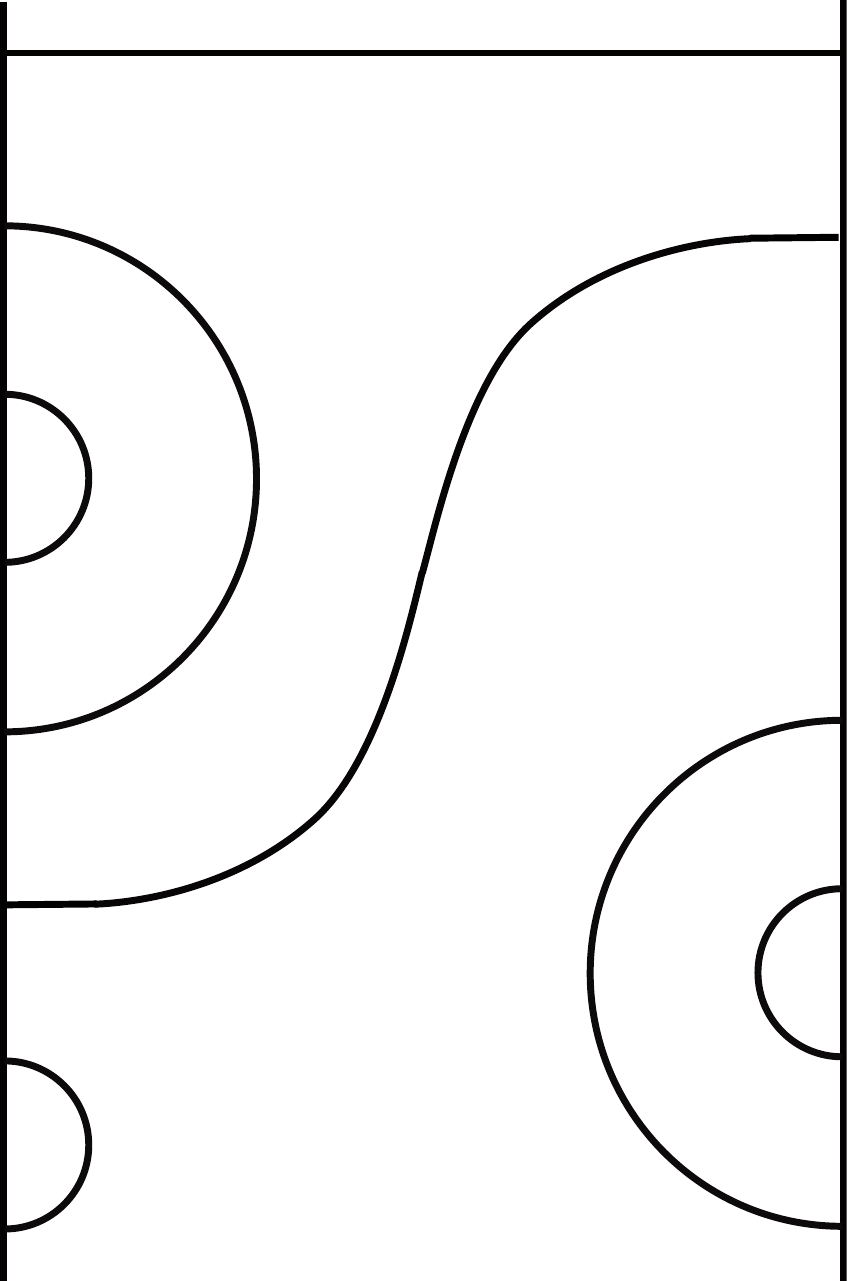}}}  \;\; \in \TL_{8}^{6} .
\end{align} 
As shown, we include the unit in the middle of the product to emphasize that $\Lgen_{i_1}$ is an $(s+2,s)$-link diagram and 
$\Rgen_{j_1}$ is an $(s,s+2)$-link diagram, in spite of the obvious relations
\begin{align}
\Lgen_i \mathbf{1}_{\TL_s} = \Lgen_i \qquad \text{and} \qquad \mathbf{1}_{\TL_s} \Rgen_j = \Rgen_j. 
\end{align} 
In~\eqref{Tword}, we order the left generators $\Lgen_i$ such that if the upper endpoint of one left link of $T$ is above the the upper endpoint of another left link, then 
the former is inserted before the latter, and similarly for the $\Rgen_i$.  This implies that
\begin{align} \label{ordering} 
i_1 < i_2 < \ldots < i_{(n-s)/2}, \qquad \text{and} \qquad j_1 < j_2 < \ldots < j_{(m-s)/2}. 
\end{align} 
We say that any product of left and right generators of the form in (\ref{Tword},~\ref{ordering}) is in \emph{standard form}.

\begin{lem} \label{StdLem} 
Each $(n,m)$-link diagram equals a unique product of left and right generators and the unit diagram in standard form (\ref{Tword},~\ref{ordering}),
 and every such product equals a unique $(n,m)$-link diagram.
\end{lem}

\begin{proof} 
It is evident that every product of the form~\eqref{Tword} equals a unique $(n,m)$-link diagram.  Also, by the above discussion, 
every $(n,m)$-link diagram $T$ equals a product  of the form~\eqref{Tword}, and ordering rule~\eqref{ordering} uniquely encodes the top-to-bottom 
ordering and nesting of the left and right links of $T$.  
\end{proof}

\begin{lem} 
The following is a complete list of independent relations satisfied by the left and right generators:
\begin{align} \label{MaxRelations} 
\Rgen_j \Lgen_i = 
\begin{cases} 
\mathbf{1}_{\TL_s}, & i = j \pm 1, \\ \nu \mathbf{1}_{\TL_s}, & i = j, \\ 
\Lgen_i \Rgen_{j-2}, &  i \leq j-2, \\ \Lgen_{i-2} \Rgen_j, &  j \leq i-2, 
\end{cases} 
\qquad \qquad 
\begin{aligned} 
& \Lgen_j \Lgen_i = \Lgen_{i+2} \Lgen_j, && j \leq i, \\ 
& \Rgen_j \Rgen_{i-2} = \Rgen_i \Rgen_j, && j \leq i, 
\end{aligned} 
\end{align} 
where $s$ is the number of crossing links in $\Lgen_i$ and $\Rgen_j$.
\end{lem}
\begin{proof} 
Each relation~\eqref{MaxRelations} is easy to verify with a diagram. Also,
relations~\eqref{MaxRelations} allow to write any word formed from the right and left generators in standard form.
Now, to see that~\eqref{MaxRelations} are all of the independent relations, we let 
\begin{align} \label{Xtra} 
\sum_{k_1, k_2, \ldots, k_l} c_{k_1, k_2, \ldots, k_l} A_{k_1} A_{k_2} \dotsm A_{k_l} = 0 , 
\qquad \text{with $c_{k_1, k_2, \ldots, k_l} \in \bC$ and $A_{k_p} \in \{ \Lgen_i, \Rgen_j \,|\, i,j \in \bZpos\}$, $k_p \in \bZpos$,} 
\end{align} 
be a relation where all terms $A_{k_1} A_{k_2} \dotsm A_{k_l}$ are in standard form.  
Then by lemma~\ref{StdLem}, each term in~\eqref{Xtra} is multiple of a link diagram particular to that term.  
Because the link diagrams are linearly independent, all of the coefficients 
$c_{k_1, k_2, \ldots, k_l}$ must vanish, so relation~\eqref{Xtra} is trivial.  This proves the assertion. 
\end{proof}

\end{appendices}

\endgroup



\newcommand{\etalchar}[1]{$^{#1}$}

\renewcommand{\bibnumfmt}[1]{\makebox[5.3em][l]{[#1]}}



\begin{thebibliography}{\kern\bibindent}

\bibitem[Bax07]{bax}
R.~J.~Baxter.
\newblock Exactly solved models in statistical mechanics.
\newblock Originally published by Academic Press, London (1982), reprinted by Dover Publications, 2007.

\bibitem[BFK99]{bfk}
J.~Bernstein, I.~Frenkel, and M.~Khovanov. 
\newblock A categorification of the Temperley-Lieb algebra and Schur quotients of $\mathsf{U}_q\mathfrak{sl}(2)$ via projective and Zuckerman functors.
\newblock {\em Selecta Math. (N.S.)}, 5(2):199--241, 1999.

\bibitem[BPZ84]{bpz}
A.~A.~Belavin, A.~M.~Polyakov, and A.~B.~Zamolodchikov.
\newblock Infinite conformal symmetry in two-dimensional quantum field theory
\newblock {\em Nucl. Phys. B}, 241(2):333--380, 1984.

\bibitem[BR99]{br}
H.~Barcelo and A.~Ram. 
\newblock Combinatorial representation theory.
\newblock New Perspectives in Geometric Combinatorics, MSRI Publications, Volume~38, 1999.

\bibitem[BSA88]{bsa}
L.~Benoit and Y.~Saint-Aubin.
\newblock Degenerate conformal field theories and explicit expressions for some null vectors.
\newblock {\em Phys. Lett.} B215(3):517--522, 1988.

\bibitem[CFS95]{cfs}
J.~S.~Carter, D.~E.~Flath, and M.~Saito. 
\newblock The classical and quantum $6j$-symbols.
\newblock Princeton University Press, 1995.

\bibitem[CKL08]{ckl}
G.~Chen, L.~H.~Kauffman and S.~J.~Lomonaco Jr. (editors). 
\newblock Mathematics of quantum computation and quantum technology.
\newblock Chapman \& Hall/CRC Applied Mathematics and Nonlinear Science Series, 2008.

\bibitem[CK12]{bck}
B.~Cooper and V.~Krushkal. 
\newblock Categorification of the Jones-Wenzl projectors.
\newblock {\em Quantum Topol.},3(2):139--180, 2012.

\bibitem[CP94]{cp}
V.~Chari and A.~Pressley.
\newblock A guide to quantum groups.
\newblock Cambridge University Press, 1994.

\bibitem[CR62]{cr}
C.~W.~ Curtis and I.~Reiner.
\newblock Representation theory of finite groups and associative algebras.
\newblock John Wiley \& Sons Inc., 1962.

\bibitem[DF84]{df1}
V.~S.~Dotsenko and V.~A.~Fateev.
\newblock Conformal algebra and multipoint correlation functions in $2D$ statistical models.
\newblock {\em Nucl. Phys. B}, 240(3):312--348, 1984.

\bibitem[DGG97]{fgg}
P.~Di Francesco, O.~Golinelli, and E.~Guitter. 
\newblock Meanders and the Temperley-Lieb algebra. 
\newblock {\em Comm. Math. Phys.}, 186(1):1--59, 1997.

\bibitem[DMS97]{fms}
P.~Di Francesco, R.~Mathieu, and D.~S\'{e}n\'{e}chal.
\newblock Conformal field theory.
\newblock Springer-Verlag New York, 1997.

\bibitem[Dub06]{dub}
J.~Dub\'{e}dat.
\newblock Euler integrals for commuting SLEs.
\newblock {\em J. Stat. Phys.}, 123(6):1183--1218, 2006.

\bibitem[FFK89]{ffk}
G.~Felder, J.~Fr{\"o}hlich, and G.~Keller.
\newblock Braid matrices and structure constants for minimal conformal models.
\newblock {\em Comm. Math. Phys.}, 124(4):647--664, 1989.

\bibitem[FK97]{fk}
I.~B.~Frenkel and M.~G.~Khovanov.
\newblock Canonical bases in tensor products and graphical calculus for $U_q(\mathfrak{sl}_2)$.
\newblock {\em Duke Math. J.}, 87(3):409--480, 1997.

\bibitem[FKK98]{fkk}
I.~Frenkel, M.~Khovanov, and A.~Kirillov~Jr.
\newblock Kazhdan-Lusztig polynomials and canonical basis.
\newblock {\em Transform. Groups}, 3(4):321--336, 1998.

\bibitem[FK15a]{sfk1}
S.~M.~Flores and P.~Kleban.
\newblock A solution space for a system of null-state partial differential equations 1.
\newblock {\em Comm. Math. Phys.}, 333(1):389--434, 2015.

\bibitem[FK15b]{sfk2}
S.~M.~Flores and P.~Kleban.
\newblock A solution space for a system of null-state partial differential equations 2.
\newblock {\em Comm. Math. Phys.}, 333(1):435--481, 2015.

\bibitem[FK15c]{sfk3}
S.~M.~Flores and P.~Kleban.
\newblock A solution space for a system of null-state partial differential equations 3.
\newblock {\em Comm. Math. Phys.}, 333(2):597--667, 2015.

\bibitem[FK15d]{sfk4}
S.~M.~Flores and P.~Kleban
\newblock A solution space for a system of null-state partial differential equations 4.
\newblock {\em Comm. Math. Phys.}, 333(2):669--715, 2015.

\bibitem[FKS06]{fks}
I.~Frenkel, M.~Khovanov, and C.~Stroppel.
\newblock A categorification of finite-dimensional irreducible representations of quantum $\mathsf{U}_q\mathfrak{sl}(2)$ and their tensor products. 
\newblock {\em Selecta Math. (N.S.)}, 12(3-4):379--431, 2006.

\bibitem[FP18a]{fp0}
S.~M.~Flores and E.~Peltola. 
\newblock Generators, projectors, and the Jones-Wenzl algebra. 
\newblock Preprint: \href{http://arxiv.org/abs/1811.12364}{arXiv:1811.12364}, 2018.

\bibitem[FP18b{\etalchar{+}}]{fp3}
S.~M.~Flores and E.~Peltola. 
\newblock Higher quantum and classical Schur-Weyl duality for $\mathfrak{sl}_2$.
\newblock In preparation.

\bibitem[FP18c{\etalchar{+}}]{fp2}
S.~M.~Flores and E.~Peltola. 
\newblock Solution spaces of the Benoit \& Saint-Aubin partial differential equations.
\newblock  In preparation.

\bibitem[FP18d{\etalchar{+}}]{fp1}
S.~M.~Flores and E.~Peltola. 
\newblock Monodromy invariant CFT correlation functions of first column Kac operators.
\newblock  In preparation.

\bibitem[FSK15]{fsk}
S.~M.~Flores, J.~J.~H.~Simmons, and P.~Kleban.
\newblock Multiple-SLE$_\kappa$ connectivity weights for rectangles, hexagons, and octagons.
\newblock Preprint: \href{http://arxiv.org/abs/1505.07756}{arXiv:1505.07756}, 2015.

\bibitem[FSS12]{fss}
I.~Frenkel, C.~Stroppel, and J.~Sussan.
\newblock Categorifying fractional Euler characteristics, Jones-Wenzl projector and $3j$-symbols.
\newblock {\em Quantum Topol.}, 3(2):181--253, 2012.

\bibitem[Fuc92]{fuc}
J.~Fuchs. 
\newblock Affine Lie algebras and quantum groups.
\newblock Cambridge Monographs on Mathematical Physics, Cambridge University Press, 1992. 

\bibitem[FW91]{fw}
G.~Felder and C.~Wieczerkowski.
\newblock Topological representations of the quantum group $U_q(\mathfrak{sl}_2)$.
\newblock {\em Comm. Math. Phys.}, 138(3):583--605, 1991. 

\bibitem[GL96]{gl}
J.~J.~Graham and G.~I.~Lehrer.
\newblock Cellular algebras.
\newblock {\em Invent. Math.}, 123(1):1--34, 1996.

\bibitem[GL98]{gl2}
J.~J.~Graham and G.~I.~Lehrer.
\newblock The representation theory of affine Temperley-Lieb algebras. 
\newblock {\em Enseign. Math.}, 44(3-4):173--218, 1998.

\bibitem[GRAS96]{gras}
C.~G{\'o}mez, M.~Ruiz-Altaba, and G.~Sierra.
\newblock Quantum groups in two-dimensional physics.
\newblock Cambridge University Press, 1996.

\bibitem[GS90]{gs}
C.~G{\'o}mez and G.~Sierra.
\newblock Quantum group meaning of the Coulomb gas.
\newblock {\em Phys. Lett. B}, 240(1-2):149--157, 1990.

\bibitem[GW93]{gwe}
F.~M.~Goodman and H.~Wenzl.
\newblock The Temperley-Lieb algebra at roots of unity. 
\newblock {\em Pacific J. Math.}, 161(2):307--334, 1993.

\bibitem[Hen99]{mh}
M.~Henkel. 
\newblock Conformal invariance and critical phenomena.
\newblock Springer-Verlag Berlin Heidelberg, 1999.

\bibitem[ILZ17]{ilz}
K.~Iohara, G.~I.~Lehrer, and R.~B.~Zhang.
\newblock Temperley-Lieb at roots of unity, a fusion category, and the Jones quotient.
\newblock Preprint: \href{http://arxiv.org/abs/1707.01196}{arXiv:1707.01196}, 2017.


\bibitem[Jim86]{mj2}
M.~Jimbo.
\newblock A $q$ analog of $\mathcal{U}(\mathfrak{gl}(N+1))$, Hecke algebra, and the Yang-Baxter equation.
\newblock {\em Lett. Math. Phys.}, 3(11):247--252, 1986. 

\bibitem[JM79]{jm79}
G.~James and G.~Murphy.
\newblock The determinant of the Gram matrix for a Specht module.
\newblock {\em J. Alg.}, 59:222--235, 1979.

\bibitem[Jon83]{vj}
V.~F.~R.~Jones.
\newblock Index for subfactors. 
\newblock {\em Invent. Math.}, 72:1--25, 1983.

\bibitem[Jon86]{vj3}
V.~F.~R.~Jones.
\newblock Braid groups, Hecke algebras, and type $II_1$ factors.
\newblock In Geometric Methods in Operator Algebras, Proc. US-Japan Seminar (Kyoto 1983),
Pitman Res. Notes Math. Ser., 123:242--273, 1986.


\bibitem[Jon89]{vj2}
V.~F.~R.~Jones.
\newblock On knot invariants related to some statistical mechanical models. 
\newblock {\em Pacific J. Math.}, 137(2):311--334, 1989.

\bibitem[Kas95]{ck}
C.~Kassel.
\newblock Quantum groups.
\newblock Springer-Verlag New York, 1995.

\bibitem[Kau87]{lk2}
L.~H.~Kauffman.
\newblock State models and the Jones polynomial.
\newblock {\em Topology}, 26:395--407, 1987. 


\bibitem[KL94]{kl}
L.~H.~Kauffman and S.~L.~Lins. 
\newblock Temperley-Lieb recoupling theory and invariants of $3$-manifolds.
\newblock Princeton University Press, 1994.

\bibitem[KM13]{km}
N.-G.~Kang and N.~Makarov. 
\newblock Gaussian free field and conformal field theory.
\newblock {\em Ast\'erisque}, 353, 2013.


\bibitem[KP16]{kp}
K.~Kyt\"ol\"a and E.~Peltola.
\newblock Pure partition functions of multiple SLEs.
\newblock {\em Comm. Math. Phys.}, 346(1):237--292, 2016.

\bibitem[KP18]{kp2}
K.~Kyt\"ol\"a and E.~Peltola.
\newblock Conformally covariant boundary correlation functions with a quantum group.
\newblock \emph{J.~Eur. Math. Soc}, to appear, 2018.
\newblock Preprint: \href{http://arxiv.org/abs/1408.1384}{arXiv:1408.1384}.

\bibitem[KRT97]{krt}
C.~Kassel, M.~Rosso, and V.~G.~Turaev.
\newblock Quantum groups and knot invariants.
\newblock American Mathematical Society, 1997. 

\bibitem[Lam91]{lam}
T.~Y.~Lam. 
\newblock A first course in non-commutative rings.
\newblock Graduate Texts in Mathematics, Springer-Verlag Berlin Heidelberg, 1991.

\bibitem[Mar91]{pm}
P.~P.~Martin. 
\newblock Potts models and related problems in statistical mechanics.
\newblock Advances in Statistical Mechanics Vol.~5, World Scientific (Singapore), 1991.

\bibitem[Mar92]{ppm}
P.~P.~Martin. 
\newblock On Schur-Weyl duality, $A_n$ Hecke algebras and quantum $\mathfrak{sl}(N)$ on $\otimes^{n+1} \bC^N$.
\newblock {\em Int. J. Mod. Phys. A}, 7: Supp.~1B, 645--673, 1992.

\bibitem[Mat99]{am}
A.~Mathas. 
\newblock Iwahori-Hecke algebras and Schur algebras of the symmetric group.
\newblock University Lecture Series, Vol.~15, American Mathematical Society (Providence RI) 1999.

\bibitem[MDRR15]{mrr}
A.~Morin-Duchesne, J.~Rasmussen, and D.~Ridout.
\newblock Boundary algebras and Kac modules for logarithmic minimal models.
\newblock {\em Nucl. Phys. B}, 889:677--769, 2015.

\bibitem[MMA92]{mma}
P.~P.~Martin and D.~McAnally.
\newblock On commutants, dual pairs and non-semisimple algebras from statistical mechanics.
\newblock {\em Int. J. Mod. Phys. A}, 7: Supp.~1B, 675--705, 1992.

\bibitem[Mor15]{sm}
S.~Morrison.
\newblock A formula for the Jones-Wenzl projections.
\newblock Preprint: \href{http://arxiv.org/abs/1503.00384}{arXiv:1503.00384}, 2015.

\bibitem[MR89]{mr}
G.~Moore and N.~Reshetikhin. 
\newblock A comment on quantum group symmetry in conformal field theory.
\newblock {\em Nucl. Phys. B}, 328(3):557--574, 1989.

\bibitem[MV94]{mv}
G.~Masbaum and P.~Vogel.
\newblock $3$-valent graphs and the Kauffman bracket. 
\newblock {\em Pacific J. Math.}, 164(2):361--381, 1994.

\bibitem[Pel18]{ep2} 
E.~Peltola.
\newblock Basis for solutions of the Benoit \& Saint-Aubin PDEs with particular asymptotic properties.
\newblock \emph{Ann. Inst. Henri Poincar\'e D,} to appear, 2018.
\newblock Preprint: \href{http://arxiv.org/abs/1605.06053}{arXiv:1605.06053}. 

\bibitem[Pen71]{pen}
R.~Penrose.
\newblock Angular momentum: An approach to combinatorial space-time.
\newblock In Quantum Theory and Beyond, 
Cambridge University Press, 1971. 

\bibitem[Rib14]{sr}
S.~Ribault.
\newblock Conformal field theory on the plane.
\newblock Preprint: \href{http://arxiv.org/abs/1406.4290}{arXiv:1406.4290}, 2014.

\bibitem[RSA14]{rsa}
D.~Ridout and Y.~Saint-Aubin.
\newblock Standard modules, induction, and the structure of the Temperley-Lieb algebra.
\newblock {\em Adv. Theor. Math. Phys.}, 18(5):957--1041, 2014.

\bibitem[Sch08]{sch}
M.~Schottenloher.
\newblock A mathematical introduction to conformal field theory.
\newblock Springer-Verlag Berlin Heidelberg, 2008.


\bibitem[SS14]{ss}
C.~Stroppel and J.~Sussan. 
\newblock Categorified Jones-Wenzl projectors: a comparison.
\newblock In Perspectives in representation theory, Contemp. Math., vol.~610, American Mathematical Society (Providence RI), 333--351, 2014.

\bibitem[TL71]{tl}
H.~Temperley and E.~Lieb. 
\newblock Relations between the `percolation' and `colouring' problem and other graph-theoretic problems
Associated with Regular Plane Lattices: Some Exact Results for the `Percolation' Problem.
\newblock In Proc. Roy. Soc. London Ser. A, 322:251--280, 1971.

\bibitem[Tur94]{vt}
V.~G.~Turaev.
\newblock Quantum invariants of knots and $3$-manifolds.
\newblock Walter de Gruyter (Berlin New York) 1994. 

\bibitem[Wen87]{hw}
H.~Wenzl.
\newblock On sequences of projections.
\newblock {\em C. R. Math. Rep. Acad. Sci. Canada}, 9(1):5--9, 1987.

\bibitem[Wes95]{bw}
B.~Westbury.
\newblock The representation theory of the Temperley-Lieb algebras.
\newblock {\em Math. Zeit.}, 219(1):539--565, 1995.

\end{thebibliography}
\end{document}